\documentclass[letterpaper]{article}

\usepackage[english]{babel}
\usepackage[utf8]{inputenc}
\usepackage[T1]{fontenc}

\usepackage{graphicx} 
\usepackage[margin=2.15cm]{geometry}
\usepackage{algorithm}
\usepackage[noend]{algpseudocode}
\usepackage{amsthm}

\usepackage{amssymb,amsmath}
\usepackage{thmtools}
\usepackage{mathtools}
\usepackage[colorlinks=true, allcolors=blue]{hyperref}
\usepackage{url}
\usepackage[capitalize]{cleveref}

\usepackage{placeins}
\usepackage{subcaption}
\usepackage{booktabs}
\usepackage{microtype}
\usepackage{newtxtext}
\usepackage{newtxmath}

\usepackage{enumitem}
\setlist[1]{labelindent=\parindent}
\setlist[enumerate]{label=(\arabic*)}
\setlist[itemize]{noitemsep}

\usepackage{titling}
\usepackage[dvipsnames]{xcolor}
\definecolor{light-gray}{gray}{0.95}
\usepackage{tikz}
\usepackage{tcolorbox}
\usepackage{wrapfig}
\usepackage{xspace}

\usepackage{authblk}

\declaretheorem[numberwithin=section]{theorem}
\declaretheorem[sibling=theorem]{lemma}
\declaretheorem[sibling=theorem]{example}
\declaretheorem[sibling=theorem]{problem}

\newcommand{\captionspacing}{-0.1cm}

\newcommand{\bst}{\textsc{BST}\xspace}
\newcommand{\wiki}{\textsc{WIKI}\xspace}

\newcommand{\leafc}[1]{L(#1)}

\DeclareMathOperator{\SA}{SA}
\DeclareMathOperator{\LCP}{LCP}

\DeclareMathOperator{\LCE}{LCE}
\DeclareMathOperator{\EV}{EV}
\DeclareMathOperator{\OD}{OD}
\DeclareMathOperator{\ISA}{ISA}

\DeclareMathOperator{\ST}{ST}

\DeclareMathOperator{\ch}{ch}
\DeclareMathOperator{\sd}{sd}
\DeclareMathOperator{\parent}{parent}
\DeclareMathOperator{\preorder}{pre}

\DeclareMathOperator{\str}{str}
\DeclareMathOperator{\occ}{occ}
\DeclareMathOperator{\loci}{\mathcal{L}}
\newcommand{\chr}{\textsc{CHR}\xspace}
\newcommand{\sdsl}{\textsc{SDSL}\xspace}
\newcommand{\sars}{\textsc{SARS}\xspace}

\newcommand{\LFCS}{\textsc{LFCS}\xspace}
\newcommand{\LCCS}{\textsc{LCCS}\xspace}
\newcommand{\TCPR}{\textsc{TCPR}\xspace}
\newcommand{\CC}{\textsc{CC}\xspace}

\newcommand{\LFCSZT}{\textsf{LFCS}$_{\mathrm{ZT}}$\xspace}
\newcommand{\LFCSBA}{\textsf{LFCS}$_{\mathrm{BA}}$\xspace}

\newcommand{\LCCSZT}{\textsf{LCCS}$_{\mathrm{ZT}}$\xspace}
\newcommand{\LCCSBA}{\textsf{LCCS}$_{\mathrm{BA}}$\xspace}

\newcommand{\TCPRZT}{\textsf{TCPR}$_{\mathrm{ZT}}$\xspace}
\newcommand{\TCPRZTminus}{\textsf{TCPR}$_{\mathrm{ZT-}}$\xspace}
\newcommand{\TCPRBA}{\textsf{TCPR}$_{\mathrm{BA}}$\xspace}

\newcommand{\CCZT}{\textsf{CC}$_{\mathrm{ZT}}$\xspace}
\newcommand{\CCBA}{\textsf{CC}$_{\mathrm{BA}}$\xspace}

\newcommand{\ZZ}{\textsf{ZZ}\xspace}

\newcommand{\ZZT}{\textsf{ZZT}\xspace}
\newcommand{\ZZA}{\textsf{ZZA}\xspace}
\newcommand{\ZZLCP}{\textsf{ZZ-LCP}\xspace}

\def\dd{\mathinner{.\,.}}

\newcommand{\cO}{\mathcal{O}}

\newif\ifrevcolor
\revcolortrue

\title{ZigZag Trie: A Novel Index for Contextual Queries}

\author[1]{Ling Li}
\author[2]{Daniel Gibney}
\author[3]{Sharma V. Thankachan}
\author[4]{Rahul Shah}
\author[1]{Grigorios Loukides}
\author[5]{\\Solon P. Pissis}
\affil[1]{King's College London, London, UK}
\affil[2]{University of Texas at Dallas, Dallas, USA}
\affil[3]{North Carolina State University, Raleigh, USA}
\affil[4]{Louisiana State University, Baton Rouge, USA}
\affil[5]{The Cyprus Institute, Nicosia, Cyprus}
\date{\today}

\begin{document}

\maketitle

\begin{abstract}
There is increasing interest in queries about the \emph{context} of a string $P$ in a longer text $T$, i.e., the set of all string pairs $(L,R)$, with $|L|=|R|=q$, for a given $q$, such that the string $LPR$ occurs in $T$. Such contextual queries are important in domains as diverse as  bioinformatics, log analysis, and text analytics but are challenging to answer efficiently. This is because the length of $T$ in applications is massive and existing indexes do not directly encode the context of a given $P$, which is key for answering retrieval queries efficiently. Our work introduces the \textsf{ZigZag Trie} (\ZZT), a new full-text index to specifically address these challenges. This index reorganizes the text so that, for any $P$, all possible strings $L$ and $R$ growing symmetrically around $P$ are grouped into a common subtree of the index, allowing their  efficient retrieval. We show how to construct the \ZZT of $T$, which has size $\cO(n)$ where $n=|T|$, in $\cO(n\log n)$ time and $\cO(n)$ space.

On top of \ZZT, we design specialized indexes that, for a query pattern $P$, answer four new types of contextual queries: (I) finding the longest string $LPR$ that occurs at least $\tau$ times in $T$, for a fixed $\tau$; (II) finding the longest string $LPR$ that occurs in at least $\tau$ texts of a text collection, for a fixed $\tau$; (III) reporting the total number of distinct contexts of $P$ in $T$; and (IV) retrieving, for a given $q$, the $k$ pairs $(L,R)$ of $P$ with the highest scores according to a given scoring function. Our indexes answer queries of type I, II, and III in \emph{optimal} time, and of type IV in \emph{near-optimal} time. Moreover, their size, construction space, and construction time are linear or near-linear in $n$, given \ZZT. Using real billion-letter datasets from different domains, we show that our indexes answer queries orders of magnitude faster than baselines, which are based on traditional text indexes or on the state of the art, and perform similarly or better in terms of index size, construction space, and construction time.
\end{abstract}

\section{Introduction}\label{sec:intro}

Strings are ubiquitous in various  application   domains such as genomics~\cite{NGS}, natural language processing~\cite{DBLP:books/lib/JurafskyM09},~log analysis~\cite{DBLP:journals/pvldb/Boncz0L20}, and  user-behavior modeling~\cite{asselin2016anomaly}. At the heart of these applications, there is \emph{text indexing}~\cite{DBLP:books/daglib/0020103}. This problem asks for preprocessing a string $T$ (\emph{text}) into a compact data structure that supports efficient pattern matching, e.g., \emph{report} the set of starting positions of a \emph{pattern} $P$ in $T$ or \emph{count} this set's size. 

There has been increasing interest (e.g.,~\cite{navarro2020contextual,abedinDCC23,contexticde,navarro2026})  in \emph{contextual queries}, whose focus is \emph{not} only on the pattern $P$ itself but on the strings $L$ and $R$ which occur immediately to the left and to the right of $P$ in $T$, respectively. The \emph{set} of such string pairs $(L, R)$ with $|L|=|R|=q$ is referred to as the \emph{context} of $P$ in $T$~\cite{navarro2020contextual} and denoted by $\mathcal{C}_{T}(P,q)$; $L$ and $R$ are referred to as the \emph{left} and \emph{right} context of $P$ in $T$, respectively. 

Motivated by the importance of contextual queries in database-driven applications, such as bioinformatics and text or log  analysis,  Navarro~\cite{navarro2020contextual,navarro2026} considered the problem of \emph{reporting} $\mathcal{C}_{T}(P,q)$, while Li et al.~\cite{contexticde} considered the problem of \emph{counting} the size of $\mathcal{C}_{T}(P,q)$. 

\subsection{The ZigZag Trie} 

We propose the \textsf{ZigZag Trie} (\ZZT), a novel \emph{full-text} index specifically designed for answering contextual queries. 
The key idea in the design of \ZZT is to define, for each position $i\in [1,n ]$ of $T$, a \emph{ZigZag} string $Z_i = T[i]T[i\!+\!1]T[i\!-\!1]T[i\!+\!2]T[i\!-\!2]\ldots$ that ``walks'' alternately to the right and to the left of $T$. 
These $n$ strings are sorted lexicographically, and the \ZZT of $T$, denoted by $\ZZT(T)$, is  
the compacted trie built over these strings. 

\begin{figure}[t]
\centering
    \begin{subfigure}[b]{0.25\textwidth}\small
    \raggedright
    $Z_1:\texttt{ba}$\\[2pt]
    $Z_2:\texttt{anba}$ \\[2pt]
    $Z_3:\texttt{naanba}$ \\[2pt]
    $Z_4:\texttt{annaa}$ \\[2pt]
    $Z_5:\texttt{naa}$ \\[2pt]
    $Z_6:\texttt{a}$ \\

    \caption{}\label{ex1:a}
    \end{subfigure}\hfill
    \hspace{+4mm}
    \begin{subfigure}[b]{0.25\textwidth}\small
    \raggedright
    $Z_6:\texttt{a}$ \\[2pt]
    $Z_2:\texttt{anba}$ \\[2pt]
    $Z_4:\texttt{annaa}$ \\[2pt]
    $Z_1:\texttt{ba}$\\[2pt]
    $Z_5:\texttt{naa}$ \\[2pt]
    $Z_3:\texttt{naanba}$ \\

    \caption{}\label{ex:b}
    \end{subfigure}\hfill
    \begin{subfigure}[b]{0.4\textwidth}
      \centering
\includegraphics[width=0.5\linewidth, trim=0 0 0 0, clip]{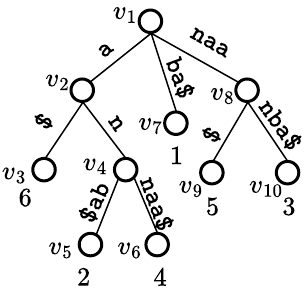}
      \caption{}\label{ex:c}
    \end{subfigure}
\caption{(a) The ZigZag strings $Z_i$ for $T=\texttt{banana}$. (b) The $Z_i$'s sorted lexicographically. (c) $\ZZT(T)$, where the subscript $i$ of each node $v_i$ denotes its \emph{preorder rank} and the special terminating symbol \texttt{\$} marks the end of a ZigZag string.}
\label{exp:intro}
\end{figure}

 \begin{example}\label{ex:zz}
Let $T=\texttt{banana}$. The ZigZag strings constructed for all positions $i\in[1,6]$ are shown in \Cref{ex1:a}. For instance, $Z_3 = T[3]T[4]T[2]T[5]T[1]T[6] = \texttt{naanba}$. The list of lexicographically sorted ZigZag strings and $\ZZT(T)$ are shown in \Cref{ex:b,ex:c}, respectively. 
\end{example}

\begin{tcolorbox}[colback=gray!10, colframe=gray!50, title=Why is the ZigZag Trie useful?, left=2pt, right=2pt]
The utility of $\ZZT(T)$ stems from its construction based on ZigZag strings: for any query pattern $P$ of length $m$, $\ZZT(T)$ \emph{directly encodes} the contextual information of $P$ in the subtree induced by $P$ in $\ZZT(T)$. Namely, for any $q$, all $(L,R)$ string pairs in the context $\mathcal{C}_T(P,q)$ are encoded as descendants of the locus  of the \emph{transformed} $P$ (i.e., the point in $\ZZT(T)$, where $\ZZ(P)=P[\lceil m/2 \rceil]P[\lceil m/2\rceil+1]P[\lceil m/2 \rceil-1]\dots$ ends), and all of these descendants are at string depth $m+2q$. This property is \emph{key} for solving diverse contextual queries that impose constraints on this subtree.
\end{tcolorbox} 

\begin{example}[Cont'd from \Cref{ex:zz}]
    For $P=\texttt{a}$, $\mathcal{C}_T(P,1)=\{(\texttt{b},\texttt{n}), (\texttt{n},\texttt{n})\}$, and the locus of $\ZZ(P)=\texttt{a}$ is the node $v_2$ in \Cref{ex:c}. The strings $Z_6$, $Z_2$, and $Z_4$, which start with $\texttt{a}$, lie in the subtree of $v_2$; 
    the next letter $\texttt{n}$ in the subtree corresponds to the two strings  $R$, $\texttt{n}$ and $\texttt{n}$, in $\mathcal{C}_T(P,1)$, and the outgoing edges of $v_4$ start with $\texttt{b}$ and $\texttt{n}$, the two  strings $L$ in $\mathcal{C}_T(P,1)$.   
\end{example}

As we will discuss later in the related work, existing 
indexes do not offer the key property of \ZZT: \emph{they do not encode all contextual pairs of $P$ within a single subtree}. Instead, using classic indexes, one has to store the left and right contexts in two separate trees, which must be traversed in a coordinated way to recover the pairs $(L,R)$ in the context of $P$. As we will see, this coordinated traversal induces a term in the query time that depends on the number of occurrences of $P$ in $T$. The latter number is huge in real-world datasets~\cite{navarro2020contextual,navarro2026,contexticde}, making the use of such an approach impractical.

\subsection{Contextual Queries}

We showcase the usefulness of \ZZT by formalizing four types of contextual queries that arise in various applications and by presenting results for query answering on real datasets.

\paragraph{Longest Frequent Contextual Superstring (\LFCS).} 

\begin{problem}[\LFCS]\label{def:lfcs} 
A text \(T\) of length \(n\) and an integer \(\tau\in[1,n]\) are given for preprocessing.
Given a pattern \(P\) of length \(m\), report a longest string \(LPR\), with \(|L|=|R|\), that occurs at least \(\tau\) times in \(T\), if such a string exists.
\end{problem}

An important application of \LFCS comes from financial computing, as it summarizes the context around a stock trading event without making stronger claims about stock price prediction or   
performance. Here, $T$ is a discretized time-series of the price of a stock over time and $P$ is a salient trading event, such as a price spike. The \LFCS output then corresponds to 
pre- and post-event price behavior associated with the event that is both persistent and recurring, and hence economically   meaningful~\cite{financial1}. Thus, \LFCS outputs can help in designing algorithmic trading strategies or detecting anomalies (e.g., when in future  data the outputs appear with a much different frequency)~\cite{financial1}. 
Another application of 
\LFCS is in log analysis. Here, $T$ is a sequence of system events, $P$ models a key system event (e.g., ``login''), and $L$ and $R$ model events that precede and succeed the key event. \LFCS can help construct a \emph{structured event template}~\cite{he2021survey} by finding the longest sequence of events $LPR$ (i.e., events in a window around $P$) with $|L|=|R|$ which occurs frequently. Such sequences capture typical behavior and are  used to detect deviations or system failures in future data. 

\begin{example}\label{exp:LFCS}
We constructed a data structure for answering \LFCS queries on the RUET OJ log dataset~\cite{weblogdataset} with $\tau=20$ and asked $3$ queries, each modeling a different key system event.    
\Cref{tab:lfcs-case-study} shows these queries,  their answers, and the typical behavior they indeed capture. 
\end{example}
\begin{table}[ht]
\caption{\LFCS on a log dataset.}
\label{tab:lfcs-case-study}
\centering
\begin{tabular}{lll}
\toprule
\multicolumn{1}{c}{\textbf{$P$}} & \multicolumn{1}{c}{\textbf{Answer}} & \multicolumn{1}{c}{\textbf{Typical behavior}} \\
\midrule
\textcolor{red}{\texttt{\small auth}} &
\texttt{\small login->\textcolor{red}{\small auth}->home} &
{\small Successful authentication flow}\\
\textcolor{red}{\texttt{\small signup}} &
\texttt{\small login->\textcolor{red}{\small signup}->action} &
{\small Successful signup flow}\\
\textcolor{red}{\texttt{\small compile}} &
\texttt{\small submit\_page->\textcolor{red}{\small compile}->submission} &
{\small Page submission pipeline}\\
\bottomrule
\end{tabular}
\end{table}

\paragraph{Longest Common Contextual Superstring (\LCCS).}

\begin{problem}[\LCCS] 
A collection \(\mathcal{C}=\{T_1,T_2,\ldots,T_c\}\) of \(c>1\) texts of total length \(N\) and an integer \(\tau\in[1,c]\) are given for preprocessing.
Given a pattern \(P\) of length \(m\), report a longest string \(LPR\), with \(|L|=|R|\), that occurs in at least \(\tau\) texts of \(\mathcal{C}\), if such a string exists.
\end{problem}

\LCCS is also motivated by applications in financial computing and log analysis. In financial computing, the setting is similar to that of \LFCS but now $T_i$'s model prices of different stocks and $LPR$ models common price behavior (e.g., $LPR$ can be a \emph{following motif} occurring when the price of one stock in a window imitates the price behavior of another~\cite{DBLP:journals/tkdd/ChinpattanakarnA25}).    
In log analysis, the setting is again similar to that of \LFCS but now $T_i$'s are logs of  different servers and $LPR$ must occur in sufficiently many server logs to be deemed typical~\cite{DBLP:conf/kdd/ZhangJJLYW24}. 

\begin{example}
We constructed a data structure for answering \LCCS queries on the same dataset as in \Cref{exp:LFCS} but now using $\tau=5$ and treating the sequence for each of the $c=5$ IP addresses as a separate text $T_i$. \Cref{tab:lccs-case-study}
shows that the answers of the $3$ queries for key system events we asked indeed capture typical behavior (performed by all $5$ IP addresses). 
\end{example}

\begin{table}[ht]
\centering
\small 
\caption{\LCCS on a log dataset.}
\label{tab:lccs-case-study}
\begin{tabular}{lll}
\toprule
\multicolumn{1}{c}{\textbf{$P$}} & \multicolumn{1}{c}{\textbf{Answer}} & \multicolumn{1}{c}{\textbf{Typical behavior}} \\
\midrule
\textcolor{red}{\texttt{\small auth}} &
\texttt{\small home->login->\textcolor{red}{\small auth}->} &
{\small Authentication moves from home,} \\
& \texttt{\small home->archive} &
 {\small to login, then to home, and next to} \\
& & {\small browsing the archive content} \\[3pt]
\textcolor{red}{\texttt{\small signup}} &
\texttt{\small home->login->\textcolor{red}{\small signup}->} &
{\small Signup after visiting home and login,}\\
& \texttt{\small action->login} &
{\small followed by an initial account}\\
& & {\small action and re-login}
\\[3pt]
\textcolor{red}{\texttt{\small submit\_page}} &
\texttt{\small problem->problem\_detail->} &
{\small Users view a problem then its} \\
& \texttt{\textcolor{red}{\small submit\_page}->compile->} &
{\small details, and then go to the submit} \\
& \texttt{\small submission} & {\small page to compile and submit it}\\[2pt]
\bottomrule
\end{tabular}%
\end{table}

\paragraph{Contextual Complexity (\CC).}

\begin{problem}[\CC] \label{pro:count}
A text \(T\) of length \(n\) is given for preprocessing.
Given a pattern \(P\) of length \(m\), report the \emph{contextual complexity} \(\CC(P,T)=\sum_{q\ge 1} |\mathcal{C}_T(P,q)|\) of $P$ in $T$.
\end{problem}

The \CC problem measures the number of distinct context pairs of $P$ of any length $q=|L|=|R|$ (i.e., how ``rich'' the context set of $P$ is). It is motivated by applications in text analytics and biology. In text analytics, it can help compare the \emph{polysemy numbers}~\cite{polysemy} of words (i.e., the number of different related meanings)~\cite{contexticde,kostic2023mapping}. Let $\occ_T(P)$ denote the set of occurrences of string $P$ in string $T$. Consider, for example, two words (patterns $P_1$, $P_2$) that have similar frequencies $|\occ_T(P_1)|\approx |\occ_T(P_2)|$ in a text $T$ but for which $\CC(P_1,T)>\CC(P_2,T)$. Then, $P_1$ is in between more different string pairs $(L,R)$ which may give more different meanings to it. As a more frequent $P$ has intuitively a greater chance to have a higher $\CC(P,T)$, one can instead use the ratio $\CC(P,T)/|\occ_T(P)|$. 
As another application, in biology, a DNA motif $P$ with large $\CC(P,T)$ occurs in between many distinct \emph{flanking-sequence}~\cite{ccmotiv1} environments. This indicates that the motif is reused across multiple cis-regulatory settings and its functional effect may be modulated by local sequence context, including nearby transcription-factor binding sites~\cite{ccmotiv1}.

\begin{example}\label{exp:CC_case_study} 
We constructed a data structure for answering \CC queries on the 
text from the book ``On the Origin of Species''~\cite{pg19}. In \Cref{tab:cc-case-study},  \texttt{place} and \texttt{origin} have similar frequencies but different  
contextual complexities. Indeed, \texttt{place} has also a larger polysemy number, as measured by
the number of WordNet synsets~\cite{miller1995wordnet}. The word \texttt{case} has the largest frequency and also the largest contextual complexity, but its ratio is in between those of the other two words. Indeed, its polysemy number is also in between theirs. \end{example}

\begin{table}[ht]
\centering
\caption{Polysemy numbers comparison based on \CC.}
\label{tab:cc-case-study}
\begin{tabular}{lrrrr}
\toprule
\textbf{$P$} & \textbf{$|\occ_T(P)|$} & \textbf{$\CC(P,T)$} & \textbf{$\CC(P,T)/|\occ_T(P)|$} & \textbf{Polysemy number} \\
\midrule
\texttt{place} & 57 & 2,691,932 & 47,226.9 & 32 \\
\texttt{origin} & 52 & 1,117,453 & 21,489.5 & 6 \\
\texttt{case} & 282 & 12,571,831 & 44,581.0 & 22 \\
\bottomrule
\end{tabular}
\end{table}

\paragraph{Top-$k$ Contextual Pattern Retrieval (\TCPR).} Notably, we lift the
Contextual Pattern Matching problem, introduced by Navarro~\cite{navarro2020contextual},
to its more general, top-$k$ counterpart.
To this end, we employ a \emph{scoring function} $f_{T}:\Sigma^*  \rightarrow \mathbb{R}_{\geq 0}$ that gets as input a substring of $T$ and outputs
a nonnegative real number. 
We will assume oracle access to
$f_{T}$ (i.e., we can read its values in $\cO(1)$ time) after an $\cO(n)$-time preprocessing of $T$.

\begin{problem}[\TCPR] \label{top-k}
A text \(T\) of length \(n\) and a scoring function \(f_T\) are given for preprocessing.
Given a pattern \(P\) of length \(m\) and integers \(q\ge 0\) and \(k>0\), report \(k\) string pairs $(L,R)$ from \(\mathcal{C}_T(P,q)\) with the highest \(f_T(LPR)\) scores, if such pairs exist. 
\end{problem}

Unlike \LFCS and \LCCS which have a single $(L,R)$ pair as part of their output, \TCPR outputs the $k$ pairs $(L,R)$ with the highest $f_{T}$ scores. The function $f_{T}$ captures the importance of pairs from $\mathcal{C}_{T}(P,q)$ using well-known notions such as their  \emph{frequency} or \emph{span}~\cite{tao2007exploration,hawking1995proximity} in $T$. A prime application of \TCPR comes from bioinformatics. Here, $P$ corresponds to a specific pattern of interest (e.g., a CpG site in a certain type of DNA region called CpG island~\cite{nucleicacidres}, which helps to study the evolutionary history of mammalian  genomes~\cite{nucleicacidres}) and the pairs $(L,R)$  correspond to \emph{flanking sequences}, which are highly relevant for understanding the role of genes~\cite{matlock2021flanker} and DNA shape fluctuations~\cite{li2024predicting}. A \TCPR query returns the $k$ most highly ranked flanking sequences, based on their frequency which is used in the study of~\cite{nucleicacidres}. Another application comes from text analysis, where $P$ corresponds to a word or phrase of interest and the top-$k$ pairs $(L,R)$ to the words or phrases that provide different, highly-important meanings to $P$~\cite{kostic2023mapping,contexticde}. For example, when $P=\texttt{bank}$ the top-$2$ pairs $(L,R)$ when $q=1$ could be $(\texttt{south}, \texttt{Thames})$ and $(\texttt{Savings}, \texttt{Ukraine})$. In the first case, $P$ corresponds to a river bank, while in the second to a financial institution. 

\begin{example}
We constructed a \TCPR data structure on the 
same text $T$ as in \Cref{exp:CC_case_study} 
with $f_T$ outputting the frequency of strings $LPR$ in $T$. For $P=\texttt{nature}$, $q=4$, and $k=20$,  \Cref{tab:tcpr-case-study} shows two representative contexts that provide entirely different meanings to the word $\texttt{nature}$.
\end{example}
\begin{table}[ht]
\centering
\caption{\TCPR on a book dataset.}
\label{tab:tcpr-case-study}
\begin{tabular}{@{}lp{0.53\textwidth}p{0.25\textwidth}@{}}
\toprule
\textbf{$P$} & \textbf{Answer} & \textbf{Meaning} \\
\midrule
\texttt{\color{red}{nature}} &
\texttt{animals and plants throughout {\color{red}{nature}} struggle for life most} &
the natural world \\
\texttt{\color{red}{nature}} &
\texttt{it is in human {\color{red}{nature}} to value any novelty} &
the inherent character of a person \\
\bottomrule
\end{tabular}
\end{table}

\subsection{Contributions}

\begin{enumerate}

\item Our central contribution is the \textsf{ZigZag Trie}. For a string $T$ of length $n$, $\ZZT(T)$ occupies $\cO(n)$ space
and locates all $|\occ_T(P)|$ occurrences of a pattern $P$ of length $m$ in $T$ in the optimal $\cO(m+|\occ_T(P)|)$ time.
We present an algorithm to construct $\ZZT(T)$
in $\cO(n\log n)$ time and $\cO(n)$ space, first for the case 
$|L|=|R|$ considered in~\cite{navarro2020contextual,navarro2026}, and then for the general case. Crucially, we show how the \ZZT can be used to construct specialized indexes that answer \LFCS, \LCCS, and \CC queries in \emph{optimal} time, and \TCPR queries in \emph{near-optimal} time.
 
\item We propose an algorithm that, 
given $\ZZT(T)$, constructs an $\cO(n)$-size
index that answers \LFCS queries in the optimal $\cO(m)$ time. This index can be constructed in $\cO(n)$ time and space given $\ZZT(T)$. To arrive at this result, we rely on the structural properties of $\ZZT(T)$.

\item We propose an algorithm that, given $\ZZT(T)$, constructs an $\cO(N)$-size index that answers \LCCS queries in the optimal $\cO(m)$ time. This index can be constructed in $\cO(N)$ time and space, given $\ZZT(T)$. To arrive at this result, we further utilize an algorithm for the \emph{color set size} problem~\cite{DBLP:conf/cpm/Hui92}.

\item We propose an algorithm that, given $\ZZT(T)$, constructs an $\cO(n)$-size index that answers \CC queries in the optimal $\cO(m)$ time. This index can be constructed in $\cO(n)$ time and space given $\ZZT(T)$.  The main idea in the construction is to count, from every possible locus of a query pattern of length $m$, and every $q$, the number of distinct contextual extensions at string depth $m+2q$, performed using dynamic programming. 

\item We propose an algorithm that, given $\ZZT(T)$, constructs an 
$\widetilde{\cO}(n)$-size index that answers \TCPR queries in near-optimal $\cO(m)+\widetilde{\cO}(k)$ time. This index can be constructed in 
 $\widetilde{\cO}(n)$ time and space, using a data structure for top-$k$ orthogonal range reporting~\cite{rahul2011efficient}. We also propose a parameterized version of our index that outperforms the non-parameterized one in all efficiency measures~\cite{contexticde}: query time, index size, construction space, and construction time. Both versions support a wide variety of $f_T$ functions from the literature~\cite{tao2007exploration,hawking1995proximity, hon2014space}.
 
\item We compare our specialized indexes to baseline indexes, as there are no \emph{out-of-the-box} solutions for our problems. The baselines for \LFCS, \LCCS, and \CC are based on the suffix tree~\cite{DBLP:conf/focs/Weiner73}, the textbook index for answering pattern matching queries in optimal time. The baseline for \TCPR uses the state-of-the-art index for reporting $\mathcal{C}_T(P,q)$ by Navarro~\cite{navarro2020contextual}. 
Unlike our indexes and as discussed above, these baselines have query times that depend on the number $|\occ_T(P)|$ of occurrences of $P$ in $T$ or on the size of $\mathcal{C}_T(P,q)$, which are on the order of thousands in practice~\cite{navarro2020contextual,contexticde}. Indeed, using real billion-letter datasets from different domains, we show that our indexes outperform the baselines by \emph{orders of magnitude} in terms of query time and perform similarly or better in terms of index size, construction space, and construction time. Thus, they offer a \emph{desirable trade-off}: answering queries much more efficiently with a similar size and one-off cost in terms of time and space to be constructed. 

\end{enumerate}

\subsection{Organization}

\Cref{sec:preliminaries} discusses some basic concepts  and \Cref{sec:ZigZagindex} our \ZZT index. \Cref{sec:LFCSmain,sec:LCCSmain,sec:CCmain,sec:tcprmain} discuss the baseline and our indexes for \LFCS, \LCCS, \CC, and \TCPR.   \Cref{sec:related}  discusses related work,  \Cref{sec:experiments} presents our experimental evaluation, and \Cref{sec:conclusion} concludes the paper. 

\section{Preliminaries}\label{sec:preliminaries}

\paragraph{Strings.} An \emph{alphabet} $\Sigma$ is a finite set of elements called \emph{letters}.
We consider throughout that $\Sigma$ is an integer alphabet.  
For a string $T=T[1\dd n] \in \Sigma^n$, we denote its length $n$ by $|T|$ and its $i$th letter by $T[i]$. A \emph{substring} of $T$ starting at position $i$ and ending at position $j$ of $T$ is denoted by $T[i\dd j]$. A substring $S$ of $T$ may have multiple occurrences in $T$.
We thus characterize an \emph{occurrence} of $S$ in $T$ by its \emph{starting position} $i\in[1,n]$; i.e.,
$S=T[i\dd i+|S|-1]$.
A substring of the form $T[i\dd n]$ is a \emph{suffix} of $T$. 
A substring of the form $T[1\dd i]$ is a \emph{prefix} of $T$.
The \emph{reverse} of $T$ is denoted by $T^R$. For example, for string   $T=\texttt{CTAAG}$, the reverse is $T^R=\texttt{GAATC}$. 
The \emph{concatenation} of strings $X$, $Y$ is denoted by $X\cdot Y$ (or by $XY$).

\paragraph{Compacted Tries, Suffix Trees, Suffix and LCP Arrays.} Let $\mathcal{X}$ be a set of strings. The \emph{compacted trie} of $\mathcal{X}$ is the trie of $\mathcal{X}$ in which each maximal branchless path is replaced by a single edge whose label is the concatenation of its edge labels. We append a terminating symbol $\texttt{\$}$ to every string, such that $\texttt{\$}$ occurs only at the last position of each string and is the lexicographically smallest. For a node $v$ in a compacted trie $\mathcal{T}$, $\str(v)$ is the concatenation of edge labels on the root-to-$v$ path. We define the \emph{string depth} of a node $v$ as $\sd(v)=|\str(v)|$. For example, consider $v_{10}$ in \Cref{ex:c}. Its edge with label \texttt{nba\$} replaces a path of edges \texttt{n}, \texttt{b}, \texttt{a}, \texttt{\$}; moreover, $\str(v_{10})=\texttt{naanba\$}$ and $\sd(v_{10})=7$. The \emph{locus} of a string $P$ in $\mathcal{T}$ is the pair $(v,|P|)$, where $v$ is the node with the smallest $\sd(v)$ such that $P$ is a prefix of $\str(v)$. When $\str(v)=P$, $(v,|P|)$ corresponds to an \emph{explicit} node; otherwise, it corresponds to an \emph{implicit} node. For example, the locus of $P=\texttt{na}$ in \Cref{ex:c} is $(v_8,2)$, since $v_8$ has the smallest $\sd$ such that $P$ is a prefix of $\str(v_8)$. We denote the \emph{parent} of a node $v$ by $\parent(v)$. 

The \emph{suffix tree} of a string $T[1\dd n]$, denoted by $\ST(T)$, is the compacted trie of the set of suffixes of $T$ with its leaves ordered in ascending lexicographical rank of the corresponding suffixes~\cite{DBLP:conf/focs/Weiner73}; see \Cref{stex} for an example. We denote by $l_i$ the leaf corresponding to the suffix $T[i\dd n]$ and associate it with id $i$. The suffix tree can be constructed in $\cO(n)$ time for integer alphabets of polynomial size in $n$~\cite{DBLP:conf/focs/Farach97}. The suffix tree of a collection $\mathcal{C}$ of strings of total length $N$ can be similarly constructed in $\cO(N)$ time. We denote it by 
$\ST(\mathcal{C})$.

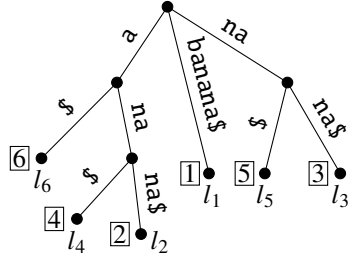
\begin{figure}[t]
\centering
\begin{tikzpicture}[-, level distance=1.5cm]
\filldraw[black] (0,0) circle (2pt); 
\filldraw[black] (1.59,-1) circle (2pt);
\filldraw[black] (-1.67,-2) circle (2pt);
\filldraw[black] (-0.67,-1) circle (2pt);
\filldraw[black] (-0.47,-2) circle (2pt);
\filldraw[black] (-1.2,-2.8) circle (2pt);
\filldraw[black] (0.55,-2.2) circle (2pt);
\filldraw[black] (-0.35,-3) circle (2pt);
\filldraw[black] (1.29,-2.2) circle (2pt);
\filldraw[black] (2.29,-2.2) circle (2pt);

\draw (0,0) -- (-0.67,-1) node[midway,sloped,above] {$\texttt{a}$};
\draw (0,0) -- (0.55,-2.2) node[midway,sloped,above] {$\texttt{banana\$}$};
\draw (0,0) -- (1.59,-1) node[midway,sloped,above] {$\texttt{na}$};

\draw (-1.67,-2) -- (-0.67,-1) node[midway,sloped,above] {$\texttt{\$}$};
\draw (-0.47,-2) -- (-0.67,-1) node[midway,sloped,above] {$\texttt{na}$};
\draw (-0.47,-2) -- (-1.2,-2.8) node[midway,sloped,above] {$\texttt{\$}$};
\draw (-0.47,-2) -- (-0.35,-3) node[midway,sloped,above] {$\texttt{na\$}$};

\draw (1.59,-1) -- (1.29,-2.2) node[midway,sloped,above] {$\texttt{\$}$};
\draw (1.59,-1) -- (2.29,-2.2) node[midway,sloped,above] {$\texttt{na\$}$};

\node[] at (-1.65,-2.3) {$l_6$};
\node[draw, rectangle, fill=white, inner sep=1pt] at (-1.96, -2) {6};

\node[] at (-1.2,-3.1) {$l_4$};
\node[draw, rectangle, fill=white, inner sep=1pt] at (-1.5, -2.8) {4};

\node[] at (-0.1,-3.1) {$l_2$};
\node[draw, rectangle, fill=white, inner sep=1pt] at (-0.63, -3) {2};

\node[] at (0.6,-2.5) {$l_1$};
\node[draw, rectangle, fill=white, inner sep=1pt] at (0.3, -2.2) {1};

\node[] at (1.3,-2.5) {$l_5$};
\node[draw, rectangle, fill=white, inner sep=1pt] at (1.02, -2.2) {5};

\node[] at (2.3,-2.5) {$l_3$};
\node[draw, rectangle, fill=white, inner sep=1pt] at (2, -2.2) {3};

\end{tikzpicture} 

\caption{$\ST(T)$ for $T=\texttt{banana}$; each leaf id is in a square.}\label{stex} 
\end{figure}
The \emph{suffix array} $\SA[1\dd n]$ of a string $T[1\dd n]$ is a permutation of $[1,n]$ such that $T[\SA[i]\dd n]$ is the $i$th smallest suffix when ordered  lexicographically~\cite{DBLP:journals/siamcomp/ManberM93}. The \emph{inverse suffix array} $\ISA[1\dd n]$
is a permutation such that $\ISA[\SA[i]]=i$. For example, for $T=\texttt{banana}$, $\SA[1\dd 6]=[6,4,2,1,5,3]$ and $\ISA[1\dd 6]=[4,3,6,2,5,1]$.

We use $\LCE_T(i,j)$ to denote the length of the \emph{longest common prefix} (LCP) of $T[i\dd n]$ and $T[j\dd n]$. For example, $\LCE_T(3,5)=2$ in \Cref{stex} because the LCP of $T[3\dd 6]=\texttt{nana}$ and $T[5\dd 6]=\texttt{na}$ is $\texttt{na}$, which has length $2$. 
Note that $\LCE_T(i,j)$ 
can be computed in $\cO(1)$ time after an $\cO(n)$-time preprocessing of the suffix tree~\cite{DBLP:journals/siamcomp/HarelT84}. We analogously use $\overleftarrow{\LCE_T}(i,j)$
to denote the length of the \emph{longest common suffix} (LCS) 
of $T[1\dd i]$ and $T[1\dd j]$. Similarly,
this quantity can be computed in $\cO(1)$ time after an $\cO(n)$-time preprocessing.

The $\LCP[1\dd n]$ array of a string $T[1\dd n]$ stores the length of the LCP of lexicographically adjacent suffixes~\cite{DBLP:journals/siamcomp/ManberM93}. For $j>1$, $\LCP[j]$ stores the length of the LCP between the suffixes starting at $\SA[j-1]$ and $\SA[j]$, and $\LCP[1]=0$. For example, for $T=\texttt{banana}$, $\LCP[1\dd 6]=[0,1,3,0,0,2]$. 
Given $\SA$, the $\LCP$ array of $T$ can be computed in $\cO(n)$ time~\cite{DBLP:conf/cpm/KasaiLAAP01}. 

\paragraph{Top-$k$ Queries for Orthogonal Range Reporting.} Let $S$ be a set of $n$ weighted points in $\mathbb{R}^d$, where each point $p \in S$ has a real-valued weight $w(p)$. Preprocess $S$ into a data structure so that given: (1) an orthogonal query box $Q = \prod_{i=1}^{d} [a_i, b_i]$; and (2) an integer $k \in [1, n]$, the data structure can efficiently report the $k$ points in $S \cap Q$ with the highest weights. Our data structure for \TCPR is  built upon the one from~\cite{rahul2011efficient}, which uses $\cO(n \log^d n)$ space with 
$\cO(\log^d n + k)$ query time for arbitrary order reporting, where $d=\cO(1)$. 

\section{The ZigZag Trie}\label{sec:ZigZagindex}

\paragraph{Definitions.} Let \(T[1\dd n]\) be a string of length \(n\) over alphabet \(\Sigma\). For every position \(i \in [1,n]\), we define the \emph{ZigZag} string \(Z_i(T)\) as the sequence
$Z_i = T[i]\,T[i+1]\,T[i-1]\,T[i+2]\,T[i-2]\cdots$,
where the construction terminates as soon as the next index falls outside of \([1,n]\); we drop \((T)\) from \(Z_i(T)\) when the context is clear. More formally, for all $i\in[1,n]$, we have
{
\abovedisplayskip=2mm
\belowdisplayskip=1mm
\[
|Z_i| = \min\left\{k \ge 1 :  i + (-1)^k\left\lfloor \frac{k}{2} \right\rfloor \notin [1, n]\right\} - 1,
\]}
and then, for all \(k\in[1,|Z_i|]\), we have
{
\abovedisplayskip=2mm
\belowdisplayskip=1mm
\[
Z_i[k] = T\!\left[i + (-1)^k\left\lfloor \frac{k}{2} \right\rfloor\right].
\]}

The above equations tell us that given \emph{any} $i,k$, we can determine whether $Z_i[k]$ is defined and, if so, access it in $\cO(1)$ time. We will thus assume oracle access to $\{Z_1, Z_2, \ldots, Z_n\}$.
From the above, it is also easy to see that the letters of the prefix $Z_i[1\dd \ell]$ of $Z_i$ are exactly those of the length-$\ell$ substring of $T$ starting at position $(i-\lceil \ell/2\rceil + 1)$.

The \textsf{ZigZag Trie} (\ZZT) of $T$, denoted by $\ZZT(T)$, is the compacted trie over all $n$ ZigZag strings $\{Z_1, Z_2, \ldots, Z_n\}$; we append a terminating symbol \(\texttt{\$}\notin\Sigma\), that is lexicographically smaller than every letter in \(\Sigma\), to each $Z_i$.
Leaf and internal nodes of $\ZZT(T)$ are defined analogously to those of $\ST(T)$. 

We further define: (I) the \textsf{ZigZag Array} (\ZZA) of $T$, denoted by $\ZZA(T)$: $\ZZA(T)[i]=j$ if $Z_j$ is the $i$th ZigZag string in lexicographic order, and (II) the \emph{ZigZag} $\LCP$ array of $T$,  denoted by $\ZZLCP(T)$: $\ZZLCP(T)[1]=0$, and $\ZZLCP(T)[j]$, for $j>1$, stores the length of the longest common prefix of $Z_{\ZZA(T)[j-1]}$ and $Z_{\ZZA(T)[j]}$. We use $\ZZA(T)$ and $\ZZLCP(T)$ to  efficiently construct $\ZZT(T)$. 

\begin{example}
For \(T=\texttt{banana}\), we have the following.
\[
\begin{array}{r|l|l|l|l}
i & Z_i & \ZZA(T)[i] & Z_{\ZZA(T)[i]} & \ZZLCP(T)[i] \\ \hline
1 & \texttt{ba}     & 6 & \texttt{a}      & 0 \\
2 & \texttt{anba}   & 2 & \texttt{anba}   & 1 \\
3 & \texttt{naanba} & 4 & \texttt{annaa}  & 2 \\
4 & \texttt{annaa}  & 1 & \texttt{ba}     & 0 \\
5 & \texttt{naa}    & 5 & \texttt{naa}    & 0 \\
6 & \texttt{a}      & 3 & \texttt{naanba} & 3
\end{array}
\]
\end{example}

\paragraph{Construction.} To construct $\ZZT(T)$, we need to sort the ZigZag strings $\{Z_1,Z_2, \ldots, Z_n\}$  lexicographically. However, explicitly constructing and then sorting them is prohibitively expensive, as it takes $\cO(\sum_{i\in[1,n]}|Z_i|)=\cO(n^2)$ time. In response, we show how they can be sorted in $\cO(n\log n)$ time \emph{without being explicitly constructed}. This leads to an $\cO(n \log n)$-time algorithm for constructing \(\ZZT(T)\). 

First, construct $\ST(T)$ and $\ST(T^R)$ to support $\cO(1)$-time LCE queries in both directions. 
To compare any two ZigZag strings $Z_i$ and $Z_j$ lexicographically, we compute:
\begin{itemize}
\item $r \coloneqq \overleftarrow{\LCE_T}(i,j)$, i.e., the length of the  longest common suffix of prefixes $T[1\dd i]$ and $T[1\dd j]$;
\item $f \coloneqq \LCE_T(i+1, j+1)$, i.e., the length of the longest common prefix of suffixes $T[i+1\dd n]$ and $T[j+1\dd n]$.
\end{itemize}

If $r \leq f$, then the two strings agree through the first $2r$ positions, and the first mismatch occurs at position $2r+1$; otherwise, they agree through position $2f+1$, and the first mismatch occurs at position $2f+2$. 
Hence, they can be compared in $\cO(1)$ time. This enables sorting the strings in $\{Z_1, Z_2, \ldots, Z_n\}$ without explicitly constructing them, in $\cO(n \log n)$ time, using any optimal comparison-based sorting algorithm, thereby producing $\ZZA(T)$. 
Then, for every pair of successive entries of $\ZZA(T)$, we use two LCE queries to infer the LCP value yielding $\ZZLCP(T)$.
Finally, we use the algorithm by Kasai et al.~\cite{DBLP:conf/cpm/KasaiLAAP01} to construct $\ZZT(T)$ from $\ZZA(T)$ and $\ZZLCP(T)$ in $\cO(n)$ time, analogously to how a suffix tree can be built from a suffix array and LCP array.

\begin{example}
Consider the comparison of \(Z_2=\texttt{anba}\) and \(Z_4=\texttt{annaa}\) for \(T=\texttt{banana}\).
We compute \(r\coloneqq\overleftarrow{\LCE_T}(2,4)=1\) and \(f\coloneqq\LCE_T(3,5)=2\).
Since \(r\le f\), the first mismatch occurs  at position $2r+1=3$: \(Z_2[3]=\texttt{b}\neq Z_4[3]=\texttt{n}\).
As $\texttt{b}<\texttt{n}$, we infer \(Z_2<Z_4\) 
in \(\cO(1)\) time. 
\end{example}

Each branching node of $\ZZT(T)$ is
augmented with a dictionary that supports $\cO(1)$-time access to the edges based on the first letter (key) of their label.
The total size of the dictionaries is $\cO(n)$, and they can be constructed in $\cO(n \log n)$ time~\cite{DBLP:conf/icalp/Ruzic08}.
We have arrived at the following result.

\begin{theorem}\label{the:con}
For every text $T$ of length $n$,
$\ZZT(T)$ can be constructed in $\cO(n \log n)$ time using $\cO(n)$ space. The data structure occupies $\cO(n)$ space.
\end{theorem}

\paragraph{Querying.} To search for a pattern \(P\) of length \(m\) in \(\ZZT(T)\), we first transform \(P\) into its ZigZag counterpart. Let \(x\coloneqq\lceil m/2\rceil\). The \emph{ZigZag transformation} \(\ZZ(P)\) is defined as:
{
\abovedisplayskip=2mm
\belowdisplayskip=1mm
\[
\ZZ(P)=P[x]P[x+1]P[x-1]P[x+2]P[x-2]\dots
\]}

\noindent with $|\ZZ(P)|=m$. More formally, for all $k\in[1,m]$, 
{
\abovedisplayskip=1mm
\belowdisplayskip=1mm
\[
\ZZ(P)[k]=P\left[x + (-1)^k \cdot \left\lfloor \frac{k}{2} \right\rfloor\right].
\]
}
\begin{example}
   $\ZZ(\texttt{abcd}) = \texttt{bcad}$ and $\ZZ(\texttt{abcde}) = \texttt{cdbea}$. 
\end{example}
This transformation is clearly implementable in $\cO(m)$ time. Then, finding the locus of \(\ZZ(P)\) in $\ZZT(T)$ takes $\cO(m)$ time. Thus, like with $\ST(T)$~\cite{DBLP:conf/focs/Weiner73}, we can locate all $|\occ_T(P)|$ occurrences of $P$ in $T$ in the optimal $\cO(m+|\occ_T(P)|)$ time using $\ZZT(T)$.
Note that it is also straightforward to implement the inverse of \(\ZZ\) in $\cO(m)$ time: for all $i\in[1,m]$,
{
\abovedisplayskip=2mm
\belowdisplayskip=2mm
\[
P[i] = \begin{cases}
\ZZ(P)[2x - 2i + 1] & \text{if \(i \leq x\)}, \\[1em]
\ZZ(P)[2(i - x)] & \text{if \(i > x\)}.
\end{cases}
\]}

Similarly to suffix trees,
we can trivially generalize the $\ZZT$ for a \emph{collection} $\mathcal{C}=\{T_1,\dots,T_c\}$ of $c$ strings of total length $N$, which we denote by $\ZZT(\mathcal{C})$. Yet, there are two differences: (I) The compacted trie is now constructed over the \emph{union} of the ZigZag strings of every string in $\mathcal{C}$. 
(II) We append a unique terminating symbol $\texttt{\$}_i\notin\Sigma$ to the ZigZag strings of $T_i$, for all $i\in[1,c]$. The construction clearly takes $\cO(N \log N)$ time. 

\paragraph{Generalization to Asymmetric Context Lengths.} Following~\cite{navarro2020contextual,navarro2026}, we considered context pairs with $|L|=|R|$,
but our construction can be readily generalized to $|L|=\alpha q$ and $|R|=\beta q$, for any
constant integers $\alpha,\beta\ge 1$. The only change is in the definition of the
ZigZag strings: after the initial letter
$T[i]$, the $k$-th step, for $k=1,2,\ldots$, appends the next $\beta$ unread letters to the right (in increasing order of position) and then the next $\alpha$ unread letters to the left (in decreasing order of position), terminating as soon as an index falls outside $[1,n]$; for $\alpha=\beta=1$ this
is exactly the definition above, and for
$\alpha=2$, $\beta=1$ we obtain $Z_i=T[i]\,T[i+1]T[i-1]T[i-2]\,T[i+2]T[i-3]T[i-4]\cdots$. Since each letter is added adjacent to the window of positions read so
far, the letters of every prefix $Z_i[1\dd \ell]$ form a contiguous
window of $T$ whose position relative to $i$ depends only on $\ell$.
Hence every pattern $P$ has a unique transformed counterpart, and the
pairs $(L,R)$ with $|R|=\beta q$ and $|L|=\alpha q$ around $P$ are again exactly
the descendants of the locus of the transformed $P$ at string depth
$|P|+(\alpha+\beta)q$. Letters of $Z_i$ are still accessed in $\cO(1)$ time, and
two ZigZag strings are still compared in $\cO(1)$ time using one forward
and one backward LCE query. Hence the ZigZag strings can be
sorted within the same time complexity, and the bounds of
\Cref{the:con} are identical.

\section{Indexes for \LFCS}\label{sec:LFCSmain}

\subsection{Baseline Index for \LFCS} \label{sec:LFCS-BA}

\paragraph{Construction.} Construct the suffix trees \(\ST(T)\) and \(\ST(T^R)\) for the text $T$ and its reverse \(T^R\). Using a DFS on \(\ST(T)\), for each node \(v\), add a pointer to a descendant \(u\) of \(v\) with maximum string depth such that the number of leaves $L(u)$ in the subtree rooted at $u$ is at least $\tau$; such a pointer may not exist, and \(u\) may be \(v\) itself. This pointer will be used to get the \emph{longest} frequent (i.e., occurring at least $\tau$ times) extension $LPR$ of $P$, for a query $P$. We then perform the same preprocessing over \(\ST(T^R)\). The key property is \Cref{lem:correct}.

\begin{lemma}\label{lem:correct}
Let $LPR$ be an optimal solution to \LFCS with $\tau\ge 2$. Consider the locus of $PR$ and $(LP)^R$ in $\ST(T)$ and $\ST(T^R)$, respectively. At least one of these loci is explicit.
\end{lemma} 

 \begin{proof}
Assume towards contradiction that both loci are implicit. Since the locus of $PR$ in $\ST(T)$ is implicit, it lies inside an edge label and there exists a \emph{unique} letter $a$ such that $PRa$ is still represented on that same edge; therefore $PRa$ occurs in $T$. Symmetrically, since the locus of $(LP)^R$ in $\ST(T^R)$ is implicit, there exists a \emph{unique} letter $b$ such that $bLP$ occurs in $T$. Thus,  every occurrence of $LPR$ in $T$ is preceded by $b$ and followed by $a$ and so $|\occ_T(bLPRa)| = |\occ_T(LPR)| \ge \tau$. Moreover, since $|L|=|R|$, $|L|+1=|R|+1$. Thus,  $bLPRa$ is a feasible solution to \LFCS. Since it is longer than $LPR$, this contradicts the optimality of $LPR$. 
\end{proof}

\paragraph{Querying.} We first locate the locus \((v,|P|)\) of the query pattern \(P\) in \(\ST(T)\) and perform a DFS on the subtree rooted at \(v\), collecting every node \(u\) with \(L(u)\ge \tau\). Each such node $u$ defines a possible frequent right extension \(PR\) of $P$.

For each \(PR\), we locate the locus \((w,|PR|)\) of \((PR)^R\) in \(\ST(T^R)\), and use the pointer of $w$ to obtain the longest frequent left extension of $PR$. We truncate \(L\) if \(|L|>|R|\) so that both have the same length \(|R|\), as required.
(If \(|L|<|R|\), then no candidate is constructed for node $u$.)
For each collected node $u$, we form a candidate $LPR$. 
Every such candidate satisfies the frequency threshold and the  condition $|L|=|R|$ required by \Cref{def:lfcs}. 
We output a \emph{longest} candidate to also meet the optimization goal. 

The query algorithm described above is performed again starting from the locus of \(P^R\) in \(\ST(T^R)\) and searching in \(\ST(T)\). The correctness of this algorithm follows by \Cref{lem:correct}. (Note that for $\tau=1$, the problem is trivial.)

\paragraph{Complexity.} The index has $\cO(n)$ size as \(\ST(T)\) and \(\ST(T^R)\) have $\cO(n)$ size. The construction time and space are $\cO(n)$, because constructing and  preprocessing \(\ST(T)\) and \(\ST(T^R)\) via DFS all take \(\cO(n)\) time and use \(\cO(n)\) space. At query time, locating \((v,|P|)\) takes \(\cO(m)\) time, and collecting all candidates \(PR\) takes \(\cO(|\occ_T(P)|)\) time. For each \(PR\), we recover the corresponding left context by traversing upward from a leaf in the subtree of \(w\), which takes \(\cO(|\occ_T(P)|)\) time. Therefore, the overall query time is \(\cO\!\left(m + |\occ_T(P)|^2\right)\).

\subsection{Our Index for \LFCS}\label{sec:lfcs}

Our index avoids the expensive ``coordinated'' traversal of the suffix trees of the baseline, by exploiting the structure of \ZZT. It adds pointers to $\ZZT(T)$ that \emph{directly} lead to nodes from which we get query  answers. This yields the optimal $\cO(m)$ query time, which does \emph{not} depend on $|\occ_T(P)|$. 

\paragraph{Construction.} Construct \(\ZZT(T)\) using \Cref{the:con}. Then perform a DFS and compute \(L(v)\), for every node \(v\) in \(\ZZT(T)\). In a second DFS, for each  node \(v\), add a pointer to a descendant \(u\) of \(v\) with maximum string depth such that \(L(u)\ge \tau\); such a pointer may not exist, and \(u\) may be \(v\) itself. This takes \(\cO(n)\) time. This pointer will be used to get the encoded longest  possible frequent $LPR$.   

\paragraph{Querying.} Given a pattern \(P\) of length \(m\), compute \(\ZZ(P)\) and locate its locus \((v,m)\) in \(\ZZT(T)\) in \(\cO(m)\) time. If \(v\) has no stored pointer, then no such \(LPR\) string exists. Otherwise, follow the pointer from \(v\) to a node \(u\): $u$ encodes the \emph{longest} frequent $LPR$ as \(\ZZ(LPR)=\str(u)\). If $|L|+|R|$ is not even, we truncate the last letter to guarantee that $|L|=|R|$.
Let $i$ be such that $\str(u)$ is a prefix of string $Z_i$.
Then $LPR$ occurs at position $(i-\lceil \ell/2\rceil + 1)$ of $T$, with $\ell=|LPR|$. The total query time is \(\cO(m)\).
This is correct by how \(\ZZT(T)\) is constructed and by how the node pointers are constructed.

\begin{example}
Let $T = \texttt{banana}$, $\tau = 2$, and $P = \texttt{n}$.
We compute $\ZZ{(P)} = \texttt{n}$ and locate its locus $(v_8,1)$ in $\ZZT(T)$ (see \Cref{fig:zigzag_with_lfcs_pointers}).
Following the pointer stored at $v_8$, we arrive at $v_8$ itself, the deepest descendant $u$
with $\leafc{u} \geq 2$, which corresponds to the ZigZag prefix
$\texttt{naa}$ of length $3 = m + 2q$ with $q = 1$.
It gives $L = \texttt{a}$, $R = \texttt{a}$, and $LPR = \texttt{ana}$,
which occurs at positions $2$ and $4$ in $T$ with frequency $2 \geq \tau$.
\end{example}

\begin{figure}[ht]
\centering
\includegraphics[width=0.25\linewidth, trim=10pt 25pt 0pt 0pt, clip]{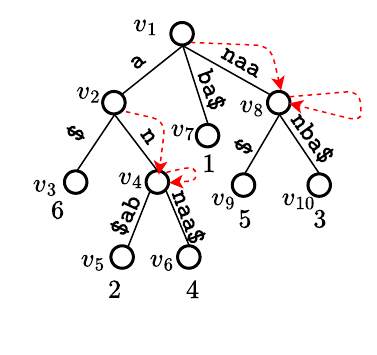}
\caption{$\ZZT(T)$ of $T=\texttt{banana}$ with $\tau=2$, annotated with pointers (dashed red arrows).}
\label{fig:zigzag_with_lfcs_pointers}
\end{figure}

\paragraph{Main Result.} From the above, we directly obtain: 

\begin{theorem}
\LFCS can be solved using an index of $\cO(n)$ size supporting $\cO(m)$-time queries. The data structure can be constructed in $\cO(n)$ time and space given $\ZZT(T)$.
\end{theorem}

\section{Indexes for \LCCS}\label{sec:LCCSmain}

\subsection{Baseline Index for \LCCS} \label{sec:LCCS-BA}

This index is conceptually similar to the baseline index for \LFCS, but now the forward suffix tree indexes the collection $\mathcal{C}$ and the reverse suffix tree indexes the reverse of each string in $\mathcal{C}$. Also, here, instead of leaf counts, we employ a coloring scheme in which \emph{color counts}  measure the number of distinct texts from $\mathcal{C}$ in which a candidate answer occurs.

\paragraph{Construction.} Construct $\ST(\mathcal{C})$ and $\ST(\mathcal{C}^R)$, where $\mathcal{C}^R$ is $\{T_1^R,T_2^R,\dots,T_c^R\}$. Using a DFS on $\ST(\mathcal{C})$, for each leaf $\ell$ originating from $T_j$, assign the color \(j\) to \(\ell\). For every node \(v\) in \(\ST(\mathcal{C})\), compute \(C(v)\), the number of distinct colors in the subtree rooted at \(v\) using the linear-time algorithm in~\cite{DBLP:conf/cpm/Hui92}. Using another DFS, for each node \(v\), add a pointer to a descendant \(u\) of \(v\) with maximum string depth such that \(C(u)\ge \tau\); such a pointer may not exist, and \(u\) may be \(v\) itself.
We then perform the same preprocessing over $\ST(\mathcal{C}^R)$. 

\paragraph{Querying.} The query algorithm is essentially the same as 
the one for the \LFCS baseline; the only difference is that we use $C(v)$ instead of $L(v)$ to construct the \(PR\) candidates.

\paragraph{Complexity.} The bounds are analogous to those of the  \LFCS baseline, but have the total string length $N$ instead of $n$. Thus, the index has $\cO(N)$ size, construction time, and construction  space. Let $\occ_{\mathcal{C}}(P)$ denote the set of occurrences of $P$ in the texts of $\mathcal{C}$, where an occurrence is a pair $(j,i)$: $P$ occurs at position $i$ of $T_j$. The query time is 
\(\cO(m+|\occ_{\mathcal{C}}(P)|^2)\), as querying is the same as for \LFCS: locating \((v,|P|)\) takes \(\cO(m)\) time, and the subsequent tree traversals cost \(\cO(|\occ_{\mathcal{C}}(P)|^2)\) in total. 

\subsection{Our Index for \LCCS}\label{sec:lccs}

Our index adds pointers to \ZZT that \emph{directly} lead to nodes from which we can get \LCCS query answers in the optimal $\cO(m)$ time, i.e., avoiding the term $|\occ_{\mathcal{C}}(P)|^2$ of the baseline.   

\paragraph{Construction.} Construct \(\ZZT(\mathcal{C})\).
Using a DFS on \(\ZZT(\mathcal{C})\), for each leaf \(\ell\) originating from \(T_j\), assign the color \(j\) to \(\ell\). For every node \(v\) in \(\ZZT(\mathcal{C})\), compute \(C(v)\), the number of distinct colors in the subtree rooted at \(v\) using the linear-time algorithm in~\cite{DBLP:conf/cpm/Hui92}. Using another DFS, for each node \(v\), add a pointer to a descendant \(u\) of \(v\) with maximum string depth such that \(C(u)\ge \tau\); such a pointer may not exist, and \(u\) may be \(v\) itself. Given \(\ZZT(\mathcal{C})\), the construction takes \(\cO(N)\) time.

\paragraph{Querying.} Given a pattern \(P\) of length \(m\), compute \(\ZZ(P)\) and locate its locus \((v,m)\) in \(\ZZT(\mathcal{C})\) in \(\cO(m)\) time. If \(v\) has no stored pointer, then no such \(LPR\) string exists. Otherwise, follow the pointer from \(v\) to a node \(u\), and report node $u$ that encodes the longest common $LPR$ as \(\ZZ(LPR)=\str(u)\). 
If $|L|+|R|$ is not even, we truncate the last letter to guarantee that $|L|=|R|$. We can locate a witness, i.e., a position where $LPR$ occurs 
in a text in $\mathcal{C}$, similar to the \LFCS query algorithm.
The total query time is \(\cO(m)\). This is correct by how \(\ZZT(\mathcal{C})\) is constructed and by how the node pointers are constructed.

\paragraph{Main Result.} From the above, we directly obtain: 

\begin{theorem}
\LCCS can be solved using an index of $\cO(N)$ size supporting $\cO(m)$-time queries. The data structure can be constructed in $\cO(N)$ time and space given \(\ZZT(\mathcal{C})\).
\end{theorem}

\section{Indexes for \CC}\label{sec:CCmain}

\subsection{Baseline Index for \CC}\label{sec:CC-BA}

\paragraph{Construction.} Construct \(\ST(T)\) and \(\ST(T^R)\).

\paragraph{Querying.} We locate the locus \((v,|P|)\) of \(P\) in \(\ST(T)\), and then locate, for increasing \(q\ge 1\), all the right extensions \(PR\) of $P$, with \(|R|=q\), in the subtree of \(v\). For each \(PR\), we locate the locus of \((PR)^R\) in \(\ST(T^R)\), and count the left extensions \(LPR\) of $PR$, with \(|L|=q\), in its subtree. Summing the counts over all \(PR\) and \(q\) yields the correct answer, as we  exhaustively consider all valid loci pairs in \(\ST(T)\) and \(\ST(T^R)\).

\paragraph{Complexity.} The construction takes \(\cO(n)\) time and space. The index size is $\cO(n)$. The query time is \(\cO\!\left(m+|\loci(v)|^2\right)\), where \(\loci(v)\) is the set of loci in the subtree rooted at \(v\) in \(\ST(T)\).

\subsection{Our Index for \CC}\label{sec:cc}

Our index exploits $\ZZT(T)$ and auxiliary information computed via dynamic programming to compute $\CC(P,T)$ in the optimal $\cO(m)$ time. This is unlike the baseline index, whose query time has a large $|\loci(v)|^2$ term, where $v$ is the locus of $P$, due to the ``coordinated'' exploration of the two suffix trees. 

\paragraph{Construction.} Construct \(\ZZT(T)\) using \Cref{the:con}. We  compute $\CC(P,T)$ via the number of contextual extensions \(LPR\) of \(P\) occurring in \(T\) over all \(q=|L|=|R|\). Each extension corresponds to a path in the subtree rooted at the locus of \(\ZZ(P)\) in \(\ZZT(T)\). Since \(|L|=|R|=q\), every extension has length \(m+2q\), where \(m=|P|\). Thus, from the locus of \(\ZZ(P)\) at string depth \(m\), we need to count the substrings in its subtree at string depths \(m+2q\), for all integers \(q\ge 0\). We do this via dynamic programming. 

To provide some intuition, let us denote a potential query string by \(P\). Let \((v,|P|)\) denote the locus of \(\ZZ(P)\). 
The distance \(d=\sd(v)-|P|\) may be even or odd. If $d$ is even, no single letter remains before reaching $v$, and the number of 
additional cross-edge length-$2$ fragments contributed by the
outgoing edges of $v$ is zero. For example, for $P = \texttt{ac}$ in \Cref{fig:CC_even}, the locus of $\ZZ(P)=\texttt{ac}$ is $(v_1, 2)$ (i.e., node $w$). As $\sd(v_1) = 4$ and $|P| = 2$,
 $d = \sd(v_1) - |P| = 2$ is even, and no single letter remains until reaching $v_1$. If $d$ is odd, after traversing \(\lfloor d/2 \rfloor\) length-\(2\) fragments along the edge from locus \((v,|P|)\) to \((v,\sd(v))\), one letter \emph{remains} before reaching \(v\). This letter can be combined with the first letter of every outgoing edge of \(v\), contributing \(|\ch(v)|\) additional cross-edge length-\(2\) fragments, where \(\ch(v)\) denotes the set of children of \(v\). For example, for $P = \texttt{bac}$ in \Cref{fig:CC_odd}, 
the locus of $\ZZ(P)=\texttt{acb}$ is $(v_1, 3)$ (i.e., node $w$), and $d = \sd(v_1) - |P| =4-3= 1$,
which is odd. After traversing $\lfloor 1/2 \rfloor = 0$ length-$2$ fragments along the 
edge from $(v_1, 3)$ to $(v_1, 4)$, one letter, namely \texttt{d}, remains before reaching $v_1$. 
This is combined with the first letter of each of the $|\ch(v_1)| = 2$
outgoing edges of $v_1$, contributing $2$ additional cross-edge
length-$2$ fragments. 

\begin{figure}[ht]
    \centering
    \begin{subfigure}[t]{0.44\textwidth}
        \centering
        \includegraphics[width=0.7\linewidth]{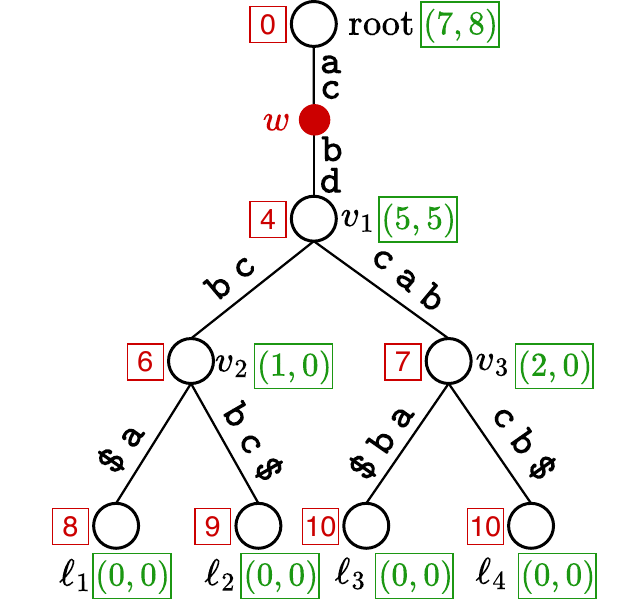}
            \vspace{+0.5mm}
        \caption{$P = \texttt{ac}$, $\ZZ(P) = \texttt{ac}$, with locus $w$ at string depth $2$}
        \label{fig:CC_even}
    \end{subfigure}
    \hspace{2mm}
    \begin{subfigure}[t]{0.44\textwidth}
        \centering
        \includegraphics[width=0.7\linewidth]{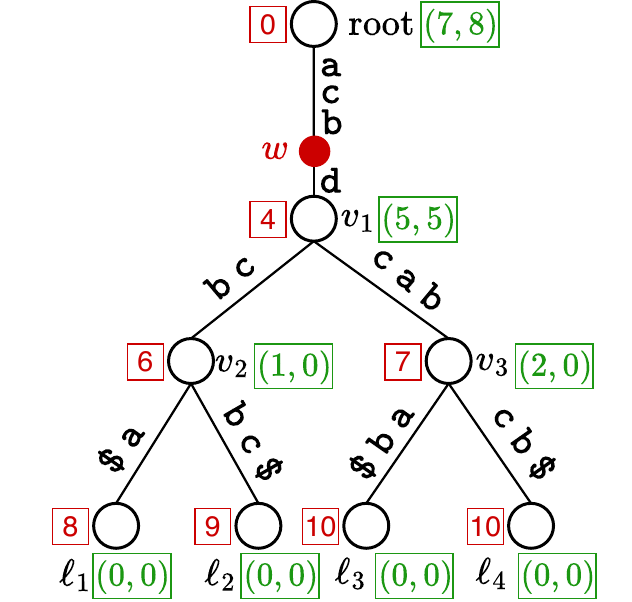}
            \vspace{+0.5mm}
        \caption{$P = \texttt{bac}$, $\ZZ(P) = \texttt{acb}$, with locus $w$ at string depth $3$}
        \label{fig:CC_odd}
    \end{subfigure}
    \vspace{+1mm}
    \caption{For each node $v$, the red square contains $\sd(v)$ and the green rectangle $(\EV(v), \OD(v))$. The locus of $\ZZ(P)$ is $w$.}
    \label{fig:CC_example}
\end{figure}

We now formalize the above intuition.
For an internal node \(v\) in \(\ZZT(T)\), let \(u_1, u_2, \ldots, u_k\) be its children, and let \(d_i = \sd(u_i) - \sd(v)\) be the length of the edge label from \(v\) to \(u_i\). To avoid counting the terminal symbol \texttt{\$}, we adjust this to \(d_i = \sd(u_i) - \sd(v) - 1\) whenever \(u_i\) is a leaf.  Let \(\mathcal{L}(v)\) be the set of loci in the subtree rooted at \(v\), excluding loci that end at a terminal symbol. 
For any locus \(x=(u,\ell) \in \mathcal{L}(v)\), let
$\delta_v(x)=\ell-\sd(v)$. 
We then define
{
\abovedisplayskip=2mm
\belowdisplayskip=2mm
$$
\EV(v)=
\left|
\left\{
x\in \mathcal{L}(v):
\delta_v(x)>0 \text{ and }
\delta_v(x)\equiv 0 \pmod{2}
\right\}
\right|,
$$
}
and
{
\abovedisplayskip=2mm
\belowdisplayskip=2mm
\[
\OD(v)=
\left|
\left\{
x\in \mathcal{L}(v):
\delta_v(x)>1
\text{ and }
\delta_v(x)\equiv 1 \pmod{2}
\right\}
\right|.
\]
}

Thus, \(\EV(v)\) counts the loci at a positive even distance from \(v\), while \(\OD(v)\) counts the loci at an odd distance greater than $1$ from \(v\); the loci at distance one are covered by the \(|\ch(v)|\) term. We define the following recurrences:
{\small
\abovedisplayskip=2mm
\belowdisplayskip=2mm
\begin{align*}
\EV(v) &= \sum_{i=1}^{k} \left( \left\lfloor \frac{d_i}{2} \right\rfloor +
\begin{cases}
\EV(u_i) & \text{if } d_i \text{ is even}, \\
\OD(u_i) + |\ch(u_i)| & \text{if } d_i \text{ is odd},
\end{cases} \right) \\
\OD(v) &= \sum_{i=1}^{k} \left( \left\lfloor \frac{d_i-1}{2} \right\rfloor +
\begin{cases}
\OD(u_i) + |\ch(u_i)| & \text{if } d_i \text{ is even}, \\
\EV(u_i) & \text{if } d_i \text{ is odd},
\end{cases} \right)
\end{align*}}

For a leaf node $v$, we set $\EV(v)=0$ and $\OD(v)=0$, which serves as the base case of the recurrence. 

We compute $\EV(v)$ and $\OD(v)$, for every $v$ in $\ZZT(T)$.
This takes $\cO(n)$ time via a single DFS traversal of $\ZZT(T)$, provided that $\ZZT(T)$ has been constructed.

\begin{lemma}\label{lem:CC}
Let \(v\) be a node of \(\ZZT(T)\). Then the quantities \(\EV(v)\) and \(\OD(v)\), defined by the recurrences above, correctly compute the contribution of the subtree rooted at \(v\) to \(\CC\) in the even and odd parity cases, respectively.
\end{lemma}

\begin{proof}
We prove the claim by induction on the size of the subtree rooted at \(v\).
If \(v\) is a leaf, then the subtree of \(v\) contains no outgoing edges, and hence contributes nothing. Therefore \(\EV(v)=\OD(v)=0\), as required. Now let \(v\) be an internal node, with children \(u_1,\dots,u_k\), and let \(d_i\) denote the length of the edge label from \(v\) to \(u_i\), as defined by the construction. Since the contribution of the subtree rooted at \(v\) is the sum of the contributions of the subtrees rooted at its children together with the contribution of the edges incident to \(v\), it suffices to verify the recurrence for each child independently.

Consider first the even parity case. The contribution accumulated along the edge from \(v\) to \(u_i\) is \(\lfloor d_i/2\rfloor\). If \(d_i\) is even, then the parity state at \(u_i\) remains even, and the contribution of the subtree rooted at \(u_i\) is \(\EV(u_i)\). If \(d_i\) is odd, then the parity state at \(u_i\) becomes odd; in this case, the contribution from the subtree rooted at \(u_i\) is \(\OD(u_i)\), together with the additional term corresponding to the children of \(u_i\), namely \(|\ch(u_i)|\). This is exactly the recurrence defining \(\EV(v)\).

The argument for \(\OD(v)\) is analogous. When the parity state at \(v\) is odd, the contribution along the edge to \(u_i\) is \(\lfloor (d_i-1)/2\rfloor\), and the parity state at \(u_i\) is determined by the parity of \(d_i\) in the same way as specified by the recurrence. The resulting contribution is therefore exactly \(\OD(v)\).

By induction, \(\EV(v)\) and \(\OD(v)\) correctly compute the contribution of the subtree rooted at \(v\) to \(\CC\).
\end{proof}

\paragraph{Querying.} Given a pattern \(P\) of length \(m\), we first construct \(\ZZ(P)\) and locate its locus \((v,m)\) in \(\ZZT(T)\) in \(\cO(m)\) time. Let \(d=\sd(v)-m\). By \Cref{lem:CC}, the contribution of the subtree rooted at \(v\) is given by \(\EV(v)\) in the even case and by \(\OD(v)\) in the odd case. It remains to account for the contribution of the path segment between the locus and \(v\), which is \(\lfloor d/2 \rfloor\), together with the additional boundary contribution \(|\ch(v)|\) in the odd case. Therefore,
{
\abovedisplayskip=2mm
\belowdisplayskip=2mm
\small
\[
\CC(P,T) =
\begin{cases}
\lfloor d/2 \rfloor + \EV(v) & \text{if } d \text{ is even}, \\[4pt]
\lfloor d/2 \rfloor + \OD(v) + |\ch(v)| & \text{if } d \text{ is odd}.
\end{cases}
\]
}
The query time is clearly \(\cO(m)\).

\begin{example} Consider the $\ZZT(T)$ in \Cref{fig:CC_example}. We compute $\EV(v)$ and $\OD(v)$ for every node $v$ via a DFS traversal.

For each leaf node $\ell_i, i\in[1,4]$, we have $\EV(\ell_i) = \OD(\ell_i) = 0$. 
Node $v_2$ has two children: $\ell_1$ with edge length $d_1 = 8 - 6 -1 = 1$, and $\ell_2$ with edge length $d_2 = 9 - 6 -1 = 2$.

{\abovedisplayskip=1mm
\belowdisplayskip=1mm
\footnotesize \begin{align*}
\EV(v_2) &=(\lfloor 1/2 \rfloor + \OD(\ell_1) + |\ch(\ell_1)|) +(\lfloor 2/2 \rfloor + \EV(\ell_2))=1
\\
\OD(v_2) &=(\lfloor 0/2 \rfloor + \EV(\ell_1))  + (\lfloor 1/2 \rfloor + \OD(\ell_2)+|\ch(\ell_2)|) = 0
\end{align*}}

Similarly, we compute $\EV(v_3) = 2$ and $\OD(v_3)  = 0$,
and $\EV(v_1)= 5$ and $\OD(v_1) = 5$; see \Cref{fig:CC_example}.
We now answer two \CC queries: (1) $P=\texttt{ac}$ in \Cref{fig:CC_even}, where $|P|$ is even. We compute $\ZZ(P)=\texttt{ac}$ and locate its locus $w=(v_1,2)$ in $\ZZT(T)$. Since $\sd(v_1)=4$, we have $d=\sd(v_1)-|P|=2$, which is even, so
$\CC(P,T) = \lfloor d/2 \rfloor + \EV(v_1) = 1 + 5 = 6$. (2) $P=\texttt{bac}$ in \Cref{fig:CC_odd}, where $|P|$ is odd. We compute $\ZZ(P)=\texttt{acb}$ and locate its locus $w=(v_1,3)$. Now $d=\sd(v_1)-|P|=1$, which is odd, so the remaining letter is combined with the first letter of every outgoing edge of $v_1$, contributing $|\ch(v_1)|=2$ additional length-$2$ fragments, giving
$
\CC(P,T) = \lfloor d/2 \rfloor + \OD(v_1) + |\ch(v_1)| = 0 + 5 + 2 = 7.
$
\end{example}

\paragraph{Main Result.} From the above, we directly obtain: 

\begin{theorem}
\CC can be solved using an index of $\cO(n)$ size supporting $\cO(m)$-time queries. The  index can be constructed in $\cO(n)$ time and space given $\ZZT(T)$.
\end{theorem}

\section{Indexes for \TCPR}\label{sec:tcprmain}

\subsection{Baseline Index for \TCPR} \label{sec:TCPR-BA}

We construct Navarro's SA-based index from~\cite{navarro2020contextual}~(see \Cref{sec:related}),  which, given $P$ and $q$, reports, for every string pair $(L,R)$ in $\mathcal{C}_{T}(P,q)$, the interval $[s,e]$, such that 
$\{\SA[s],\ldots,\SA[e]\}$ are exactly the occurrences of $LPR$ in $T$. Upon a \TCPR query, we evaluate the scoring function $f_T(LPR)$ for each $(L,R)$ pair and select the top-$k$ candidates, breaking ties arbitrarily. The construction takes $\cO(n)$ time and space. 
The query time is $\cO(m+|\mathcal{C}_T(P,q)|\cdot \mathcal{T}_{f_T})$, where $\mathcal{T}_{f_T}$ is the time required to evaluate $f_T$ on $LPR$; and the top-$k$ selection can then be performed in linear time with the algorithm of~\cite{blum1973time}. Note that the \(LPR\) strings are not known in advance;  precomputing their scores would therefore require enumerating all possible pairs \((P,q)\), and thus $\Theta(n^2)$ admissible \(LPR\) strings, which is prohibitively expensive.

\subsection{Our Index for \TCPR}\label{sec:tcpr}

Our index exploits \ZZT, which compactly represents the distinct string pairs $(L, R)$ for every $P$, and a geometric data structure to quickly retrieve the top-$k$ of these pairs w.r.t. the scoring function $f_T$, which comprise the \TCPR query answer.  

\paragraph{Construction.} Construct \(\ZZT(T)\) using \Cref{the:con}. For each  node \(w\) of \(\ZZT(T)\), compute the rank of $w$ in a preorder traversal of \(\ZZT(T)\).  This rank is called the \emph{preorder rank} of $w$ and denoted by \(\preorder(w)\). The root has preorder rank \(1\). The preorder rank is used to locate pairs in the query answer. 
Then, for every non-root node \(w\) in \(\ZZT(T)\), construct the following three-dimensional point
{
\abovedisplayskip=2mm
\belowdisplayskip=2mm
\[
\bigl(\preorder(w),\; \sd(w),\; \sd(\parent(w))\bigr),
\]
}

\noindent and weigh it by \(f_T(\str(w))\), where \(f_T\) is the chosen scoring function; see subsection ``Scoring Functions'' below.  
The first dimension will ensure that every reported element is of the form $LPR$ (i.e., the locus of $LPR$ ``descends'' from the locus of $P$), and the second and third dimensions will ensure that the length constraints are satisfied (i.e., $|L|=|R|=q$ and $m=|P|$).
Let \(S\) be the set of all such points. We preprocess \(S\) into the data structure
described in \Cref{sec:preliminaries} for top-\(k\) orthogonal range reporting, which can be constructed in $\widetilde{\cO}(n)$ time~\cite{rahul2011efficient}. Thus, the total construction time is \(\widetilde{\cO}(n)\). 

\paragraph{Querying.} Given a pattern \(P\) of length \(m\), we construct \(\ZZ(P)\) and locate its locus \((u,m)\) in \(\ZZT(T)\) in \(\cO(m)\) time. For the given \(q\geq 0\), every descendant node \(w\) of \(u\) with 
{
\abovedisplayskip=3mm
\belowdisplayskip=3mm
\[
\sd(\parent(w)) < m + 2q \le \sd(w)
\]}

\noindent corresponds to a distinct element of \(\mathcal{C}_T(P,q)\). Consider the locus of such an element. It corresponds to a string of the form $LPR$ with $|L|=|R|=q$ and $m=|P|$. 
Accordingly, we report the top-\(k\) weighted points in $S \cap Q$, where
{
\abovedisplayskip=3mm
\belowdisplayskip=3mm
\[
Q=[\preorder(u),\preorder(u')] \times [m+2q,\infty) \times [0,m+2q),
\]}

\noindent and \(u'\) is the rightmost leaf descendant of \(u\). 
We now argue for correctness. 
By the preorder rank assignment, 
we have that $\preorder(w)\in[\preorder(u),\preorder(u')]$ if and only if $w$ is a descendant of $u$; and then the second and third dimensions of the query box guarantee that $|LPR|$ is exactly $m+2q$ as desired. We have that the prefix of $\str(w)$ of length $m+2q$ is $\ZZ(LPR)$.
Let $i$ be such that $\ZZ(LPR)$ is a prefix of string $Z_i$. Then $LPR$ occurs at position $(i-\lceil \ell/2\rceil + 1)$ of $T$, with $\ell=m+2q$. 

The time for reporting the top-$k$ points is \(\widetilde{\cO}(k)\), and so the total query time becomes \(\cO(m)+\widetilde{\cO}(k)\).

\begin{figure}[ht]
  \centering
  \includegraphics[width=0.5\linewidth]{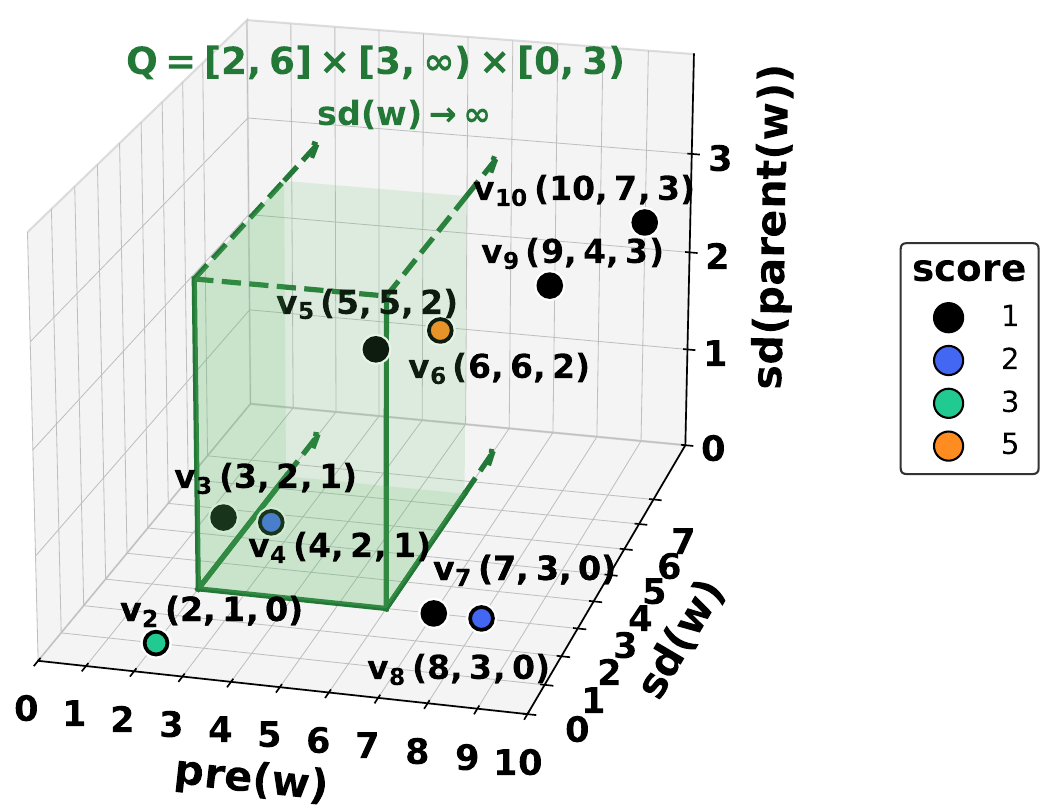}
  \vspace{+2mm}
  \caption{Each point 
in $S$ is colored based on its weight (score), assigned by $f_T$. The green hyperrectangle corresponds to $Q$.}
  \label{fig:visu_TCPR}
\end{figure}

\begin{example}\label{exp:TCPR}
Let $T=\texttt{banana}$, the  $\ZZT(T)$ in \Cref{ex:c}, and a scoring function $f_T$.  We perform a preorder traversal of $\ZZT(T)$ and assign rank $\preorder(v_i)=i$ to each node $v_i, i \in[1,10]$. For every node $w\in \ZZT(T)$ except the root node $v_1$, we construct a three-dimensional point $\bigl( \preorder(w), \sd(w), \sd(\parent(w))\bigr)$, which is weighed by $f_T(\str(w))$, as shown in \Cref{fig:visu_TCPR}. For example, for $v_4$, we have $\preorder(v_4) = 4$, $\sd(v_4)=|\str(v_4)|=|\texttt{an}|=2$, and $\sd(\parent(v_4)) = \sd(v_2)=|\texttt{a}|=1$. Thus, for $v_4$, we construct a three-dimensional point $(4, 2, 1)$, which is assigned a weight $f_T(\str(v_4))=2$. The set $S$ consists of all such three-dimensional points constructed based on $v_2, \ldots, v_{10}$, and then we construct the data structure from~\cite{rahul2011efficient} on $S$. The points in it have coordinates in $[0,10]\times[0,7]\times[0,3]$. Let $P = \texttt{a}$ be a query pattern of length $m=1$, $q = 1$, and $k=1$.
We construct $\ZZ(P) = \texttt{a}$ and locate its locus $(v_2,1)$, which has $\preorder(v_2) = 2$ in $\ZZT(T)$. 
As $m + 2q = 3$ and the rightmost leaf descendant of $v_2$ is $v_6$, whose $\preorder(v_6) = 6$, we have 
{
\abovedisplayskip=2mm
\belowdisplayskip=2mm
\begin{align*}
  Q &= [\preorder(v_2),\,\preorder(v_6)] \times [m+2q,\,\infty) \times [0,\,m+2q) \\
    &= [2,\,6] \times [3,\,\infty) \times [0,\,3).
\end{align*}}
Since $S\cap Q=\{(5,5,2),(6,6,2)\}$ and $k=1$, we report the top-$1$ point, namely $(6,6,2)$. This point corresponds to $v_6$, which is assigned a weight $f_T(\str(v_6))=5$.

\end{example}

\paragraph{Main Result.} From the above, we obtain: 

\begin{theorem}
\TCPR can be solved using an index of $\widetilde{\cO}(n)$ size supporting \(\cO(m)+\widetilde{\cO}(k)\)-time queries. The index can be constructed in $\widetilde{\cO}(n)$ time and space given $\ZZT(T)$. 
\end{theorem}

\paragraph{Bounded-Length Optimization.} Motivated by the fact that there is often an upper bound on the length of $LPR$ in applications~\cite{contexticde}, we propose a parameterized version of our index that answers queries faster, has a much smaller size in practice,  and can also be constructed faster and using less space. This index is  parameterized by an integer $B\geq m+2q$. 

Let \(T[1\dd n]\) be a string of length \(n\) over alphabet \(\Sigma\). For every position \(i \in [1,n]\), we define the \emph{\(B\)-bounded ZigZag} string \(Z_i(T)\) as the sequence
\(
Z_i = T[i]\,T[i+1]\,T[i-1]\,T[i+2]\,T[i-2]\cdots,
\)
truncated after at most \(B\) letters; we drop \((T)\) from \(Z_i(T)\) when the context is clear. More precisely, we have
$|Z_i| = \min\!\left(B,\; \min\left\{k \ge 1 : i + (-1)^k\left\lfloor \frac{k}{2} \right\rfloor \notin [1,n] \right\}-1\right)$. 
Then, for all \(k\in[1,|Z_i|]\), we have
$Z_i[k] = T\!\left[i + (-1)^k\left\lfloor \frac{k}{2} \right\rfloor\right]$.

The counterpart of \(\ZZT(T)\) over the \(n\) \(B\)-bounded ZigZag strings of $T$ can be constructed in \(\cO(n\log n)\) time using the same algorithm described in \Cref{sec:ZigZagindex}. After the index construction,  
we can proceed with the construction of the \TCPR index as described in the subsection ``Construction'' above. 
The query algorithm works as described in the subsection ``Querying'' above provided that $m+2q \leq B$.

\paragraph{Scoring Functions.} Recall from \Cref{sec:intro} that we can support any function $f_T:\Sigma^*\rightarrow \mathbb{R}_{\geq 0}$ computable in $\cO(1)$ time after $\cO(n)$-time preprocessing. Our \ZZT-based index relies on the fact that the inputs to $f_T$ are precisely the $LPR$ strings corresponding to explicit nodes in $\ZZT(T)$. This permits using well-known importance measures such as \emph{term frequency} (\textsf{TF})~\cite{ricbook}, computed as $|\occ_T(LPR)|$, and \emph{span} (\textsf{SP})~\cite{tao2007exploration,hawking1995proximity}, computed as $(o_r-o_l)$, where $o_r$ and $o_l$ are the rightmost and leftmost occurrences of $LPR$ in $T$, respectively. Other functions, such as \emph{term proximity} (\textsf{TP})~\cite{hon2014space,rasolofo2003term,cummins2009learning}, computed as $\min_{i\in [1,t)}(o_{i+1}-o_i)$ for occurrences $o_1,\ldots,o_t$ of $LPR$ in ascending order, can also be used as $f_T$, but their preprocessing cost must be accounted for in construction.  

\section{Related Work}\label{sec:related}

String indexing and pattern retrieval have been extensively studied, including work on full-text indexes and on top-$k$ retrieval; see~\cite{ricbook} for a textbook.  
Our work falls into the emerging area of contextual pattern matching. 

As discussed in \Cref{sec:intro}, Navarro~\cite{navarro2020contextual} introduced the problem of reporting the context of a given pattern, and Li et al.~\cite{contexticde} introduced the problem of counting the context size of a pattern and of mining all patterns of a given length with a sufficiently large context. These works employ classic indexes, such as suffix or prefix trees~\cite{DBLP:conf/focs/Weiner73}. There are four indexes 
that \emph{report} all pairs in $\mathcal{C}_T(P,q)$ for a pattern $P$.   
The first index~\cite{navarro2020contextual} uses a suffix array (SA) and other auxiliary data structures; it has \(\cO(|P|+|\mathcal{C}_T(P,q)|)\) query time and uses \(\cO(n)\) space.  
The second index~\cite{navarro2020contextual} has 
\(\cO(|P|\log\log n + |\mathcal{C}_T(P,q)|\log n)\) query time and uses \(\cO(\bar{r}\log (n/\bar{r}))\) space, where \(\bar{r}\) is the maximum number of equal-letter runs in the BWT (Burrows-Wheeler transform) of \(T\) and of its reverse. The third index~\cite{abedinDCC23} has \(\cO(|P|+|\mathcal{C}_T(P,q)|\log q \log(n/r))\) query time, where \(r\) is the number of equal-letter runs in the BWT of \(T\), and uses \(\cO(r\log(n/r))\) space. The query time for the third index is slightly worse than that for the second. The fourth index~\cite{navarro2026} has \(\cO(|P|+|\mathcal{C}_T(P,q)|)\) query time and uses \(\cO(\bar{e})\) space, where $\bar{e}$ is the space occupied by the Symmetric Compact Directed Acyclic Word Graph (SCDAWG)~\cite{DBLP:journals/jacm/BlumerBHME87} of $T$.  
Our work differs substantially from the above works in that it introduces: (I) several novel contextual queries (which cannot be answered by the above indexes) and specialized indexes to answer them; and (II) \ZZT, a general index for answering contextual queries efficiently, which  allows directly retrieving all the contextual information for a pattern.

The family of affix trees~\cite{DBLP:journals/algorithmica/Maass03,DBLP:conf/dcc/CanovasR17} and affix arrays~\cite{DBLP:journals/tcs/Strothmann07} is closely related to \ZZT: these indexes support bidirectional search, allowing a matched pattern 
to be extended both to the left and to the right, while maintaining its occurrences during the search. Similar functionality is provided by 
BWT-based bidirectional indexes~\cite{DBLP:conf/bibm/LamLTWWY09,DBLP:conf/spire/BelazzouguiC20,DBLP:conf/cpm/ArakawaNS22}. 
However, none of these indexes \emph{is an alternative} to \ZZT for our purposes. 
These indexes are designed to navigate from a known pattern to its bidirectional extensions, whereas our contextual queries require efficiently aggregating 
over all (balanced) contexts \(LPR\), with \(|L|=|R|\), around a given pattern 
\(P\). In particular, these indexes do not by themselves provide the contextual information associated with \(P\) in the compacted form provided by \ZZT.

\section{Experimental Evaluation}\label{sec:experiments}

\begin{table}[t]
\centering
\caption{Dataset characteristics.} \label{tab:data}
\begin{tabular}{llrr}
\toprule
\textbf{Dataset} & \textbf{Domain} & \textbf{Length $n$} & \textbf{Alphabet size $|\Sigma|$} \\
\midrule
\wiki \cite{pizzachili} & biographical articles & 474,214,798 & 36 \\
\bst \cite{boost} & GitHub repository & 3,000,000,000 & 68 \\
\sdsl \cite{sdsl} & GitHub repository & 3,000,000,000 & 67 \\
\sars \cite{NCBI} & bioinformatics & 3,000,000,000 & 4 \\
\chr \cite{chr} & bioinformatics & 3,000,000,000 & 4 \\
\bottomrule
\end{tabular}
\end{table}

\paragraph{Datasets.} We used $5$ publicly available, large-scale datasets of up to about $3$ billion letters  that were also used in related work~\cite{contexticde}; see \Cref{tab:data}. 
\wiki is a collection of 
Wikipedia articles~\cite{pizzachili}. 
\bst and \sdsl represent all commits from two GitHub repositories (Boost~\cite{boost} and \sdsl~\cite{sdsl}), retrieved from~\cite{Getgit}. \sars is a genomic dataset~\cite{NCBI}.
\chr is a dataset 
from the $1,000$ Genomes Project~\cite{GenomesProject} 
representing chromosome 19 sequences from $1,000$ human haplotypes~\cite{chr}. 

We report below results for one dataset per query type; the results for the other datasets are in the appendix. 

\paragraph{Methods.} Since \emph{no} out-of-the-box solution  exists for the type of queries we consider (see \Cref{sec:related}), we compared our \LFCS, \LCCS, and \CC indexes, referred to as \LFCSZT, \LCCSZT, and \CCZT, to their corresponding baselines, \LFCSBA, \LCCSBA, and \CCBA. 

We consider our \TCPR index both without and with the bounded-length optimization, referred to as \TCPRZTminus and \TCPRZT, respectively. In both cases, we did not use the data structure from \Cref{sec:tcpr} for top-$k$ queries, since its $\cO(n\log^d n)$ space usage~\cite{rahul2011efficient} is prohibitive for our dataset sizes; here, we have \(d=3\). We instead used an R$^*$-tree~\cite{beckmann1990r} to index the three-dimensional points for orthogonal range reporting (ORR).  The R$^*$-tree was shown to  
outperform other ORR data structures~\cite{contexticde}, such as Range trees~\cite{bentley1978decomposable,DBLP:conf/focs/Lueker78} and KD trees~\cite{bentley1975multidimensional}. The R$^*$-tree uses $\cO(n)$ space; and the $f_T$ scores were stored as satellite information. 
After querying the R$^*$-tree for ORR, we obtain all points in $S\cap Q$ (see \Cref{sec:tcpr}), and then select the top-$k$ ones by their $f_T$ scores. 

We compared \TCPRZT to the baseline from \Cref{sec:TCPR-BA}, referred to as \TCPRBA, using 
the \textsf{TF}, \textsf{SP}, and \textsf{TP} scoring functions from \Cref{sec:tcpr}. \TCPRZTminus was worse than \TCPRZT in all efficiency measures and worse than \TCPRBA in all measures except query time; see \Cref{sec:exp:ablation}. 

\paragraph{Setup.} We used all relevant measures of efficiency~\cite{pvldb23} (i.e., query time; index size; construction time; and construction space) in our evaluation, each for the problem parameter(s) that affect it. We measured query and construction time using the \texttt{chrono} C++ library, construction space using \texttt{/usr/bin/time -v}, and index size using \texttt{mallinfo2}. 

Note that for the pattern length $m$, we used the same values as in~\cite{contexticde}, and $m=9$ by default. When we varied $n$, we used increasingly longer prefixes of the dataset. As query patterns, we used the $1000$ most frequent length-$m$ substrings (in the collection $\mathcal{C}$ for \LCCS and in the text $T$ for all other query types). In \LFCS and \LCCS, we used by default $\tau=100$ and $\tau=5$, respectively. In \TCPR, by default we used $q=9$ and varied it from $3$ to $15$, and $k=10$ and varied it from $5$ to $80$. We set $B=m+2q$, which gave the best performance across all measures; by default, $B = m + 2q=27$. These values reflect realistic pattern and context sizes~\cite{contexticde}. 

All experiments were conducted on an AMD EPYC 7282 CPU with $1$\,TB of RAM and $1$\,TB of disk space.  
All indexes were implemented in \texttt{C++}. 
\emph{Our code and datasets are in}~\cite{githubZT}.

\subsection{\LFCS}\label{experiments:lfcs}

\begin{figure}[t]
  \begin{subfigure}[t]{0.16\linewidth}
    \includegraphics[width=\linewidth]{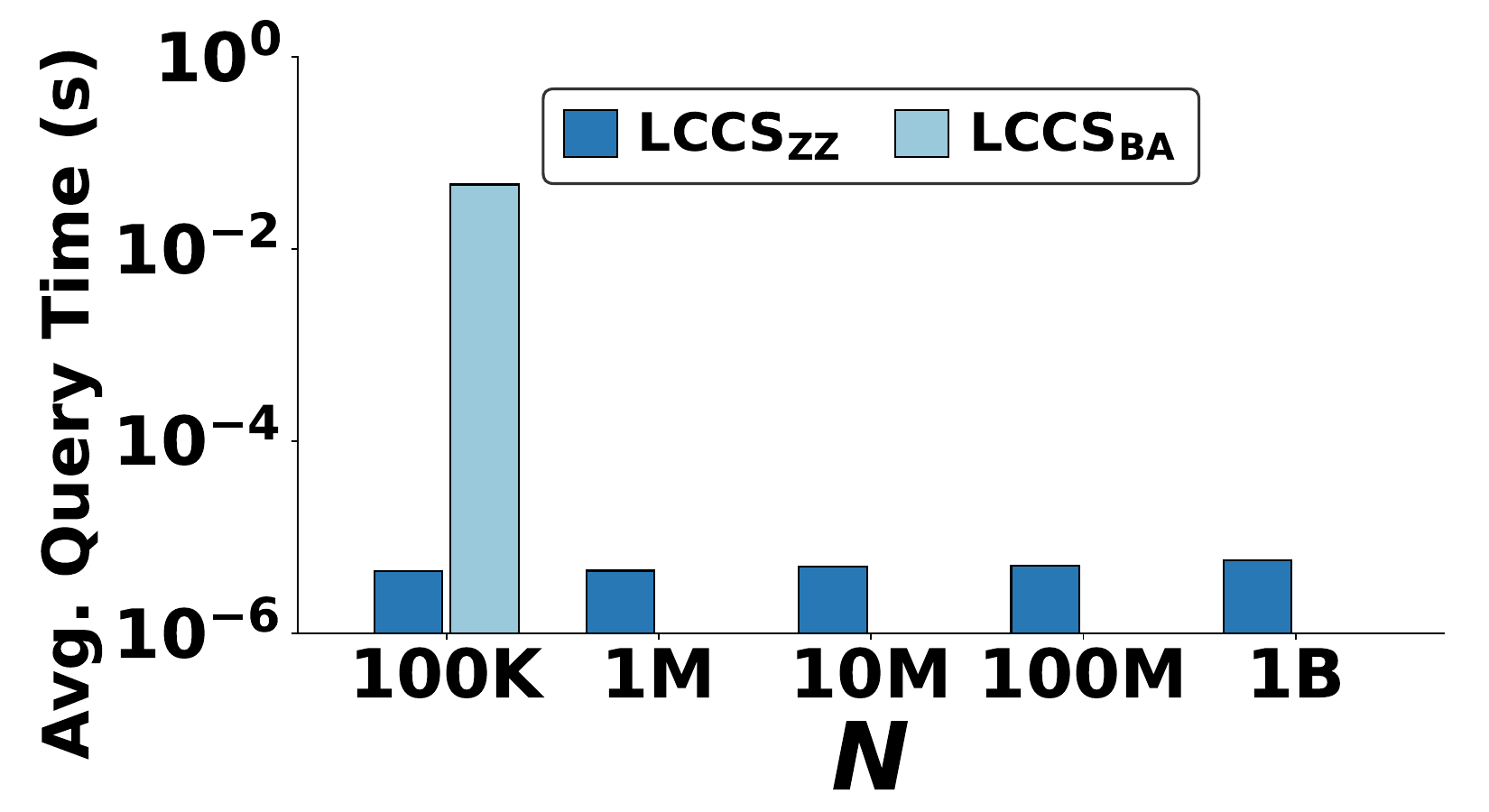}
    \caption{Query time vs. $N$}\label{fig:LCCS_qt_N}
  \end{subfigure}\hfill
  \begin{subfigure}[t]{0.16\linewidth}
    \includegraphics[width=\linewidth]{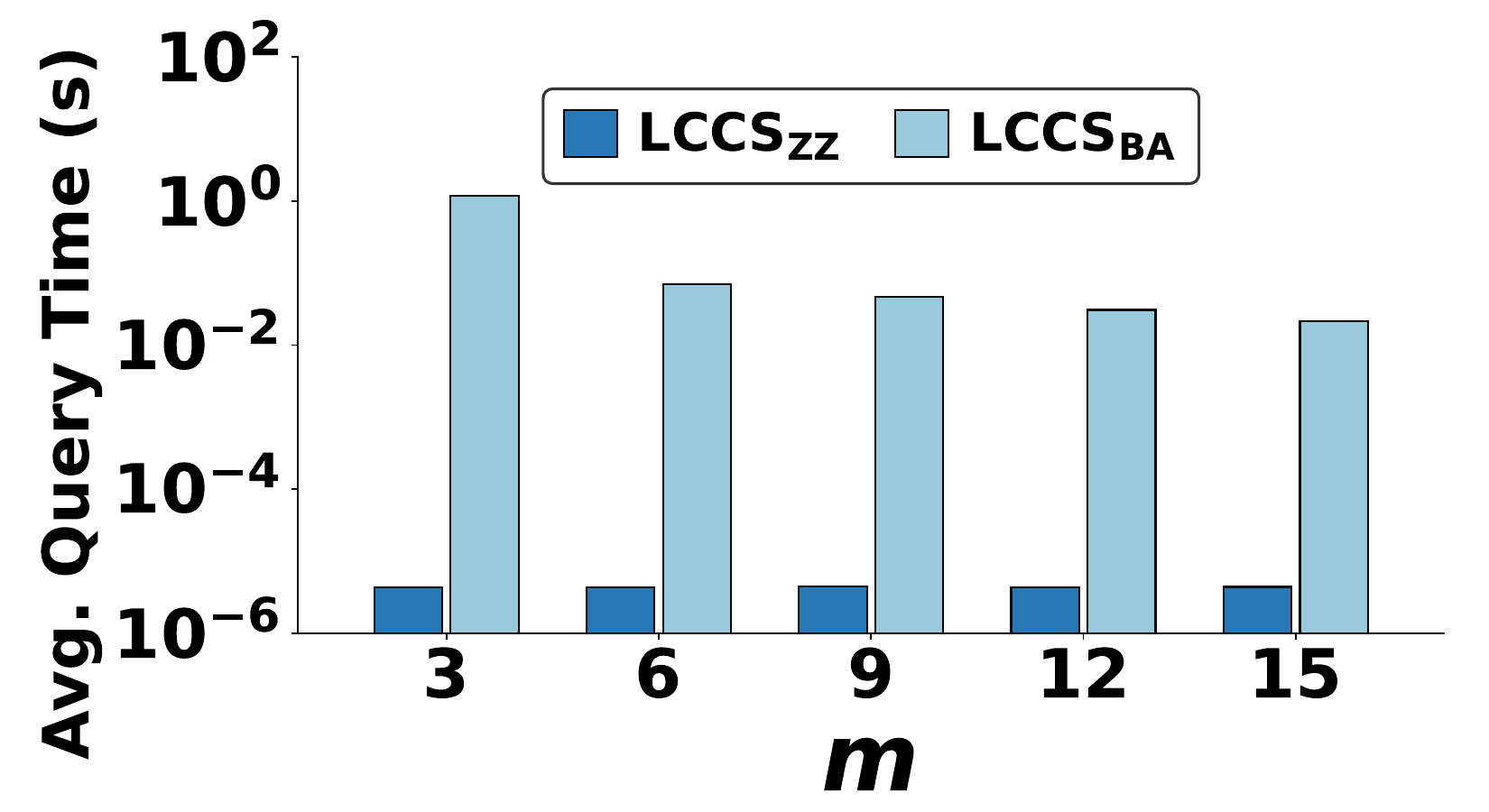}
    \caption{Query time vs. $m$}\label{fig:LCCS_qt_m}
  \end{subfigure}\hfill
  \begin{subfigure}[t]{0.16\linewidth}
    \includegraphics[width=\linewidth]{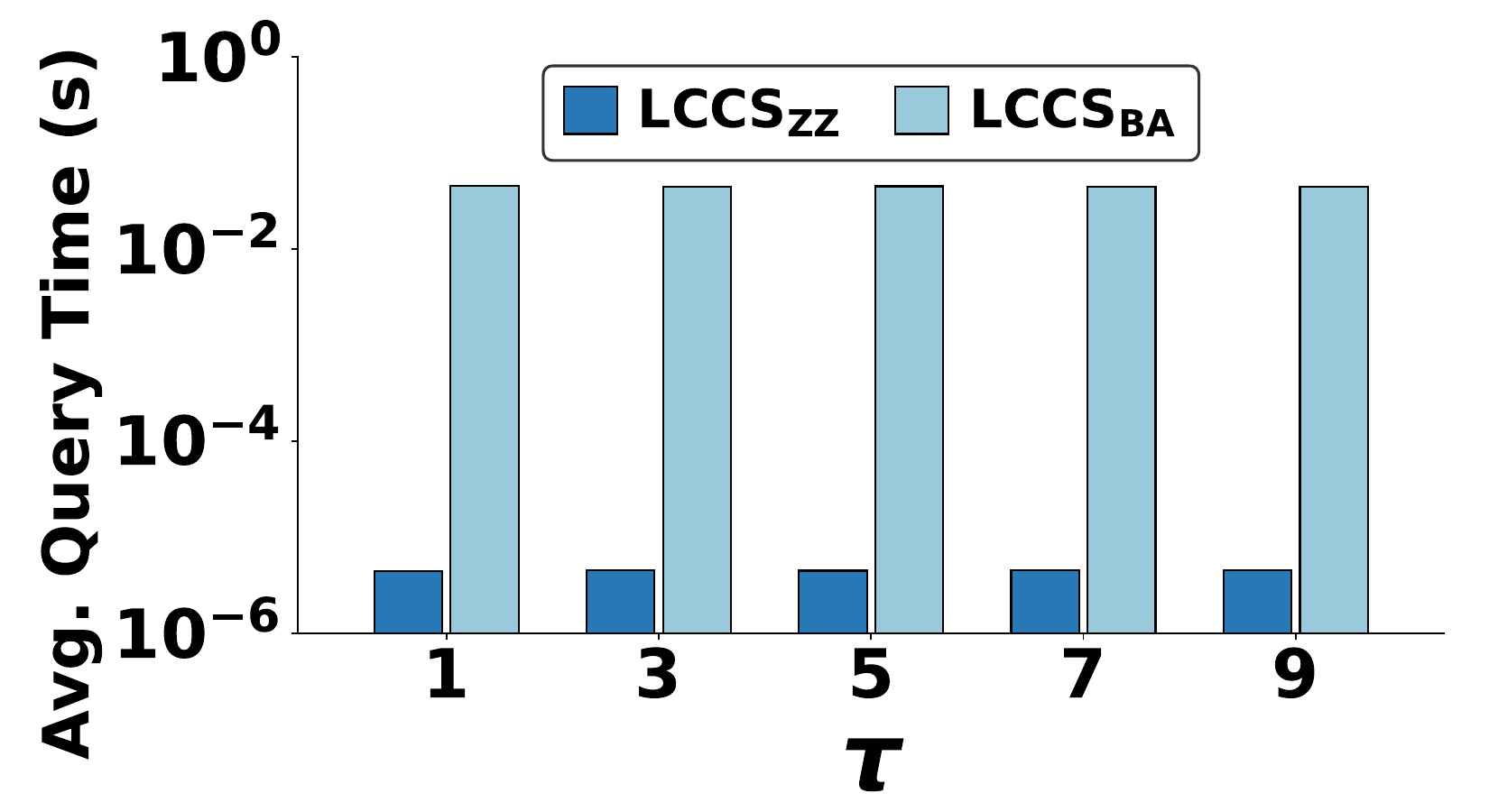}
    \caption{Query time vs. $\tau$}\label{fig:LCCS_qt_tau}
  \end{subfigure}\hfill
  \begin{subfigure}[t]{0.16\linewidth}
    \includegraphics[width=\linewidth]{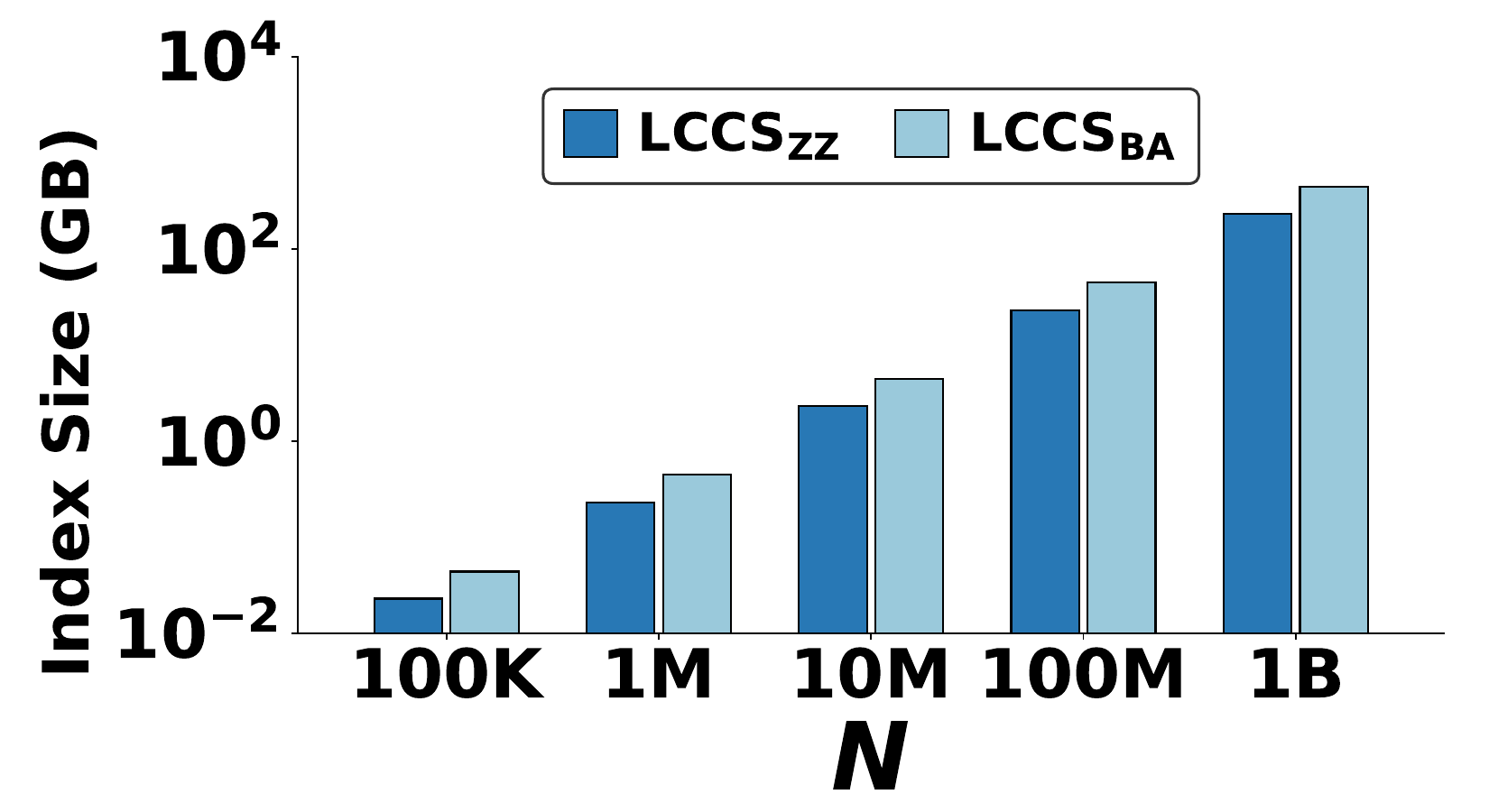}
    \caption{Index size vs. $N$}\label{fig:LCCS_is}
  \end{subfigure}\hfill
  \begin{subfigure}[t]{0.161\linewidth}
    \includegraphics[width=\linewidth]{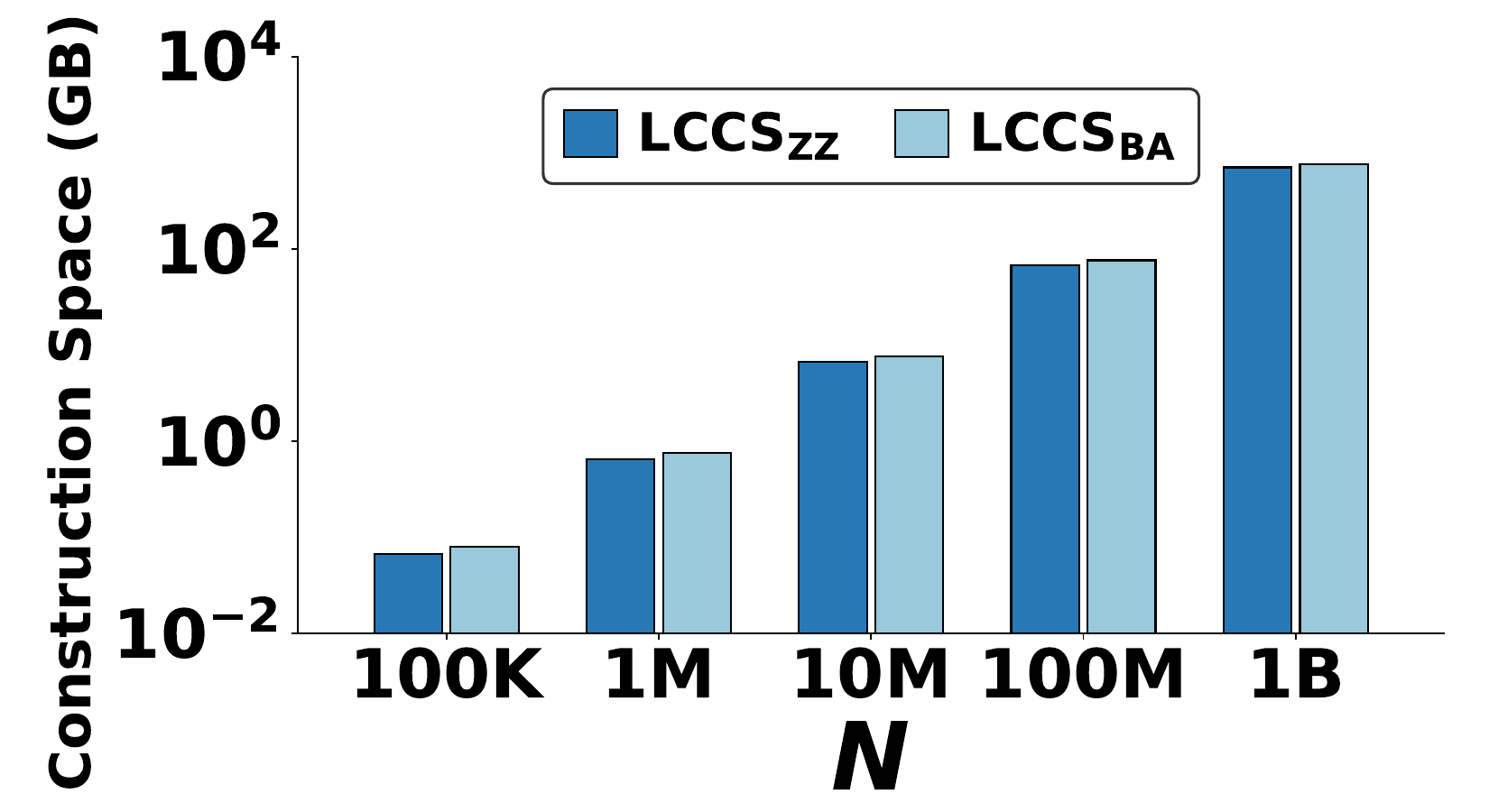}
    \caption{Constr.\ space vs. $N$}\label{fig:LCCS_cs}
  \end{subfigure}\hfill
  \begin{subfigure}[t]{0.16\linewidth}
    \includegraphics[width=\linewidth]{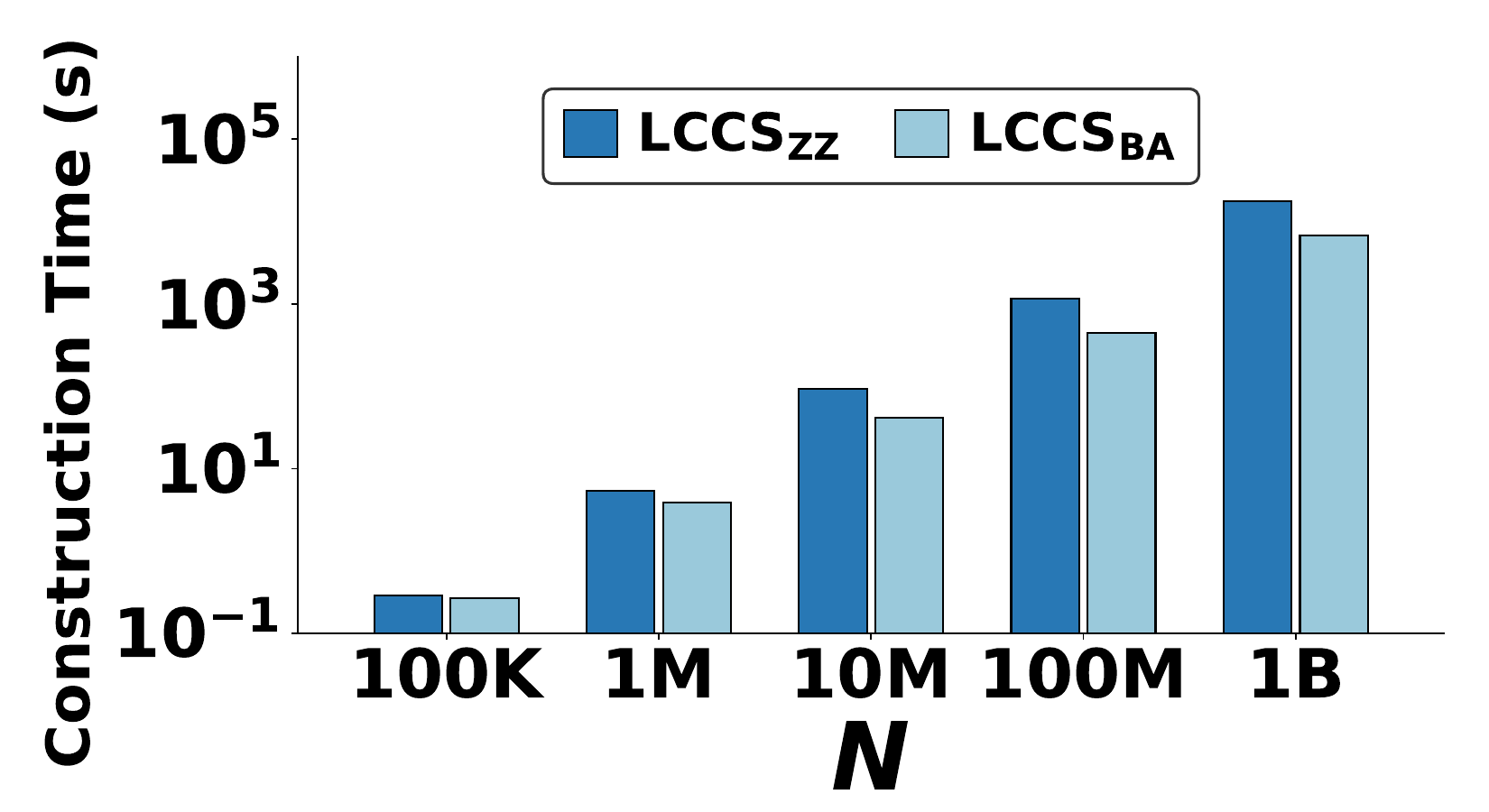}   \caption{Constr.\ time vs. $N$}\label{fig:LCCS_ct}
  \end{subfigure}
  \vspace{\captionspacing}
  \caption{Our \LCCS index vs. \LCCSBA on \sars. In \Cref{fig:LCCS_qt_N}, \LCCSBA did not terminate within $24$ hours for $N\geq 10^6$.}
\end{figure}

\begin{figure}[t]
  \centering
  \begin{subfigure}[t]{0.16\linewidth}
    \includegraphics[width=\linewidth]{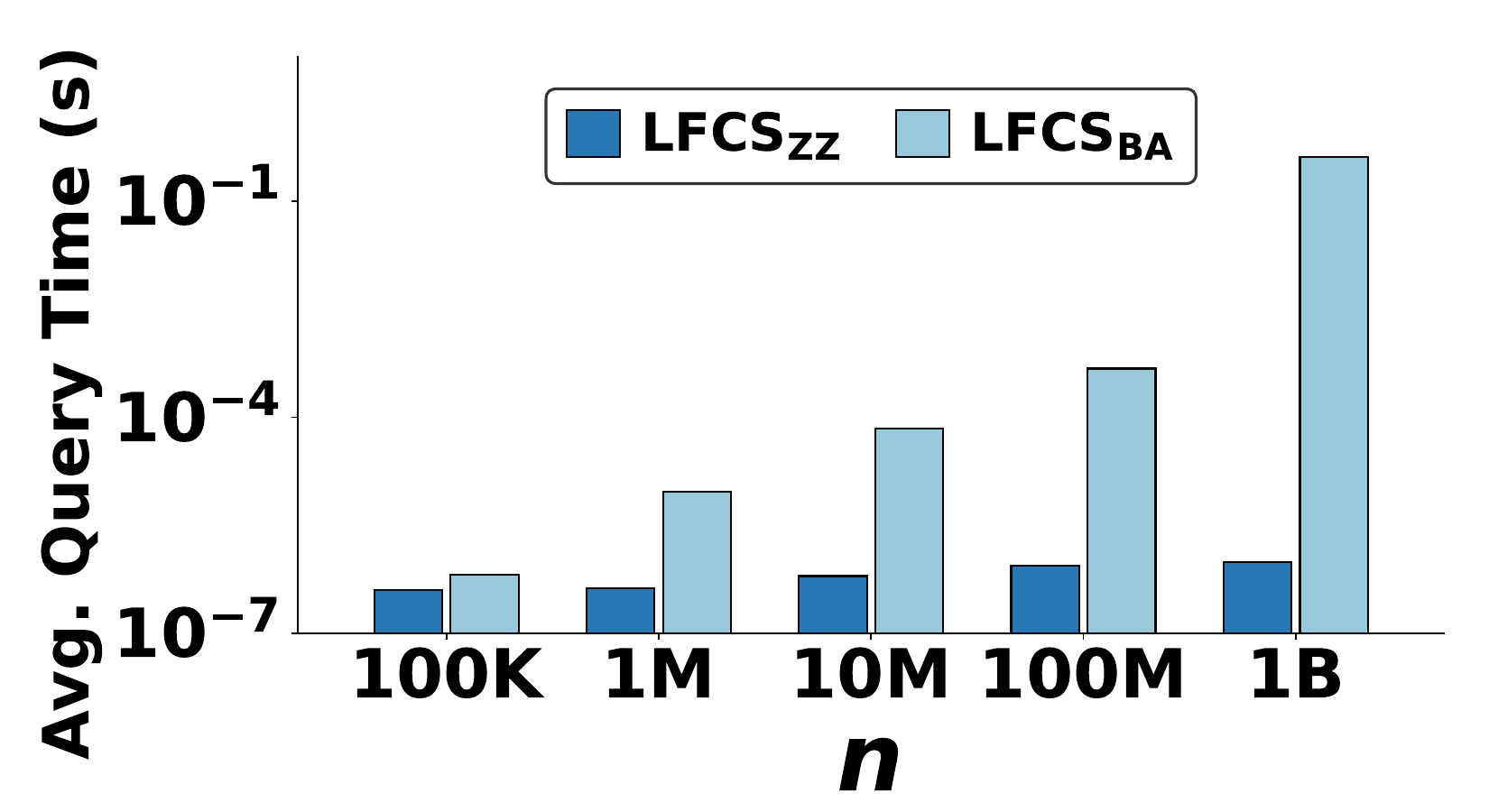}
    \caption{Query time vs. $n$}
    \label{fig:LFCS_SDSL_n_query}
  \end{subfigure}\hfill
  \begin{subfigure}[t]{0.16\linewidth}
    \includegraphics[width=\linewidth]{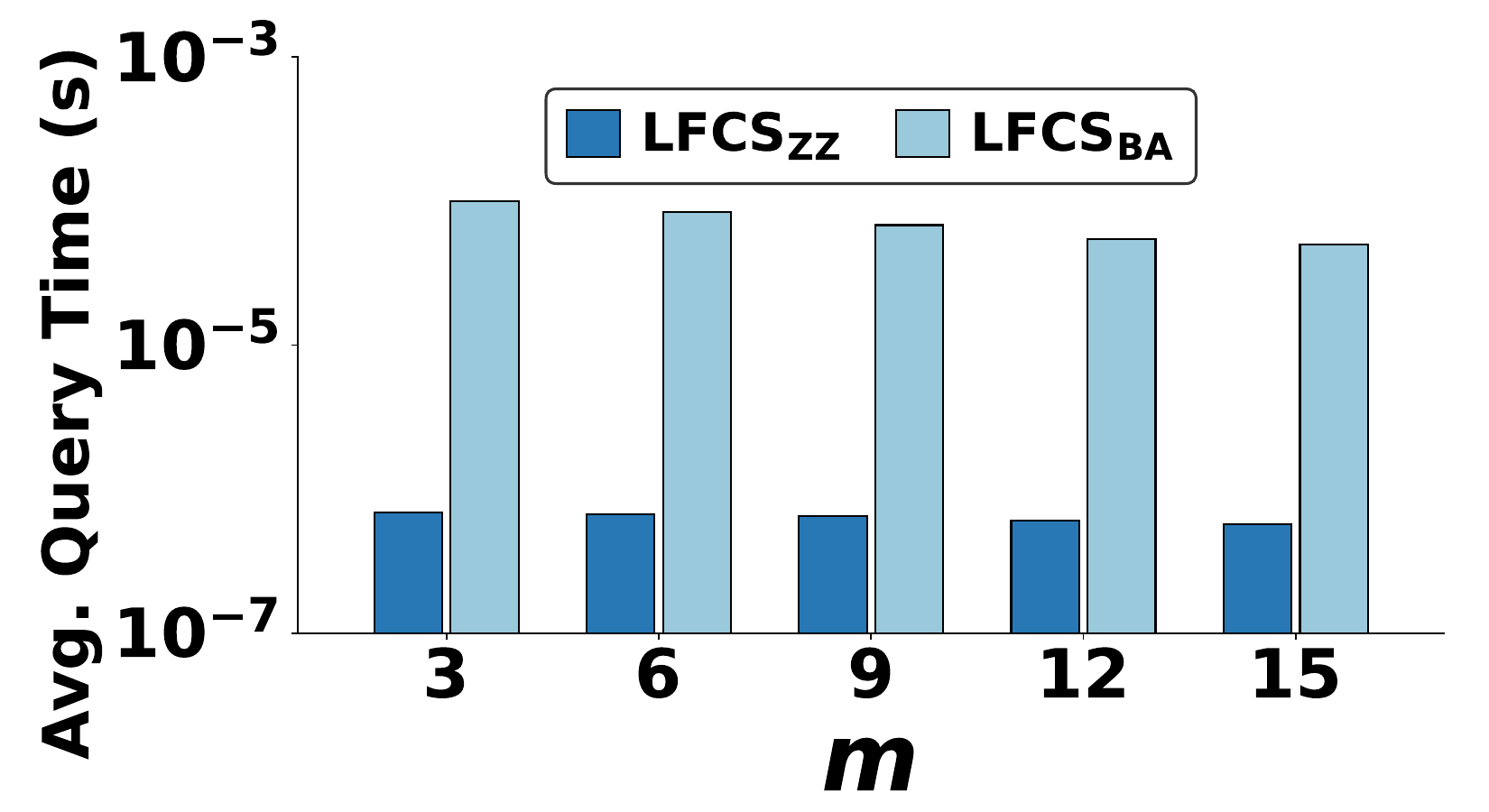}
    \caption{Query time vs. $m$}\label{fig:LFCS_SDSL_m_query}
  \end{subfigure}\hfill
  \begin{subfigure}[t]{0.16\linewidth}
    \includegraphics[width=\linewidth]{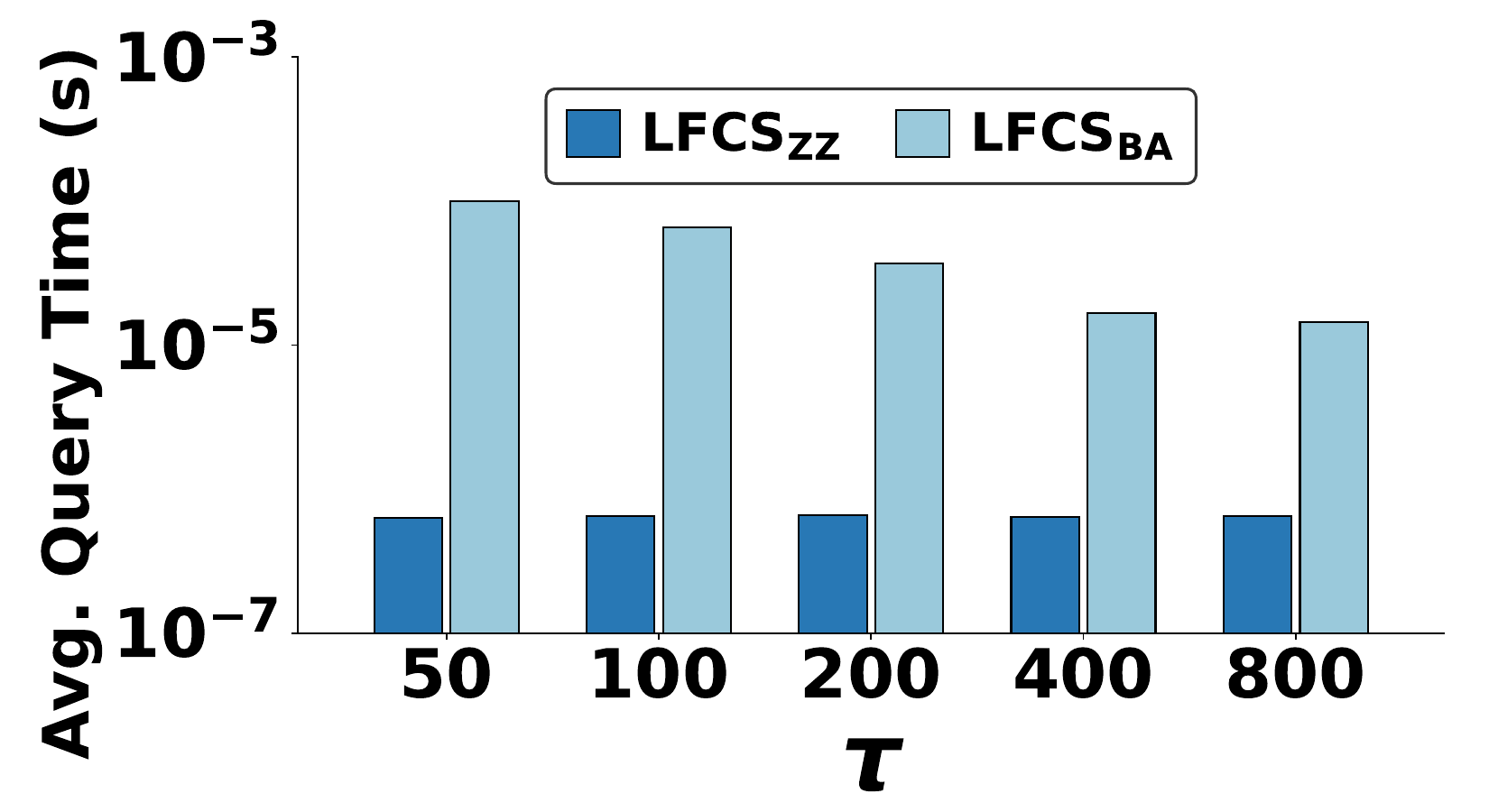}
    \caption{Query time vs. $\tau$}
    \label{fig:LFCS_SDSL_tau_query}
  \end{subfigure}\hfill
  \begin{subfigure}[t]{0.16\linewidth}
    \includegraphics[width=\linewidth]{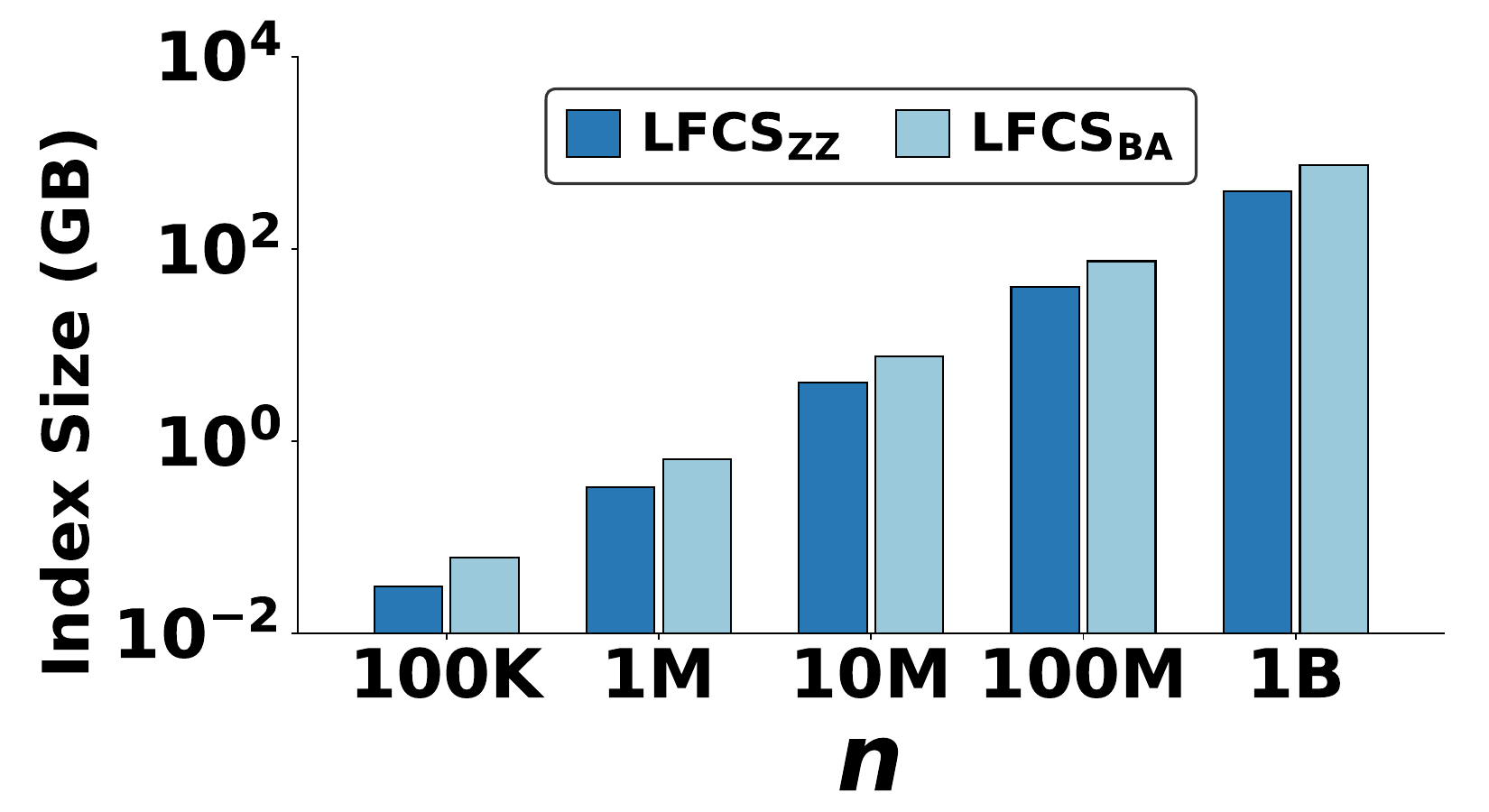}
    \caption{Index size vs. $n$}
    \label{fig:LFCS_SDSL_n_index}
  \end{subfigure}\hfill
  \begin{subfigure}[t]{0.16\linewidth}
    \includegraphics[width=\linewidth]{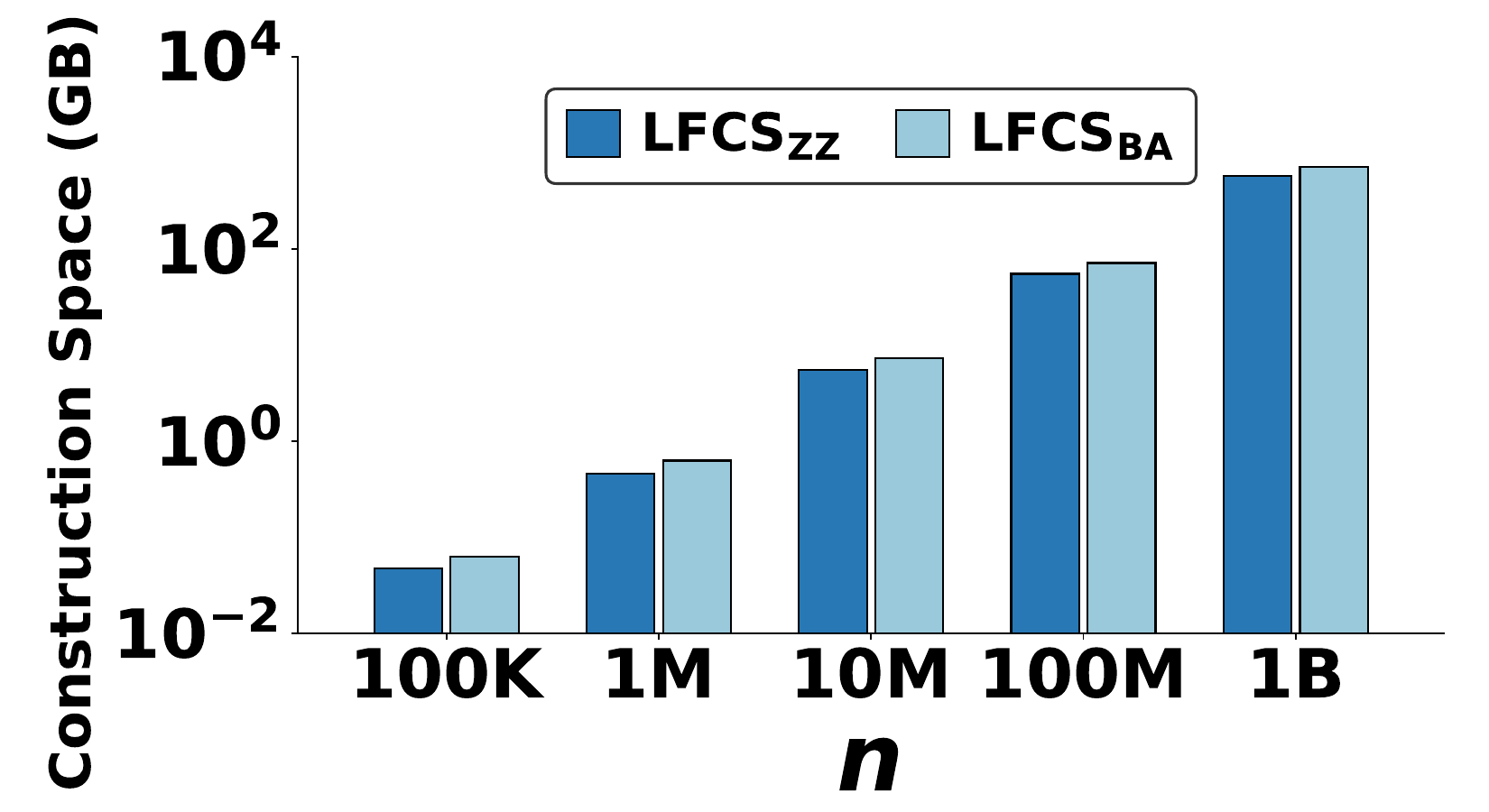}
    \caption{Constr.\ space vs. $n$}
        \label{fig:LFCS_SDSL_n_rss}
  \end{subfigure}\hfill
  \begin{subfigure}[t]{0.16\linewidth}
    \includegraphics[width=\linewidth]{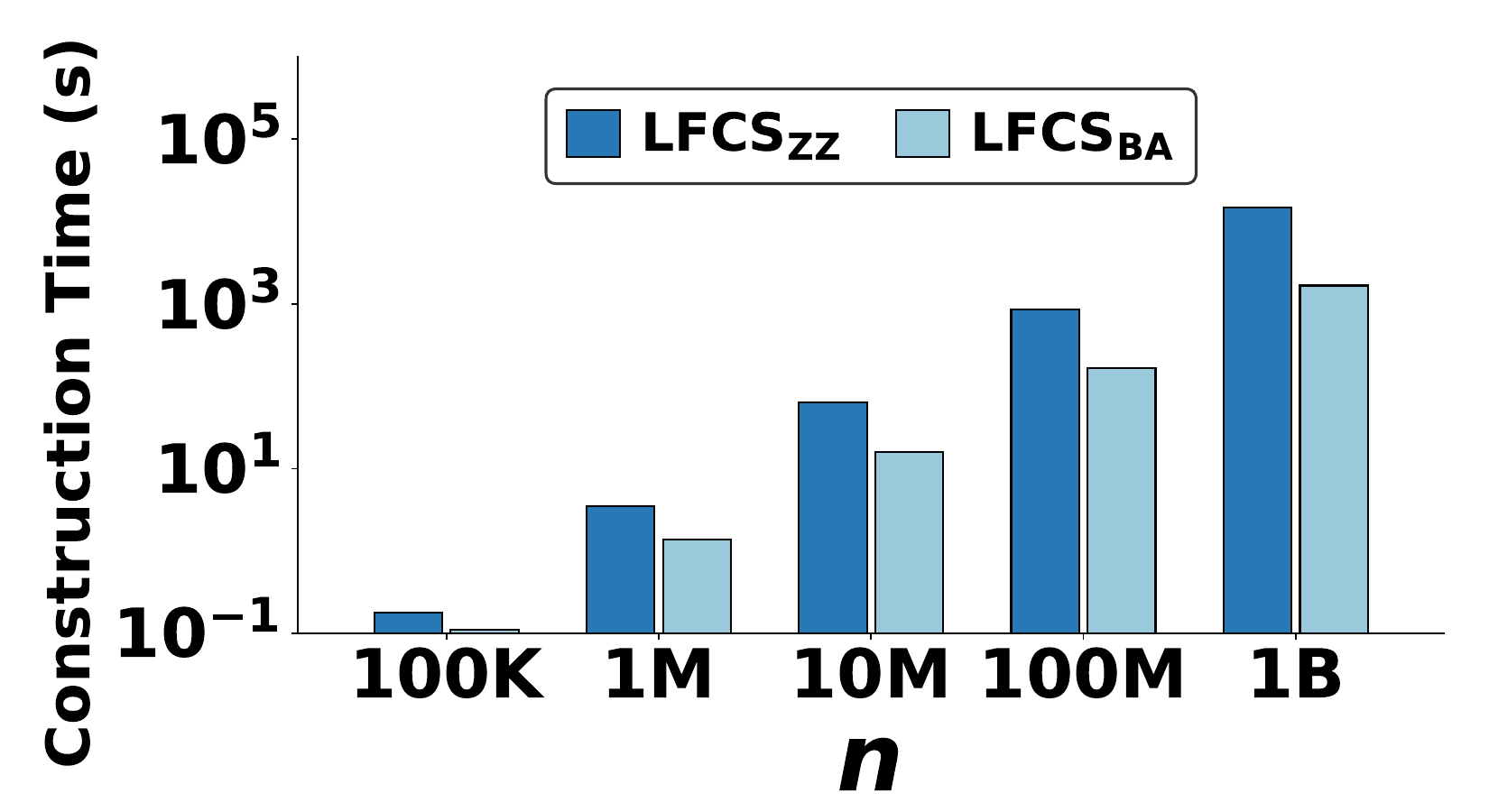}
    \caption{Constr.\ time vs. $n$}
        \label{fig:LFCS_SDSL_n_build}
  \end{subfigure}
  \vspace{\captionspacing}
  \caption{Our \LFCS index vs. \LFCSBA on \sdsl.}
\end{figure}

\paragraph{Query Time.} \Cref{fig:LFCS_SDSL_n_query} shows that, for varying $n$, our \LFCSZT index offers \emph{more than two orders of magnitude faster query time} on average (and up to five orders of magnitude faster)  than \LFCSBA, and its query time is \emph{unaffected by} $n$, as expected by its complexity. This experiment  shows the key benefit of our index: its query time is independent of $|\occ_T(P)|$. This term increases with $n$ and it is very large (e.g., for $n=10^9$, it is $127,560$ on average over all queries), making \LFCSBA impractical for realistically long strings. 

\Cref{fig:LFCS_SDSL_m_query} shows that \LFCSZT is \emph{$12$ times faster} than \LFCSBA on average over all values of $m$ and its query time grows more slowly than linearly in $m$. \LFCSBA becomes faster as $m$ increases, since  longer patterns occur fewer times in $T$ (i.e., the term $|\occ_T(P)|$ in the complexity of \LFCSBA becomes smaller) and induce fewer right extension candidates.

\Cref{fig:LFCS_SDSL_tau_query} shows that \LFCSZT is \emph{$8$ times faster} on average over all $\tau$ values. Its query time is unaffected by $\tau$, unlike that of \LFCSBA. The query time of \LFCSBA increases when $\tau$ decreases, as the DFS it employs prunes a smaller fraction of the subtree of the locus of $P$ in the suffix tree $\ST(T)$.

\paragraph{Index Size.} \Cref{fig:LFCS_SDSL_n_index} shows that our \LFCS index \emph{is $48\%$ smaller} on average over all $n$ values, and \emph{$46\%$ smaller} when $n=10^9$. The reason \LFCSBA occupies more space is that it uses two suffix trees. The size of both indexes scales linearly with $n$, as expected by their complexities.  

\paragraph{Construction Space.} \Cref{fig:LFCS_SDSL_n_rss} shows analogous results to those of \Cref{fig:LFCS_SDSL_n_index}. For example, \LFCSZT needs \emph{$20\%$ less space to be constructed} than \LFCSBA for $n=10^9$.  
\LFCSBA needs more space to be constructed due to its two suffix trees. 

\paragraph{Construction Time.} \Cref{fig:LFCS_SDSL_n_build} shows that \LFCSBA is constructed $4$ times faster than \LFCSZT on average over all $n$ values. 
This is consistent with the complexities of these indexes. The construction time of our \LFCS index was heavily influenced by that of $\ZZT(T)$ (e.g., constructing $\ZZT(T)$ took $85\%$ of the total time for $n=10^9$), which, however,  is the key to its excellent query time performance. 

\subsection{\LCCS}\label{experiments:LCCS}

\paragraph{Query Time.} \Cref{fig:LCCS_qt_N} shows that our \LCCSZT index offers query answering that is \emph{four orders of magnitude faster} than \LCCSBA and its query time is \emph{unaffected by} $N$, as expected by its complexity. 
In fact, \LCCSBA \emph{did not terminate within $24$ hours} for $N\geq 10^6$, as the term $|\occ_{\mathcal{C}}(P)|$ in its complexity is very large (e.g., $2,323$ on average over all queries).   
Thus, in the remainder of the section we used $N=10^5$. 

\Cref{fig:LCCS_qt_m} shows that \LCCSZT offers \emph{four orders of magnitude faster query time} than \LCCSBA on average and its query time grows more slowly than linearly in $m$. \LCCSBA becomes faster as $m$ increases, because  the term $|\occ_{\mathcal{C}}(P)|$ in its complexity becomes smaller as $m$ increases. 

\Cref{fig:LCCS_qt_tau} shows that \LCCSZT offers \emph{four orders of magnitude faster query time} on average. Furthermore, its query time is unaffected by $\tau$, unlike that of \LCCSBA. In this experiment, $\tau$ did not affect the query time of \LCCSBA, as the tested values were too small for DFS pruning to have a noticeable impact.

\paragraph{Index Size.} \Cref{fig:LCCS_is} shows that \LCCSZT \emph{is $48\%$ smaller} than \LCCSBA on average over all $N$ values, and \emph{$48\%$ smaller when $N= 10^9$}. The reason is the use of \ZZT; the suffix trees in \LCCSBA occupy much more  space. The size of both indexes scales linearly with $N$, as expected by their complexities.   

\paragraph{Construction Space.} \Cref{fig:LCCS_cs} shows  analogous results to those of \Cref{fig:LCCS_is}. For example, \LCCSZT needs \emph{$7\%$ less space to be constructed} than \LCCSBA for $N=10^9$. The larger construction space of \LCCSBA is again due to  its suffix trees. 

\paragraph{Construction Time.} \Cref{fig:LCCS_ct} shows that \LCCSBA is constructed two times faster than \LCCSZT on average over all $N$ values, as expected by the complexities of these indexes. Recall, however, that \LCCSBA is prohibitively expensive to query for realistically large string collections. The bottleneck for \LCCSZT is the construction of \ZZT. 

\begin{figure}[t]
  \centering
  \begin{subfigure}[t]{0.16\linewidth}
    \includegraphics[width=\linewidth]{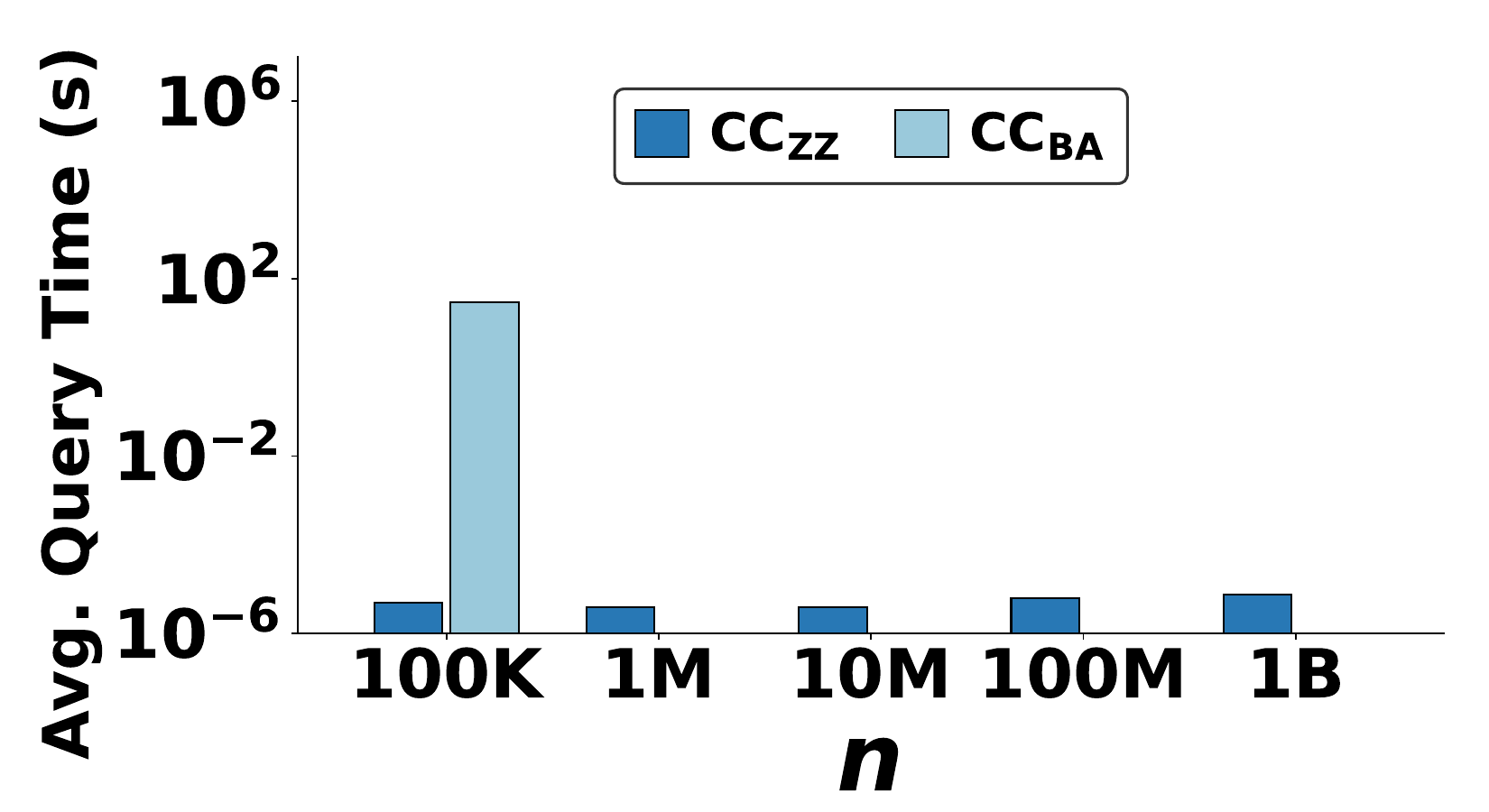}
    \caption{Query time vs. $n$}\label{cc_qt_n}
  \end{subfigure}\hfill
  \begin{subfigure}[t]{0.16\linewidth}
    \includegraphics[width=\linewidth]{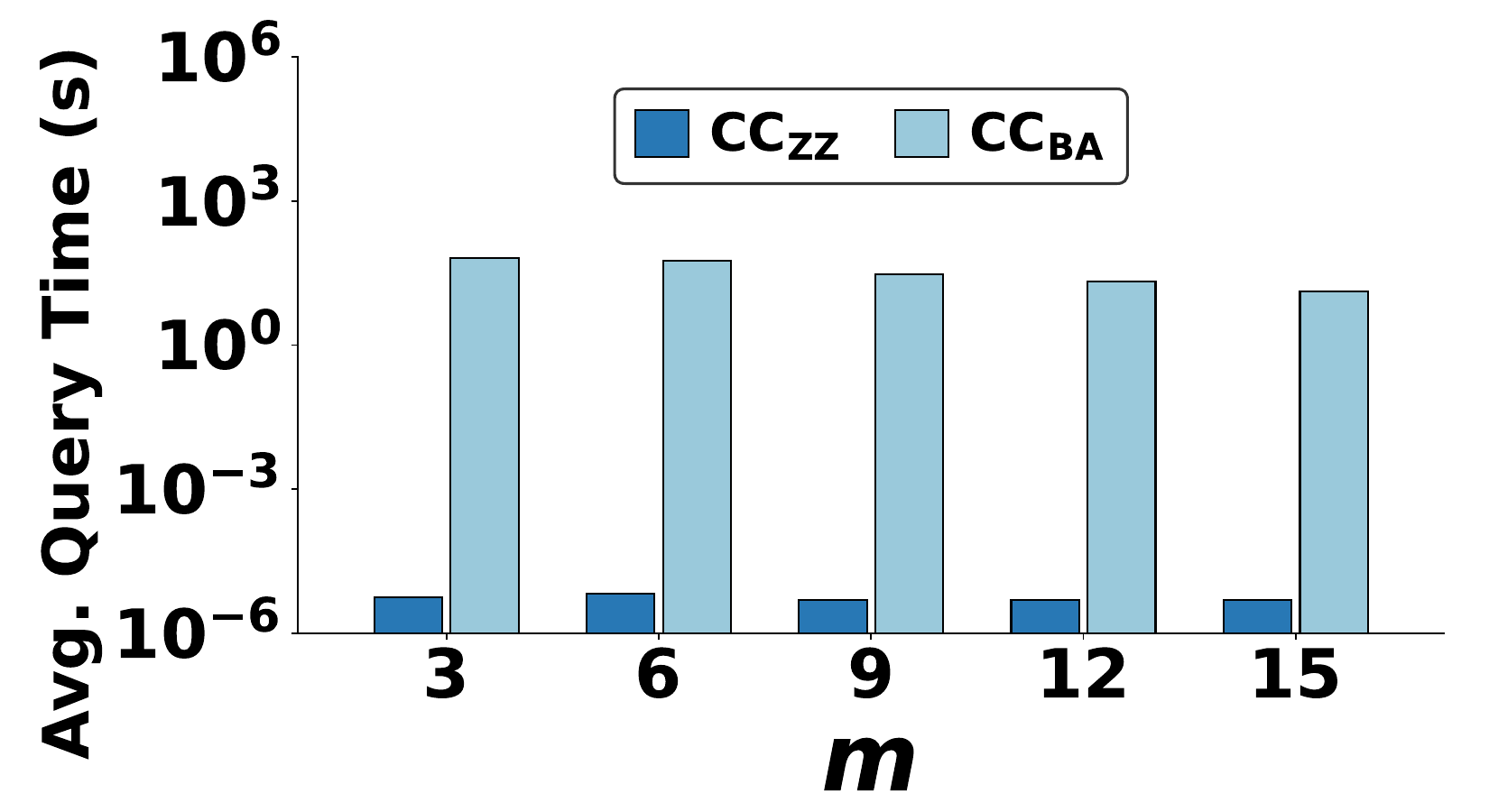}
    \caption{Query time vs. $m$}\label{cc_qt_m}
  \end{subfigure}\hfill
  \begin{subfigure}[t]{0.16\linewidth}
    \includegraphics[width=\linewidth]{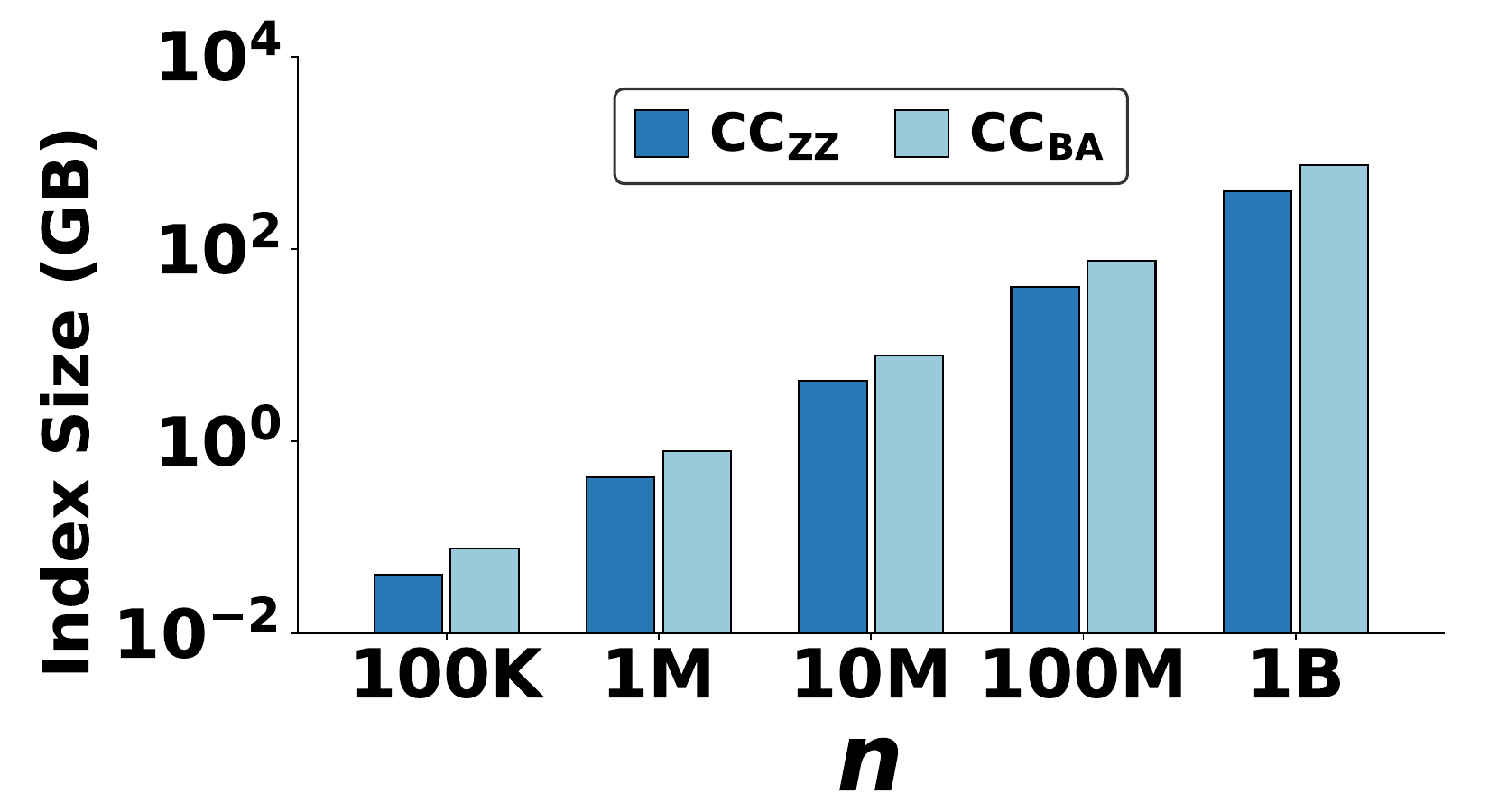}
    \caption{Index size vs. $n$}\label{cc_is}
  \end{subfigure}\hfill
  \begin{subfigure}[t]{0.16\linewidth}
    \includegraphics[width=\linewidth]{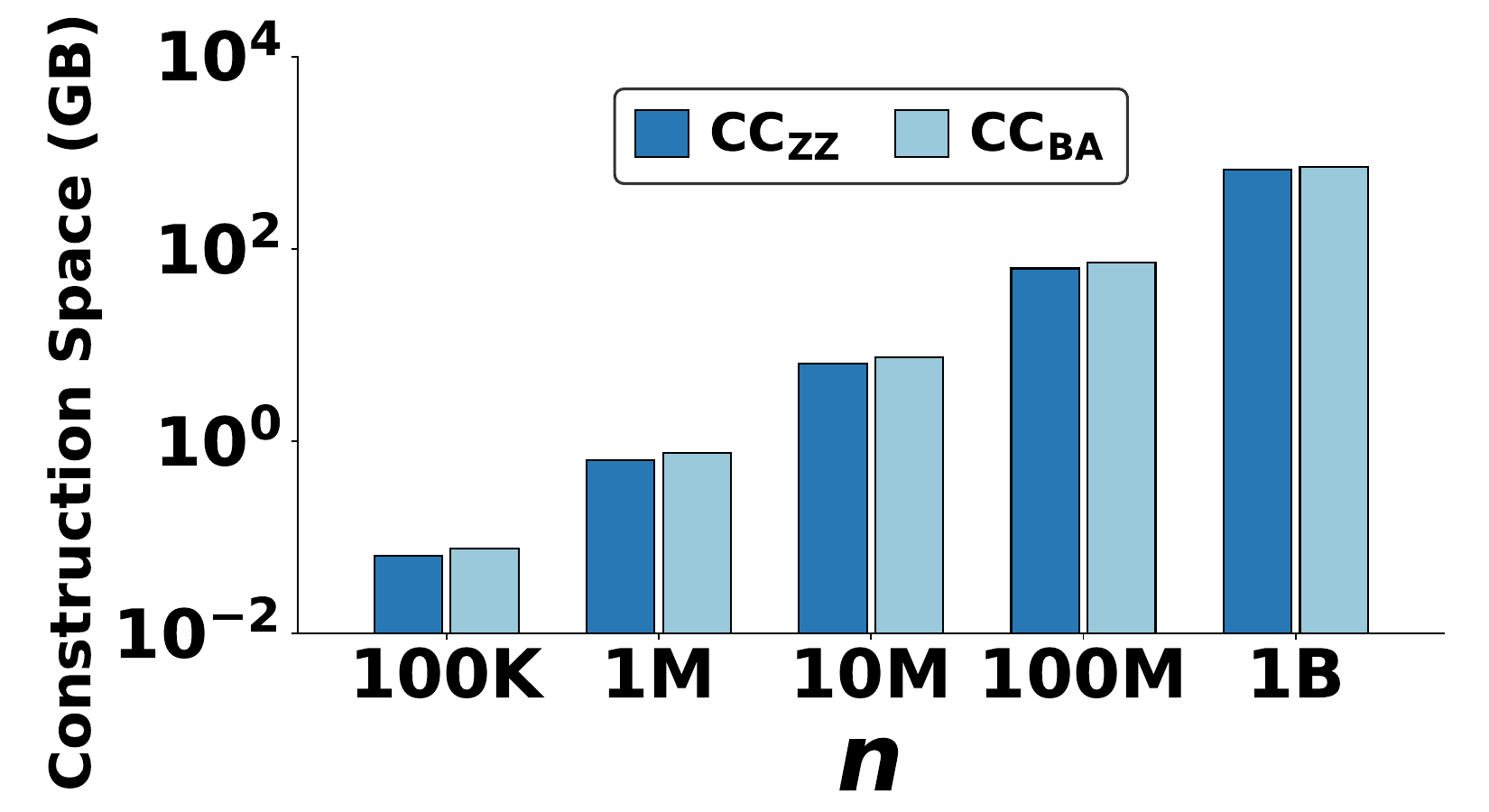}
    \caption{Constr.\ space vs. $n$}\label{cc_cs}
  \end{subfigure}\hfill
  \begin{subfigure}[t]{0.16\linewidth}
    \includegraphics[width=\linewidth]{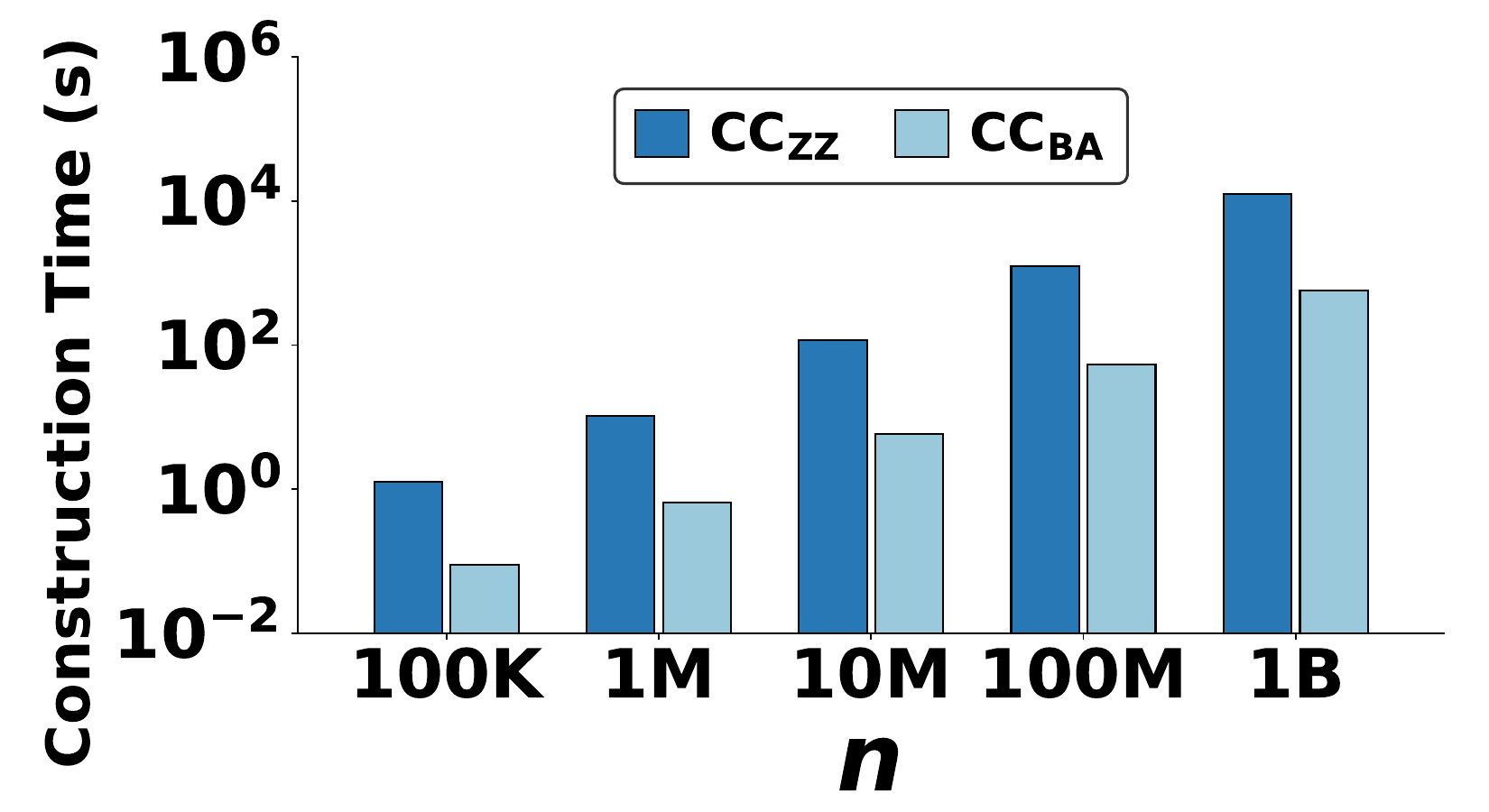}
    \caption{Constr.\ time vs. $n$}\label{cc_ct}
  \end{subfigure}
  \vspace{\captionspacing}
  \caption{Our \CC index vs. \CCBA on \bst. In \Cref{cc_qt_n}, \CCBA did not terminate within $24$ hours for $n\geq  10^6$.}
\end{figure}

\begin{figure}[t]
  \centering
  \begin{subfigure}[t]{0.16\linewidth}
    \includegraphics[width=\linewidth]{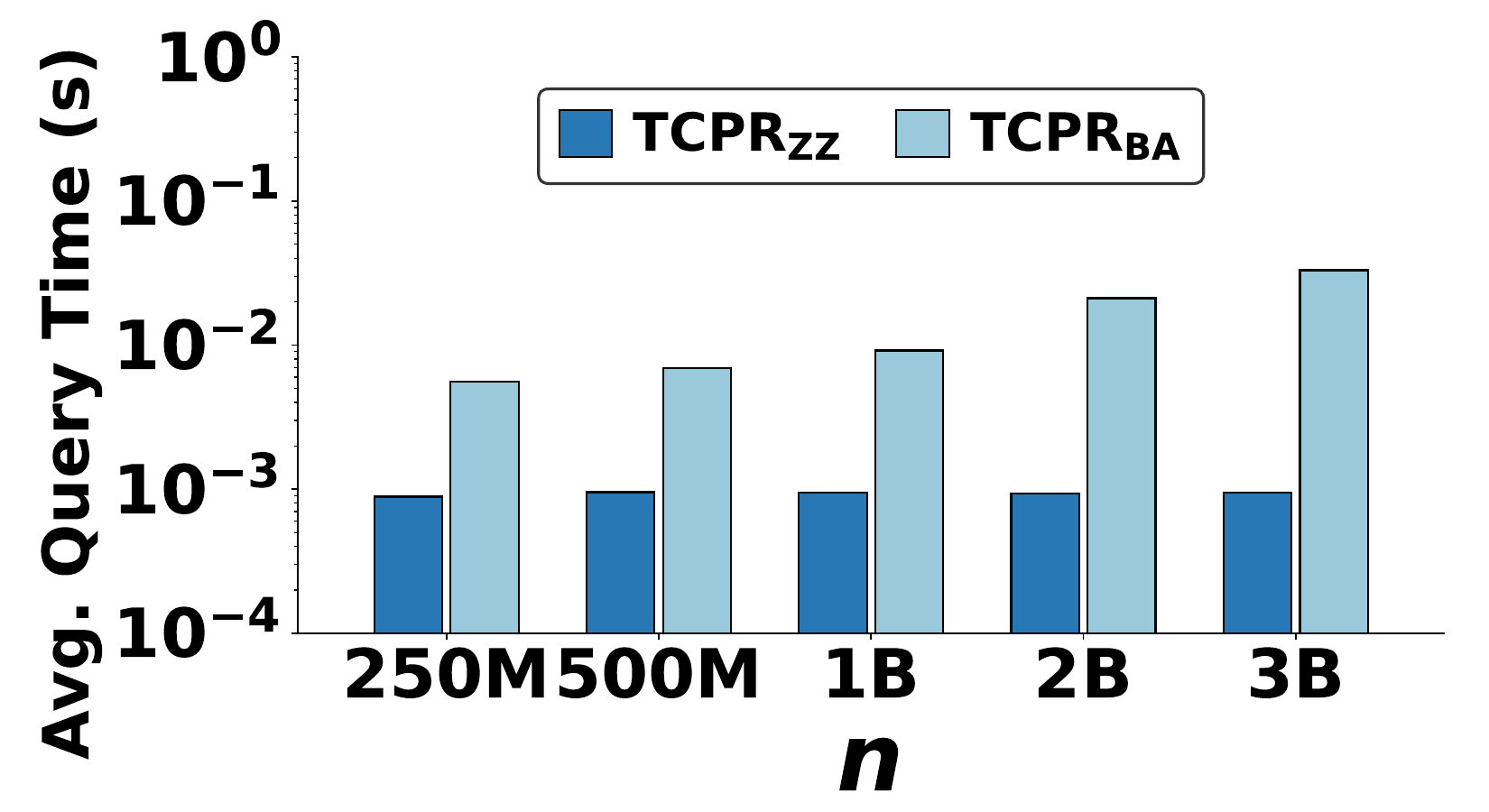}
    \caption{Query time vs. $n$}\label{fig:qt_tcpr_n1}
  \end{subfigure}\hfill
  \begin{subfigure}[t]{0.16\linewidth}
    \includegraphics[width=\linewidth]{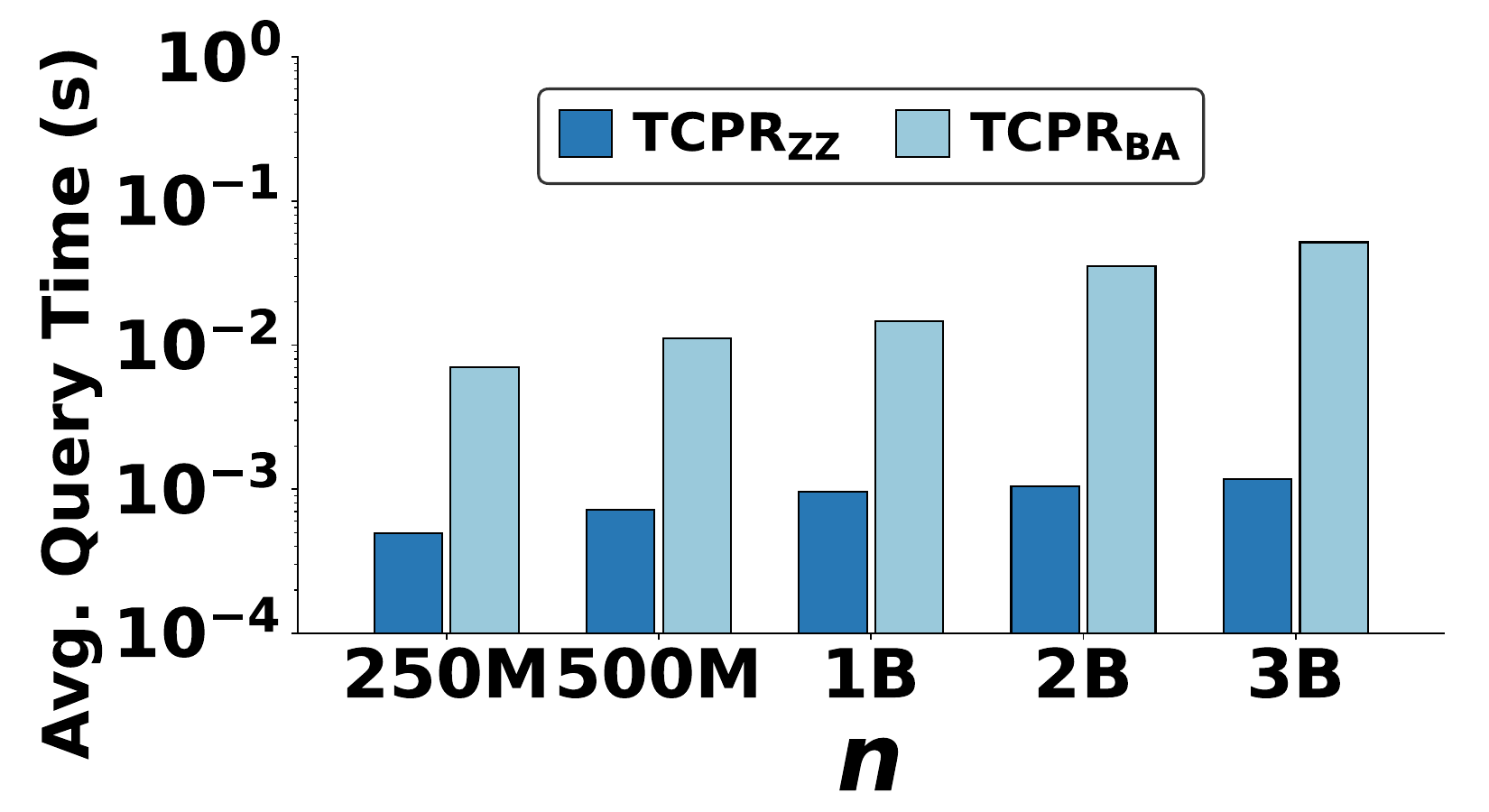}
    \caption{Query time vs. $n$}\label{fig:qt_tcpr_n2}
  \end{subfigure}\hfill
  \begin{subfigure}[t]{0.16\linewidth}
    \includegraphics[width=\linewidth]{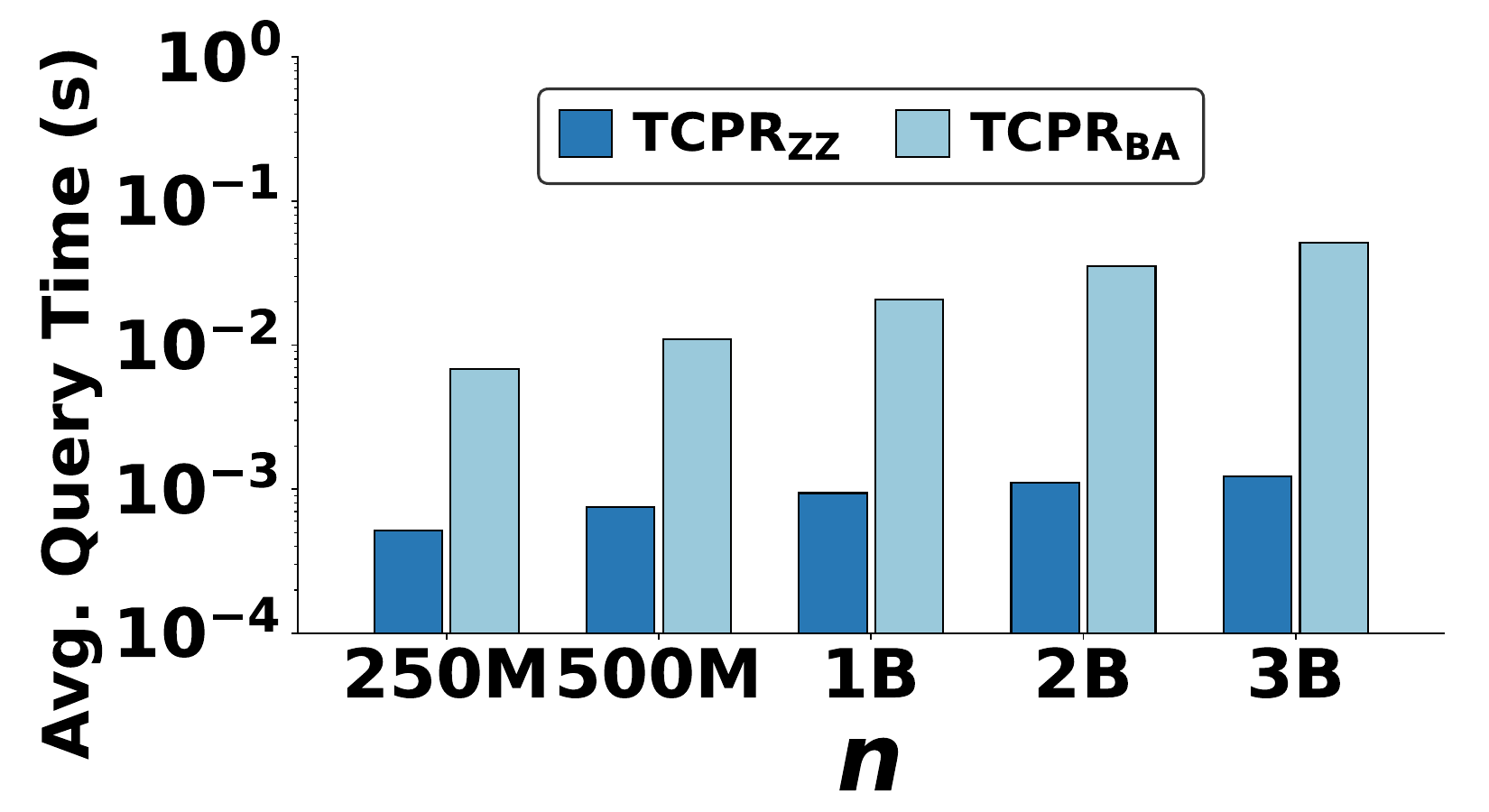}
    \caption{Query time vs. $n$}\label{fig:qt_tcpr_n3}
  \end{subfigure}
  \begin{subfigure}[t]{0.16\linewidth}
    \includegraphics[width=\linewidth]{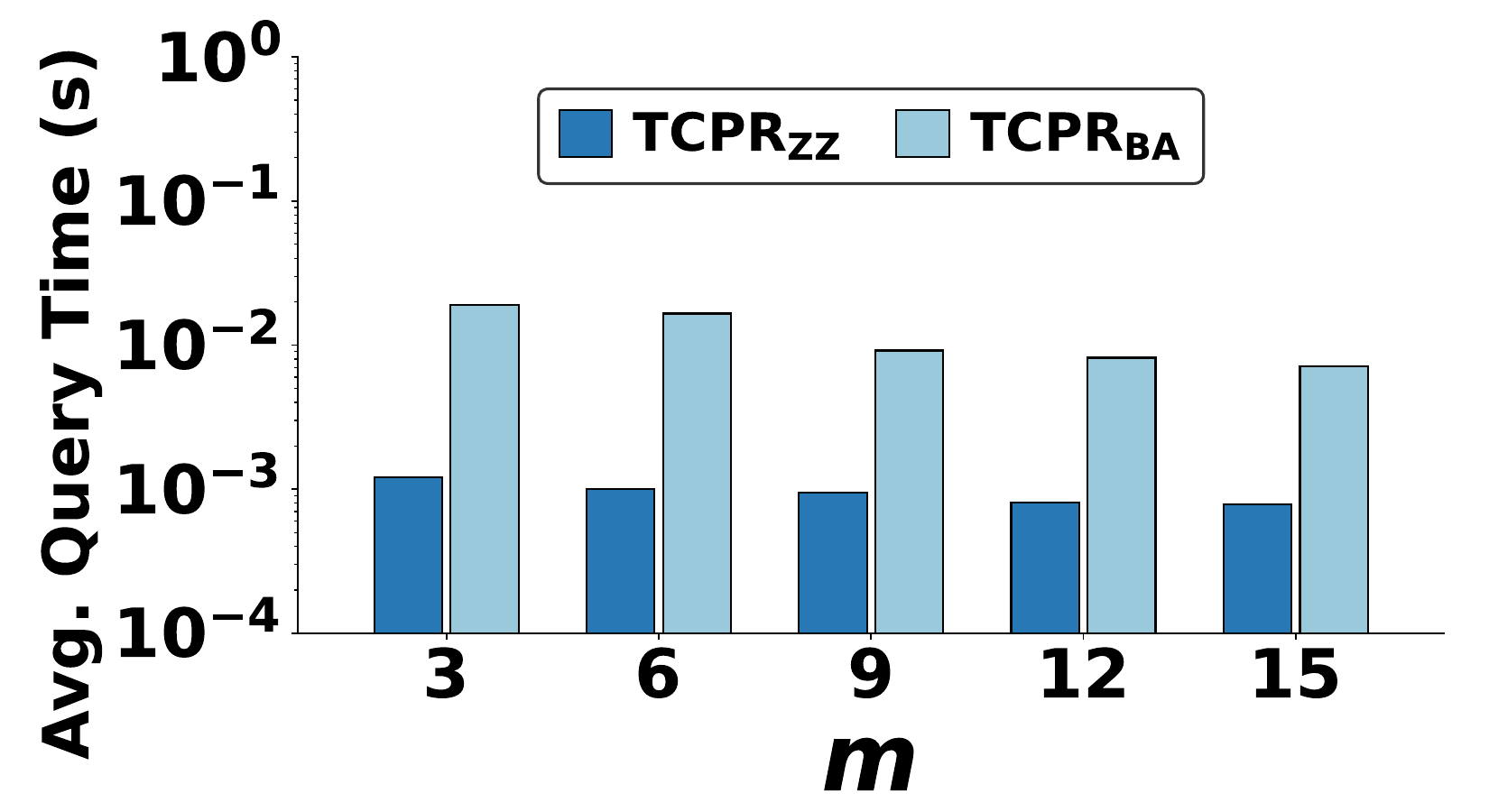}
    \caption{Query time vs. $m$}\label{fig:qt_tcpr_m1}
  \end{subfigure}\hfill
  \begin{subfigure}[t]{0.16\linewidth}
    \includegraphics[width=\linewidth]{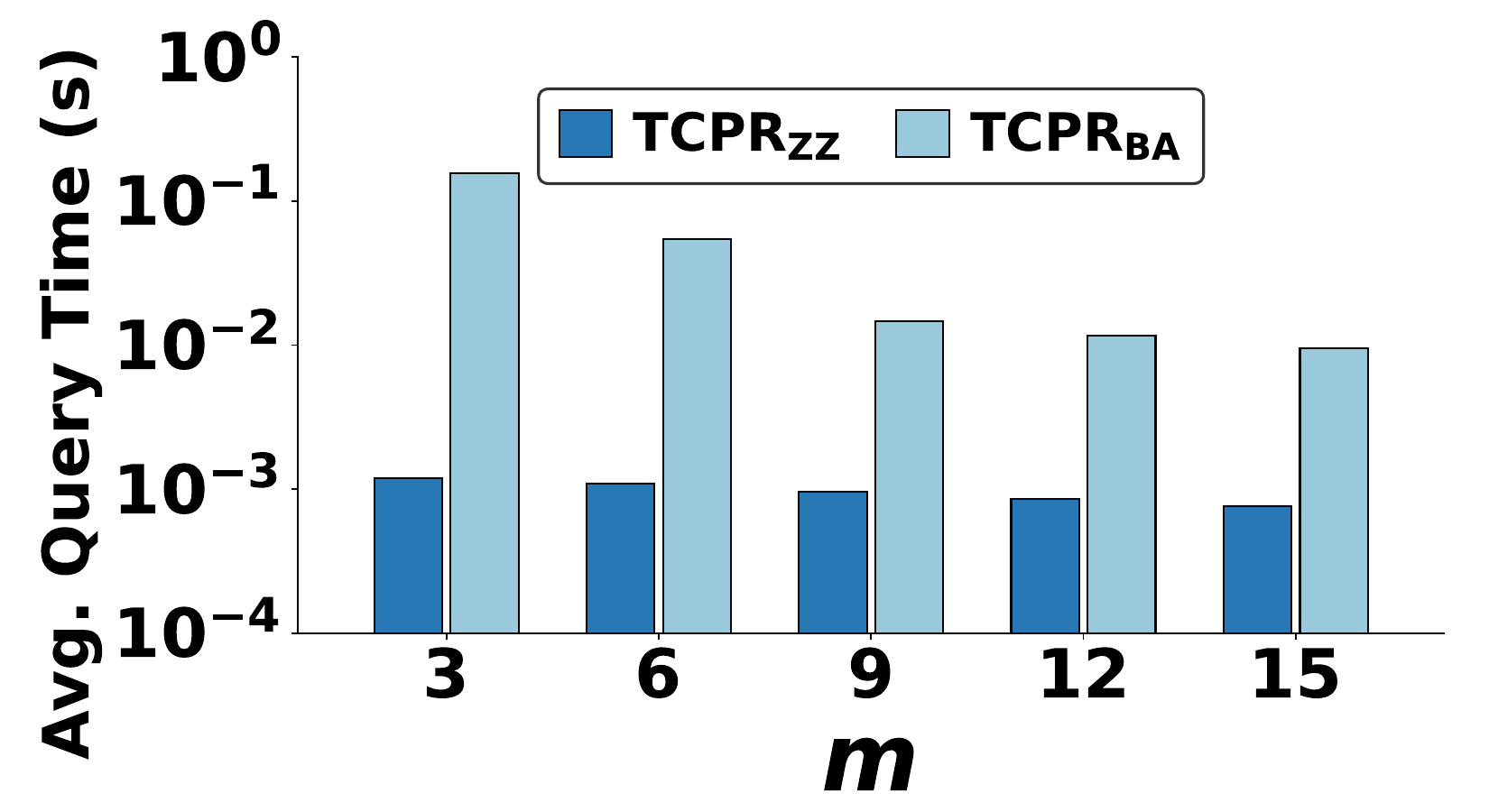}
    \caption{Query time vs. $m$}\label{fig:qt_tcpr_m2}
  \end{subfigure}\hfill
  \begin{subfigure}[t]{0.16\linewidth}
    \includegraphics[width=\linewidth]{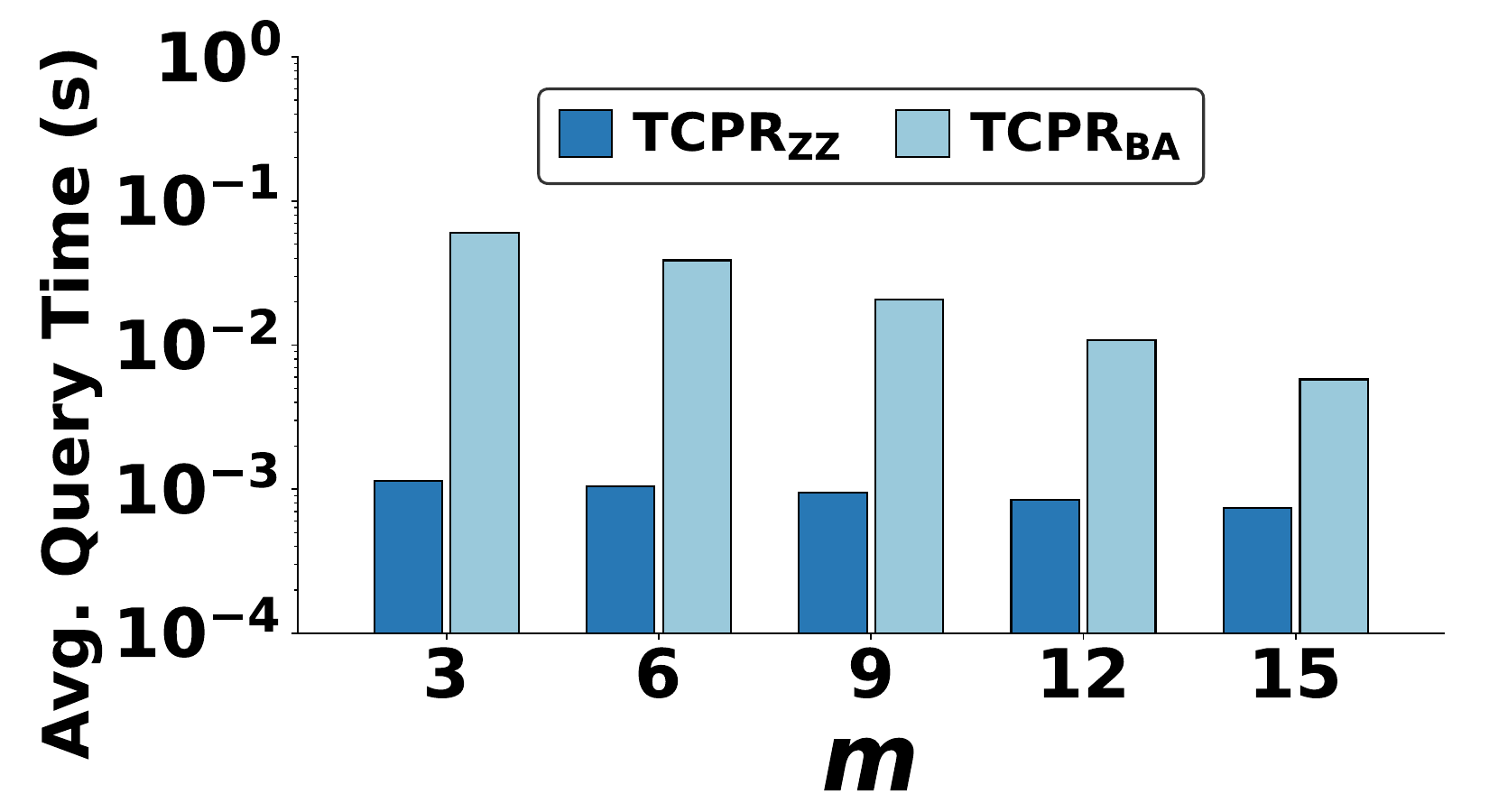}
    \caption{Query time vs. $m$}\label{fig:qt_tcpr_m3}
  \end{subfigure}
  \vspace{\captionspacing}
  \caption{Query time of our \TCPR index vs. the baseline on \chr using the (a) \textsf{TF}, (b) \textsf{SP}, and (c) \textsf{TP} scoring function vs. $n$; or using (d) \textsf{TF}, (e) \textsf{SP}, and (f) \textsf{TP} scoring function vs. $m$.}\label{TCPR_query_nm}

  \centering
  \begin{subfigure}[t]{0.16\linewidth}
    \includegraphics[width=\linewidth]{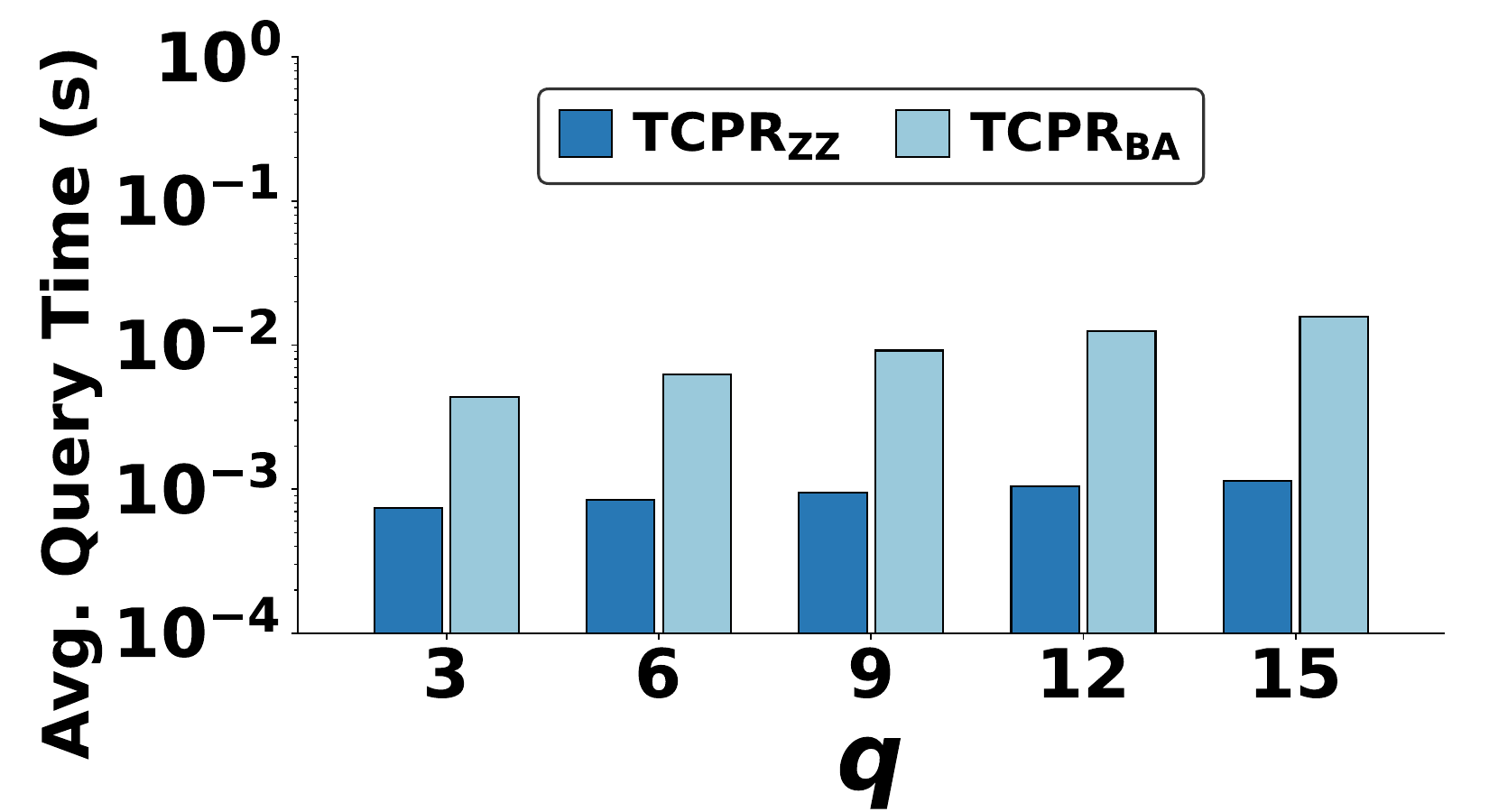}
    \caption{Query time vs. $q$}\label{fig:qt_tcpr_q1}
  \end{subfigure}\hfill
  \begin{subfigure}[t]{0.16\linewidth}
    \includegraphics[width=\linewidth]{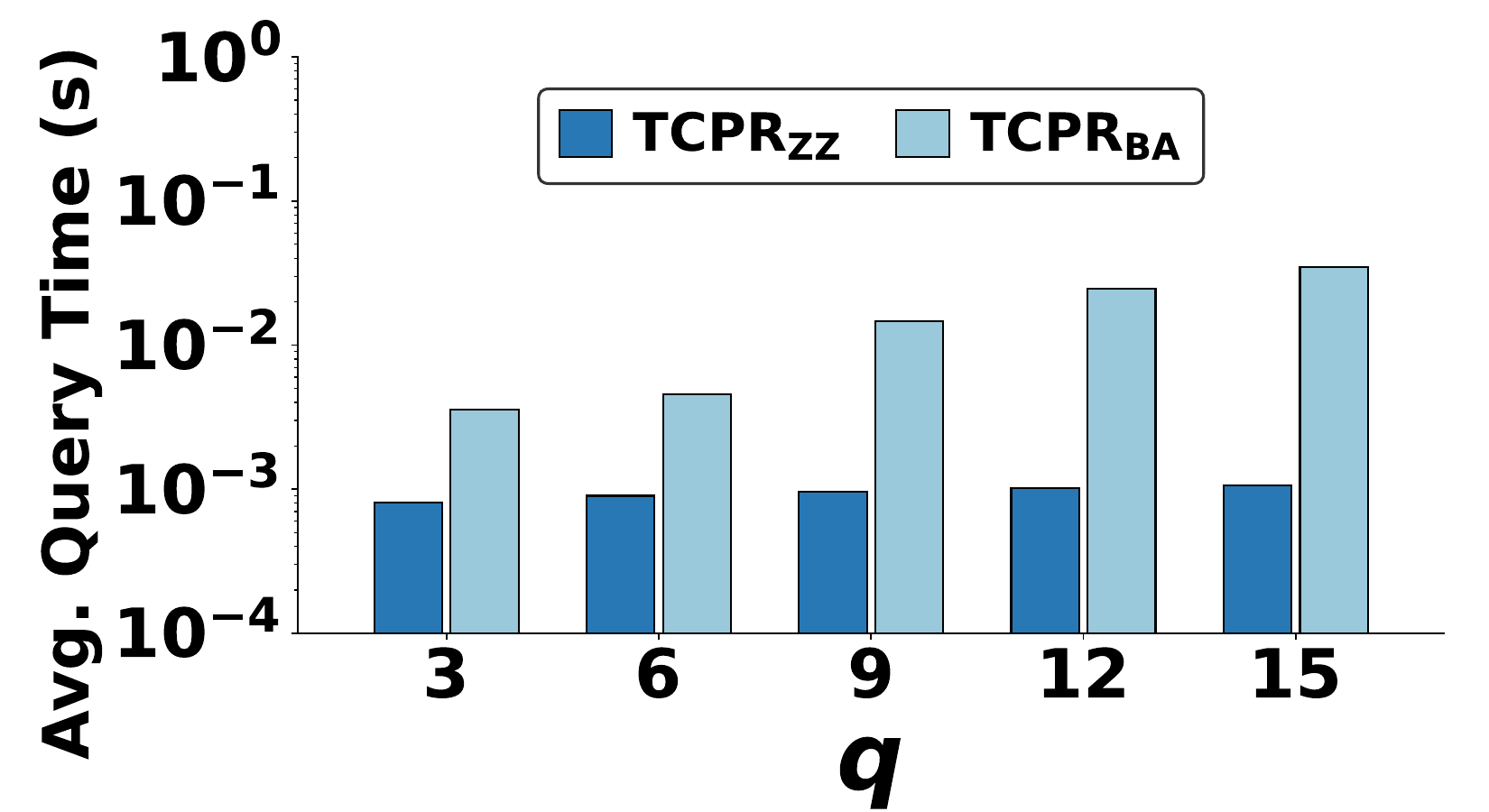}
    \caption{Query time vs. $q$}\label{fig:qt_tcpr_q2}
  \end{subfigure}\hfill
  \begin{subfigure}[t]{0.16\linewidth}
    \includegraphics[width=\linewidth]{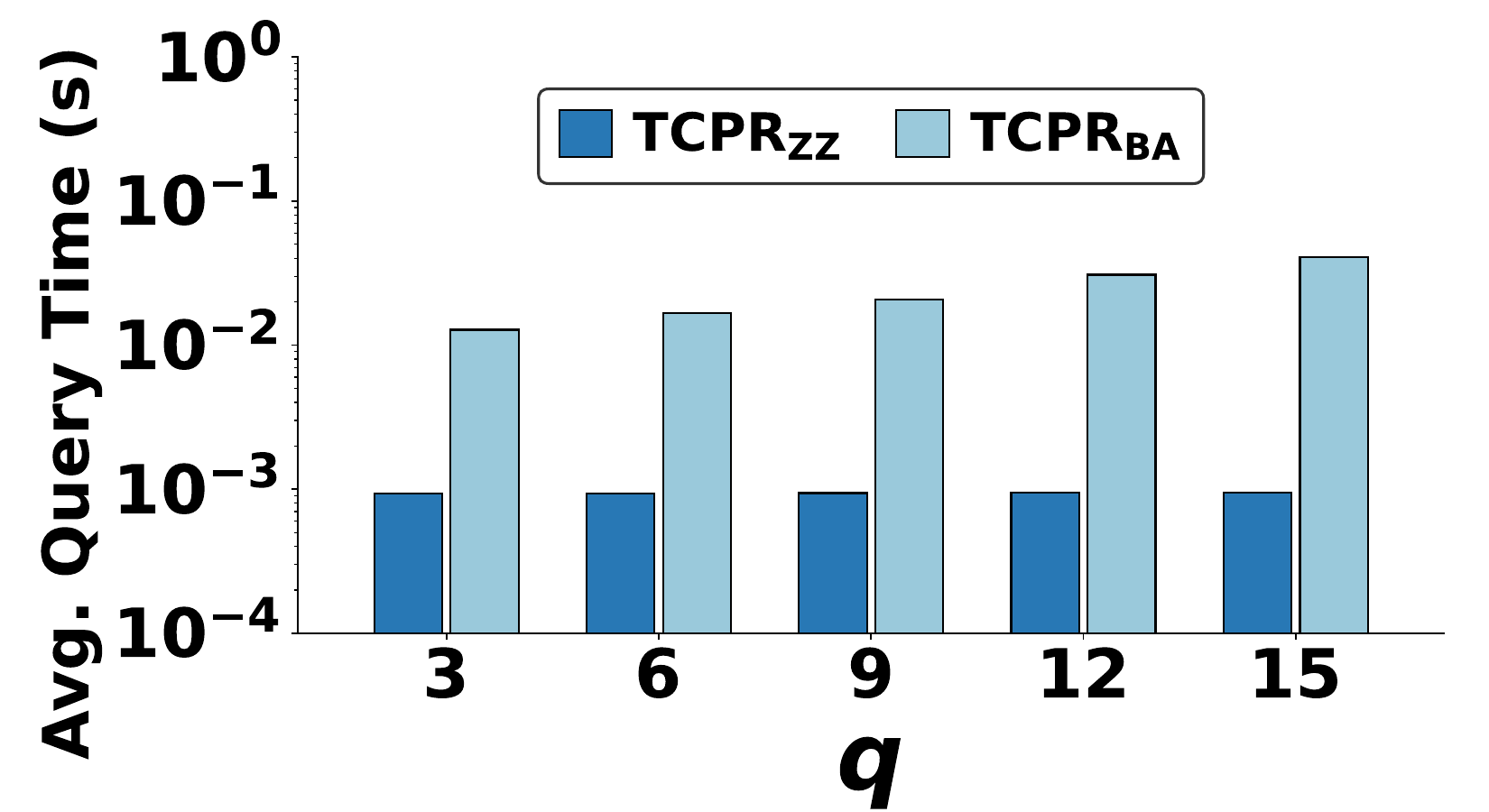}
    \caption{Query time vs. $q$}\label{fig:qt_tcpr_q3}
  \end{subfigure}
  \begin{subfigure}[t]{0.16\linewidth}
    \includegraphics[width=\linewidth]{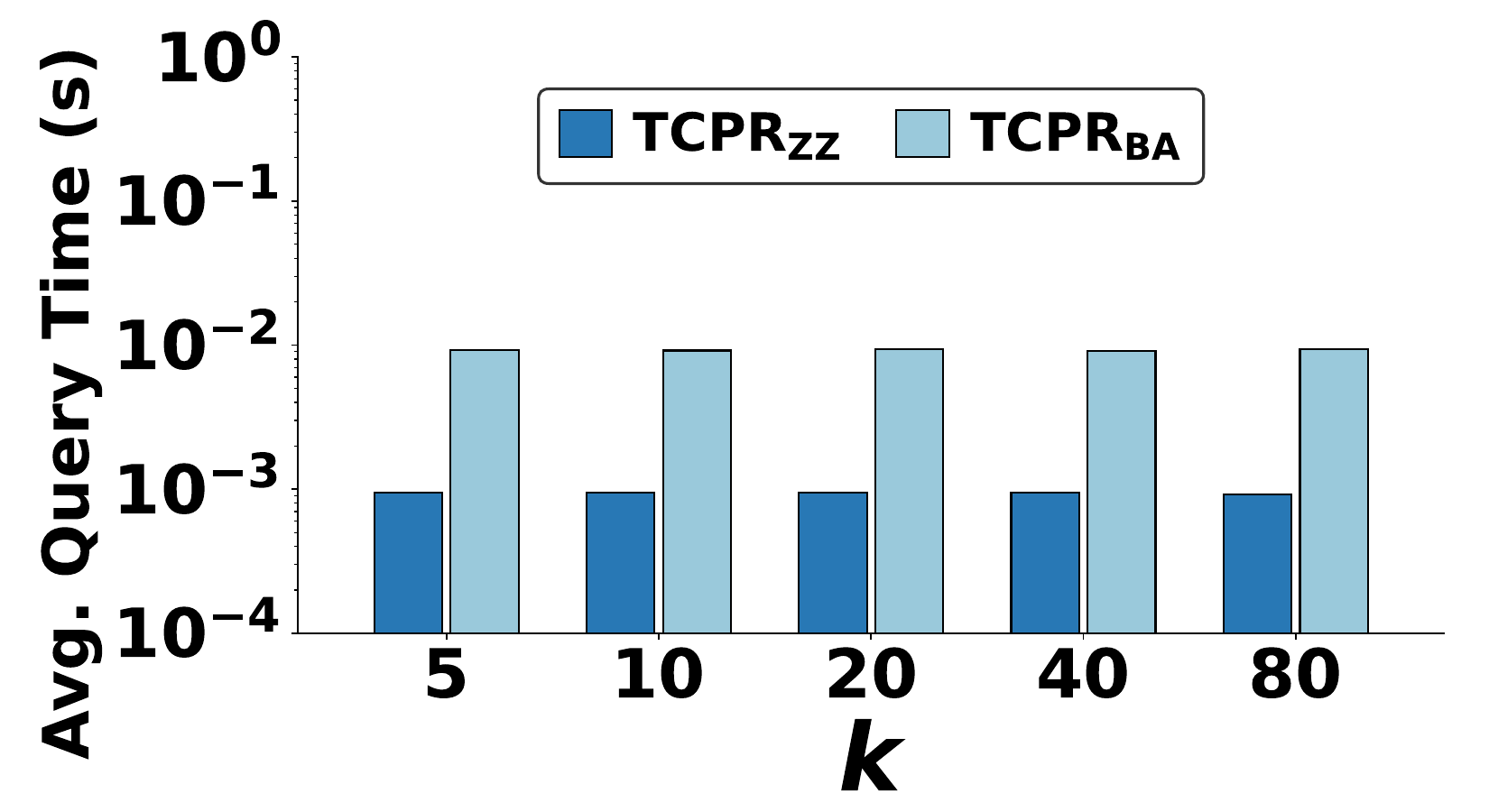}
    \caption{Query time vs. $k$}\label{fig:qt_tcpr_k1}
  \end{subfigure}\hfill
  \begin{subfigure}[t]{0.16\linewidth}
    \includegraphics[width=\linewidth]{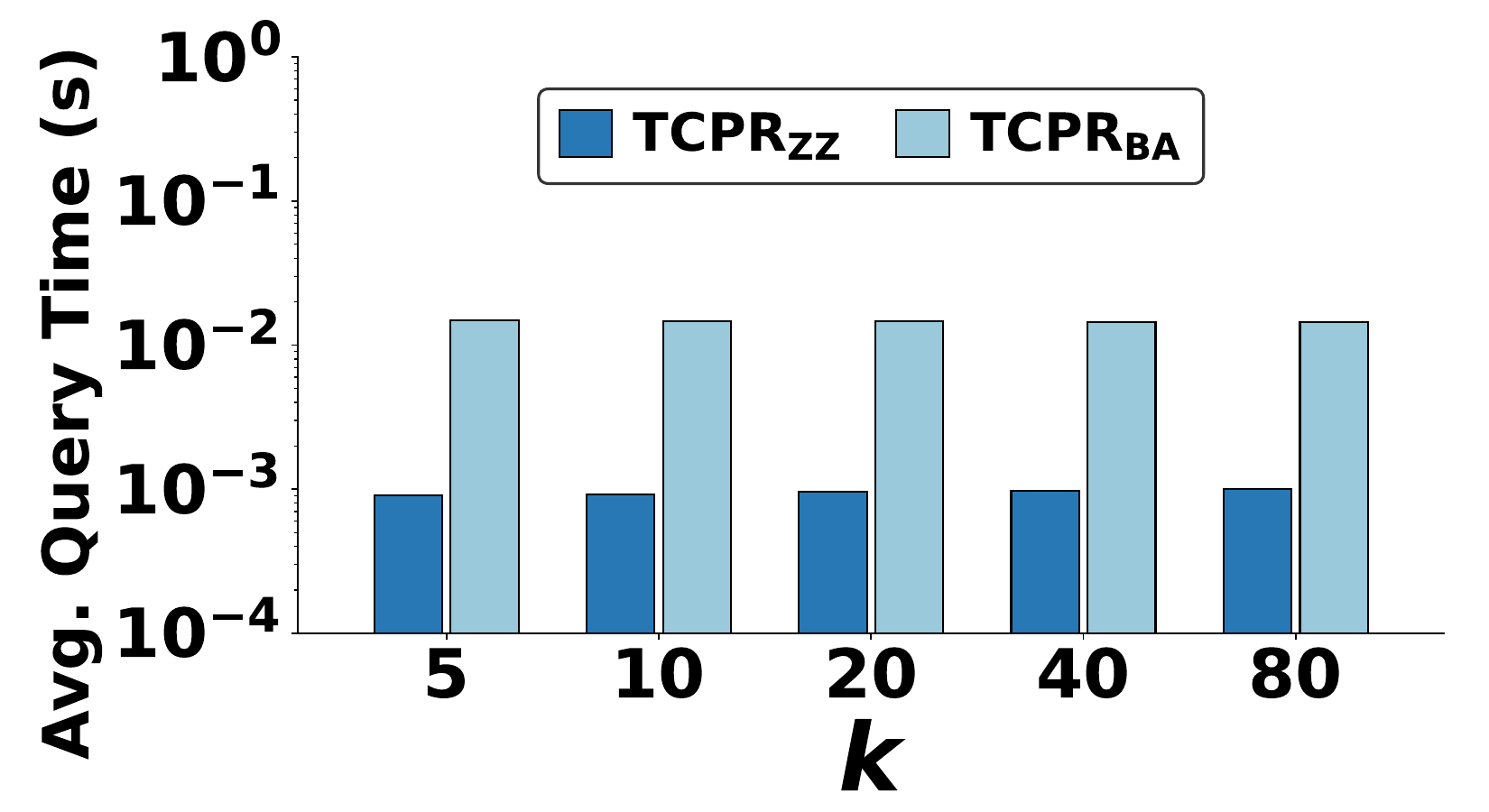}
    \caption{Query time vs. $k$}\label{fig:qt_tcpr_k2}
  \end{subfigure}\hfill
  \begin{subfigure}[t]{0.16\linewidth}
    \includegraphics[width=\linewidth]{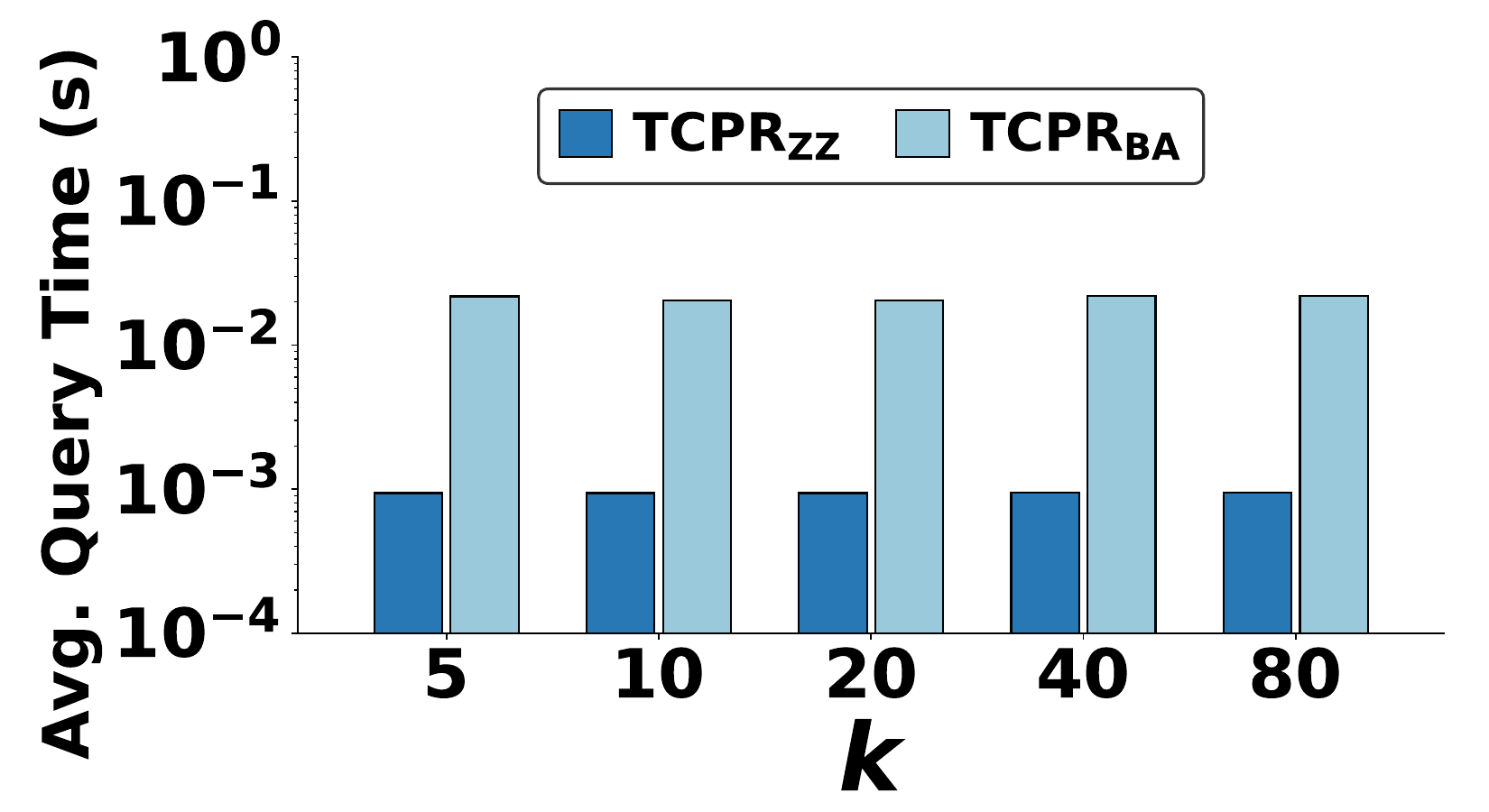}
    \caption{Query time vs. $k$}\label{fig:qt_tcpr_k3}
  \end{subfigure}
  \vspace{\captionspacing}
  \vspace{+2mm}
  \caption{Query time of our \TCPR index vs. the baseline on \chr using the (a) \textsf{TF}, (b) \textsf{SP}, and (c) \textsf{TP} scoring function vs. $q$; or using (d) \textsf{TF}, (e) \textsf{SP}, and (f) \textsf{TP} scoring function vs. $k$.}\label{TCPR_query_qk}

  \centering
  \begin{subfigure}[t]{0.16\linewidth}
    \includegraphics[width=\linewidth]{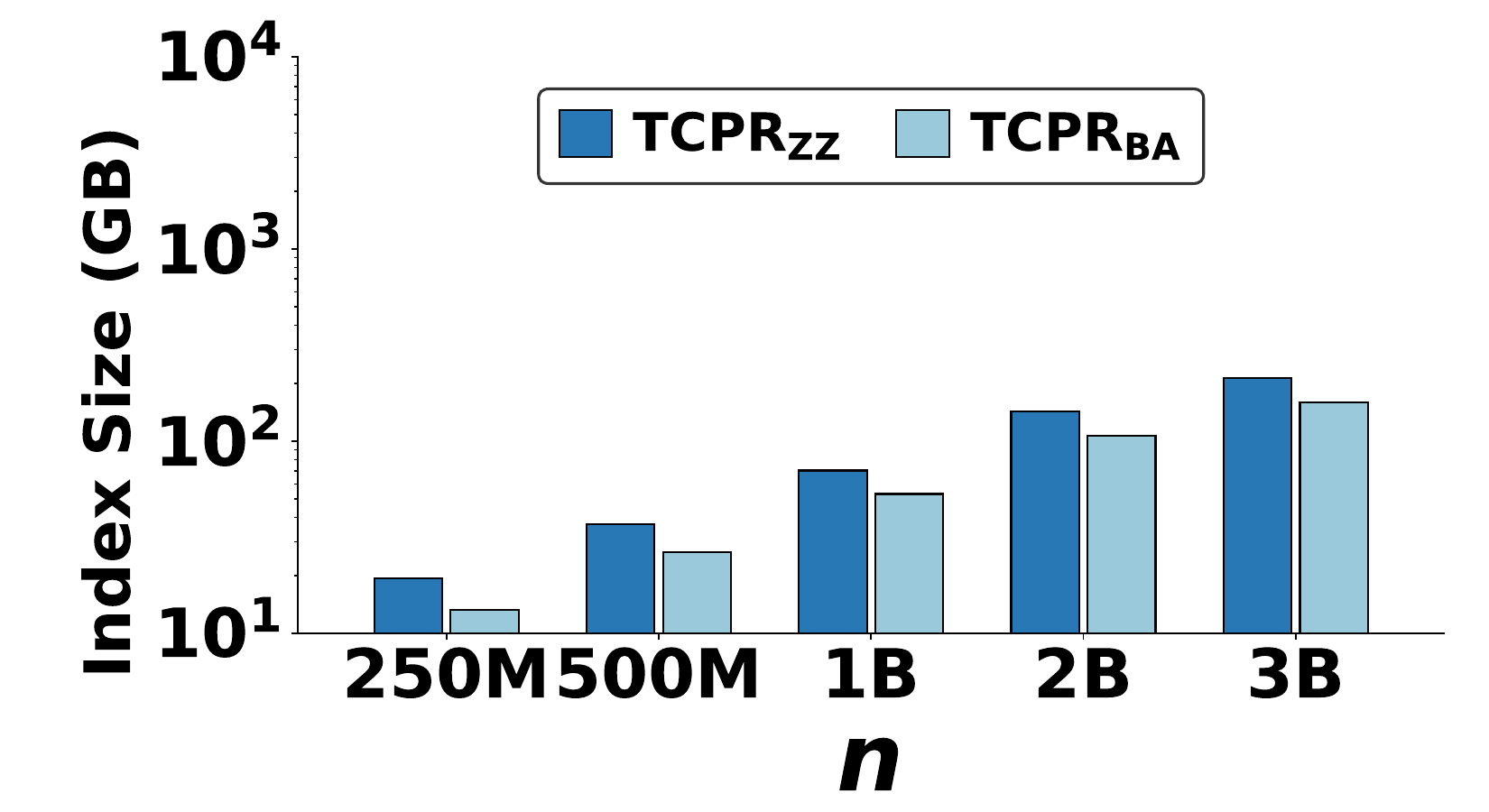}
    \caption{Index size vs. $n$}\label{fig:tcpr_is_n1}
  \end{subfigure}\hfill
  \begin{subfigure}[t]{0.16\linewidth}
    \includegraphics[width=\linewidth]{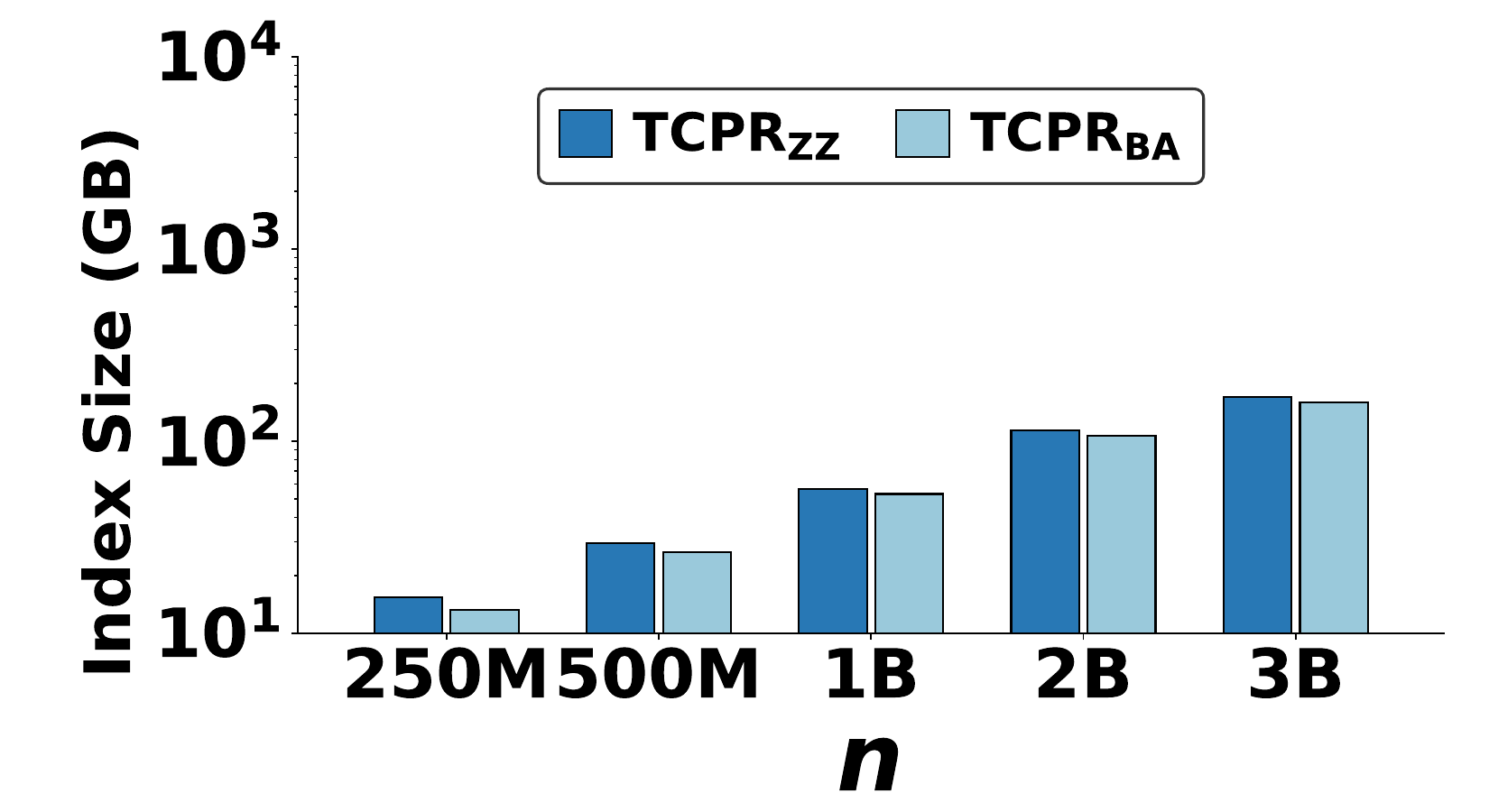}
    \caption{Index size vs. $n$}\label{fig:tcpr_is_n2}
  \end{subfigure}\hfill
  \begin{subfigure}[t]{0.16\linewidth}
    \includegraphics[width=\linewidth]{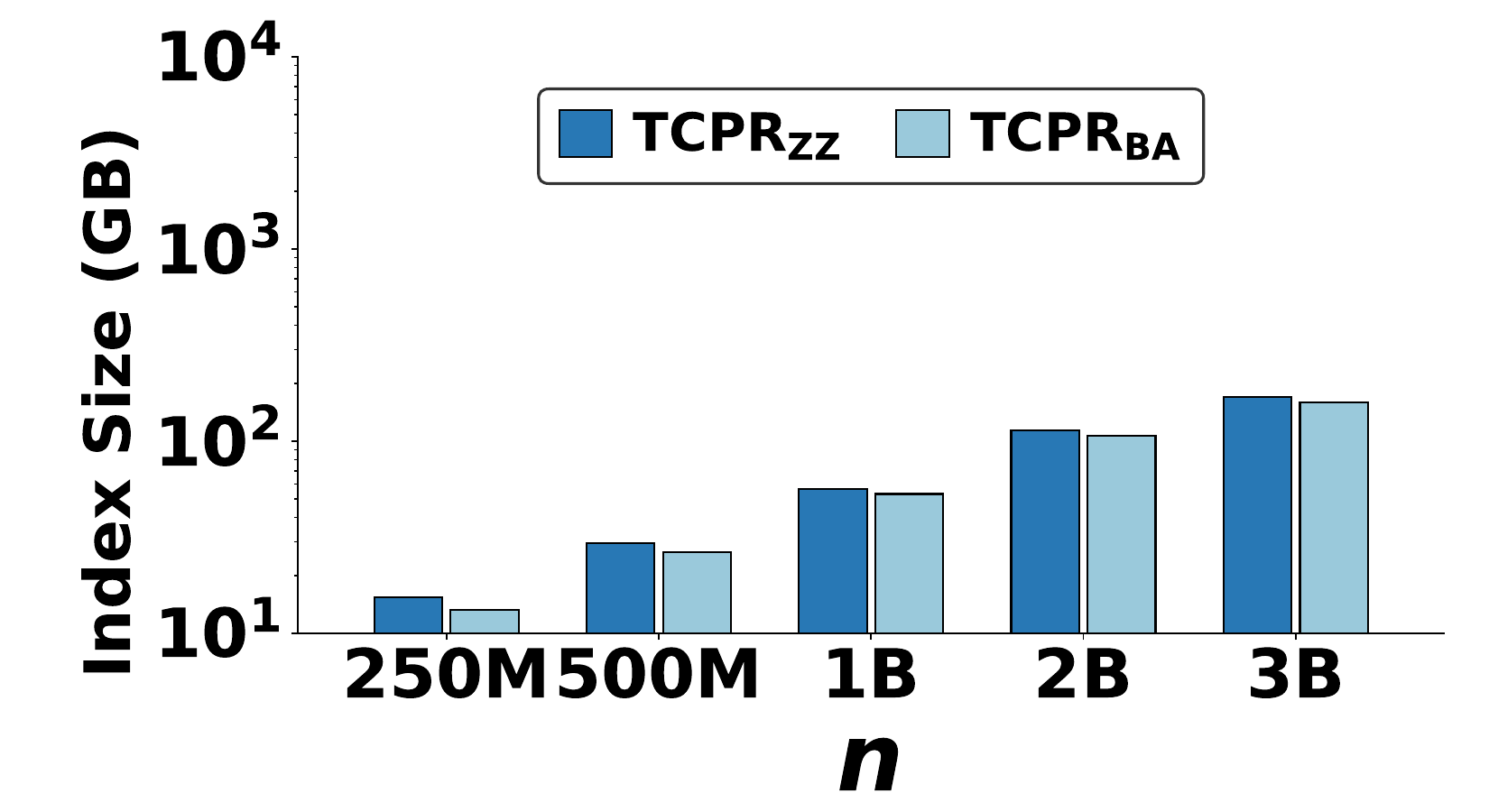}
    \caption{Index size vs. $n$}\label{fig:tcpr_is_n3}
  \end{subfigure}
  \begin{subfigure}[t]{0.16\linewidth}
    \includegraphics[width=\linewidth]{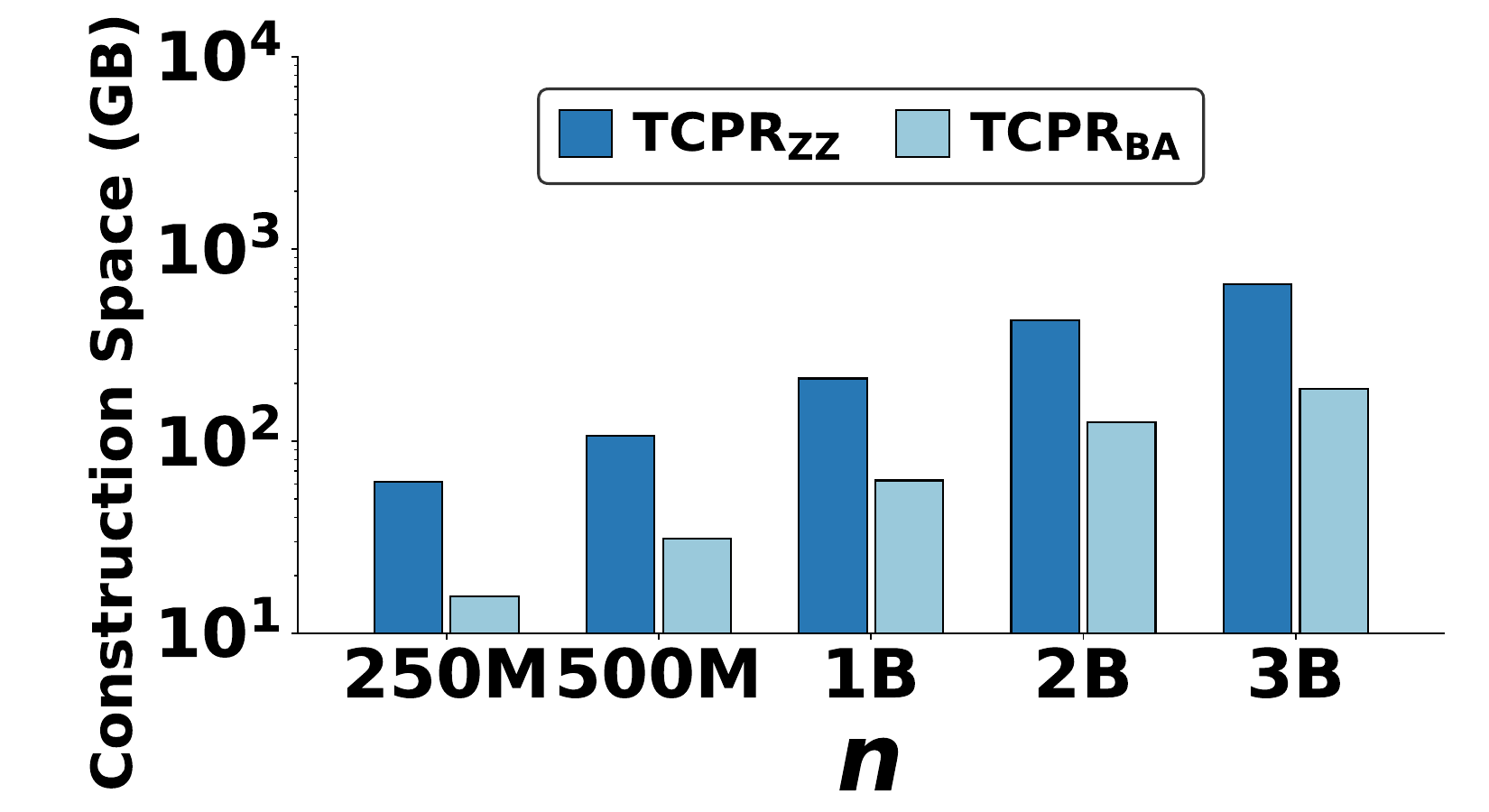}
    \caption{Constr. space vs. $n$}\label{fig:tcpr_cs_n1}
  \end{subfigure}\hfill
  \begin{subfigure}[t]{0.16\linewidth}
    \includegraphics[width=\linewidth]{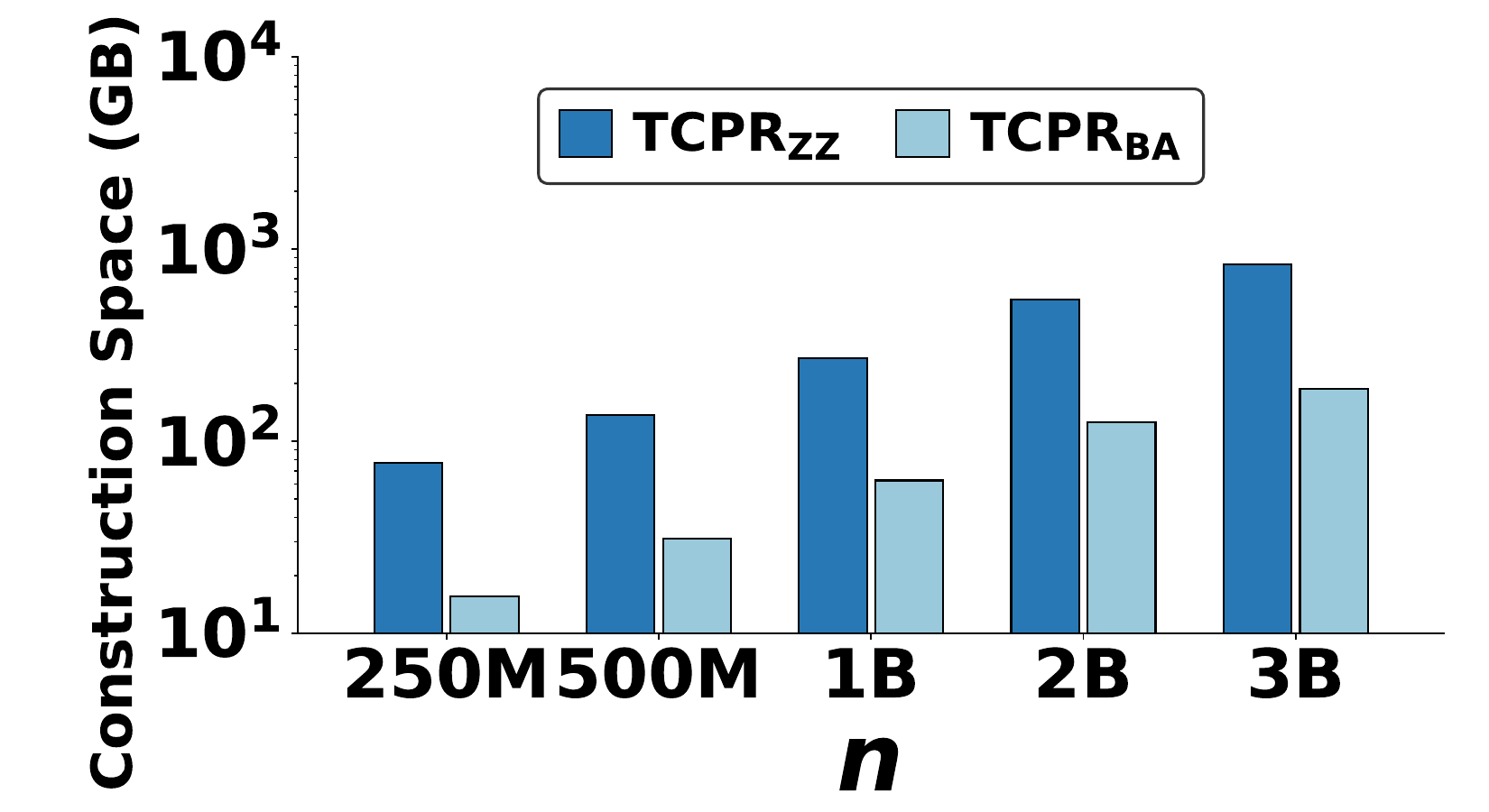}
    \caption{Constr. space vs. $n$}\label{fig:tcpr_cs_n2}
  \end{subfigure}\hfill
  \begin{subfigure}[t]{0.16\linewidth}
    \includegraphics[width=\linewidth]{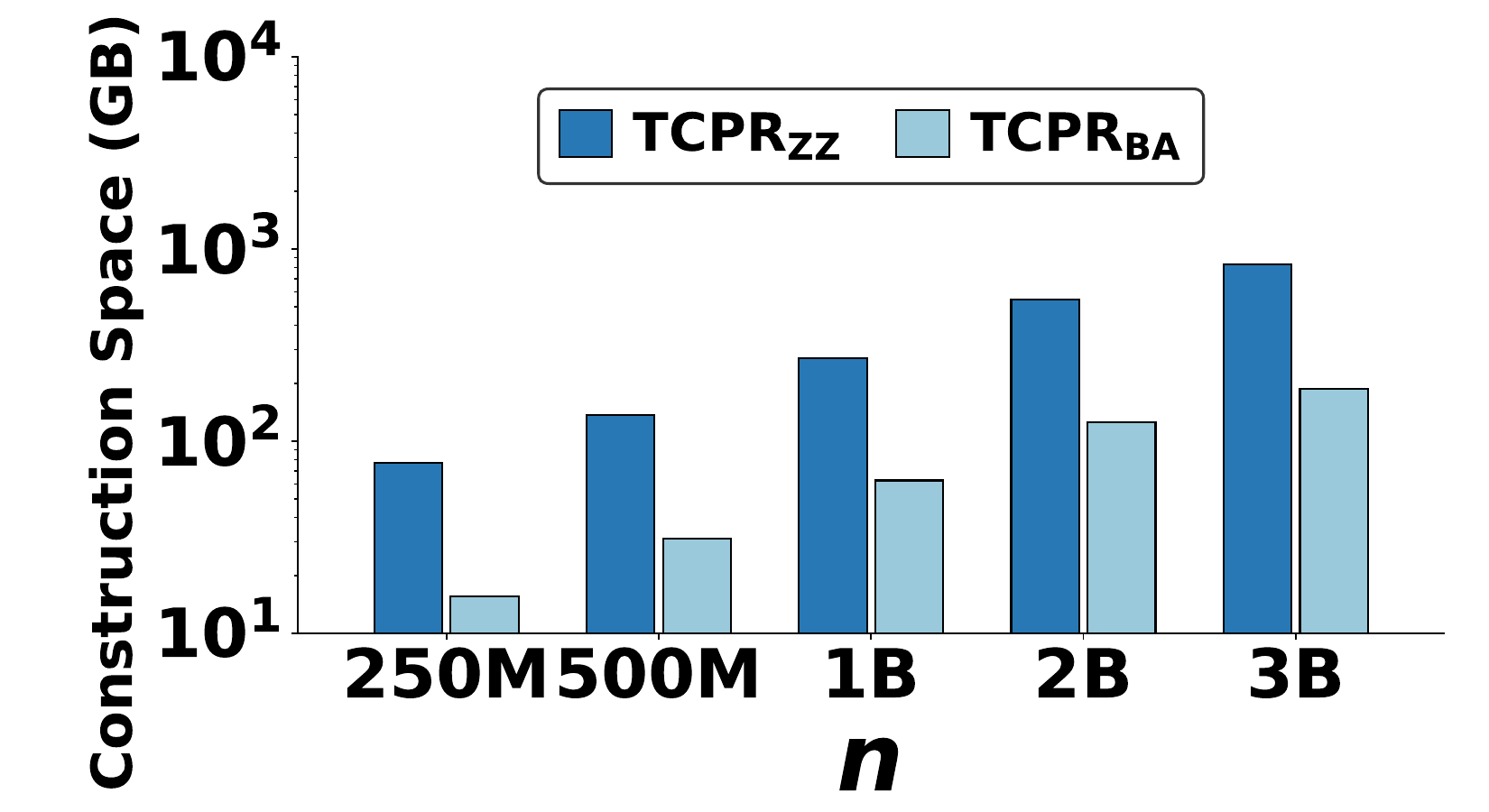}
    \caption{Constr. space vs. $n$}\label{fig:tcpr_cs_n3}
  \end{subfigure}
   \vspace{\captionspacing}
   \vspace{+2mm}
  \caption{Index size of our \TCPR index vs. the baseline on \chr using the (a) \textsf{TF}, (b) \textsf{SP}, and (c) \textsf{TP} scoring function vs. $n$; construction space of our \TCPR index vs. the baseline on \chr using the (d) \textsf{TF}, (e) \textsf{SP}, and (f) \textsf{TP} scoring function vs. $n$.
  }\label{TCPR_contru}
 \vspace{\captionspacing}
\end{figure}

\begin{figure}[t]
  \centering
  \begin{subfigure}[t]{0.20\linewidth}
    \includegraphics[width=\linewidth]{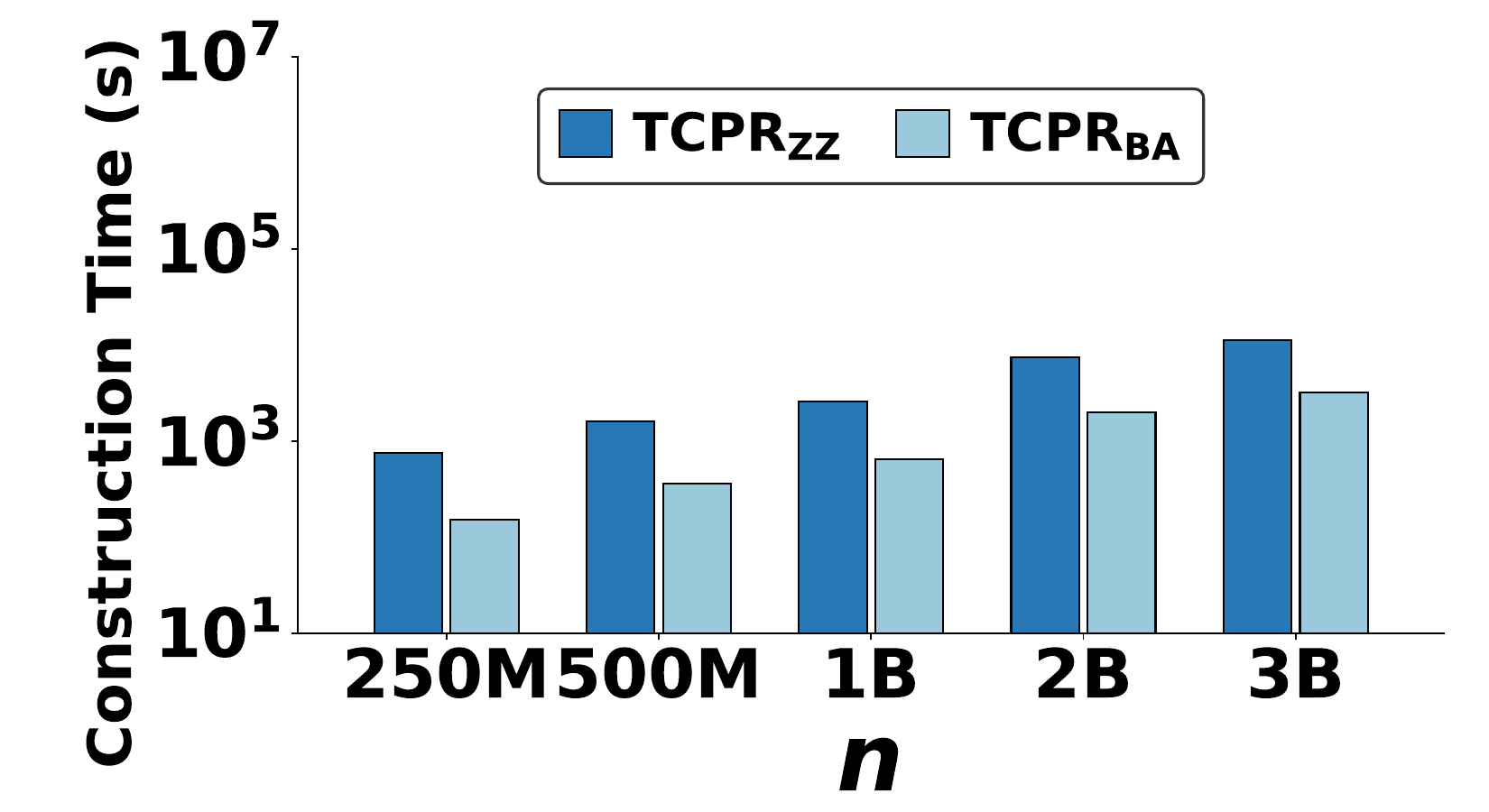}
    \caption{Constr. time vs. $n$}\label{fig:tcpr_build1}
  \end{subfigure}
  \begin{subfigure}[t]{0.20\linewidth}
    \includegraphics[width=\linewidth]{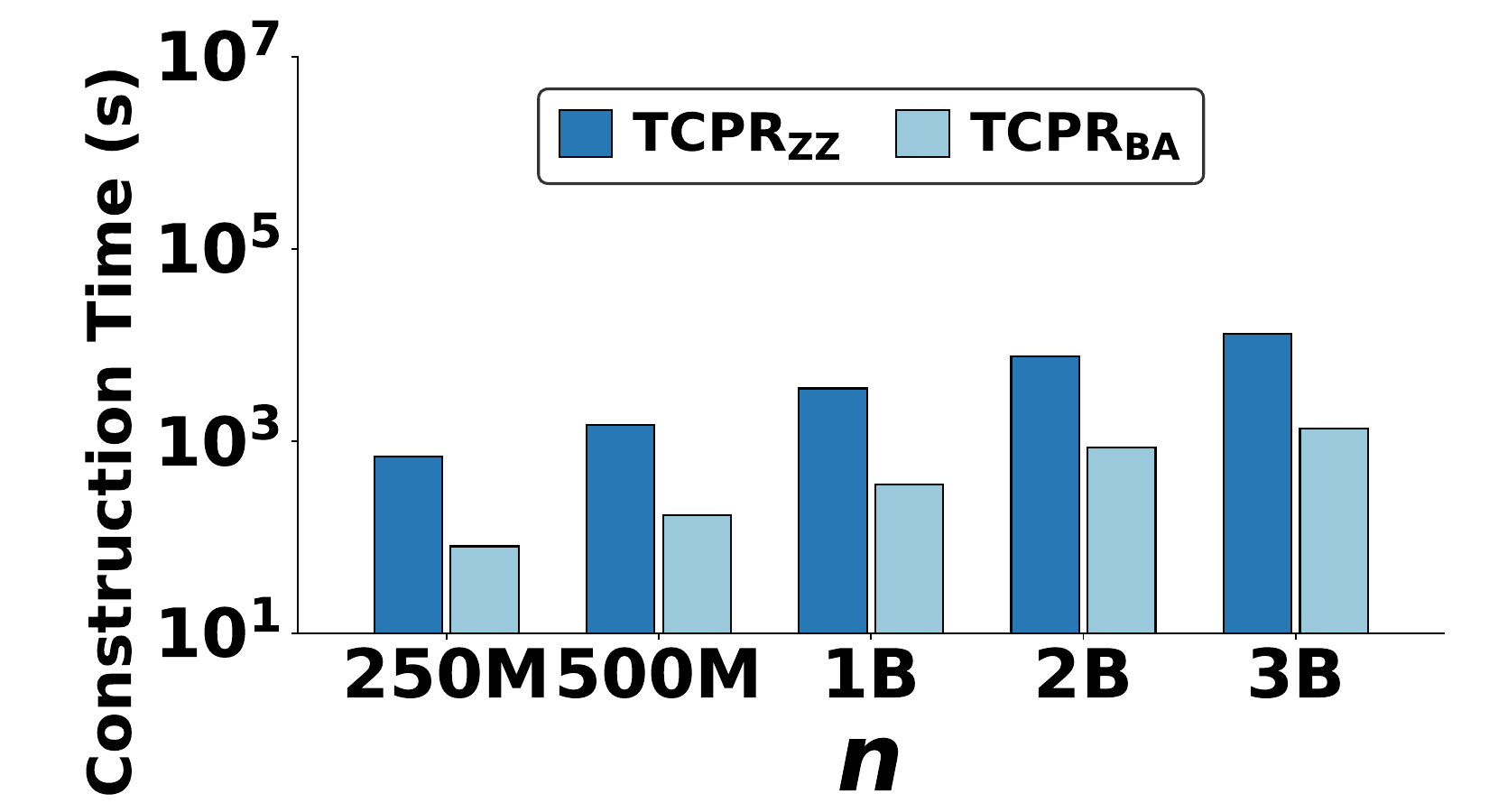}
    \caption{Constr. time vs. $n$}\label{fig:tcpr_build2}
  \end{subfigure}
  \begin{subfigure}[t]{0.20\linewidth}
    \includegraphics[width=\linewidth]{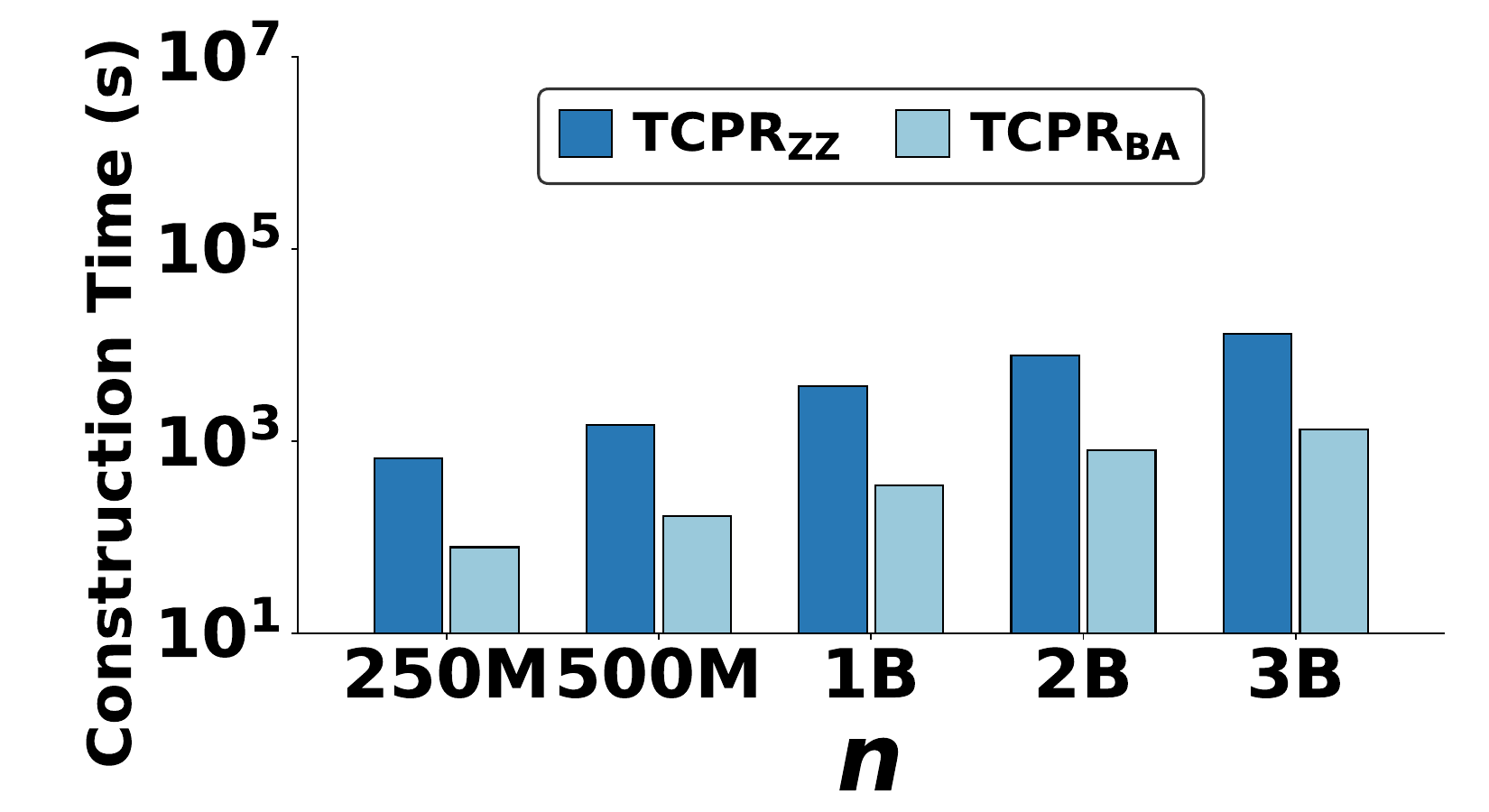}
    \caption{Constr. time vs. $n$}\label{fig:tcpr_build3}
  \end{subfigure}
    \vspace{\captionspacing}
  \caption{Construction time of our \TCPR index vs. the baseline on \chr using (a) \textsf{TF}, (b) \textsf{SP}, and (c)  \textsf{TP} vs. $n$.
  }\label{exp:TCPR_constr_time}
 \vspace{\captionspacing}
\end{figure}

\subsection{\CC} \label{experiments:CC}

\paragraph{Query Time.} \Cref{cc_qt_n} shows that our \CCZT index achieves query times that are \emph{five orders of magnitude faster} than \CCBA and its query time is \emph{unaffected by} $n$, as expected by its complexity.  \CCBA \emph{did not terminate within $24$ hours} for $n\geq 10^6$, as the term $|\loci(v)|$ in its query time bound is very large, which makes it impractical. For example, $|\loci(v)|$ is  $2\cdot 10^7$ on average over all query patterns 
when $n=10^6$. Thus, in the remainder of the section, we used $n=10^5$. 

\Cref{cc_qt_m} shows that the query time for \CCZT is \emph{at least five orders of magnitude} faster than that of \CCBA and its query time grows more slowly than linearly in $m$. \CCBA gets faster as $m$ increases, because the term 
$|\loci(v)|$ decreases. 

\paragraph{Index Size.} \Cref{cc_is} shows that \CCZT \emph{is $46\%$ smaller} on average over all $n$ values and \emph{$46\%$ smaller} when $n= 10^9$, since \ZZT occupies smaller space than the suffix trees used in \CCBA.  The size of both \CCZT and \CCBA scales linearly with $n$, as expected by their complexities.  

\paragraph{Construction Space.} In \Cref{cc_cs}, we report analogous results to those of \Cref{cc_is}, which are again due to the use of \ZZT vs. the suffix trees. For example, \CCZT needs \emph{$8\%$ less space to be constructed} than \CCBA for $n=10^9$. 

\paragraph{Construction Time.} \Cref{cc_ct} shows that \CCBA can be constructed faster than \CCZT, by $19$ times on average over all $n$ values, as expected by their complexities. The bottleneck of \CCZT is, as expected, the construction of \ZZT. Recall, however, that the query time of \CCBA is prohibitive, and that, unlike query time, construction time is a one-off cost.  

\subsection{\TCPR}\label{exp:tcpr}

\newcommand{\rowsep}{4pt}
\newcommand{\colsep}{\hfill}
\newcommand{\rowlabel}[1]{%
\makebox[0pt][r]{\rotatebox{90}{\small\textsf{#1}}\hspace{8pt}}}

\paragraph{Query Time.} \Crefrange{fig:qt_tcpr_n1}{fig:qt_tcpr_n3} show that our \TCPRZT index achieves query times that are \emph{at least $6$ and up to $44$ times faster} than \TCPRBA across all $n$ values and scoring functions. Since $|\mathcal{C}_T(P,q)|$ increases with $n$ (e.g., from $652$ to $10,967$ on average over all patterns), \TCPRBA becomes substantially slower (e.g., $6$ times slower as $n$ increases from $2.5\cdot 10^8$ to $3\cdot 10^9$) and scales worse than \TCPRZT. \TCPRZT performs even better for the \textsf{SP} and \textsf{TP} scoring functions, as \TCPRBA evaluates them at query time, and these functions are more expensive to evaluate than \textsf{TF}; see \Cref{sec:tcpr}. 

\Crefrange{fig:qt_tcpr_m1}{fig:qt_tcpr_m3} show that the query time of \TCPRZT is \emph{at least $4$ and up to $60$ times faster} than that of \TCPRBA as $m$ varies. The query time for both indexes  scales linearly with $m$, as expected by their complexities (here $k$ is fixed). The query time of \TCPRBA is larger for \textsf{SP} and \textsf{TP}, as their computation time is larger compared to that of \textsf{TF}. 

\Crefrange{fig:qt_tcpr_q1}{fig:qt_tcpr_q3} show that for varying $q$  the query time of \TCPRZT is \emph{at least $3$ and up to $72$ times faster} than that of \TCPRBA. The query time of both indexes increases with $q$, and that of \TCPRZT scales better; for \TCPRZT the reason is that $B$ increases (recall that $B=m+2q$ in our experiments), while for \TCPRBA the reason is that $|\mathcal{C}_T(P,q)|$ increases. 

\Crefrange{fig:qt_tcpr_k1}{fig:qt_tcpr_k3} show that for varying $k$ the query time of \TCPRZT is \emph{at least $5$ and up to $53$ times faster} than \TCPRBA across all scoring functions. The query time of \TCPRBA is not affected substantially by $k$, as the only part that changes with $k$ is the linear selection of pairs, which is very fast. 

\paragraph{Index Size.} \Crefrange{fig:tcpr_is_n1}{fig:tcpr_is_n3} show that the index size of  
both \TCPRZT and \TCPRBA scales linearly with $n$, as expected by the complexities, and that the size of the former is $16\%$ larger. 

\paragraph{Construction Space.} The results in \Crefrange{fig:tcpr_cs_n1}{fig:tcpr_cs_n3} are analogous to those for index size in \Crefrange{fig:tcpr_is_n1}{fig:tcpr_is_n3}. The gap between the two indexes is larger, as \ZZT needs more space to be constructed. 

\paragraph{Construction Time.} \Cref{exp:TCPR_constr_time} shows that \TCPRBA can be constructed faster, but the gap from \TCPRZT becomes much smaller as $n$ increases. Most of the construction time for \TCPRZT is spent in the \ZZT construction. However, \ZZT is also the key to the much faster query times of our index.

\begin{figure}[!ht]
  \centering
    \begin{subfigure}[t]{0.20\linewidth}
    \includegraphics[width=\linewidth]{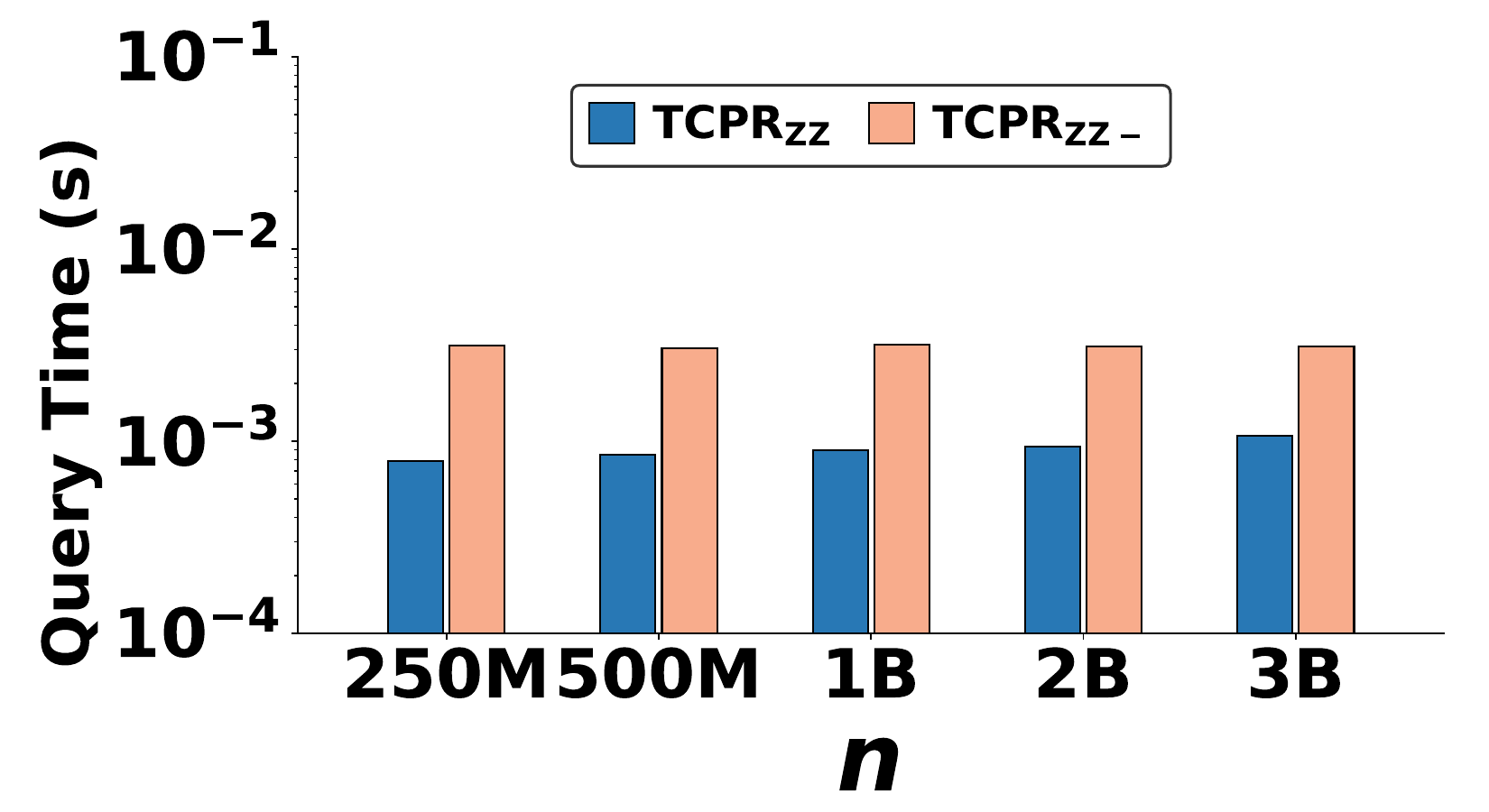}
    \caption{Query time vs. $n$}
     \label{fig:ablation_TCPR-TF_CHR_n_qt}
  \end{subfigure}
    \begin{subfigure}[t]{0.20\linewidth}
    \includegraphics[width=\linewidth]{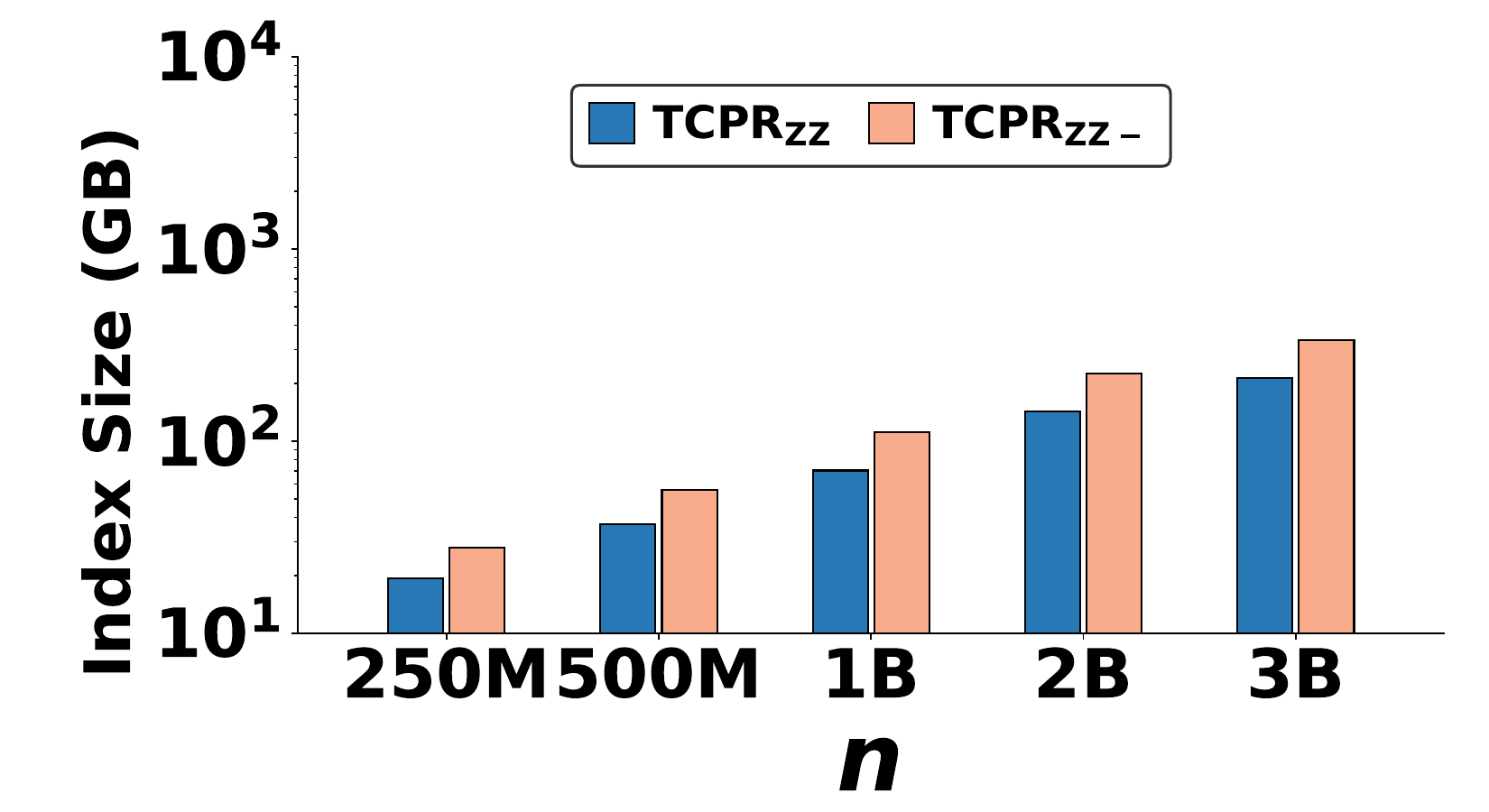}
    \caption{Index size vs. $n$}
         \label{fig:ablation_TCPR-TF_CHR_n_is}
  \end{subfigure}
    \begin{subfigure}[t]{0.20\linewidth}
    \includegraphics[width=\linewidth]{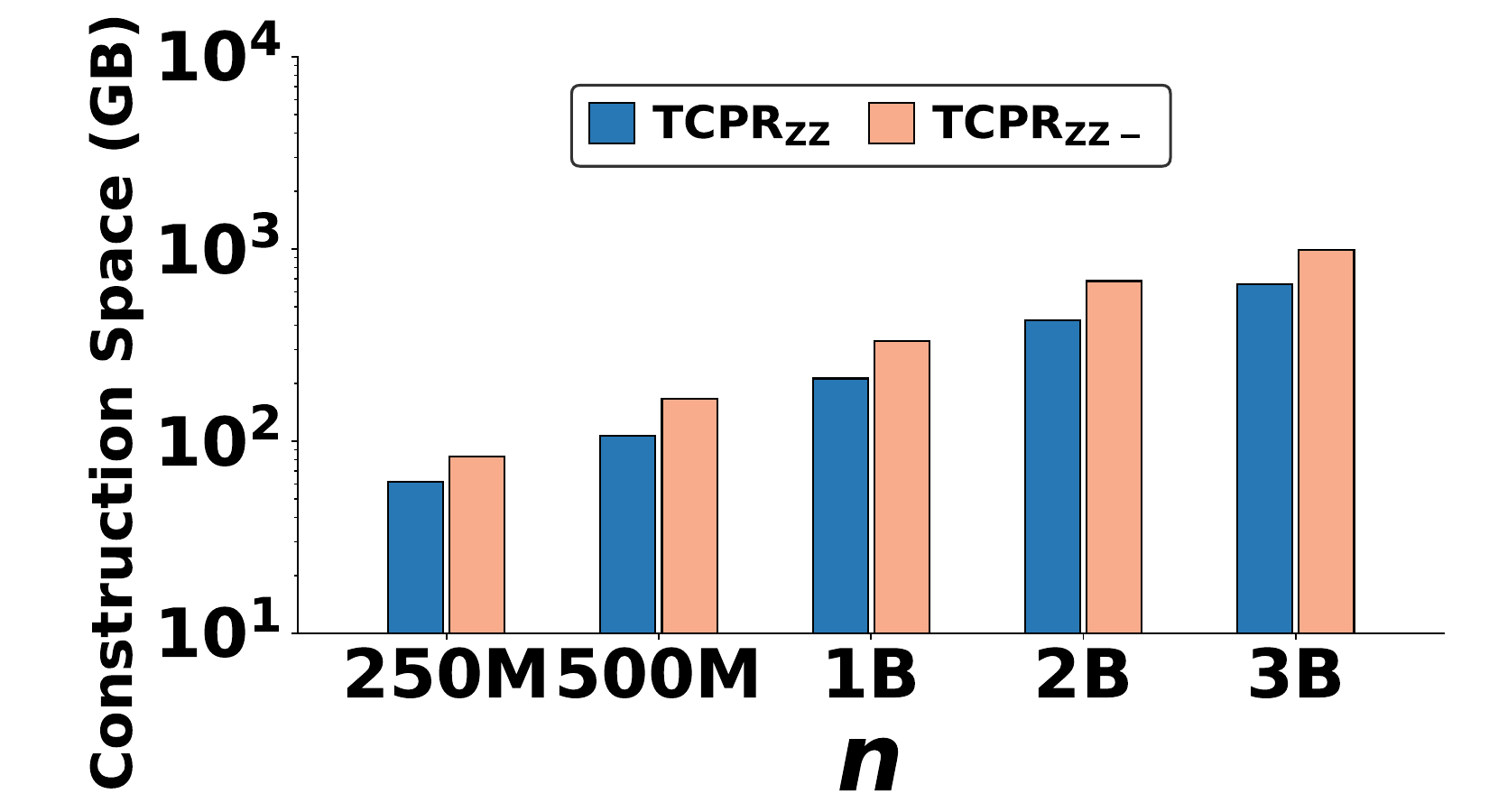}
    \caption{Constr. space vs. $n$}
         \label{fig:ablation_TCPR-TF_CHR_n_cs}
  \end{subfigure}
  \begin{subfigure}[t]{0.20\linewidth}
    \includegraphics[width=\linewidth]{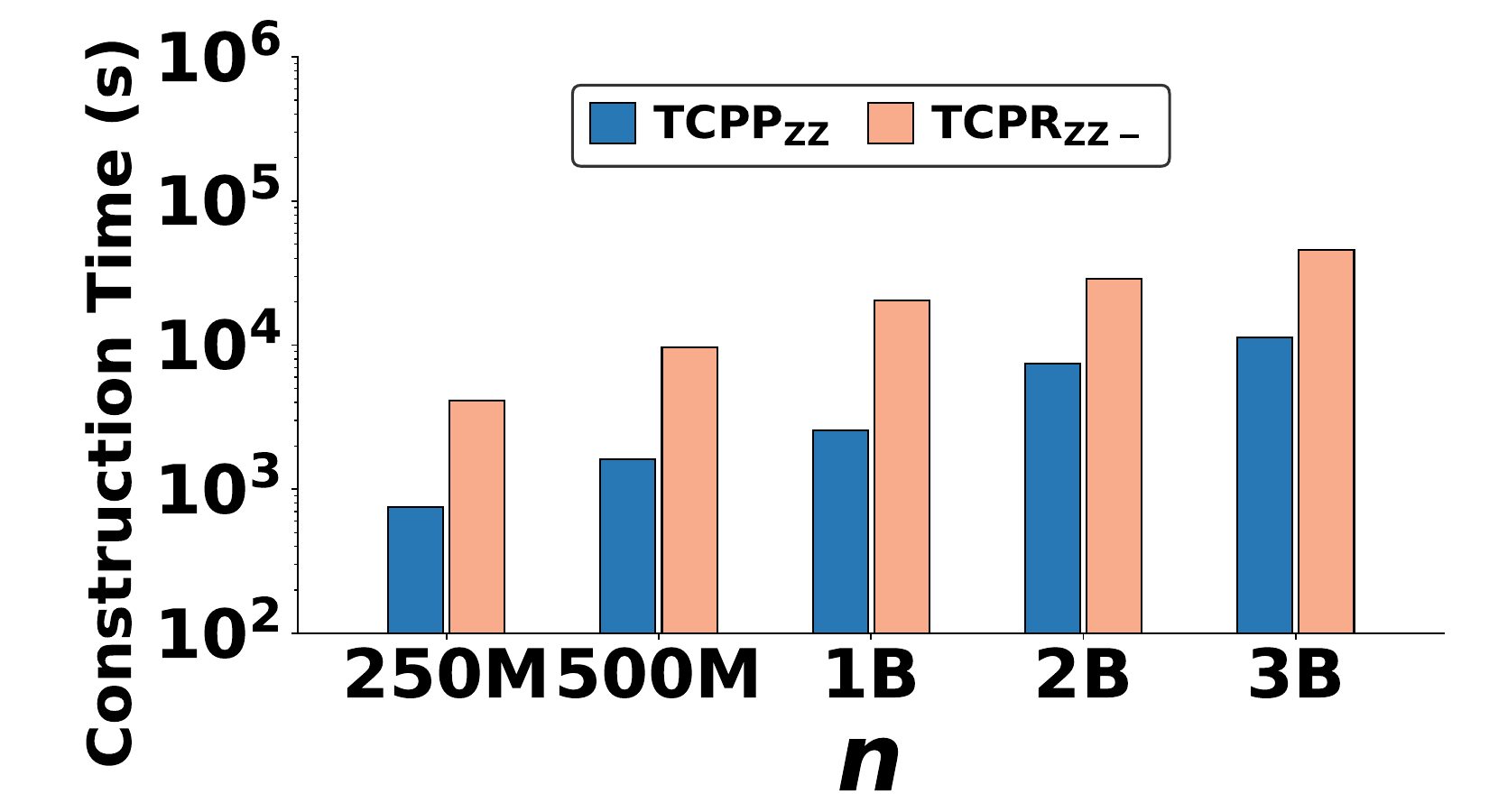}
    \caption{Constr. time vs. $n$}
         \label{fig:ablation_TCPR-TF_CHR_n_ct}
  \end{subfigure}
  \caption{(a) Query time, (b) index size, (c) construction space, and (d) construction time of \TCPRZT and \TCPRZTminus on \chr using the \textsf{TF} scoring function vs.\ $n$.}
  \label{fig:ablation_TCPR-TF_CHR_n}
\end{figure}

\subsection{Ablation Study}\label{sec:exp:ablation}

We show that the bounded-length optimization is very effective: it substantially improved our \TCPR index in all efficiency measures. We report results for the \textsf{TF} scoring function (those for \textsf{SP} and \textsf{TP} are analogous), and the default $B=27$ in \TCPRZT.  \Cref{fig:ablation_TCPR-TF_CHR_n_qt} shows that the query time of \TCPRZT is \emph{at least $2.4$ and up to $3.8$} times faster than that of \TCPRZTminus. The index size, construction space, and construction time of \TCPRZT  are smaller than those of  \TCPRZTminus by $1.5$, $1.6$, and $3.9$ times on average over all $n$ values, respectively; see \Cref{fig:ablation_TCPR-TF_CHR_n_is,fig:ablation_TCPR-TF_CHR_n_cs,fig:ablation_TCPR-TF_CHR_n_ct}. 

\section{Conclusion}\label{sec:conclusion}

We introduced \ZZT, a novel index that is useful for answering various types of contextual queries efficiently, and used it as a basis to build specialized indexes for four practical contextual query types. A natural next step is to support dynamic texts. Since practical full-text indexes for dynamic texts are rather undeveloped, designing dynamic versions of \ZZT or our specialized indexes is a challenging problem left for future work. 

\bibliographystyle{alphaurl}
\bibliography{main}
\clearpage

\appendix
\section{Additional Experimental Results}\label{appendix}

For each of our four indexes, \Cref{sec:experiments} reported results 
on one representative dataset. We report here the results on the
remaining four datasets. The behavior of every index is analogous, and hence all conclusions drawn in \Cref{sec:experiments} carry over.

\Cref{fig:app:LFCS:query} shows the query time of \LFCSZT vs.\ \LFCSBA for
varying $n$, $m$, and $\tau$ on \bst, \chr, \sars, and \wiki. The results in this
figure are analogous to those of
\Crefrange{fig:LFCS_SDSL_n_query}{fig:LFCS_SDSL_tau_query} in
\Cref{experiments:lfcs}.

\Cref{fig:app:LFCS:cost} shows the index size, construction space, and
construction time of \LFCSZT vs.\ \LFCSBA for varying $n$ on \bst, \chr, \sars,
and \wiki. Its results are analogous to those in
\Crefrange{fig:LFCS_SDSL_n_index}{fig:LFCS_SDSL_n_build} in
\Cref{experiments:lfcs}.

\Cref{fig:app:LCCS:query} shows the query time of \LCCSZT vs.\ \LCCSBA for
varying $N$, $m$, and $\tau$ on \bst, \chr, \sdsl, and \wiki. The results in this
figure are analogous to those of
\Crefrange{fig:LCCS_qt_N}{fig:LCCS_qt_tau} in
\Cref{experiments:LCCS}.

\Cref{fig:app:LCCS:cost} shows the index size, construction space, and
construction time of \LCCSZT vs.\ \LCCSBA for varying $N$ on \bst, \chr, \sdsl,
and \wiki. Its results are analogous to those in
\Crefrange{fig:LCCS_is}{fig:LCCS_ct} in \Cref{experiments:LCCS}.

\Cref{fig:app:CC:query} shows the query time of \CCZT vs.\ \CCBA for
varying $n$ and $m$ on \chr, \sars, \sdsl, and \wiki. The results in this figure
are analogous to those of \Cref{cc_qt_n,cc_qt_m} in
\Cref{experiments:CC}. As in \Cref{cc_qt_n}, \CCBA did not
terminate within $24$ hours for $n\geq 10^6$ on any dataset.

\Cref{fig:app:CC:cost} shows the index size, construction space, and
construction time of \CCZT vs.\ \CCBA for varying $n$ on \chr, \sars, \sdsl, and
\wiki. Its results are analogous to those in \Crefrange{cc_is}{cc_ct} in
\Cref{experiments:CC}.

\Cref{fig:app:TF:query} shows the query time of \TCPRZT vs.\ \TCPRBA,
using the \textsf{TF} scoring function, for varying $n$, $m$, $q$, and $k$ on \bst,
\sars, \sdsl, and \wiki. The results in this figure are analogous to those of
\Cref{fig:qt_tcpr_n1,fig:qt_tcpr_m1,fig:qt_tcpr_q1,fig:qt_tcpr_k1} in \Cref{exp:tcpr}.

\Cref{fig:app:TF:cost} shows the index size, construction space, and
construction time of \TCPRZT vs.\ \TCPRBA, using the \textsf{TF} scoring
function, for varying $n$ on \bst, \sars, \sdsl, and \wiki. Its results are analogous
to those in \Cref{fig:tcpr_is_n1,fig:tcpr_cs_n1,fig:tcpr_build1} in \Cref{exp:tcpr}.

\Cref{fig:app:SP:query} shows the query time of \TCPRZT vs.\ \TCPRBA,
using the \textsf{SP} scoring function, for varying $n$, $m$, $q$, and $k$ on \bst,
\sars, \sdsl, and \wiki. The results in this figure are analogous to those of
\Cref{fig:qt_tcpr_n2,fig:qt_tcpr_m2,fig:qt_tcpr_q2,fig:qt_tcpr_k2} in \Cref{exp:tcpr}.

\Cref{fig:app:SP:cost} shows the index size, construction space, and
construction time of \TCPRZT vs.\ \TCPRBA, using the \textsf{SP} scoring
function, for varying $n$ on \bst, \sars, \sdsl, and \wiki. Its results are analogous
to those in \Cref{fig:tcpr_is_n2,fig:tcpr_cs_n2,fig:tcpr_build2} in \Cref{exp:tcpr}.

\Cref{fig:app:TP:query} shows the query time of \TCPRZT vs.\ \TCPRBA,
using the \textsf{TP} scoring function, for varying $n$, $m$, $q$, and $k$ on \bst,
\sars, \sdsl, and \wiki. The results in this figure are analogous to those of
\Cref{fig:qt_tcpr_n3,fig:qt_tcpr_m3,fig:qt_tcpr_q3,fig:qt_tcpr_k3} in \Cref{exp:tcpr}.

\Cref{fig:app:TP:cost} shows the index size, construction space, and
construction time of \TCPRZT vs.\ \TCPRBA, using the \textsf{TP} scoring
function, for varying $n$ on \bst, \sars, \sdsl, and \wiki. Its results are analogous
to those in \Cref{fig:tcpr_is_n3,fig:tcpr_cs_n3,fig:tcpr_build3} in \Cref{exp:tcpr}.

\newcommand{\appfigwidth}{0.2\linewidth}

\begin{figure}[ht]
  \centering
  \begin{subfigure}[t]{\appfigwidth}
    \includegraphics[width=\linewidth]{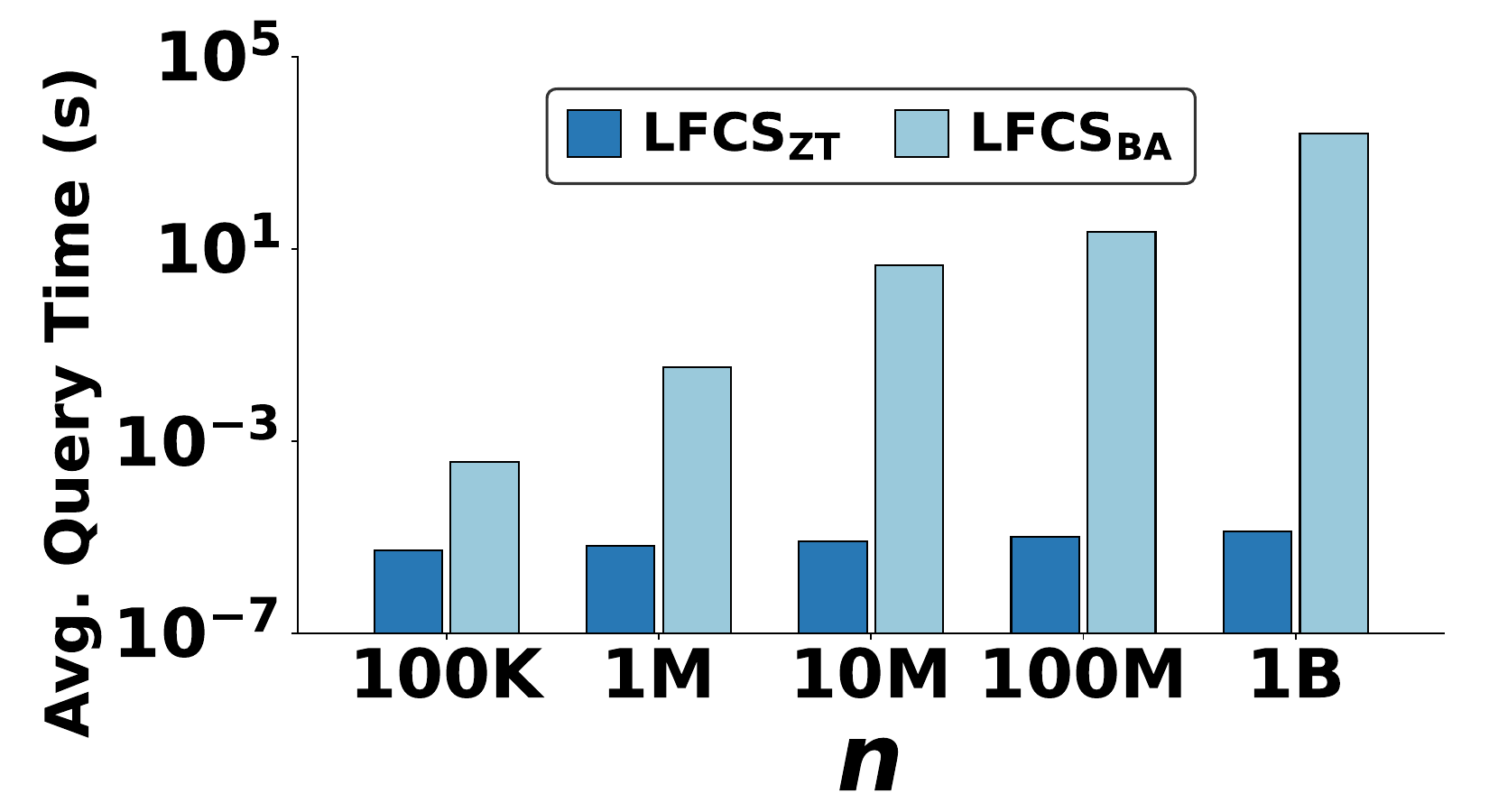}
    \caption{Query time vs. $n$}\label{fig:app:LFCS:n:query:BST}
  \end{subfigure}
  \begin{subfigure}[t]{\appfigwidth}
    \includegraphics[width=\linewidth]{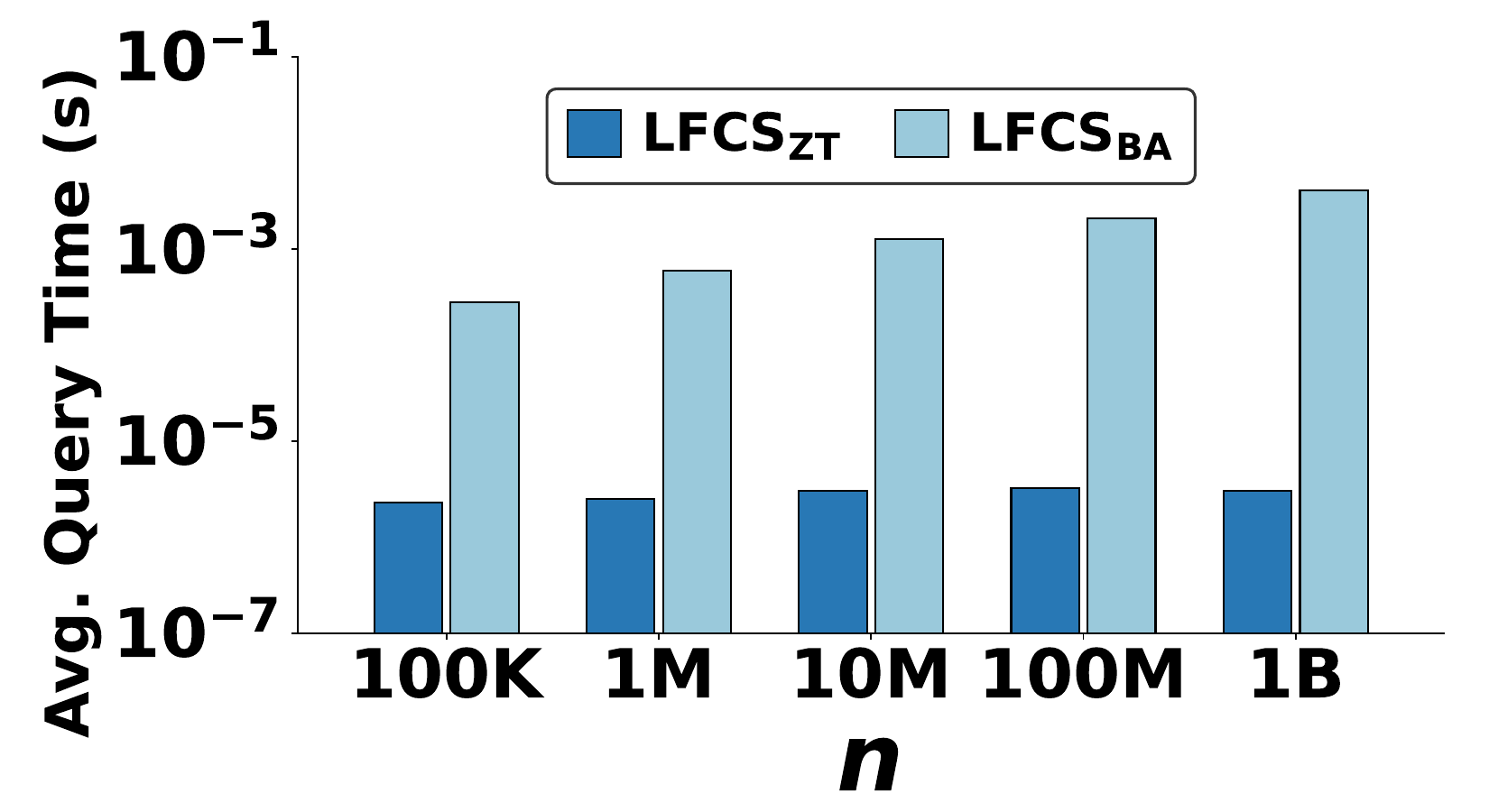}
    \caption{Query time vs. $n$}\label{fig:app:LFCS:n:query:CHR}
  \end{subfigure}
  \begin{subfigure}[t]{\appfigwidth}
    \includegraphics[width=\linewidth]{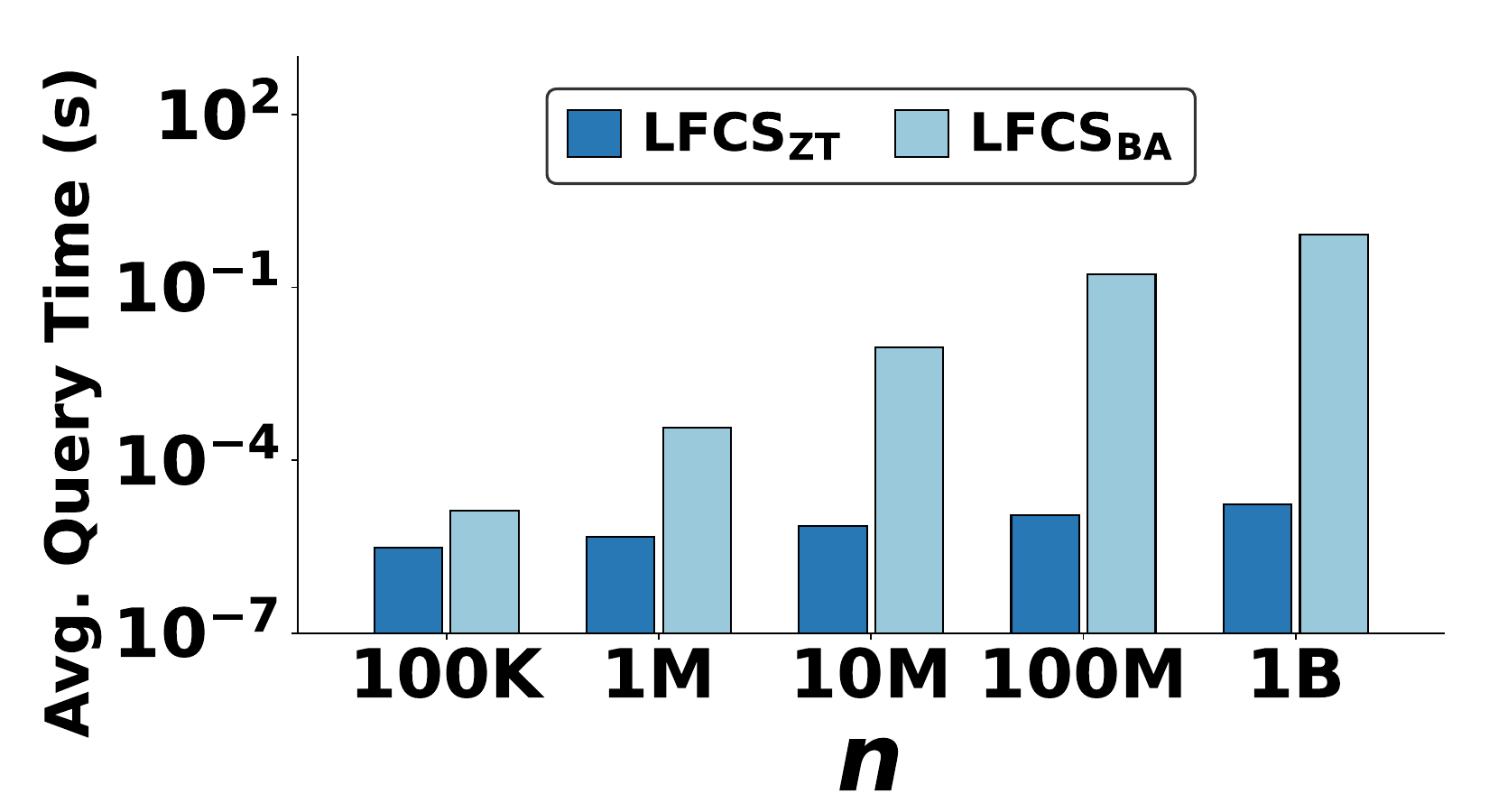}
    \caption{Query time vs. $n$}\label{fig:app:LFCS:n:query:SARS}
  \end{subfigure}
  \begin{subfigure}[t]{\appfigwidth}
    \includegraphics[width=\linewidth]{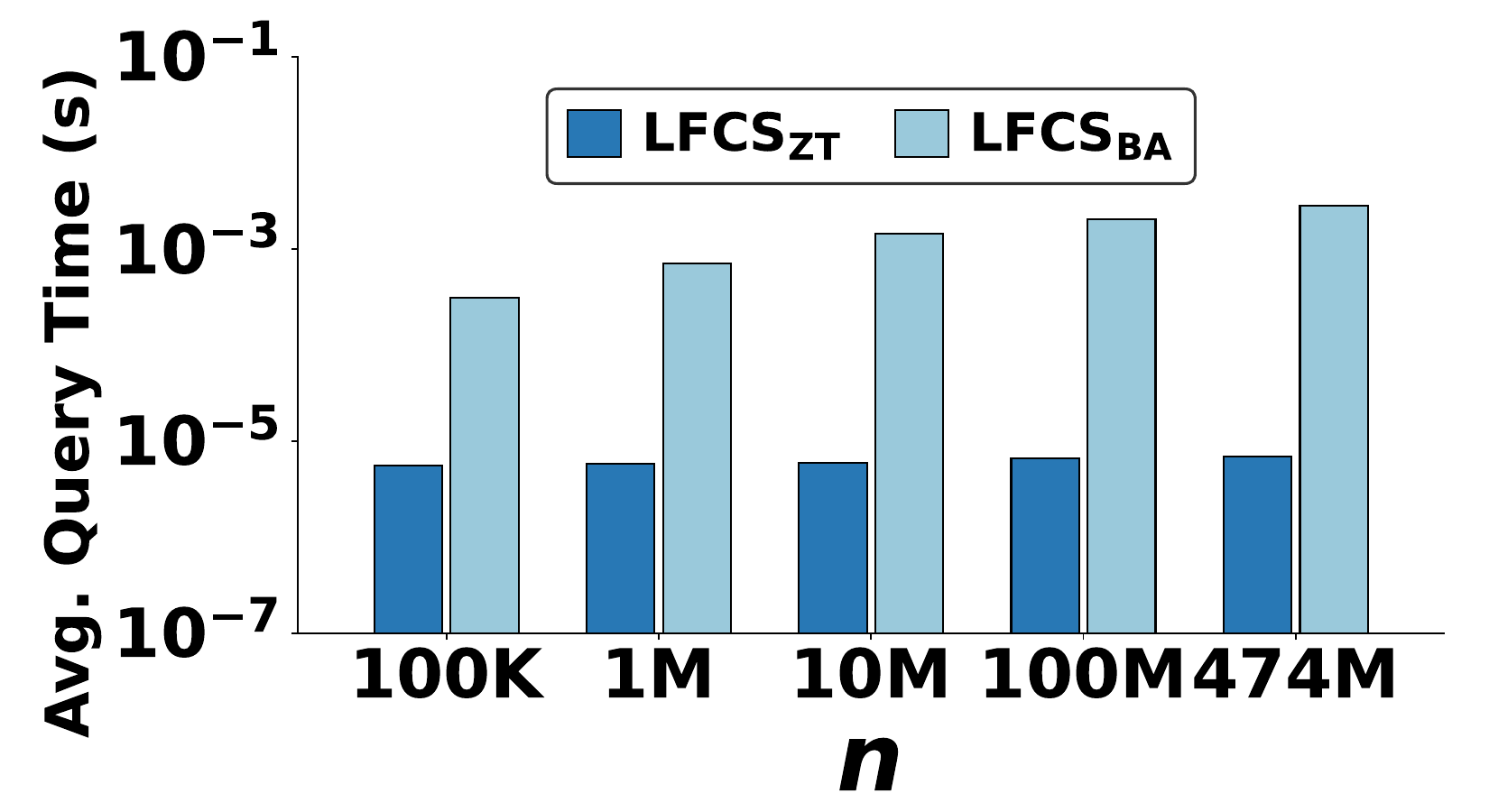}
    \caption{Query time vs. $n$}\label{fig:app:LFCS:n:query:WIKI}
  \end{subfigure}\\[0pt]
  \begin{subfigure}[t]{\appfigwidth}
    \includegraphics[width=\linewidth]{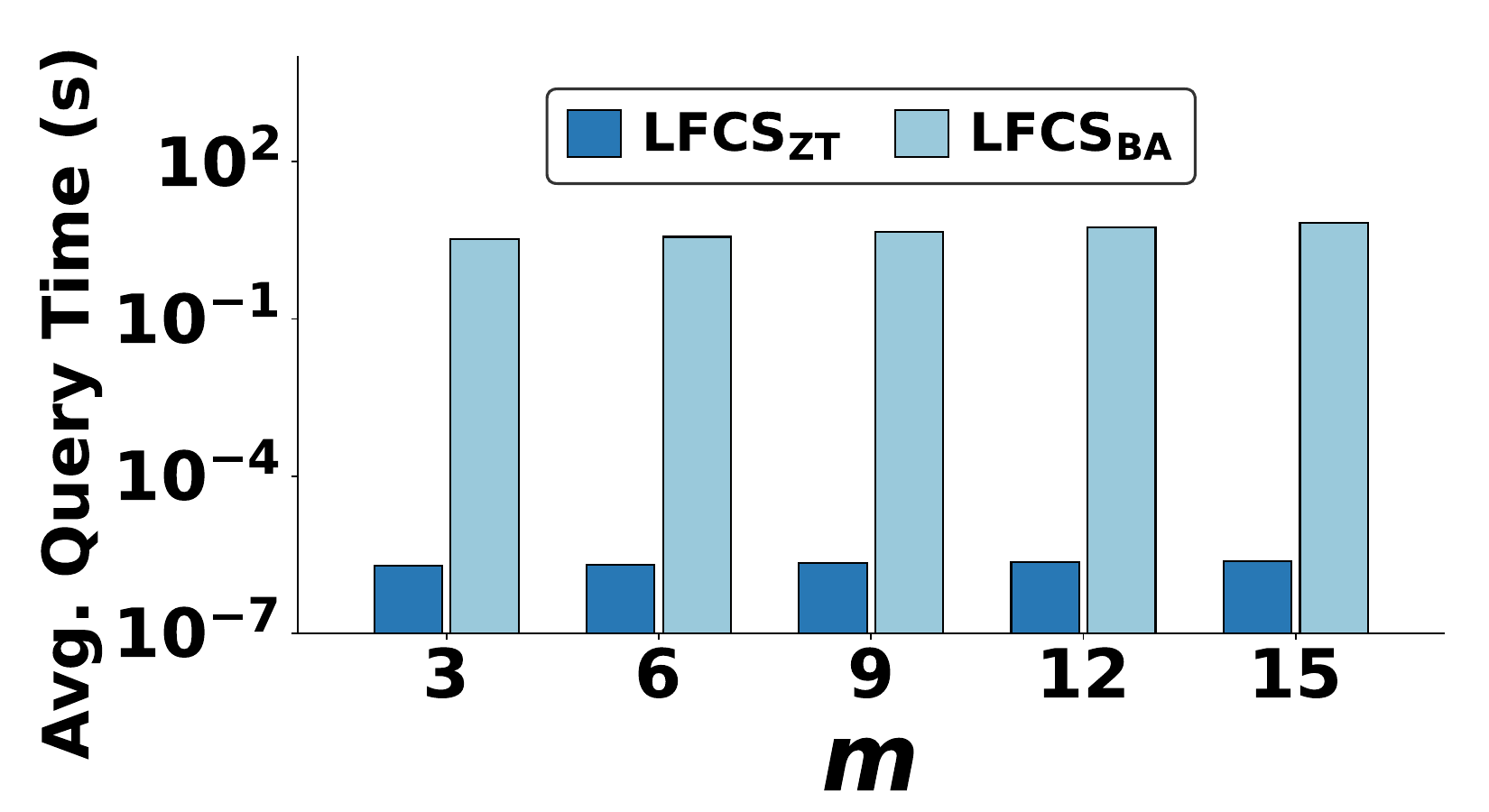}
    \caption{Query time vs. $m$}\label{fig:app:LFCS:m:query:BST}
  \end{subfigure}
  \begin{subfigure}[t]{\appfigwidth}
    \includegraphics[width=\linewidth]{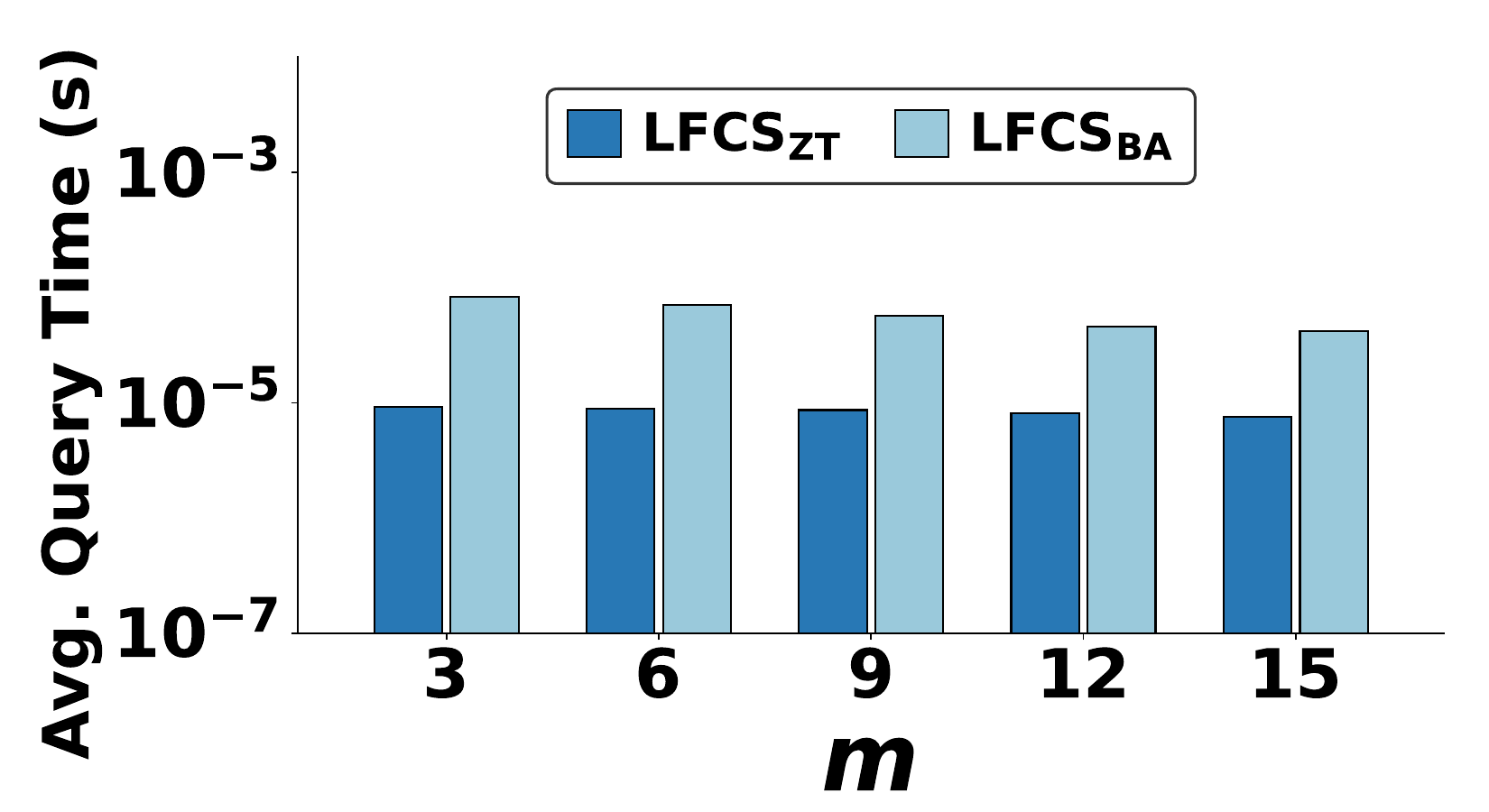}
    \caption{Query time vs. $m$}\label{fig:app:LFCS:m:query:CHR}
  \end{subfigure}
  \begin{subfigure}[t]{\appfigwidth}
    \includegraphics[width=\linewidth]{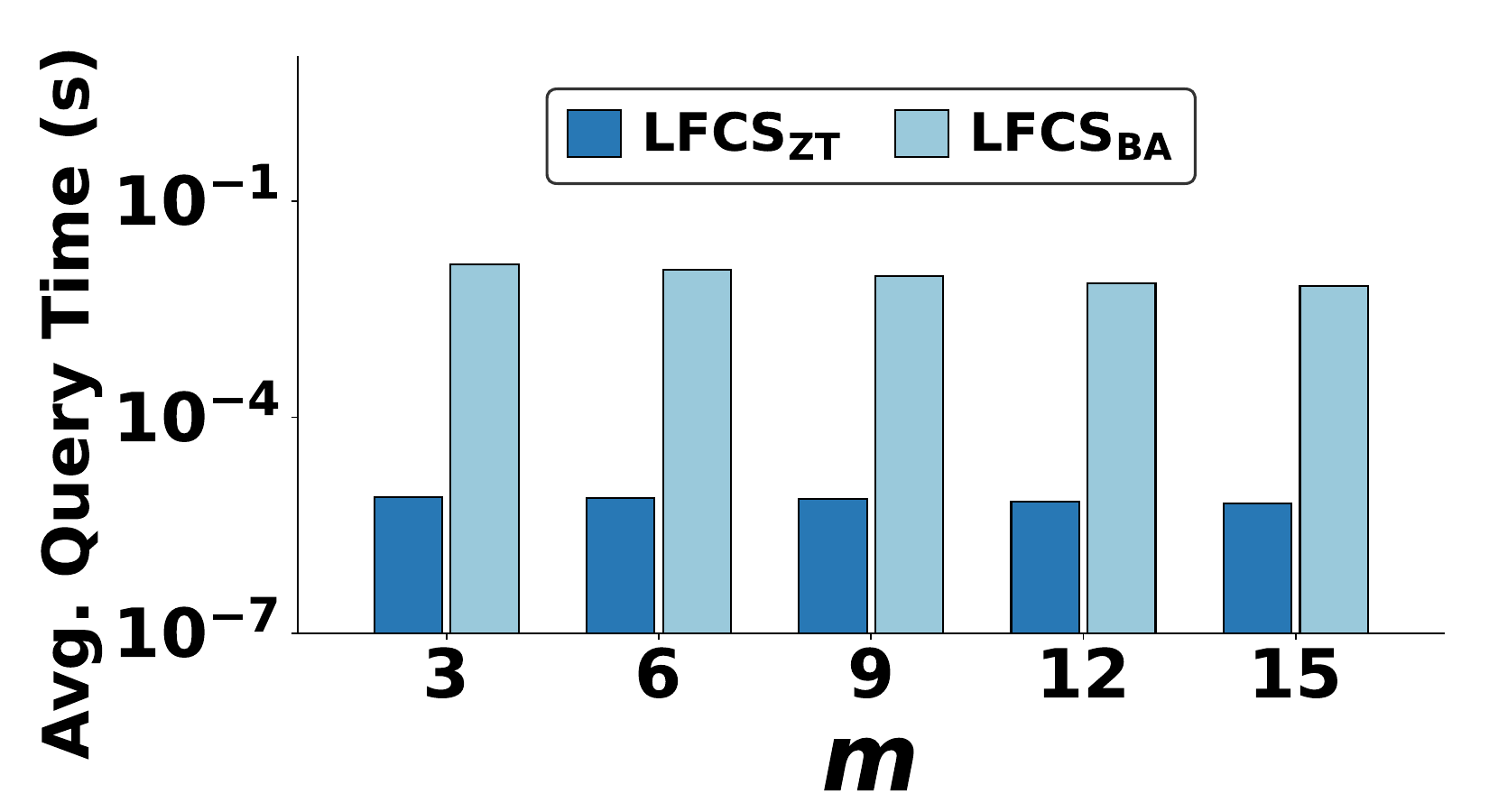}
    \caption{Query time vs. $m$}\label{fig:app:LFCS:m:query:SARS}
  \end{subfigure}
  \begin{subfigure}[t]{\appfigwidth}
    \includegraphics[width=\linewidth]{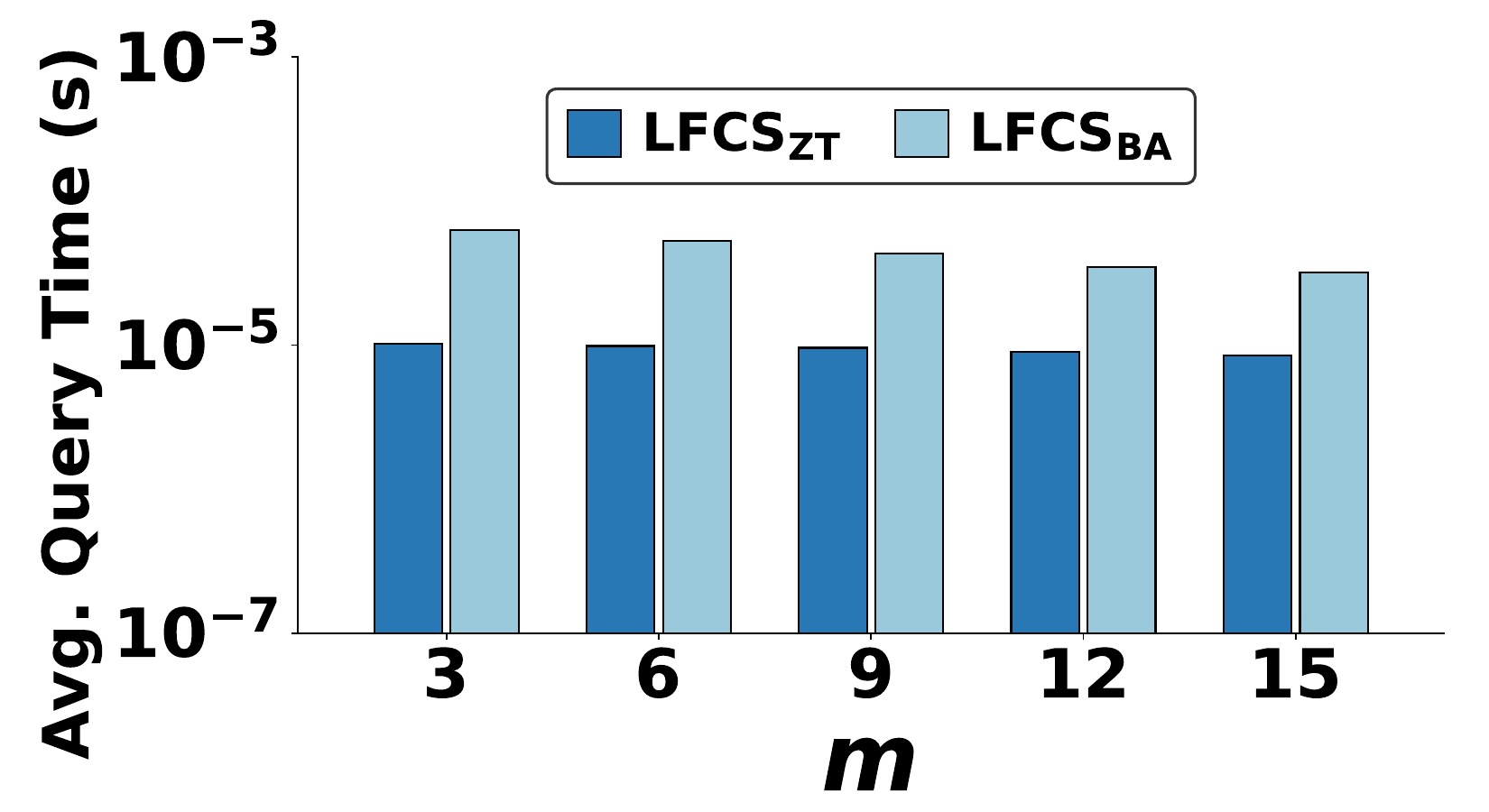}
    \caption{Query time vs. $m$}\label{fig:app:LFCS:m:query:WIKI}
  \end{subfigure}\\[0pt]
  \begin{subfigure}[t]{\appfigwidth}
    \includegraphics[width=\linewidth]{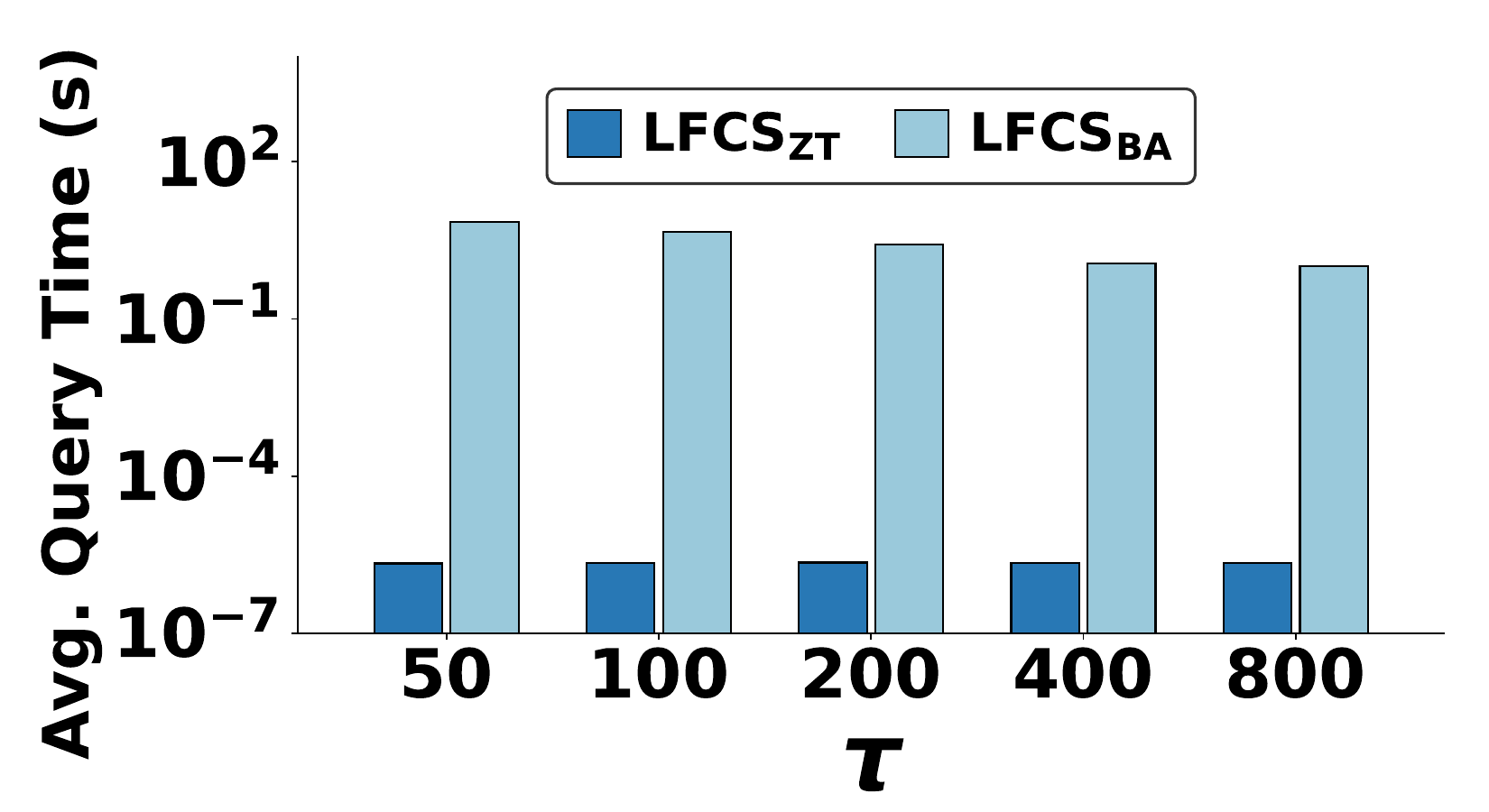}
    \caption{Query time vs. $\tau$}\label{fig:app:LFCS:tau:query:BST}
  \end{subfigure}
  \begin{subfigure}[t]{\appfigwidth}
    \includegraphics[width=\linewidth]{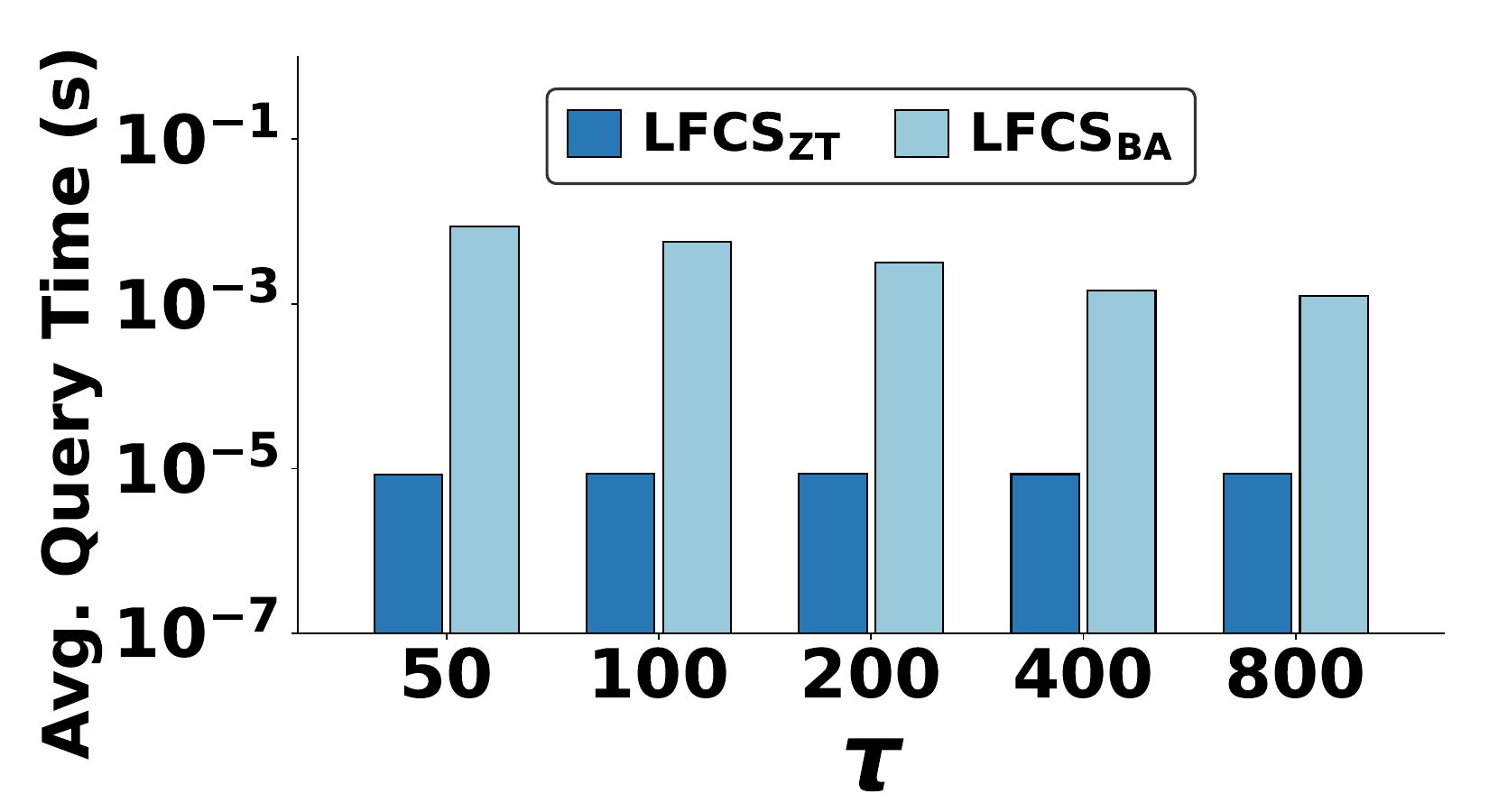}
    \caption{Query time vs. $\tau$}\label{fig:app:LFCS:tau:query:CHR}
  \end{subfigure}
  \begin{subfigure}[t]{\appfigwidth}
    \includegraphics[width=\linewidth]{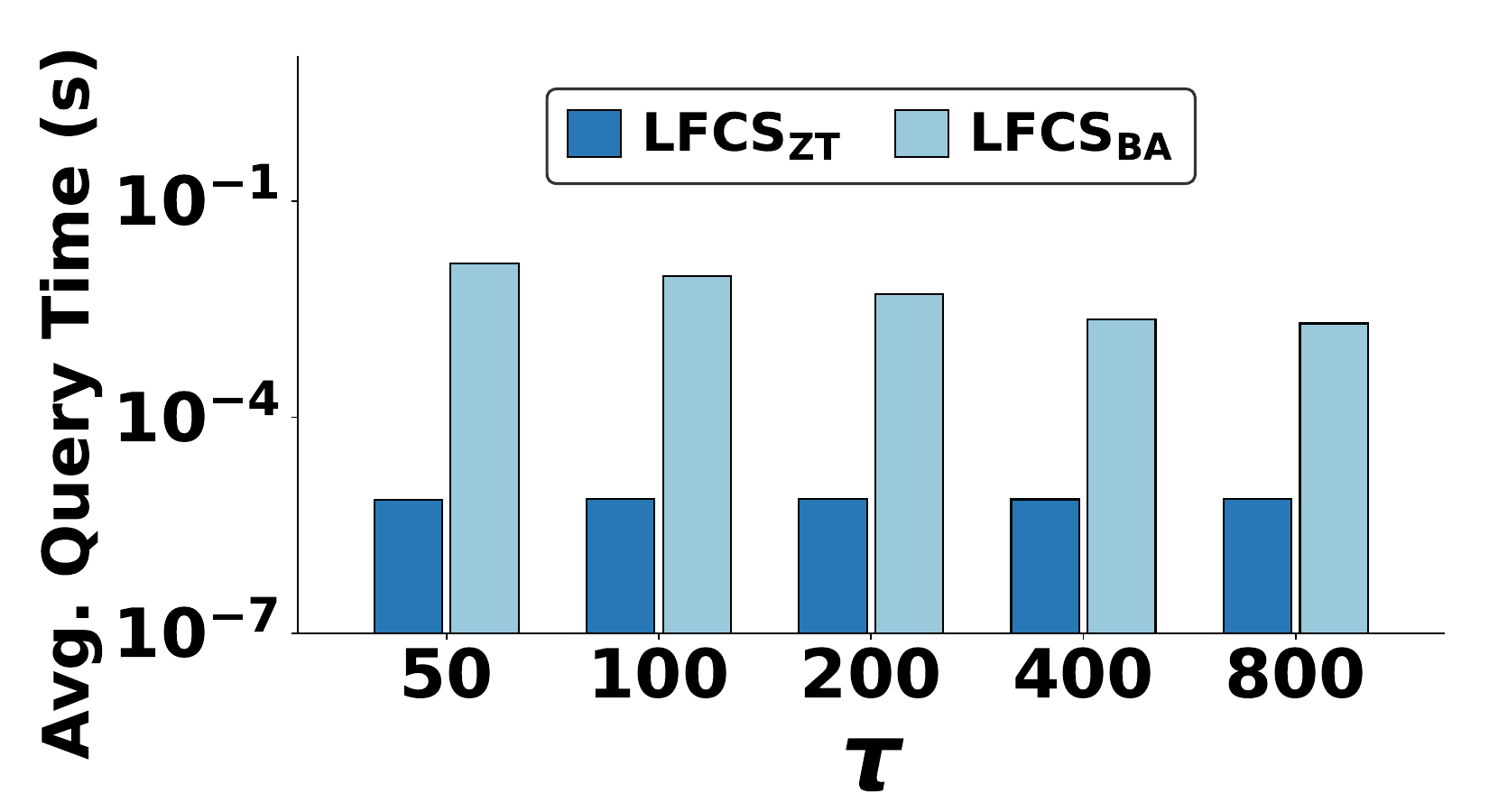}
    \caption{Query time vs. $\tau$}\label{fig:app:LFCS:tau:query:SARS}
  \end{subfigure}
  \begin{subfigure}[t]{\appfigwidth}
    \includegraphics[width=\linewidth]{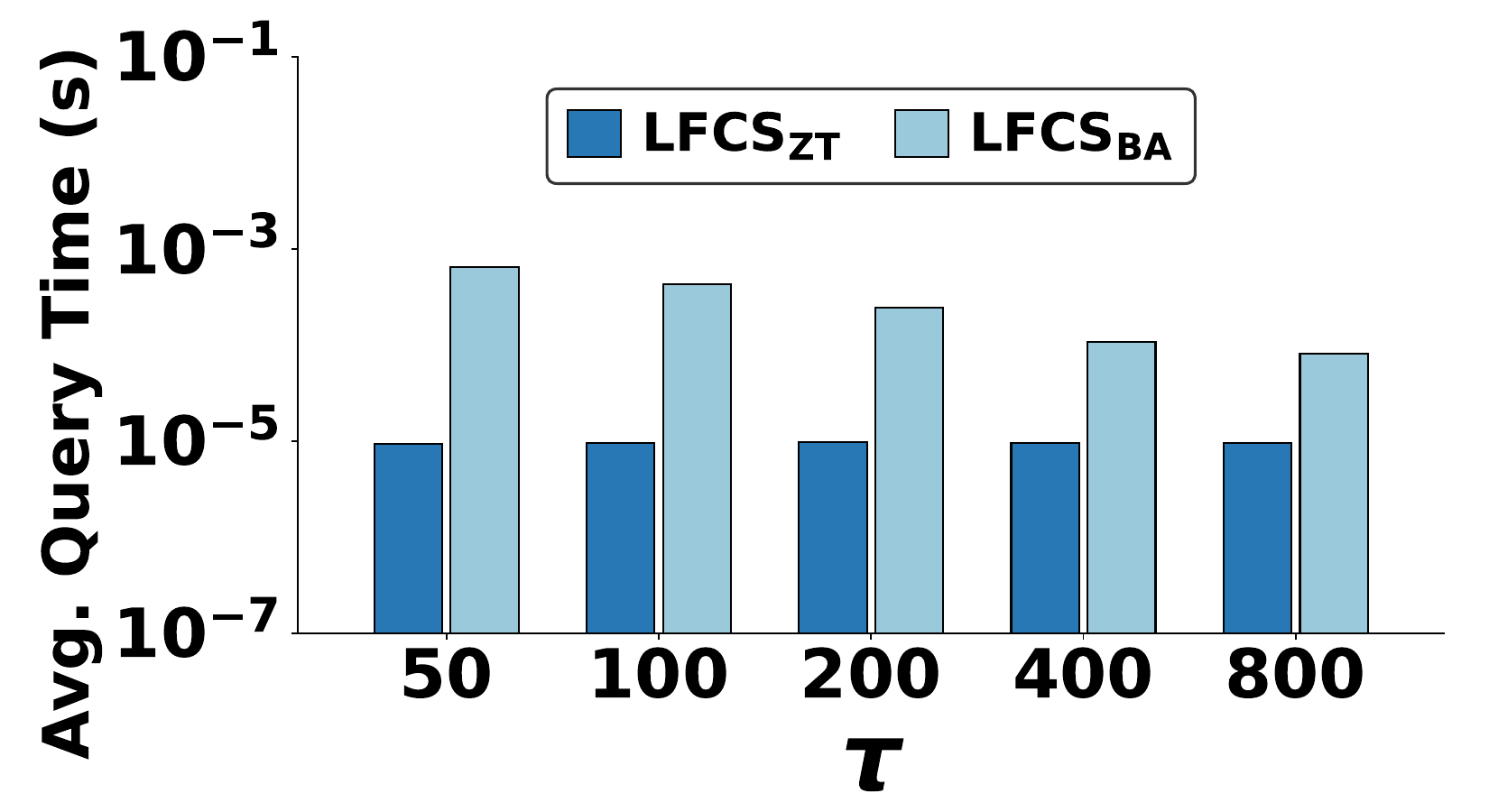}
    \caption{Query time vs. $\tau$}\label{fig:app:LFCS:tau:query:WIKI}
  \end{subfigure}
  \vspace{\captionspacing}
  \vspace{+2mm}
  \caption{Query time of our \LFCS index vs. \LFCSBA on (a) \bst, (b) \chr, (c) \sars, and (d) \wiki vs. $n$; on (e) \bst, (f) \chr, (g) \sars, and (h) \wiki vs. $m$; on (i) \bst, (j) \chr, (k) \sars, and (l) \wiki vs. $\tau$.}\label{fig:app:LFCS:query}
\end{figure}

\begin{figure}[ht]
  \centering
  \begin{subfigure}[t]{\appfigwidth}
    \includegraphics[width=\linewidth]{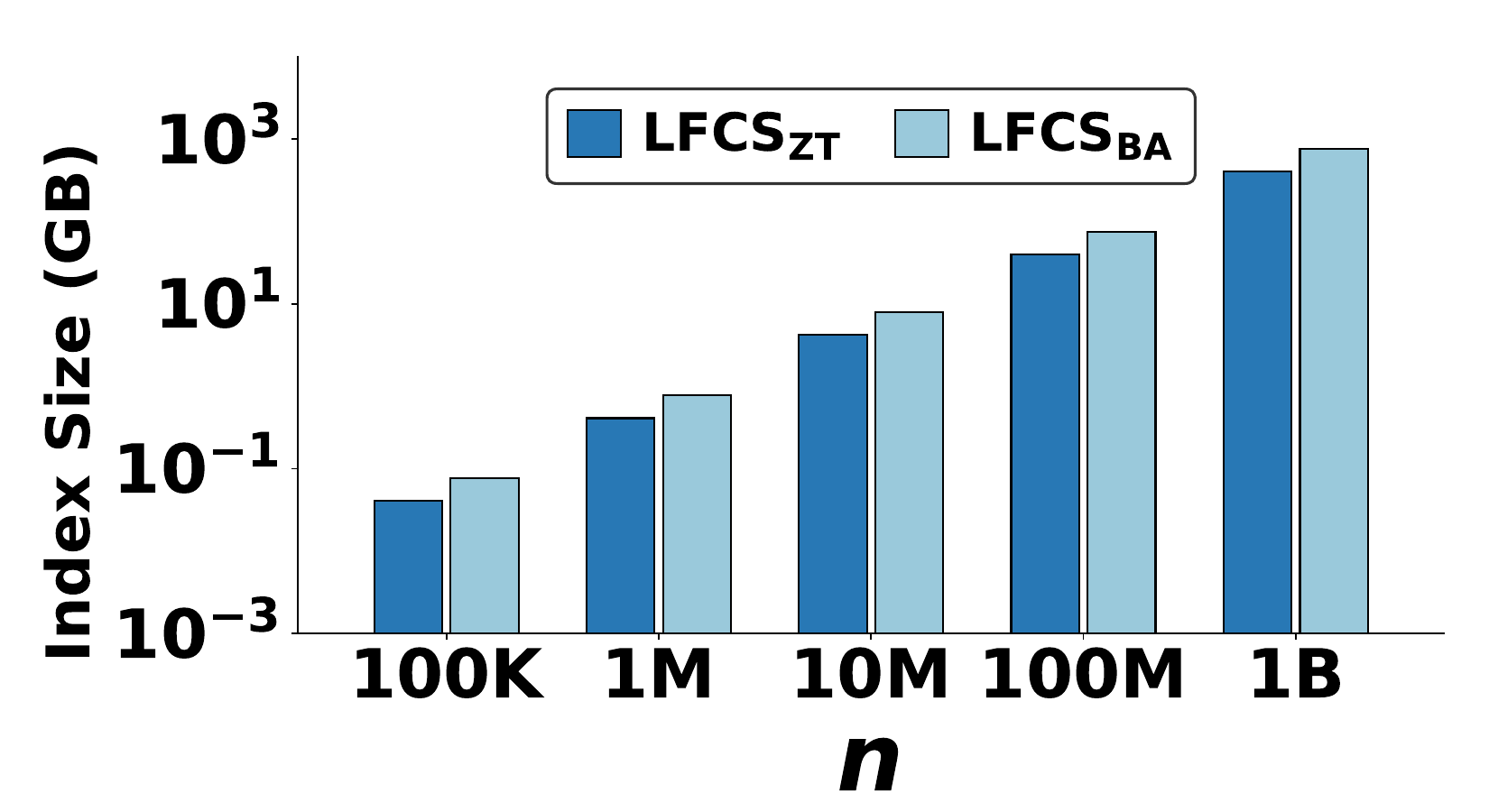}
    \caption{Index size vs. $n$}\label{fig:app:LFCS:n:index:BST}
  \end{subfigure}
  \begin{subfigure}[t]{\appfigwidth}
    \includegraphics[width=\linewidth]{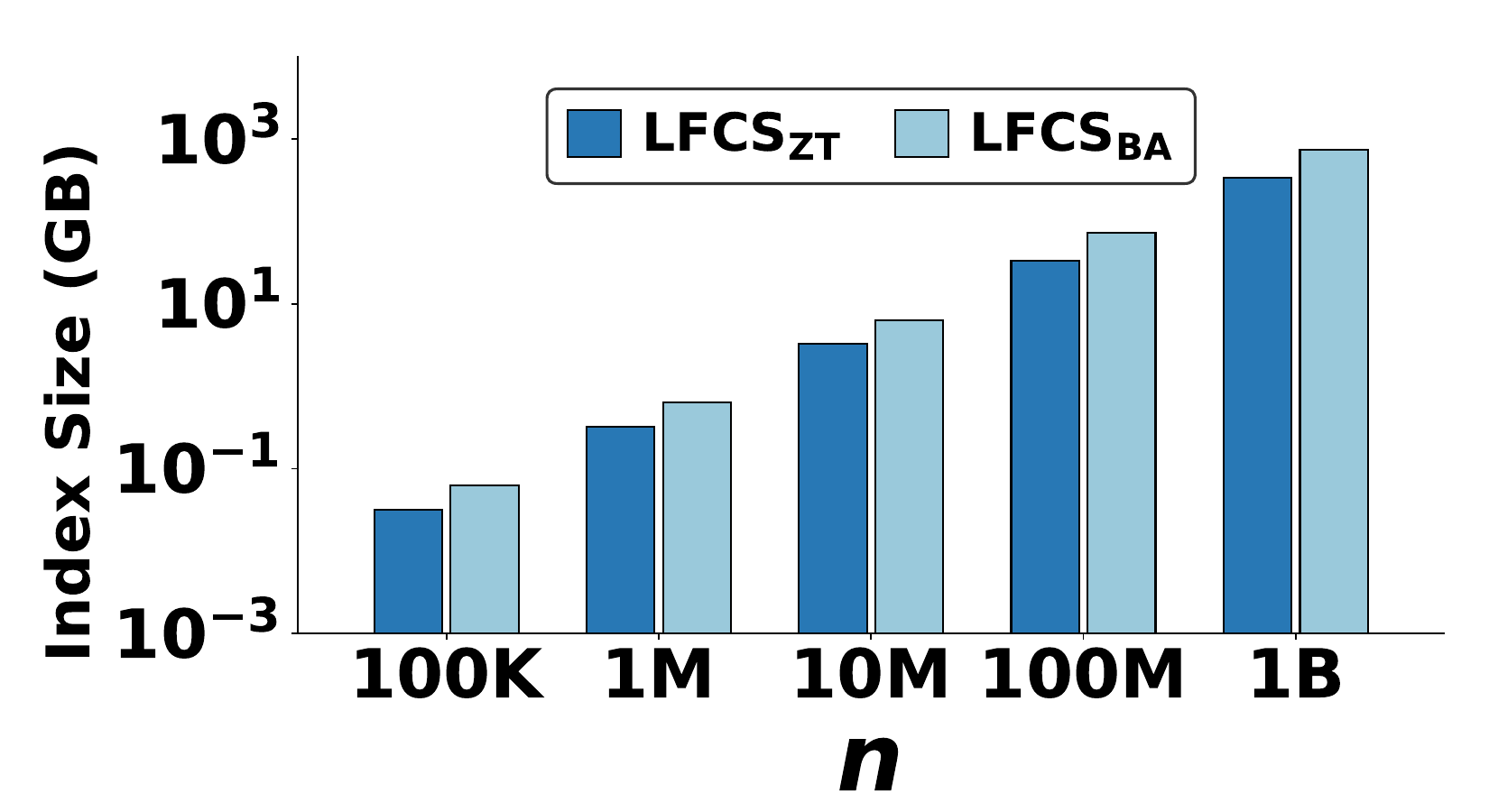}
    \caption{Index size vs. $n$}\label{fig:app:LFCS:n:index:CHR}
  \end{subfigure}
  \begin{subfigure}[t]{\appfigwidth}
    \includegraphics[width=\linewidth]{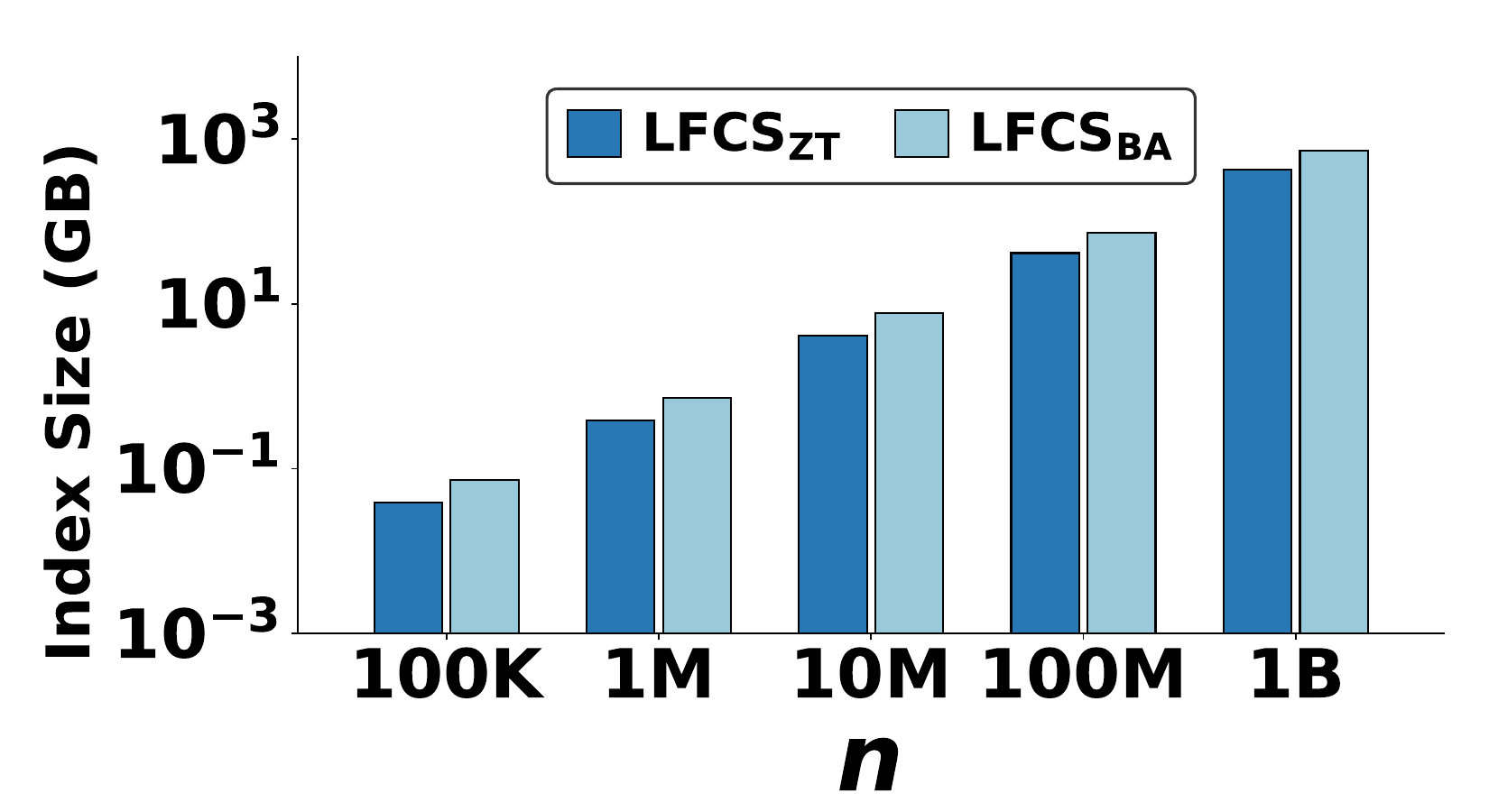}
    \caption{Index size vs. $n$}\label{fig:app:LFCS:n:index:SARS}
  \end{subfigure}
  \begin{subfigure}[t]{\appfigwidth}
    \includegraphics[width=\linewidth]{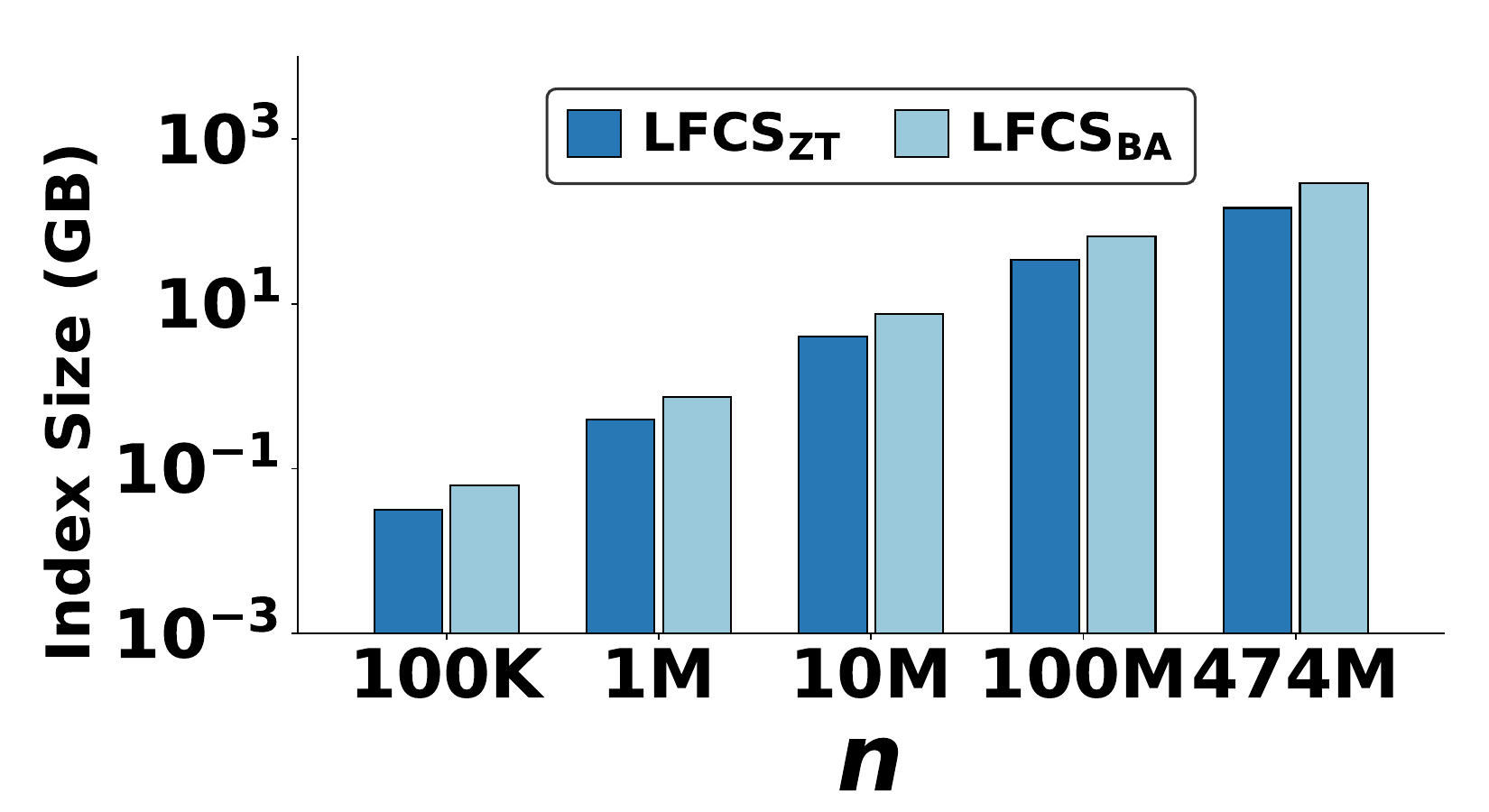}
    \caption{Index size vs. $n$}\label{fig:app:LFCS:n:index:WIKI}
  \end{subfigure}\\[0pt]
  \begin{subfigure}[t]{\appfigwidth}
    \includegraphics[width=\linewidth]{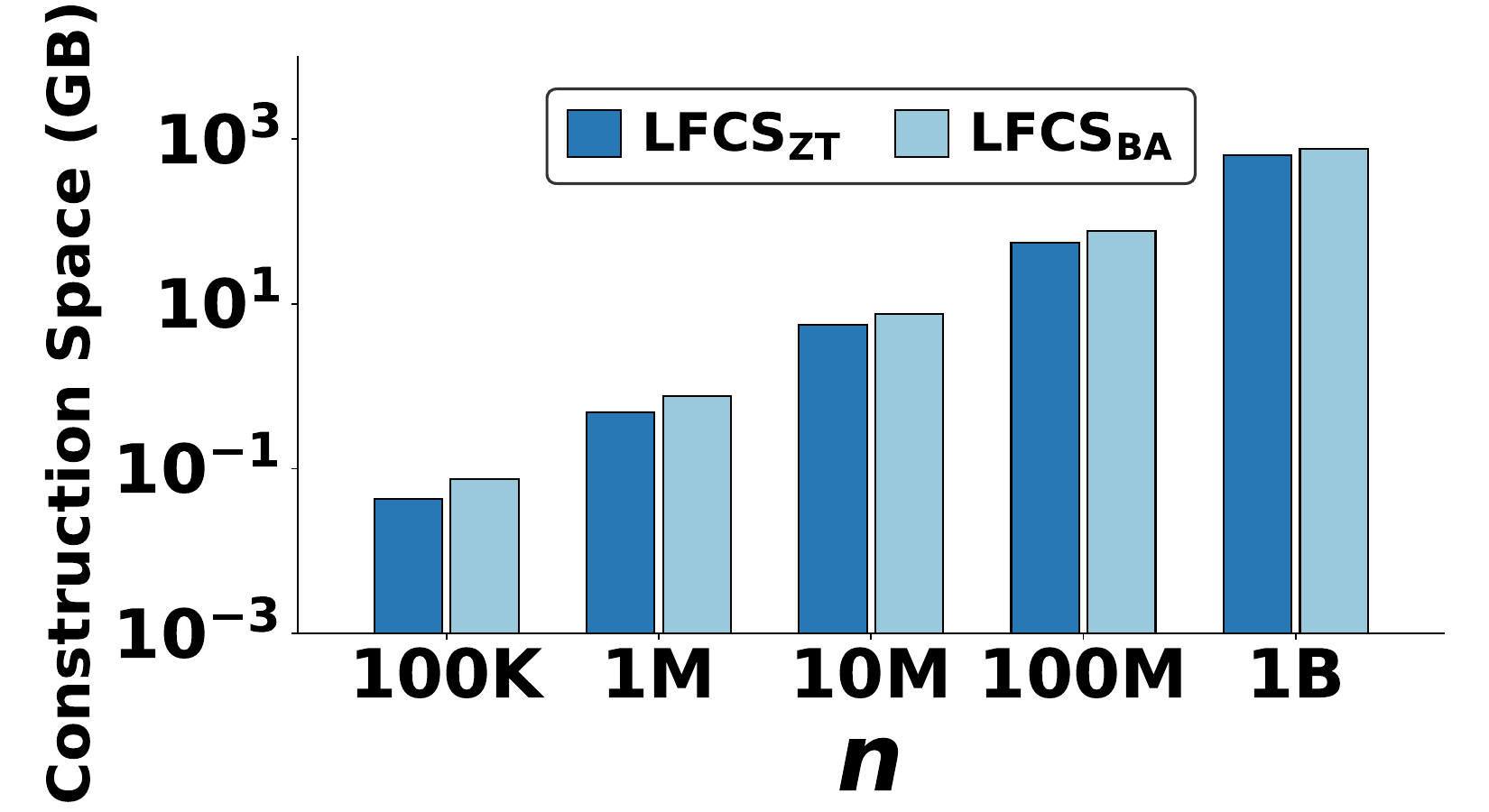}
    \caption{Constr.\ space vs. $n$}\label{fig:app:LFCS:n:rss:BST}
  \end{subfigure}
  \begin{subfigure}[t]{\appfigwidth}
    \includegraphics[width=\linewidth]{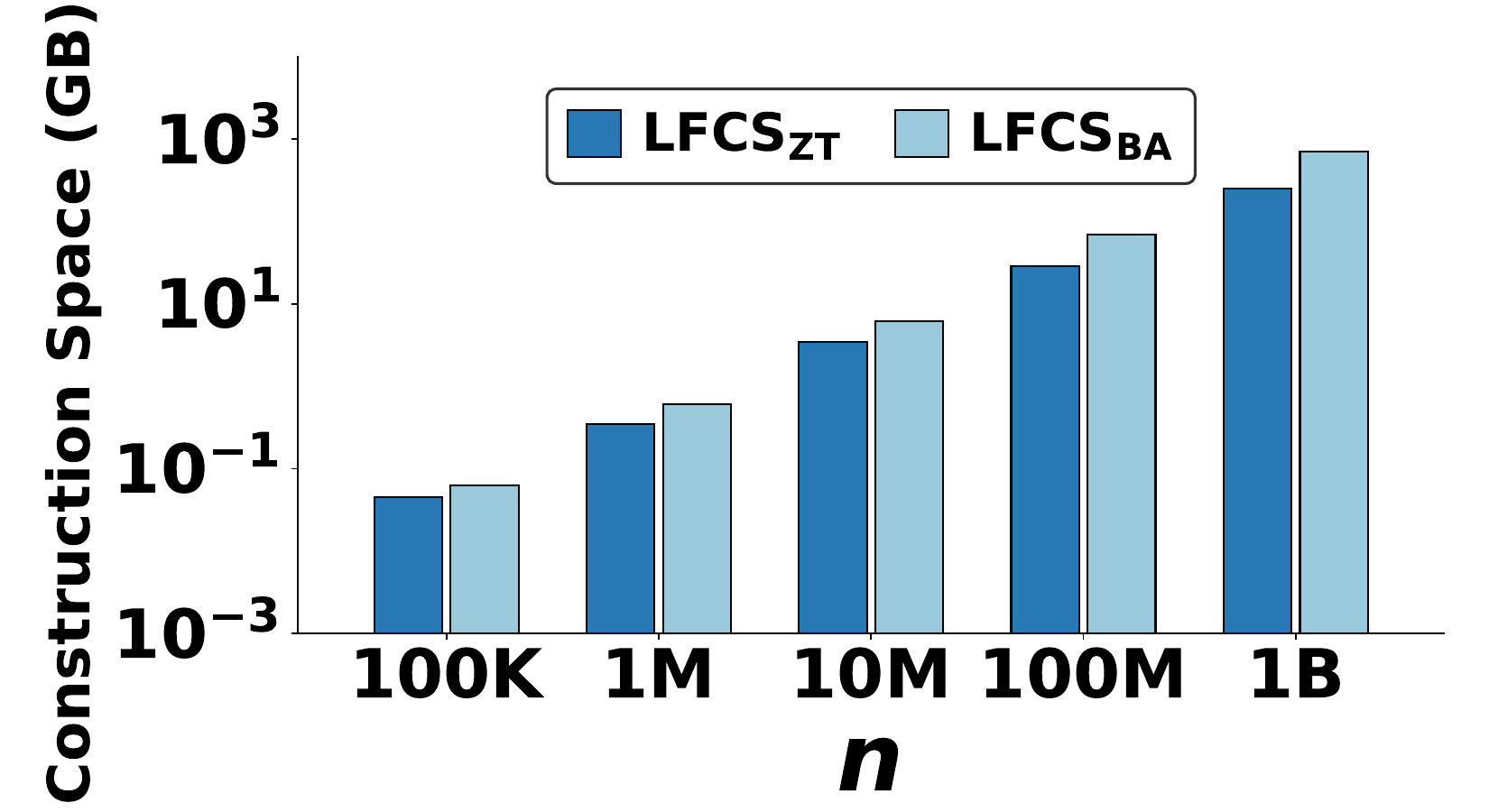}
    \caption{Constr.\ space vs. $n$}\label{fig:app:LFCS:n:rss:CHR}
  \end{subfigure}
  \begin{subfigure}[t]{\appfigwidth}
    \includegraphics[width=\linewidth]{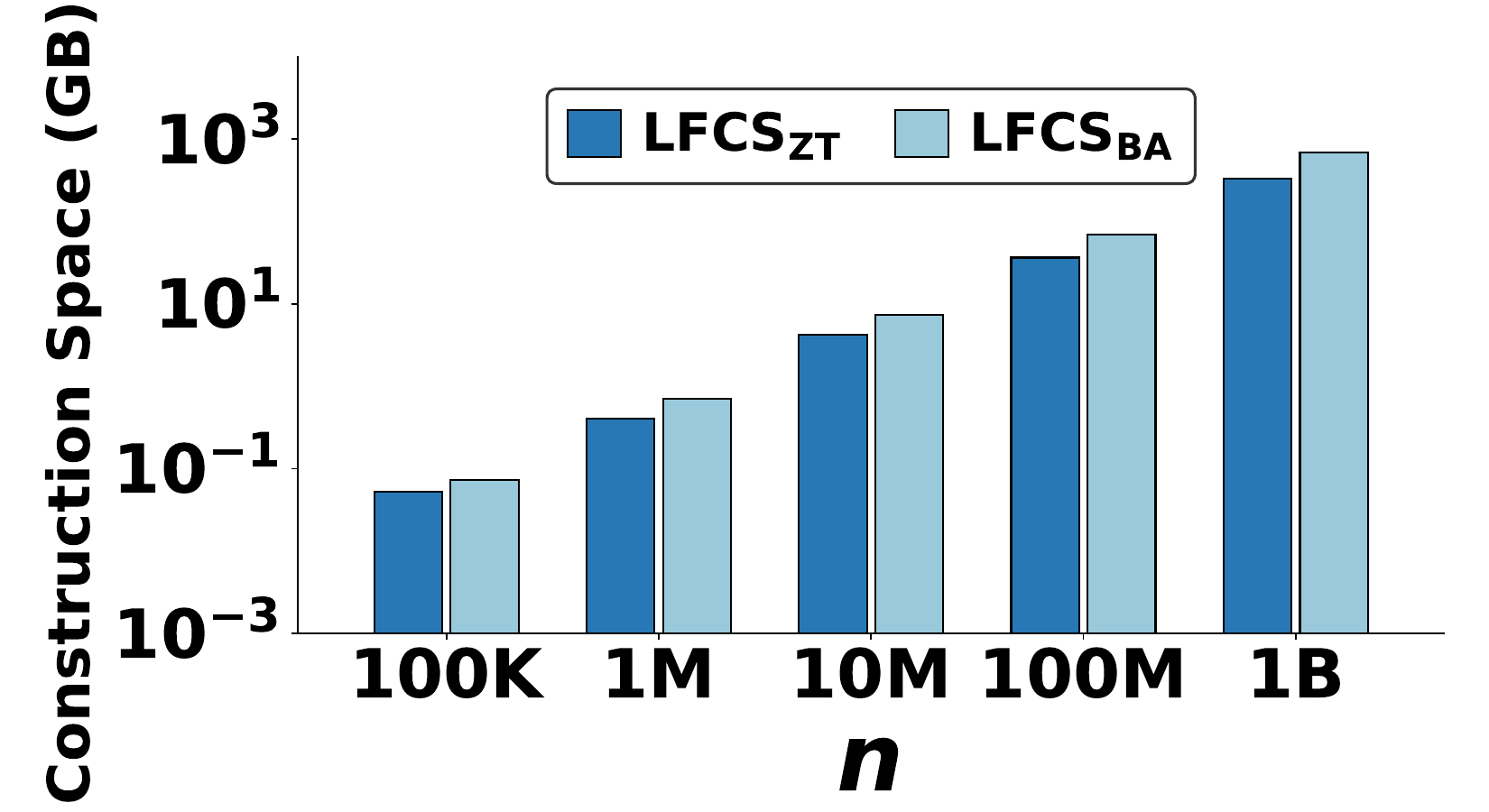}
    \caption{Constr.\ space vs. $n$}\label{fig:app:LFCS:n:rss:SARS}
  \end{subfigure}
  \begin{subfigure}[t]{\appfigwidth}
    \includegraphics[width=\linewidth]{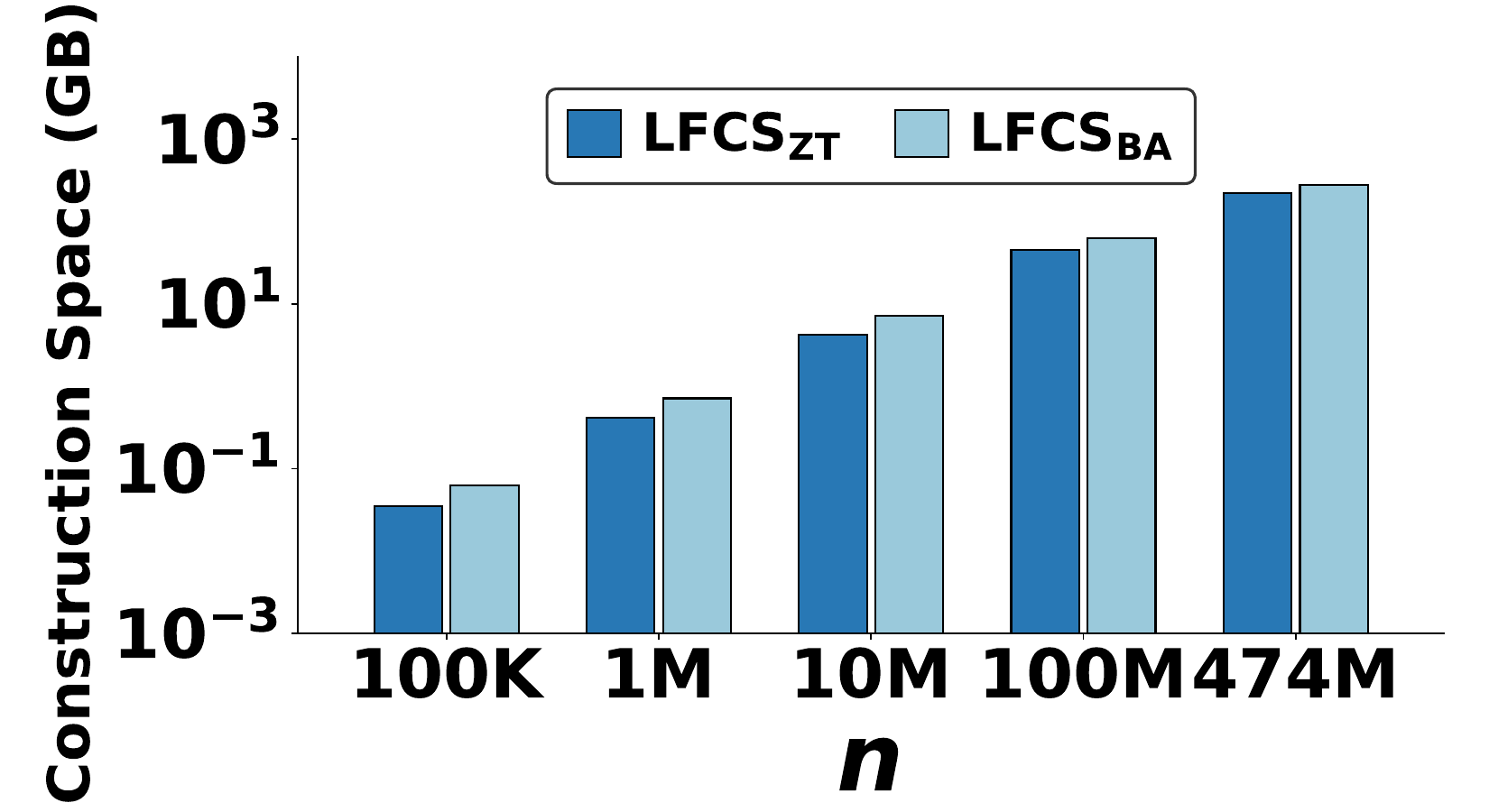}
    \caption{Constr.\ space vs. $n$}\label{fig:app:LFCS:n:rss:WIKI}
  \end{subfigure}\\[0pt]
  \begin{subfigure}[t]{\appfigwidth}
    \includegraphics[width=\linewidth]{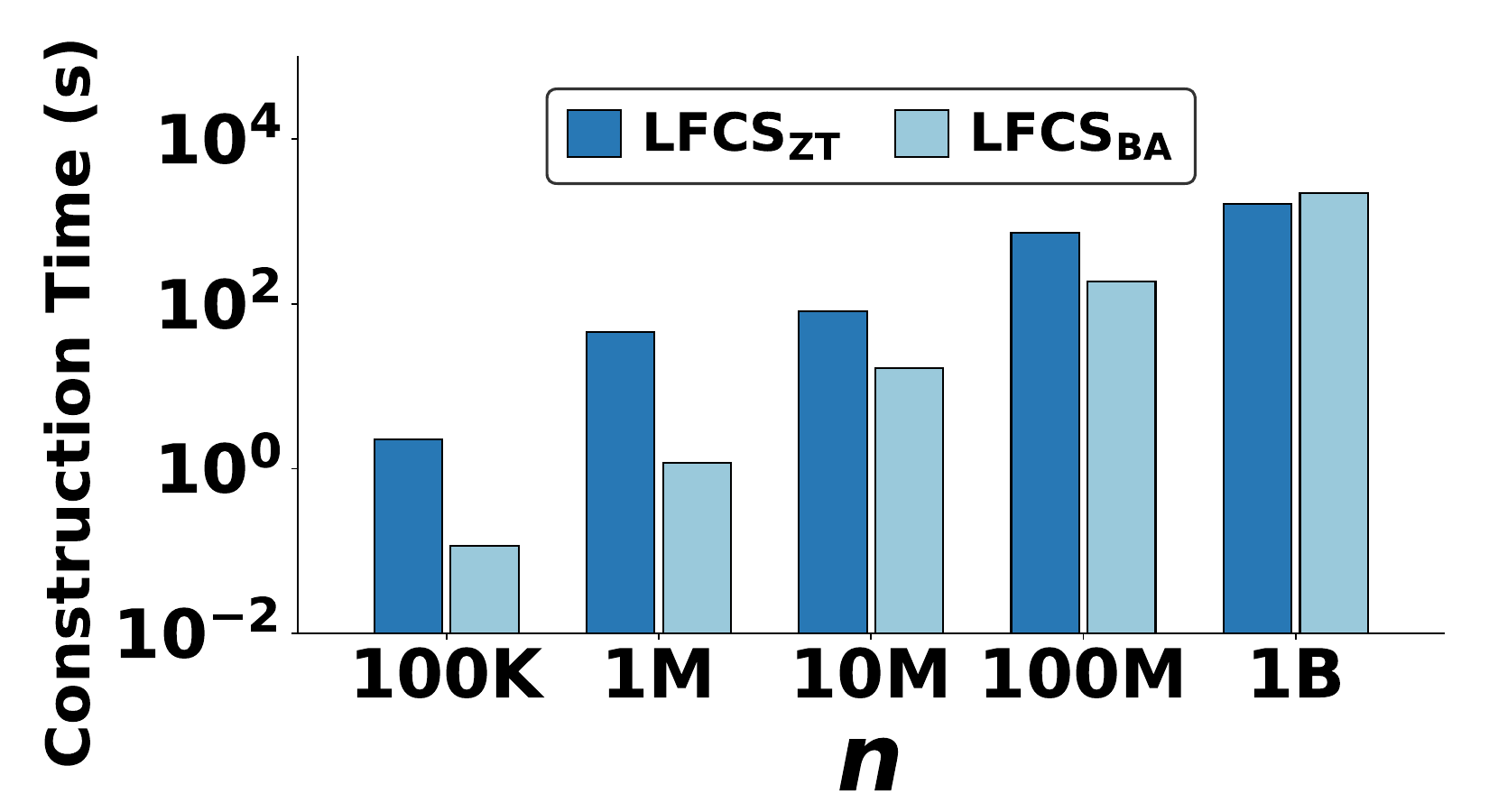}
    \caption{Constr.\ time vs. $n$}\label{fig:app:LFCS:n:build:BST}
  \end{subfigure}
  \begin{subfigure}[t]{\appfigwidth}
    \includegraphics[width=\linewidth]{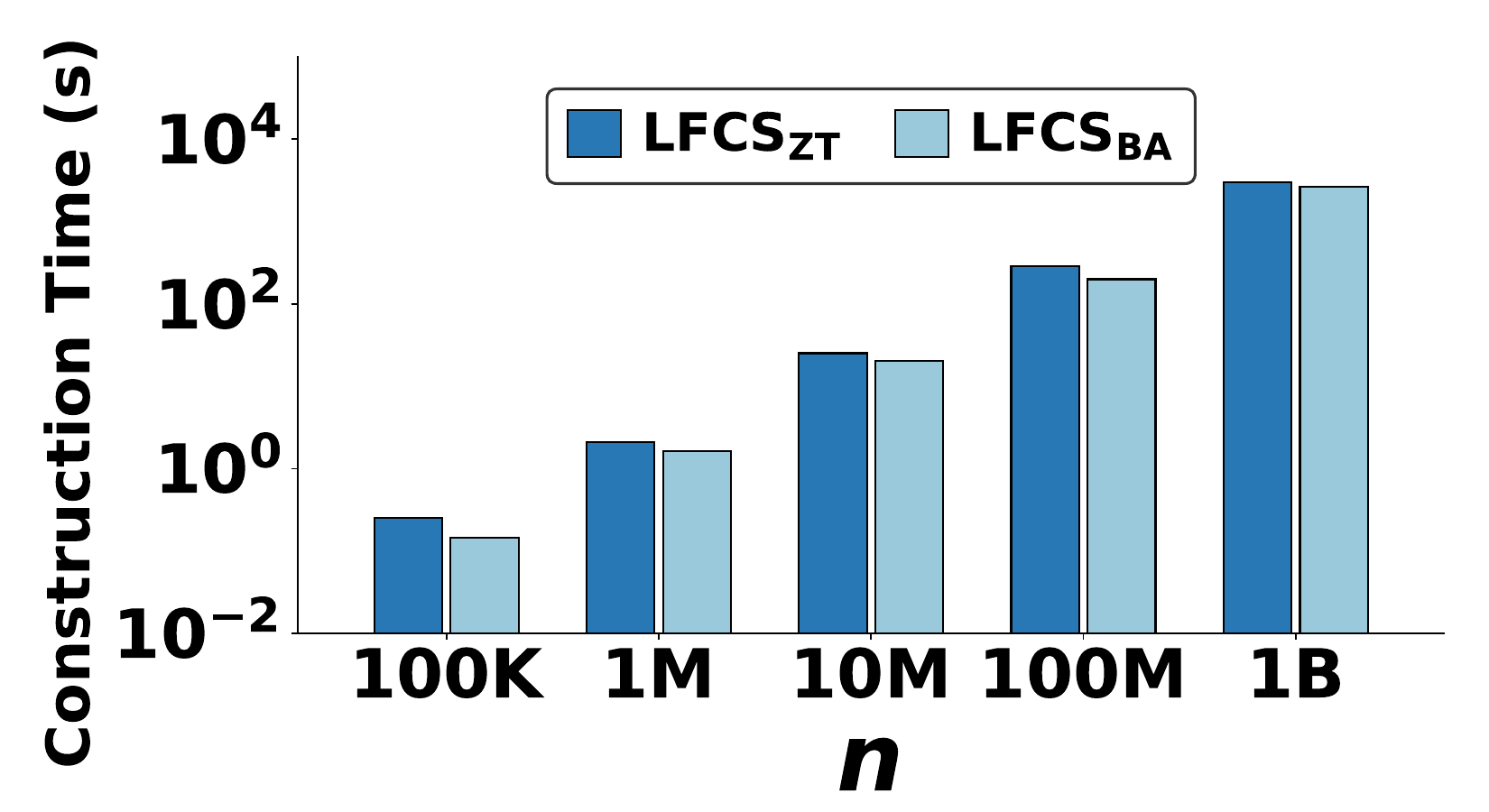}
    \caption{Constr.\ time vs. $n$}\label{fig:app:LFCS:n:build:CHR}
  \end{subfigure}
  \begin{subfigure}[t]{\appfigwidth}
    \includegraphics[width=\linewidth]{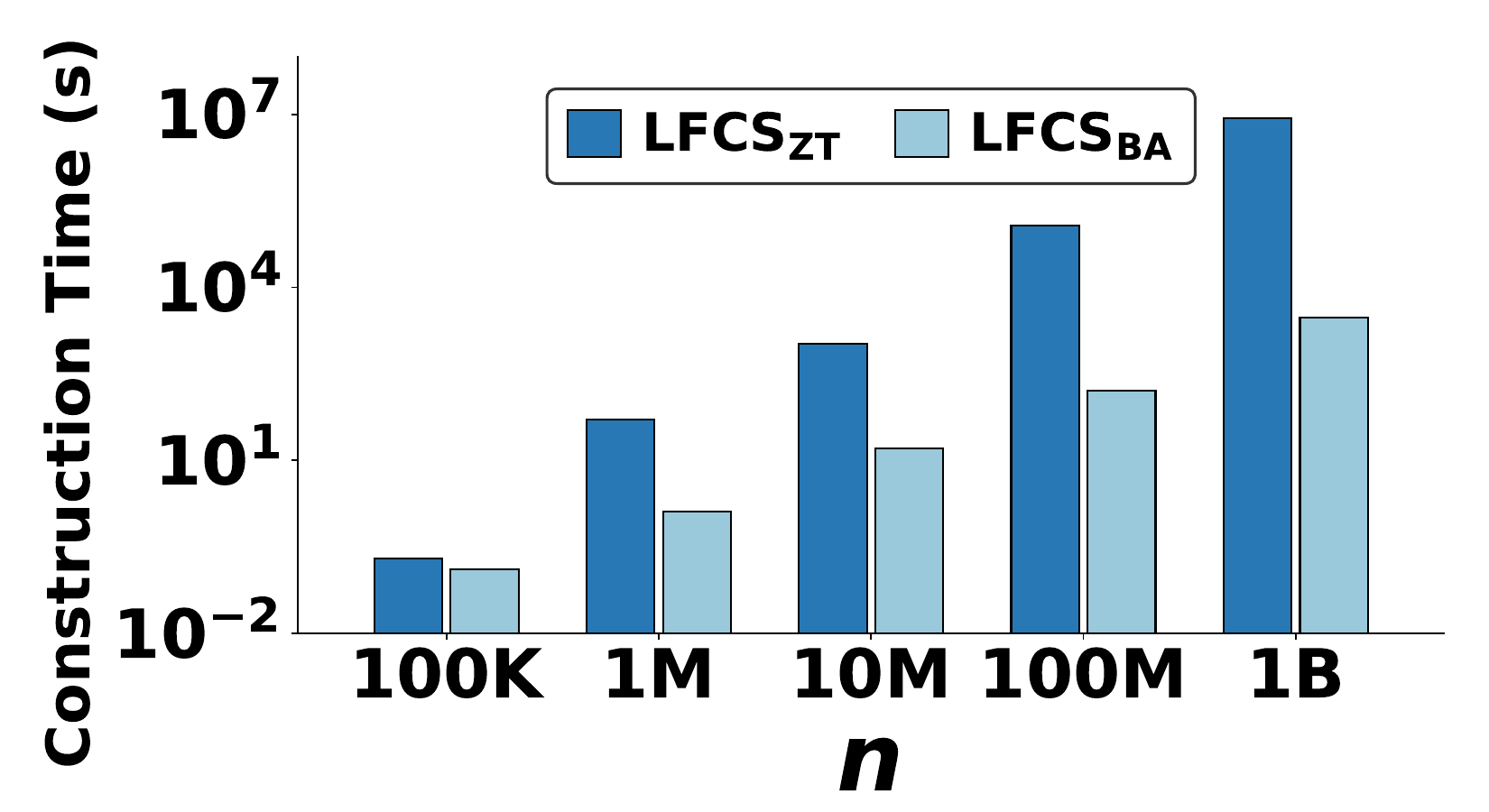}
    \caption{Constr.\ time vs. $n$}\label{fig:app:LFCS:n:build:SARS}
  \end{subfigure}
  \begin{subfigure}[t]{\appfigwidth}
    \includegraphics[width=\linewidth]{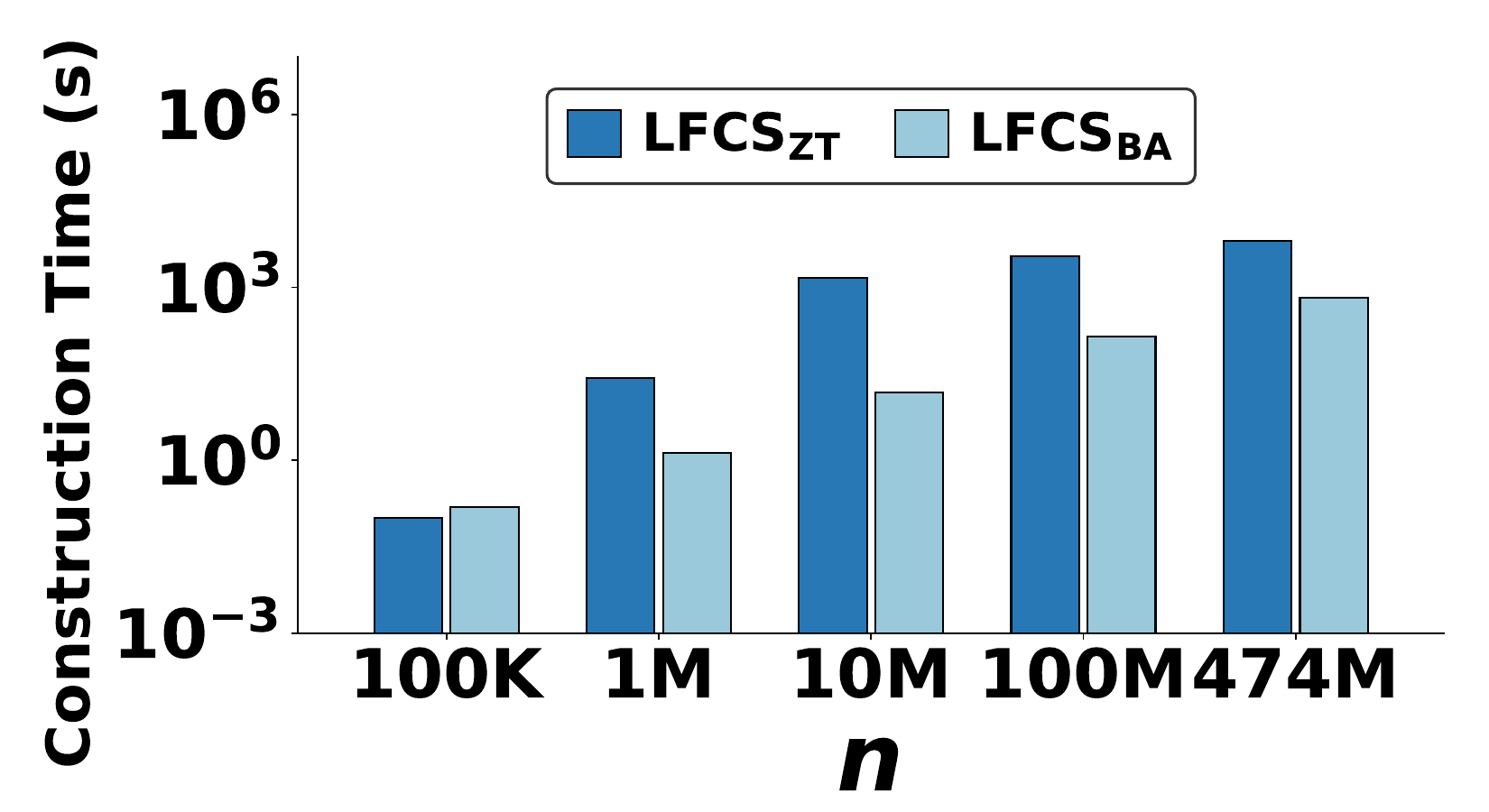}
    \caption{Constr.\ time vs. $n$}\label{fig:app:LFCS:n:build:WIKI}
  \end{subfigure}
  \vspace{\captionspacing}
  \vspace{+2mm}
  \caption{Index size of our \LFCS index vs. \LFCSBA on (a) \bst, (b) \chr, (c) \sars, and (d) \wiki vs. $n$; construction space of our \LFCS index vs. \LFCSBA on (e) \bst, (f) \chr, (g) \sars, and (h) \wiki vs. $n$; construction time of our \LFCS index vs. \LFCSBA on (i) \bst, (j) \chr, (k) \sars, and (l) \wiki vs. $n$.}\label{fig:app:LFCS:cost}
\end{figure}

\begin{figure}[ht]
  \centering
  \begin{subfigure}[t]{\appfigwidth}
    \includegraphics[width=\linewidth]{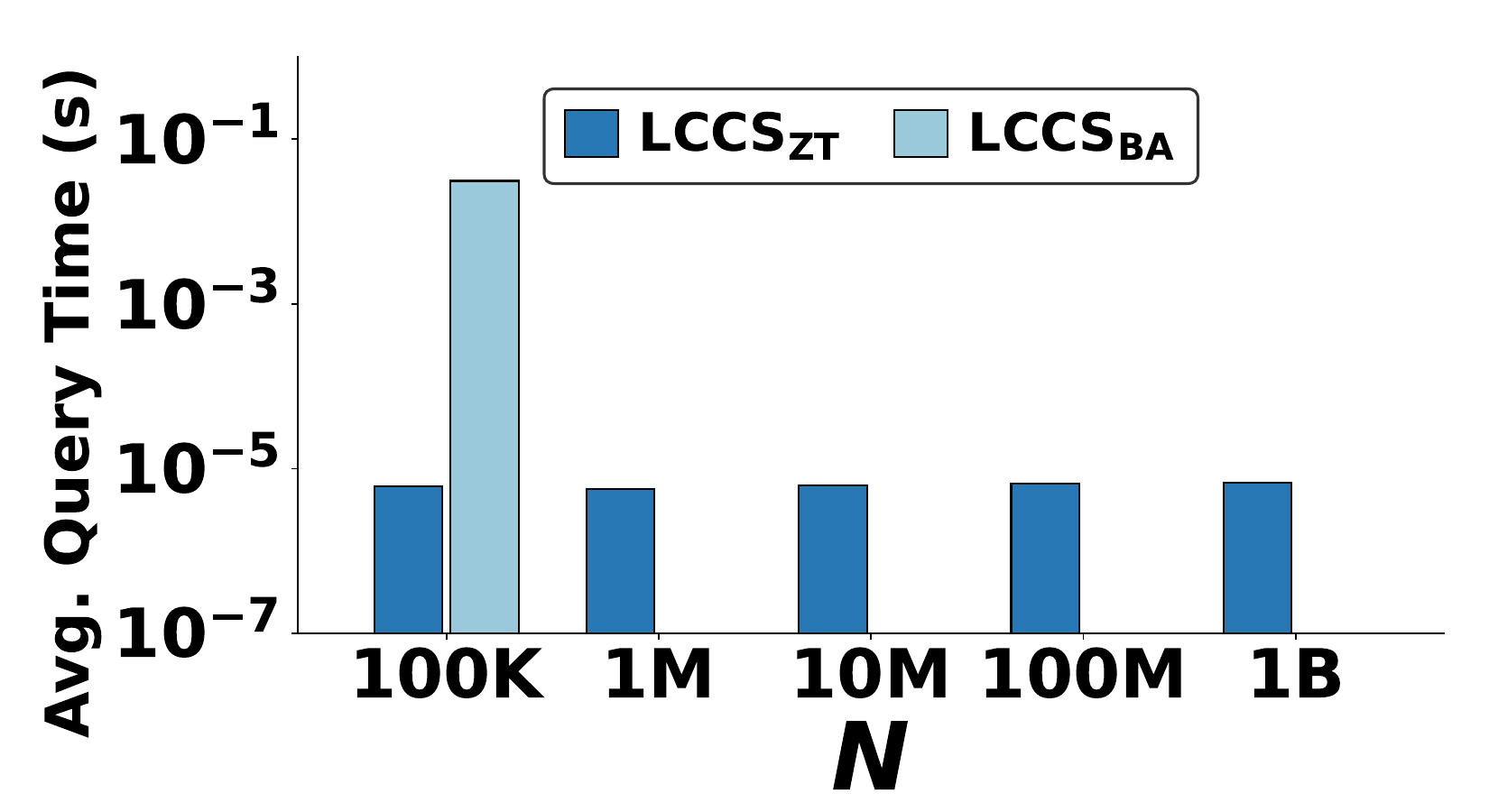}
    \caption{Query time vs. $N$}\label{fig:app:LCCS:n:query:BST}
  \end{subfigure}
  \begin{subfigure}[t]{\appfigwidth}
    \includegraphics[width=\linewidth]{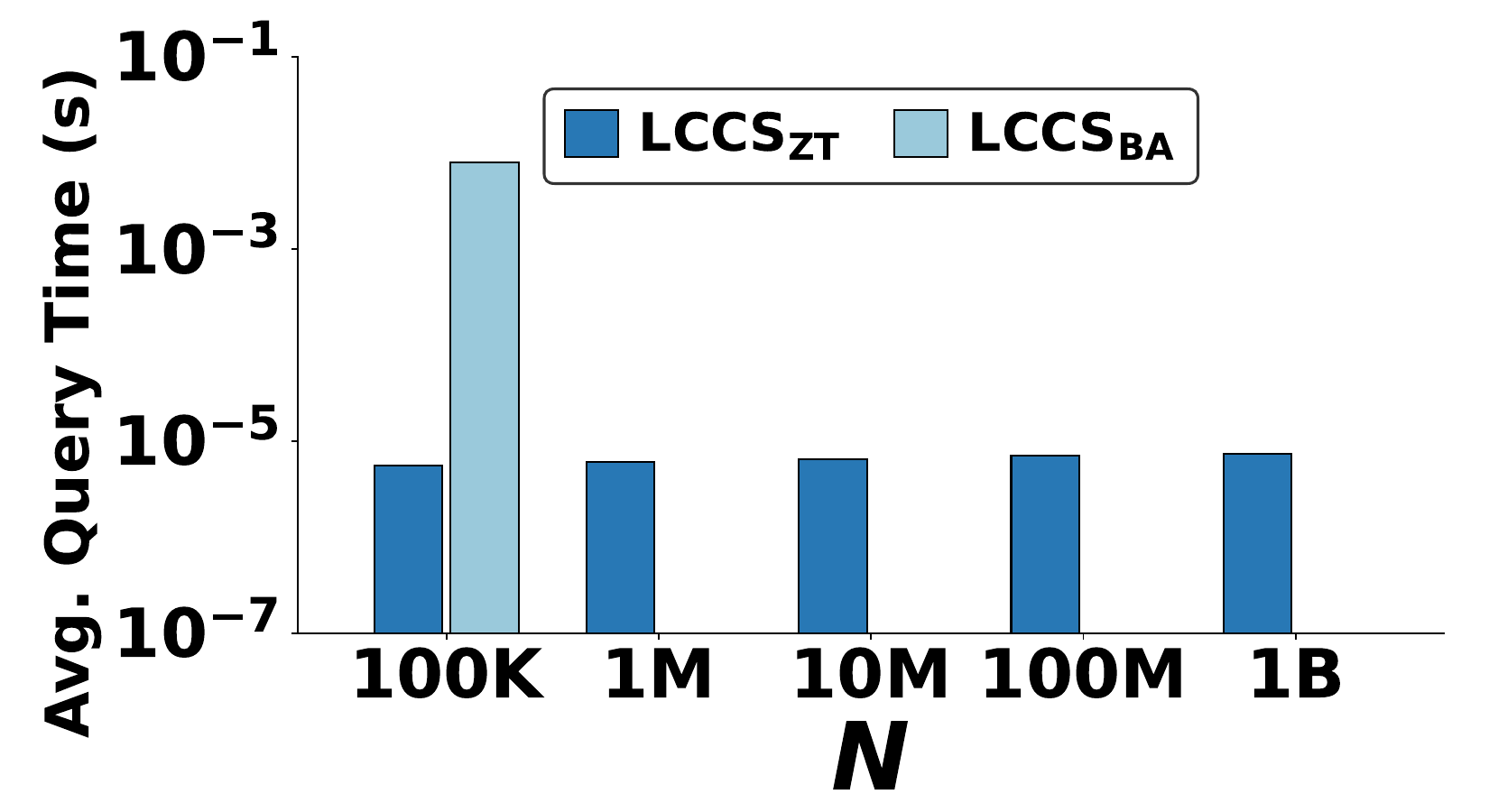}
    \caption{Query time vs. $N$}\label{fig:app:LCCS:n:query:CHR}
  \end{subfigure}
  \begin{subfigure}[t]{\appfigwidth}
    \includegraphics[width=\linewidth]{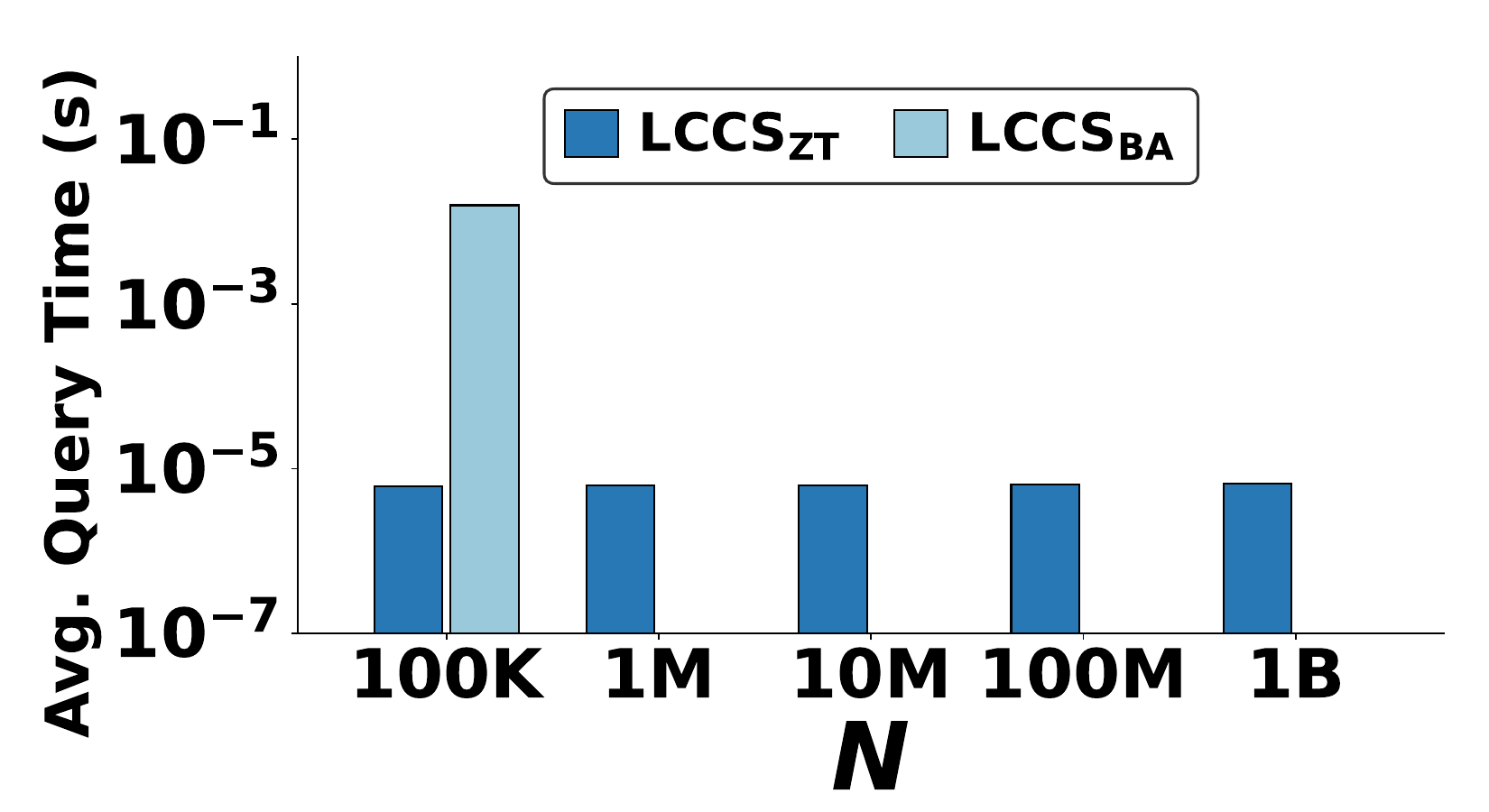}
    \caption{Query time vs. $N$}\label{fig:app:LCCS:n:query:SDSL}
  \end{subfigure}
  \begin{subfigure}[t]{\appfigwidth}
    \includegraphics[width=\linewidth]{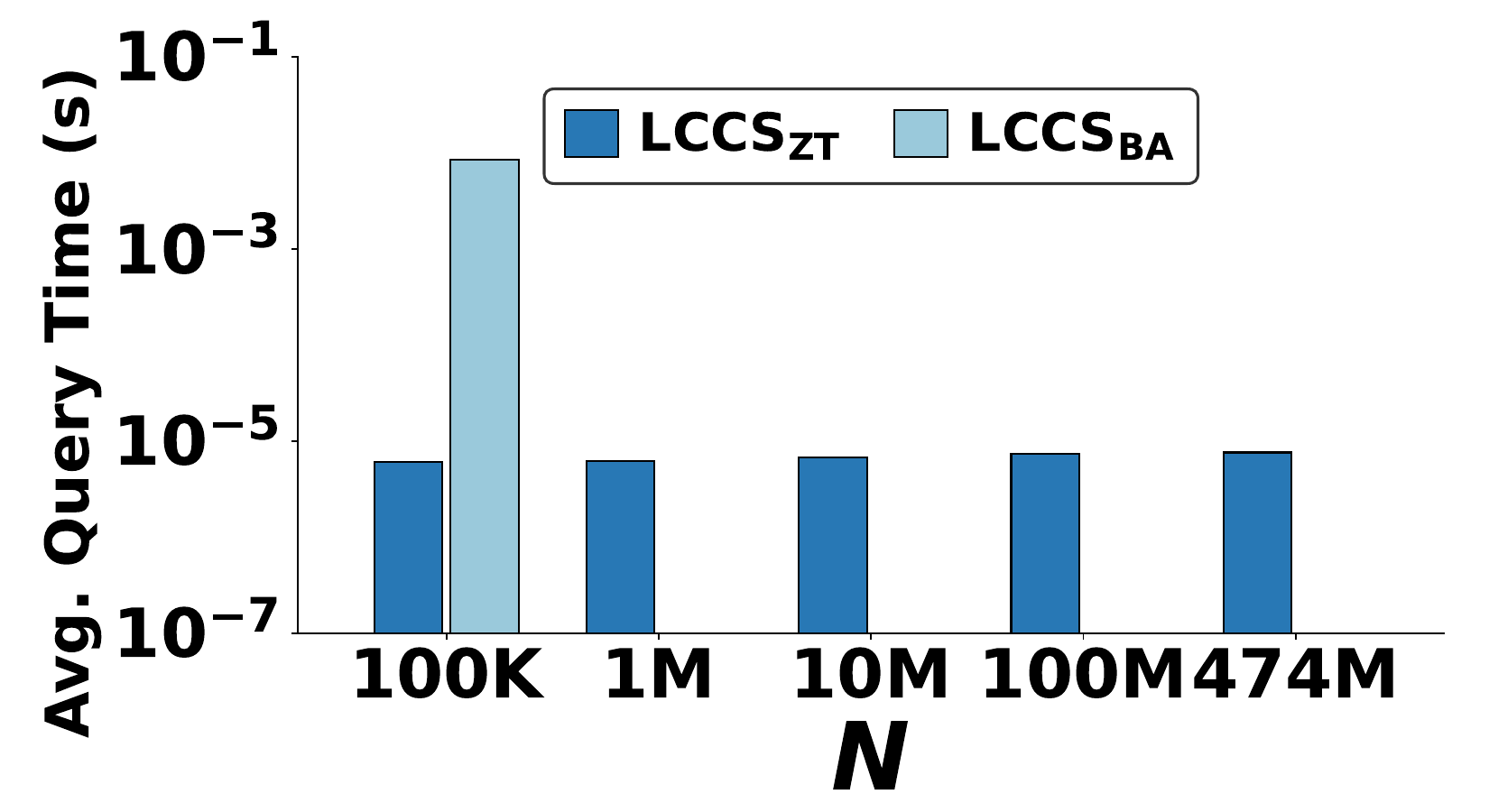}
    \caption{Query time vs. $N$}\label{fig:app:LCCS:n:query:WIKI}
  \end{subfigure}\\[0pt]
  \begin{subfigure}[t]{\appfigwidth}
    \includegraphics[width=\linewidth]{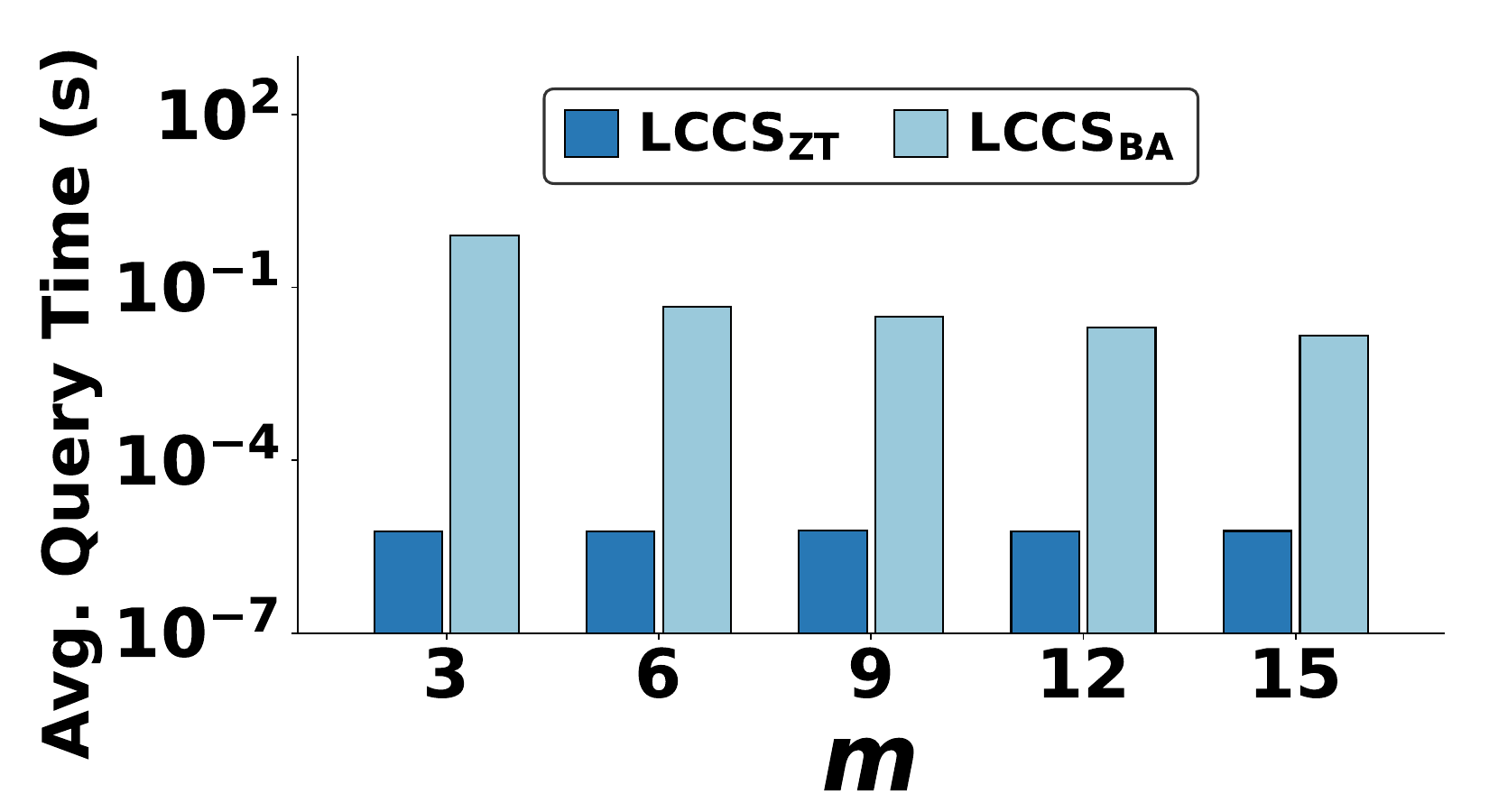}
    \caption{Query time vs. $m$}\label{fig:app:LCCS:m:query:BST}
  \end{subfigure}
  \begin{subfigure}[t]{\appfigwidth}
    \includegraphics[width=\linewidth]{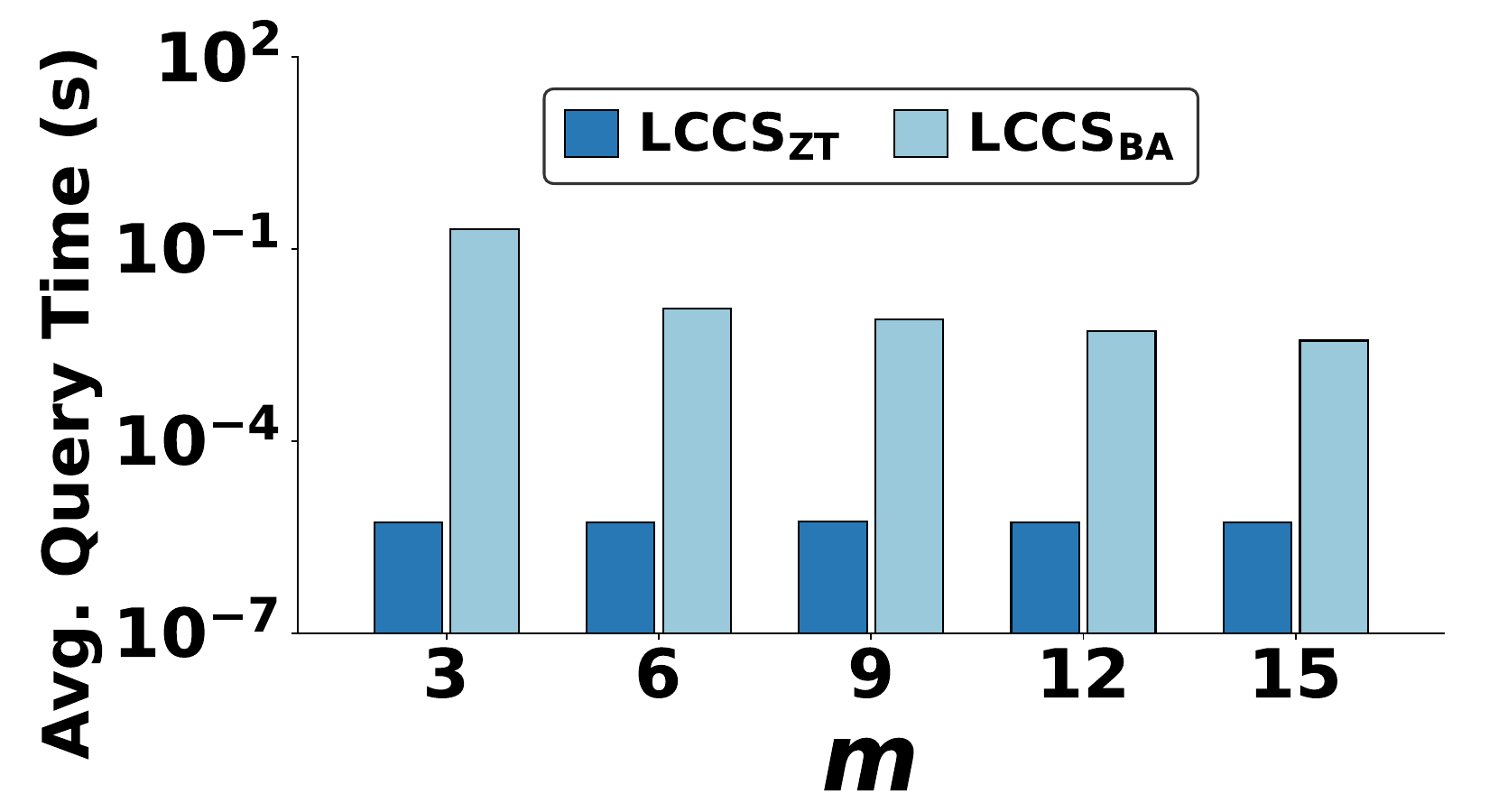}
    \caption{Query time vs. $m$}\label{fig:app:LCCS:m:query:CHR}
  \end{subfigure}
  \begin{subfigure}[t]{\appfigwidth}
    \includegraphics[width=\linewidth]{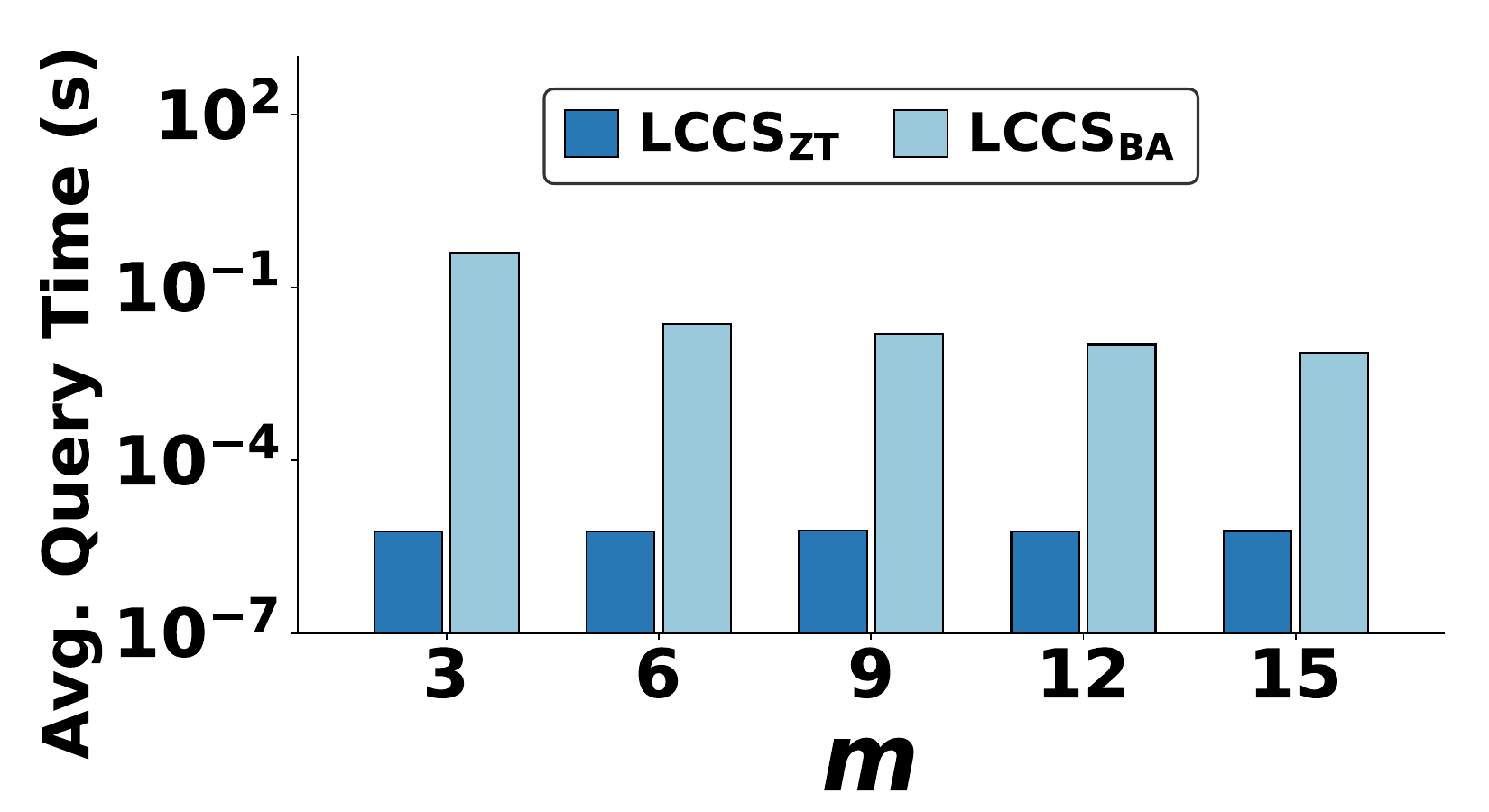}
    \caption{Query time vs. $m$}\label{fig:app:LCCS:m:query:SDSL}
  \end{subfigure}
  \begin{subfigure}[t]{\appfigwidth}
    \includegraphics[width=\linewidth]{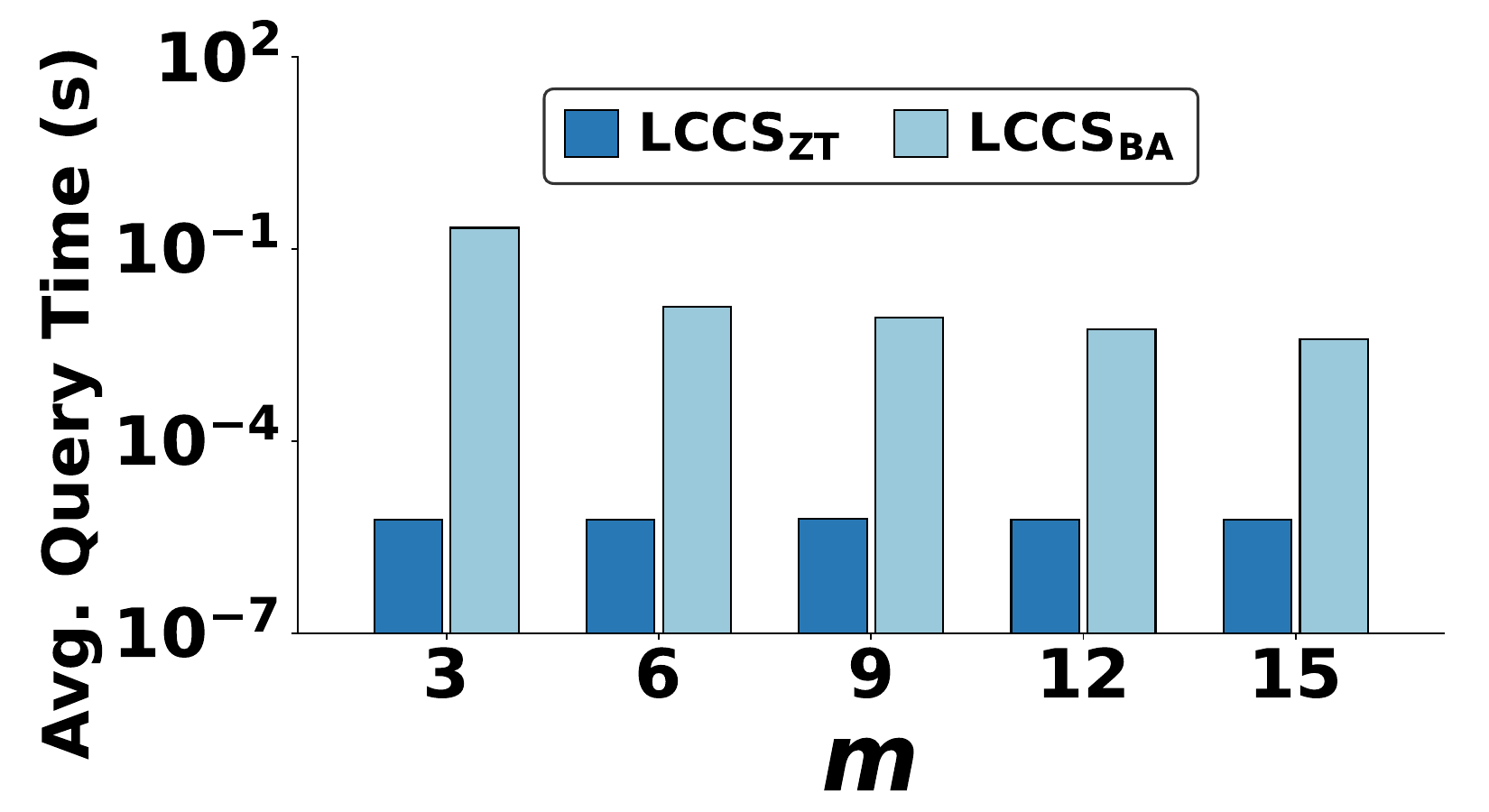}
    \caption{Query time vs. $m$}\label{fig:app:LCCS:m:query:WIKI}
  \end{subfigure}\\[0pt]
  \begin{subfigure}[t]{\appfigwidth}
    \includegraphics[width=\linewidth]{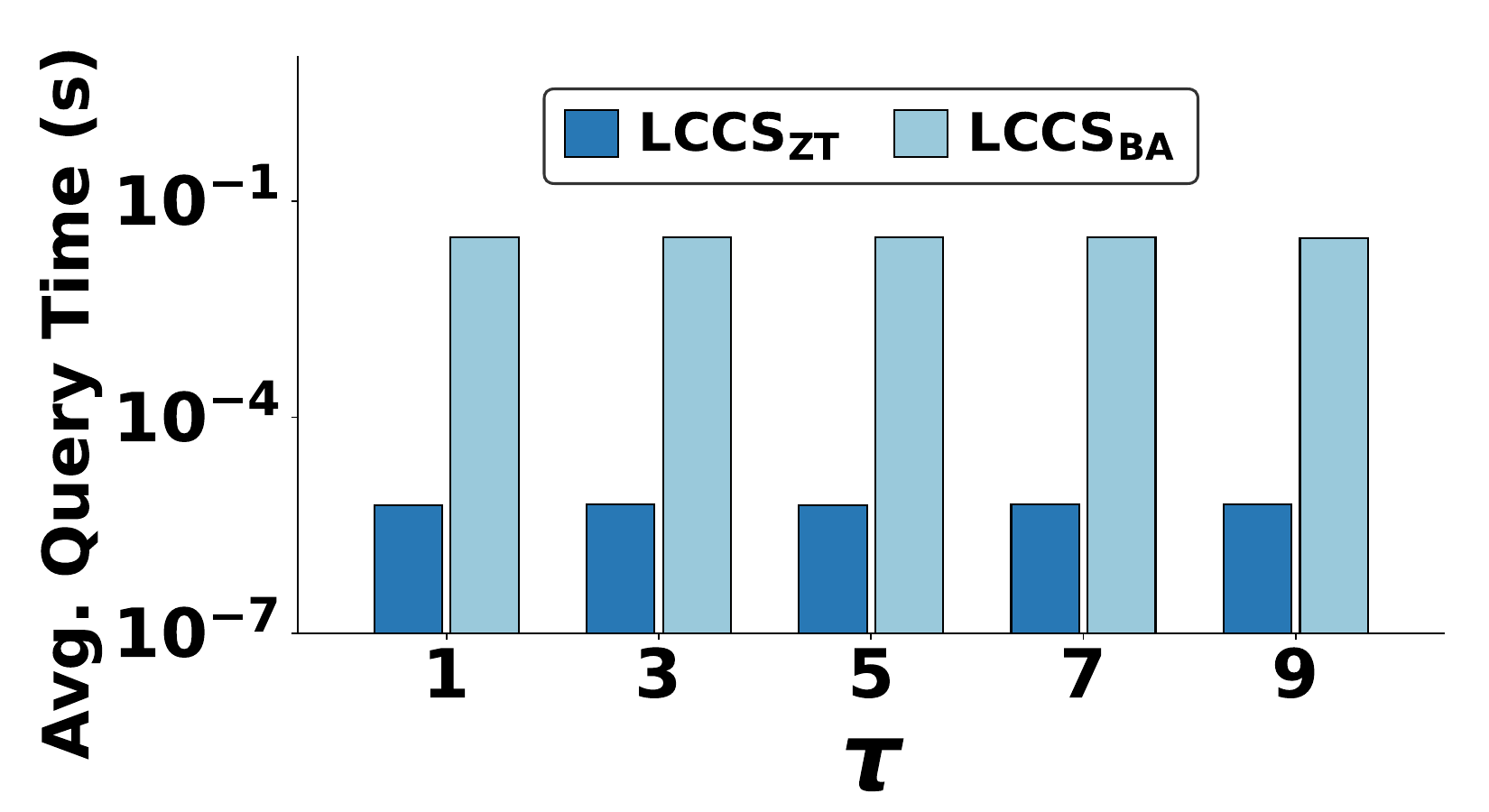}
    \caption{Query time vs. $\tau$}\label{fig:app:LCCS:tau:query:BST}
  \end{subfigure}
  \begin{subfigure}[t]{\appfigwidth}
    \includegraphics[width=\linewidth]{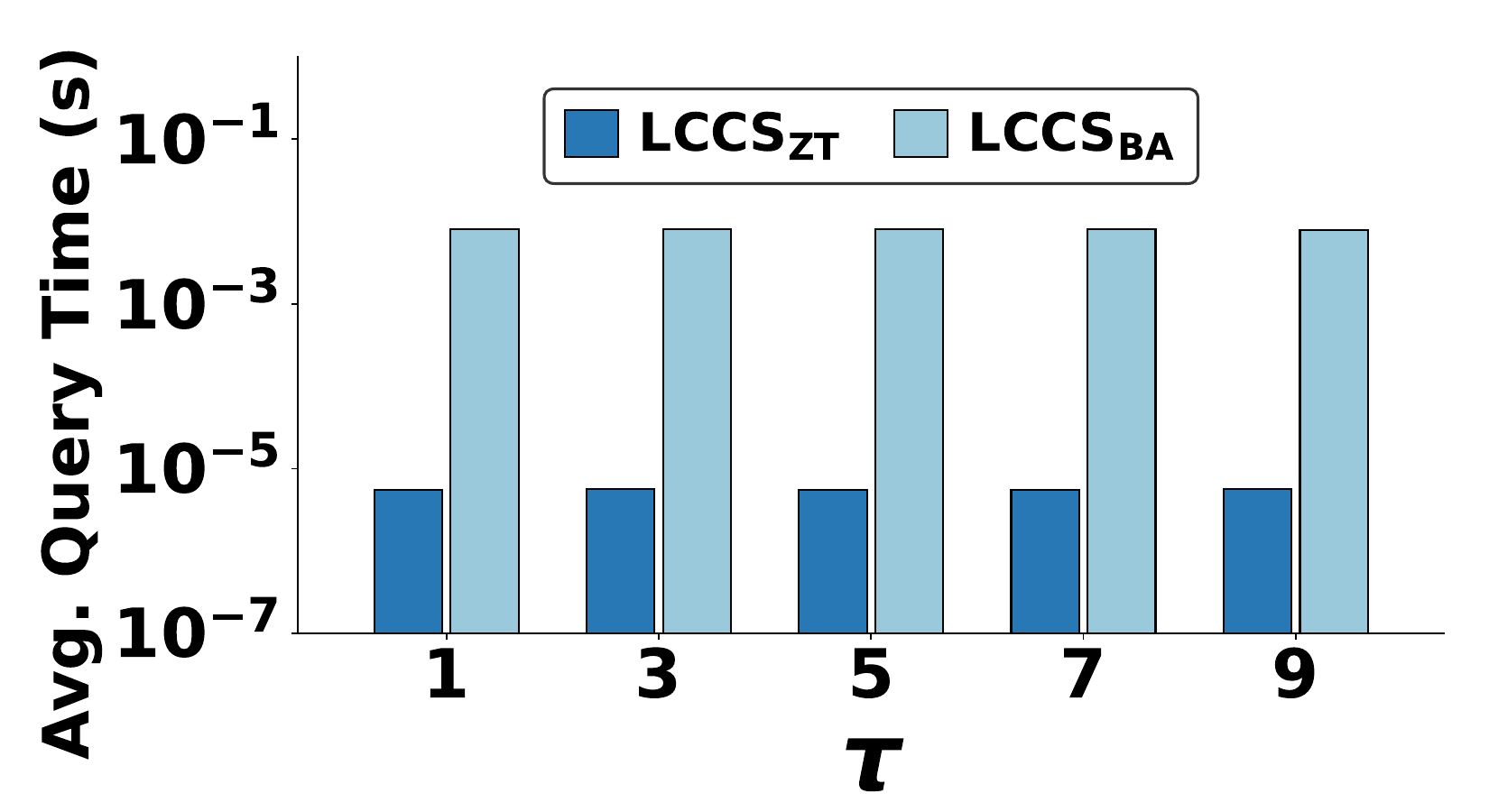}
    \caption{Query time vs. $\tau$}\label{fig:app:LCCS:tau:query:CHR}
  \end{subfigure}
  \begin{subfigure}[t]{\appfigwidth}
    \includegraphics[width=\linewidth]{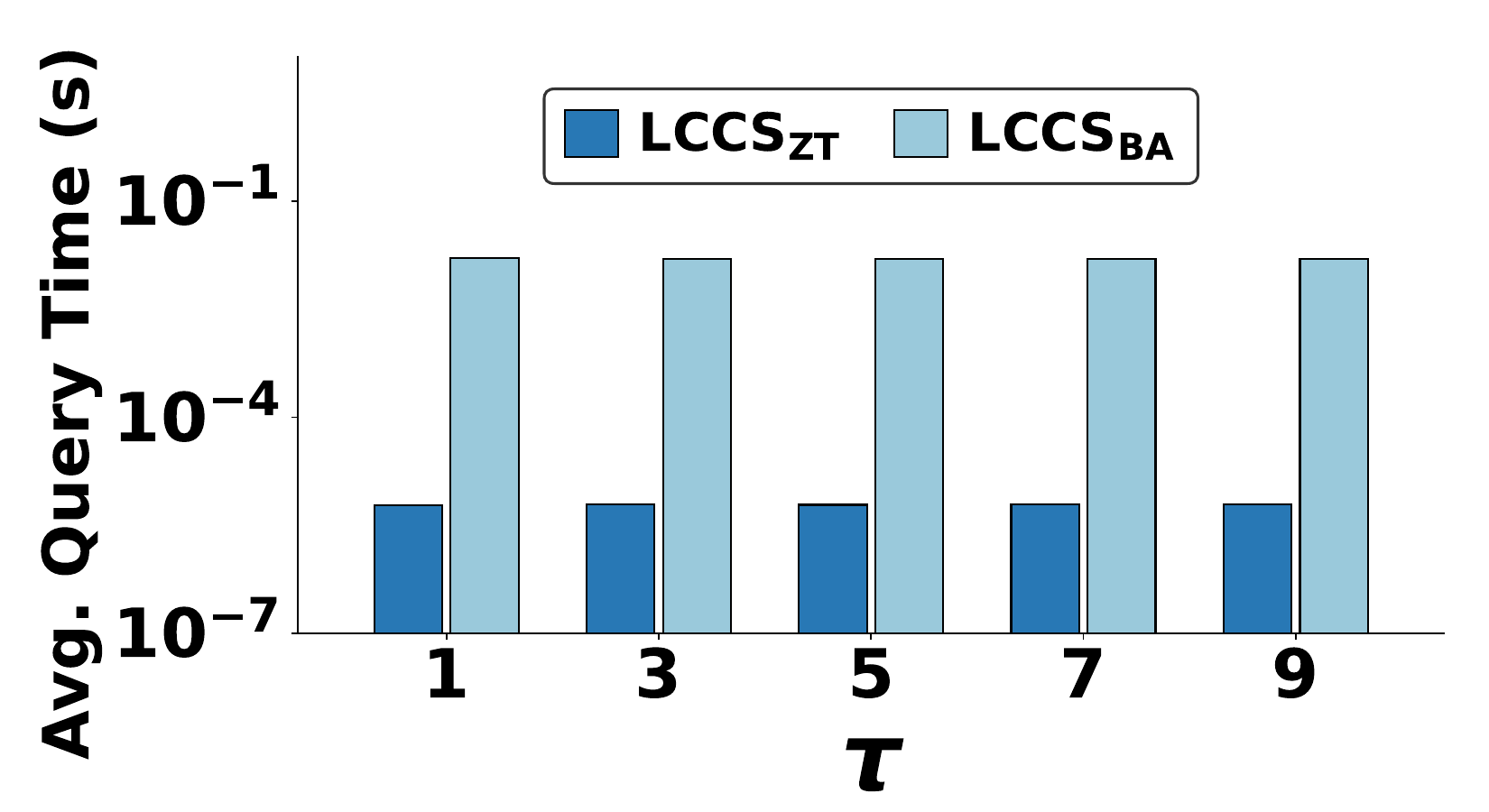}
    \caption{Query time vs. $\tau$}\label{fig:app:LCCS:tau:query:SDSL}
  \end{subfigure}
  \begin{subfigure}[t]{\appfigwidth}
    \includegraphics[width=\linewidth]{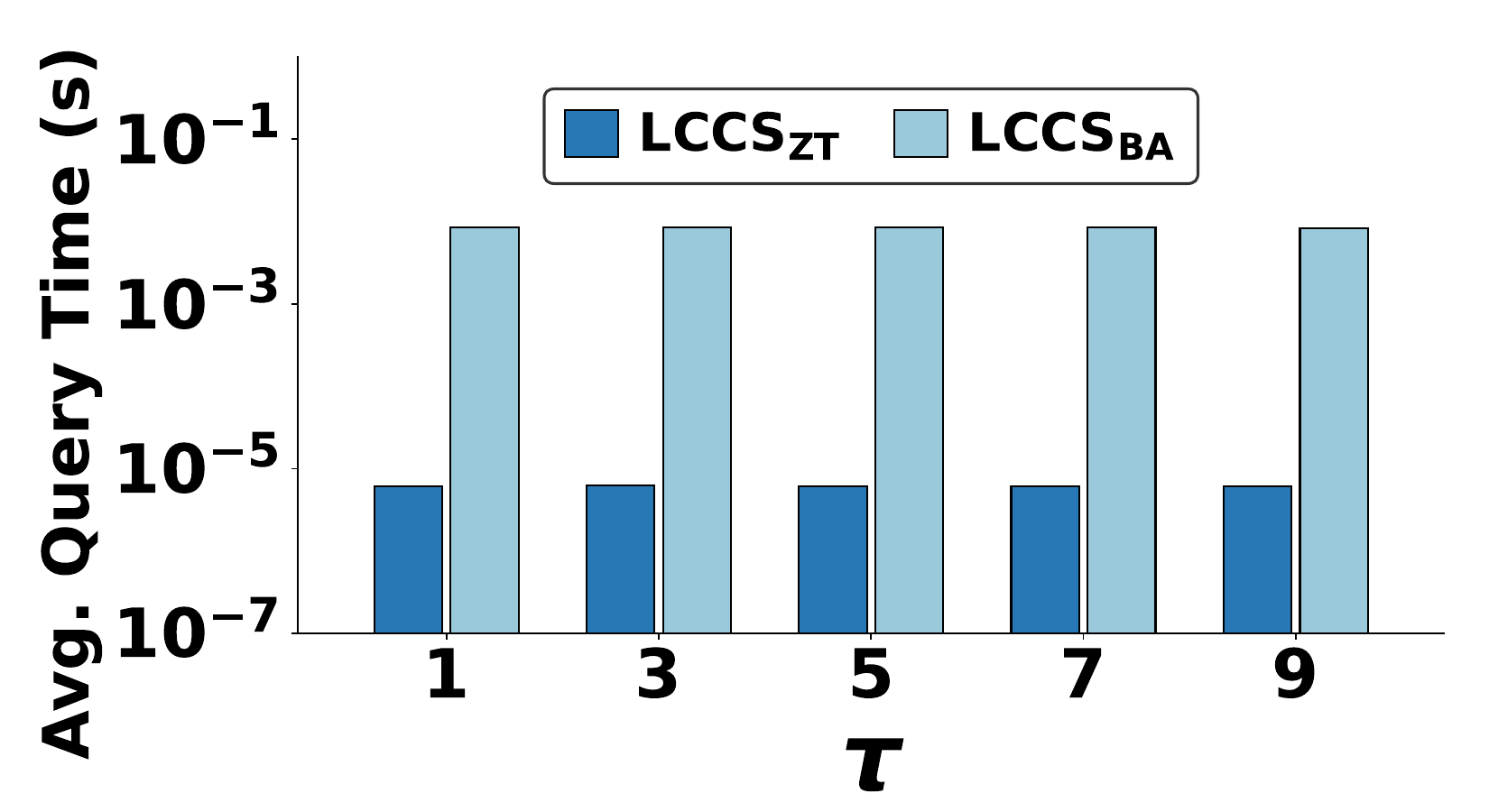}
    \caption{Query time vs. $\tau$}\label{fig:app:LCCS:tau:query:WIKI}
  \end{subfigure}
  \vspace{\captionspacing}
  \vspace{+2mm}
  \caption{Query time of our \LCCS index vs. \LCCSBA on (a) \bst, (b) \chr, (c) \sdsl, and (d) \wiki vs. $N$; on (e) \bst, (f) \chr, (g) \sdsl, and (h) \wiki vs. $m$; on (i) \bst, (j) \chr, (k) \sdsl, and (l) \wiki vs. $\tau$.}\label{fig:app:LCCS:query}
\end{figure}

\begin{figure}[ht]
  \centering
  \begin{subfigure}[t]{\appfigwidth}
    \includegraphics[width=\linewidth]{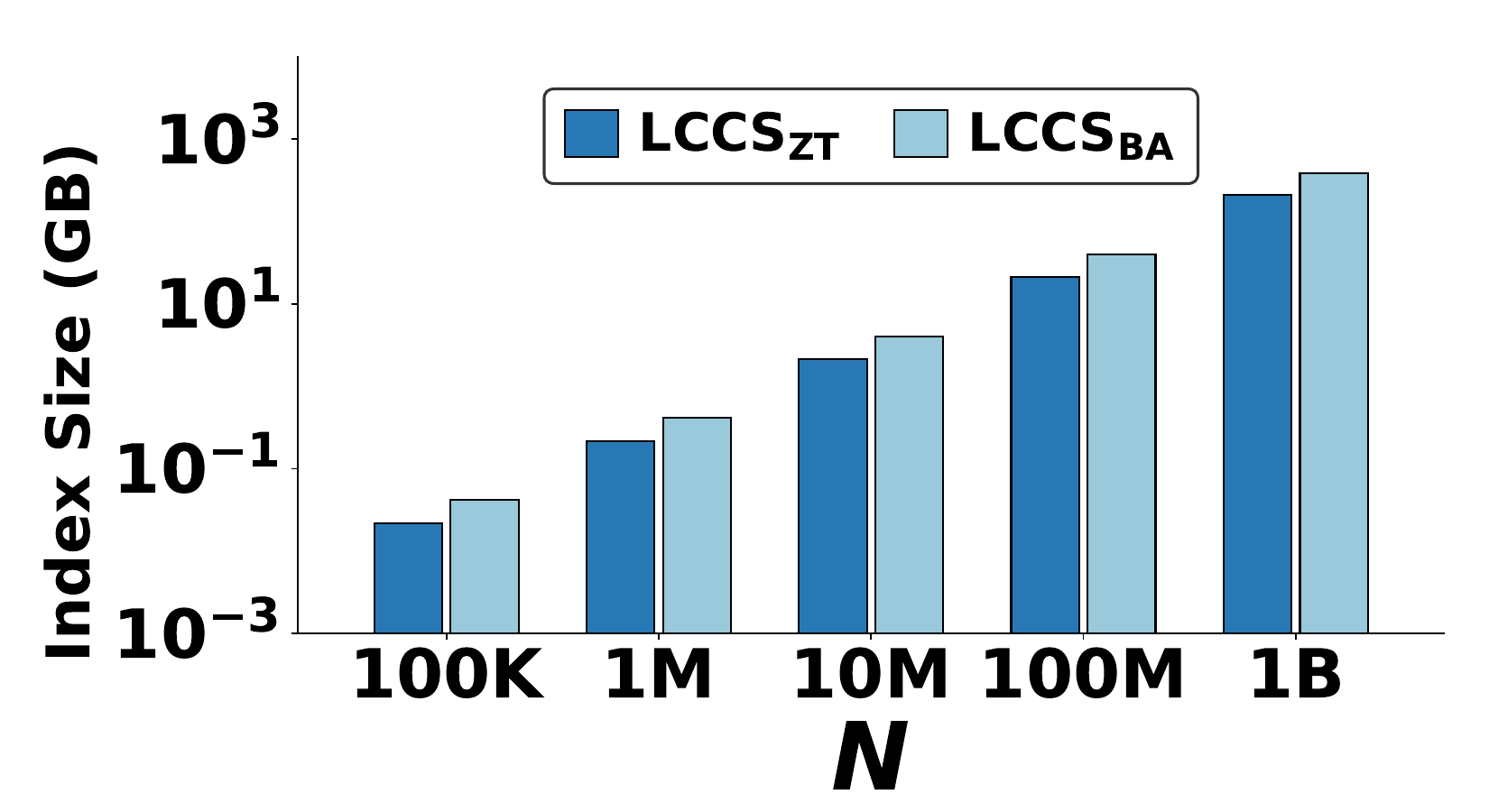}
    \caption{Index size vs. $N$}\label{fig:app:LCCS:n:index:BST}
  \end{subfigure}
  \begin{subfigure}[t]{\appfigwidth}
    \includegraphics[width=\linewidth]{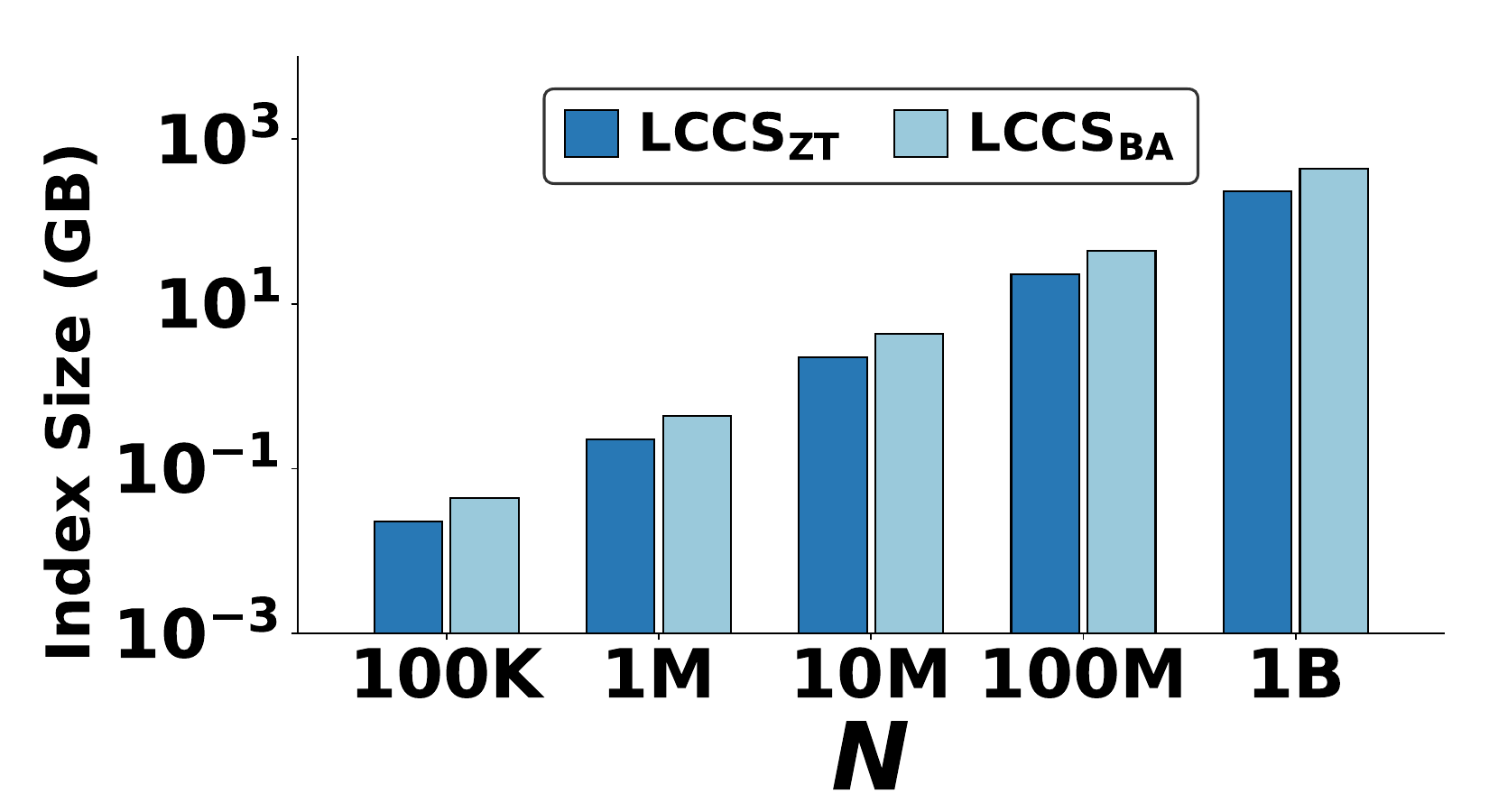}
    \caption{Index size vs. $N$}\label{fig:app:LCCS:n:index:CHR}
  \end{subfigure}
  \begin{subfigure}[t]{\appfigwidth}
    \includegraphics[width=\linewidth]{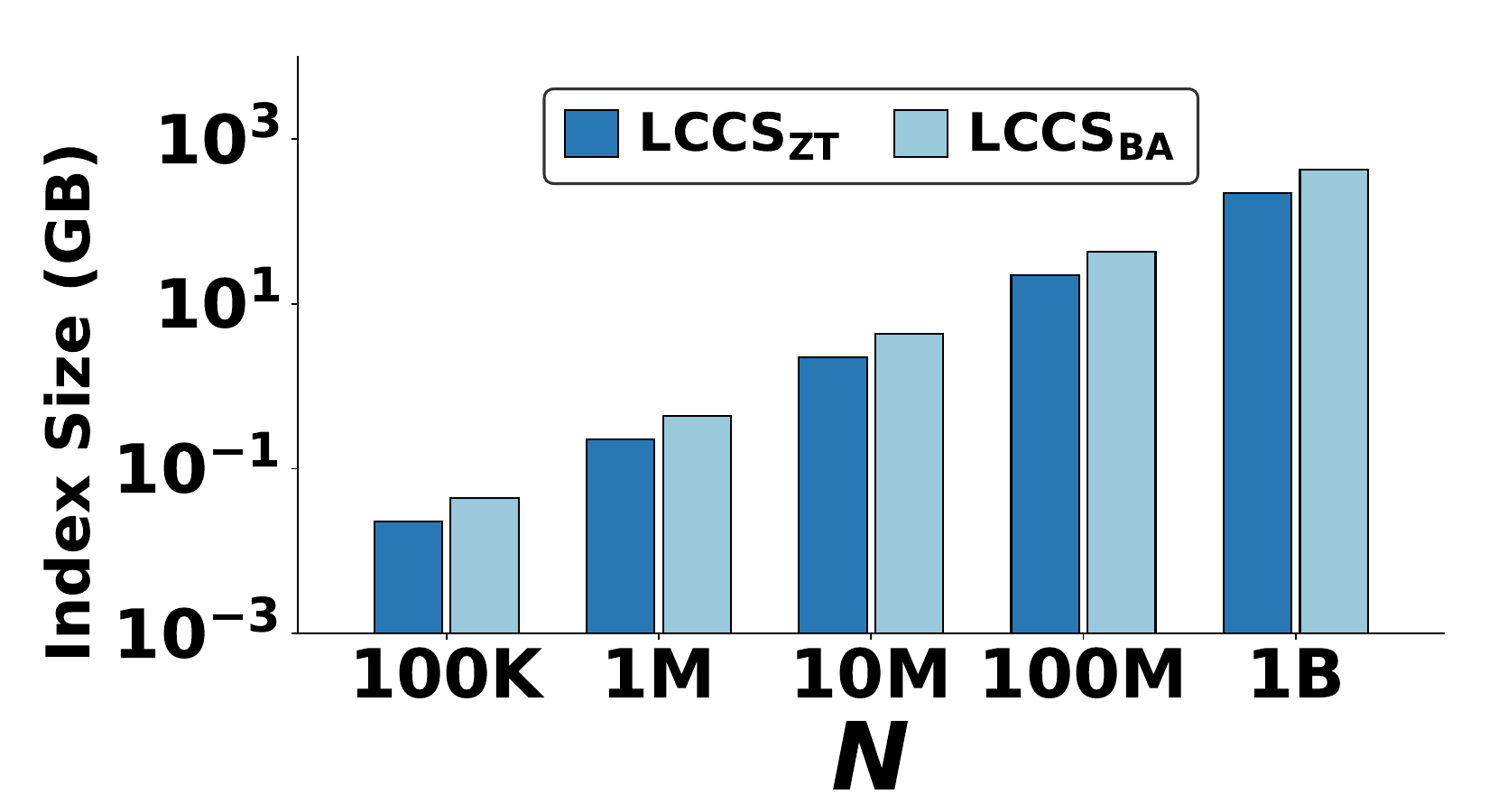}
    \caption{Index size vs. $N$}\label{fig:app:LCCS:n:index:SDSL}
  \end{subfigure}
  \begin{subfigure}[t]{\appfigwidth}
    \includegraphics[width=\linewidth]{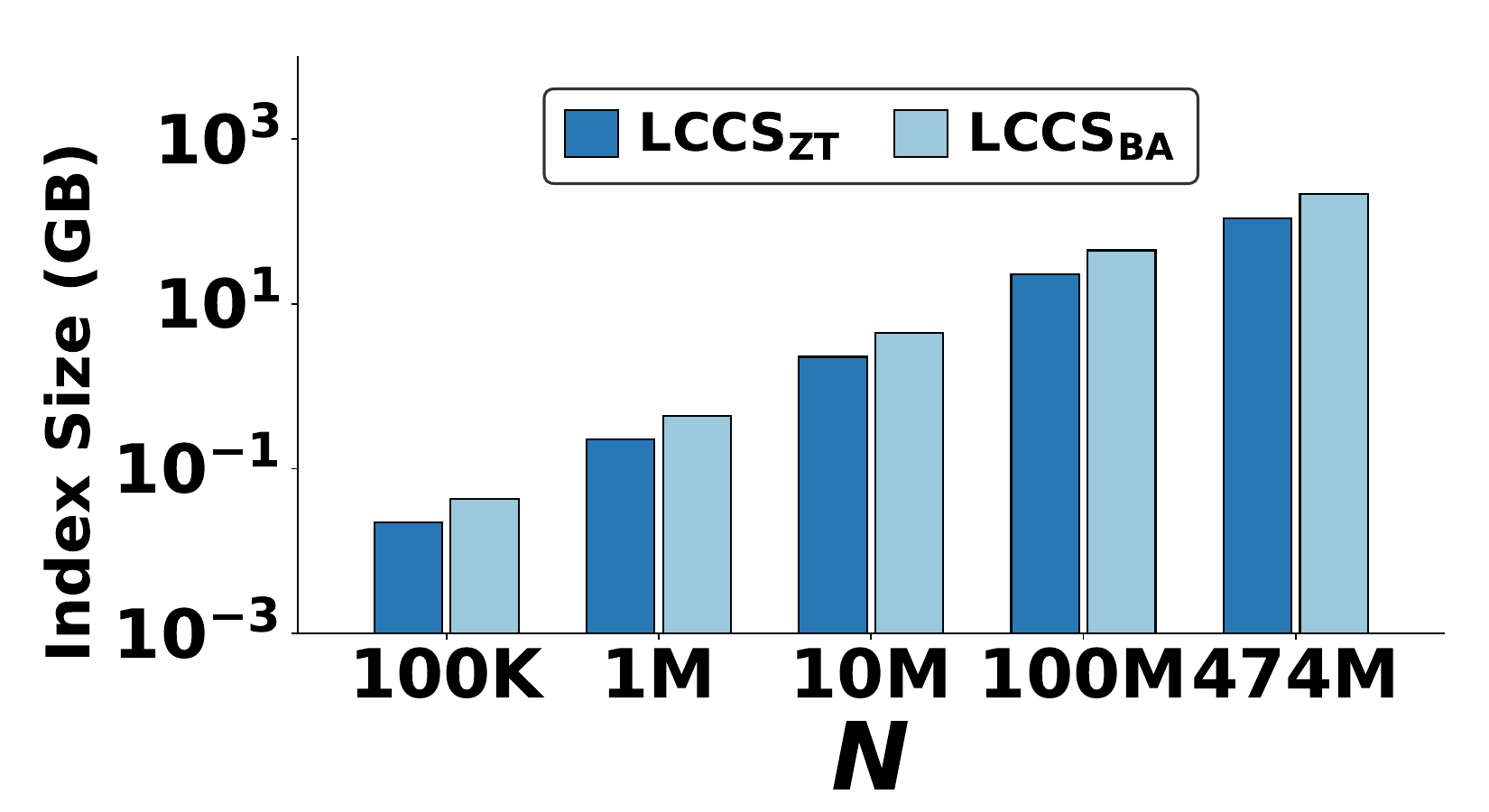}
    \caption{Index size vs. $N$}\label{fig:app:LCCS:n:index:WIKI}
  \end{subfigure}\\[0pt]
  \begin{subfigure}[t]{\appfigwidth}
    \includegraphics[width=\linewidth]{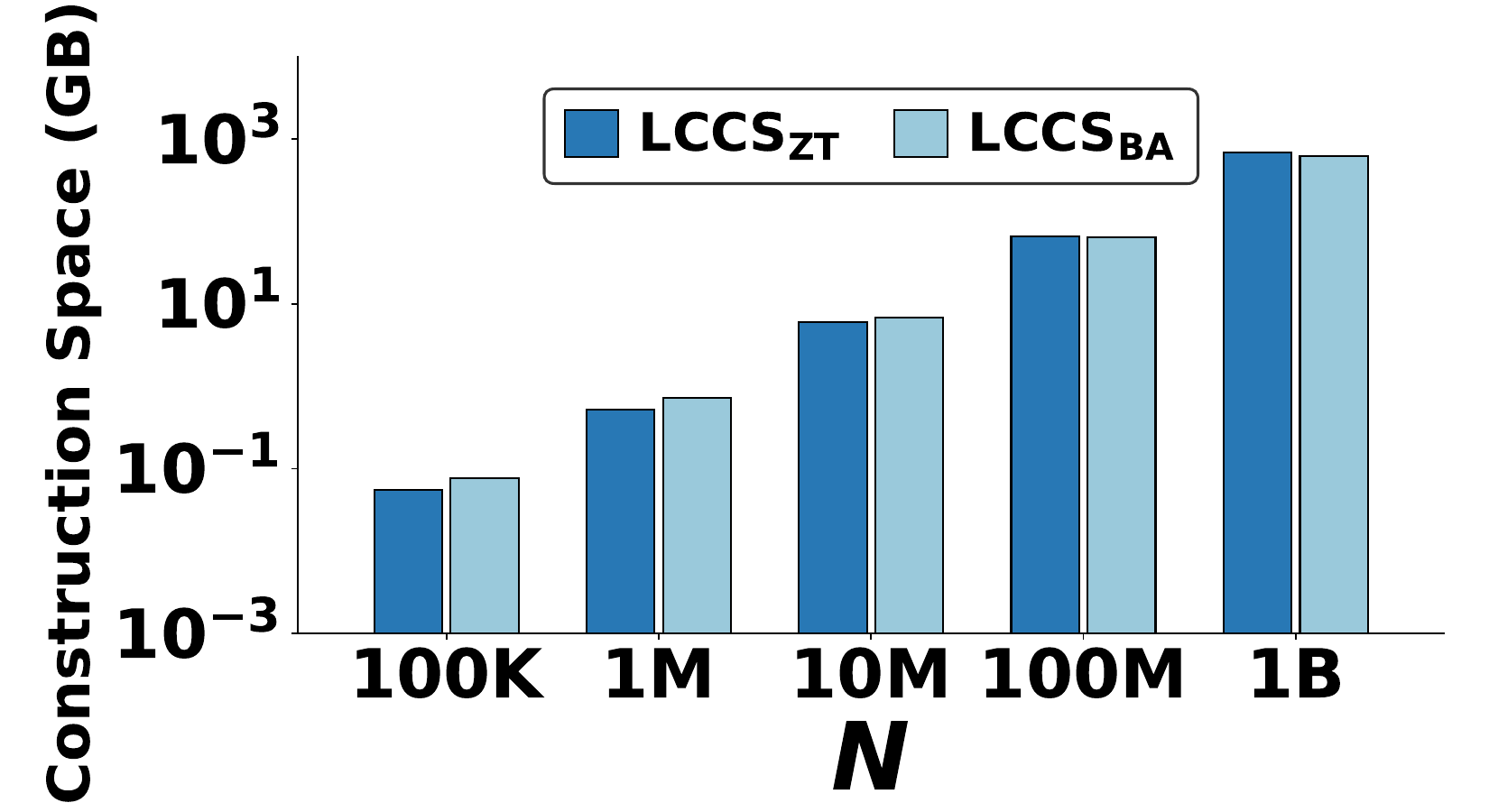}
    \caption{Constr.\ space vs. $N$}\label{fig:app:LCCS:n:rss:BST}
  \end{subfigure}
  \begin{subfigure}[t]{\appfigwidth}
    \includegraphics[width=\linewidth]{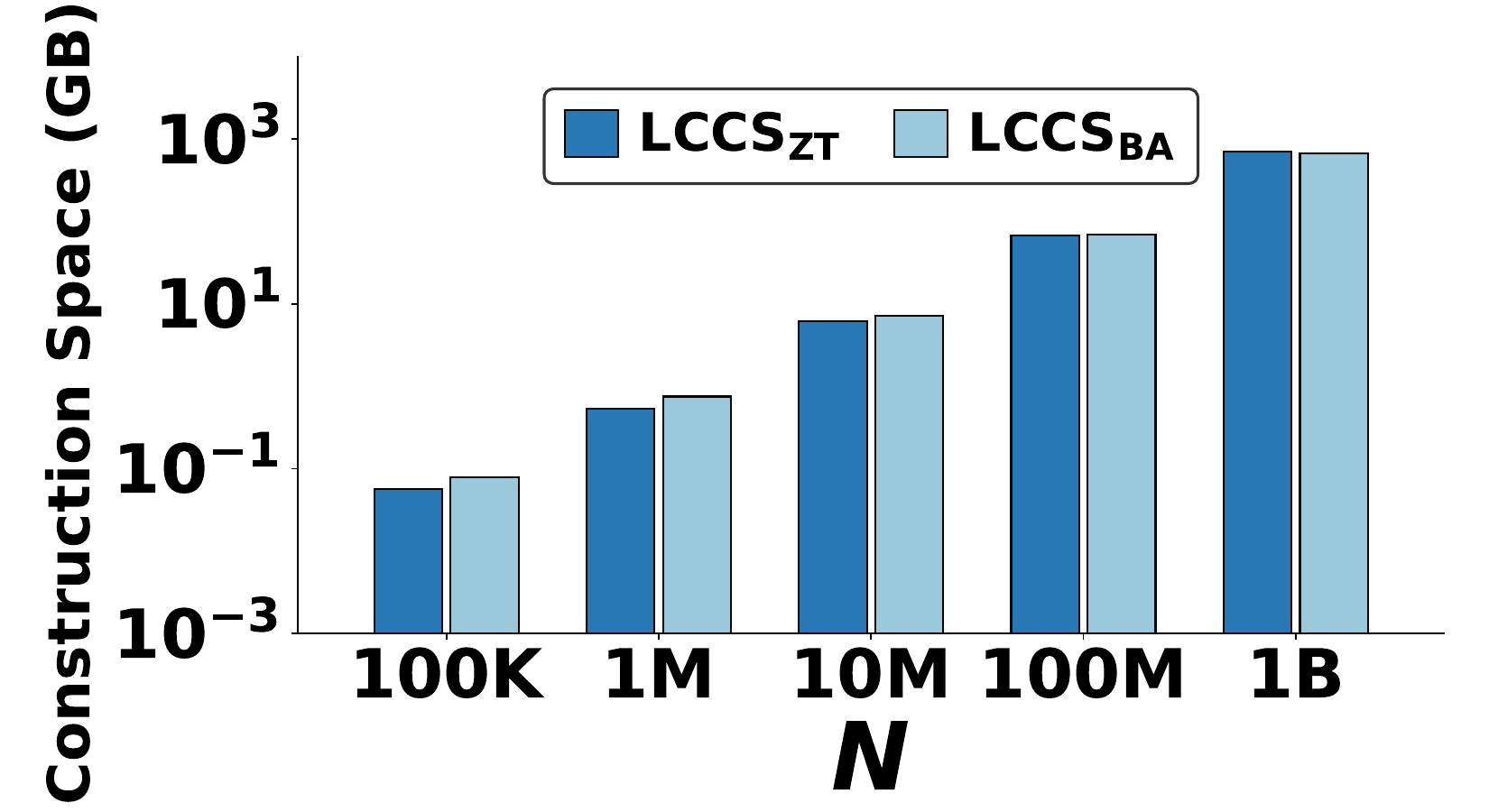}
    \caption{Constr.\ space vs. $N$}\label{fig:app:LCCS:n:rss:CHR}
  \end{subfigure}
  \begin{subfigure}[t]{\appfigwidth}
    \includegraphics[width=\linewidth]{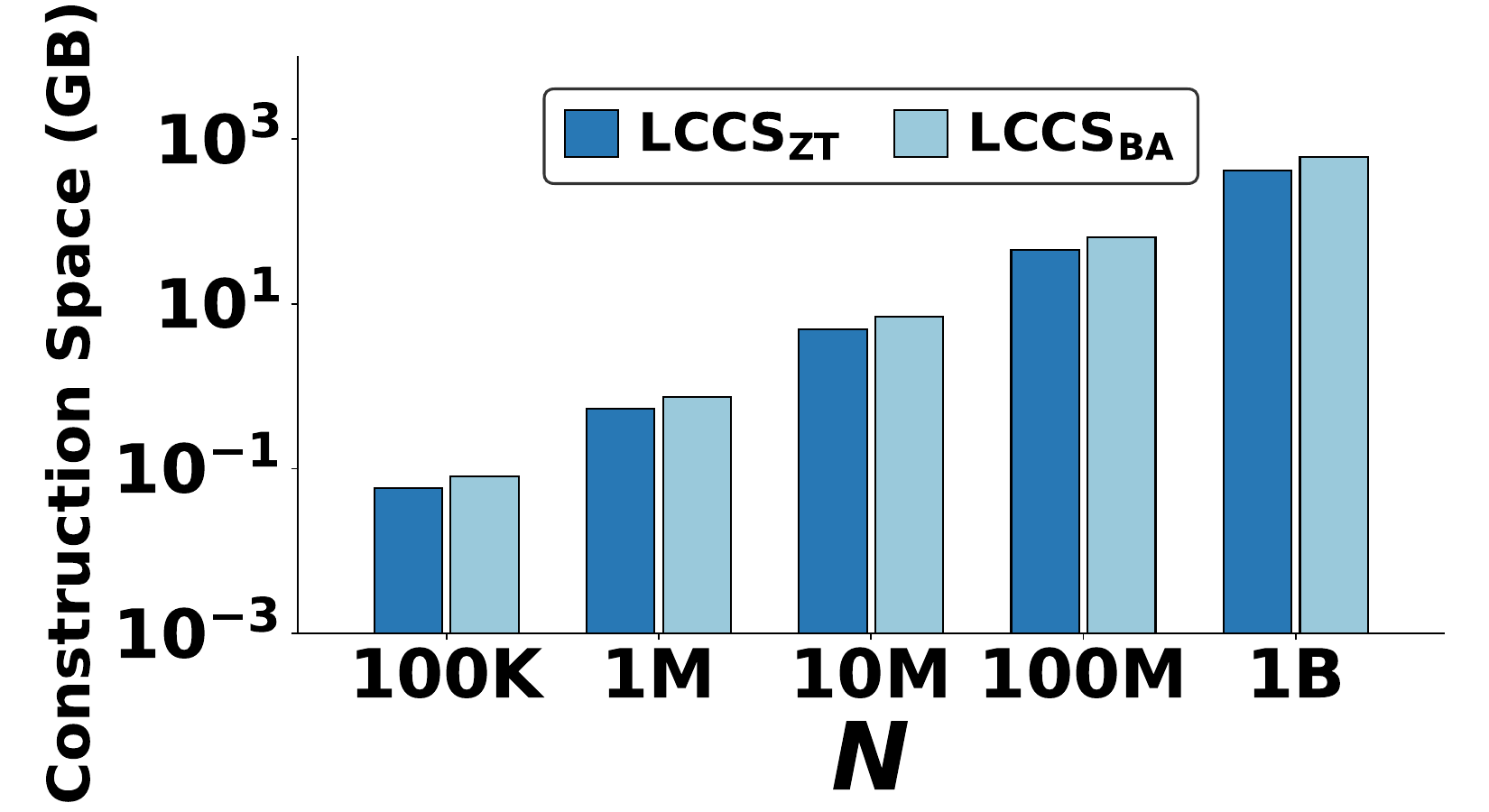}
    \caption{Constr.\ space vs. $N$}\label{fig:app:LCCS:n:rss:SDSL}
  \end{subfigure}
  \begin{subfigure}[t]{\appfigwidth}
    \includegraphics[width=\linewidth]{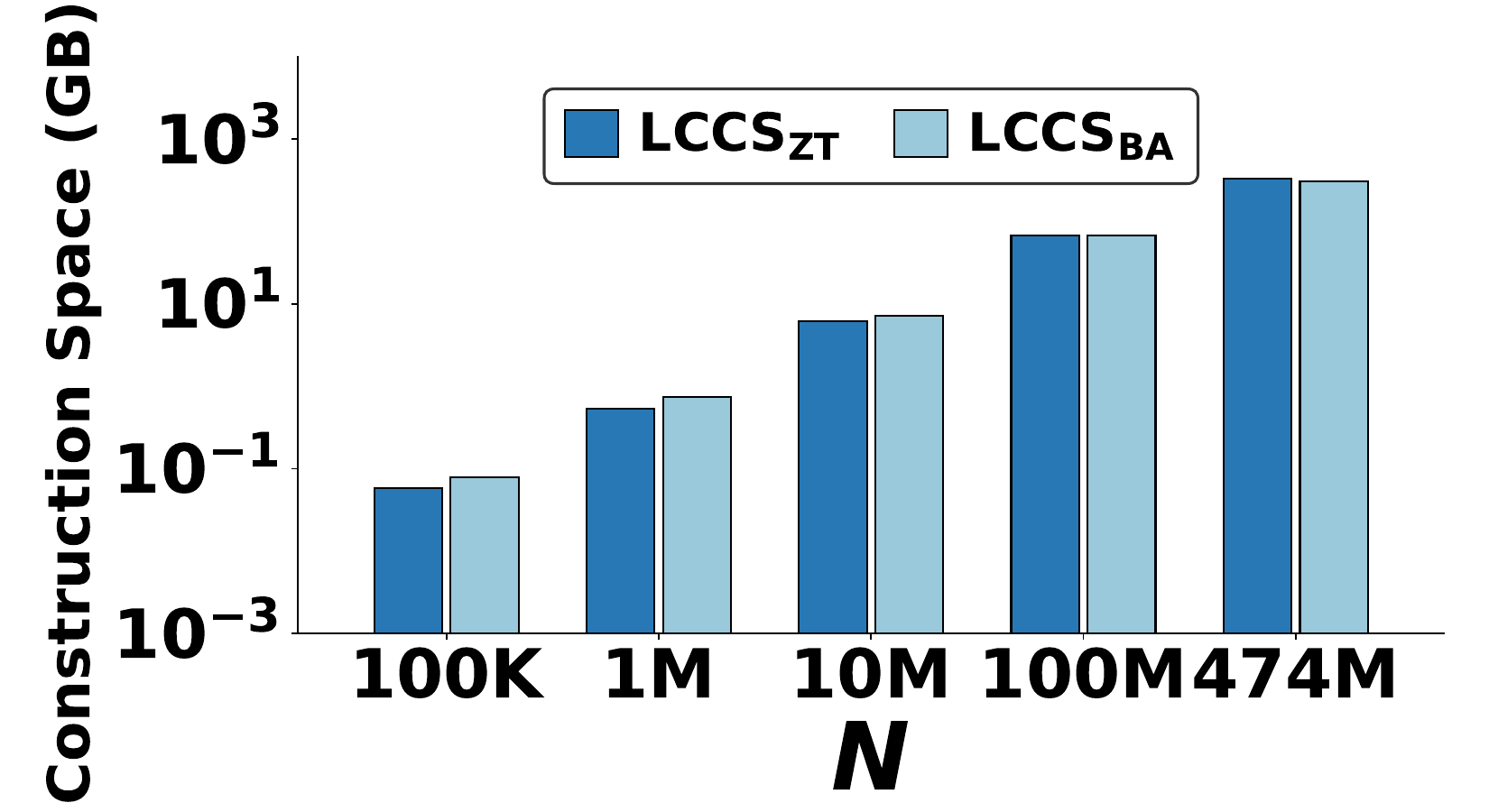}
    \caption{Constr.\ space vs. $N$}\label{fig:app:LCCS:n:rss:WIKI}
  \end{subfigure}\\[0pt]
  \begin{subfigure}[t]{\appfigwidth}
    \includegraphics[width=\linewidth]{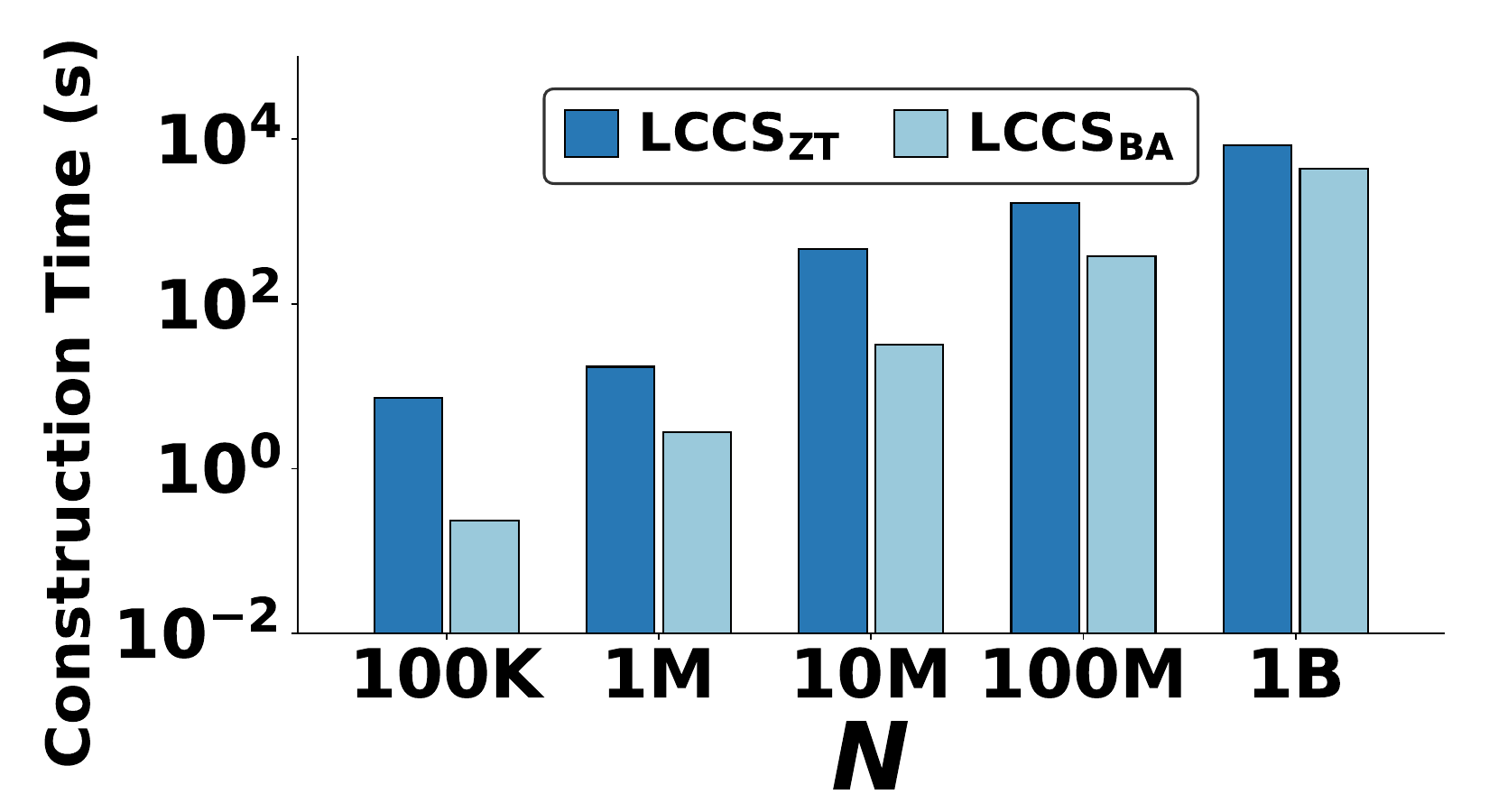}
    \caption{Constr.\ time vs. $N$}\label{fig:app:LCCS:n:build:BST}
  \end{subfigure}
  \begin{subfigure}[t]{\appfigwidth}
    \includegraphics[width=\linewidth]{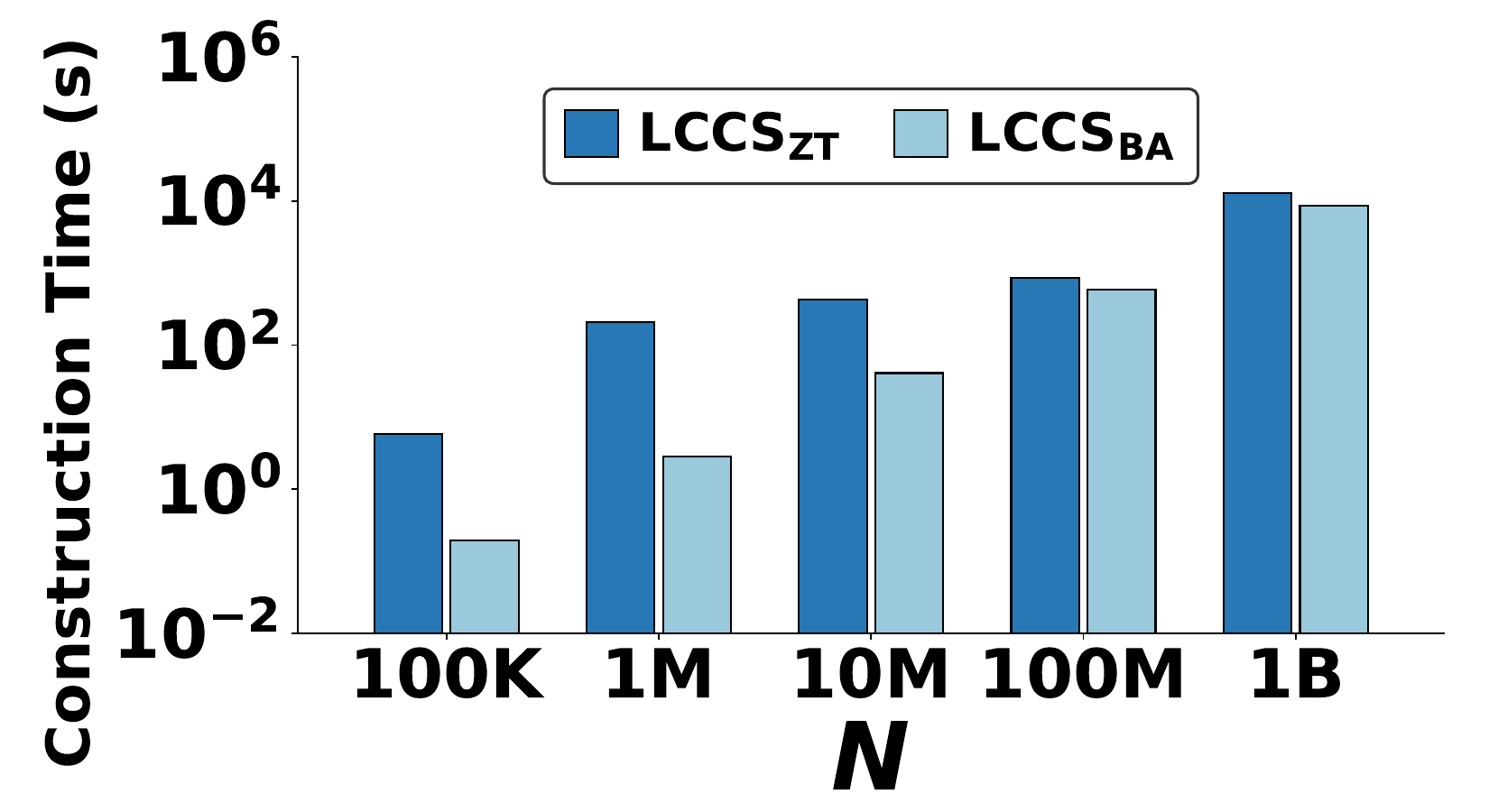}
    \caption{Constr.\ time vs. $N$}\label{fig:app:LCCS:n:build:CHR}
  \end{subfigure}
  \begin{subfigure}[t]{\appfigwidth}
    \includegraphics[width=\linewidth]{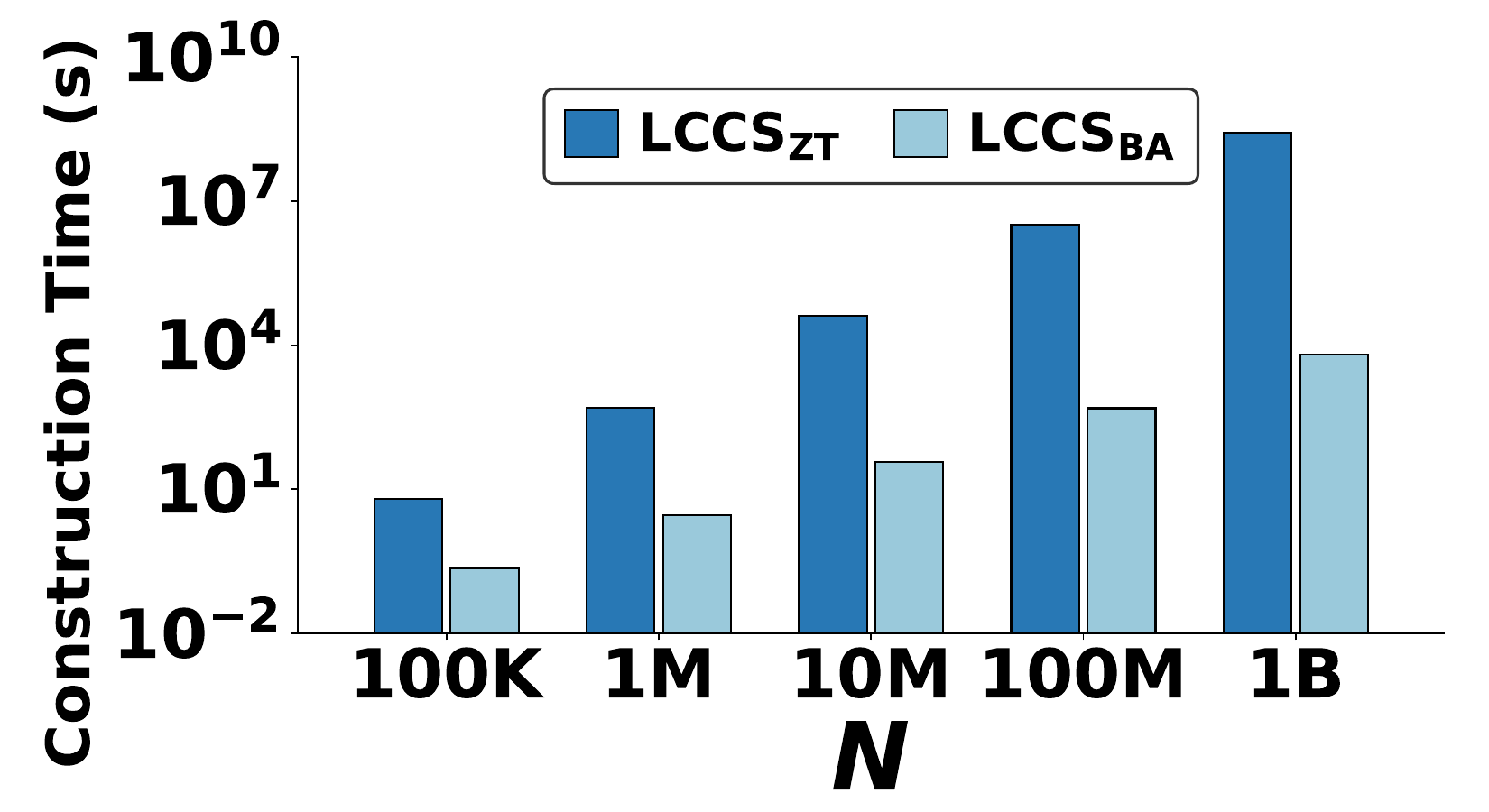}
    \caption{Constr.\ time vs. $N$}\label{fig:app:LCCS:n:build:SDSL}
  \end{subfigure}
  \begin{subfigure}[t]{\appfigwidth}
    \includegraphics[width=\linewidth]{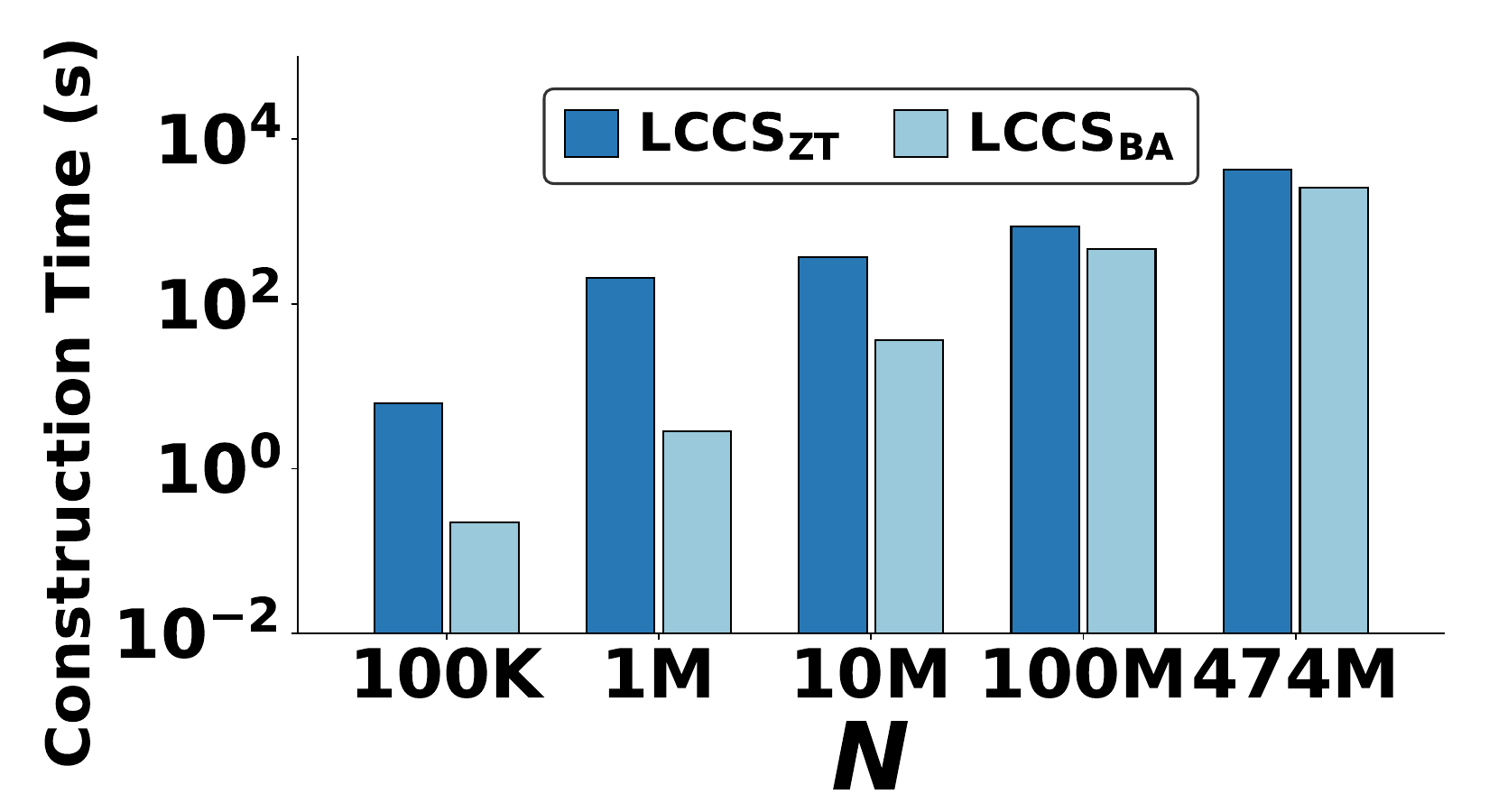}
    \caption{Constr.\ time vs. $N$}\label{fig:app:LCCS:n:build:WIKI}
  \end{subfigure}
  \vspace{\captionspacing}
  \vspace{+2mm}
  \caption{Index size of our \LCCS index vs. \LCCSBA on (a) \bst, (b) \chr, (c) \sdsl, and (d) \wiki vs. $N$; construction space of our \LCCS index vs. \LCCSBA on (e) \bst, (f) \chr, (g) \sdsl, and (h) \wiki vs. $N$; construction time of our \LCCS index vs. \LCCSBA on (i) \bst, (j) \chr, (k) \sdsl, and (l) \wiki vs. $N$.}\label{fig:app:LCCS:cost}
\end{figure}

\begin{figure}[ht]
  \centering
  \begin{subfigure}[t]{\appfigwidth}
    \includegraphics[width=\linewidth]{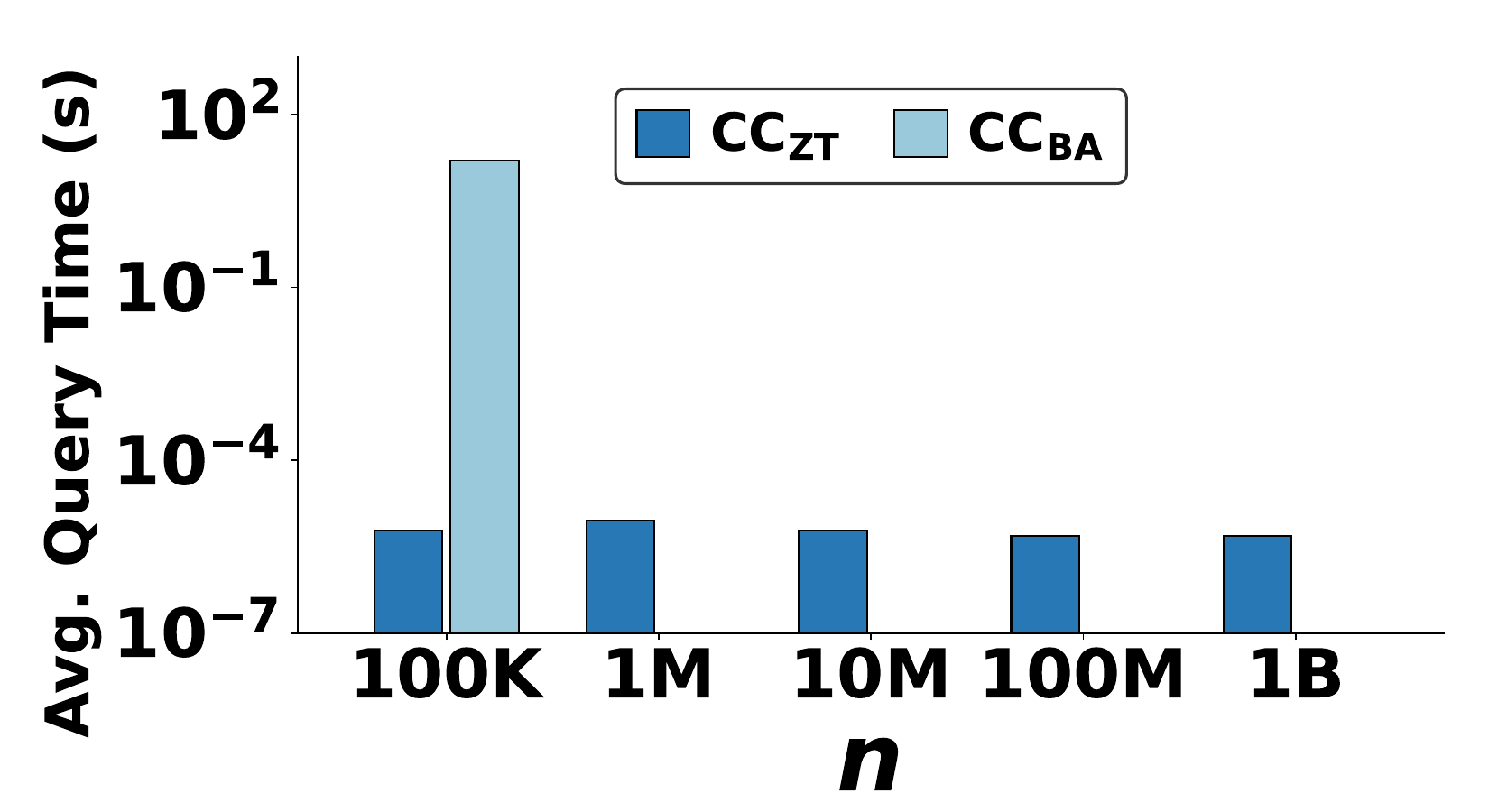}
    \caption{Query time vs. $n$}\label{fig:app:CC:n:query:CHR}
  \end{subfigure}
  \begin{subfigure}[t]{\appfigwidth}
    \includegraphics[width=\linewidth]{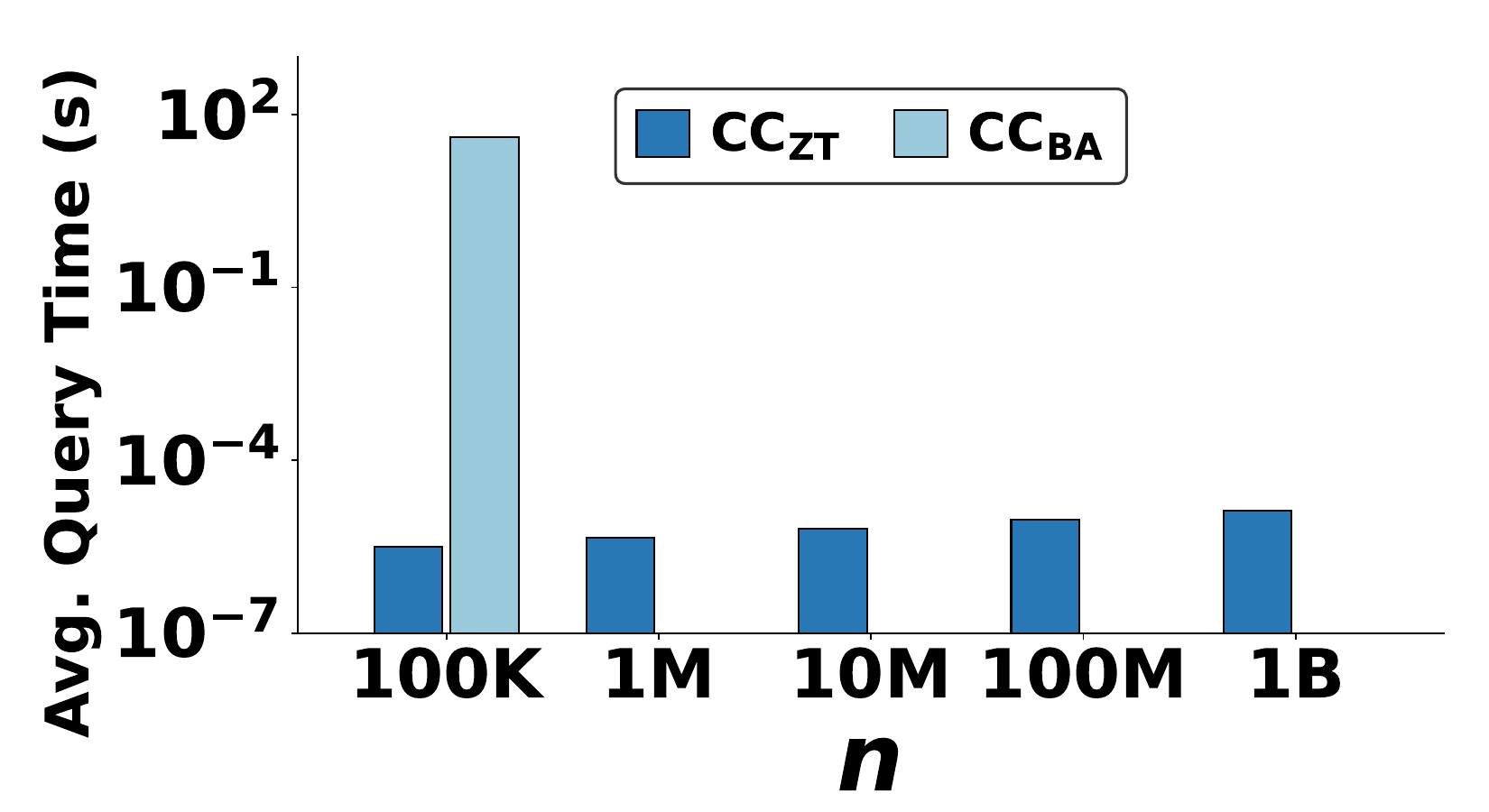}
    \caption{Query time vs. $n$}\label{fig:app:CC:n:query:SARS}
  \end{subfigure}
  \begin{subfigure}[t]{\appfigwidth}
    \includegraphics[width=\linewidth]{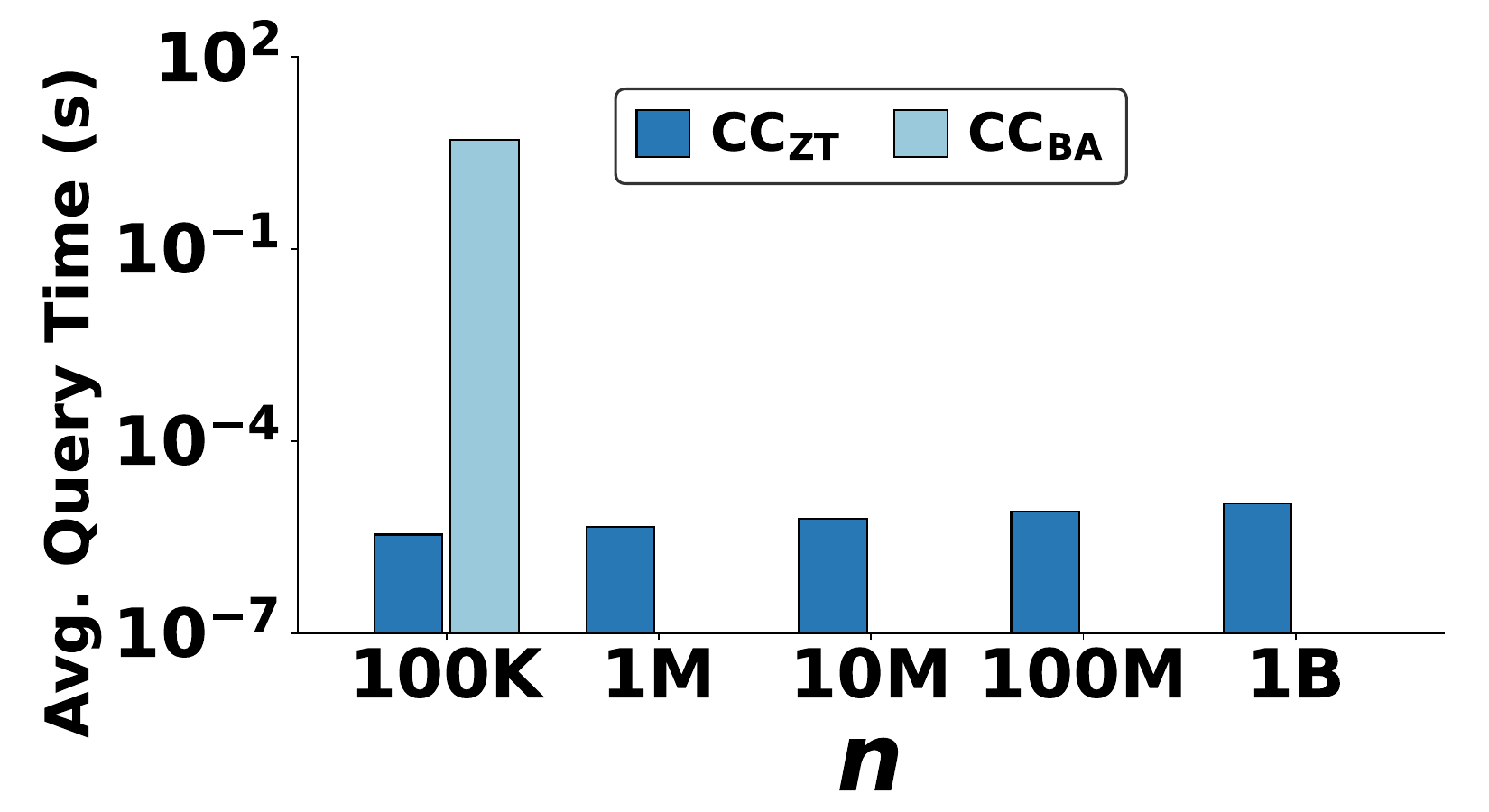}
    \caption{Query time vs. $n$}\label{fig:app:CC:n:query:SDSL}
  \end{subfigure}
  \begin{subfigure}[t]{\appfigwidth}
    \includegraphics[width=\linewidth]{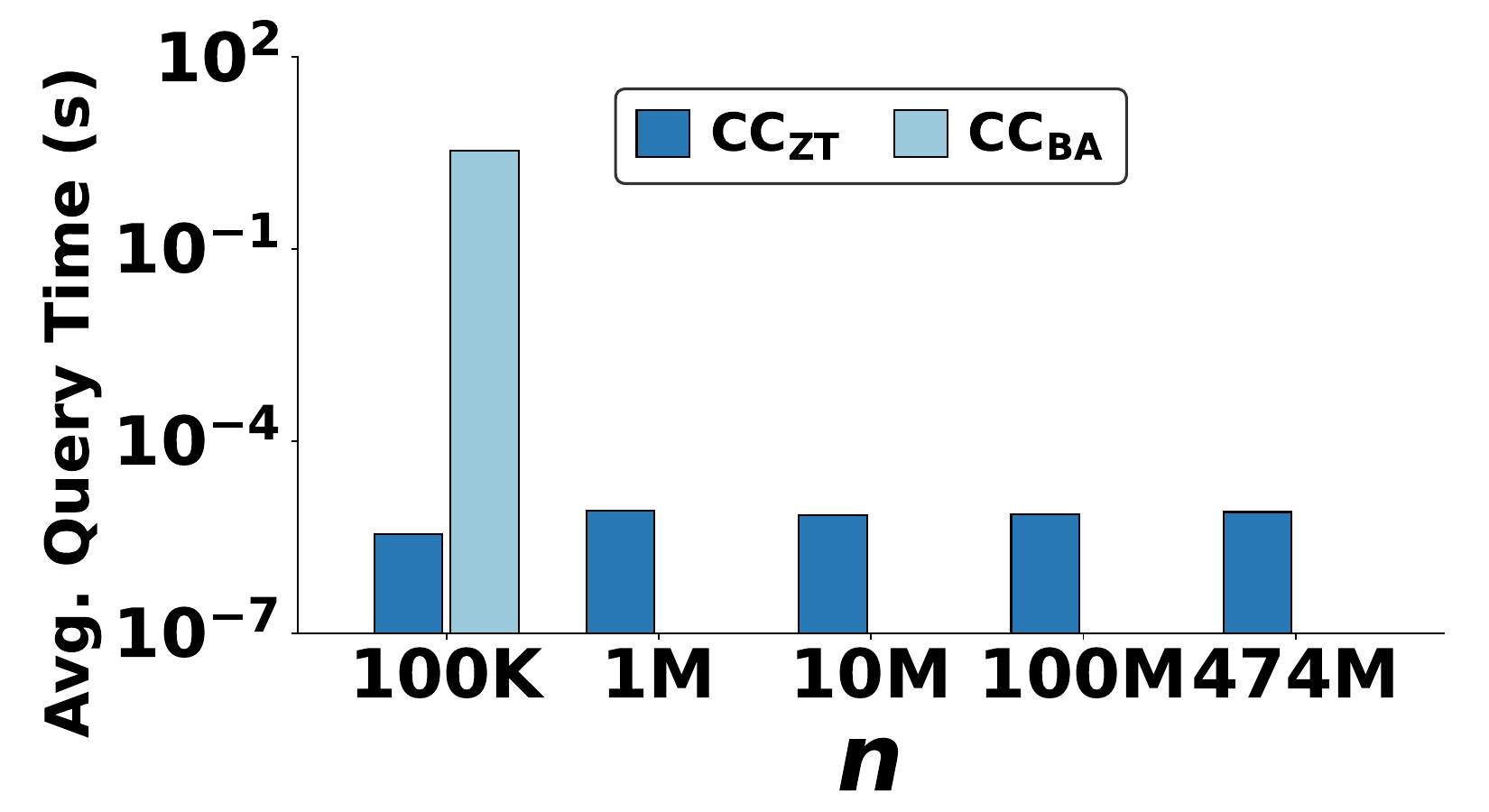}
    \caption{Query time vs. $n$}\label{fig:app:CC:n:query:WIKI}
  \end{subfigure}\\[0pt]
  \begin{subfigure}[t]{\appfigwidth}
    \includegraphics[width=\linewidth]{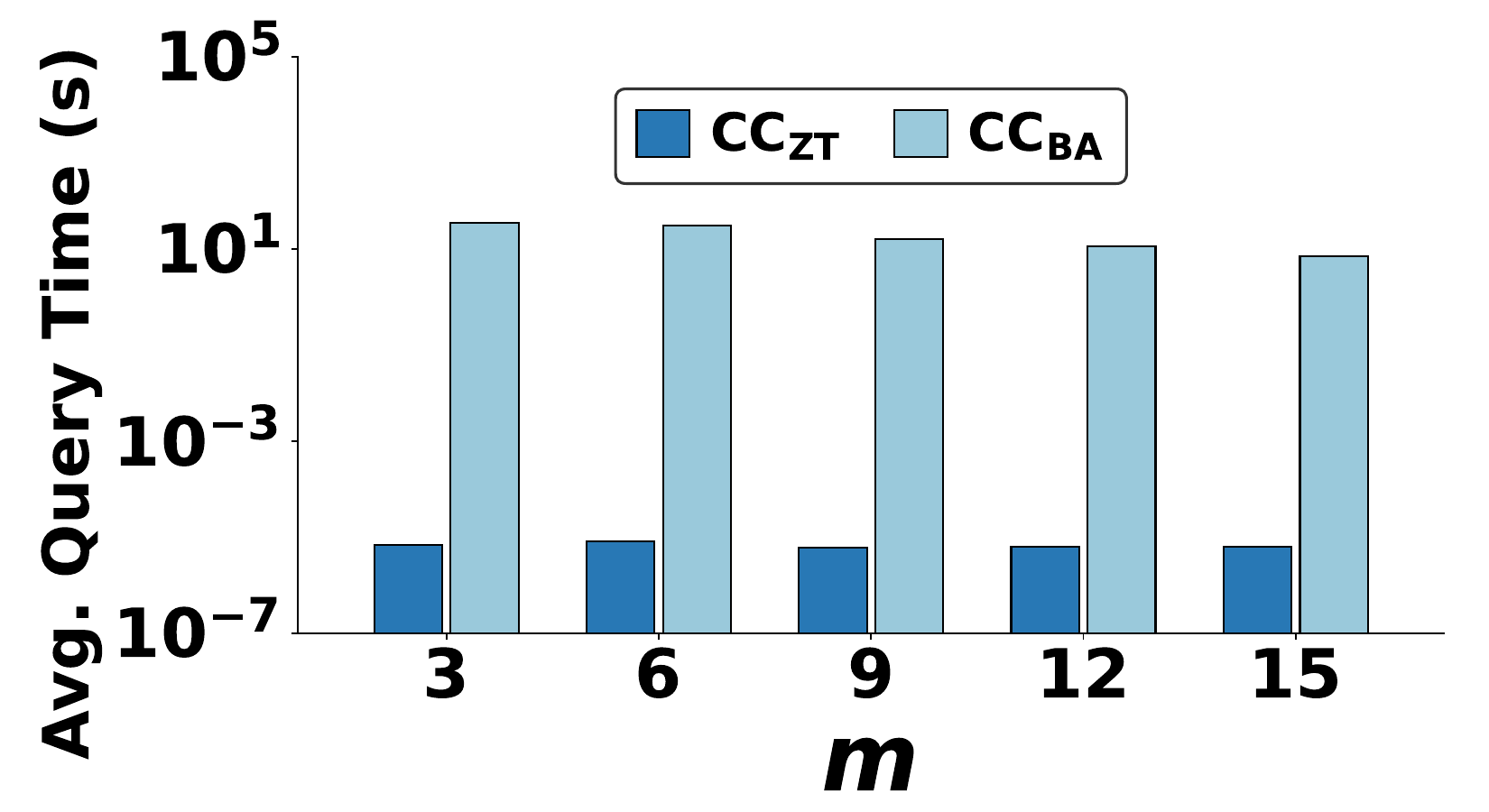}
    \caption{Query time vs. $m$}\label{fig:app:CC:m:query:CHR}
  \end{subfigure}
  \begin{subfigure}[t]{\appfigwidth}
    \includegraphics[width=\linewidth]{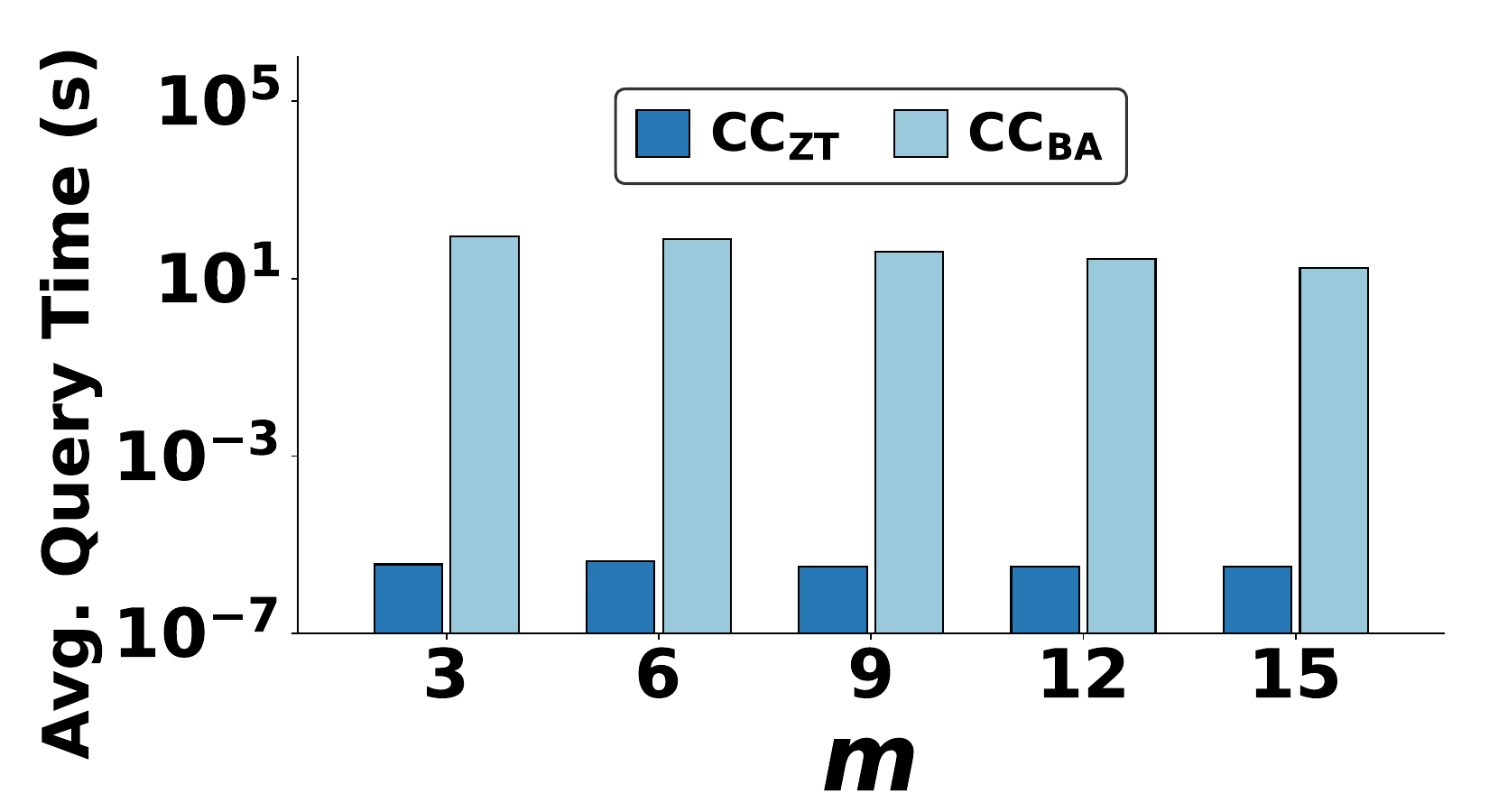}
    \caption{Query time vs. $m$}\label{fig:app:CC:m:query:SARS}
  \end{subfigure}
  \begin{subfigure}[t]{\appfigwidth}
    \includegraphics[width=\linewidth]{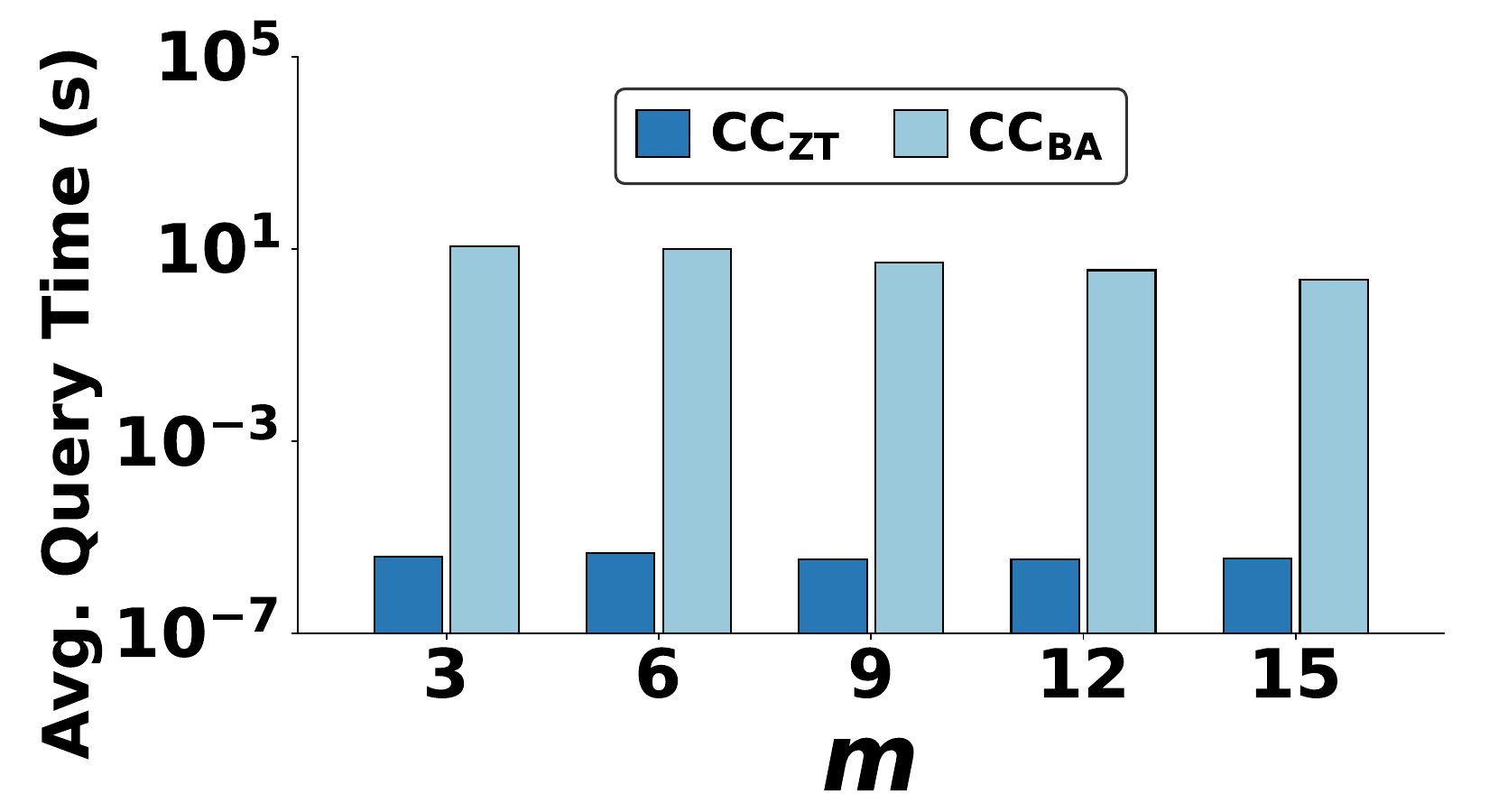}
    \caption{Query time vs. $m$}\label{fig:app:CC:m:query:SDSL}
  \end{subfigure}
  \begin{subfigure}[t]{\appfigwidth}
    \includegraphics[width=\linewidth]{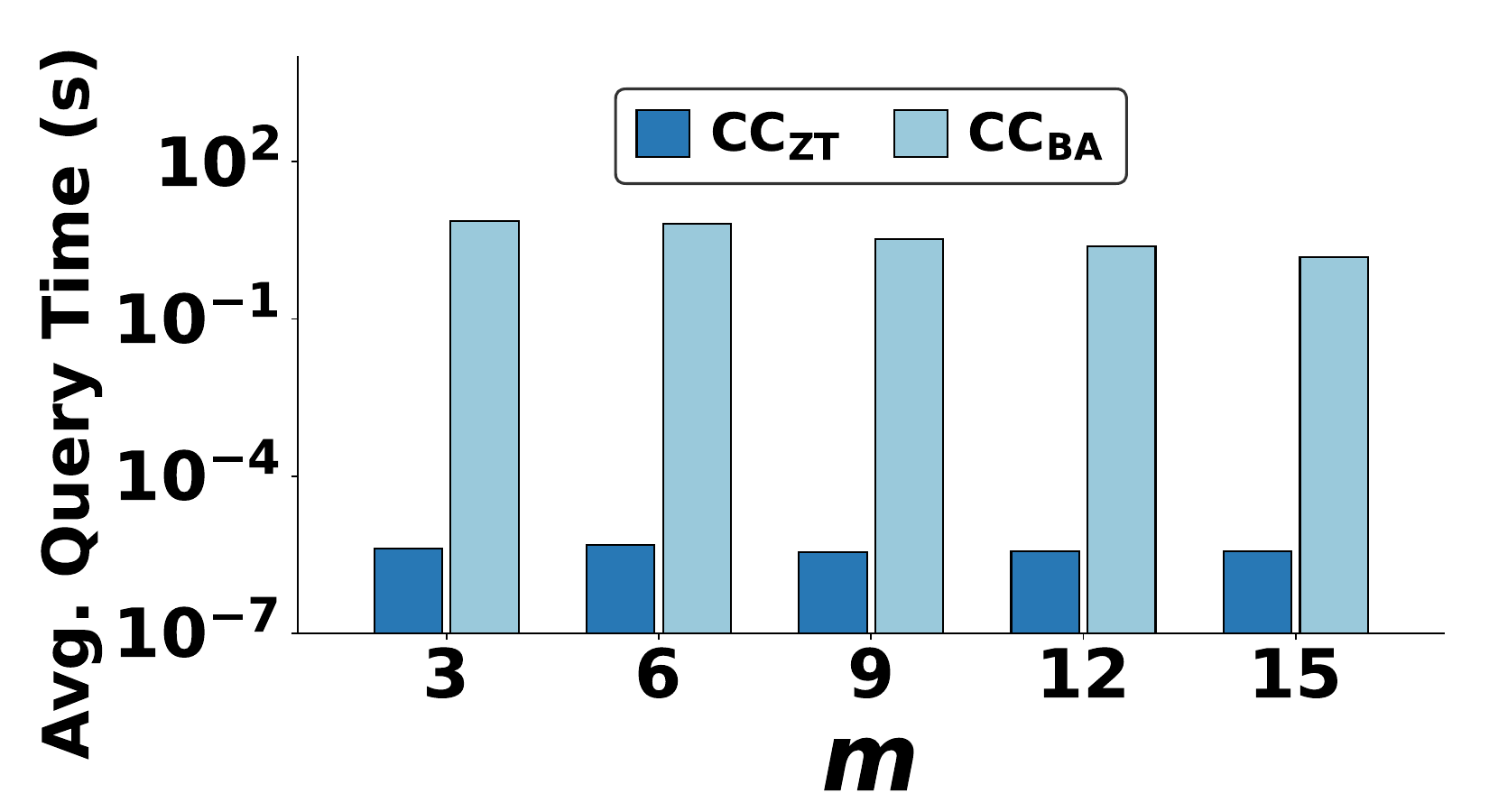}
    \caption{Query time vs. $m$}\label{fig:app:CC:m:query:WIKI}
  \end{subfigure}
  \vspace{\captionspacing}
  \vspace{+2mm}
  \caption{Query time of our \CC index vs. \CCBA on (a) \chr, (b) \sars, (c) \sdsl, and (d) \wiki vs. $n$; on (e) \chr, (f) \sars, (g) \sdsl, and (h) \wiki vs. $m$.}\label{fig:app:CC:query}
\end{figure}

\begin{figure}[ht]
  \centering
  \begin{subfigure}[t]{\appfigwidth}
    \includegraphics[width=\linewidth]{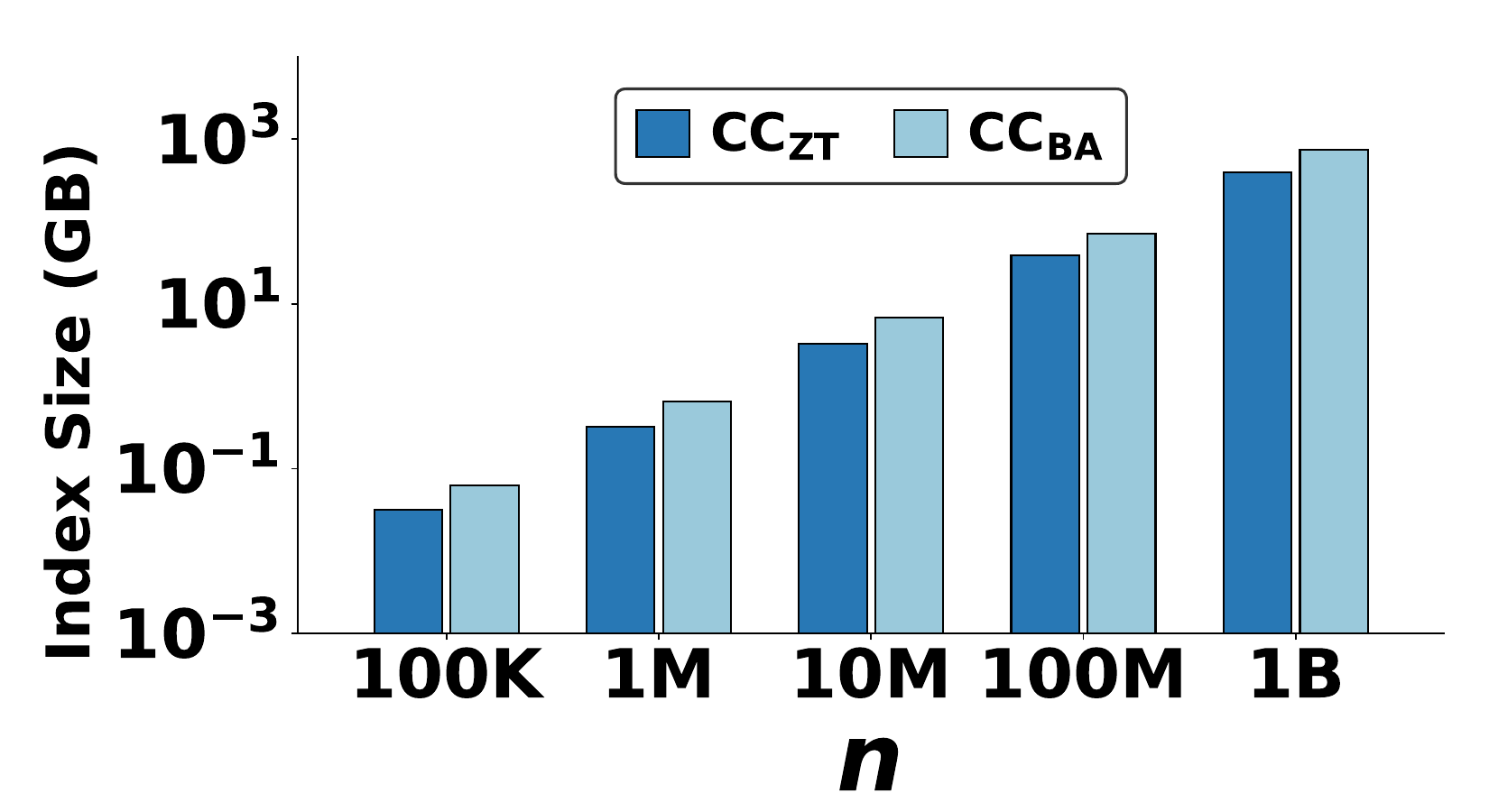}
    \caption{Index size vs. $n$}\label{fig:app:CC:n:index:CHR}
  \end{subfigure}
  \begin{subfigure}[t]{\appfigwidth}
    \includegraphics[width=\linewidth]{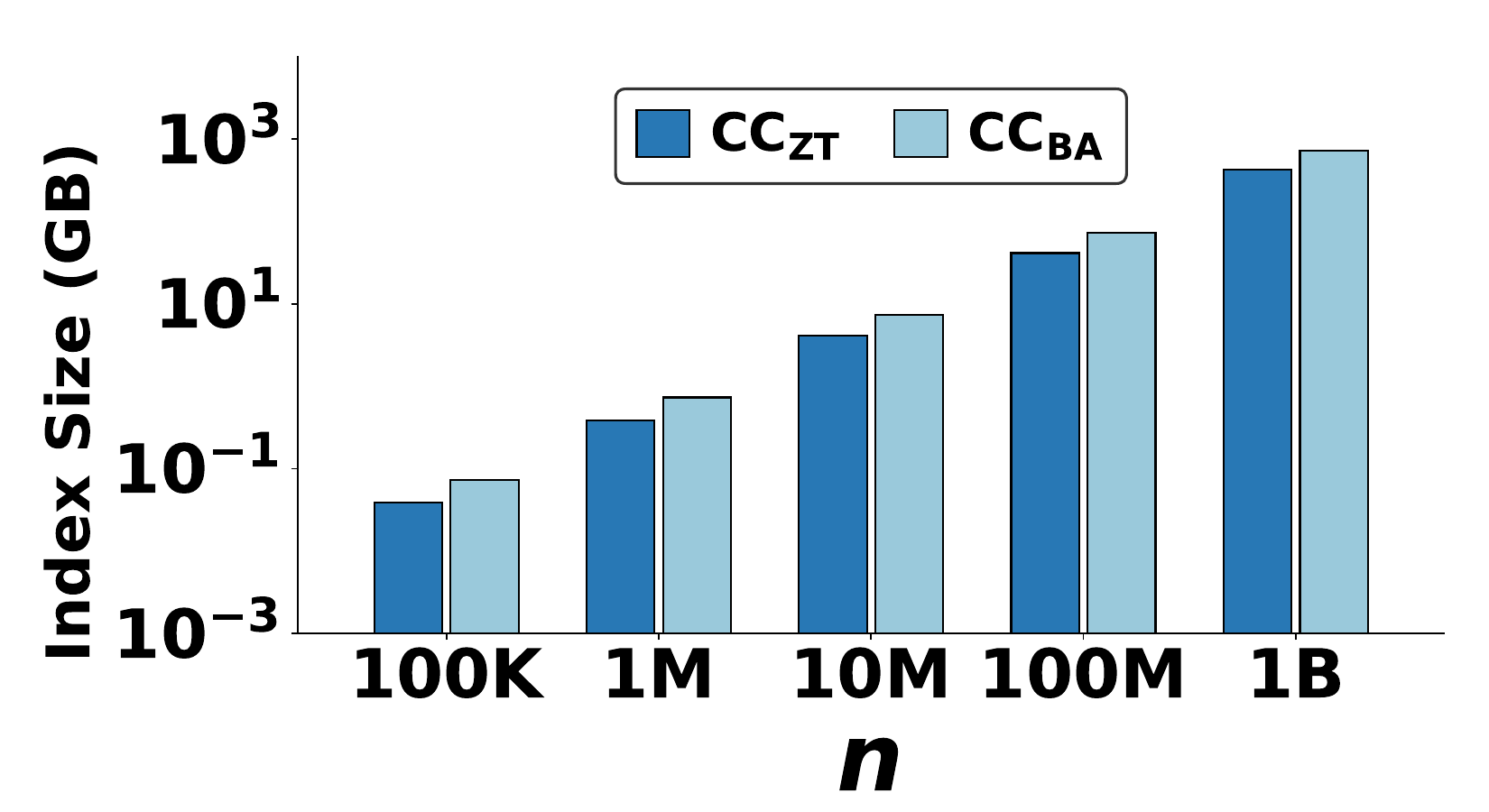}
    \caption{Index size vs. $n$}\label{fig:app:CC:n:index:SARS}
  \end{subfigure}
  \begin{subfigure}[t]{\appfigwidth}
    \includegraphics[width=\linewidth]{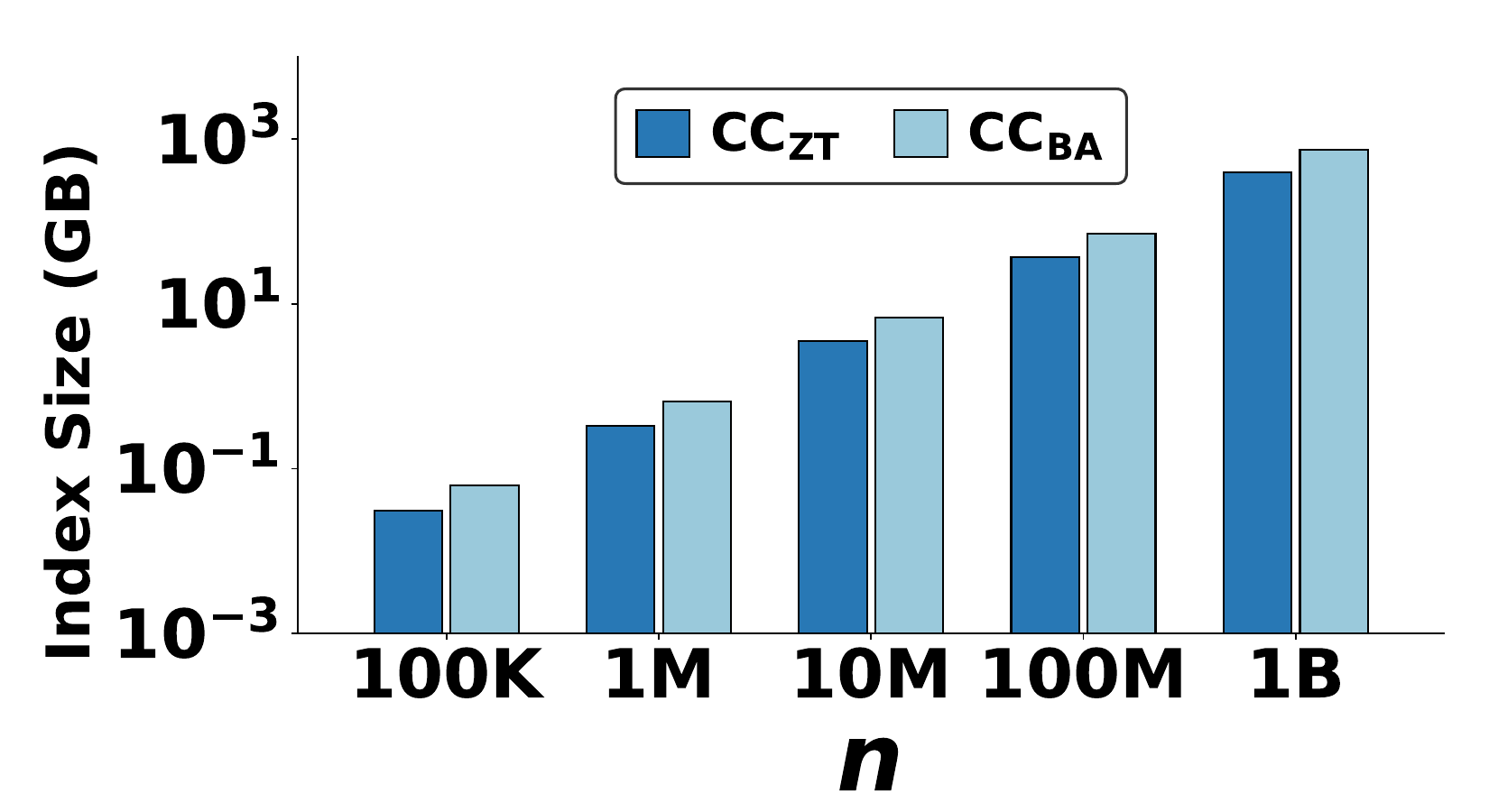}
    \caption{Index size vs. $n$}\label{fig:app:CC:n:index:SDSL}
  \end{subfigure}
  \begin{subfigure}[t]{\appfigwidth}
    \includegraphics[width=\linewidth]{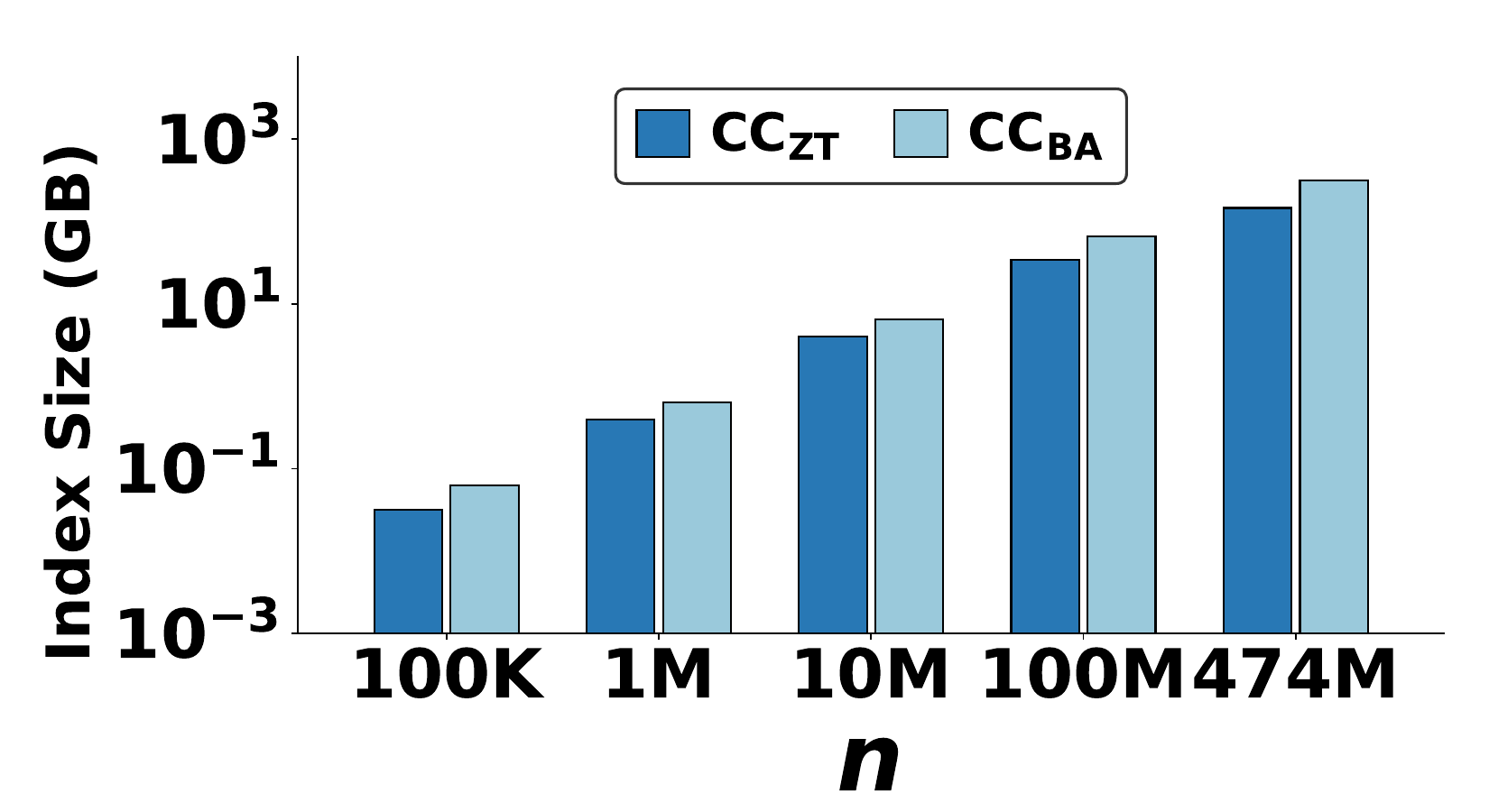}
    \caption{Index size vs. $n$}\label{fig:app:CC:n:index:WIKI}
  \end{subfigure}\\[0pt]
  \begin{subfigure}[t]{\appfigwidth}
    \includegraphics[width=\linewidth]{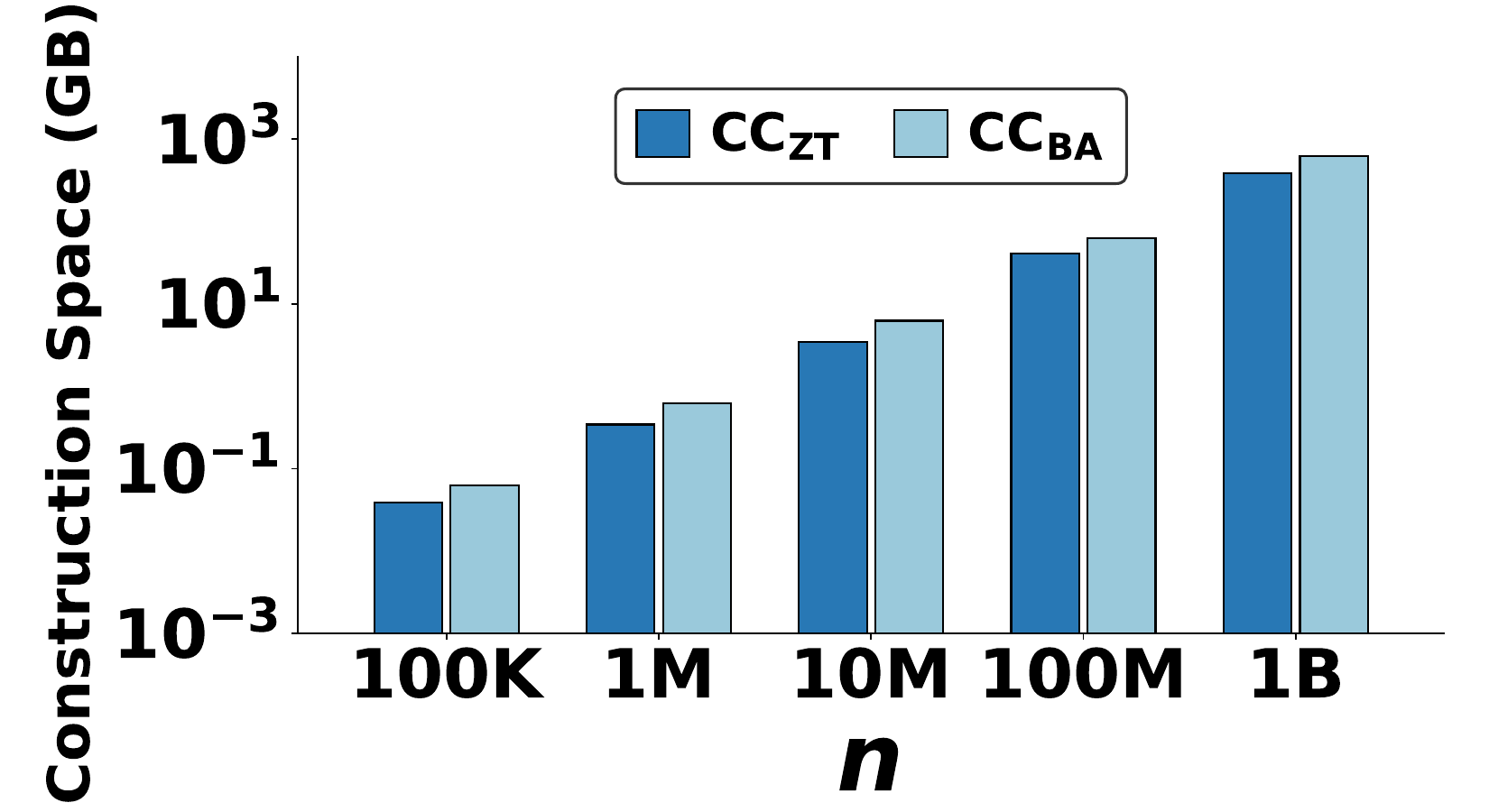}
    \caption{Constr.\ space vs. $n$}\label{fig:app:CC:n:rss:CHR}
  \end{subfigure}
  \begin{subfigure}[t]{\appfigwidth}
    \includegraphics[width=\linewidth]{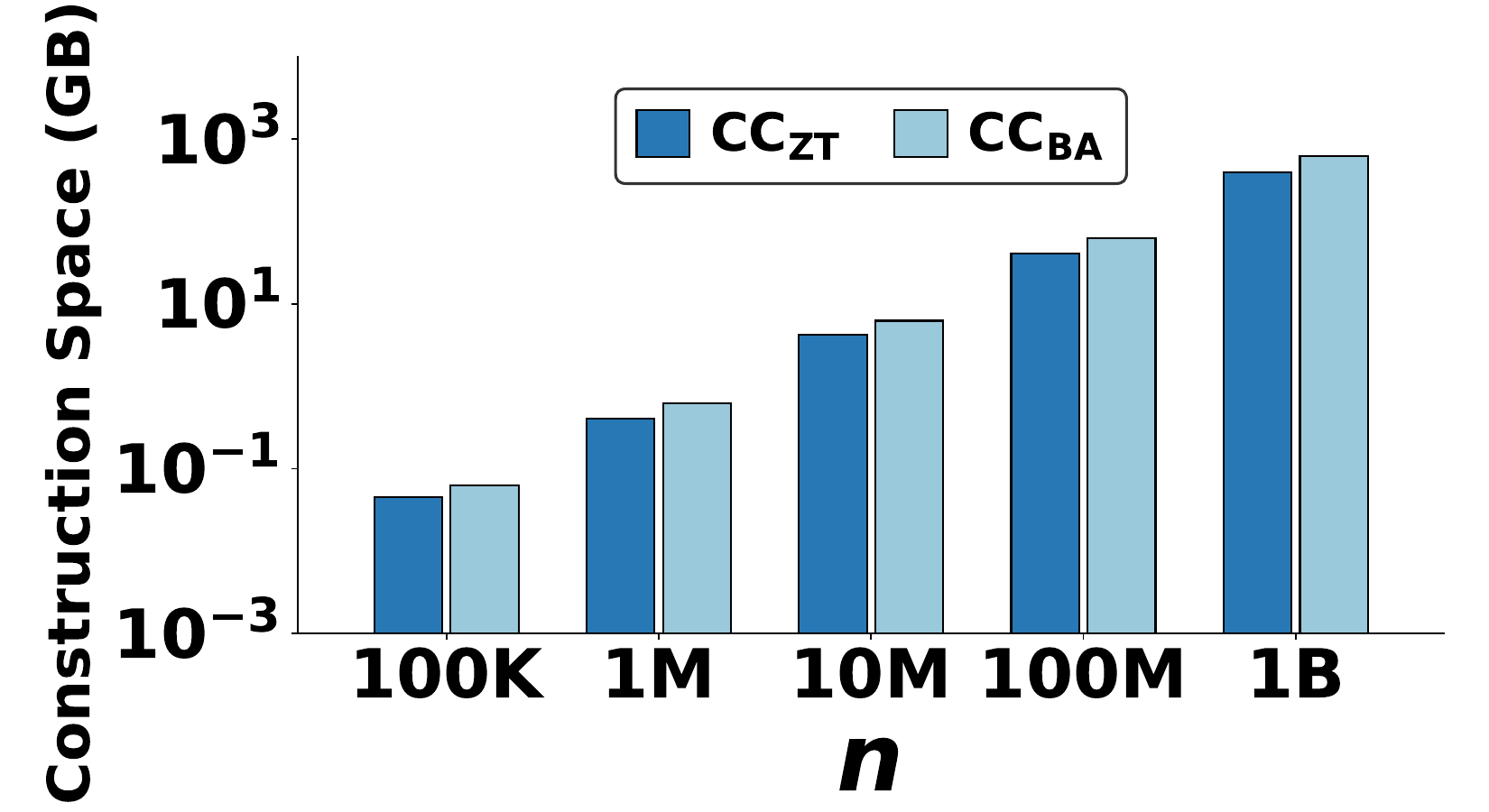}
    \caption{Constr.\ space vs. $n$}\label{fig:app:CC:n:rss:SARS}
  \end{subfigure}
  \begin{subfigure}[t]{\appfigwidth}
    \includegraphics[width=\linewidth]{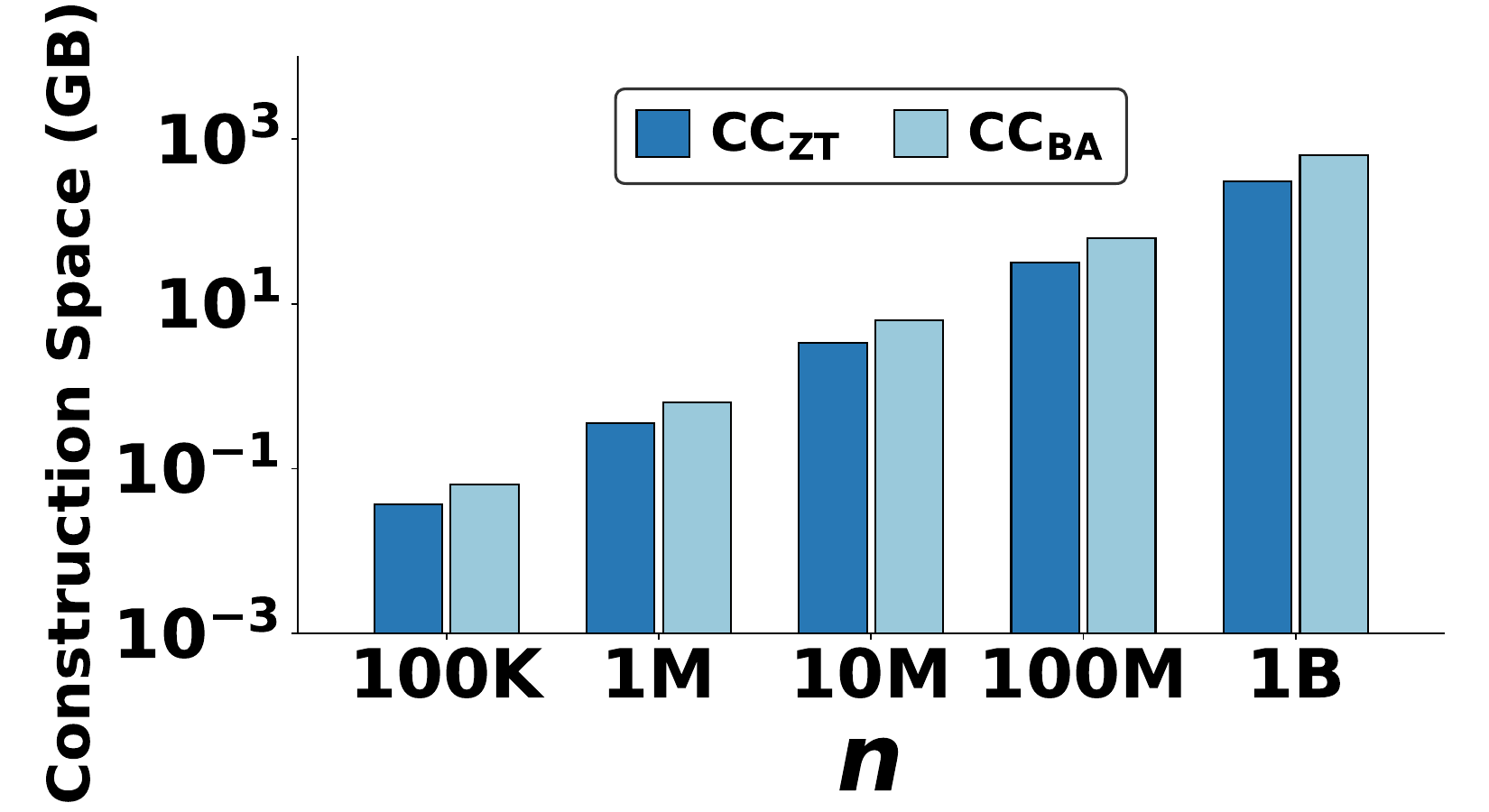}
    \caption{Constr.\ space vs. $n$}\label{fig:app:CC:n:rss:SDSL}
  \end{subfigure}
  \begin{subfigure}[t]{\appfigwidth}
    \includegraphics[width=\linewidth]{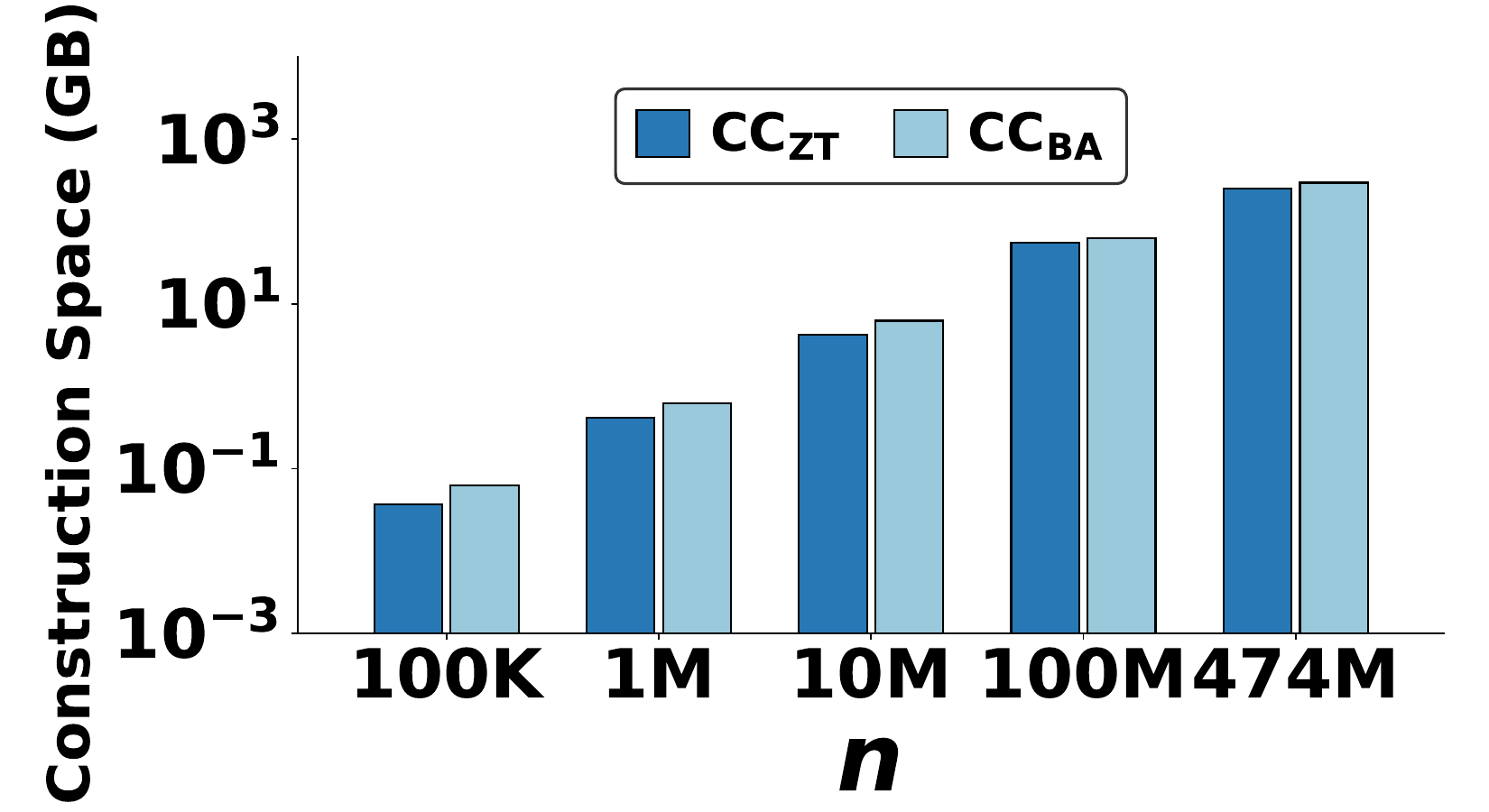}
    \caption{Constr.\ space vs. $n$}\label{fig:app:CC:n:rss:WIKI}
  \end{subfigure}\\[0pt]
  \begin{subfigure}[t]{\appfigwidth}
    \includegraphics[width=\linewidth]{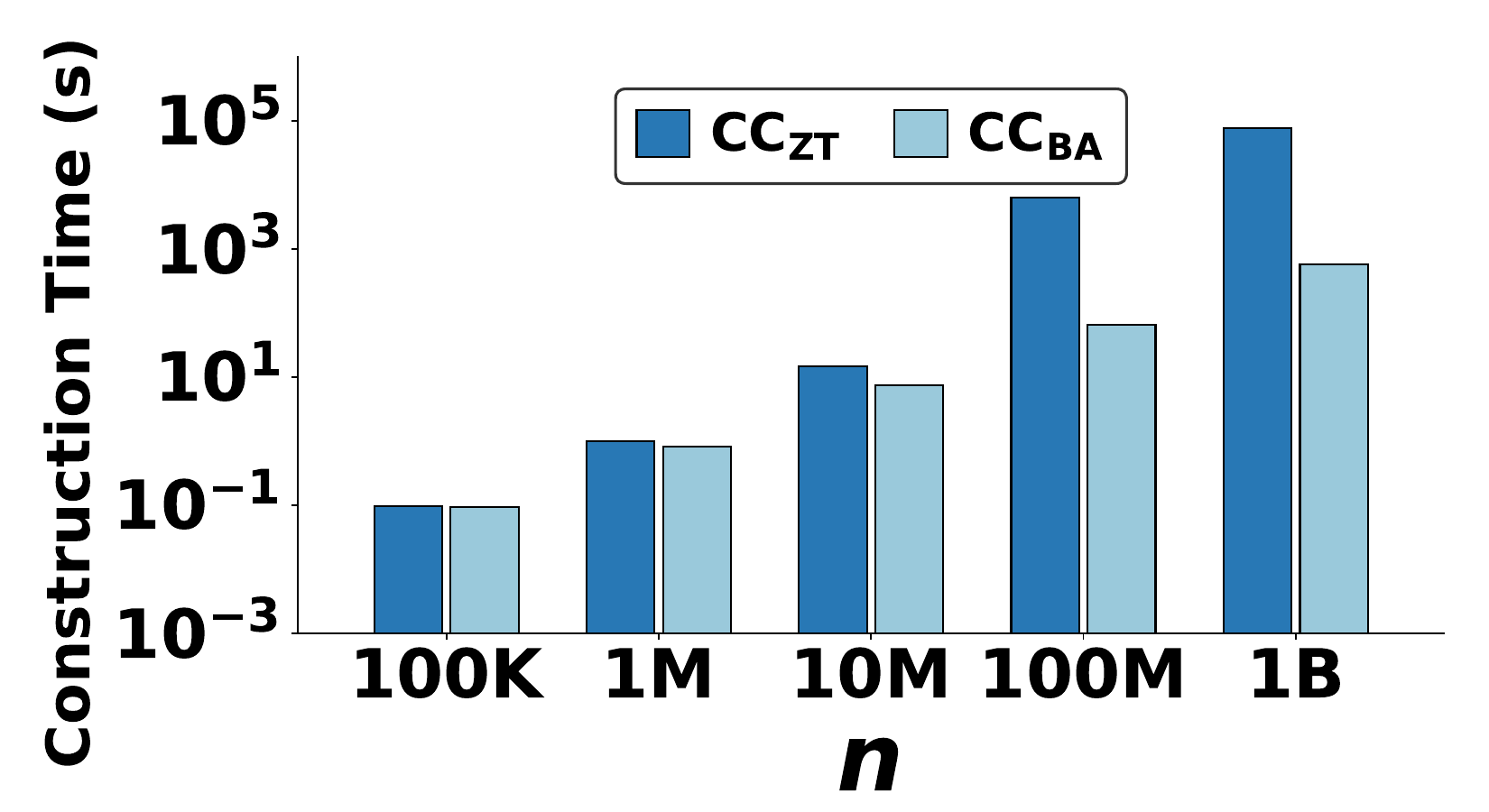}
    \caption{Constr.\ time vs. $n$}\label{fig:app:CC:n:build:CHR}
  \end{subfigure}
  \begin{subfigure}[t]{\appfigwidth}
    \includegraphics[width=\linewidth]{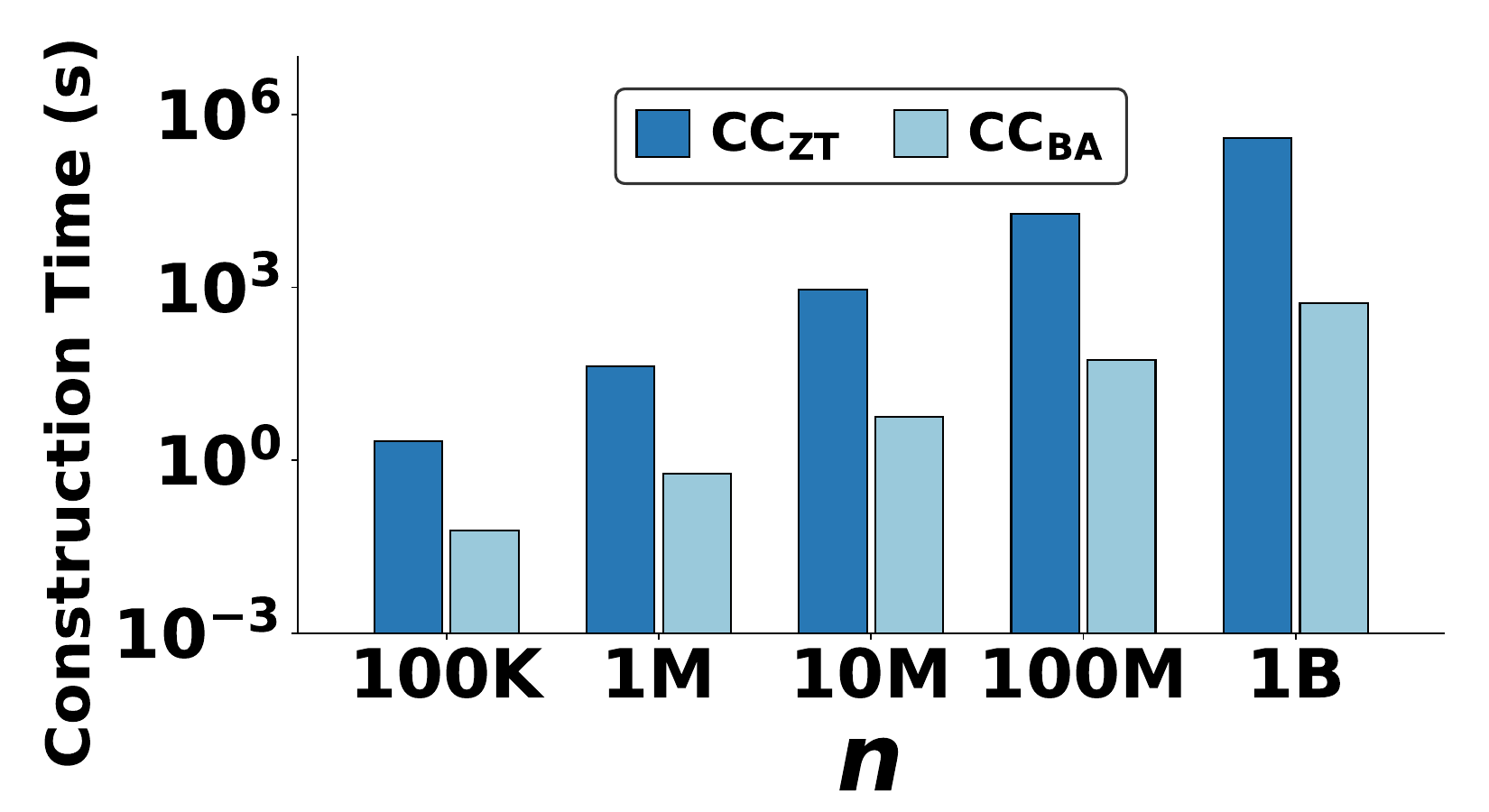}
    \caption{Constr.\ time vs. $n$}\label{fig:app:CC:n:build:SARS}
  \end{subfigure}
  \begin{subfigure}[t]{\appfigwidth}
    \includegraphics[width=\linewidth]{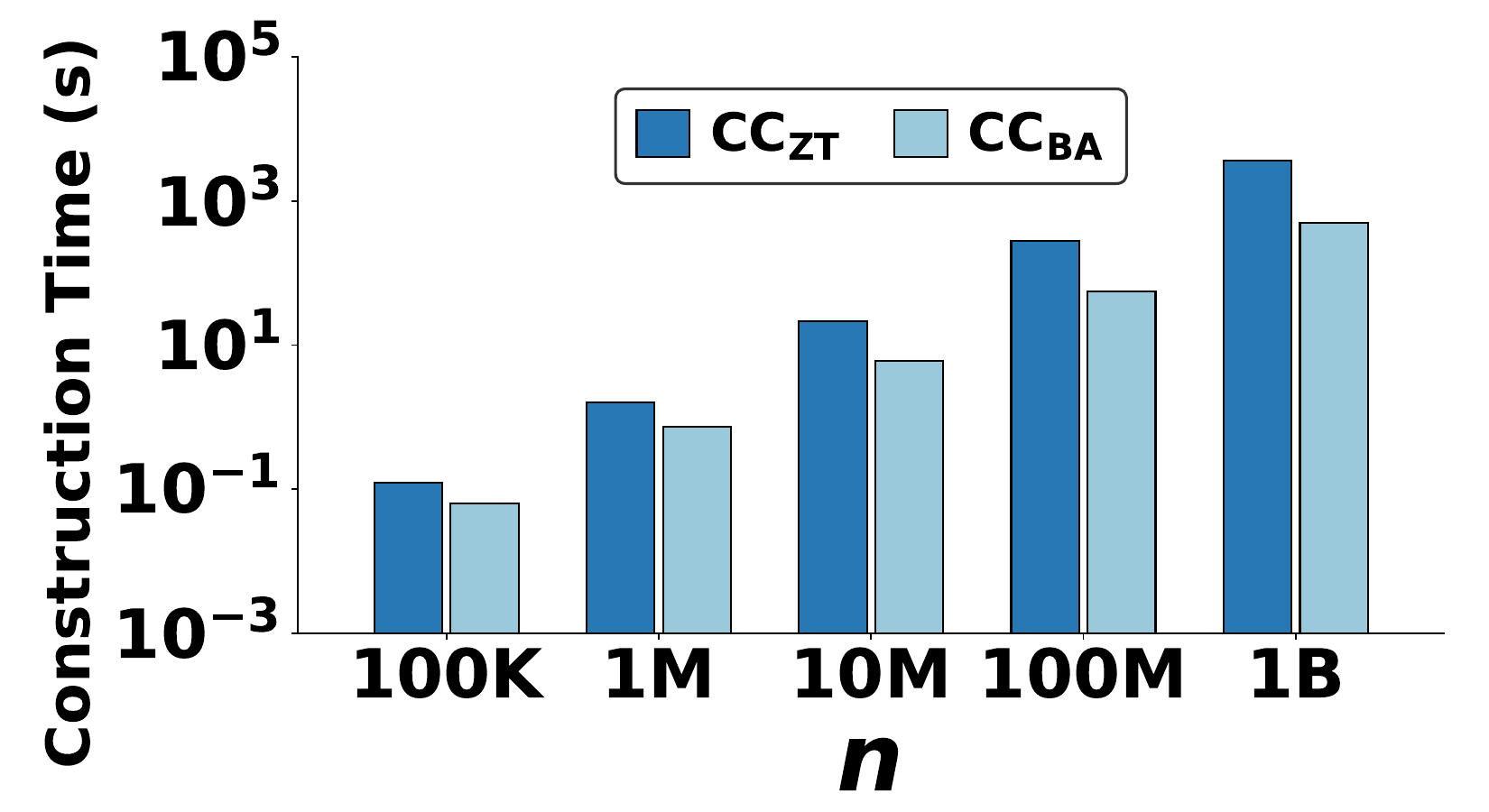}
    \caption{Constr.\ time vs. $n$}\label{fig:app:CC:n:build:SDSL}
  \end{subfigure}
  \begin{subfigure}[t]{\appfigwidth}
    \includegraphics[width=\linewidth]{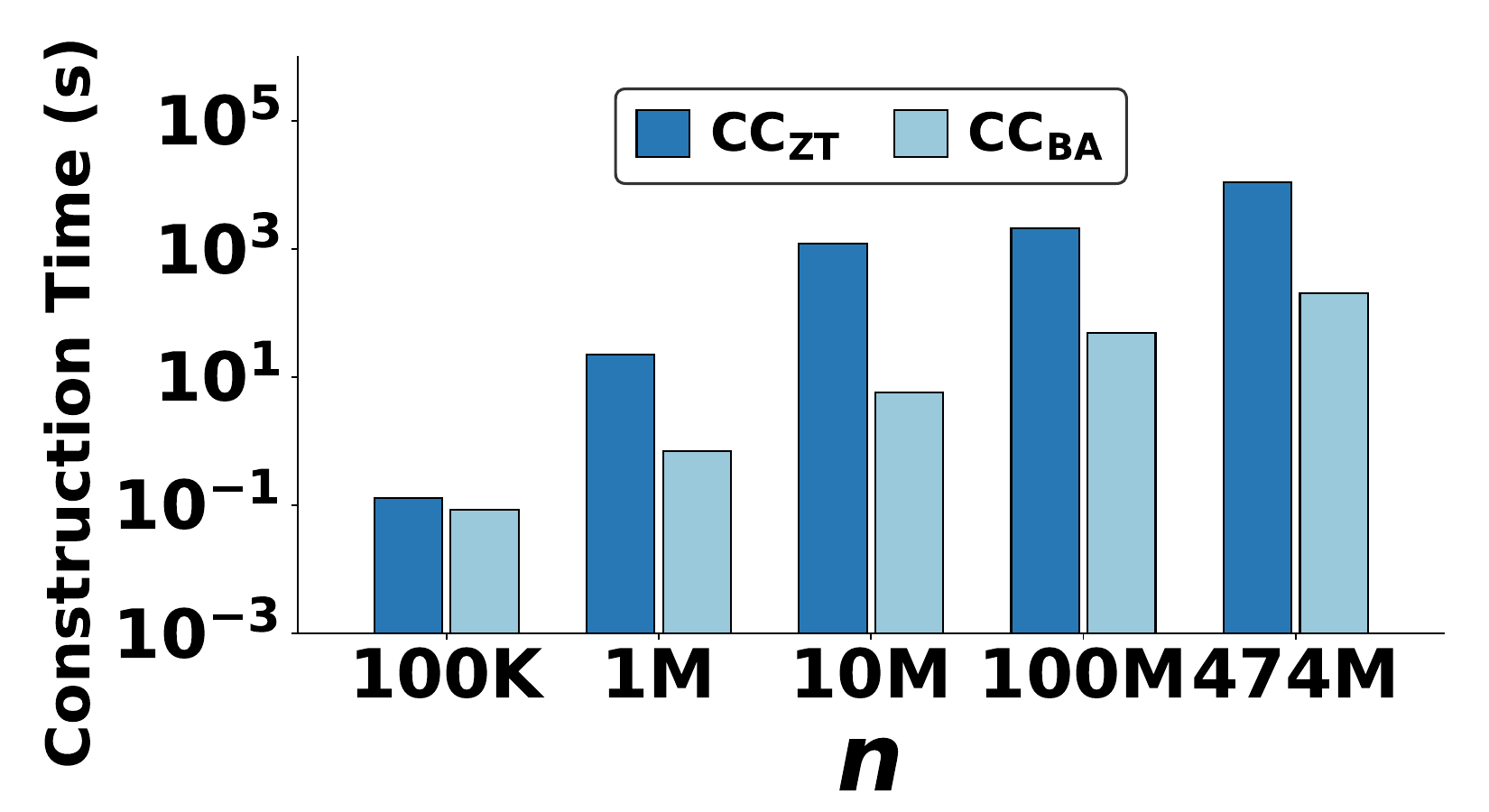}
    \caption{Constr.\ time vs. $n$}\label{fig:app:CC:n:build:WIKI}
  \end{subfigure}
  \vspace{\captionspacing}
  \vspace{+2mm}
  \caption{Index size of our \CC index vs. \CCBA on (a) \chr, (b) \sars, (c) \sdsl, and (d) \wiki vs. $n$; construction space of our \CC index vs. \CCBA on (e) \chr, (f) \sars, (g) \sdsl, and (h) \wiki vs. $n$; construction time of our \CC index vs. \CCBA on (i) \chr, (j) \sars, (k) \sdsl, and (l) \wiki vs. $n$.}\label{fig:app:CC:cost}
\end{figure}

\begin{figure}[ht]
  \centering
  \begin{subfigure}[t]{\appfigwidth}
    \includegraphics[width=\linewidth]{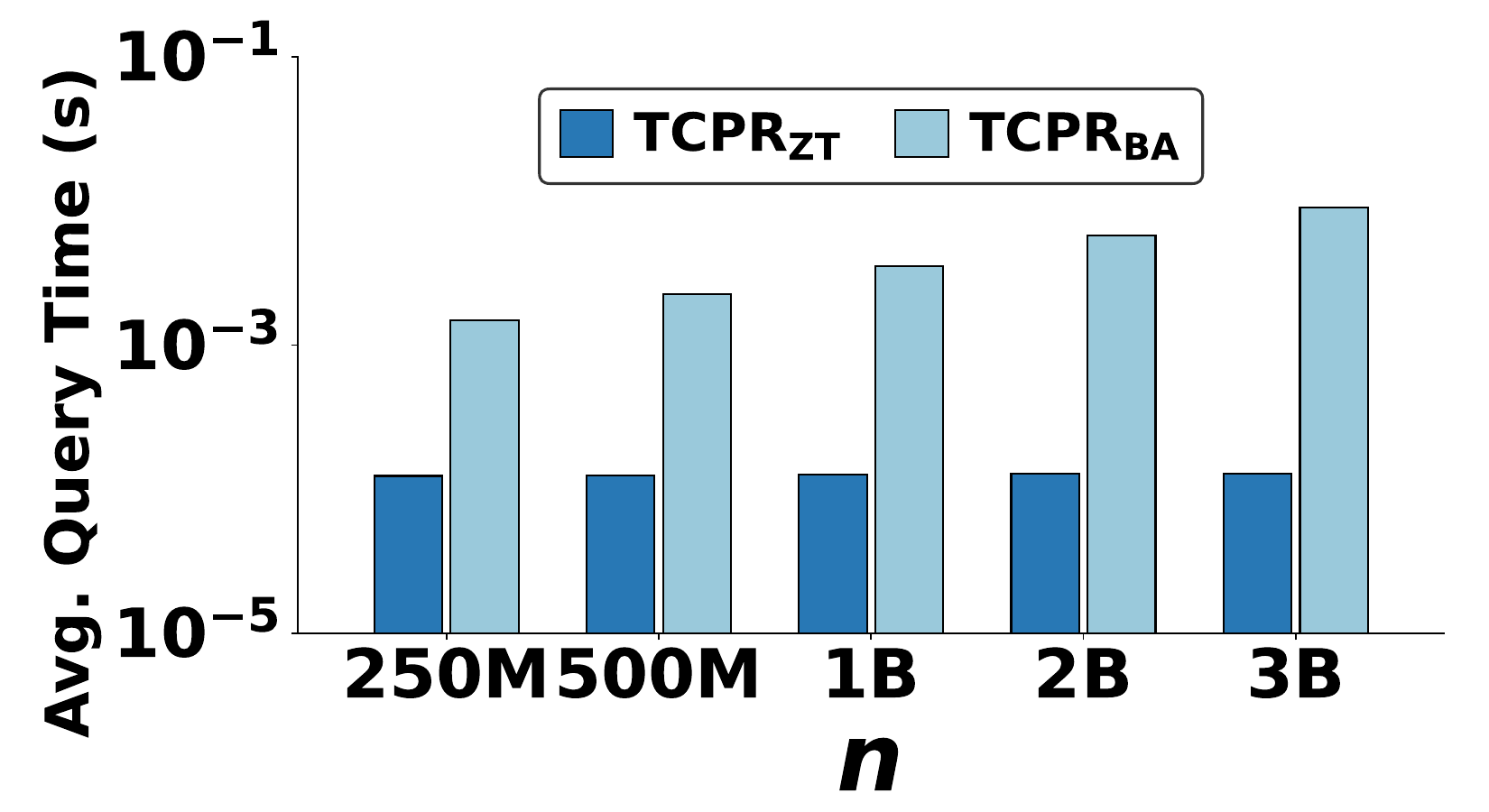}
    \caption{Query time vs. $n$}\label{fig:app:TF:n:query:BST}
  \end{subfigure}
  \begin{subfigure}[t]{\appfigwidth}
    \includegraphics[width=\linewidth]{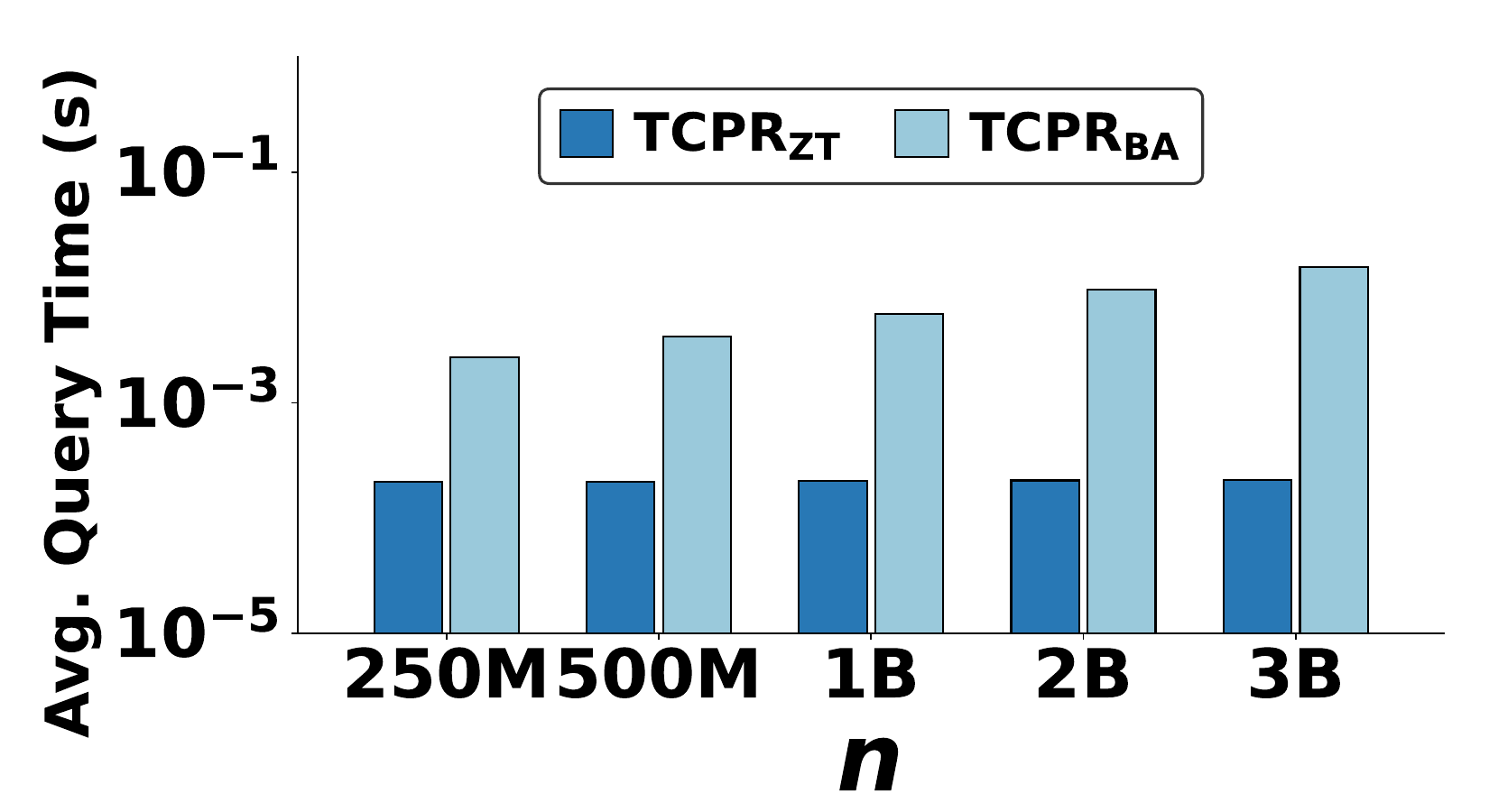}
    \caption{Query time vs. $n$}\label{fig:app:TF:n:query:SARS}
  \end{subfigure}
  \begin{subfigure}[t]{\appfigwidth}
    \includegraphics[width=\linewidth]{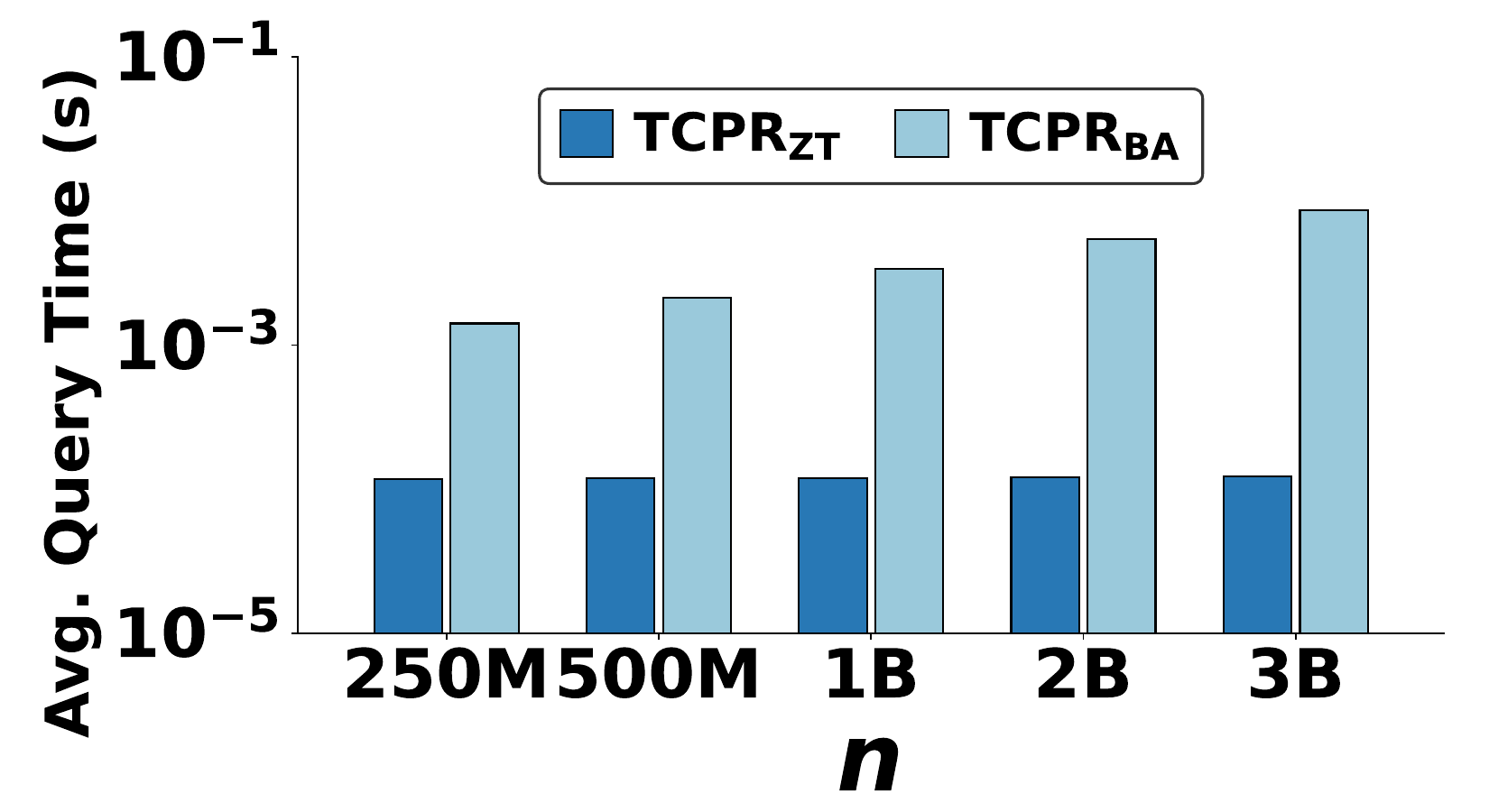}
    \caption{Query time vs. $n$}\label{fig:app:TF:n:query:SDSL}
  \end{subfigure}
  \begin{subfigure}[t]{\appfigwidth}
    \includegraphics[width=\linewidth]{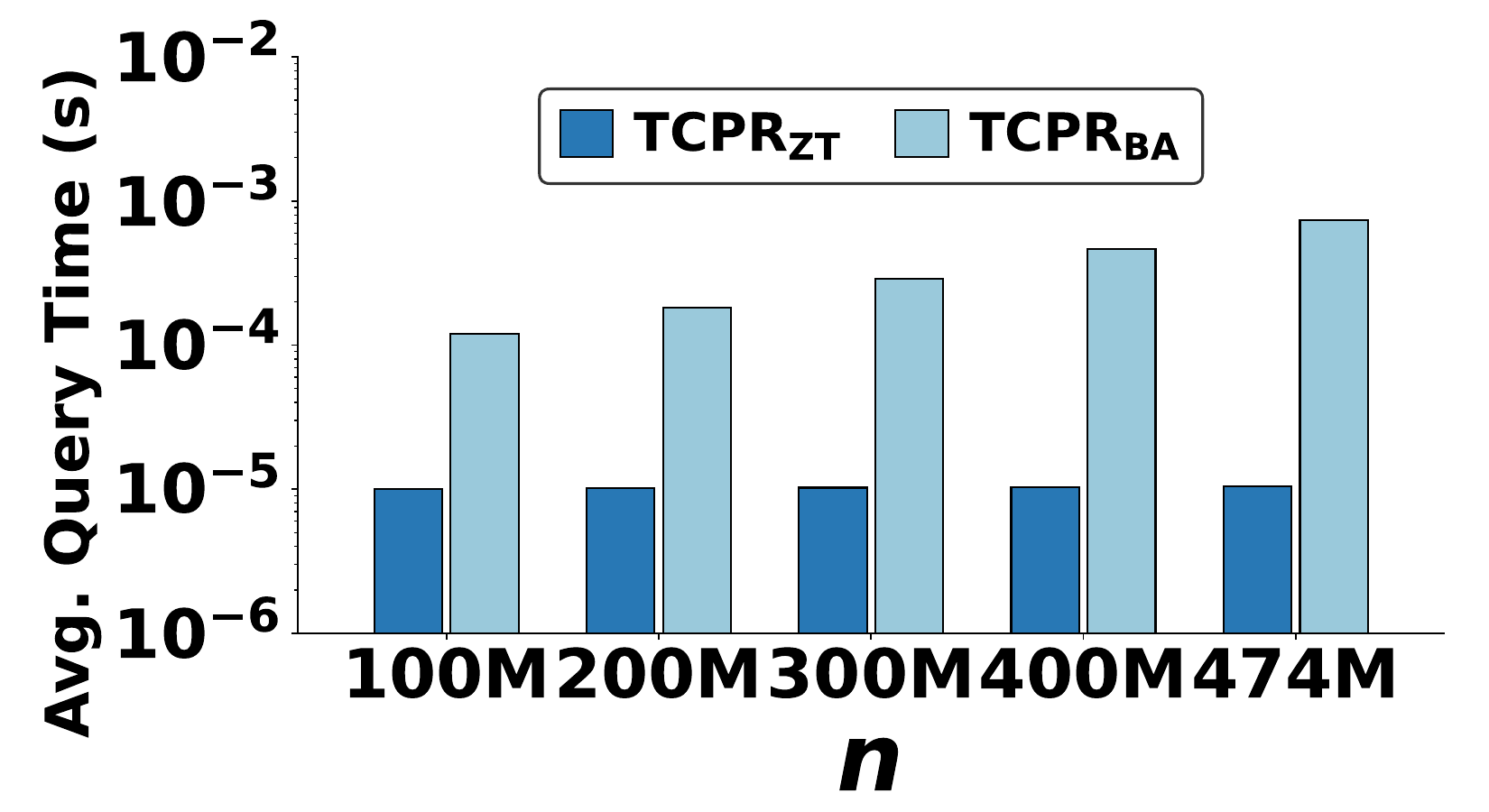}
    \caption{Query time vs. $n$}\label{fig:app:TF:n:query:WIKI}
  \end{subfigure}\\[0pt]
  \begin{subfigure}[t]{\appfigwidth}
    \includegraphics[width=\linewidth]{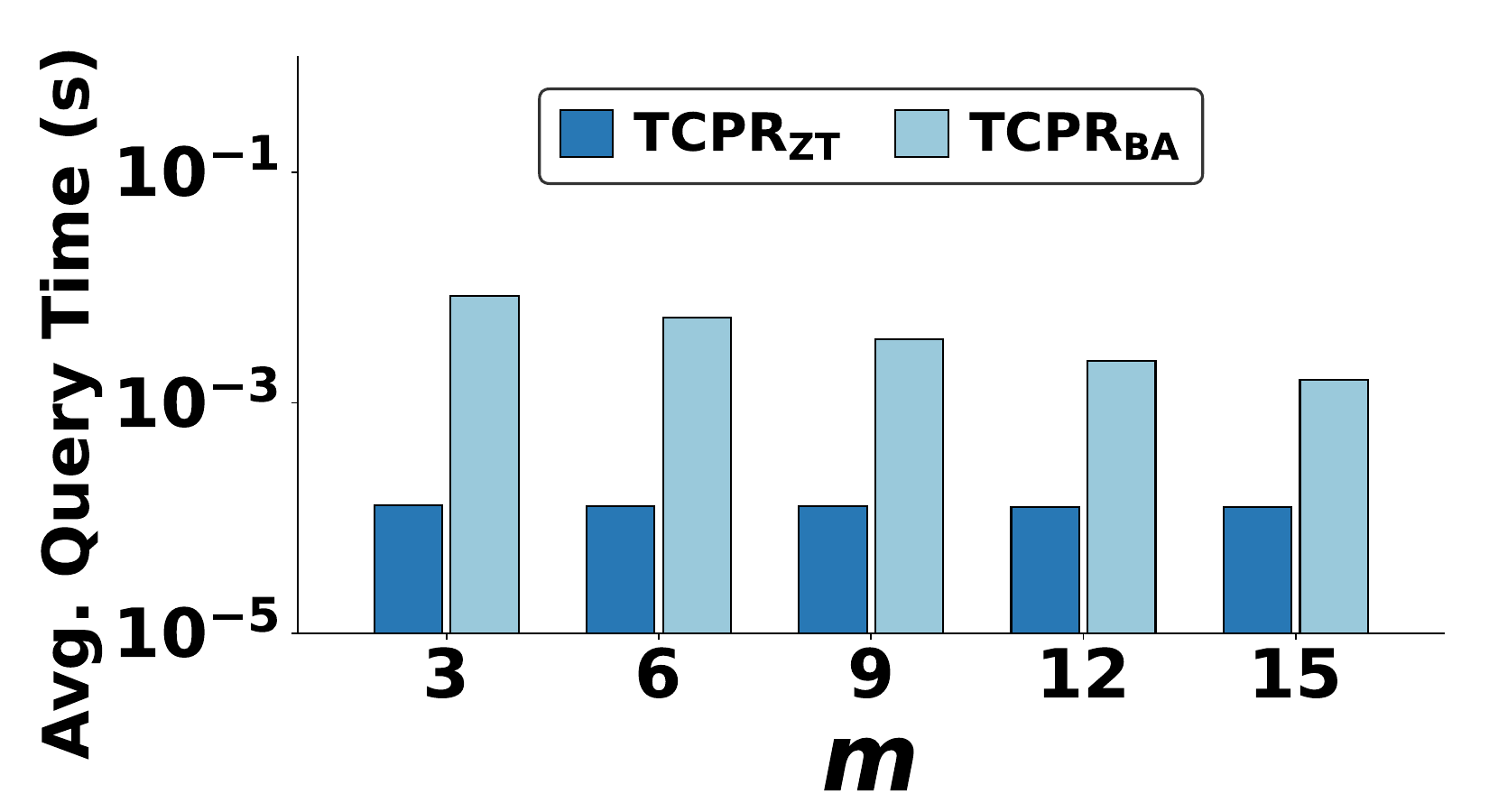}
    \caption{Query time vs. $m$}\label{fig:app:TF:m:query:BST}
  \end{subfigure}
  \begin{subfigure}[t]{\appfigwidth}
    \includegraphics[width=\linewidth]{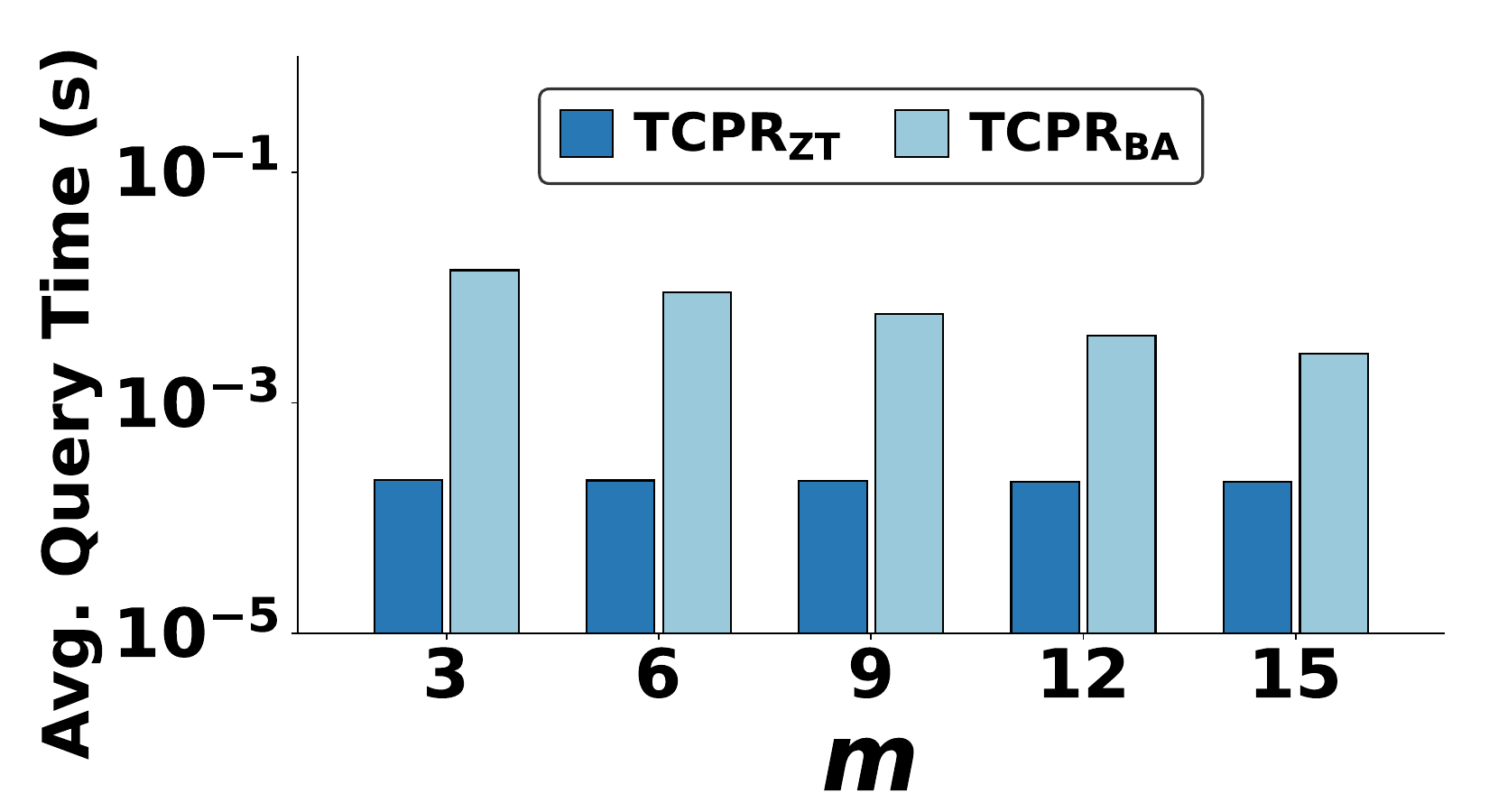}
    \caption{Query time vs. $m$}\label{fig:app:TF:m:query:SARS}
  \end{subfigure}
  \begin{subfigure}[t]{\appfigwidth}
    \includegraphics[width=\linewidth]{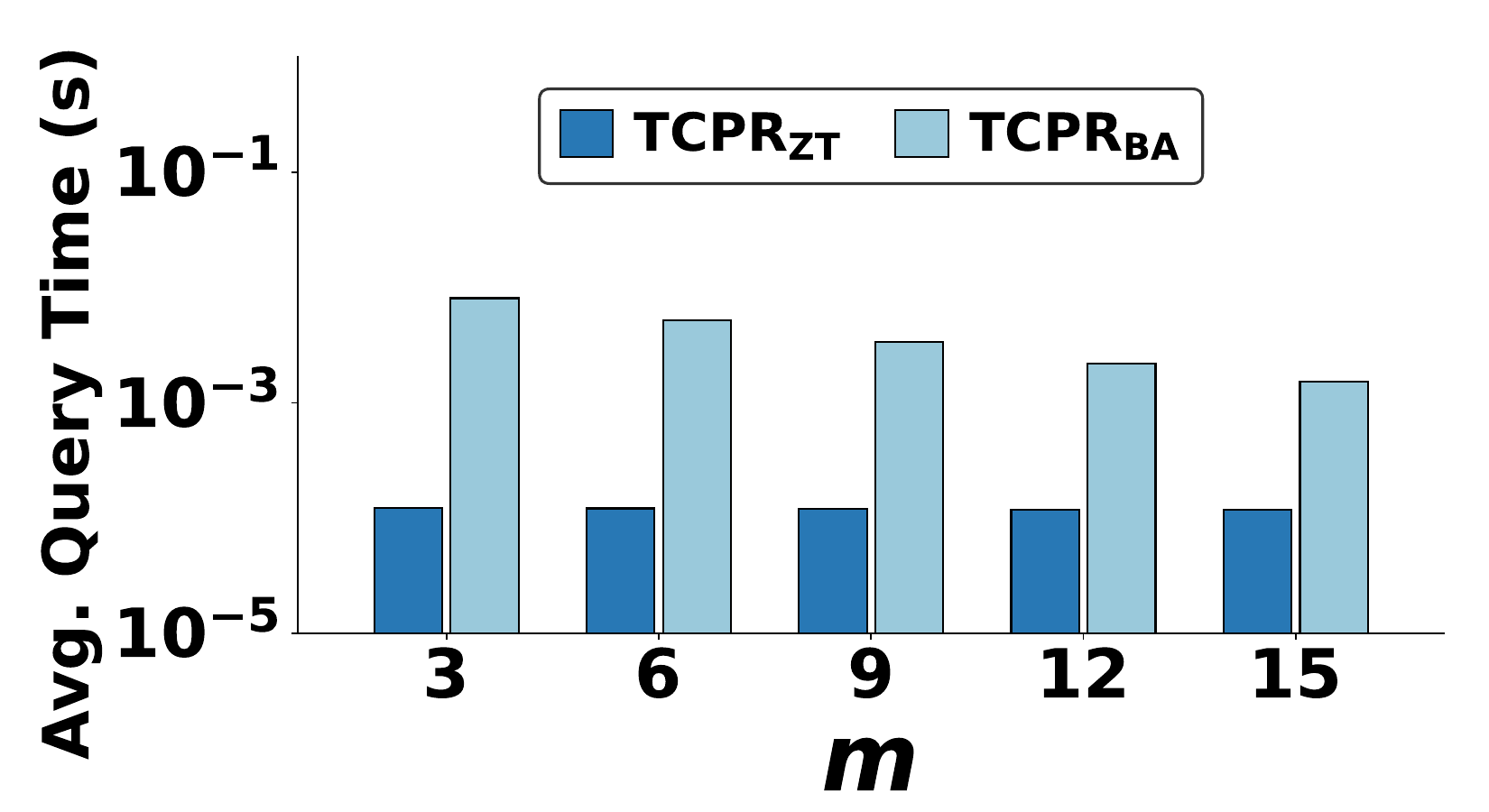}
    \caption{Query time vs. $m$}\label{fig:app:TF:m:query:SDSL}
  \end{subfigure}
  \begin{subfigure}[t]{\appfigwidth}
    \includegraphics[width=\linewidth]{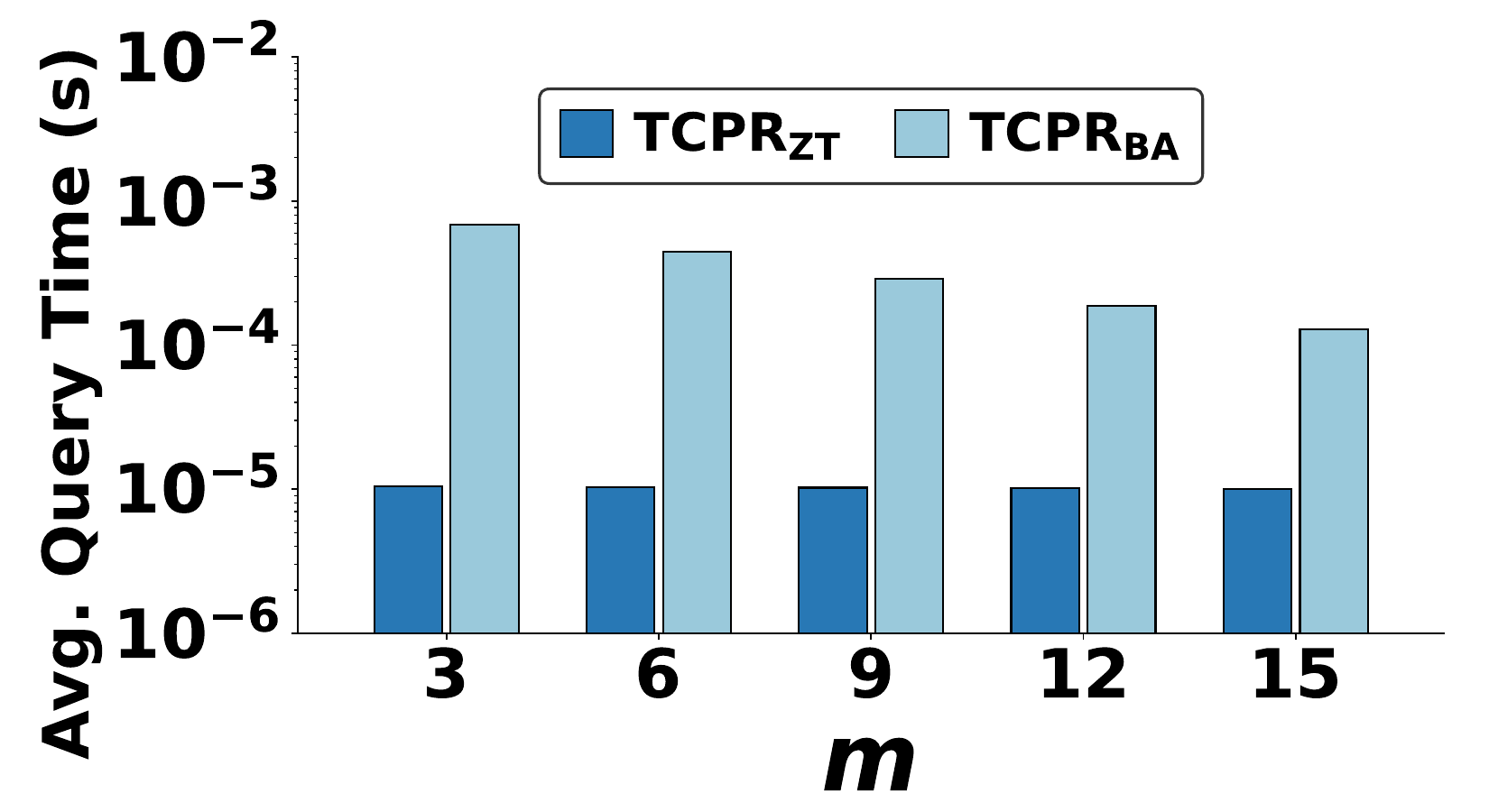}
    \caption{Query time vs. $m$}\label{fig:app:TF:m:query:WIKI}
  \end{subfigure}\\[0pt]
  \begin{subfigure}[t]{\appfigwidth}
    \includegraphics[width=\linewidth]{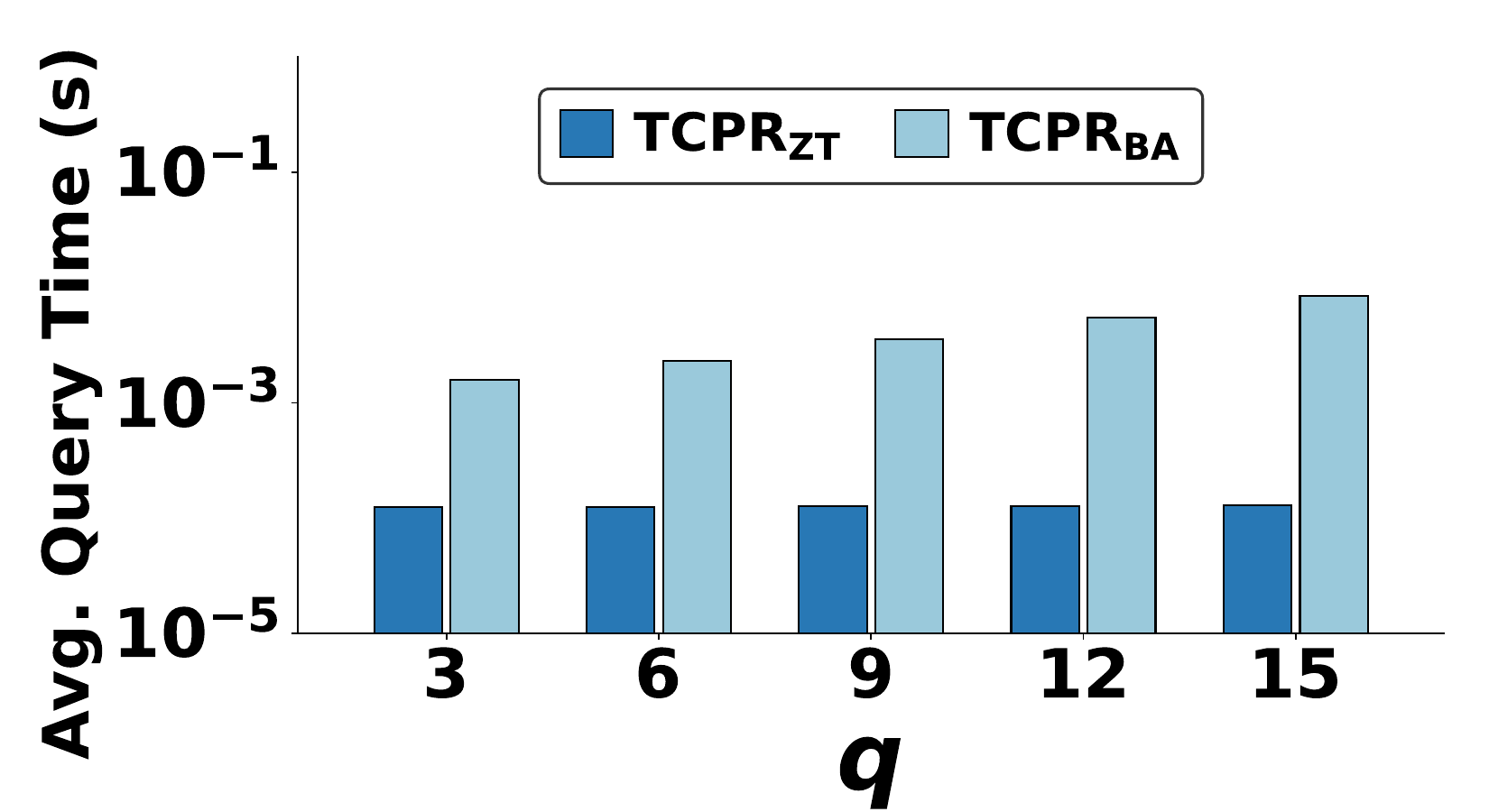}
    \caption{Query time vs. $q$}\label{fig:app:TF:q:query:BST}
  \end{subfigure}
  \begin{subfigure}[t]{\appfigwidth}
    \includegraphics[width=\linewidth]{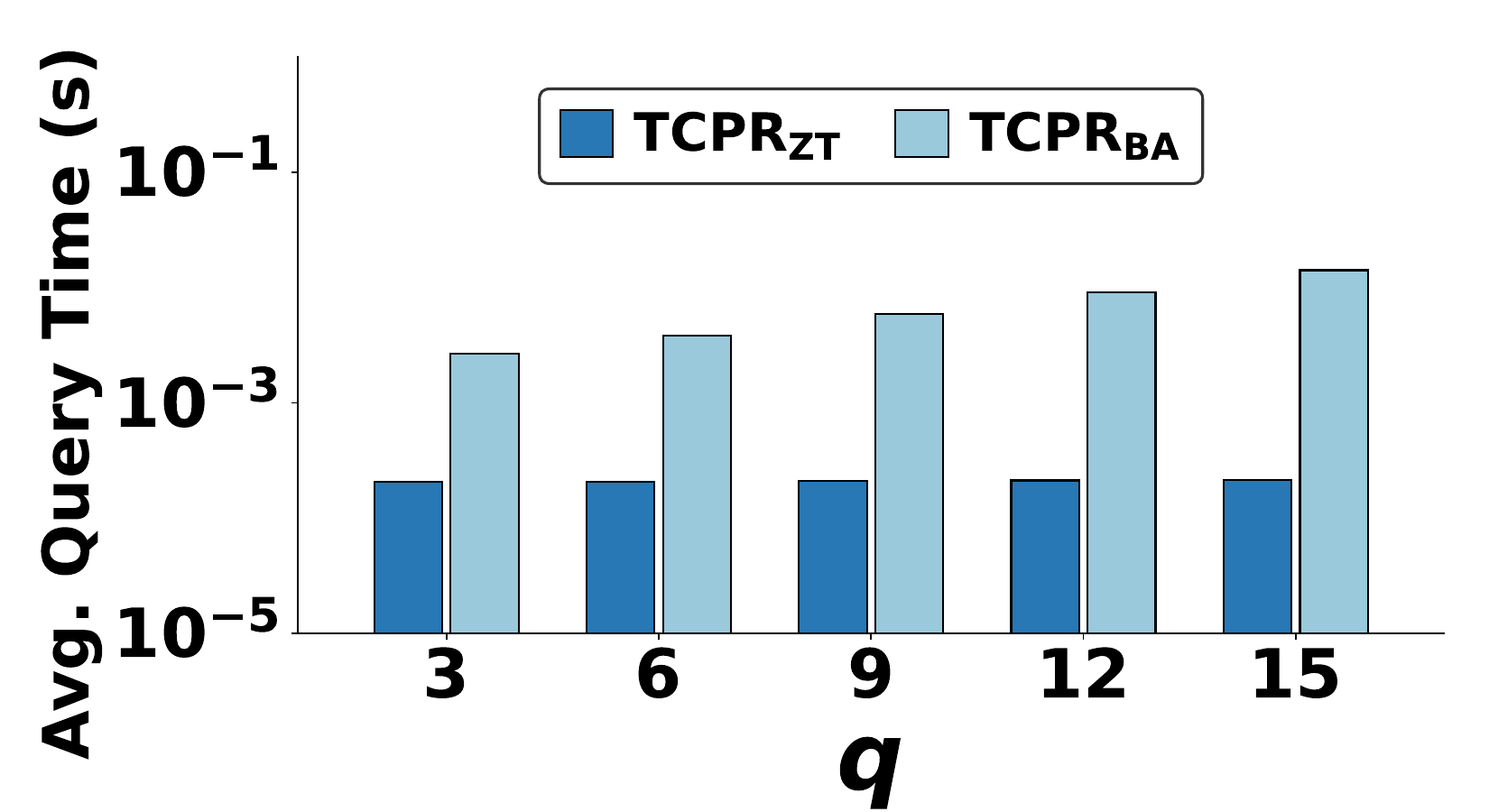}
    \caption{Query time vs. $q$}\label{fig:app:TF:q:query:SARS}
  \end{subfigure}
  \begin{subfigure}[t]{\appfigwidth}
    \includegraphics[width=\linewidth]{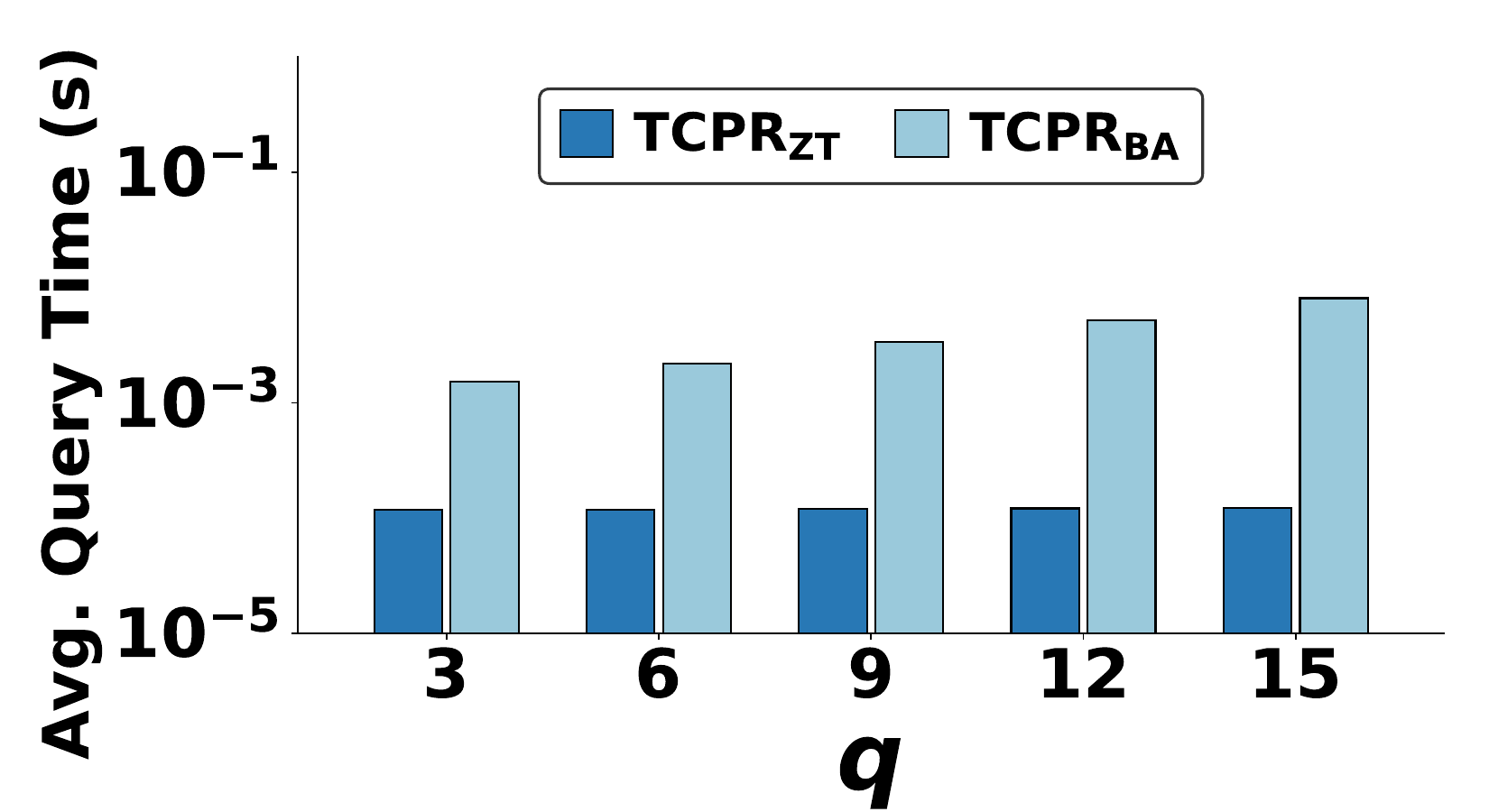}
    \caption{Query time vs. $q$}\label{fig:app:TF:q:query:SDSL}
  \end{subfigure}
  \begin{subfigure}[t]{\appfigwidth}
    \includegraphics[width=\linewidth]{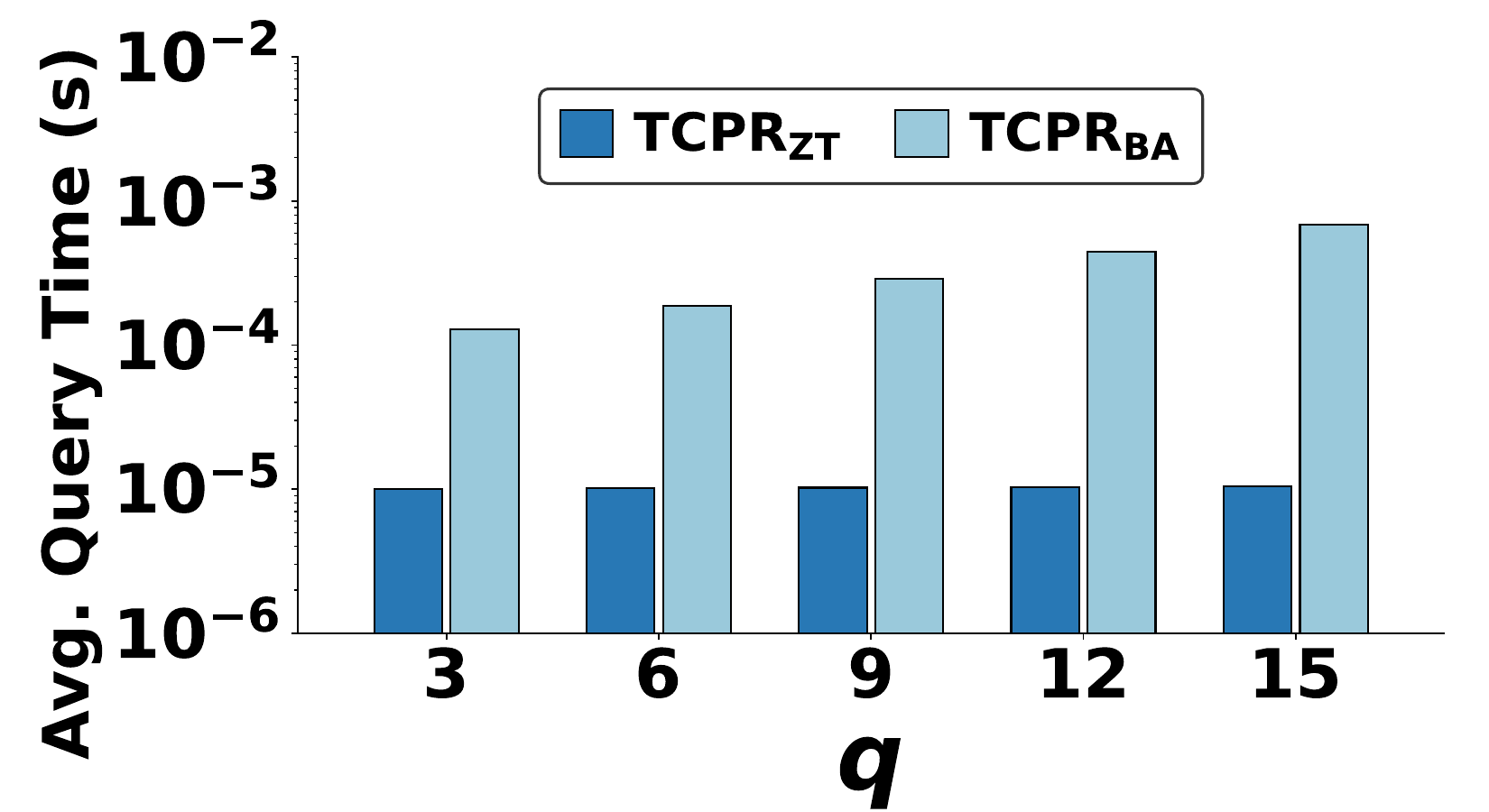}
    \caption{Query time vs. $q$}\label{fig:app:TF:q:query:WIKI}
  \end{subfigure}\\[0pt]
  \begin{subfigure}[t]{\appfigwidth}
    \includegraphics[width=\linewidth]{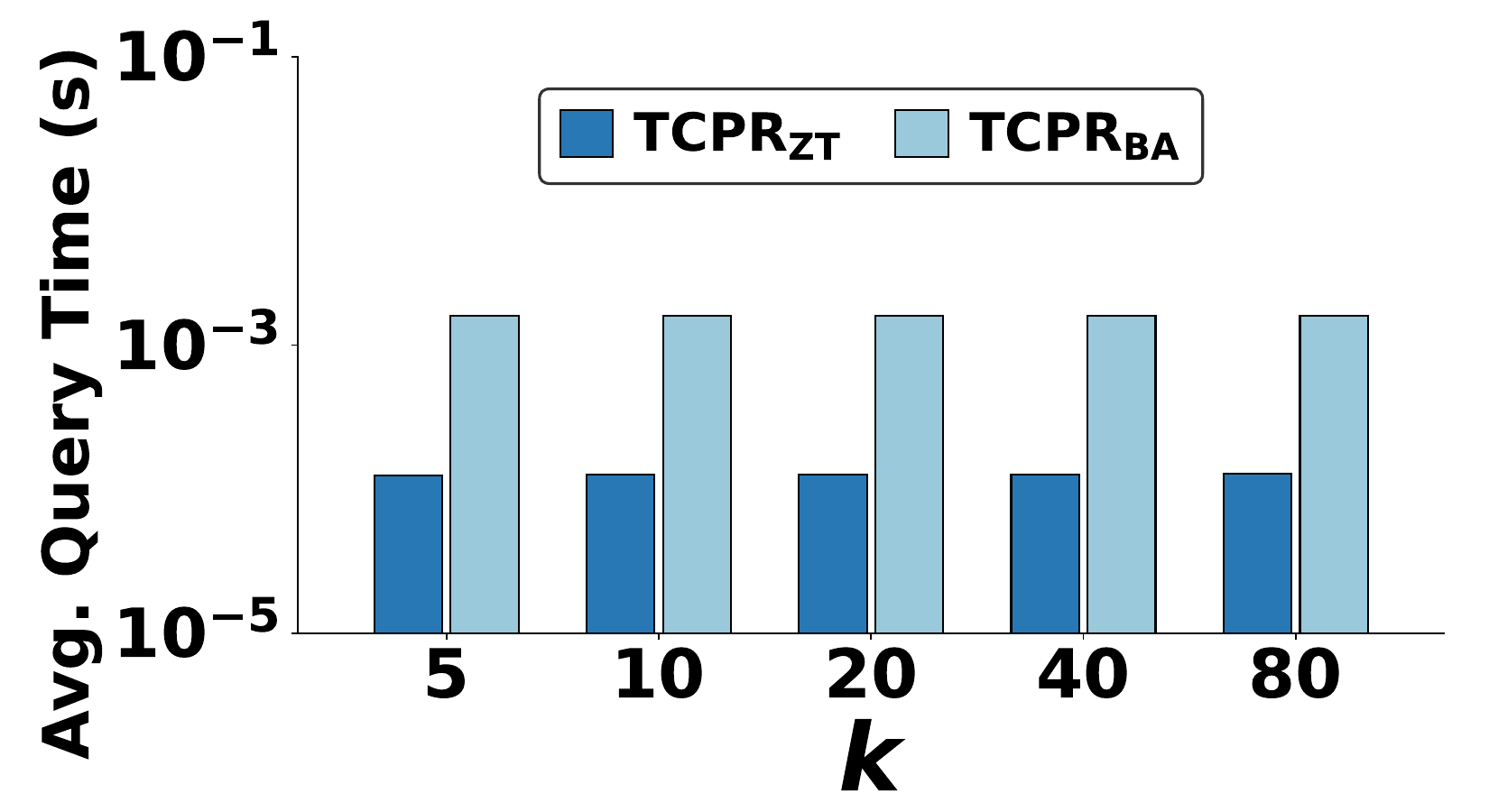}
    \caption{Query time vs. $k$}\label{fig:app:TF:k:query:BST}
  \end{subfigure}
  \begin{subfigure}[t]{\appfigwidth}
    \includegraphics[width=\linewidth]{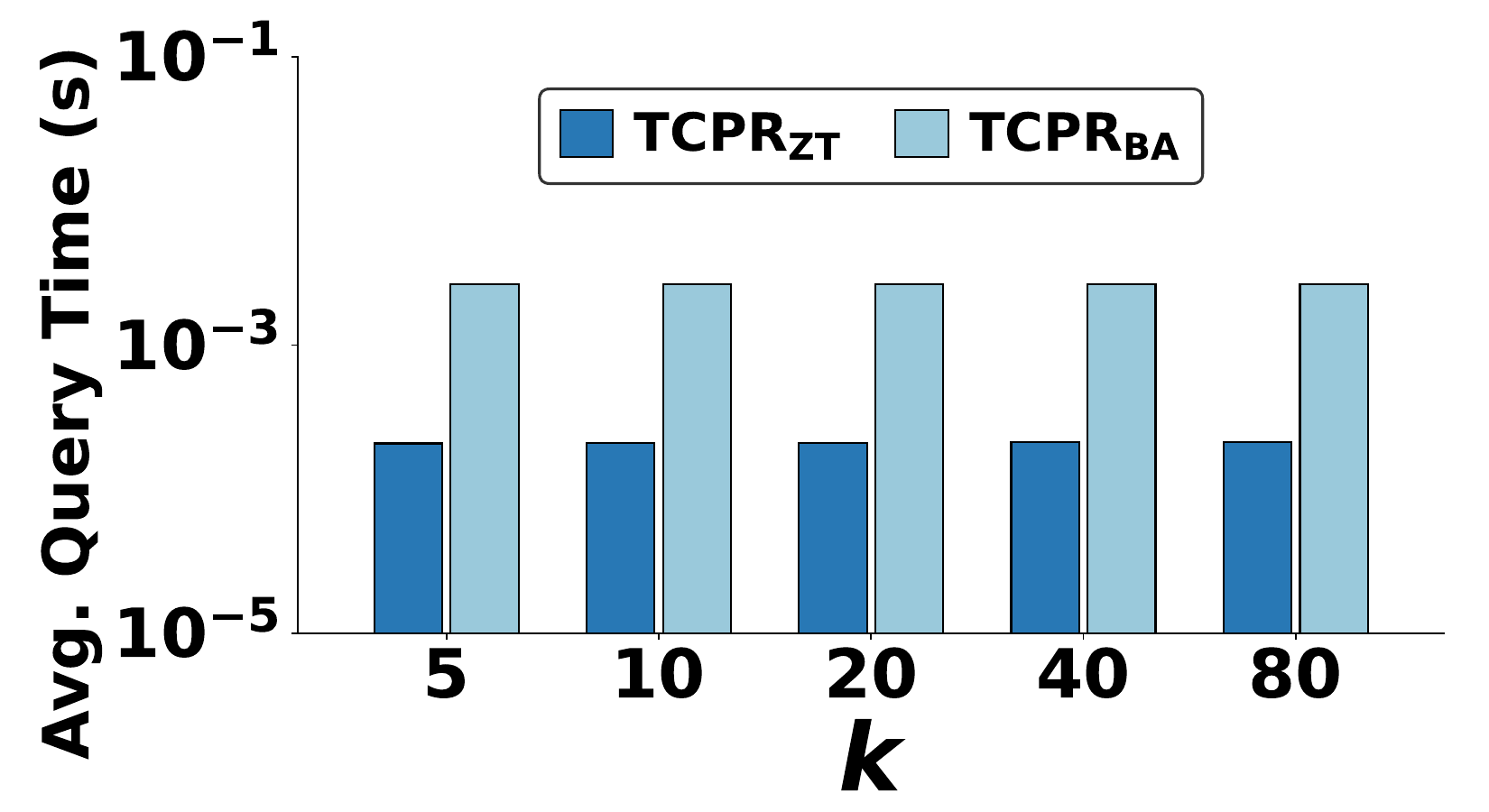}
    \caption{Query time vs. $k$}\label{fig:app:TF:k:query:SARS}
  \end{subfigure}
  \begin{subfigure}[t]{\appfigwidth}
    \includegraphics[width=\linewidth]{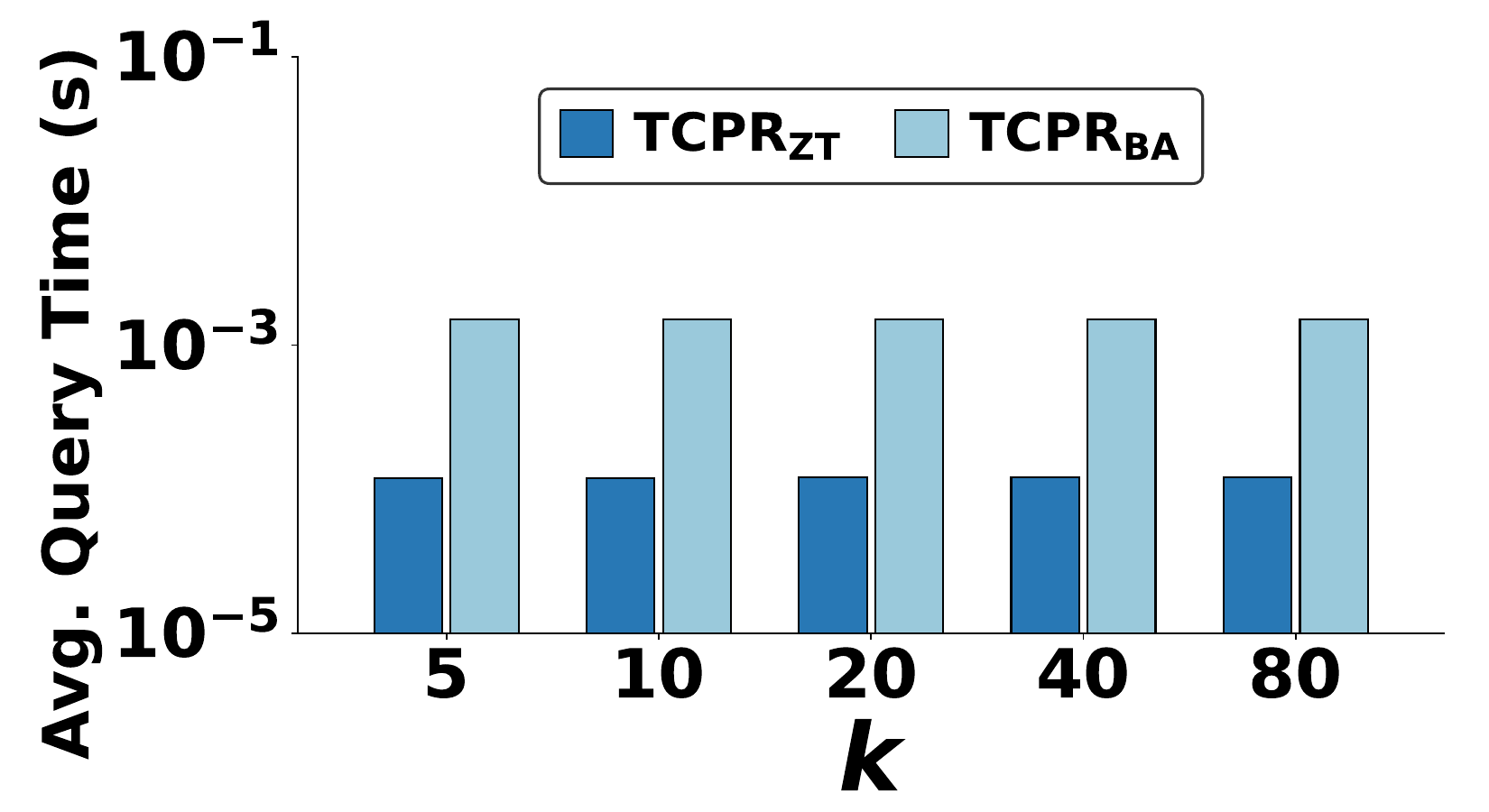}
    \caption{Query time vs. $k$}\label{fig:app:TF:k:query:SDSL}
  \end{subfigure}
  \begin{subfigure}[t]{\appfigwidth}
    \includegraphics[width=\linewidth]{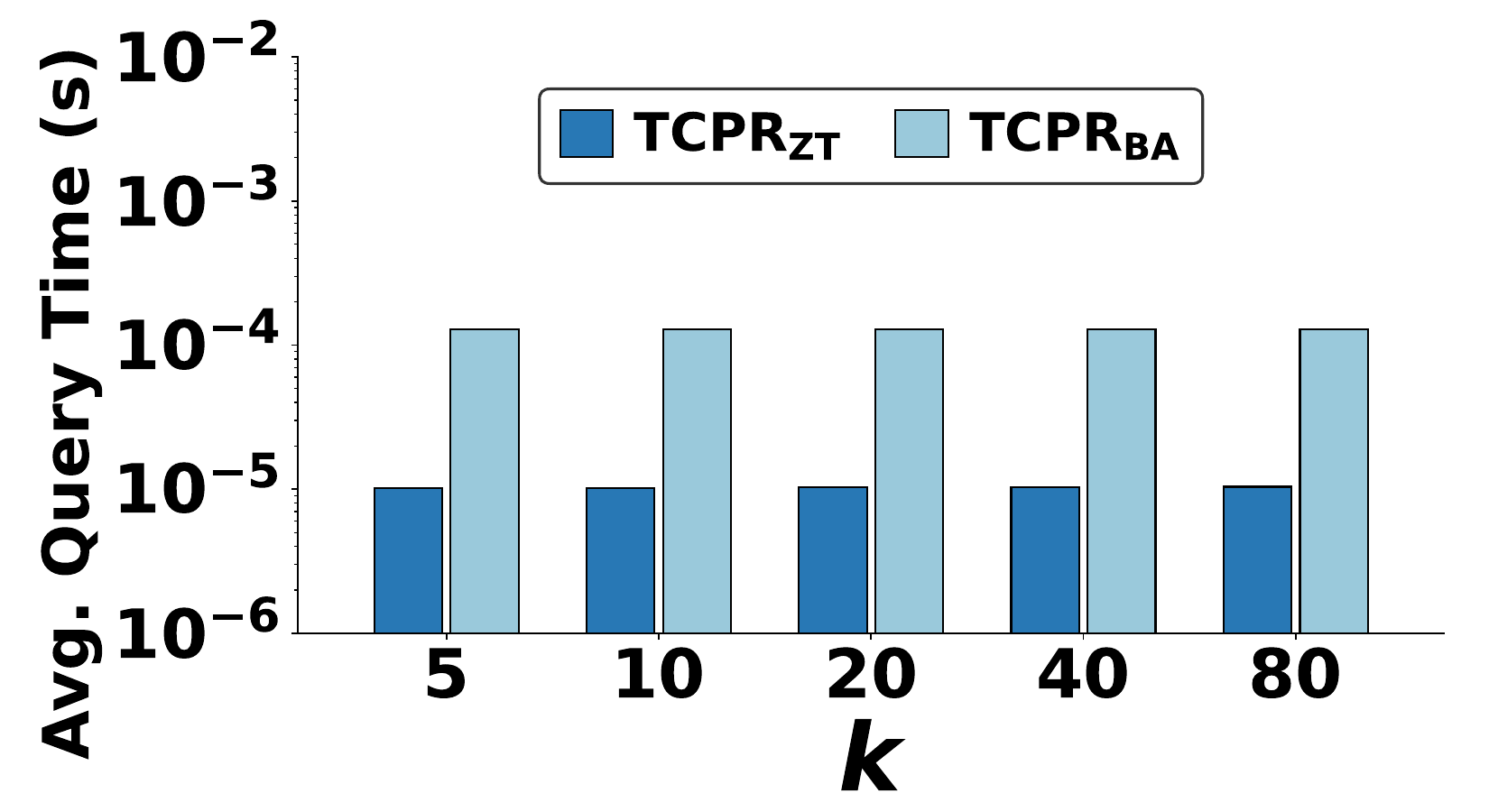}
    \caption{Query time vs. $k$}\label{fig:app:TF:k:query:WIKI}
  \end{subfigure}
  \vspace{\captionspacing}
  \vspace{+2mm}
  \caption{Query time of our \TCPR index with the \textsf{TF} scoring function vs. \TCPRBA on (a) \bst, (b) \sars, (c) \sdsl, and (d) \wiki vs. $n$; on (e) \bst, (f) \sars, (g) \sdsl, and (h) \wiki vs. $m$; on (i) \bst, (j) \sars, (k) \sdsl, and (l) \wiki vs. $q$; on (m) \bst, (n) \sars, (o) \sdsl, and (p) \wiki vs. $k$.}\label{fig:app:TF:query}
\end{figure}

\begin{figure}[ht]
  \centering
  \begin{subfigure}[t]{\appfigwidth}
    \includegraphics[width=\linewidth]{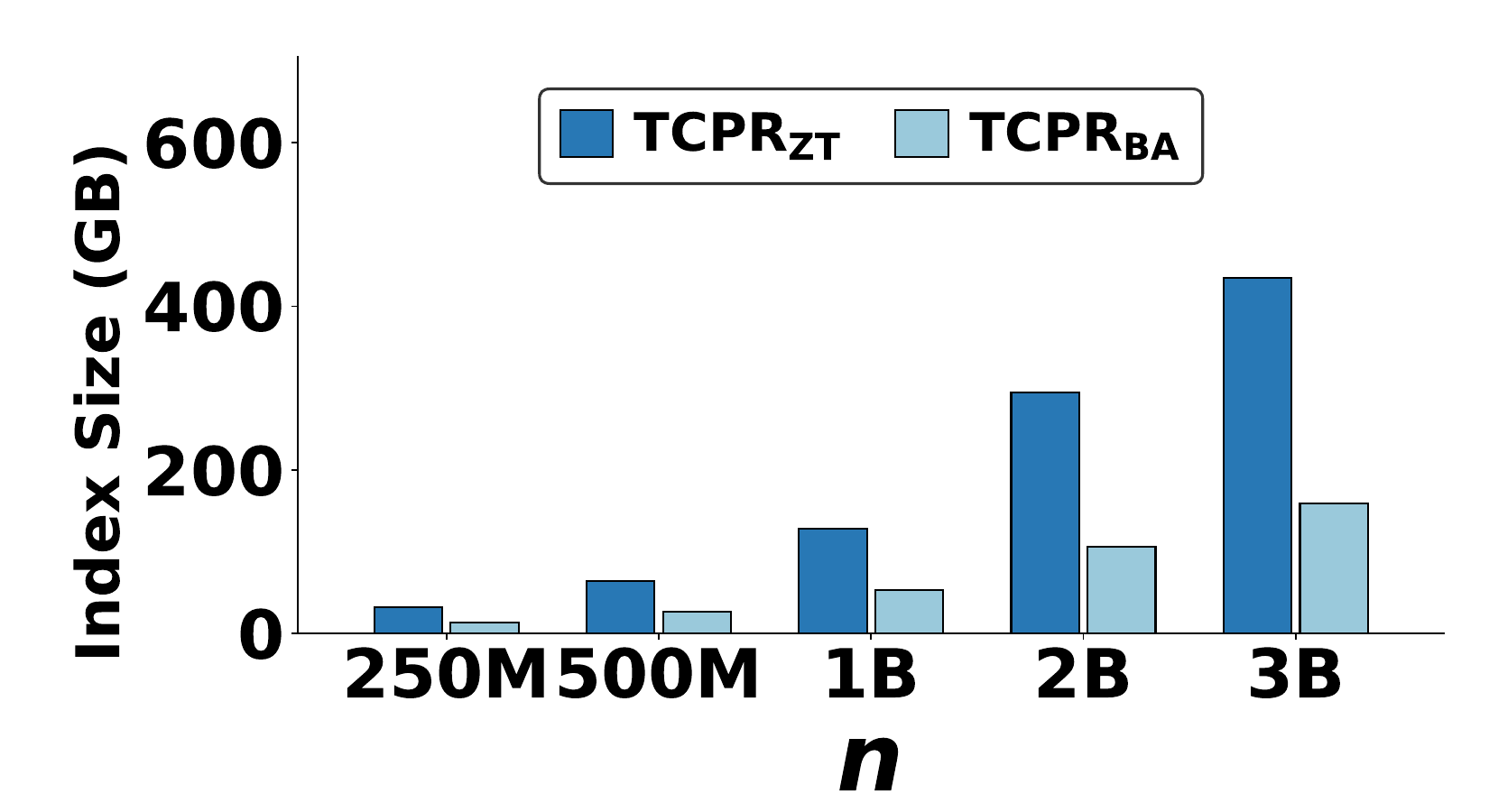}
    \caption{Index size vs. $n$}\label{fig:app:TF:n:index:BST}
  \end{subfigure}
  \begin{subfigure}[t]{\appfigwidth}
    \includegraphics[width=\linewidth]{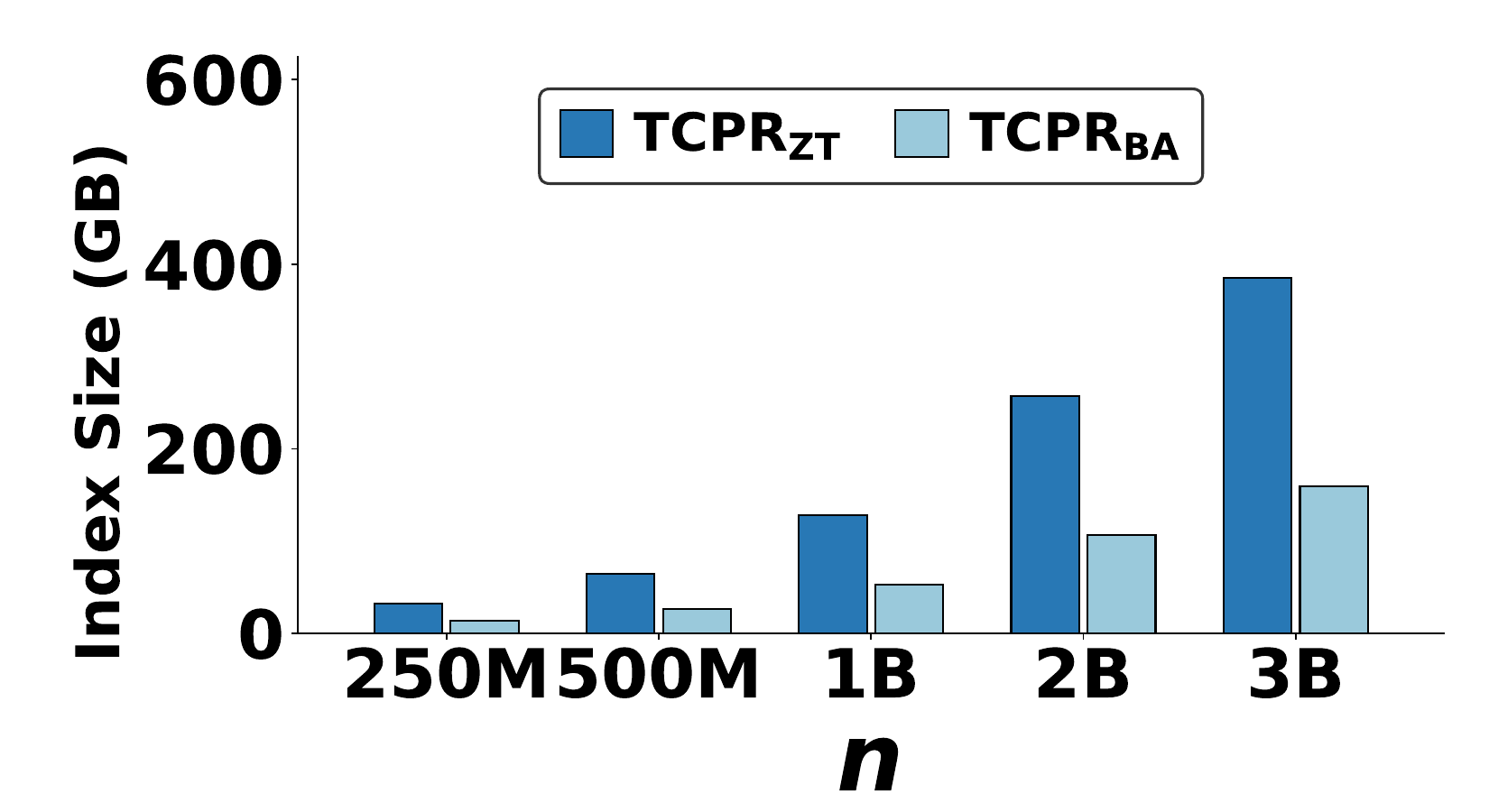}
    \caption{Index size vs. $n$}\label{fig:app:TF:n:index:SARS}
  \end{subfigure}
  \begin{subfigure}[t]{\appfigwidth}
    \includegraphics[width=\linewidth]{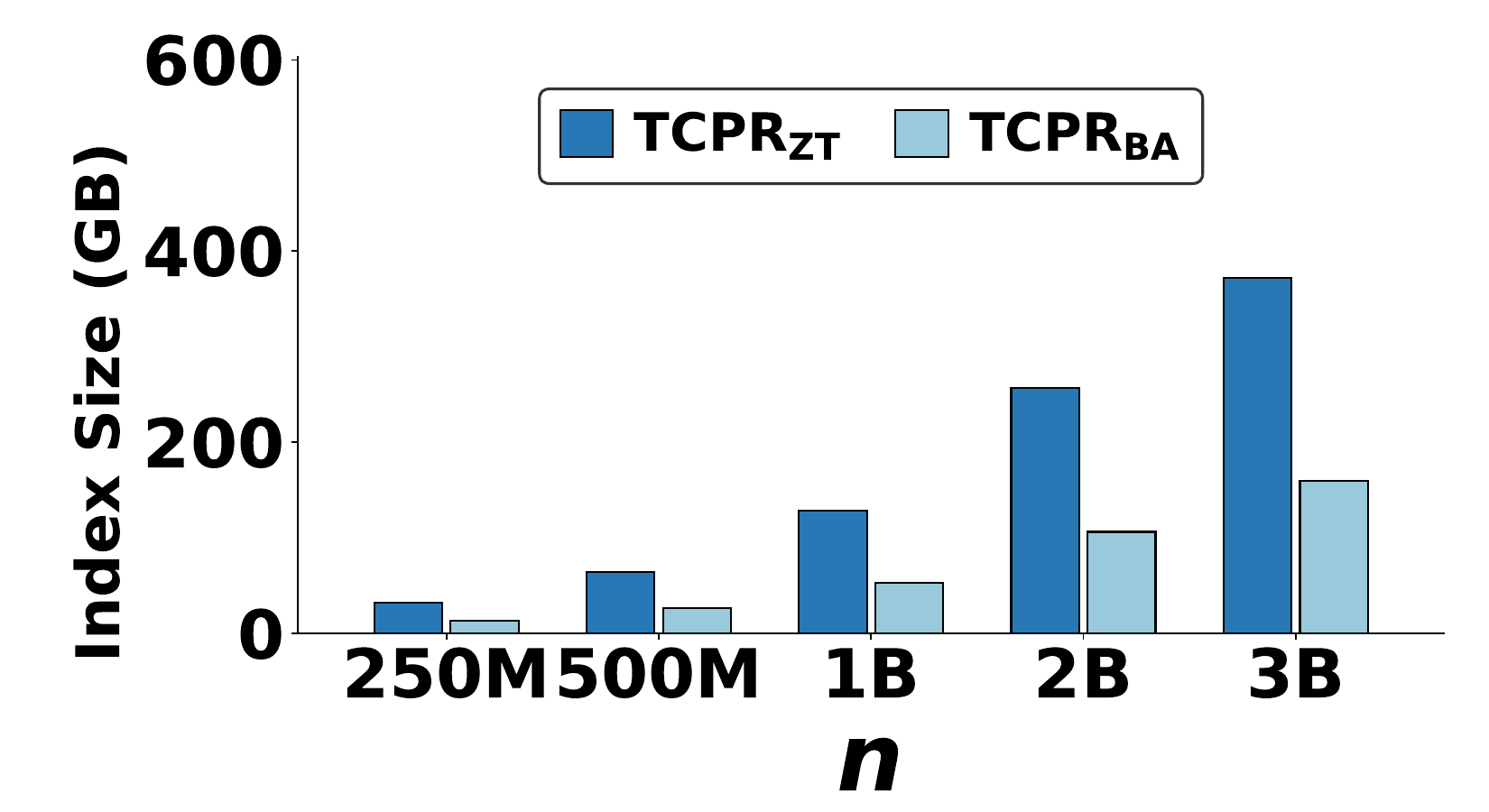}
    \caption{Index size vs. $n$}\label{fig:app:TF:n:index:SDSL}
  \end{subfigure}
  \begin{subfigure}[t]{\appfigwidth}
    \includegraphics[width=\linewidth]{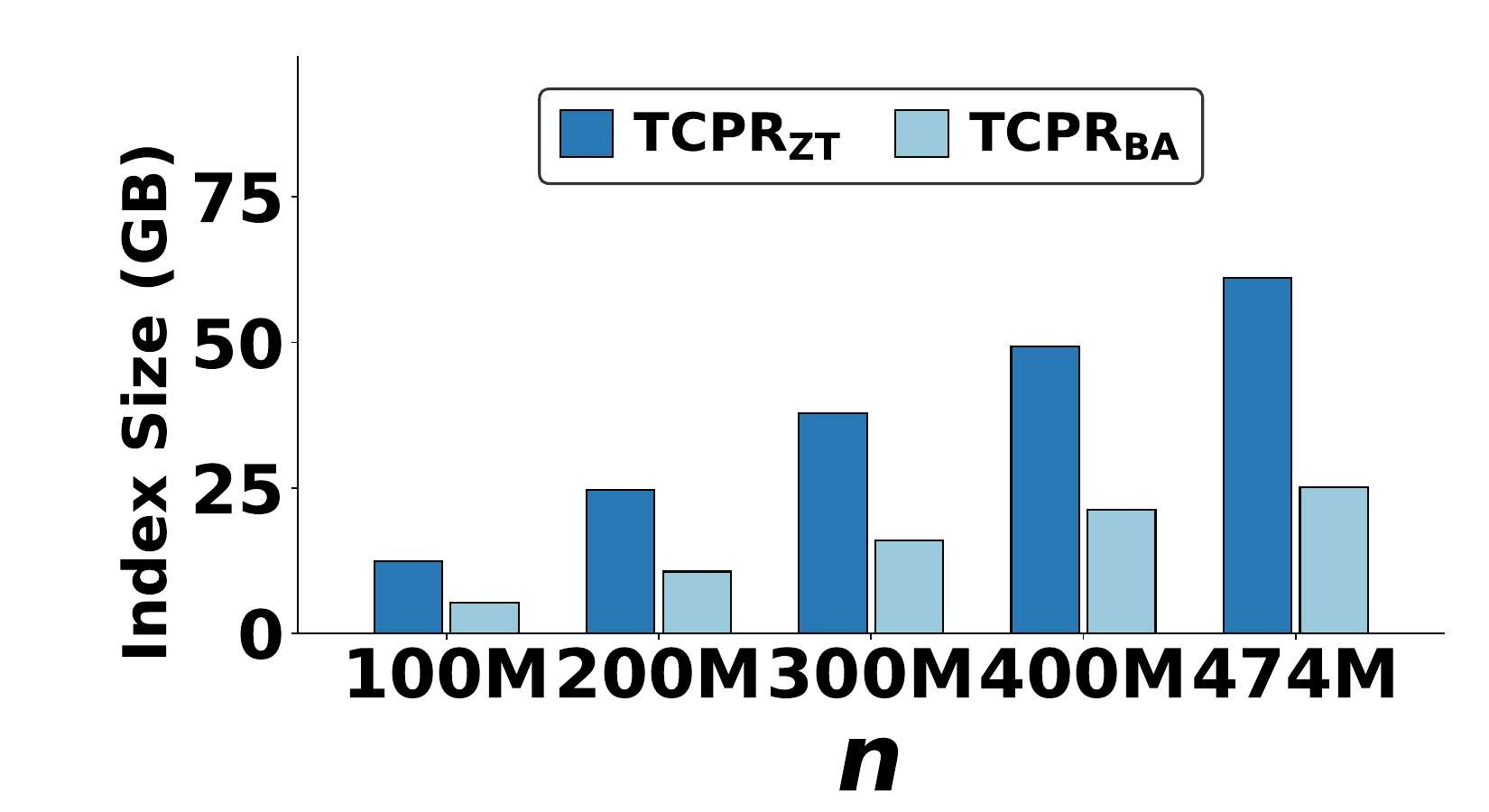}
    \caption{Index size vs. $n$}\label{fig:app:TF:n:index:WIKI}
  \end{subfigure}\\[0pt]
  \begin{subfigure}[t]{\appfigwidth}
    \includegraphics[width=\linewidth]{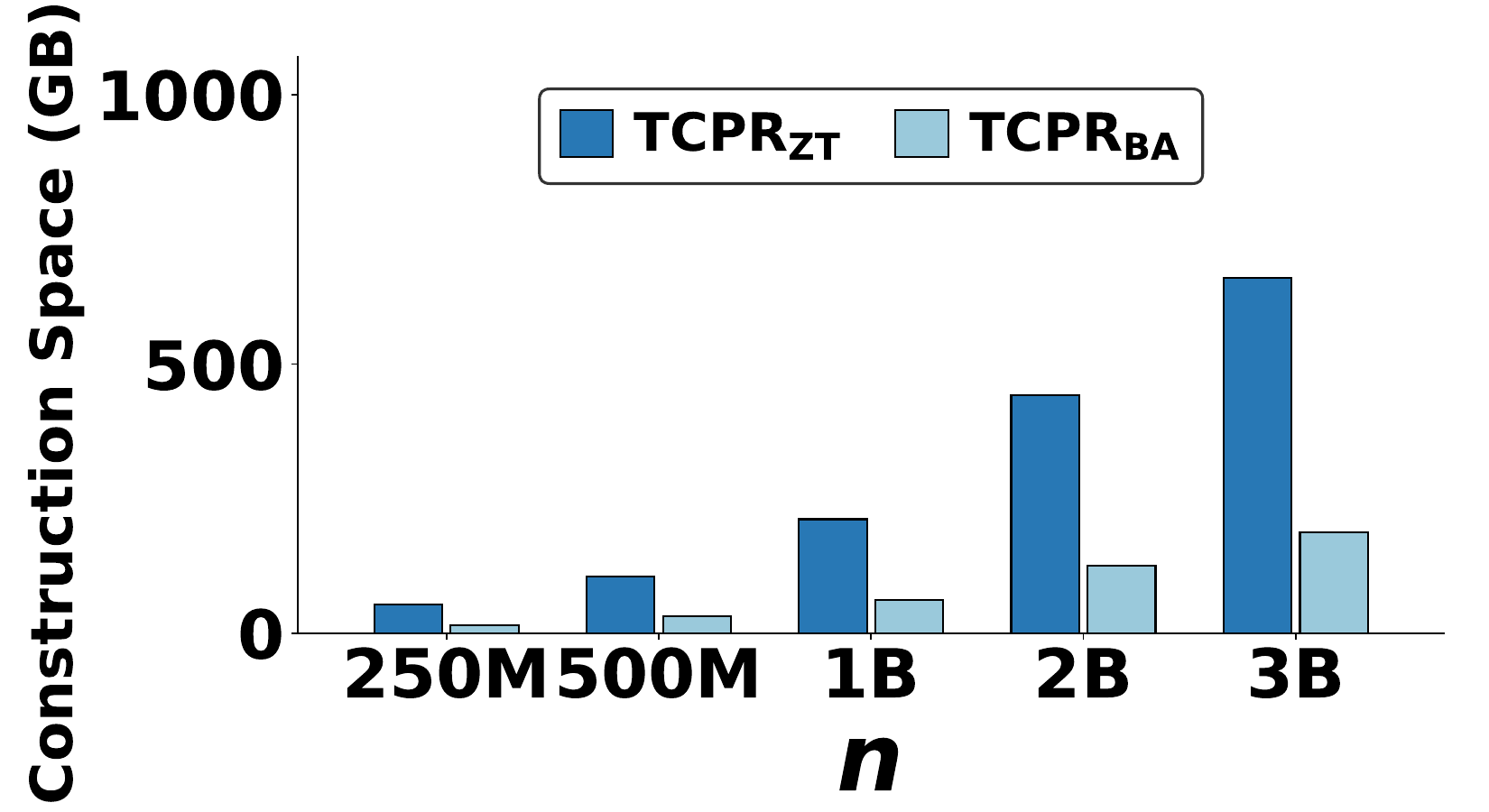}
    \caption{Constr.\ space vs. $n$}\label{fig:app:TF:n:rss:BST}
  \end{subfigure}
  \begin{subfigure}[t]{\appfigwidth}
    \includegraphics[width=\linewidth]{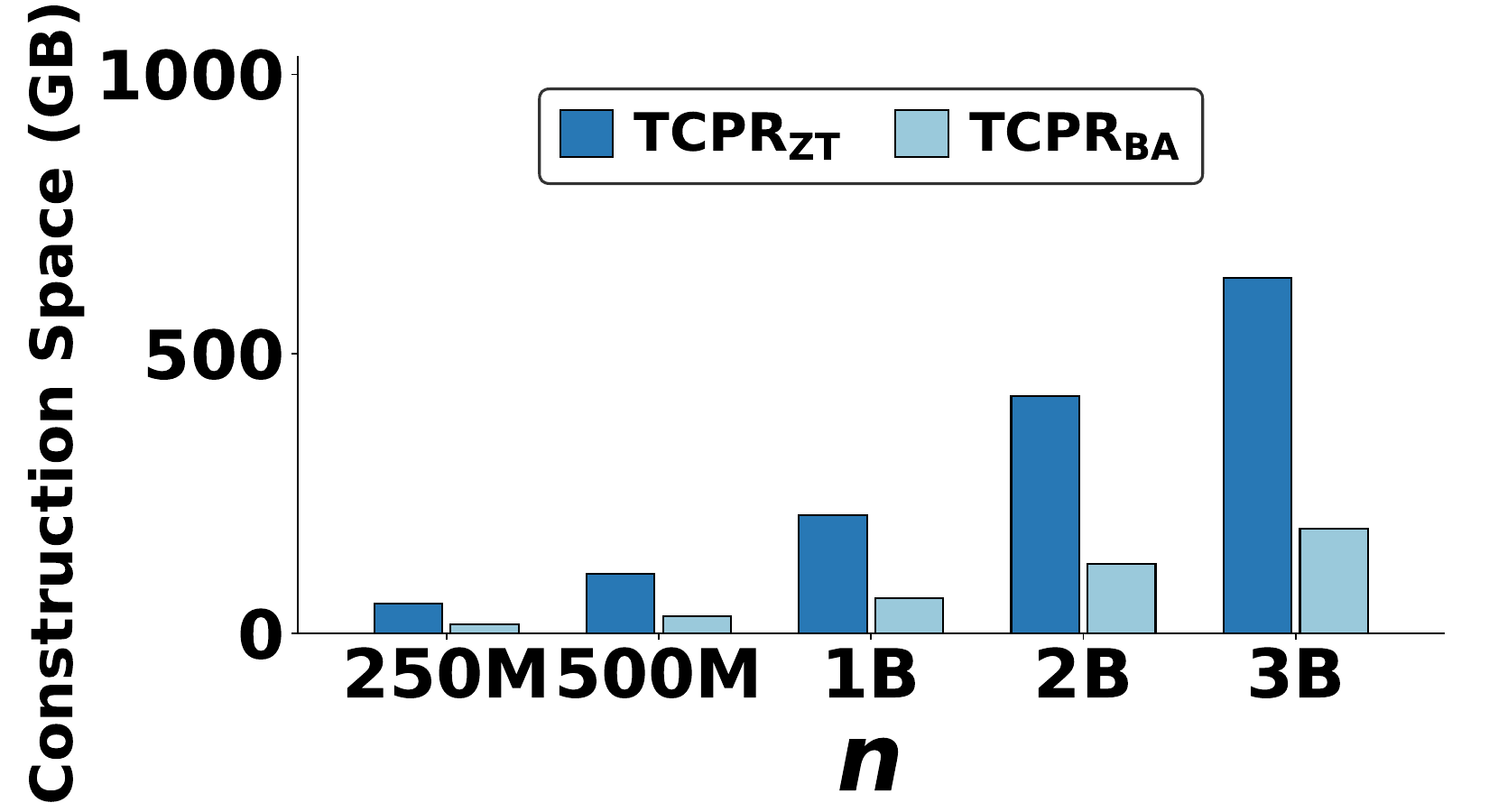}
    \caption{Constr.\ space vs. $n$}\label{fig:app:TF:n:rss:SARS}
  \end{subfigure}
  \begin{subfigure}[t]{\appfigwidth}
    \includegraphics[width=\linewidth]{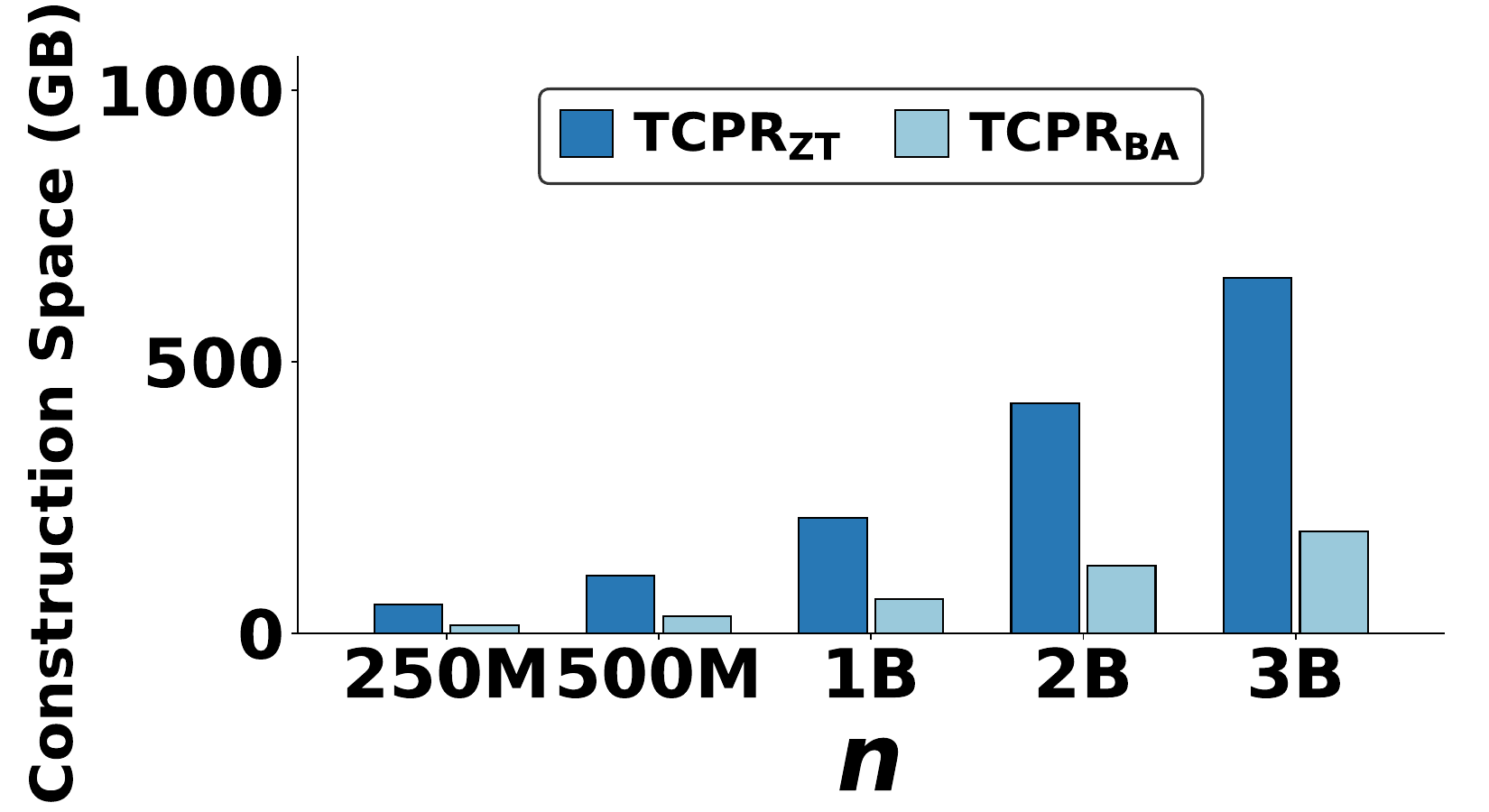}
    \caption{Constr.\ space vs. $n$}\label{fig:app:TF:n:rss:SDSL}
  \end{subfigure}
  \begin{subfigure}[t]{\appfigwidth}
    \includegraphics[width=\linewidth]{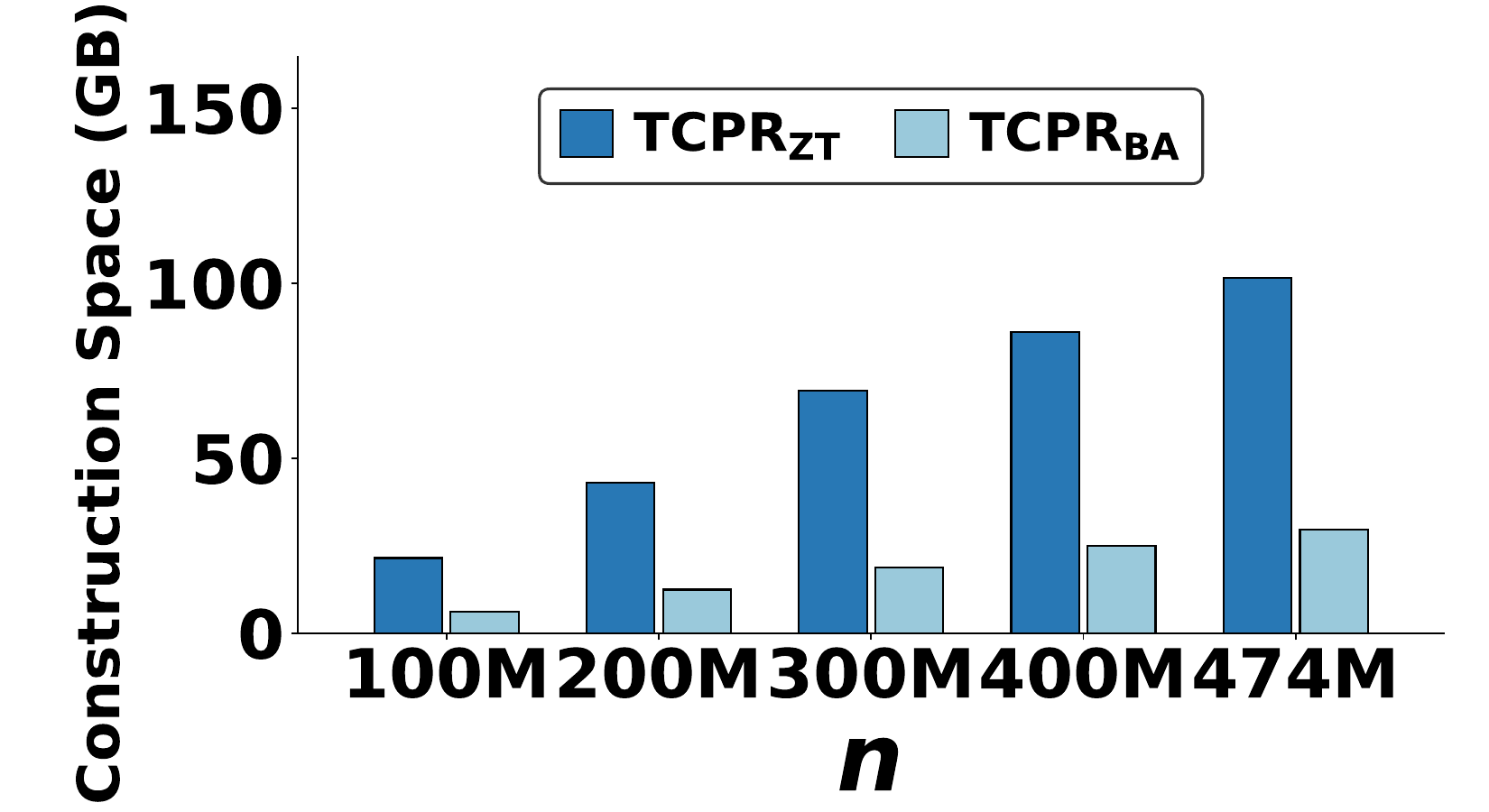}
    \caption{Constr.\ space vs. $n$}\label{fig:app:TF:n:rss:WIKI}
  \end{subfigure}\\[0pt]
  \begin{subfigure}[t]{\appfigwidth}
    \includegraphics[width=\linewidth]{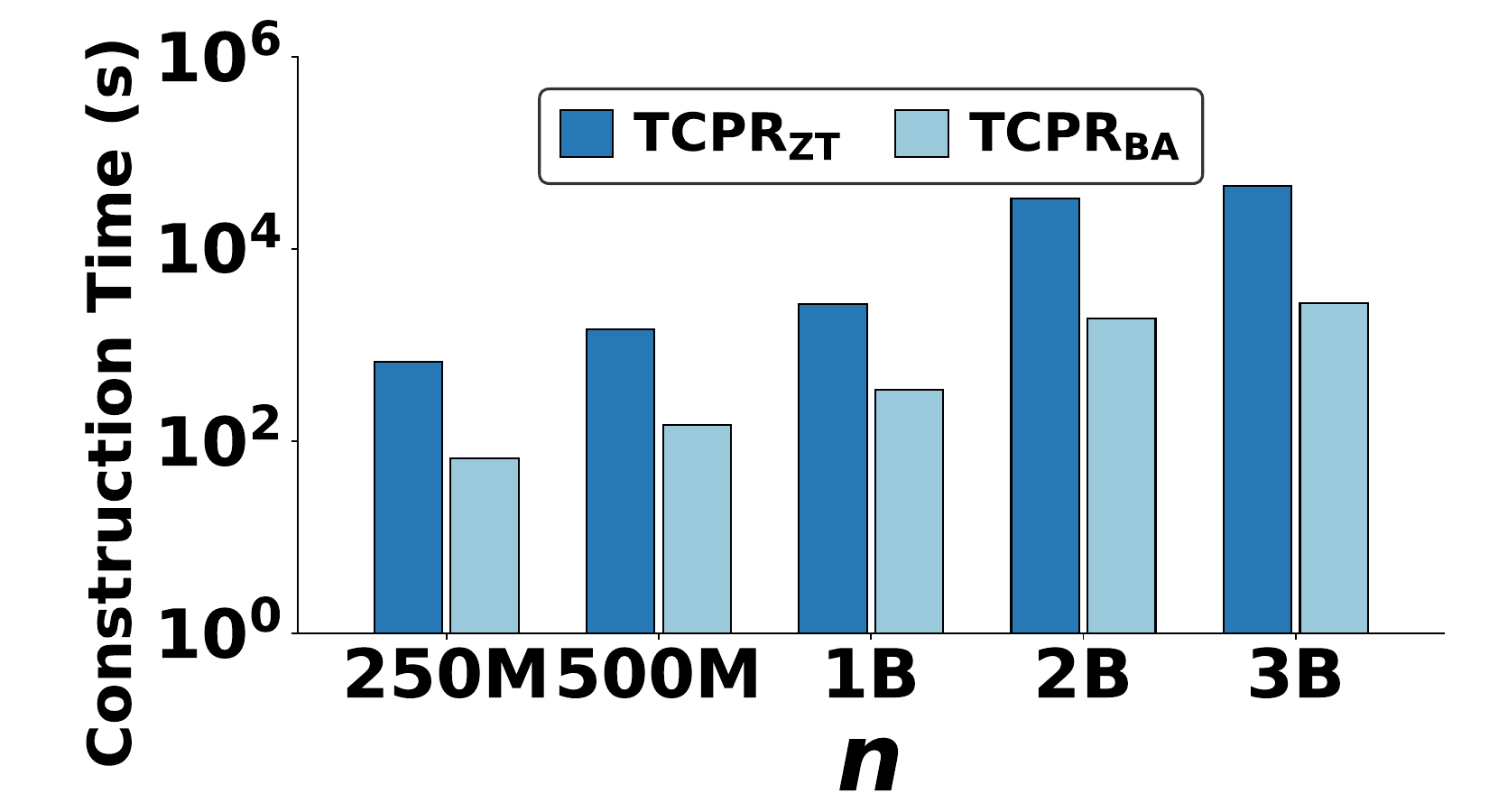}
    \caption{Constr.\ time vs. $n$}\label{fig:app:TF:n:build:BST}
  \end{subfigure}
  \begin{subfigure}[t]{\appfigwidth}
    \includegraphics[width=\linewidth]{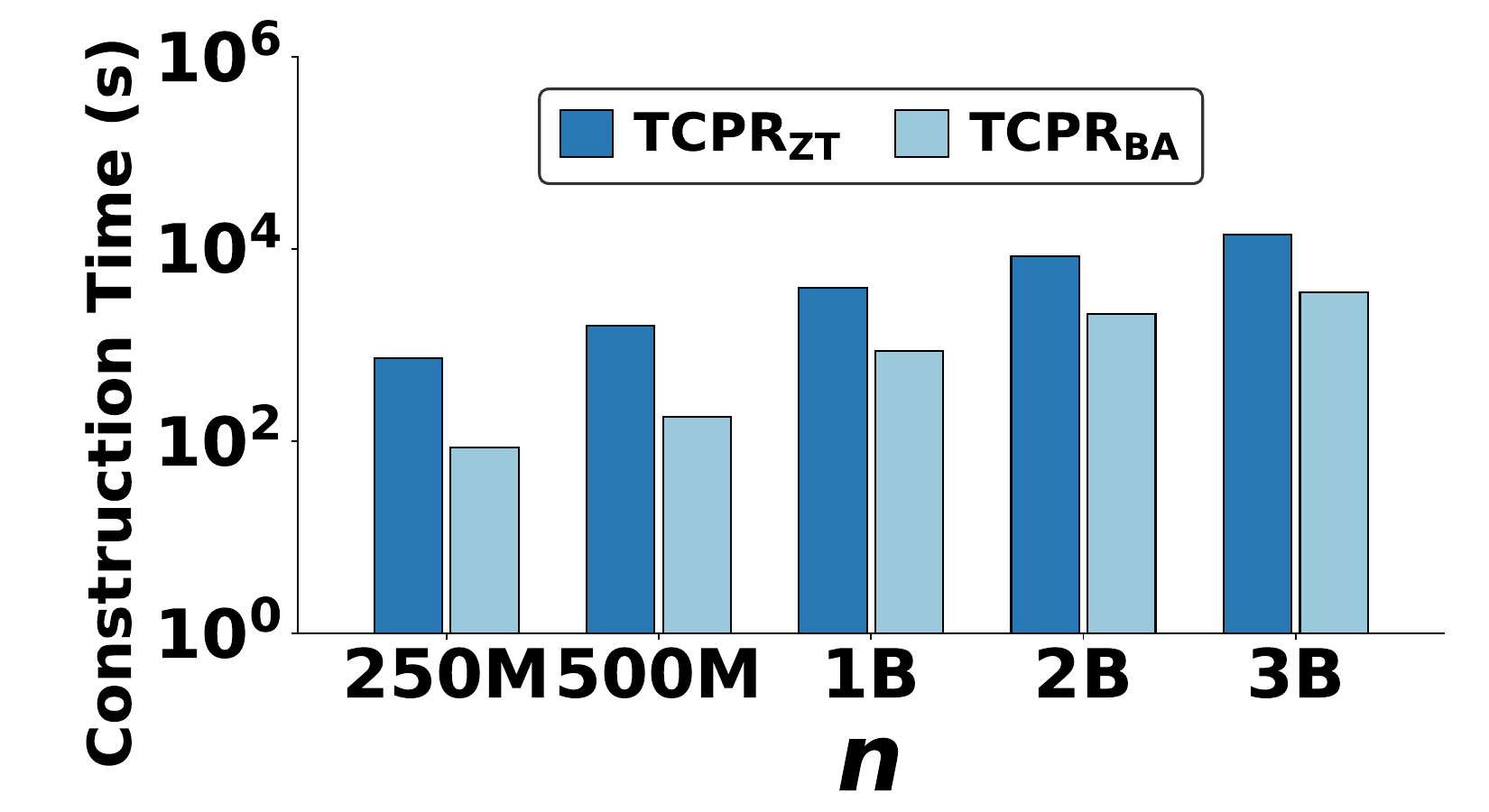}
    \caption{Constr.\ time vs. $n$}\label{fig:app:TF:n:build:SARS}
  \end{subfigure}
  \begin{subfigure}[t]{\appfigwidth}
    \includegraphics[width=\linewidth]{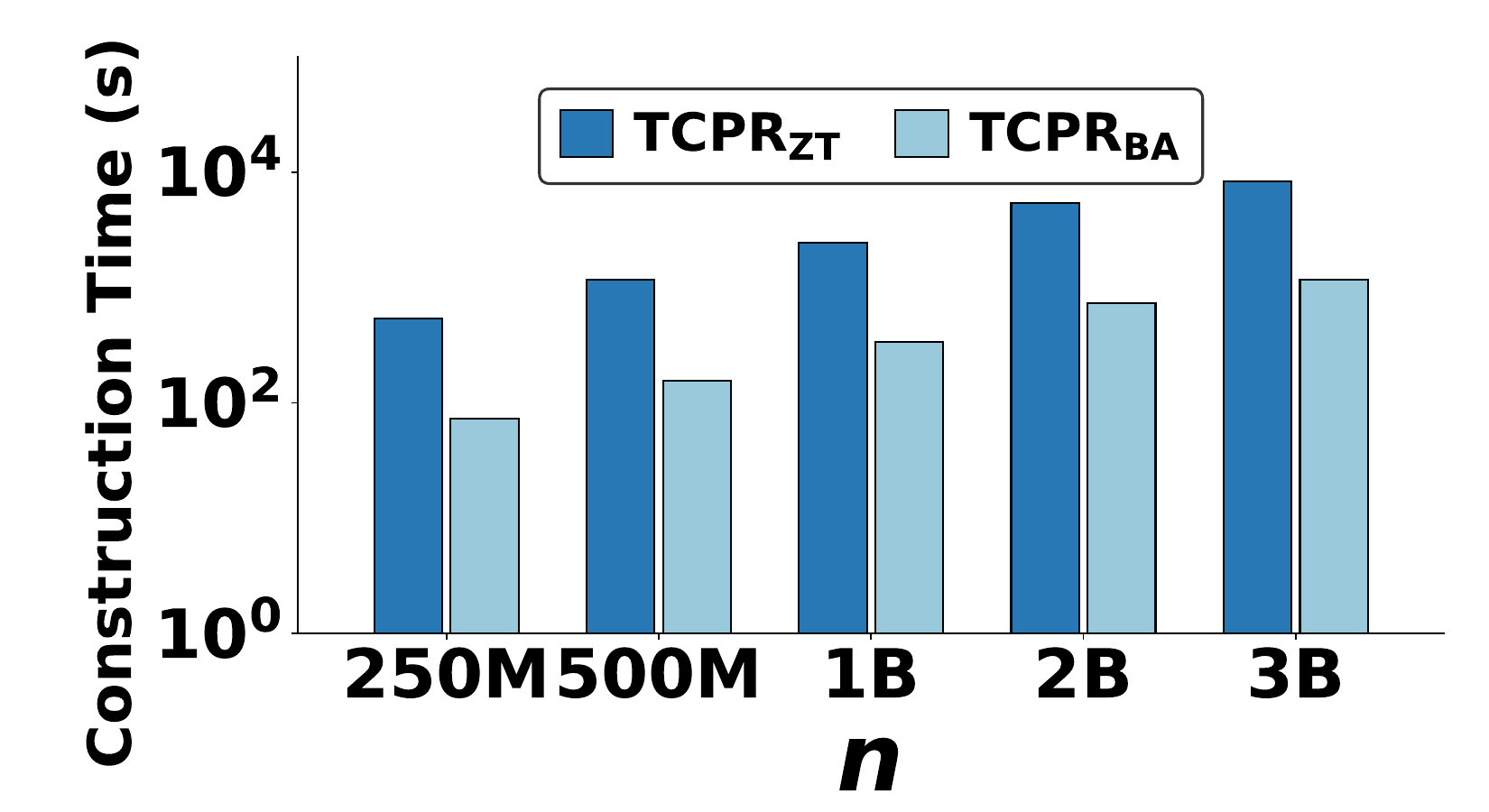}
    \caption{Constr.\ time vs. $n$}\label{fig:app:TF:n:build:SDSL}
  \end{subfigure}
  \begin{subfigure}[t]{\appfigwidth}
    \includegraphics[width=\linewidth]{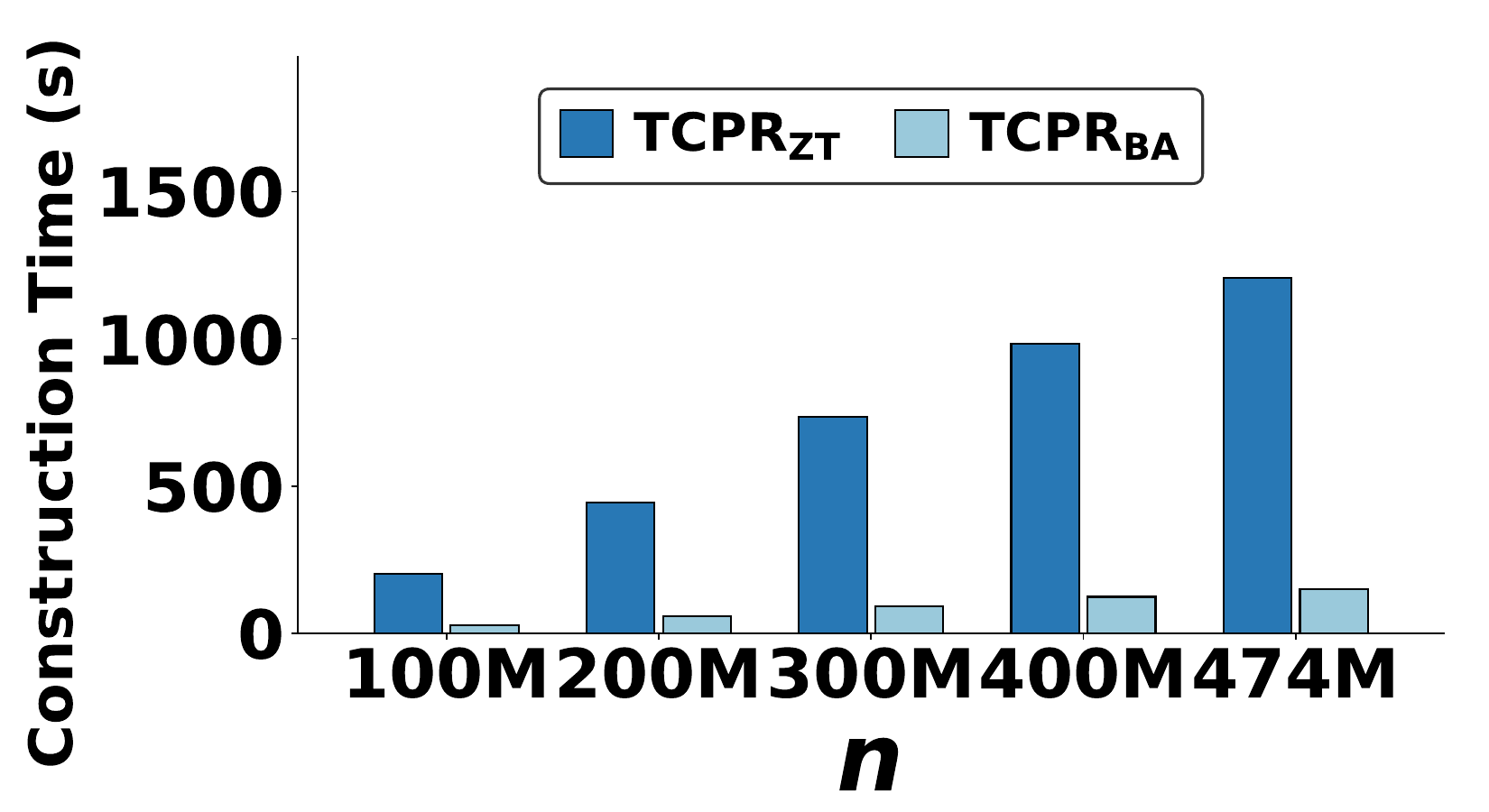}
    \caption{Constr.\ time vs. $n$}\label{fig:app:TF:n:build:WIKI}
  \end{subfigure}
  \vspace{\captionspacing}
  \vspace{+2mm}
  \caption{Index size of our \TCPR index with the \textsf{TF} scoring function vs. \TCPRBA on (a) \bst, (b) \sars, (c) \sdsl, and (d) \wiki vs. $n$; construction space of our \TCPR index with the \textsf{TF} scoring function vs. \TCPRBA on (e) \bst, (f) \sars, (g) \sdsl, and (h) \wiki vs. $n$; construction time of our \TCPR index with the \textsf{TF} scoring function vs. \TCPRBA on (i) \bst, (j) \sars, (k) \sdsl, and (l) \wiki vs. $n$.}\label{fig:app:TF:cost}
\end{figure}

\begin{figure}[ht]
  \centering
  \begin{subfigure}[t]{\appfigwidth}
    \includegraphics[width=\linewidth]{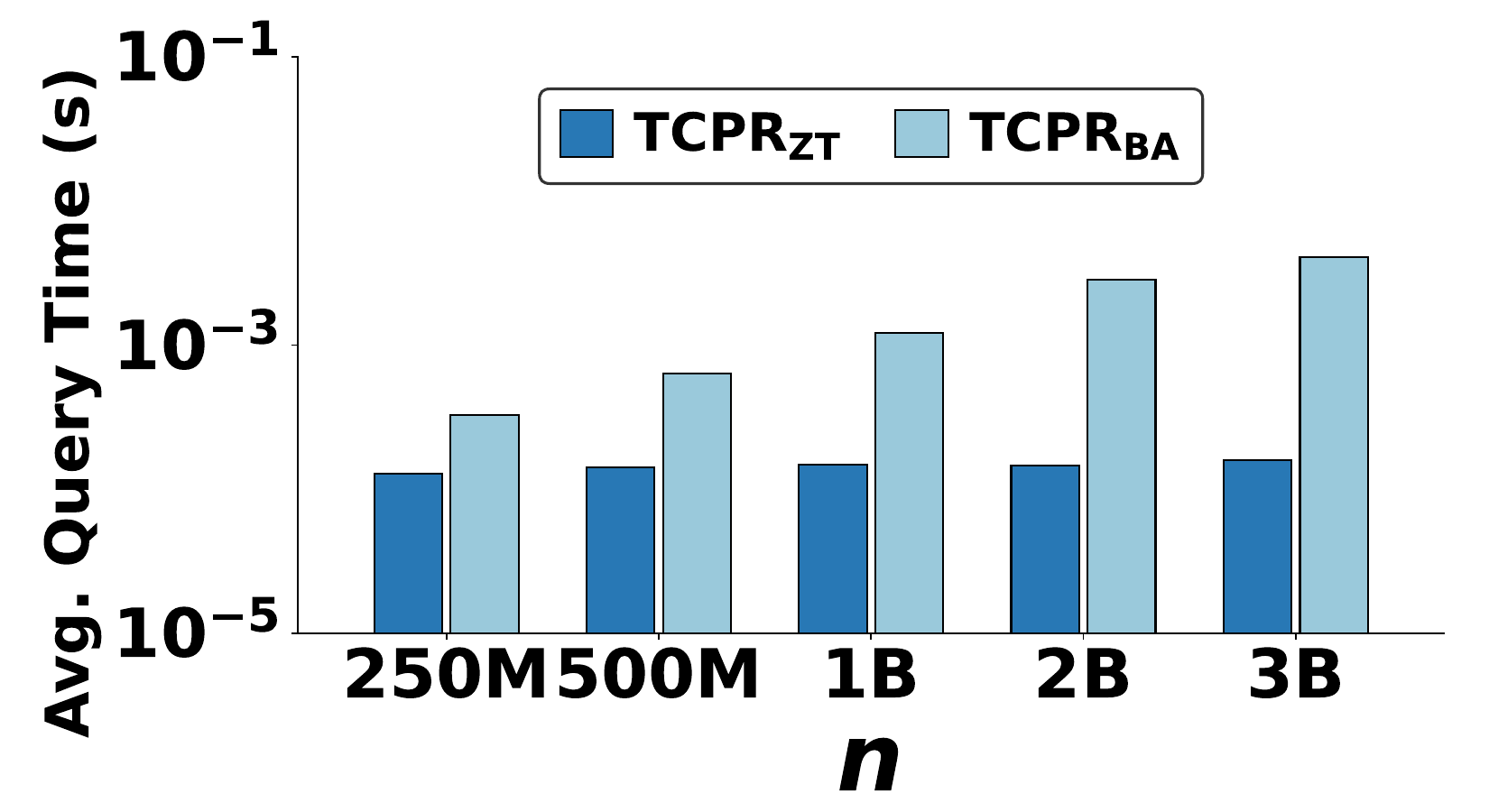}
    \caption{Query time vs. $n$}\label{fig:app:SP:n:query:BST}
  \end{subfigure}
  \begin{subfigure}[t]{\appfigwidth}
    \includegraphics[width=\linewidth]{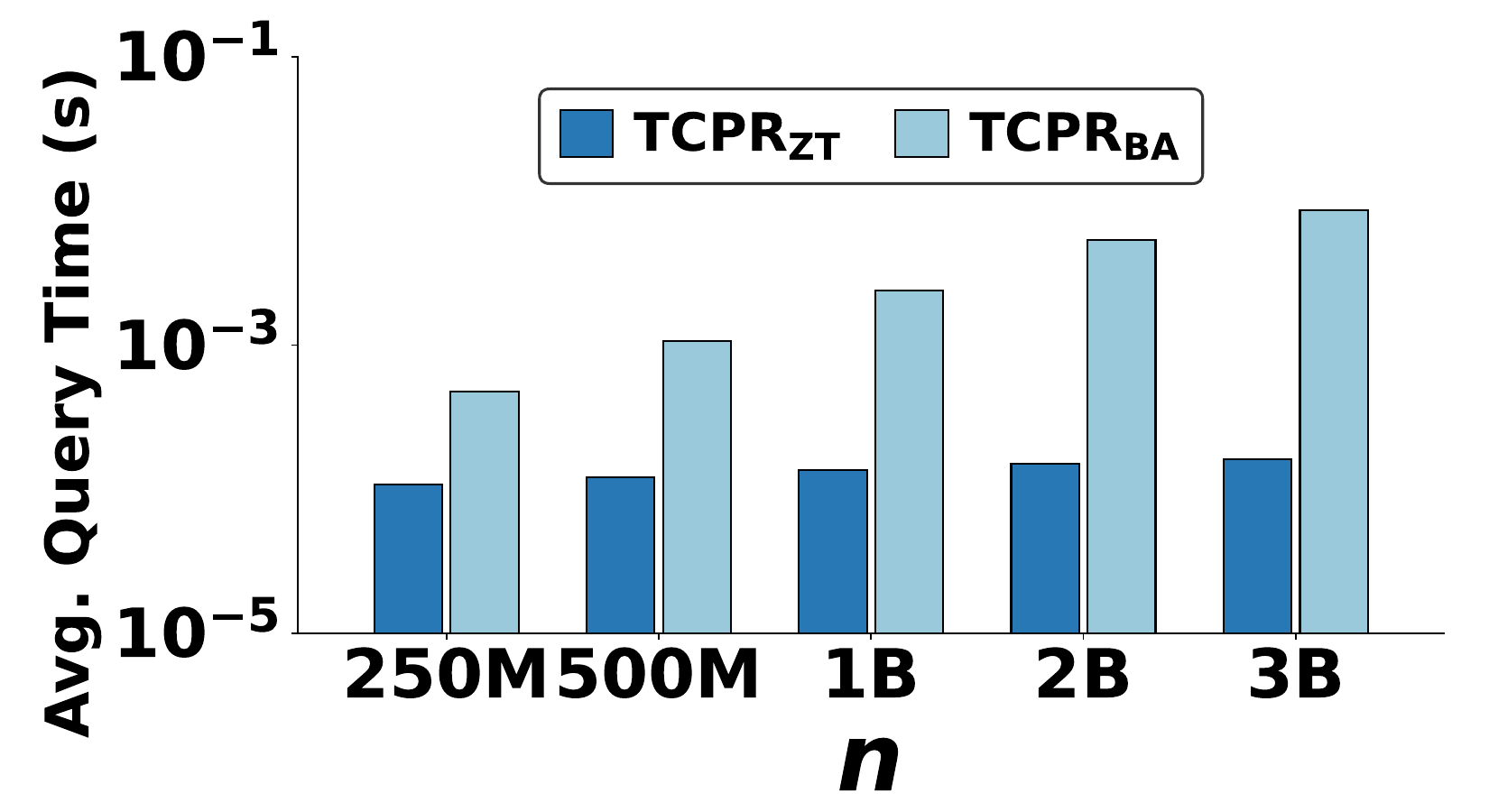}
    \caption{Query time vs. $n$}\label{fig:app:SP:n:query:SARS}
  \end{subfigure}
  \begin{subfigure}[t]{\appfigwidth}
    \includegraphics[width=\linewidth]{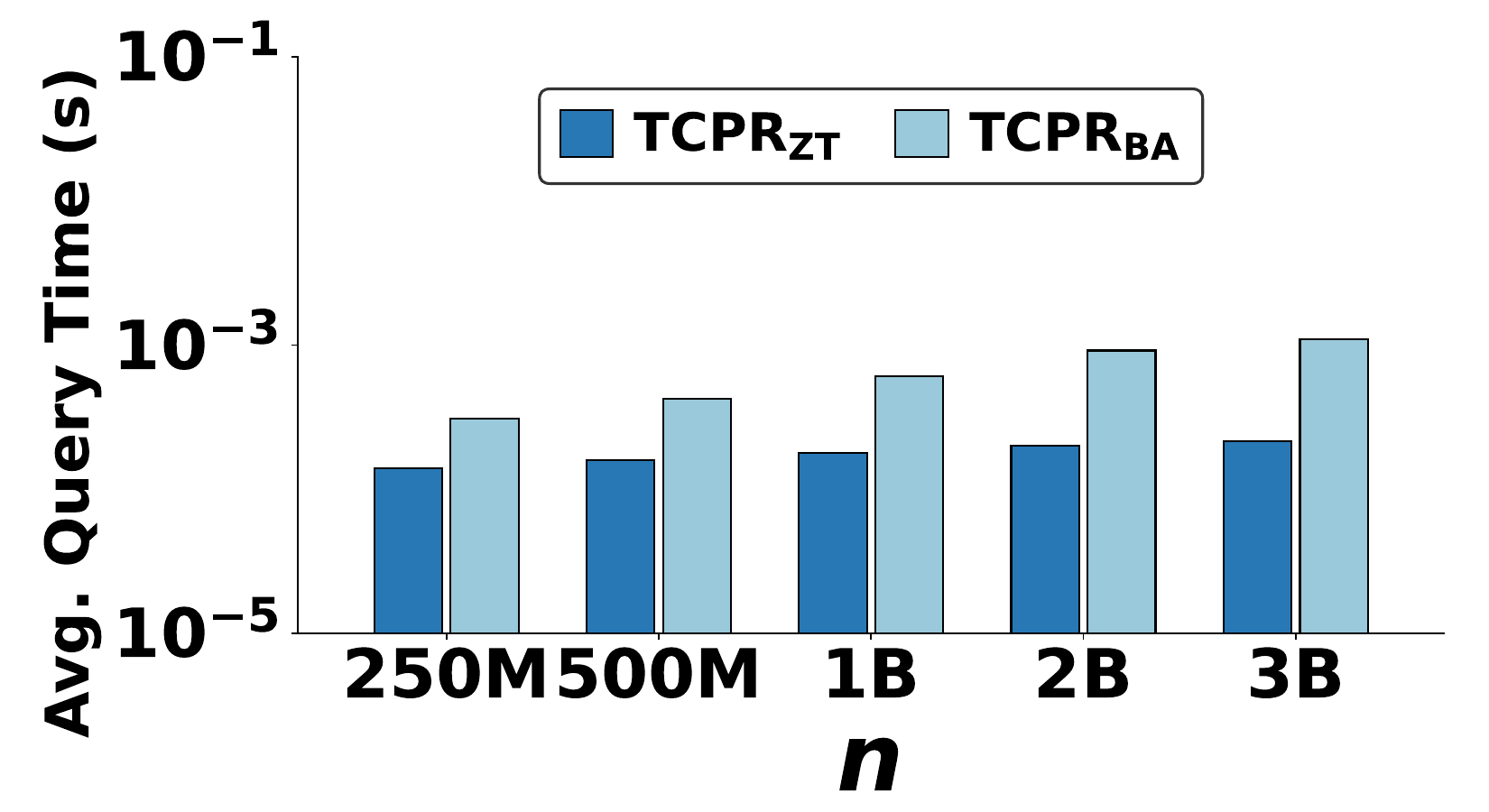}
    \caption{Query time vs. $n$}\label{fig:app:SP:n:query:SDSL}
  \end{subfigure}
  \begin{subfigure}[t]{\appfigwidth}
    \includegraphics[width=\linewidth]{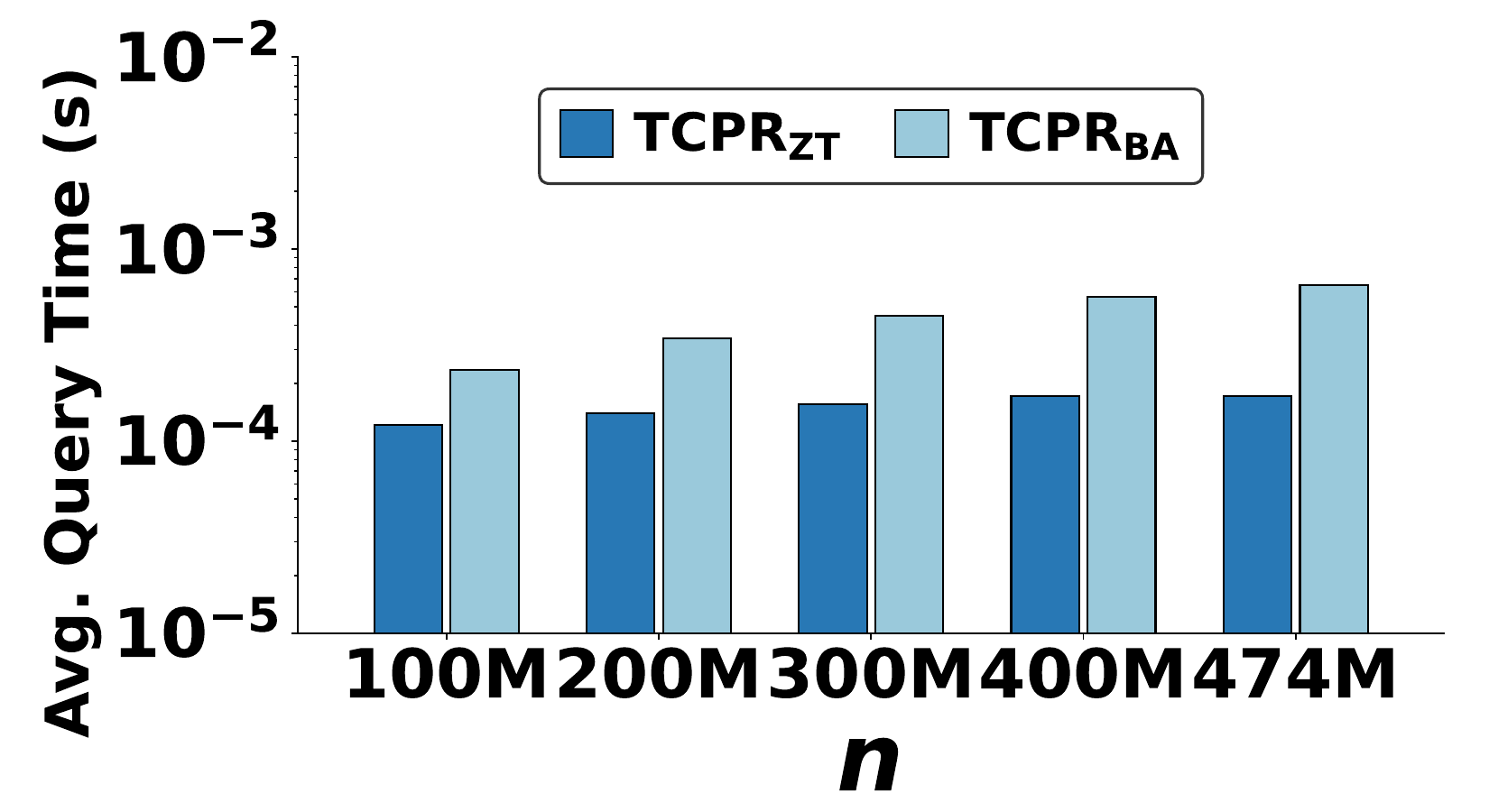}
    \caption{Query time vs. $n$}\label{fig:app:SP:n:query:WIKI}
  \end{subfigure}\\[0pt]
  \begin{subfigure}[t]{\appfigwidth}
    \includegraphics[width=\linewidth]{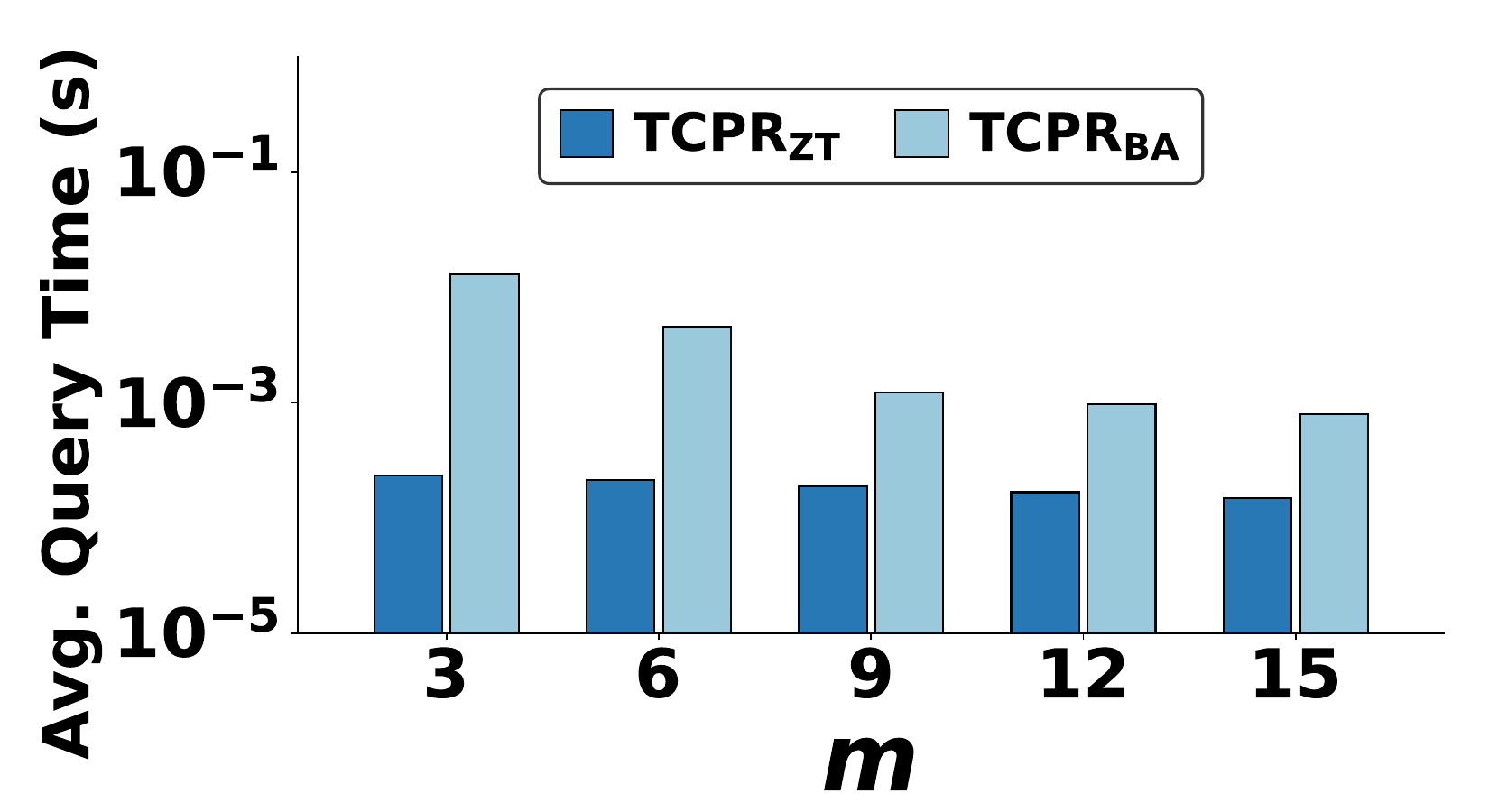}
    \caption{Query time vs. $m$}\label{fig:app:SP:m:query:BST}
  \end{subfigure}
  \begin{subfigure}[t]{\appfigwidth}
    \includegraphics[width=\linewidth]{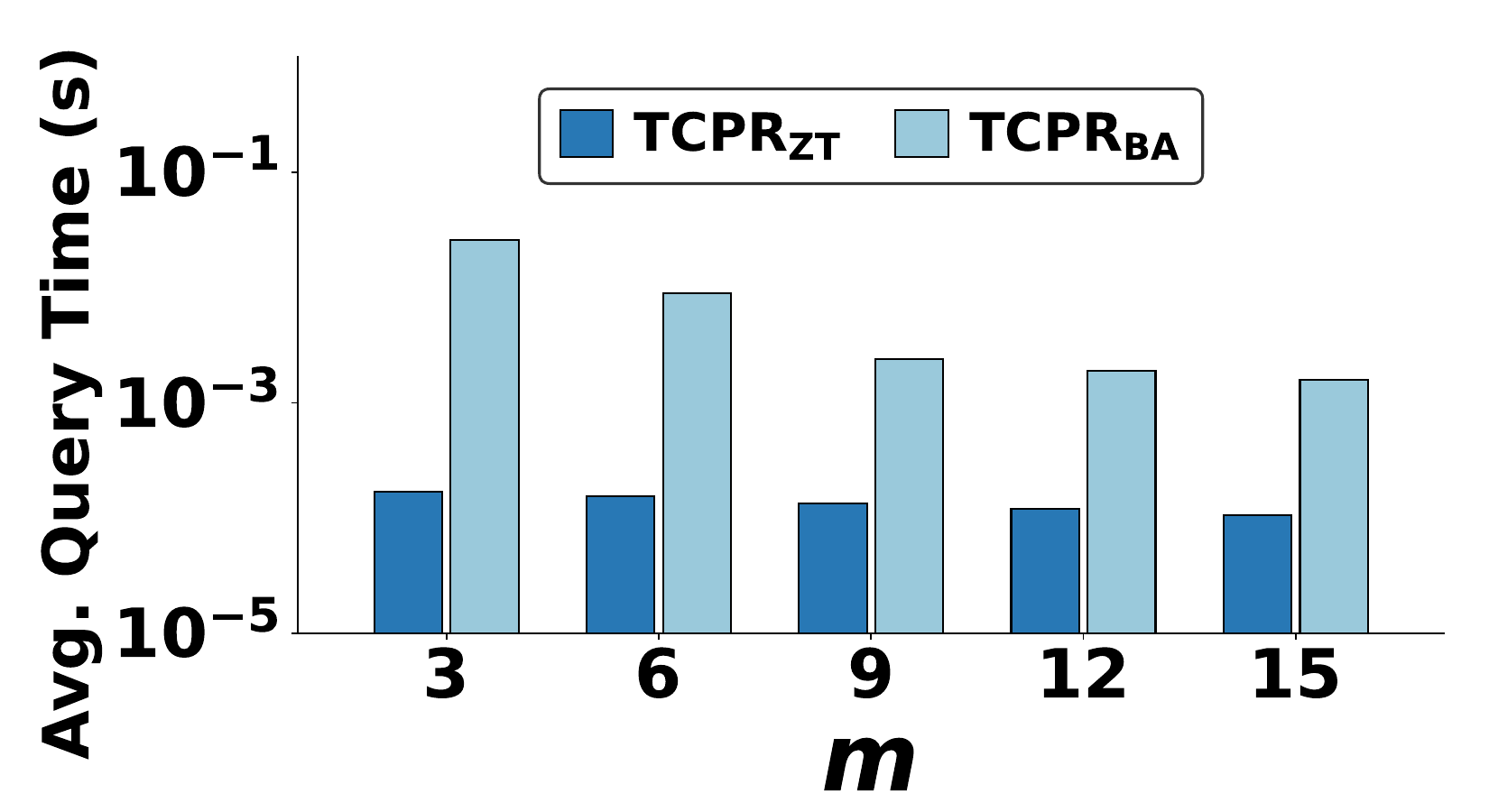}
    \caption{Query time vs. $m$}\label{fig:app:SP:m:query:SARS}
  \end{subfigure}
  \begin{subfigure}[t]{\appfigwidth}
    \includegraphics[width=\linewidth]{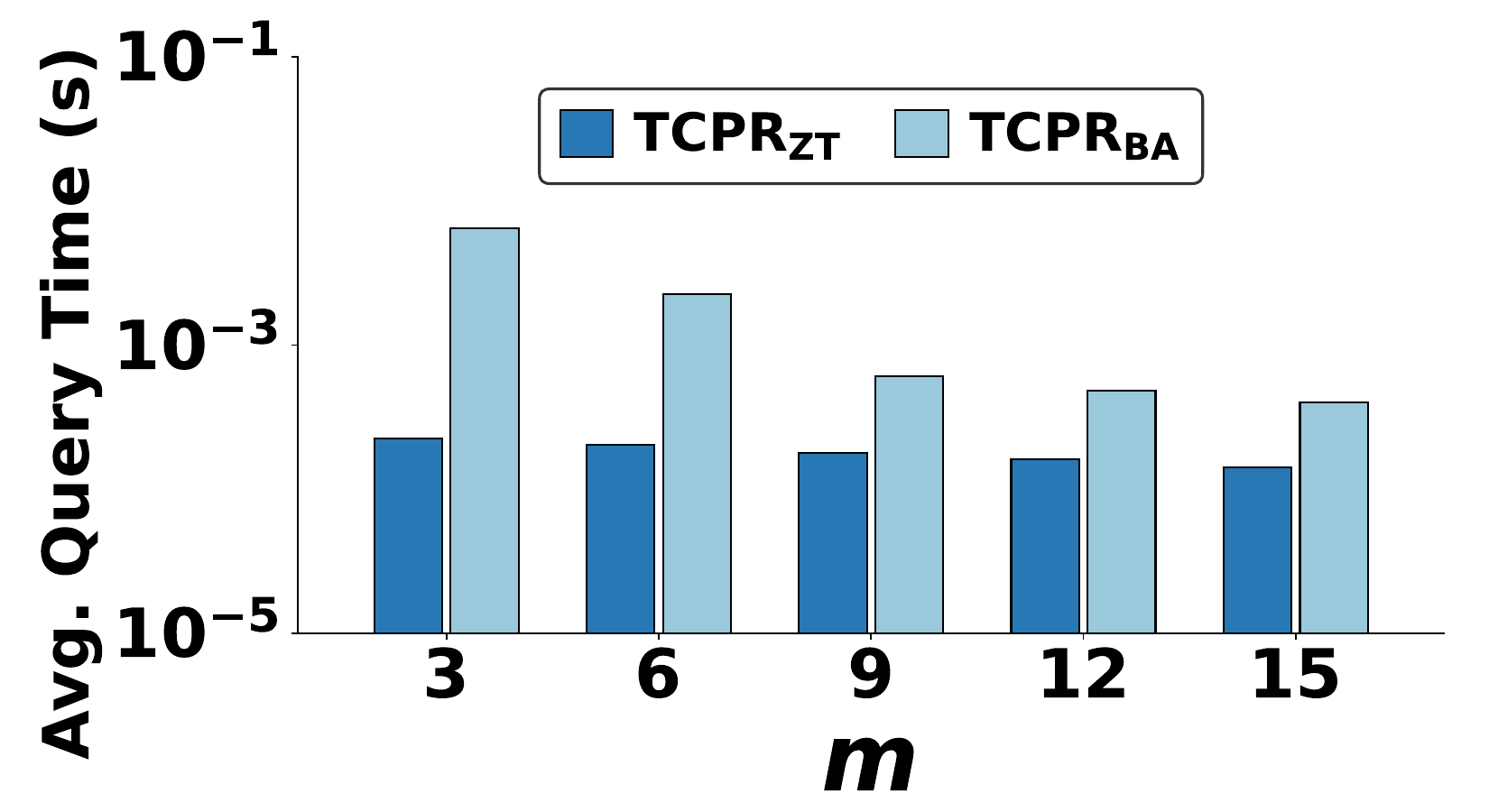}
    \caption{Query time vs. $m$}\label{fig:app:SP:m:query:SDSL}
  \end{subfigure}
  \begin{subfigure}[t]{\appfigwidth}
    \includegraphics[width=\linewidth]{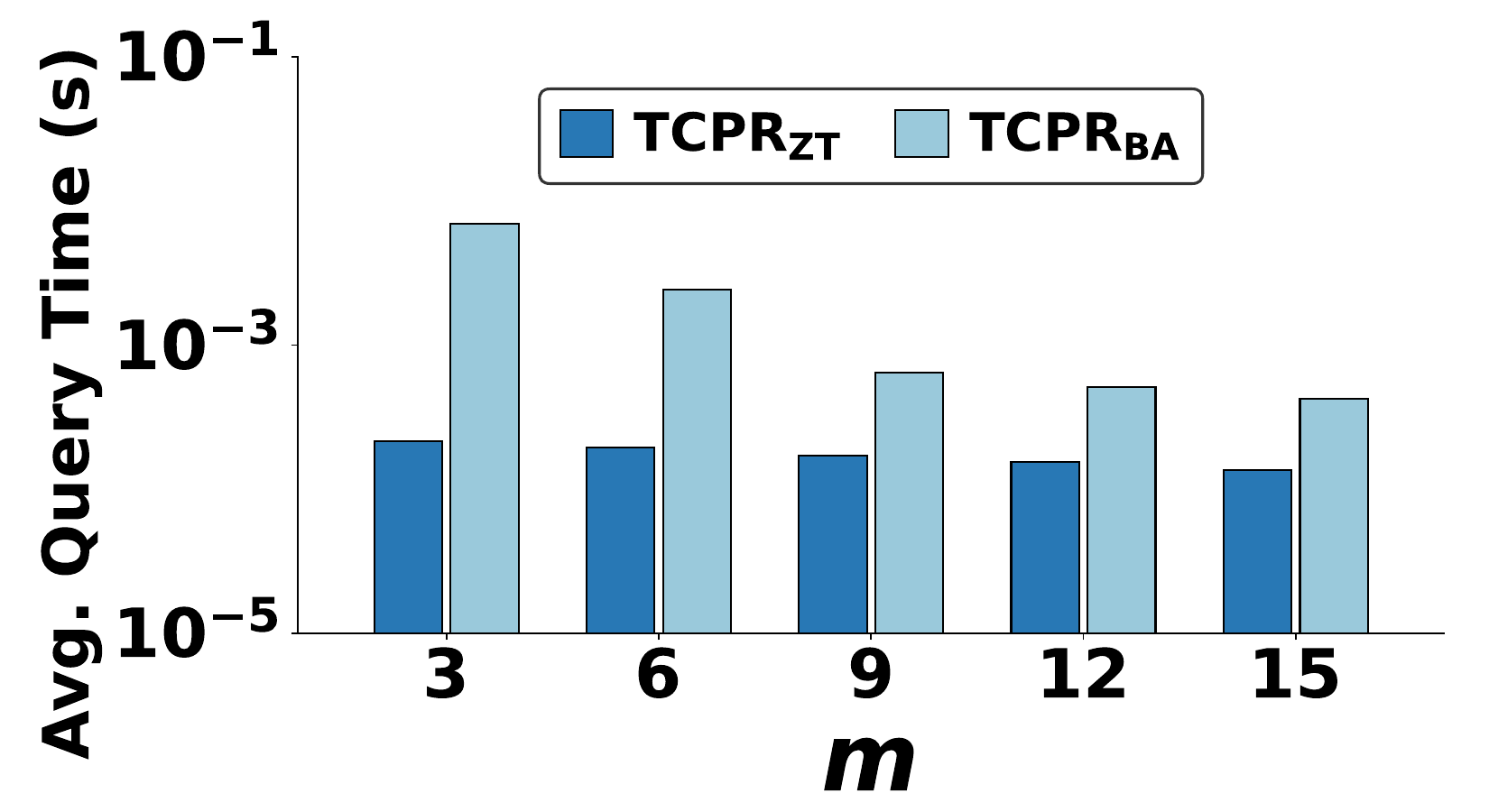}
    \caption{Query time vs. $m$}\label{fig:app:SP:m:query:WIKI}
  \end{subfigure}\\[0pt]
  \begin{subfigure}[t]{\appfigwidth}
    \includegraphics[width=\linewidth]{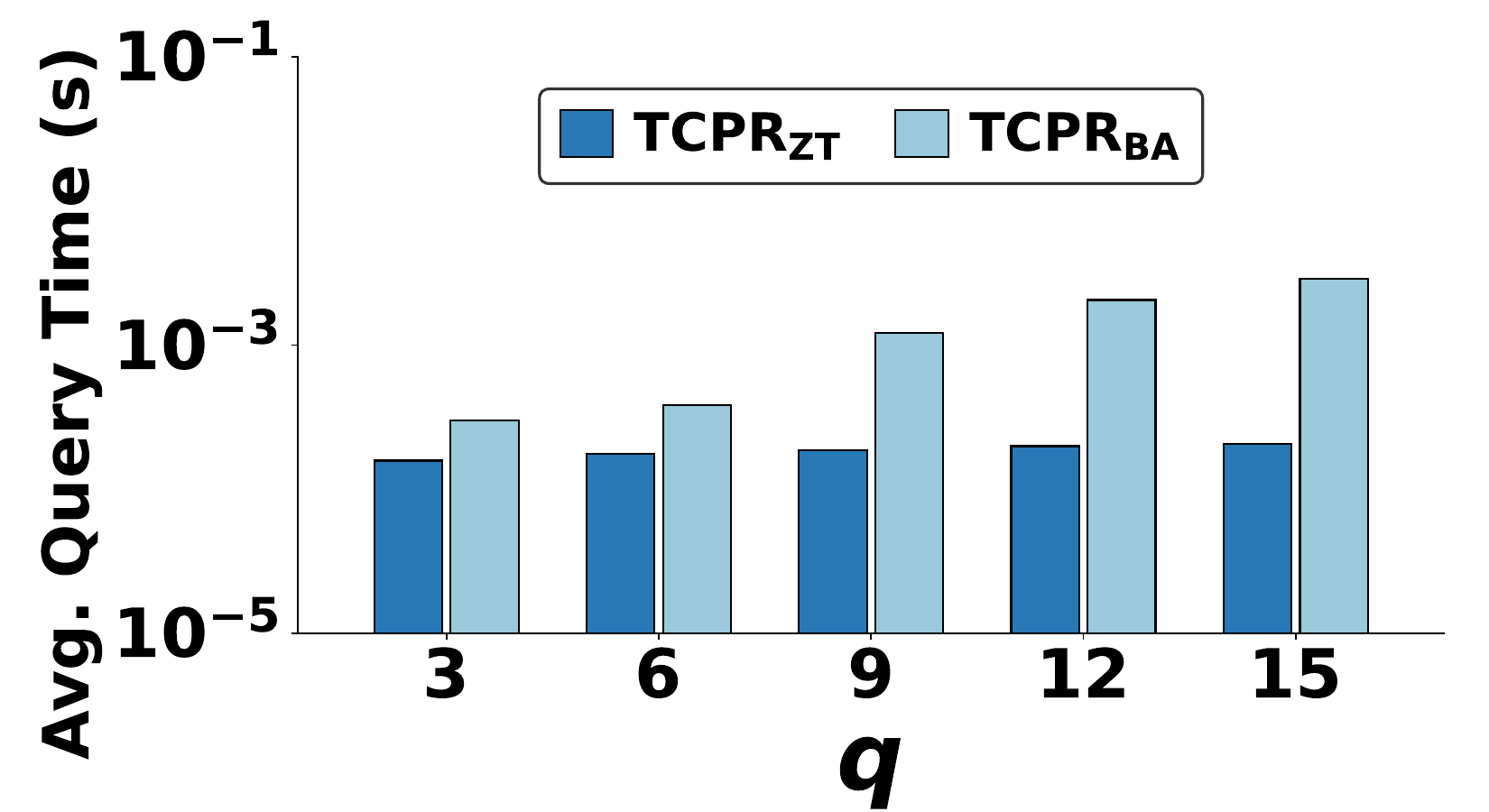}
    \caption{Query time vs. $q$}\label{fig:app:SP:q:query:BST}
  \end{subfigure}
  \begin{subfigure}[t]{\appfigwidth}
    \includegraphics[width=\linewidth]{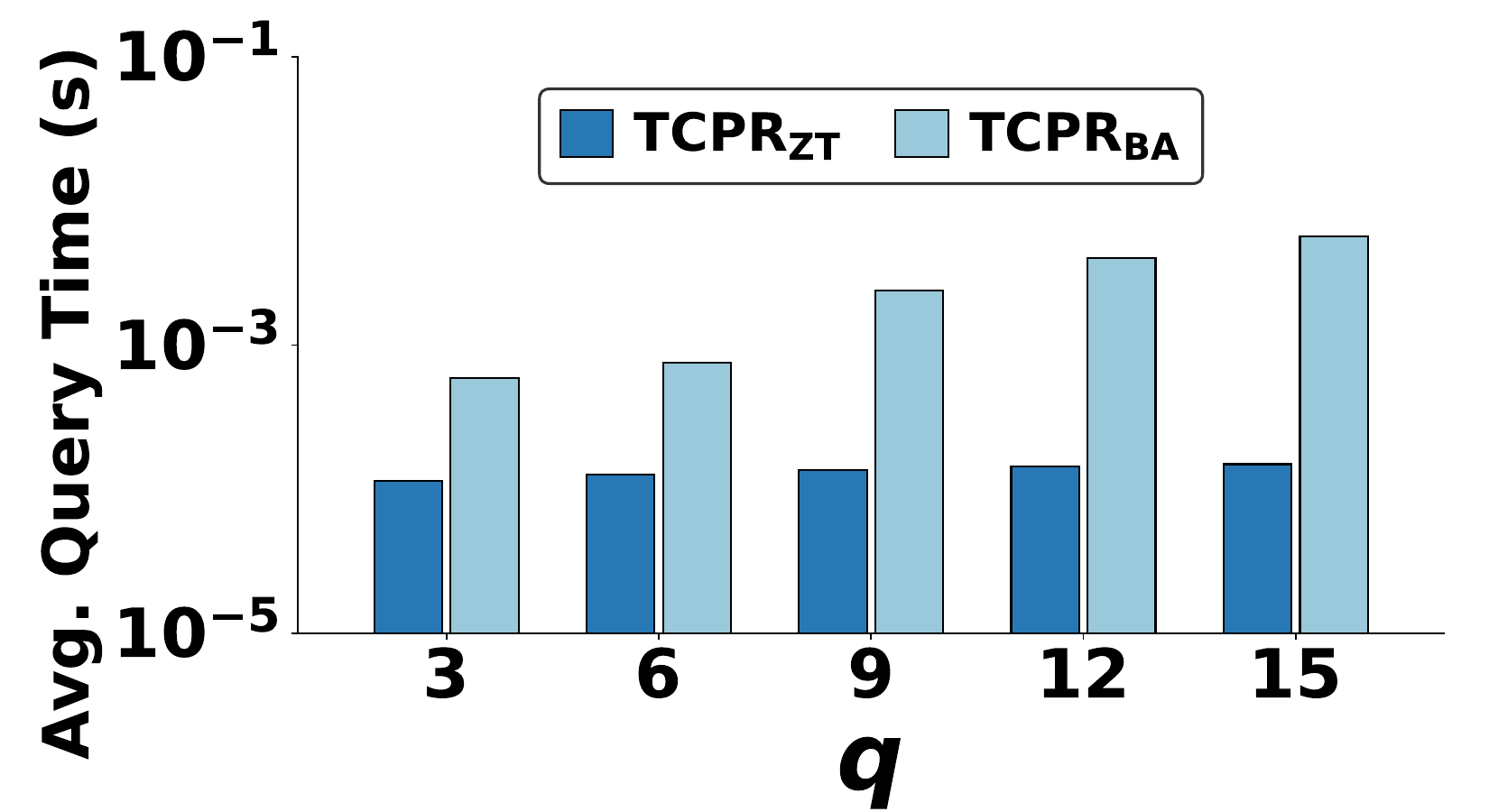}
    \caption{Query time vs. $q$}\label{fig:app:SP:q:query:SARS}
  \end{subfigure}
  \begin{subfigure}[t]{\appfigwidth}
    \includegraphics[width=\linewidth]{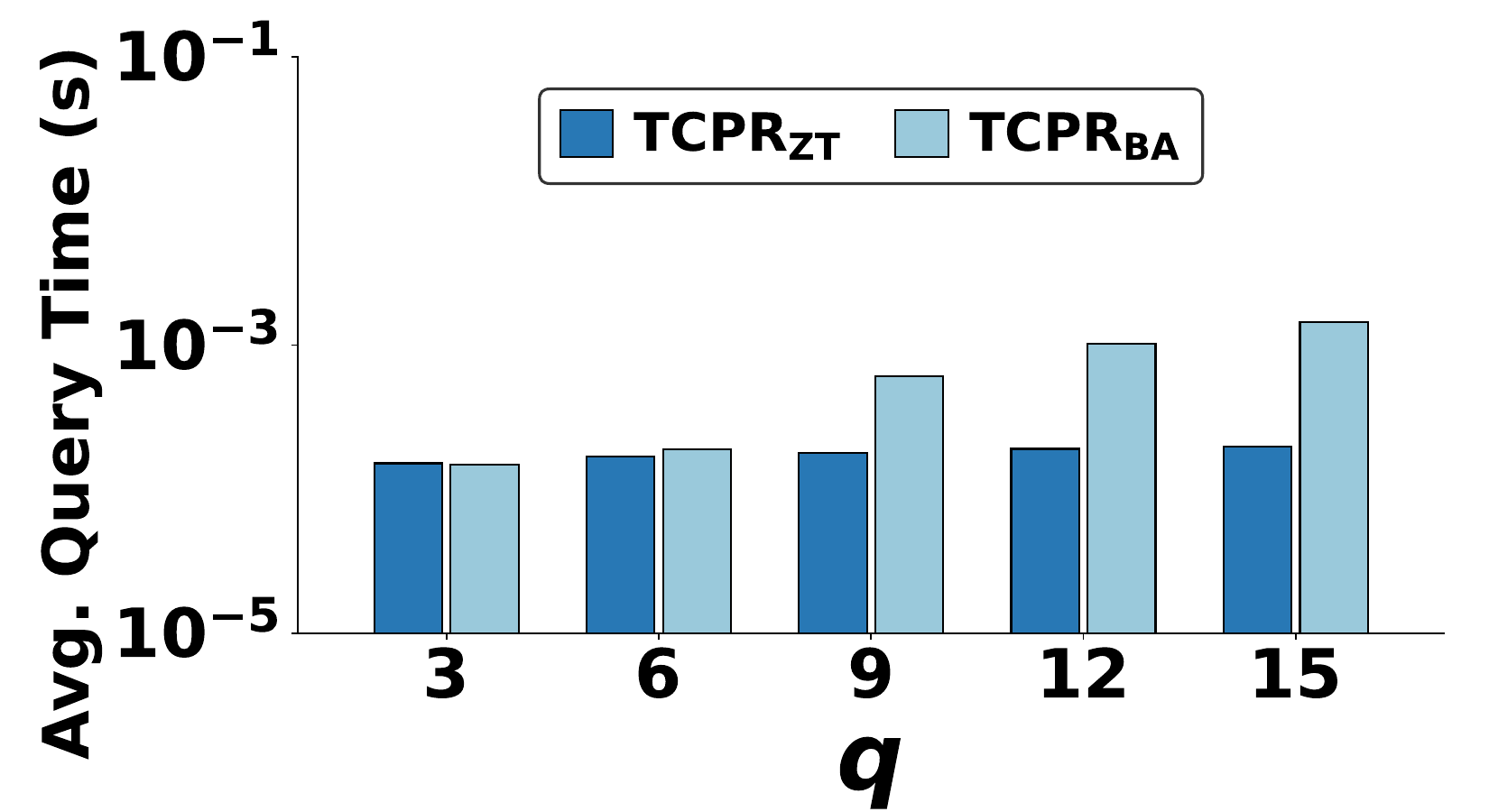}
    \caption{Query time vs. $q$}\label{fig:app:SP:q:query:SDSL}
  \end{subfigure}
  \begin{subfigure}[t]{\appfigwidth}
    \includegraphics[width=\linewidth]{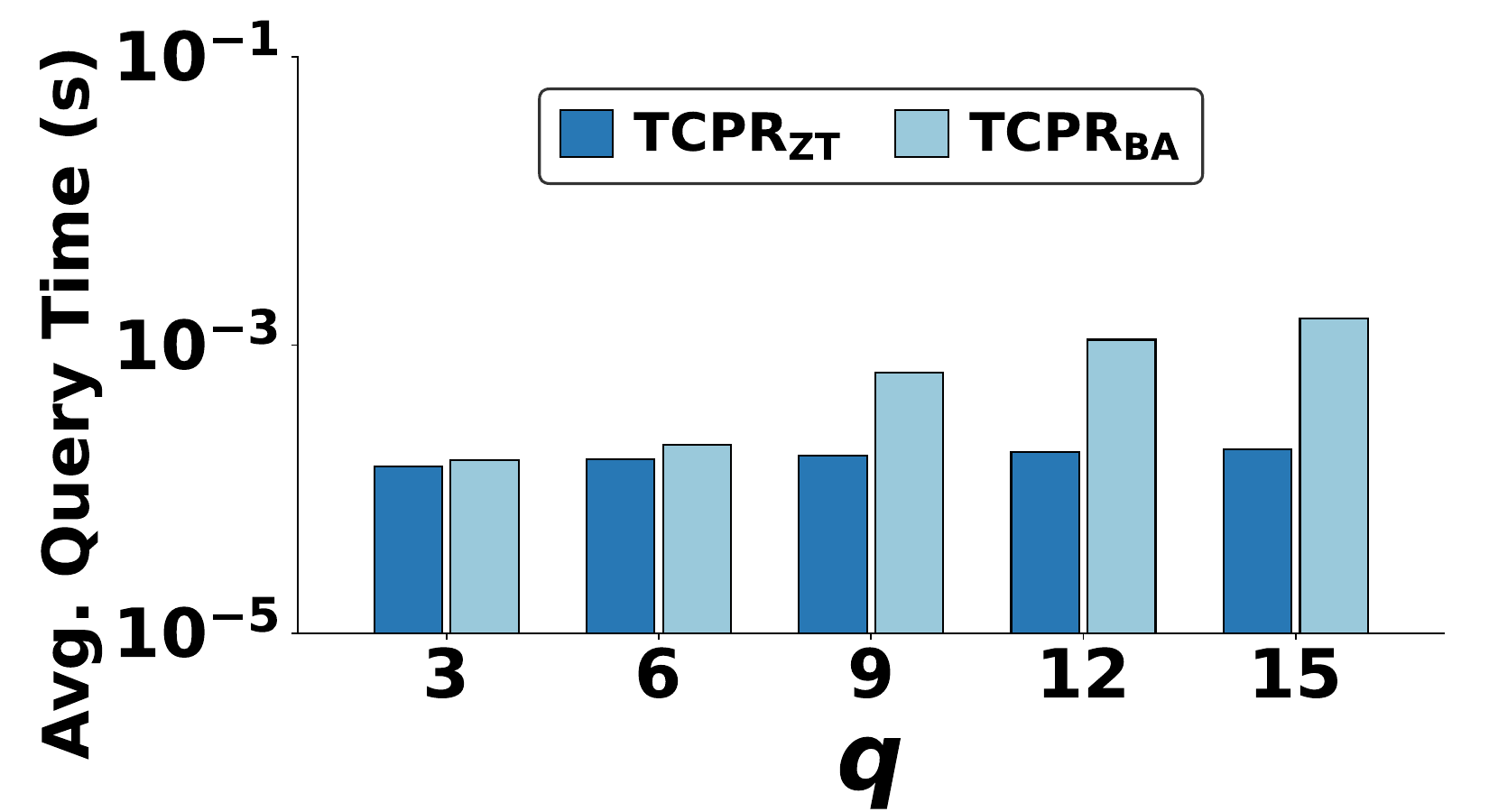}
    \caption{Query time vs. $q$}\label{fig:app:SP:q:query:WIKI}
  \end{subfigure}\\[0pt]
  \begin{subfigure}[t]{\appfigwidth}
    \includegraphics[width=\linewidth]{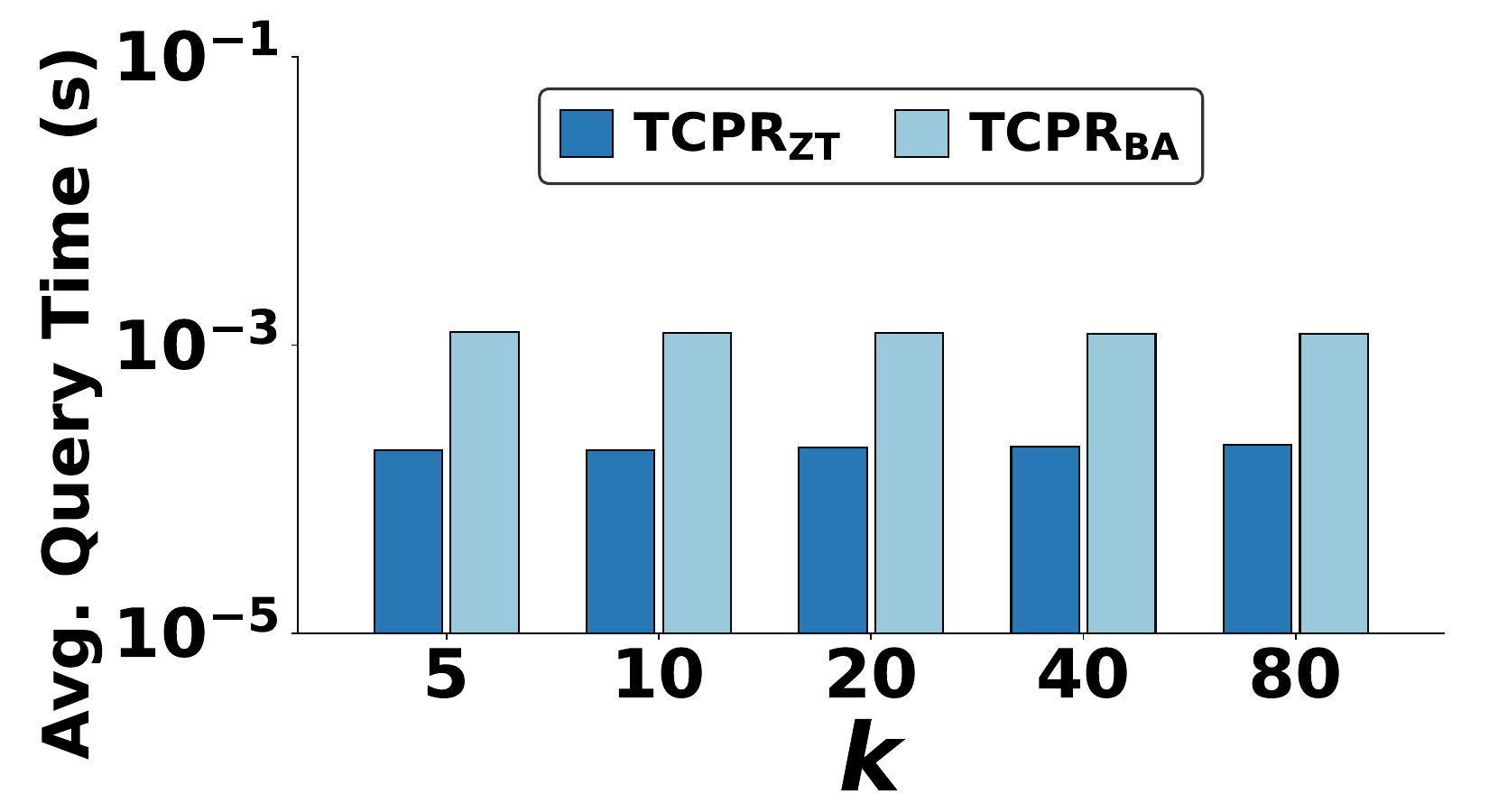}
    \caption{Query time vs. $k$}\label{fig:app:SP:k:query:BST}
  \end{subfigure}
  \begin{subfigure}[t]{\appfigwidth}
    \includegraphics[width=\linewidth]{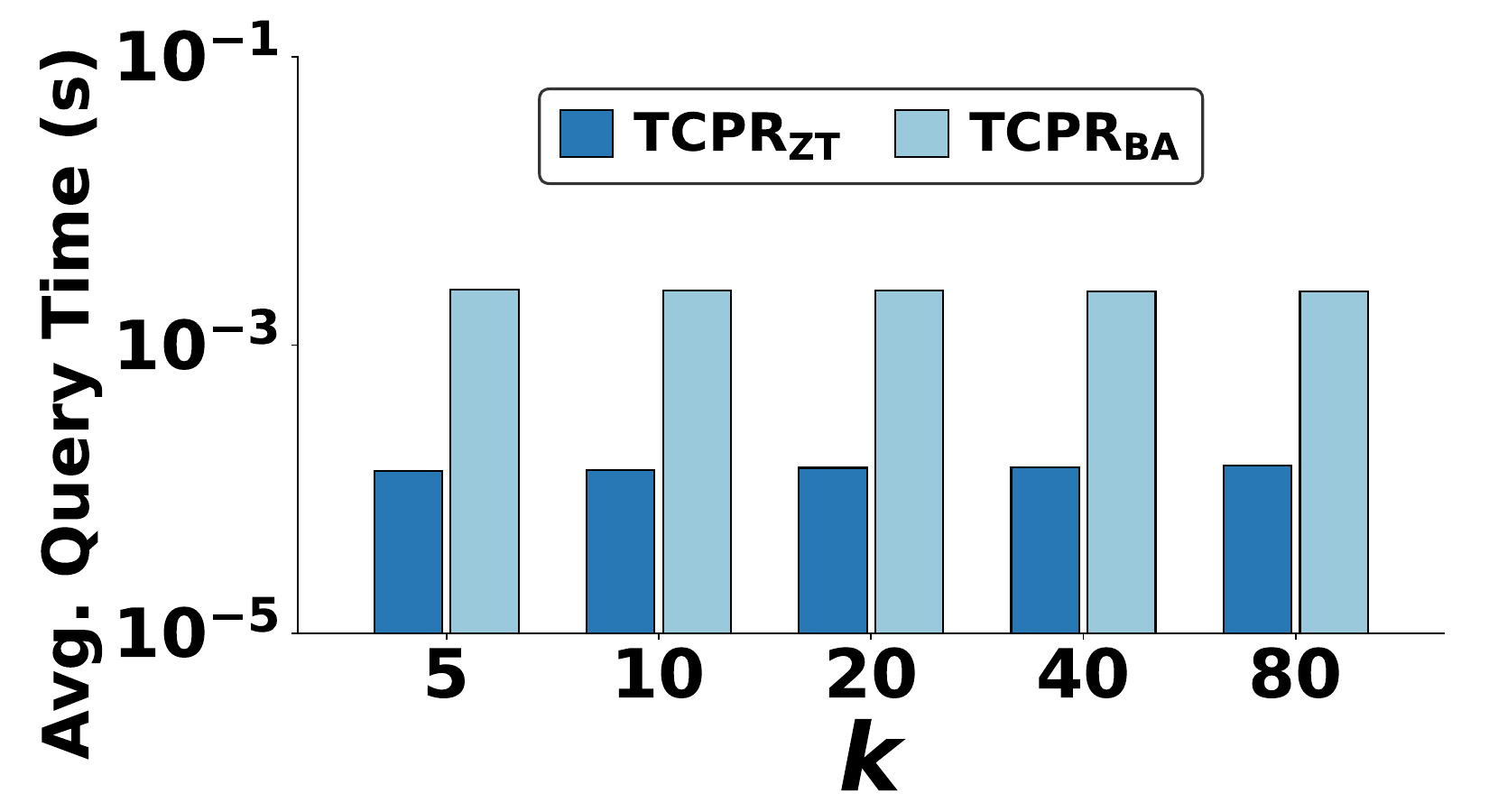}
    \caption{Query time vs. $k$}\label{fig:app:SP:k:query:SARS}
  \end{subfigure}
  \begin{subfigure}[t]{\appfigwidth}
    \includegraphics[width=\linewidth]{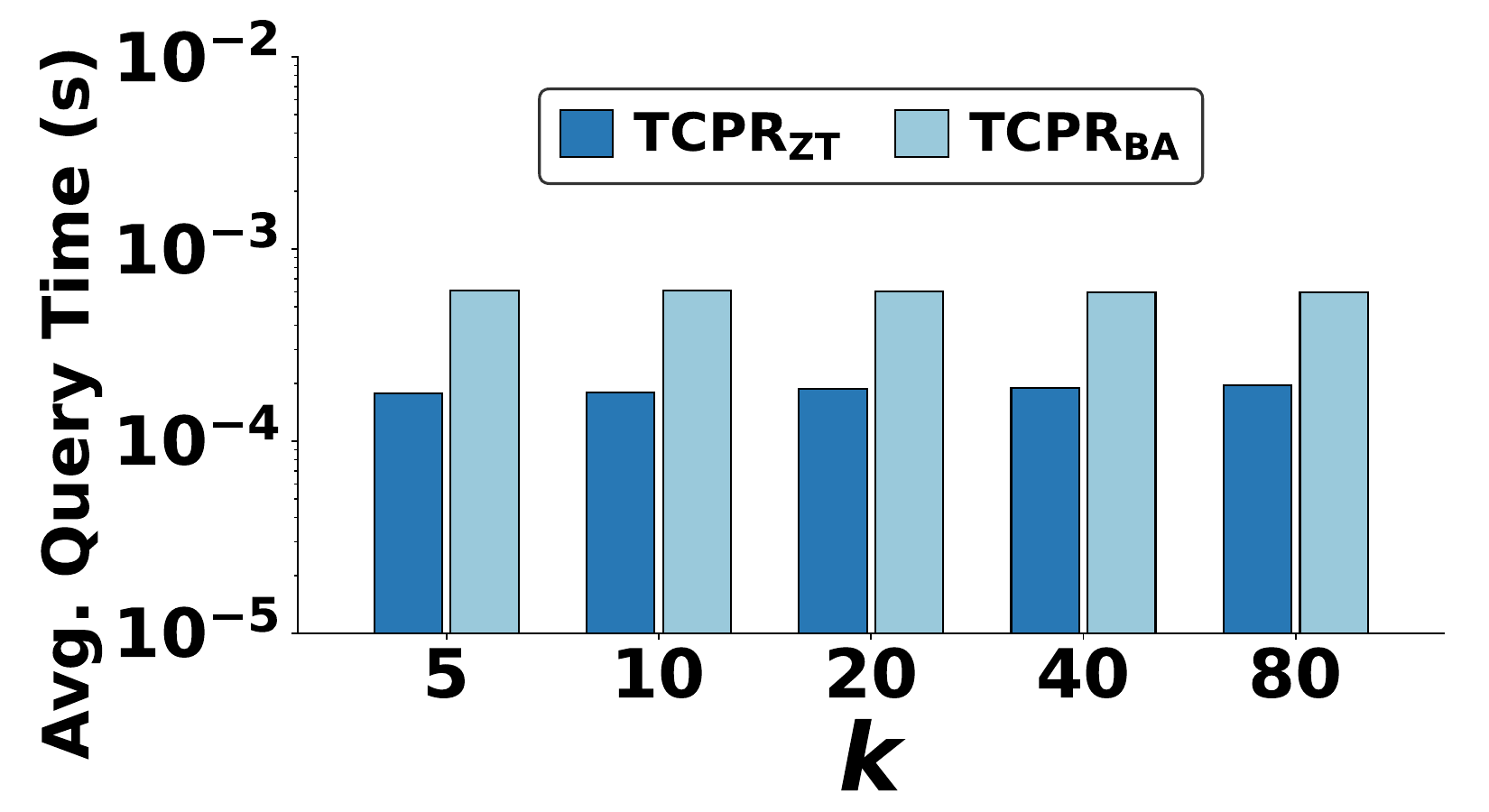}
    \caption{Query time vs. $k$}\label{fig:app:SP:k:query:SDSL}
  \end{subfigure}
  \begin{subfigure}[t]{\appfigwidth}
    \includegraphics[width=\linewidth]{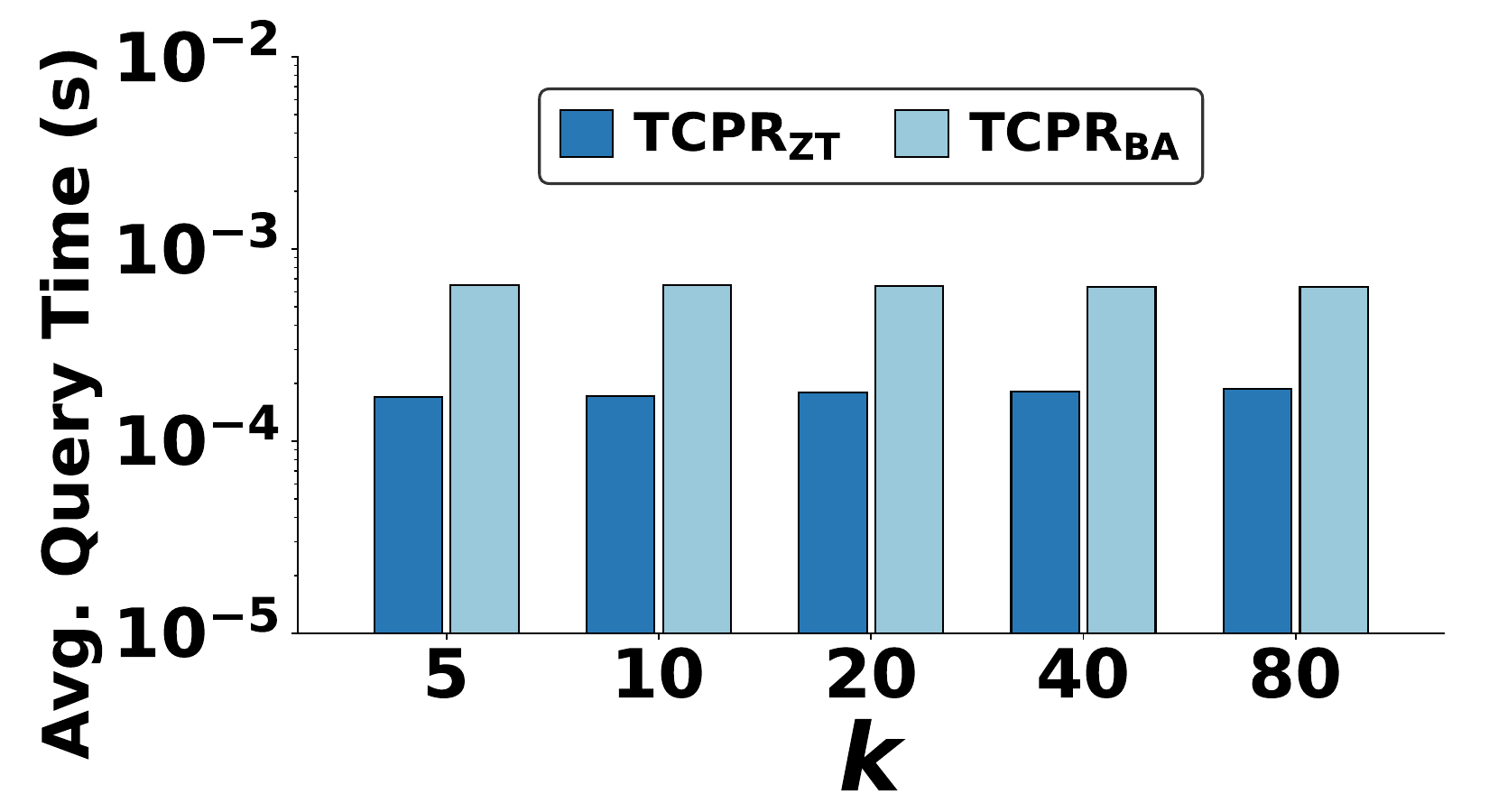}
    \caption{Query time vs. $k$}\label{fig:app:SP:k:query:WIKI}
  \end{subfigure}
  \vspace{\captionspacing}
  \vspace{+2mm}
  \caption{Query time of our \TCPR index with the \textsf{SP} scoring function vs. \TCPRBA on (a) \bst, (b) \sars, (c) \sdsl, and (d) \wiki vs. $n$; on (e) \bst, (f) \sars, (g) \sdsl, and (h) \wiki vs. $m$; on (i) \bst, (j) \sars, (k) \sdsl, and (l) \wiki vs. $q$; on (m) \bst, (n) \sars, (o) \sdsl, and (p) \wiki vs. $k$.}\label{fig:app:SP:query}
\end{figure}

\begin{figure}[ht]
  \centering
  \begin{subfigure}[t]{\appfigwidth}
    \includegraphics[width=\linewidth]{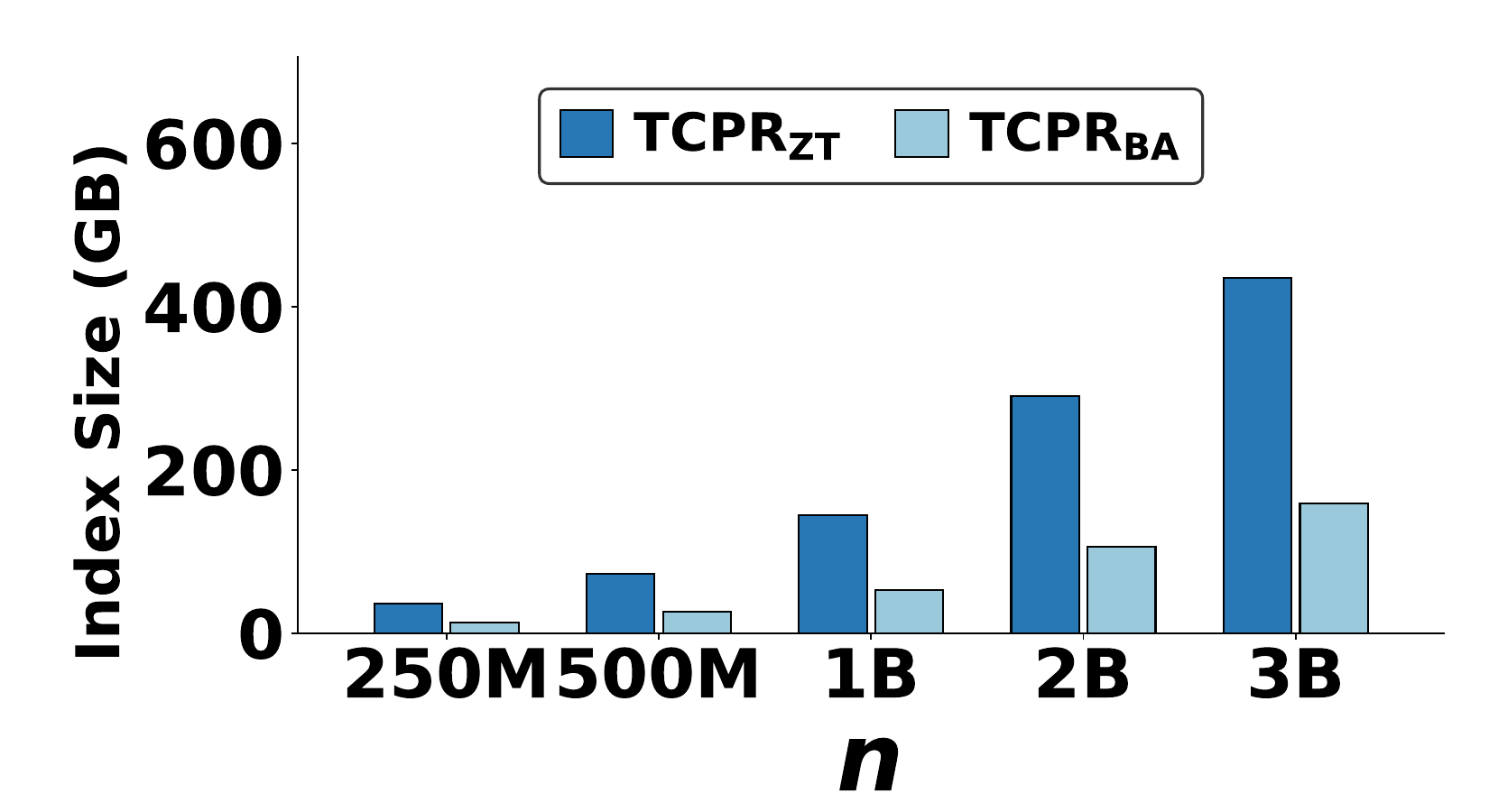}
    \caption{Index size vs. $n$}\label{fig:app:SP:n:index:BST}
  \end{subfigure}
  \begin{subfigure}[t]{\appfigwidth}
    \includegraphics[width=\linewidth]{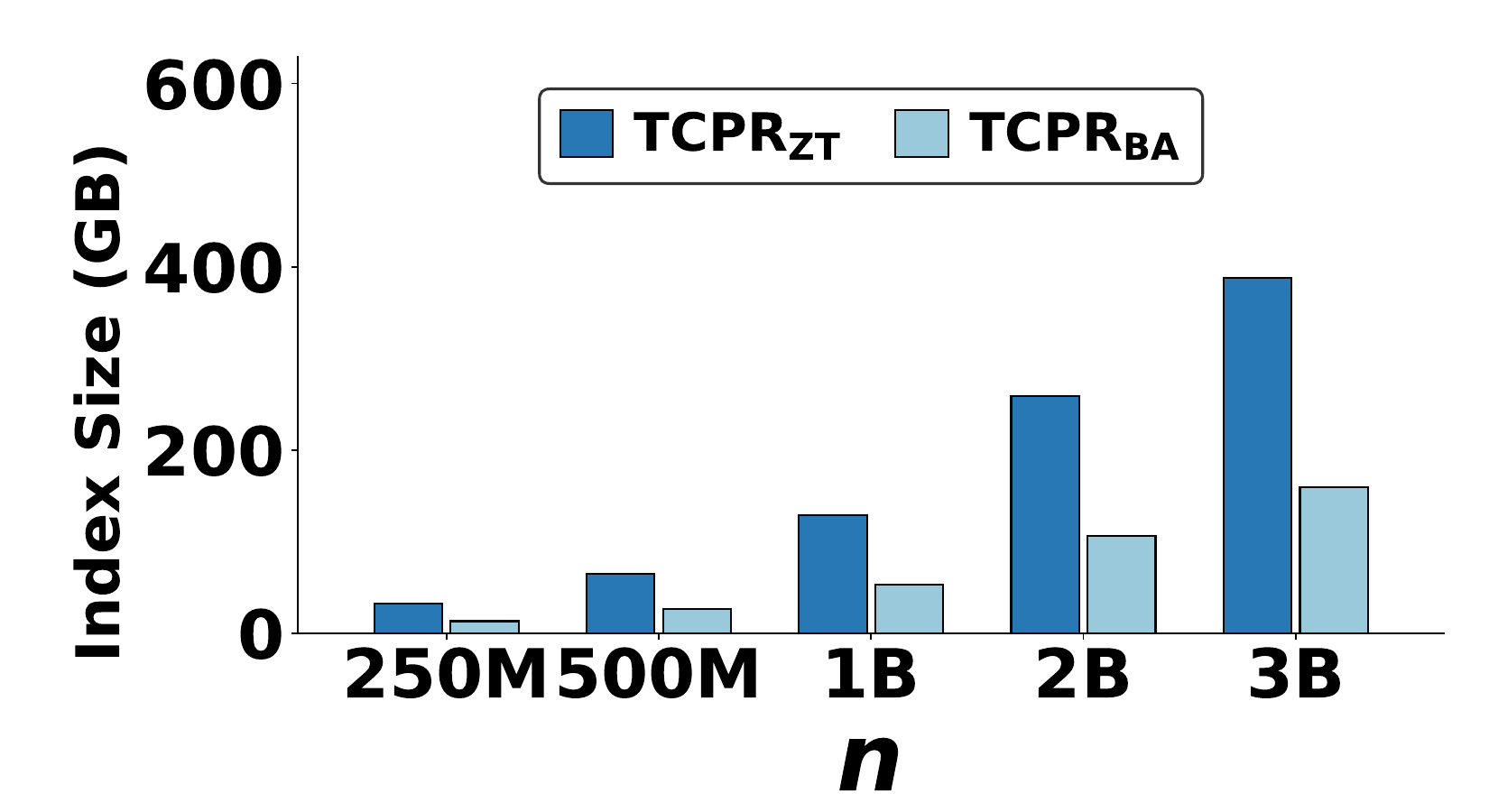}
    \caption{Index size vs. $n$}\label{fig:app:SP:n:index:SARS}
  \end{subfigure}
  \begin{subfigure}[t]{\appfigwidth}
    \includegraphics[width=\linewidth]{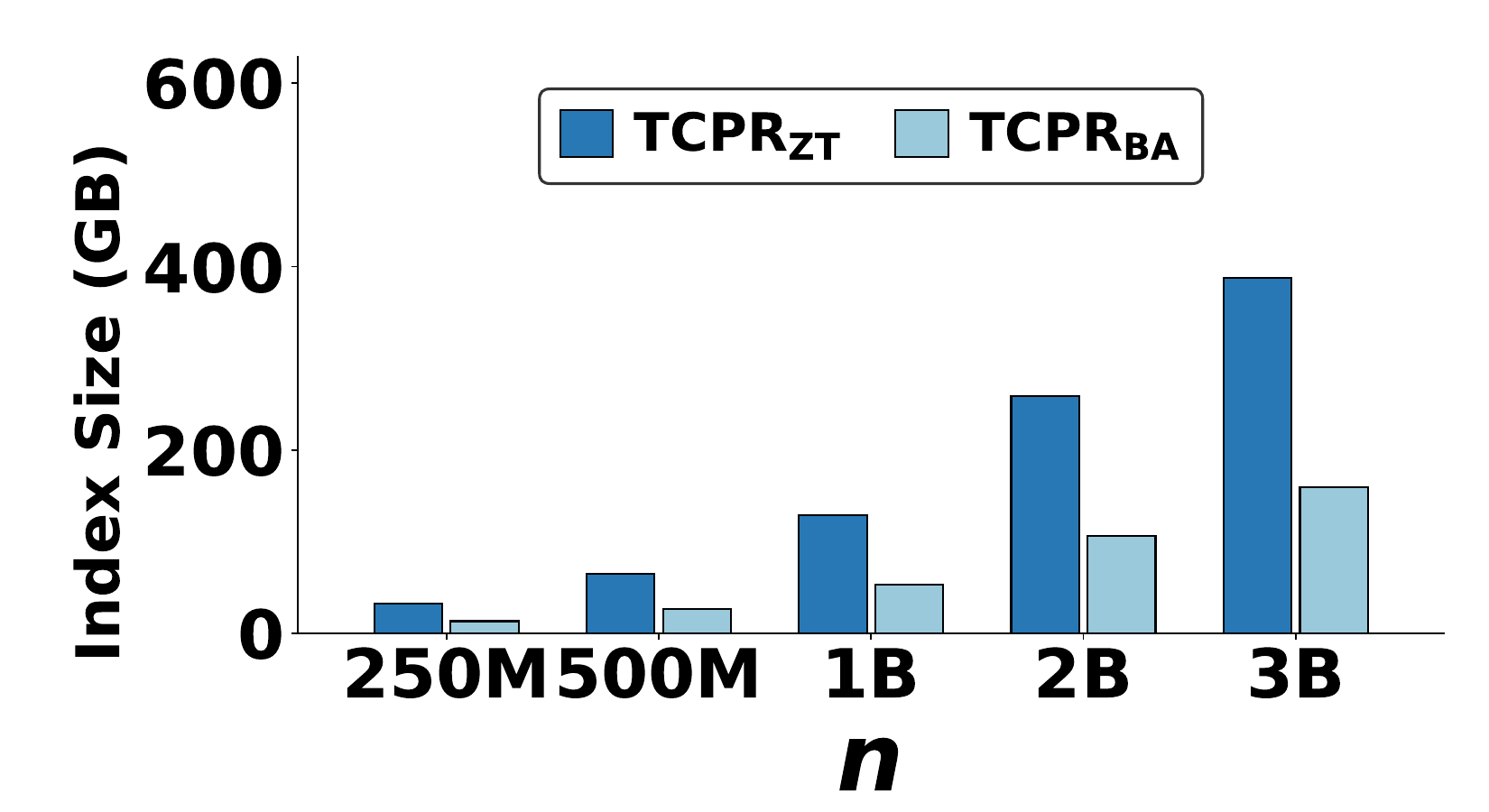}
    \caption{Index size vs. $n$}\label{fig:app:SP:n:index:SDSL}
  \end{subfigure}
  \begin{subfigure}[t]{\appfigwidth}
    \includegraphics[width=\linewidth]{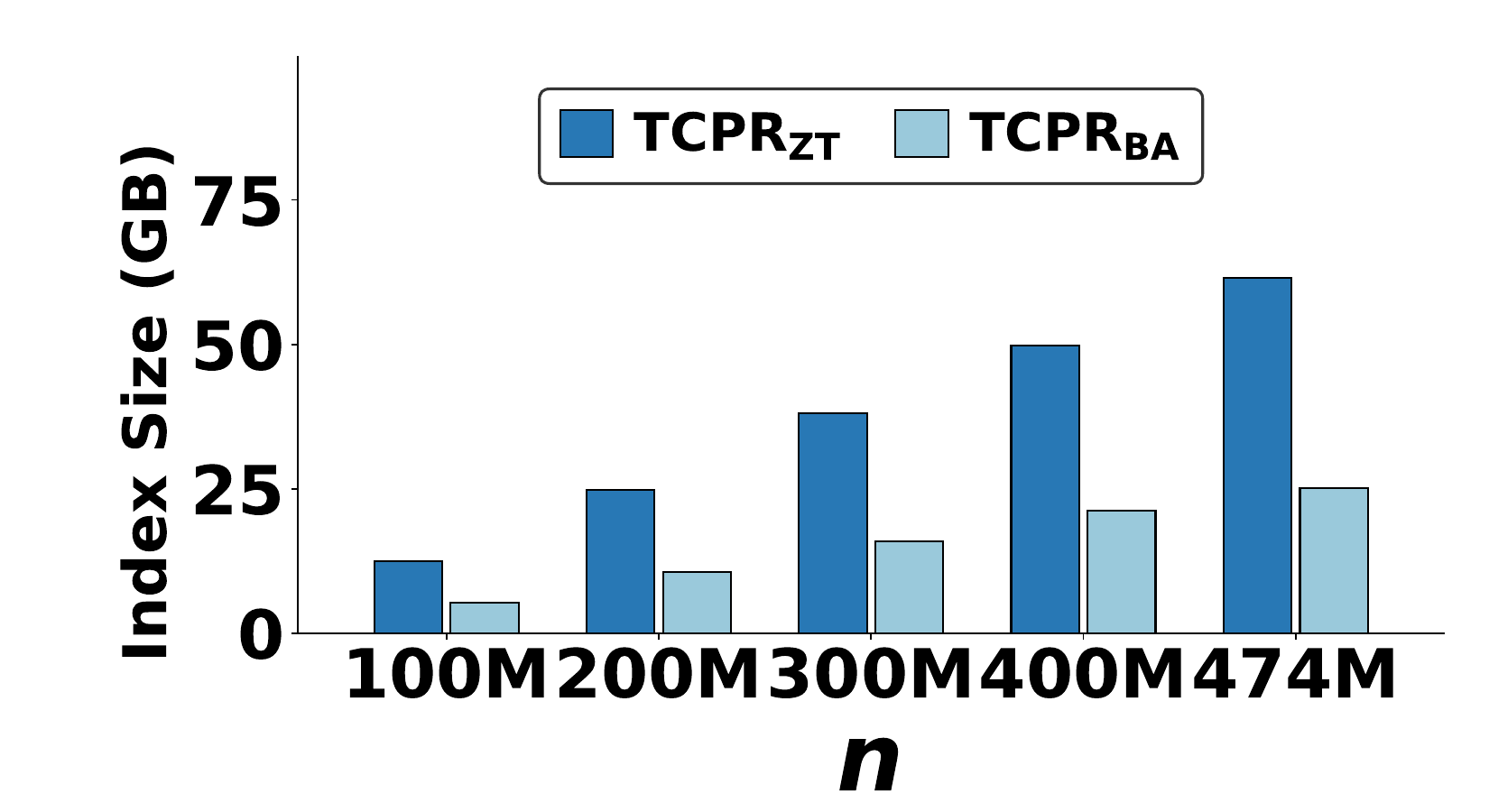}
    \caption{Index size vs. $n$}\label{fig:app:SP:n:index:WIKI}
  \end{subfigure}\\[0pt]
  \begin{subfigure}[t]{\appfigwidth}
    \includegraphics[width=\linewidth]{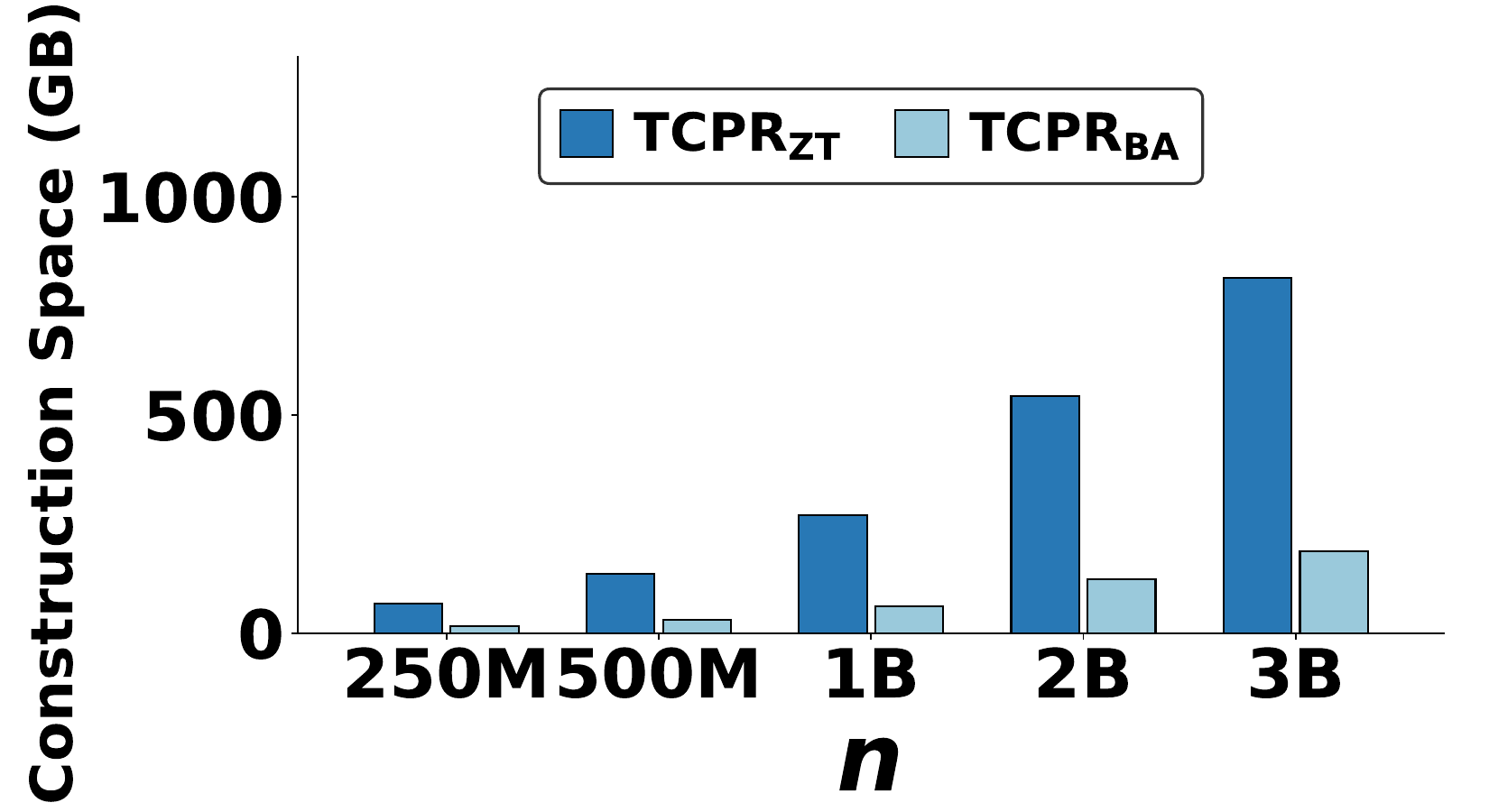}
    \caption{Constr.\ space vs. $n$}\label{fig:app:SP:n:rss:BST}
  \end{subfigure}
  \begin{subfigure}[t]{\appfigwidth}
    \includegraphics[width=\linewidth]{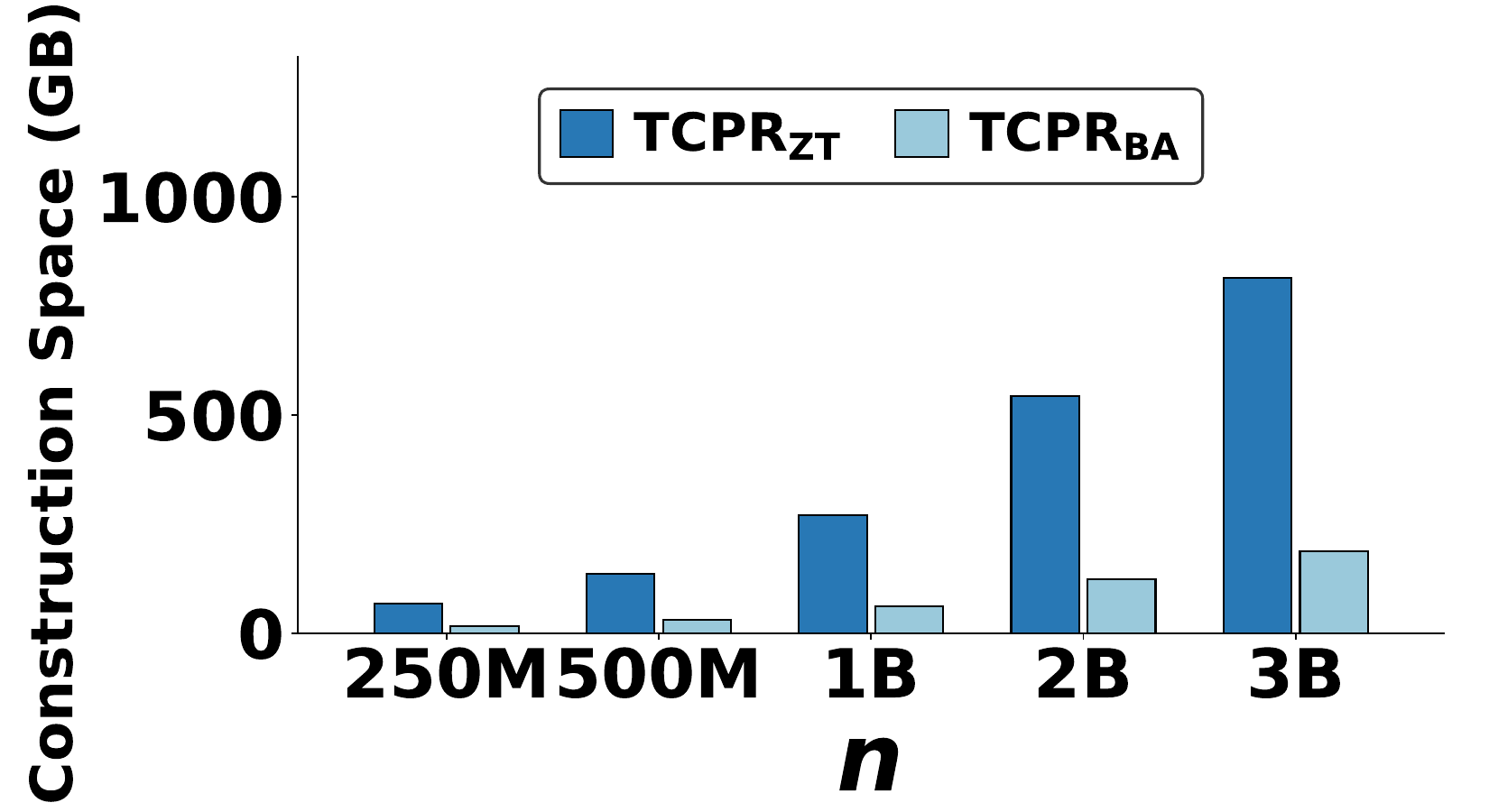}
    \caption{Constr.\ space vs. $n$}\label{fig:app:SP:n:rss:SARS}
  \end{subfigure}
  \begin{subfigure}[t]{\appfigwidth}
    \includegraphics[width=\linewidth]{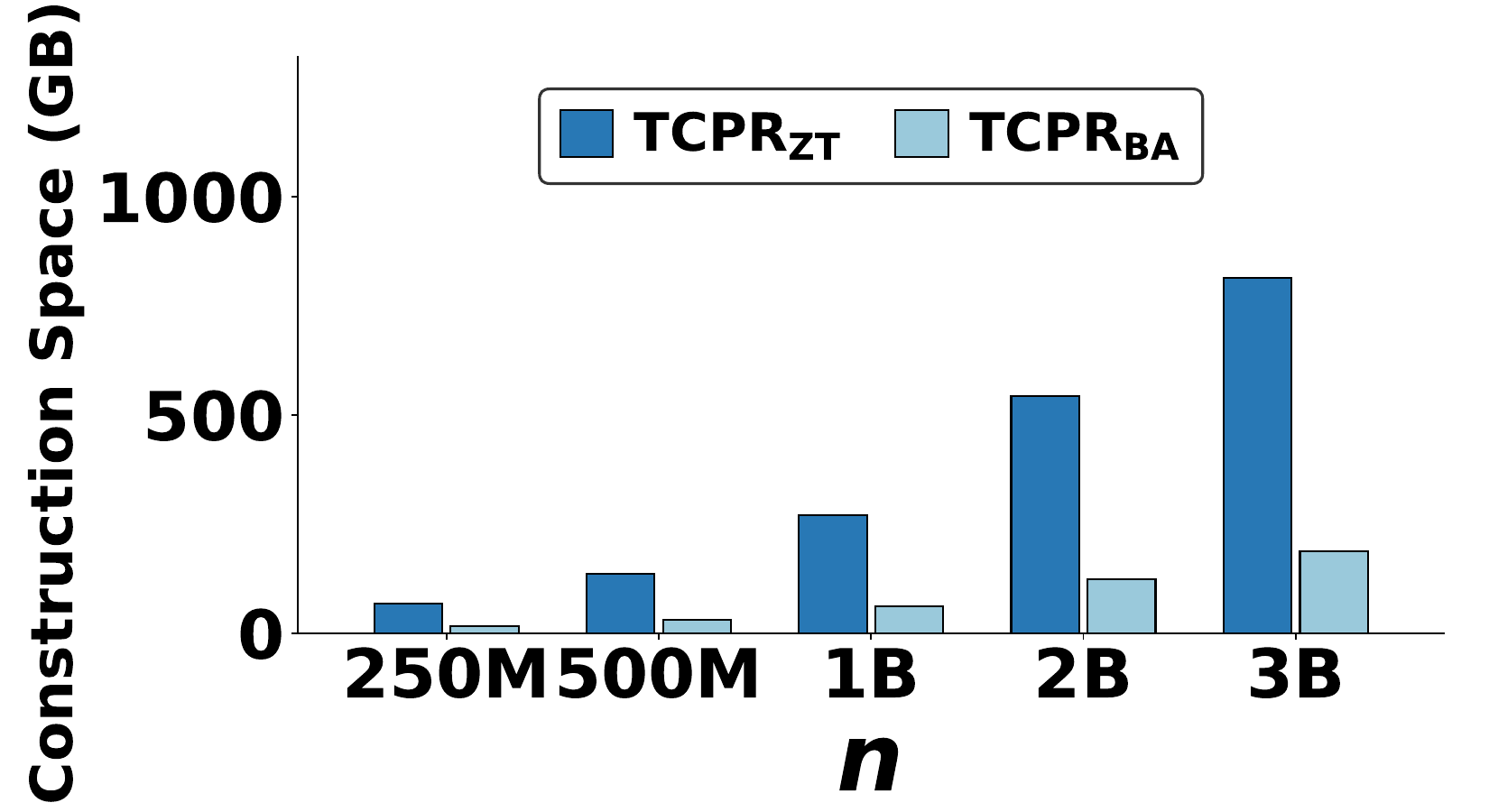}
    \caption{Constr.\ space vs. $n$}\label{fig:app:SP:n:rss:SDSL}
  \end{subfigure}
  \begin{subfigure}[t]{\appfigwidth}
    \includegraphics[width=\linewidth]{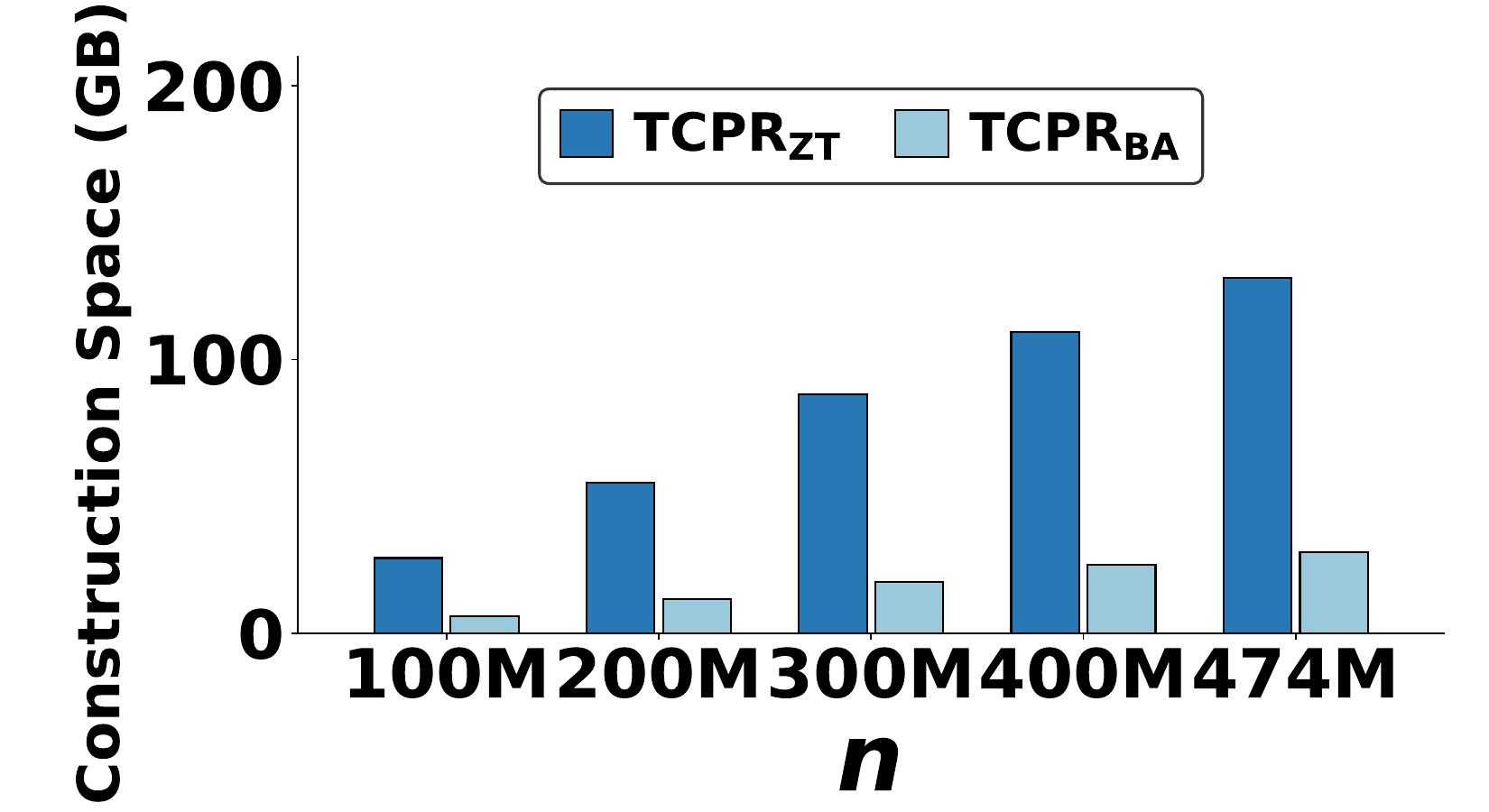}
    \caption{Constr.\ space vs. $n$}\label{fig:app:SP:n:rss:WIKI}
  \end{subfigure}\\[0pt]
  \begin{subfigure}[t]{\appfigwidth}
    \includegraphics[width=\linewidth]{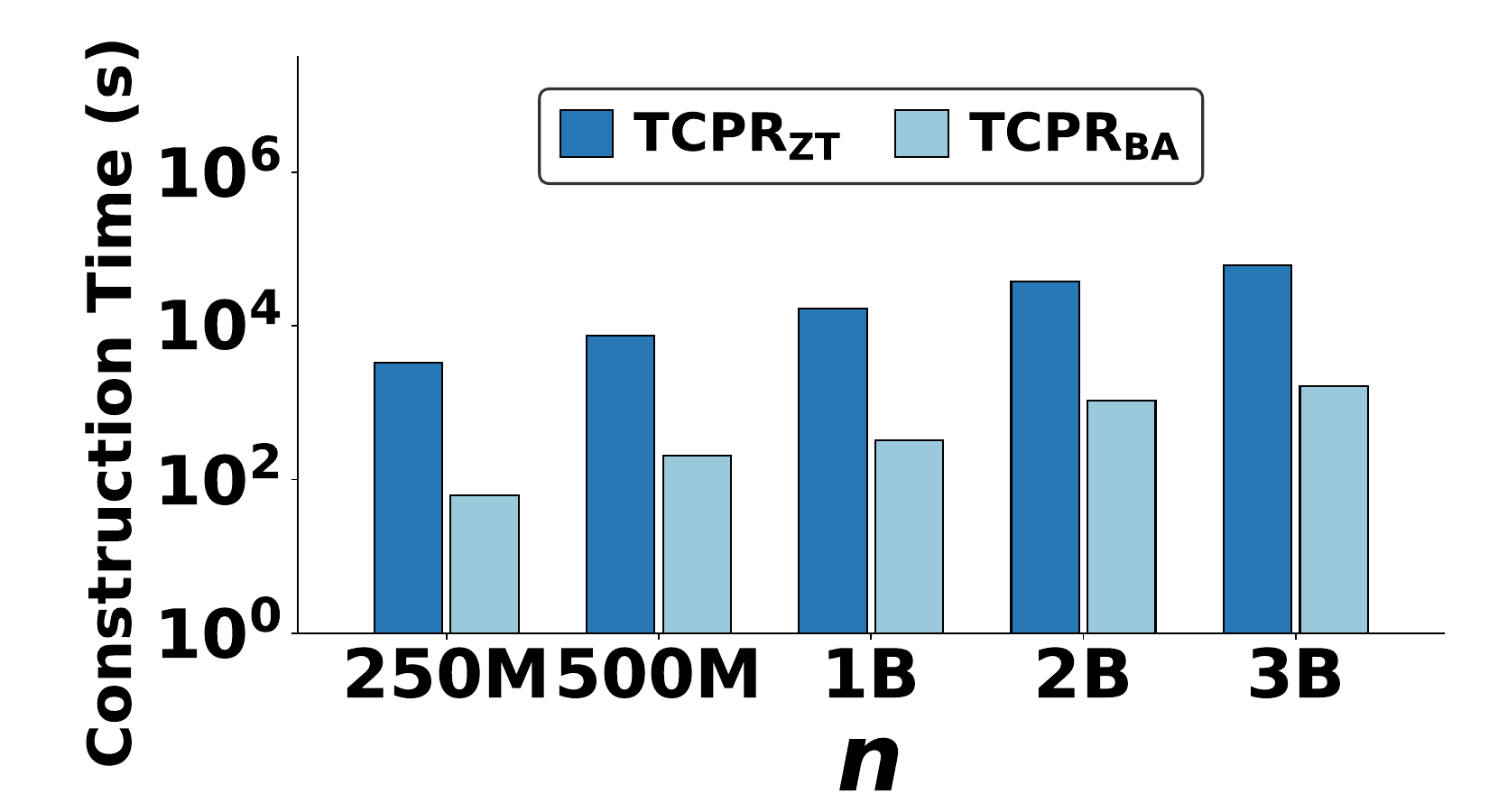}
    \caption{Constr.\ time vs. $n$}\label{fig:app:SP:n:build:BST}
  \end{subfigure}
  \begin{subfigure}[t]{\appfigwidth}
    \includegraphics[width=\linewidth]{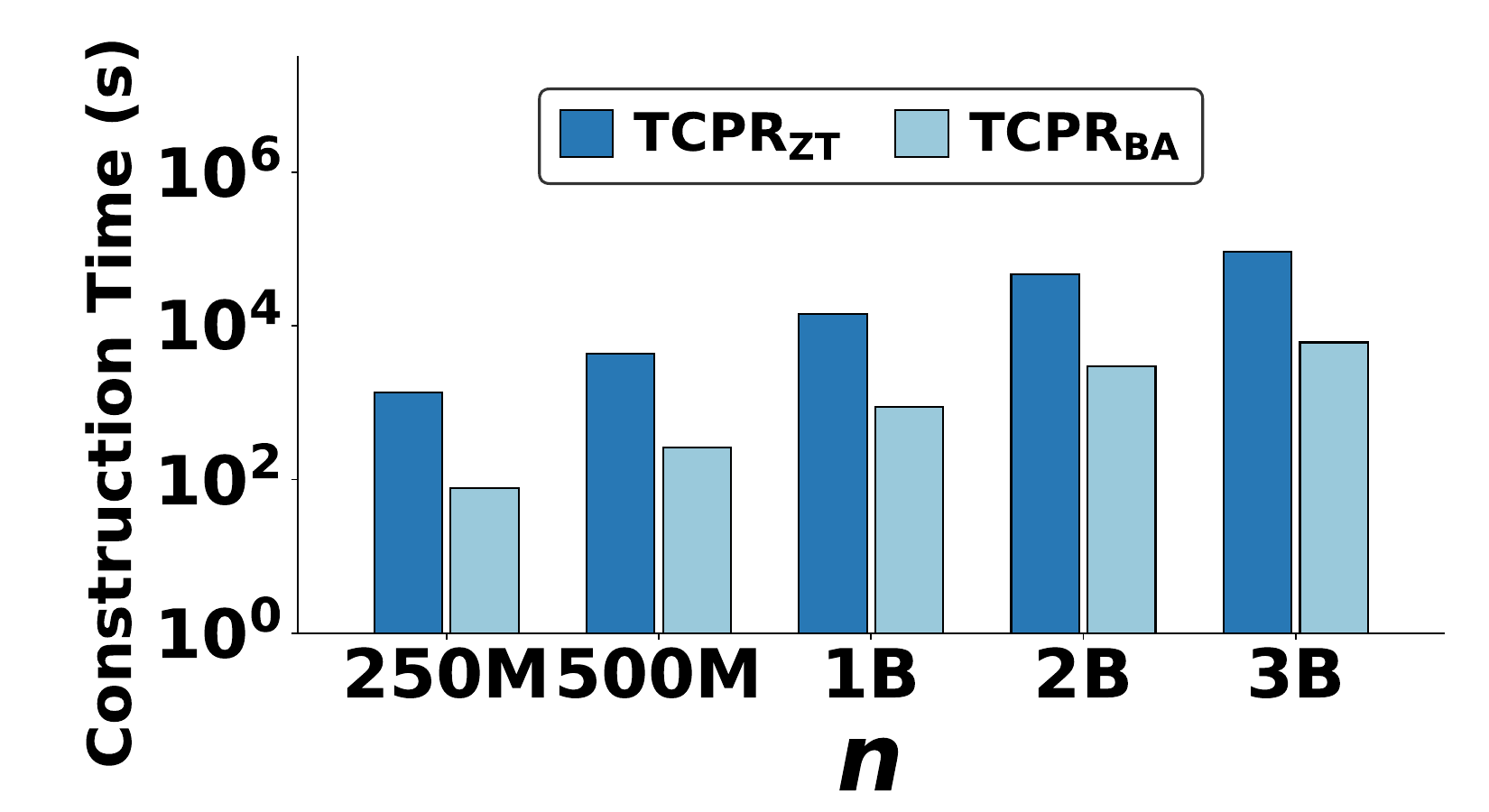}
    \caption{Constr.\ time vs. $n$}\label{fig:app:SP:n:build:SARS}
  \end{subfigure}
  \begin{subfigure}[t]{\appfigwidth}
    \includegraphics[width=\linewidth]{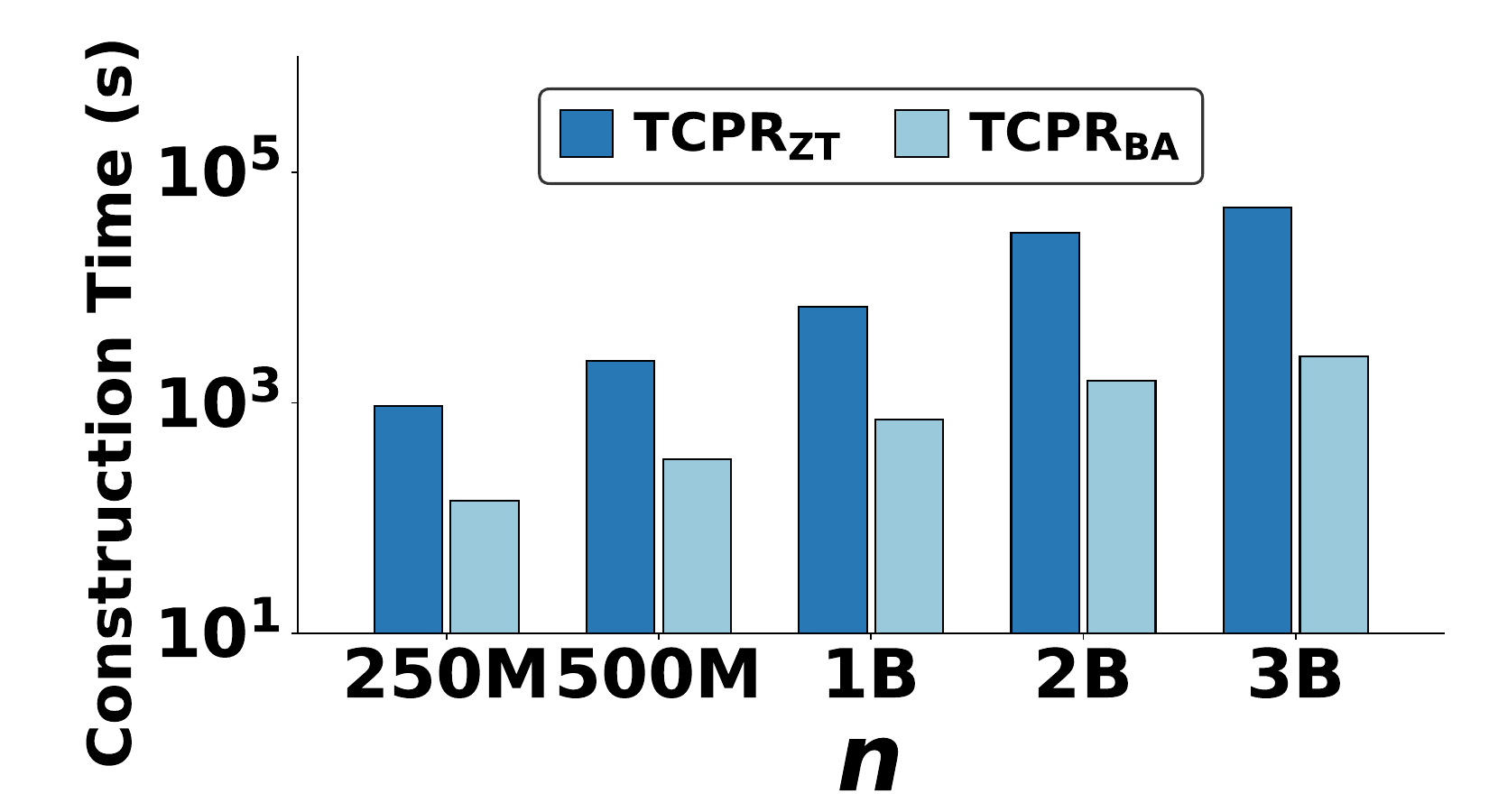}
    \caption{Constr.\ time vs. $n$}\label{fig:app:SP:n:build:SDSL}
  \end{subfigure}
  \begin{subfigure}[t]{\appfigwidth}
    \includegraphics[width=\linewidth]{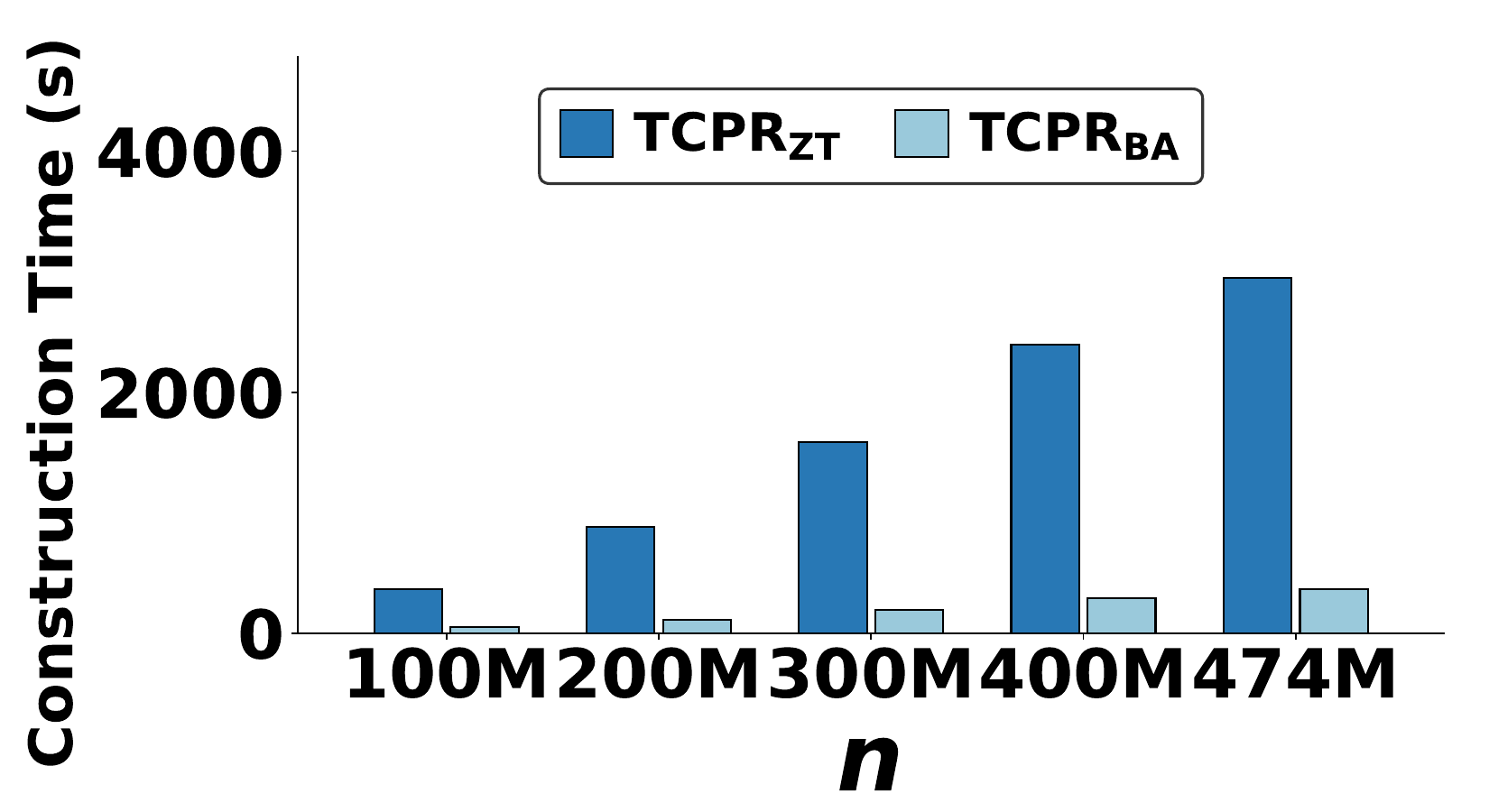}
    \caption{Constr.\ time vs. $n$}\label{fig:app:SP:n:build:WIKI}
  \end{subfigure}
  \vspace{\captionspacing}
  \vspace{+2mm}
  \caption{Index size of our \TCPR index with the \textsf{SP} scoring function vs. \TCPRBA on (a) \bst, (b) \sars, (c) \sdsl, and (d) \wiki vs. $n$; construction space of our \TCPR index with the \textsf{SP} scoring function vs. \TCPRBA on (e) \bst, (f) \sars, (g) \sdsl, and (h) \wiki vs. $n$; construction time of our \TCPR index with the \textsf{SP} scoring function vs. \TCPRBA on (i) \bst, (j) \sars, (k) \sdsl, and (l) \wiki vs. $n$.}\label{fig:app:SP:cost}
\end{figure}

\begin{figure}[ht]
  \centering
  \begin{subfigure}[t]{\appfigwidth}
    \includegraphics[width=\linewidth]{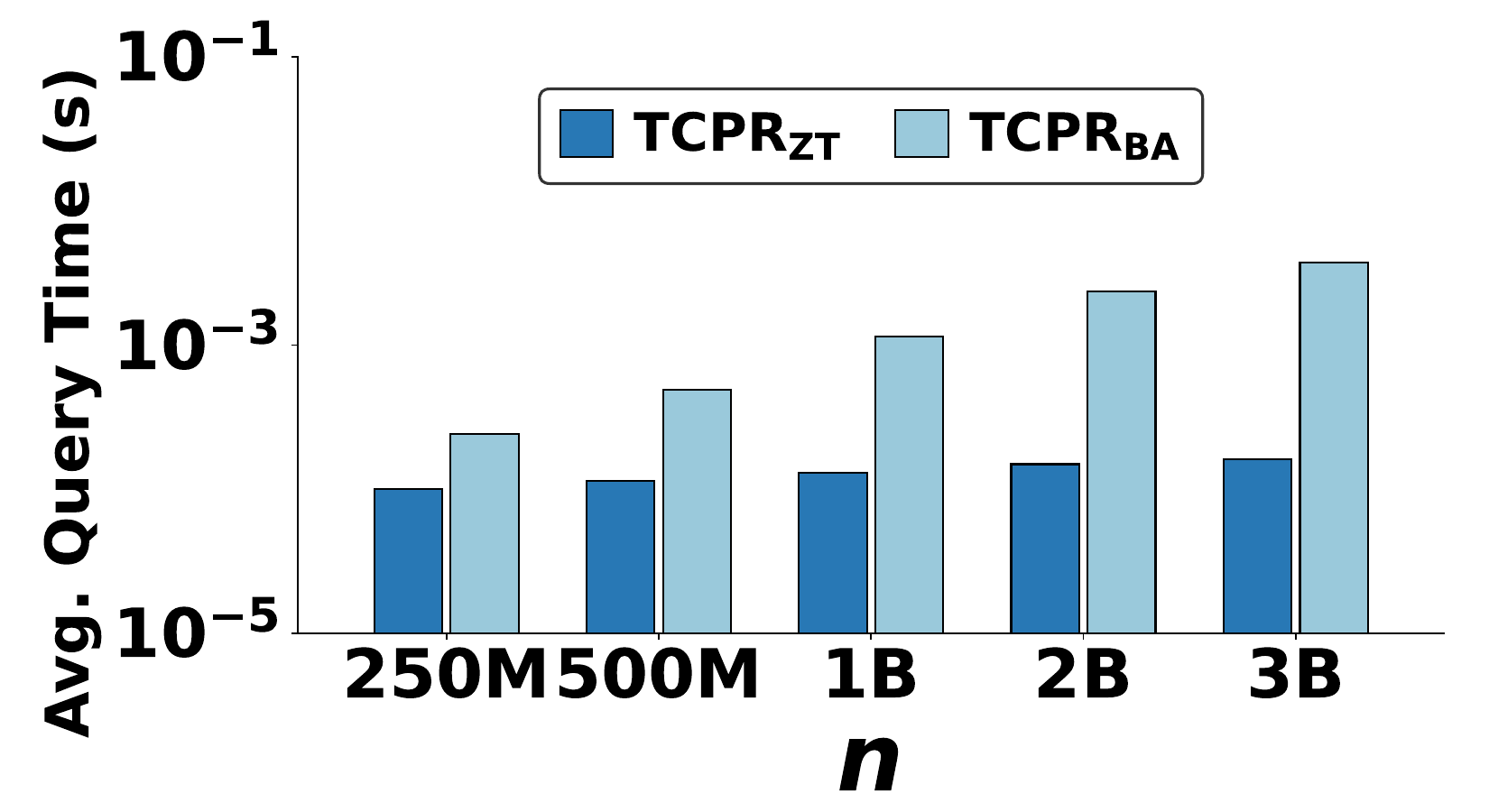}
    \caption{Query time vs. $n$}\label{fig:app:TP:n:query:BST}
  \end{subfigure}
  \begin{subfigure}[t]{\appfigwidth}
    \includegraphics[width=\linewidth]{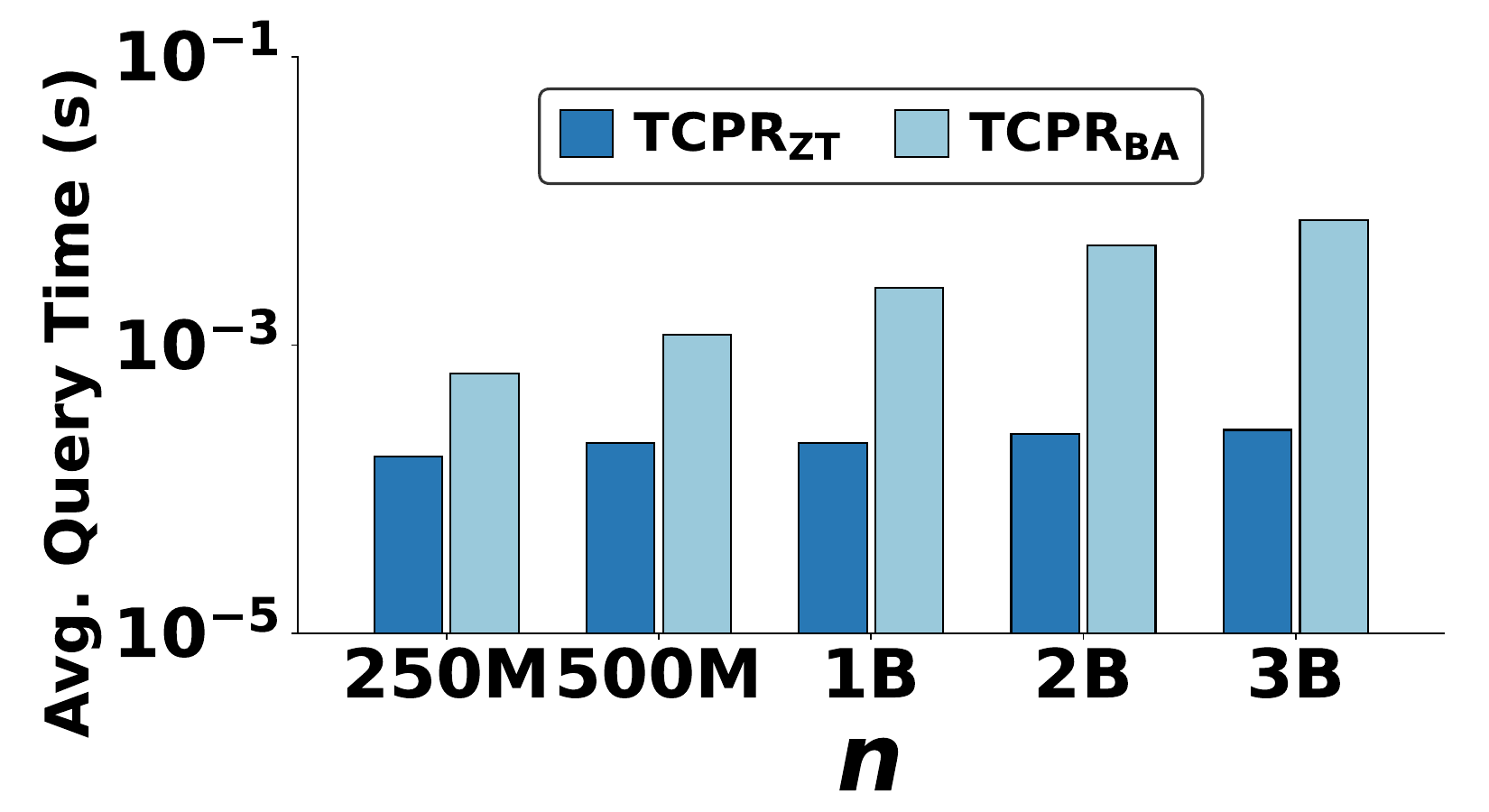}
    \caption{Query time vs. $n$}\label{fig:app:TP:n:query:SARS}
  \end{subfigure}
  \begin{subfigure}[t]{\appfigwidth}
    \includegraphics[width=\linewidth]{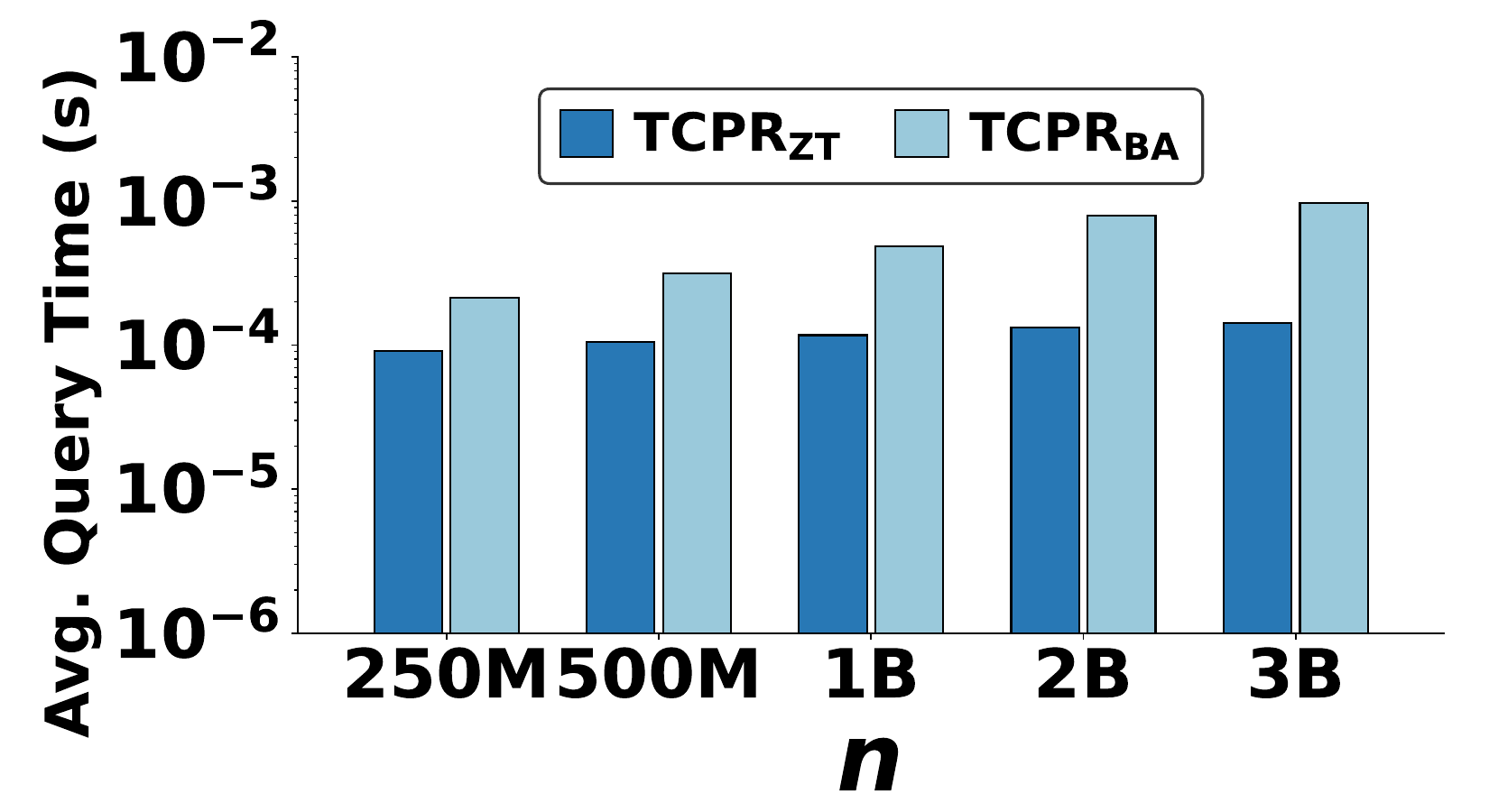}
    \caption{Query time vs. $n$}\label{fig:app:TP:n:query:SDSL}
  \end{subfigure}
  \begin{subfigure}[t]{\appfigwidth}
    \includegraphics[width=\linewidth]{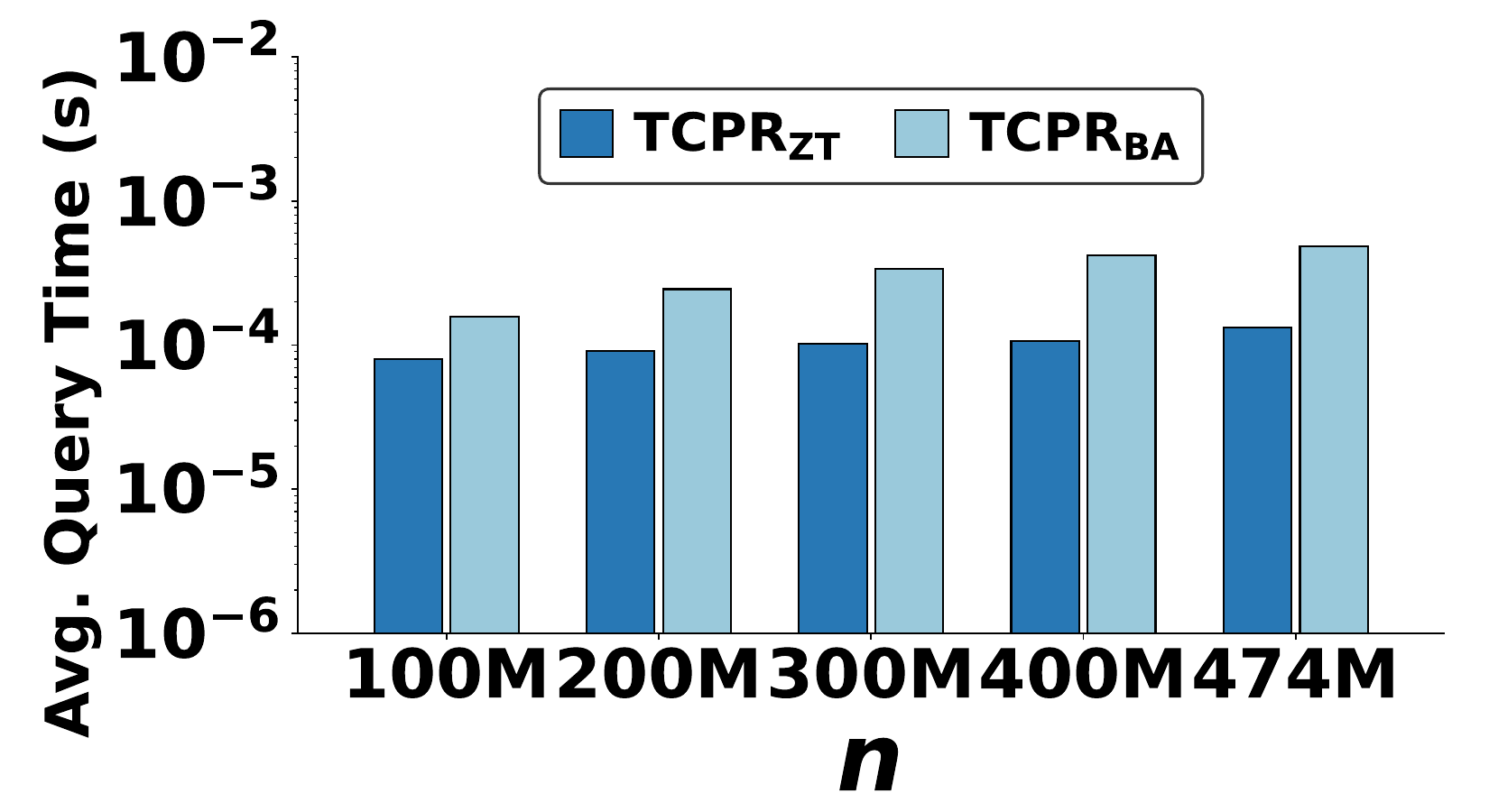}
    \caption{Query time vs. $n$}\label{fig:app:TP:n:query:WIKI}
  \end{subfigure}\\[0pt]
  \begin{subfigure}[t]{\appfigwidth}
    \includegraphics[width=\linewidth]{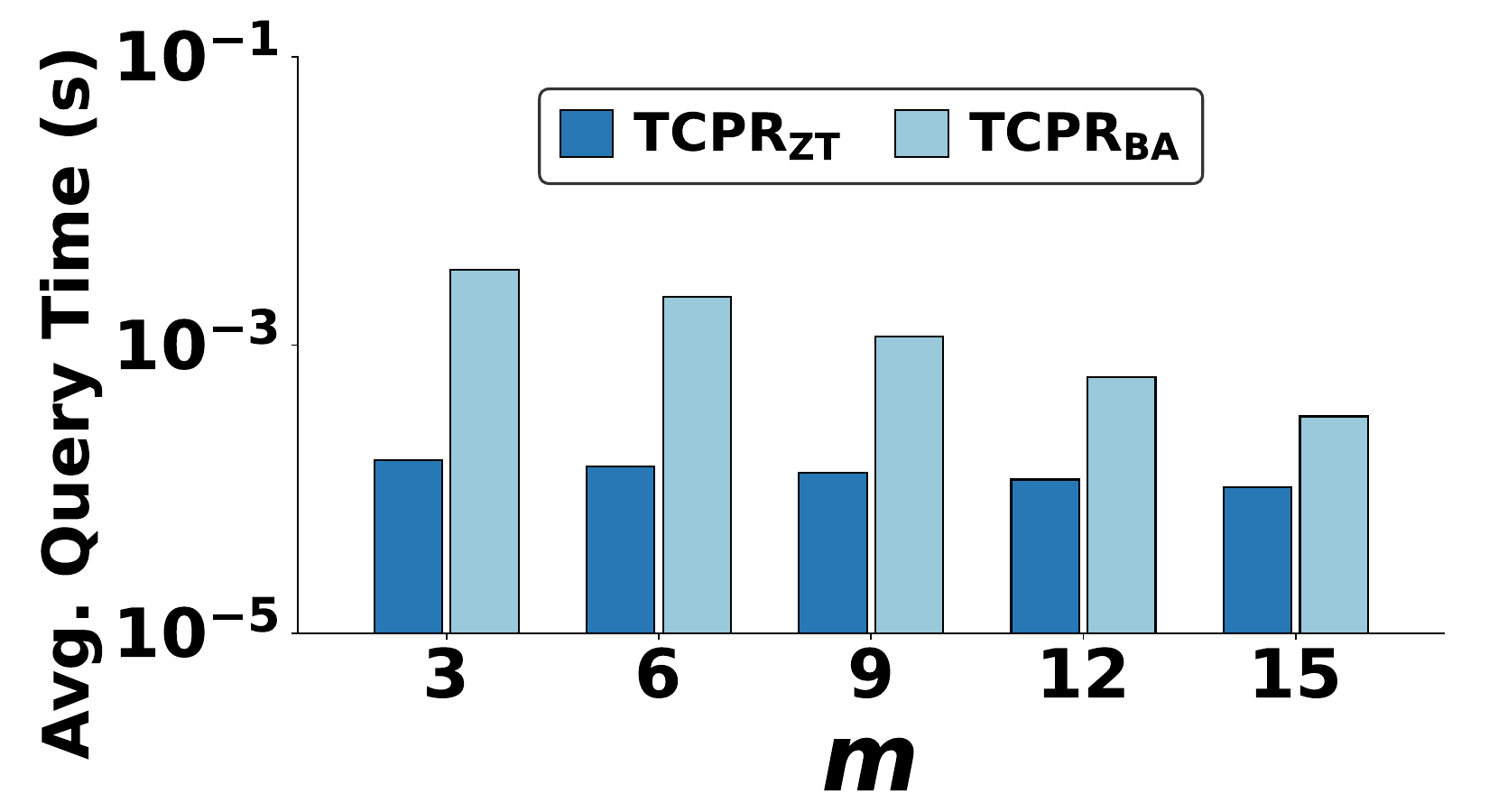}
    \caption{Query time vs. $m$}\label{fig:app:TP:m:query:BST}
  \end{subfigure}
  \begin{subfigure}[t]{\appfigwidth}
    \includegraphics[width=\linewidth]{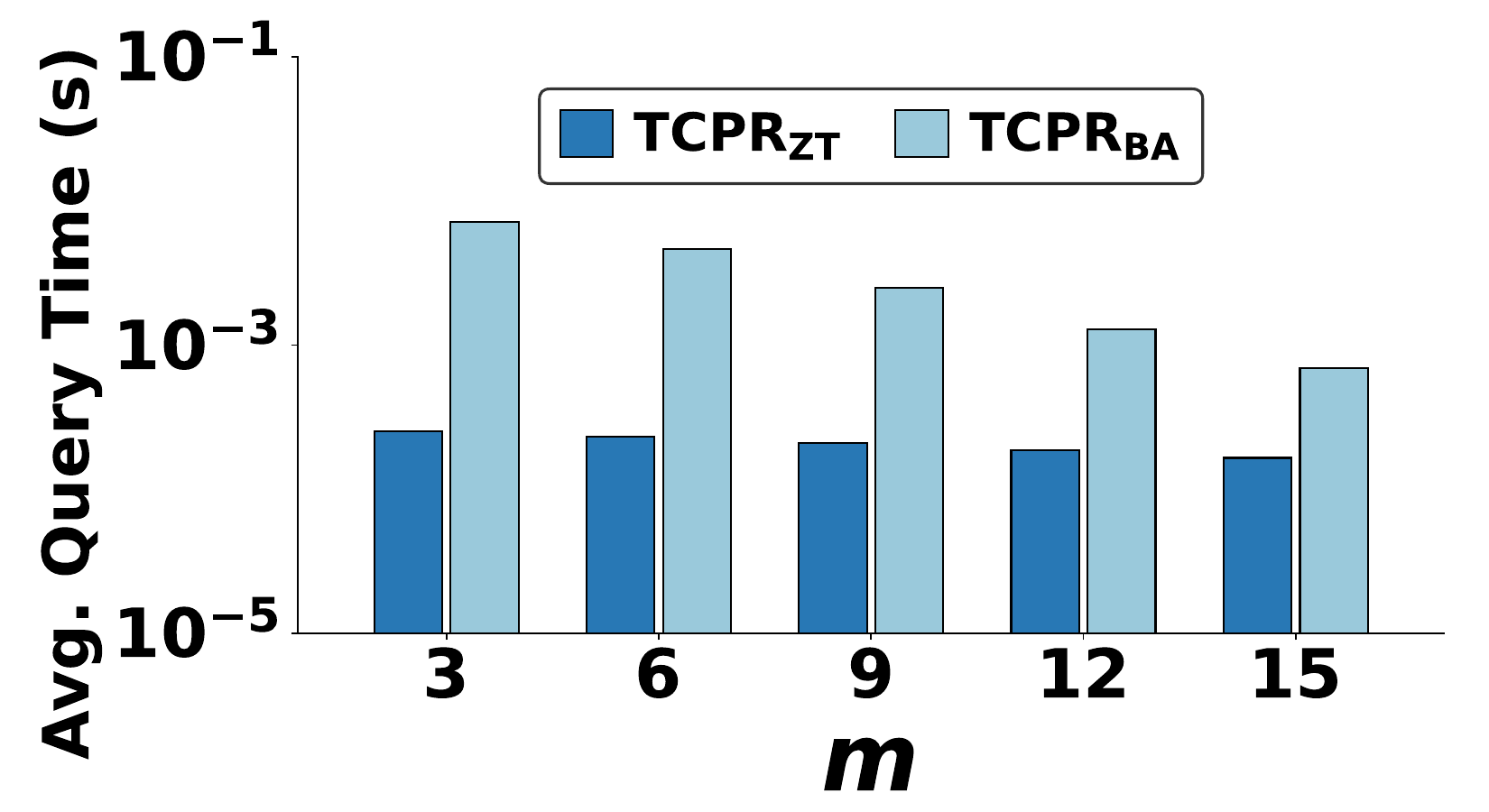}
    \caption{Query time vs. $m$}\label{fig:app:TP:m:query:SARS}
  \end{subfigure}
  \begin{subfigure}[t]{\appfigwidth}
    \includegraphics[width=\linewidth]{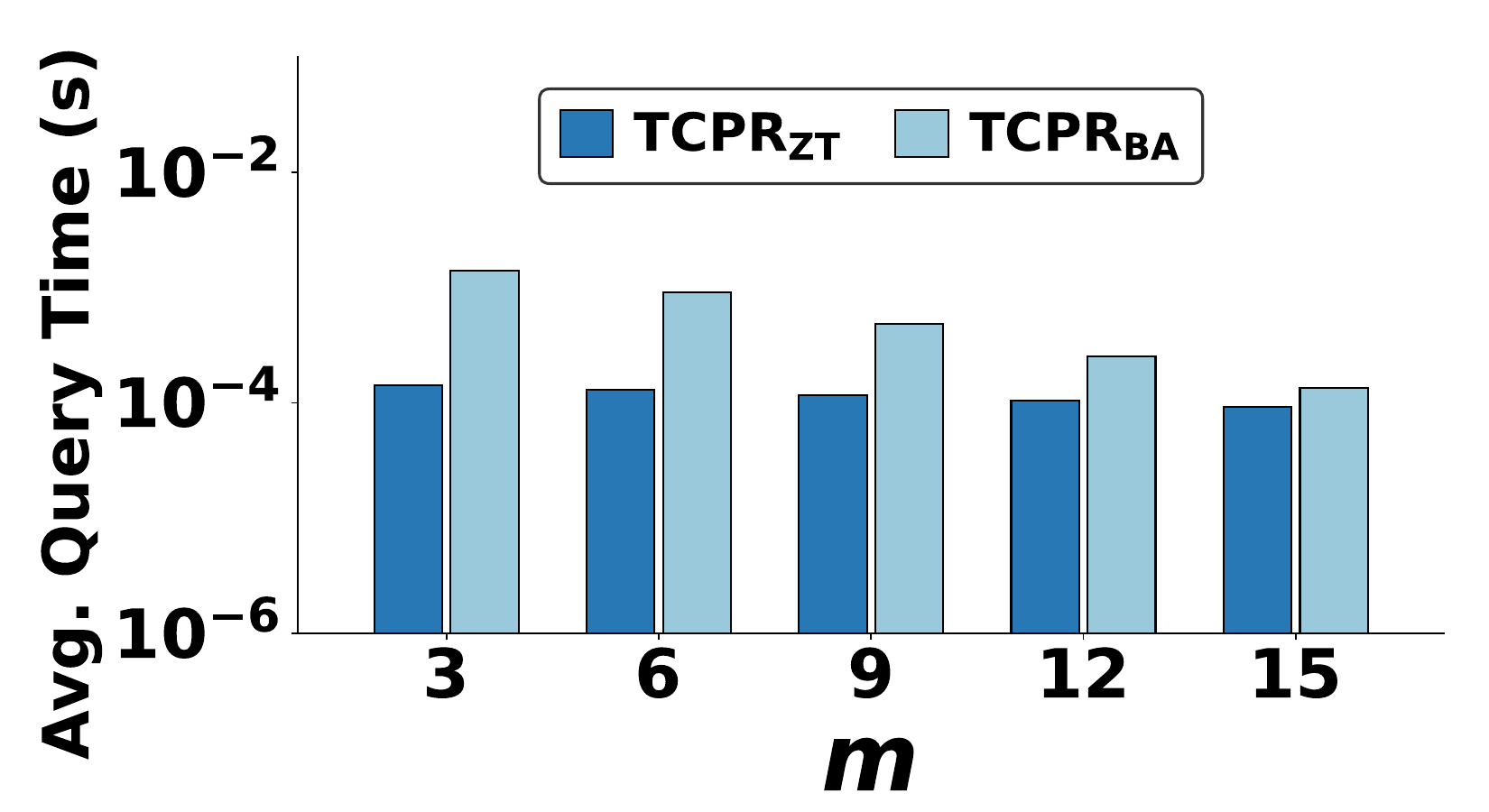}
    \caption{Query time vs. $m$}\label{fig:app:TP:m:query:SDSL}
  \end{subfigure}
  \begin{subfigure}[t]{\appfigwidth}
    \includegraphics[width=\linewidth]{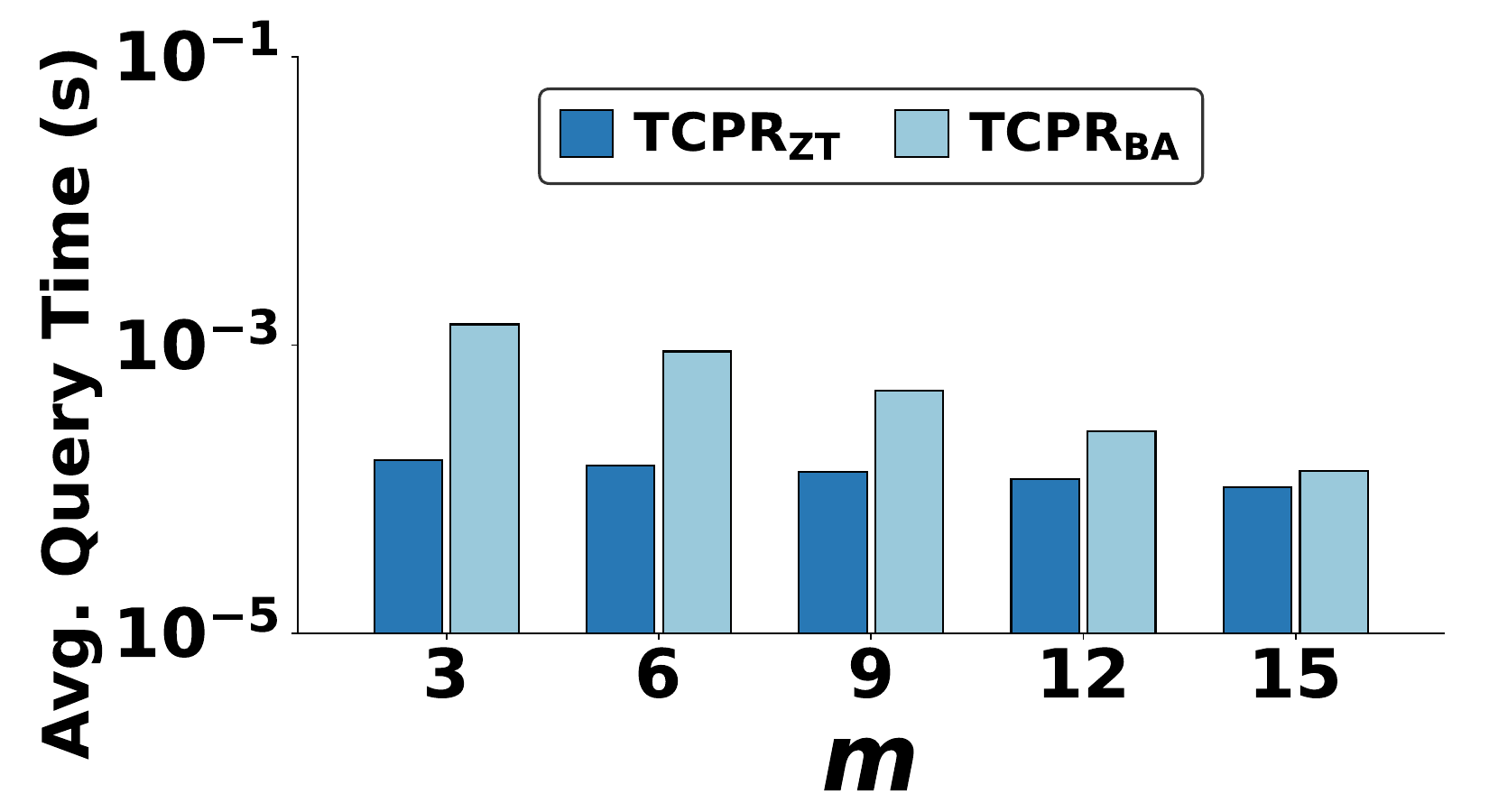}
    \caption{Query time vs. $m$}\label{fig:app:TP:m:query:WIKI}
  \end{subfigure}\\[0pt]
  \begin{subfigure}[t]{\appfigwidth}
    \includegraphics[width=\linewidth]{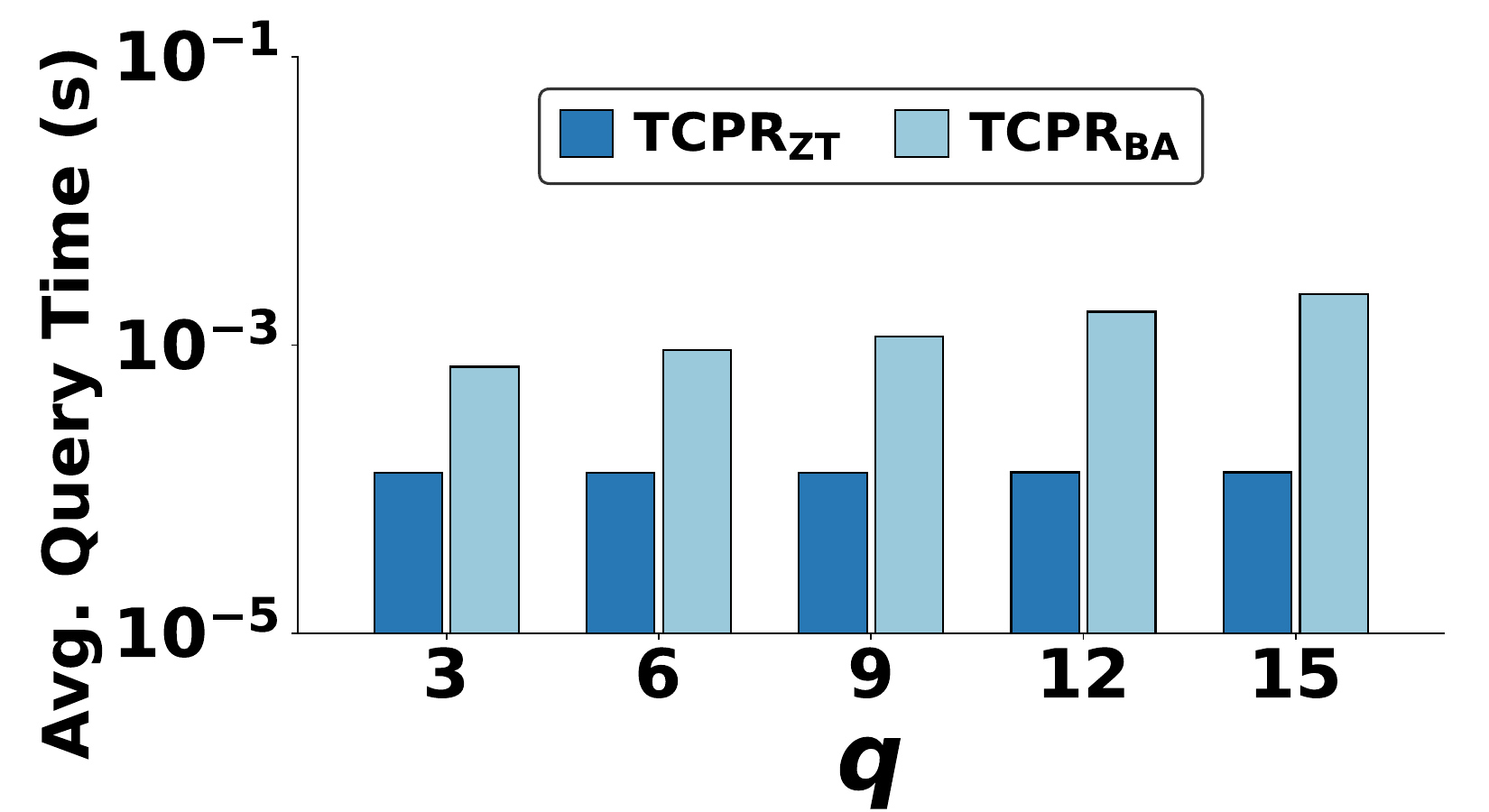}
    \caption{Query time vs. $q$}\label{fig:app:TP:q:query:BST}
  \end{subfigure}
  \begin{subfigure}[t]{\appfigwidth}
    \includegraphics[width=\linewidth]{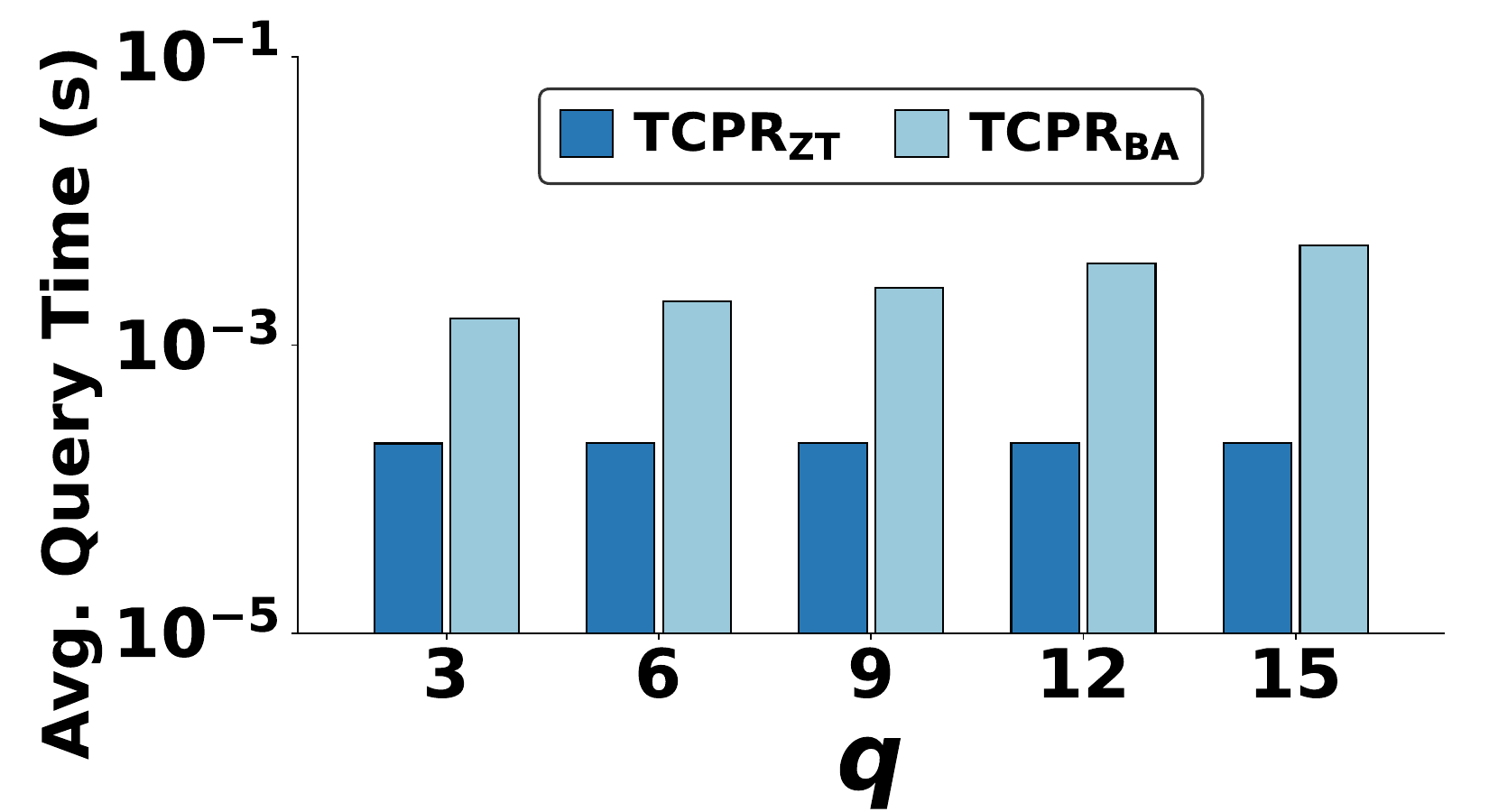}
    \caption{Query time vs. $q$}\label{fig:app:TP:q:query:SARS}
  \end{subfigure}
  \begin{subfigure}[t]{\appfigwidth}
    \includegraphics[width=\linewidth]{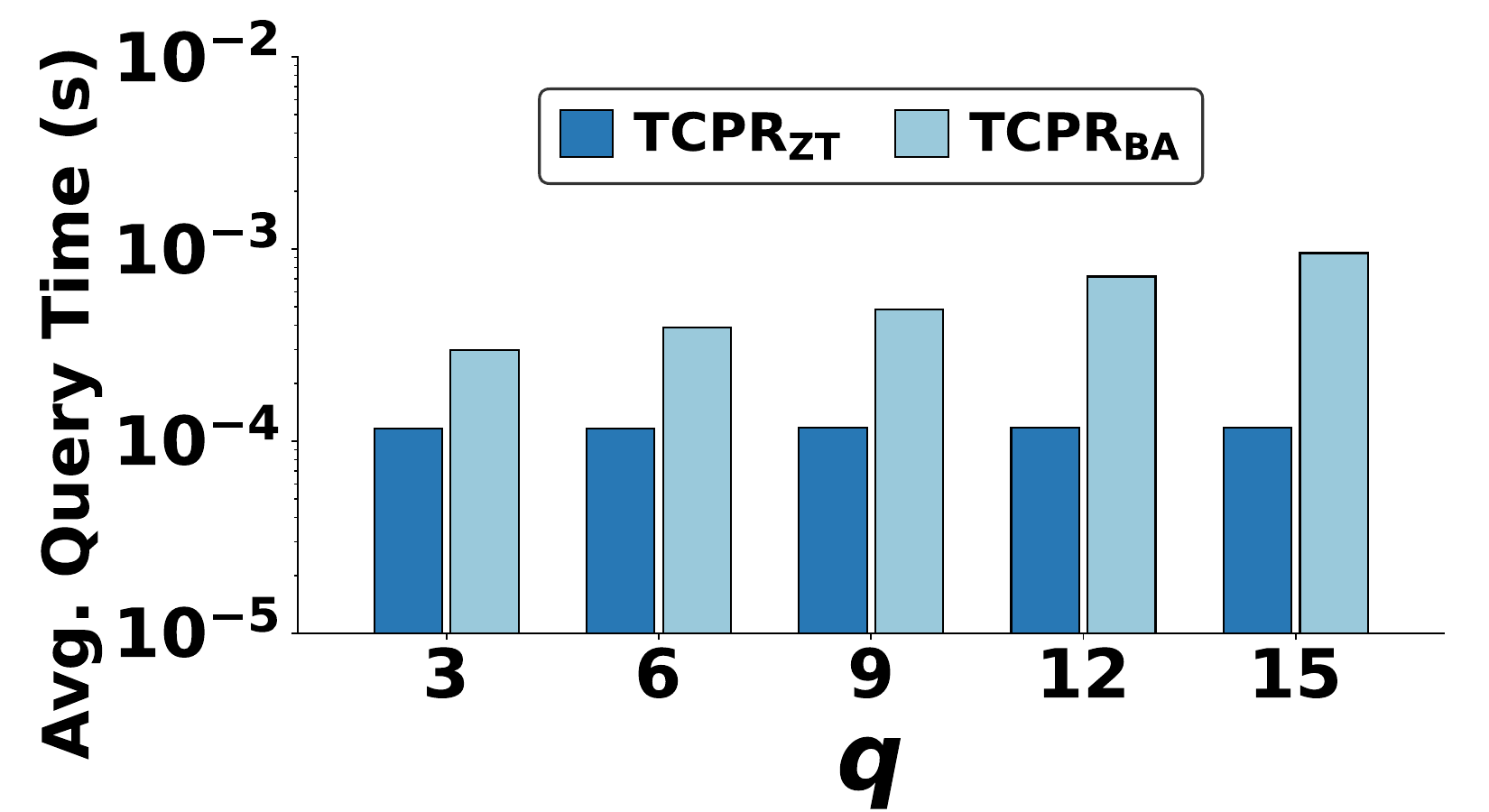}
    \caption{Query time vs. $q$}\label{fig:app:TP:q:query:SDSL}
  \end{subfigure}
  \begin{subfigure}[t]{\appfigwidth}
    \includegraphics[width=\linewidth]{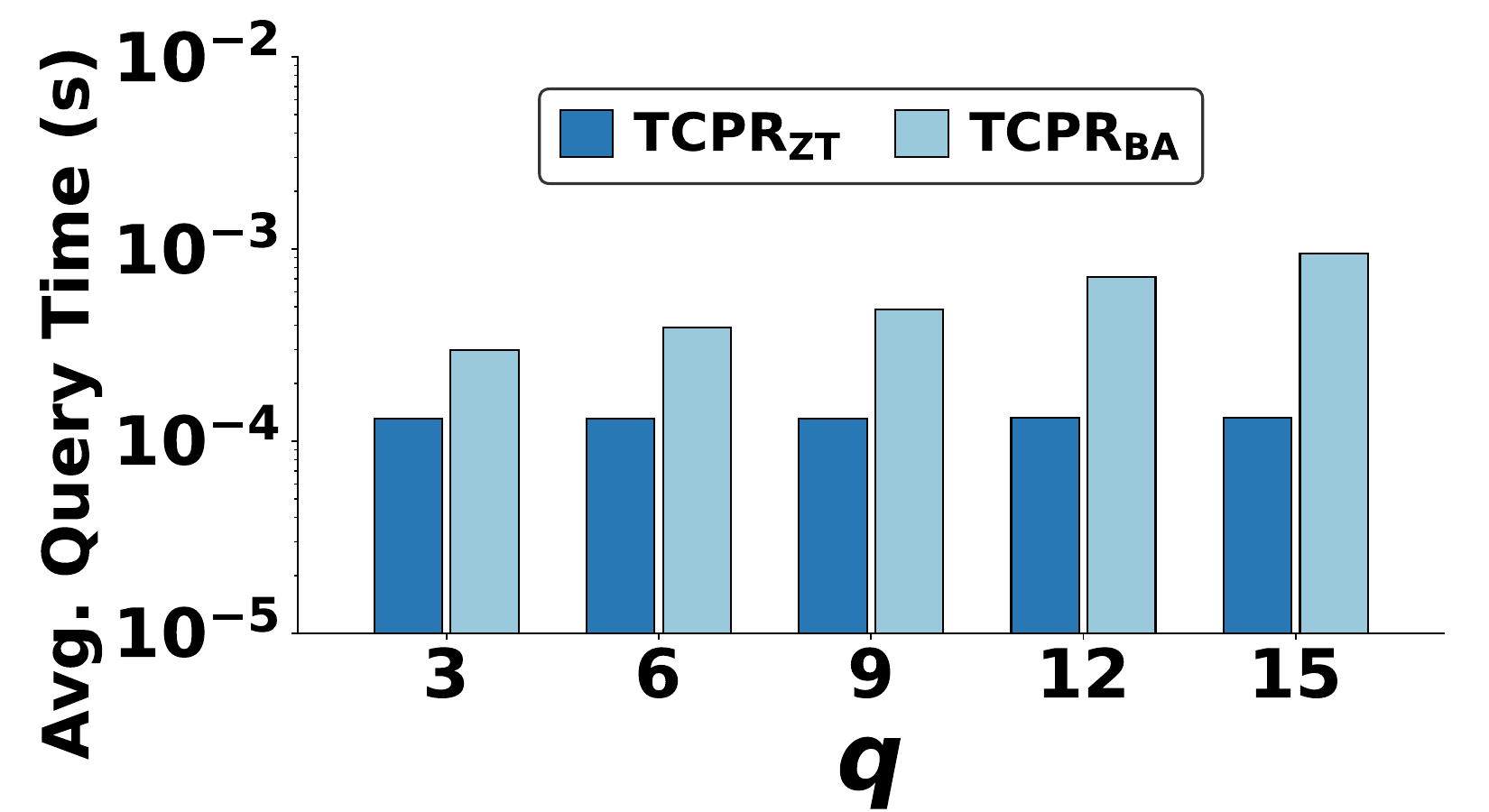}
    \caption{Query time vs. $q$}\label{fig:app:TP:q:query:WIKI}
  \end{subfigure}\\[0pt]
  \begin{subfigure}[t]{\appfigwidth}
    \includegraphics[width=\linewidth]{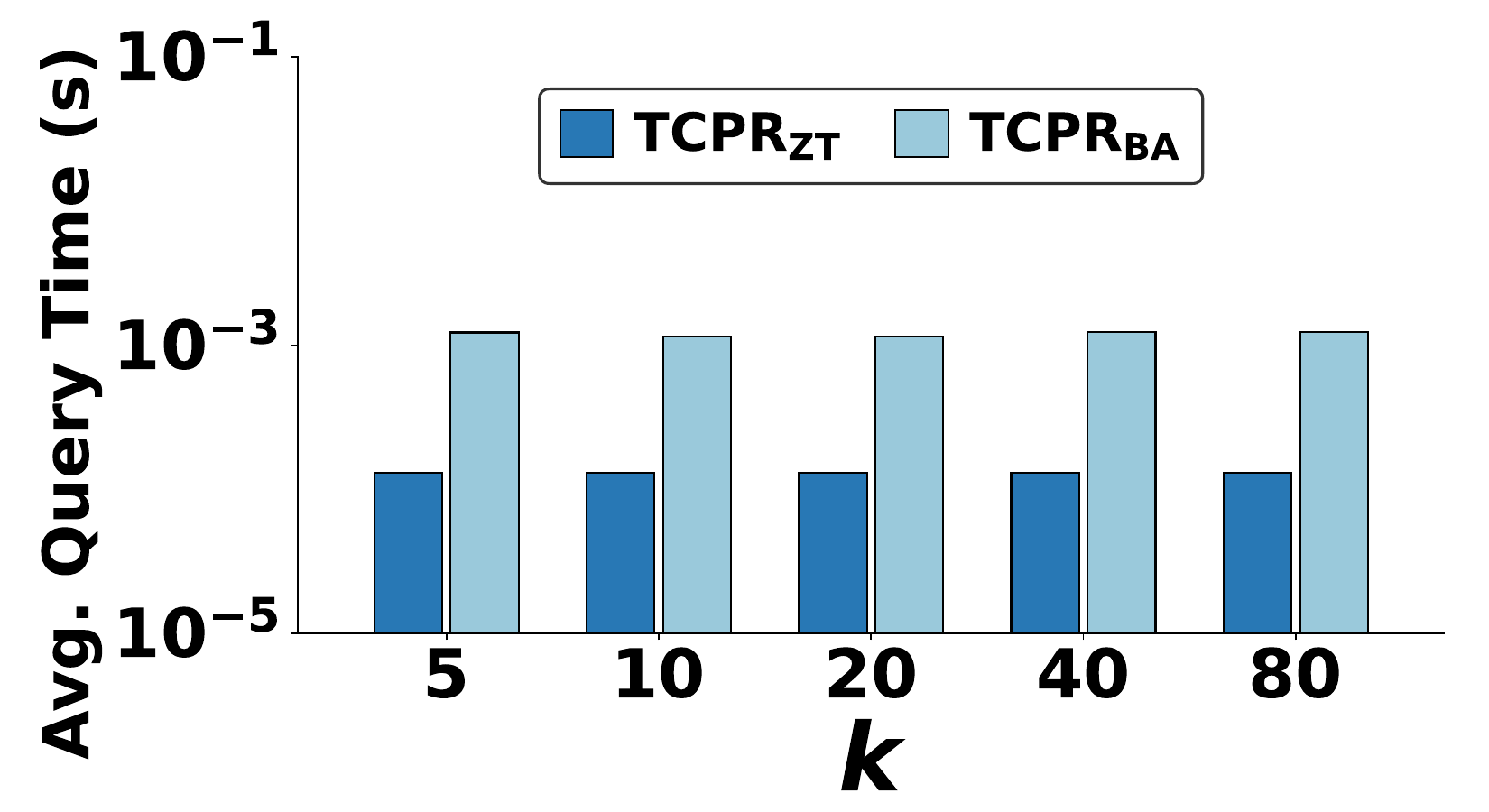}
    \caption{Query time vs. $k$}\label{fig:app:TP:k:query:BST}
  \end{subfigure}
  \begin{subfigure}[t]{\appfigwidth}
    \includegraphics[width=\linewidth]{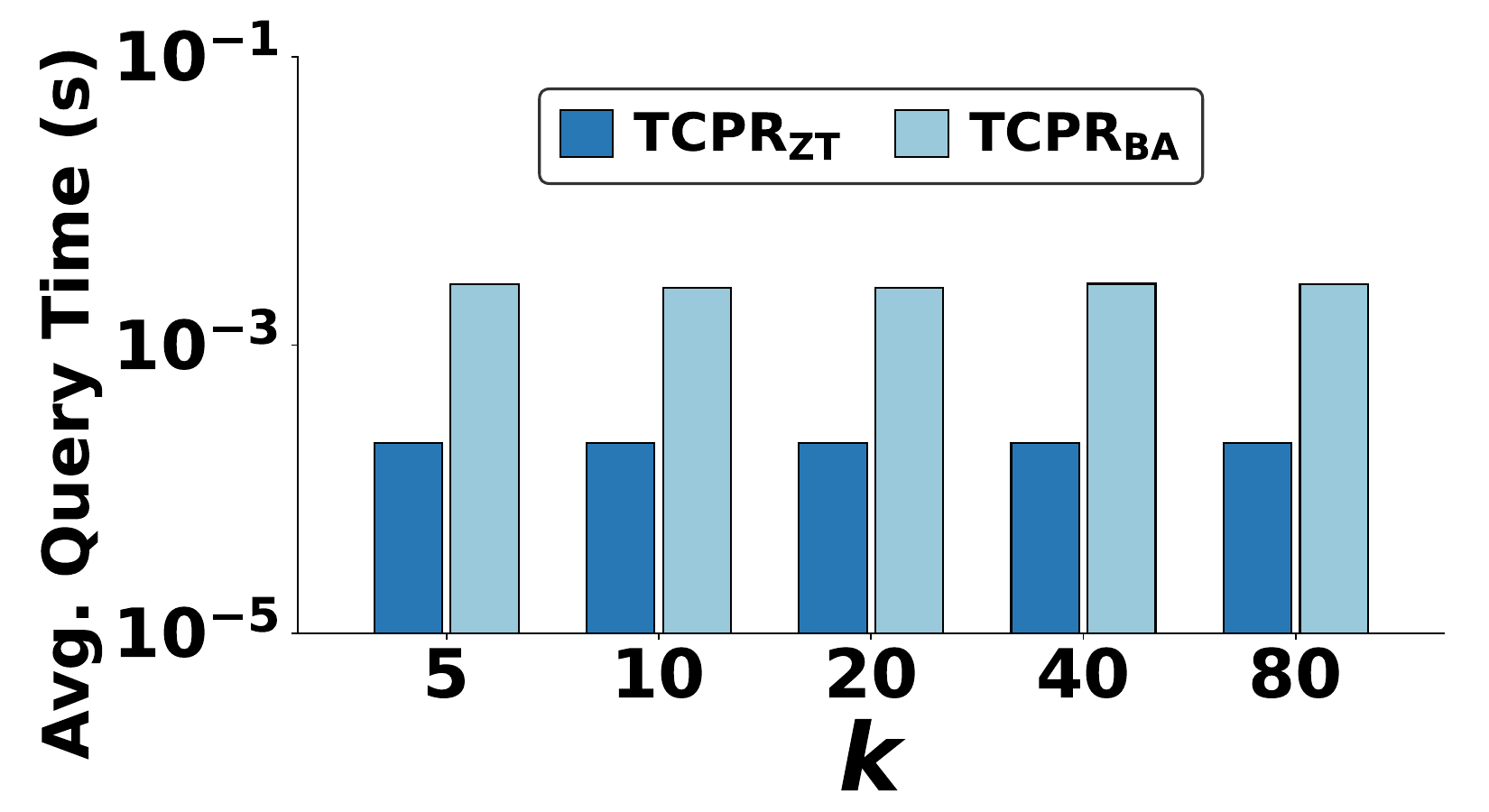}
    \caption{Query time vs. $k$}\label{fig:app:TP:k:query:SARS}
  \end{subfigure}
  \begin{subfigure}[t]{\appfigwidth}
    \includegraphics[width=\linewidth]{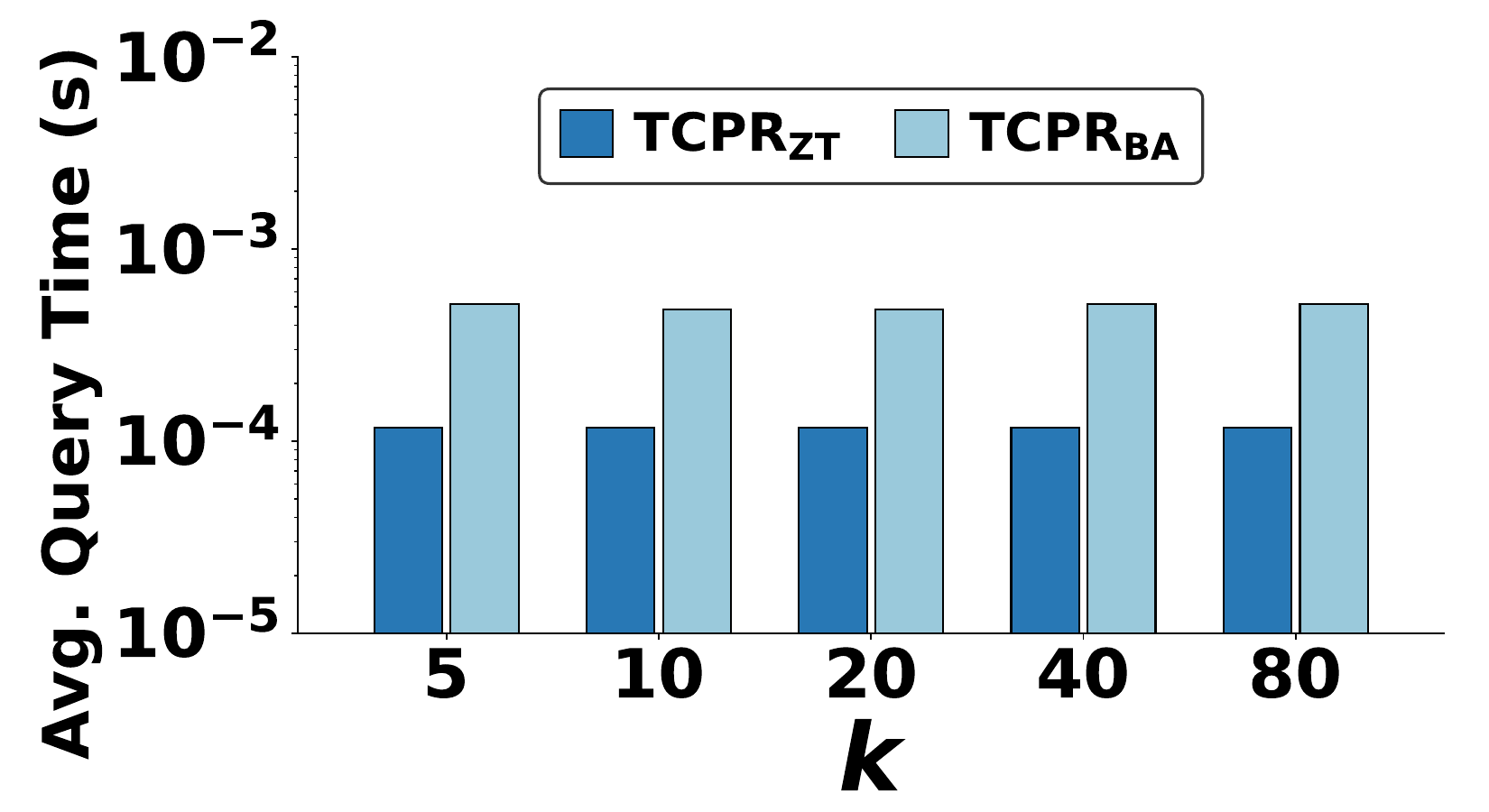}
    \caption{Query time vs. $k$}\label{fig:app:TP:k:query:SDSL}
  \end{subfigure}
  \begin{subfigure}[t]{\appfigwidth}
    \includegraphics[width=\linewidth]{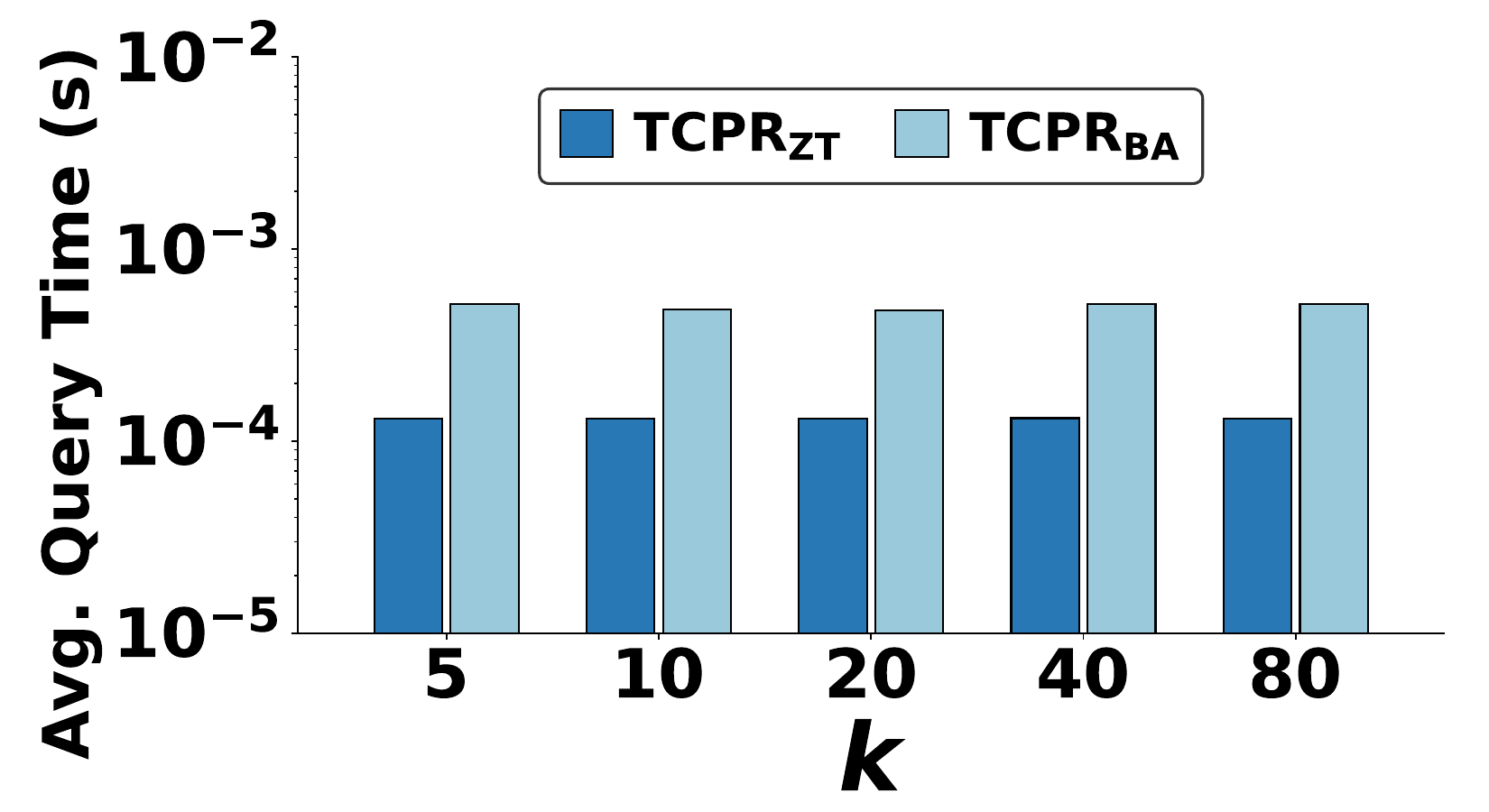}
    \caption{Query time vs. $k$}\label{fig:app:TP:k:query:WIKI}
  \end{subfigure}
  \vspace{\captionspacing}
  \vspace{+2mm}
  \caption{Query time of our \TCPR index with the \textsf{TP} scoring function vs. \TCPRBA on (a) \bst, (b) \sars, (c) \sdsl, and (d) \wiki vs. $n$; on (e) \bst, (f) \sars, (g) \sdsl, and (h) \wiki vs. $m$; on (i) \bst, (j) \sars, (k) \sdsl, and (l) \wiki vs. $q$; on (m) \bst, (n) \sars, (o) \sdsl, and (p) \wiki vs. $k$.}\label{fig:app:TP:query}
\end{figure}

\begin{figure}[ht]
  \centering
  \begin{subfigure}[t]{\appfigwidth}
    \includegraphics[width=\linewidth]{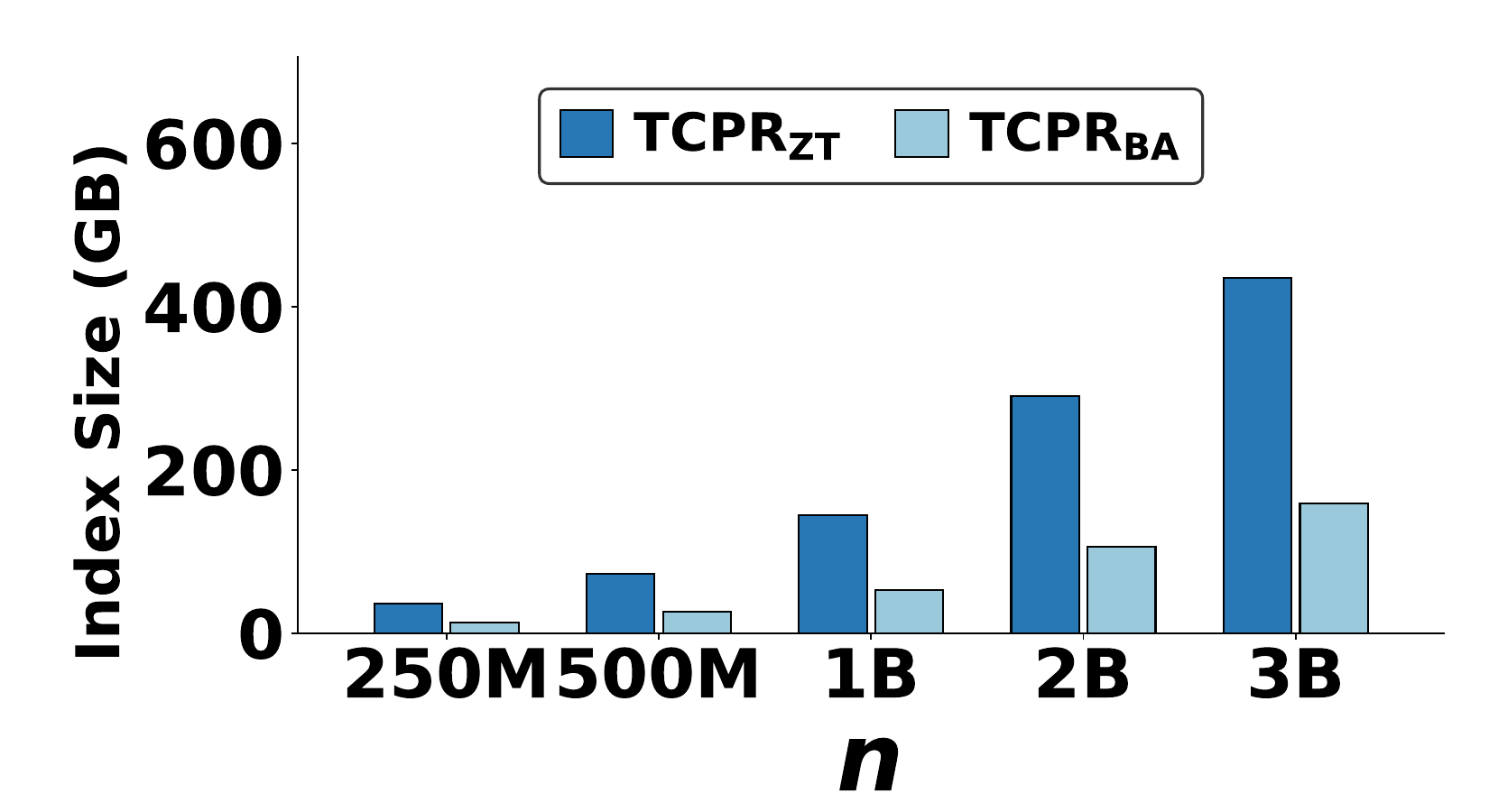}
    \caption{Index size vs. $n$}\label{fig:app:TP:n:index:BST}
  \end{subfigure}
  \begin{subfigure}[t]{\appfigwidth}
    \includegraphics[width=\linewidth]{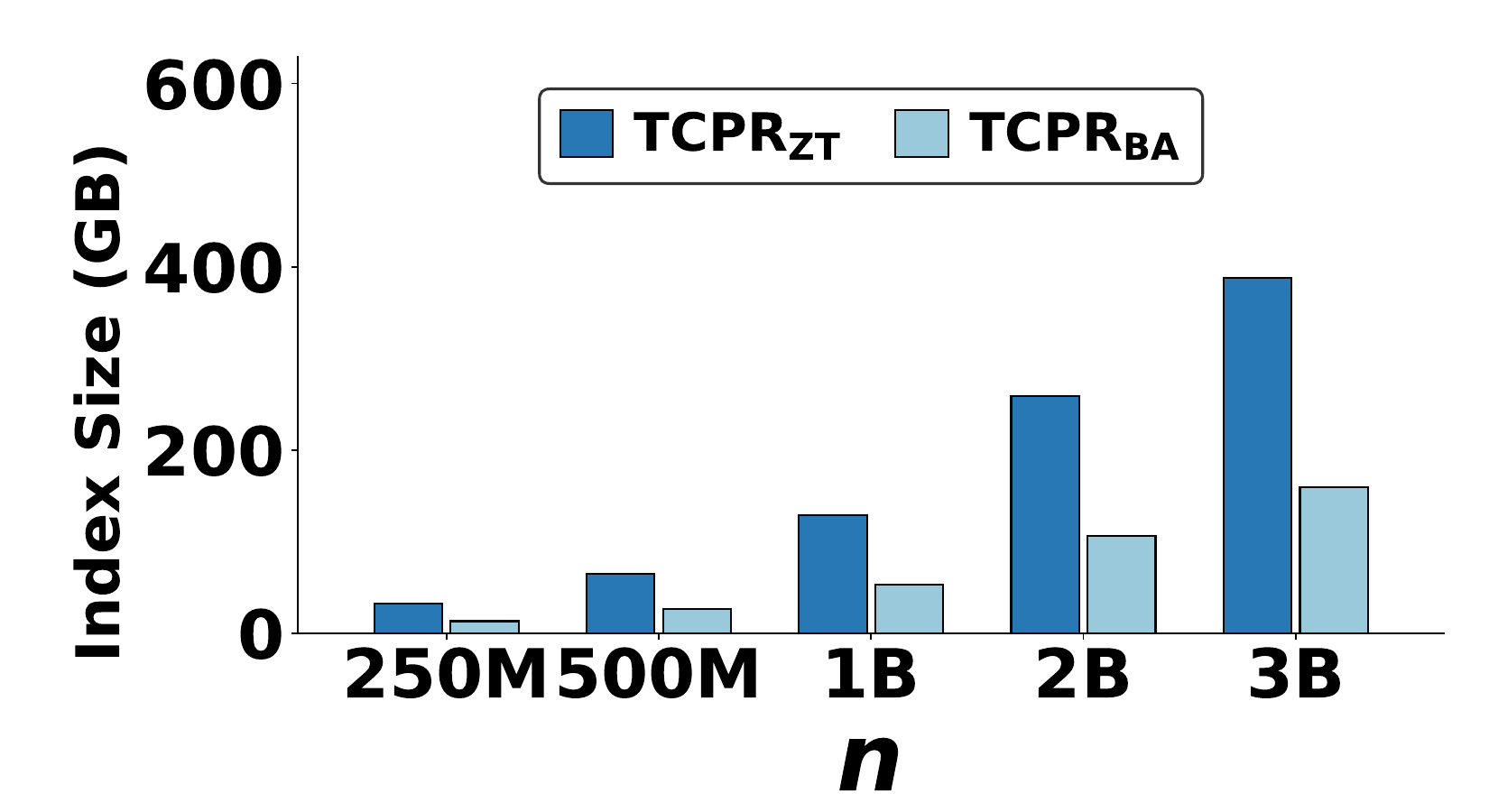}
    \caption{Index size vs. $n$}\label{fig:app:TP:n:index:SARS}
  \end{subfigure}
  \begin{subfigure}[t]{\appfigwidth}
    \includegraphics[width=\linewidth]{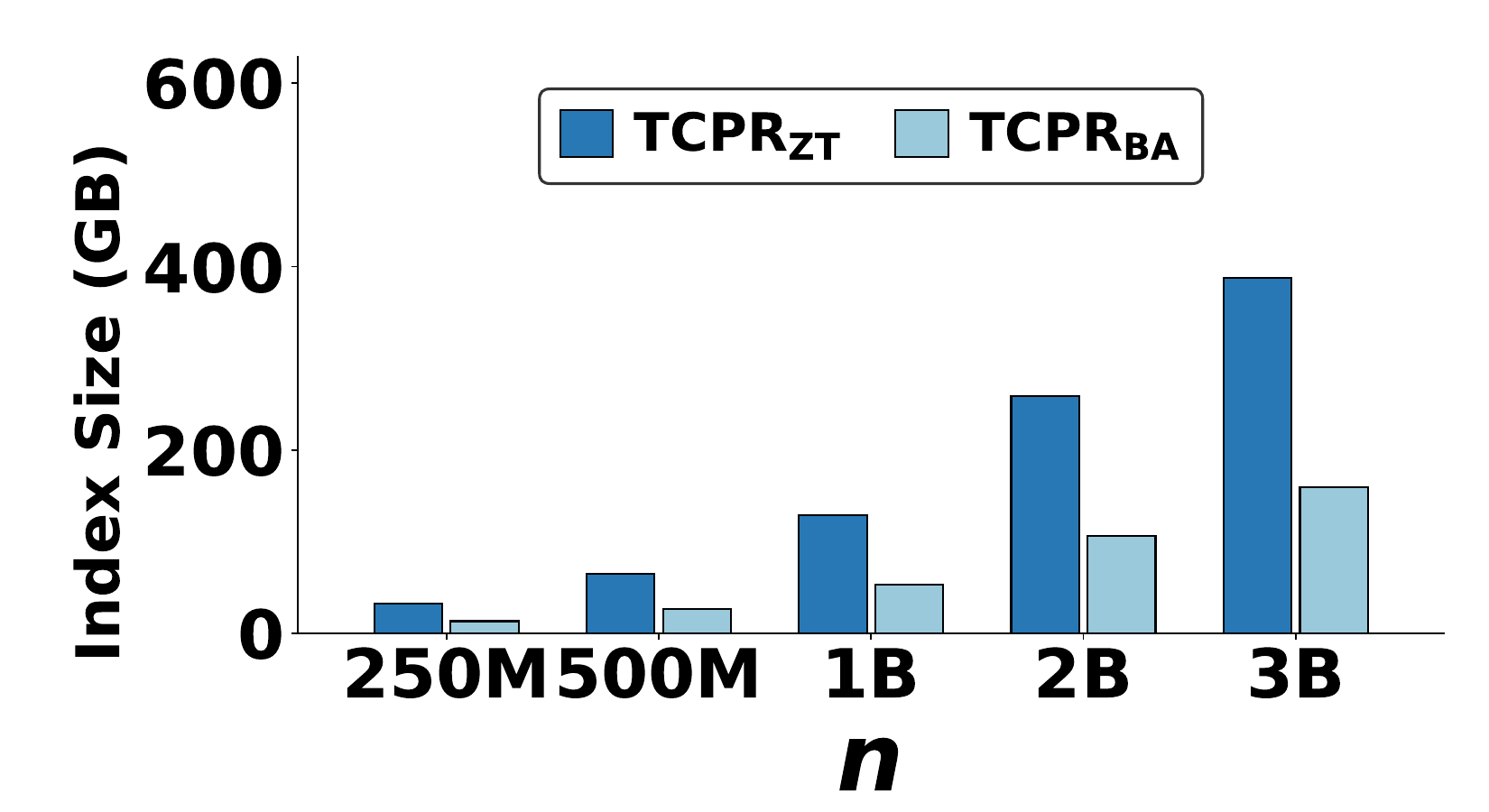}
    \caption{Index size vs. $n$}\label{fig:app:TP:n:index:SDSL}
  \end{subfigure}
  \begin{subfigure}[t]{\appfigwidth}
    \includegraphics[width=\linewidth]{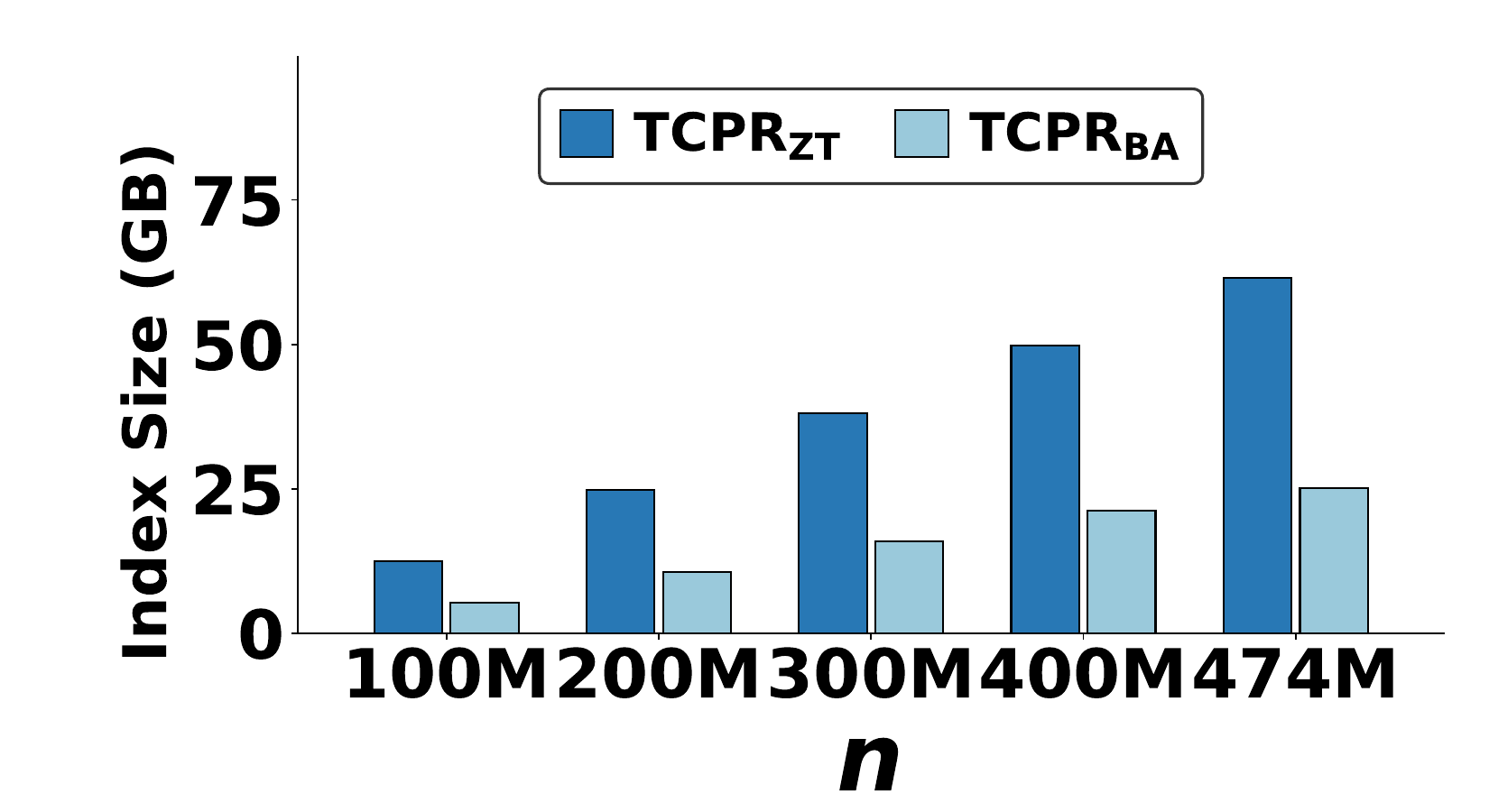}
    \caption{Index size vs. $n$}\label{fig:app:TP:n:index:WIKI}
  \end{subfigure}\\[0pt]
  \begin{subfigure}[t]{\appfigwidth}
    \includegraphics[width=\linewidth]{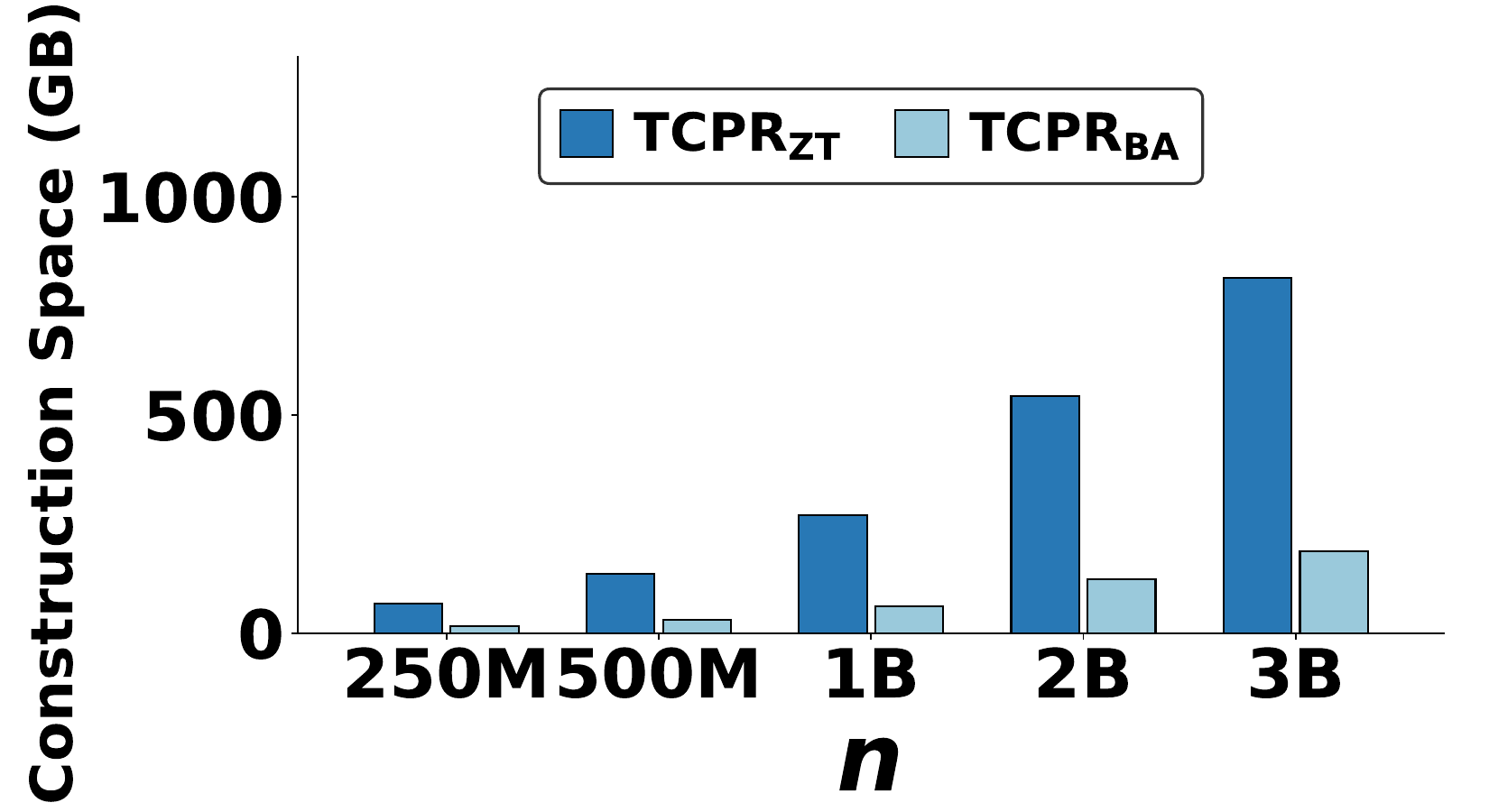}
    \caption{Constr.\ space vs. $n$}\label{fig:app:TP:n:rss:BST}
  \end{subfigure}
  \begin{subfigure}[t]{\appfigwidth}
    \includegraphics[width=\linewidth]{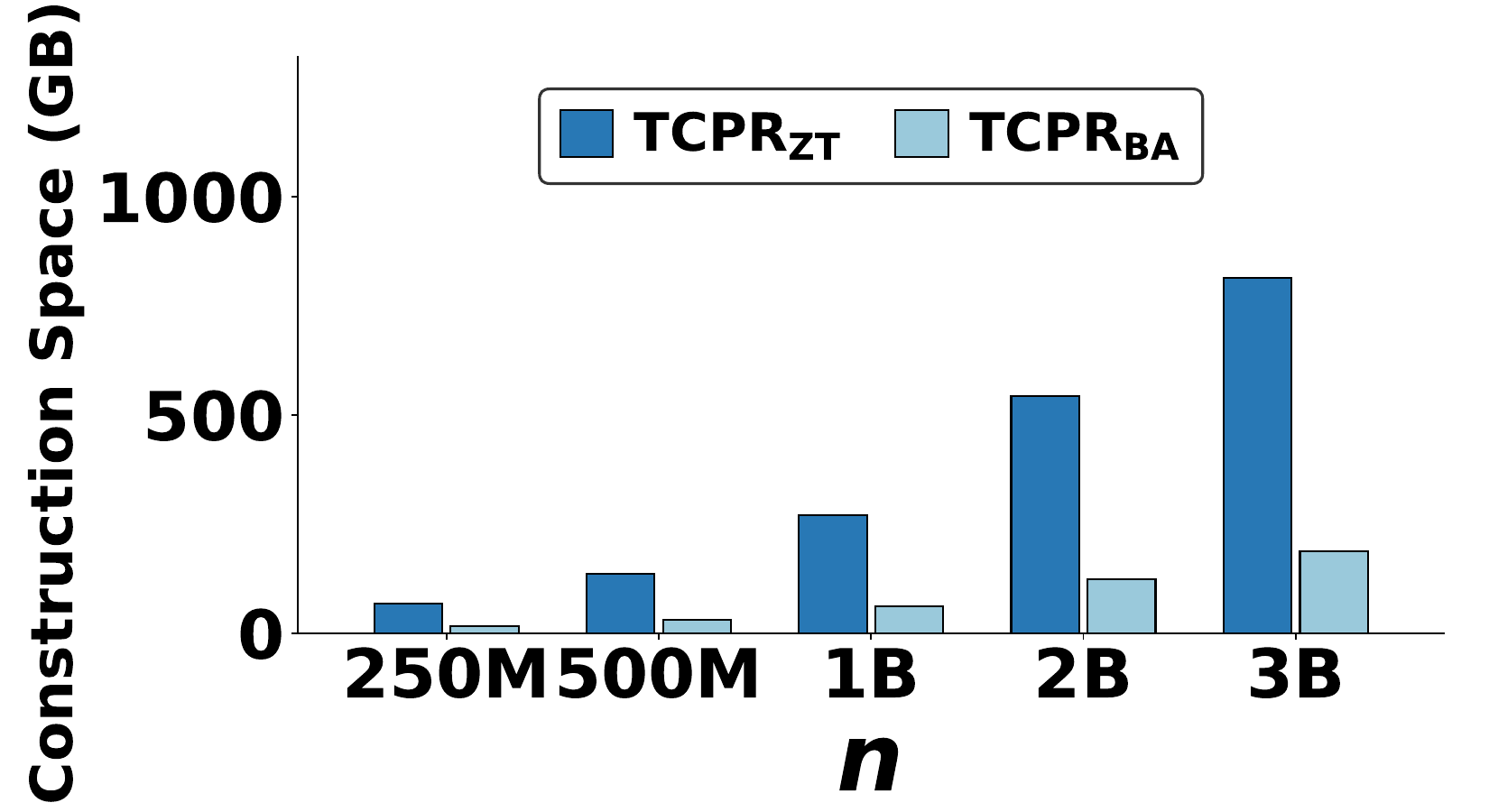}
    \caption{Constr.\ space vs. $n$}\label{fig:app:TP:n:rss:SARS}
  \end{subfigure}
  \begin{subfigure}[t]{\appfigwidth}
    \includegraphics[width=\linewidth]{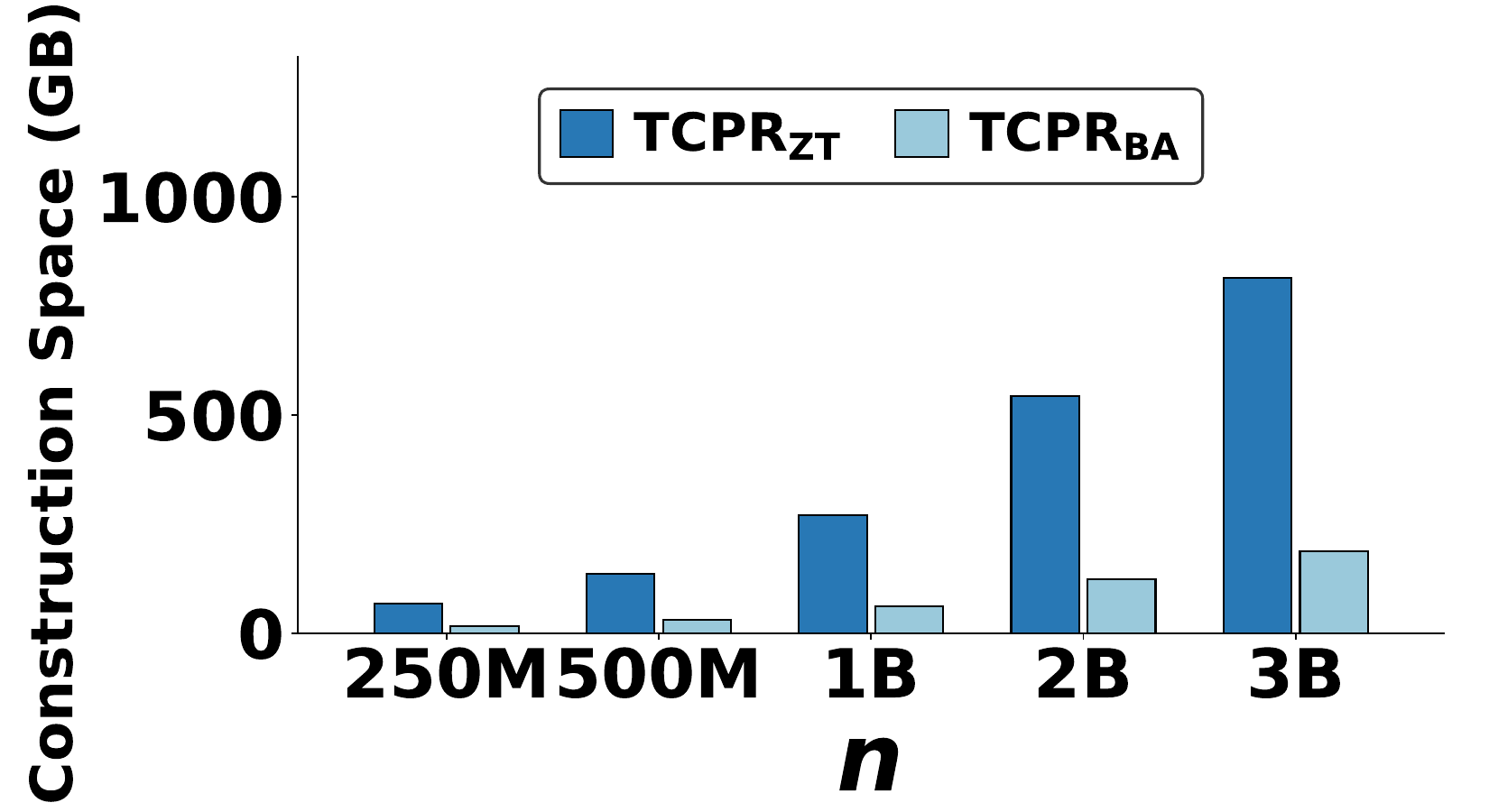}
    \caption{Constr.\ space vs. $n$}\label{fig:app:TP:n:rss:SDSL}
  \end{subfigure}
  \begin{subfigure}[t]{\appfigwidth}
    \includegraphics[width=\linewidth]{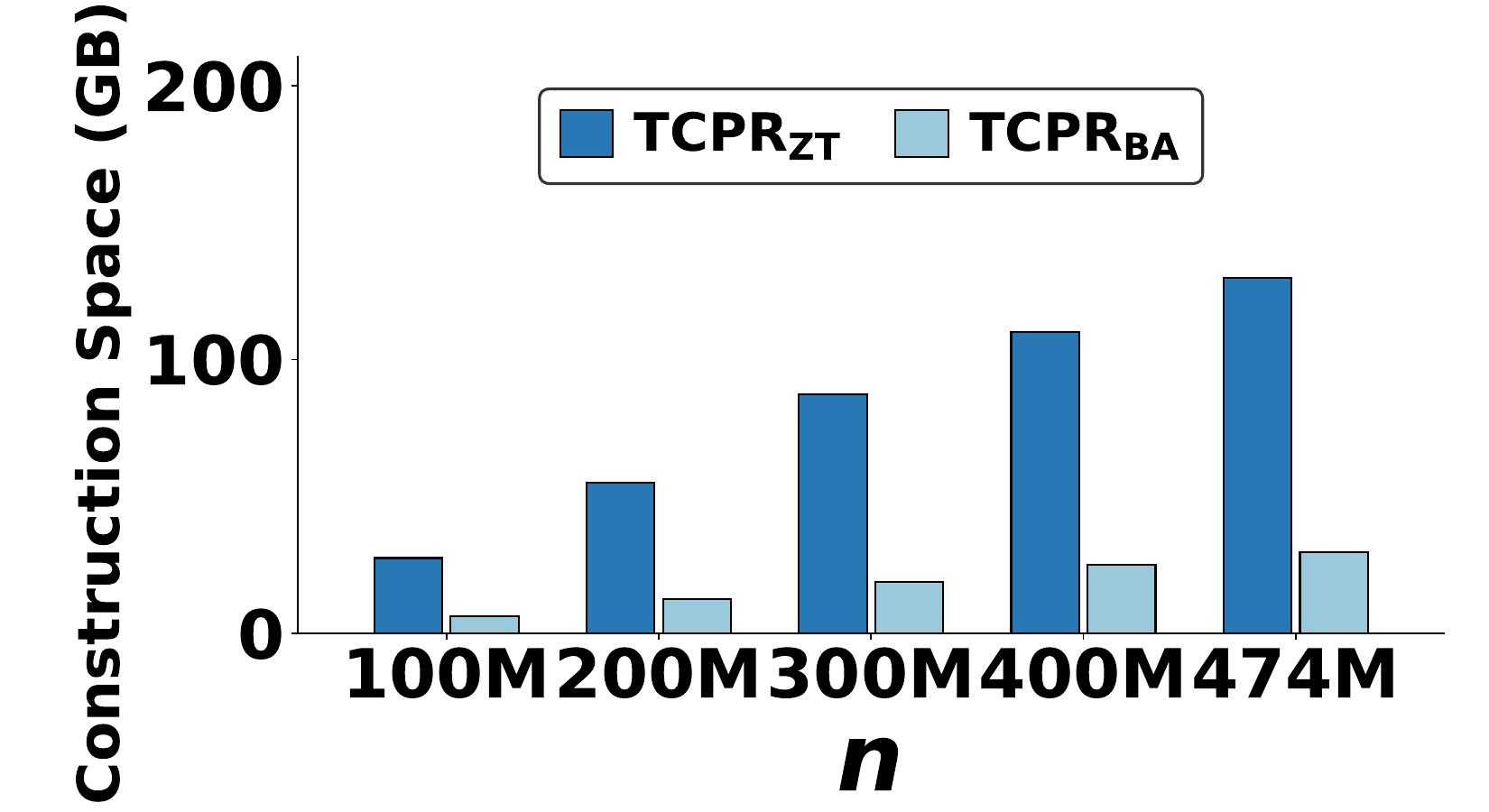}
    \caption{Constr.\ space vs. $n$}\label{fig:app:TP:n:rss:WIKI}
  \end{subfigure}\\[0pt]
  \begin{subfigure}[t]{\appfigwidth}
    \includegraphics[width=\linewidth]{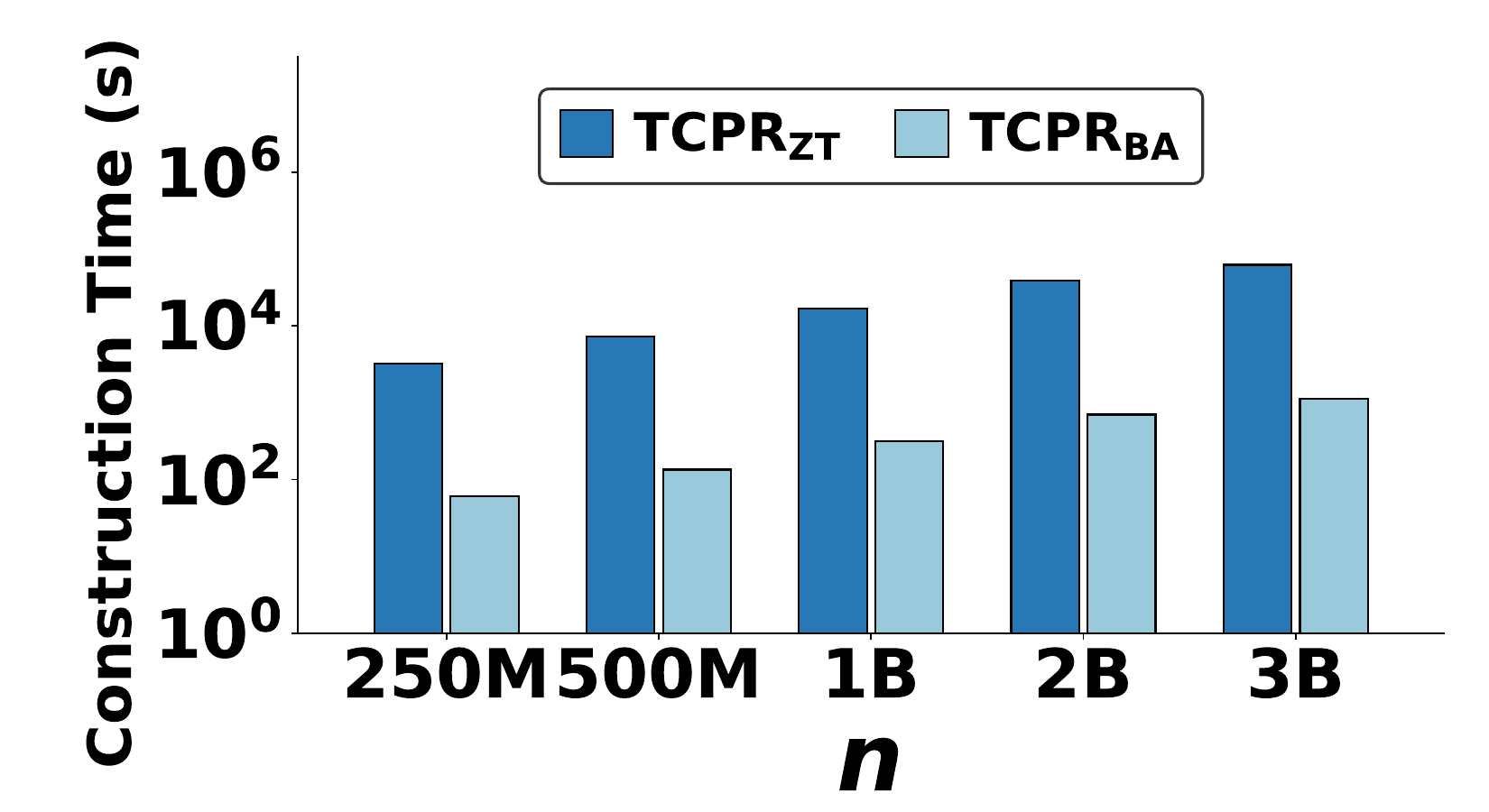}
    \caption{Constr.\ time vs. $n$}\label{fig:app:TP:n:build:BST}
  \end{subfigure}
  \begin{subfigure}[t]{\appfigwidth}
    \includegraphics[width=\linewidth]{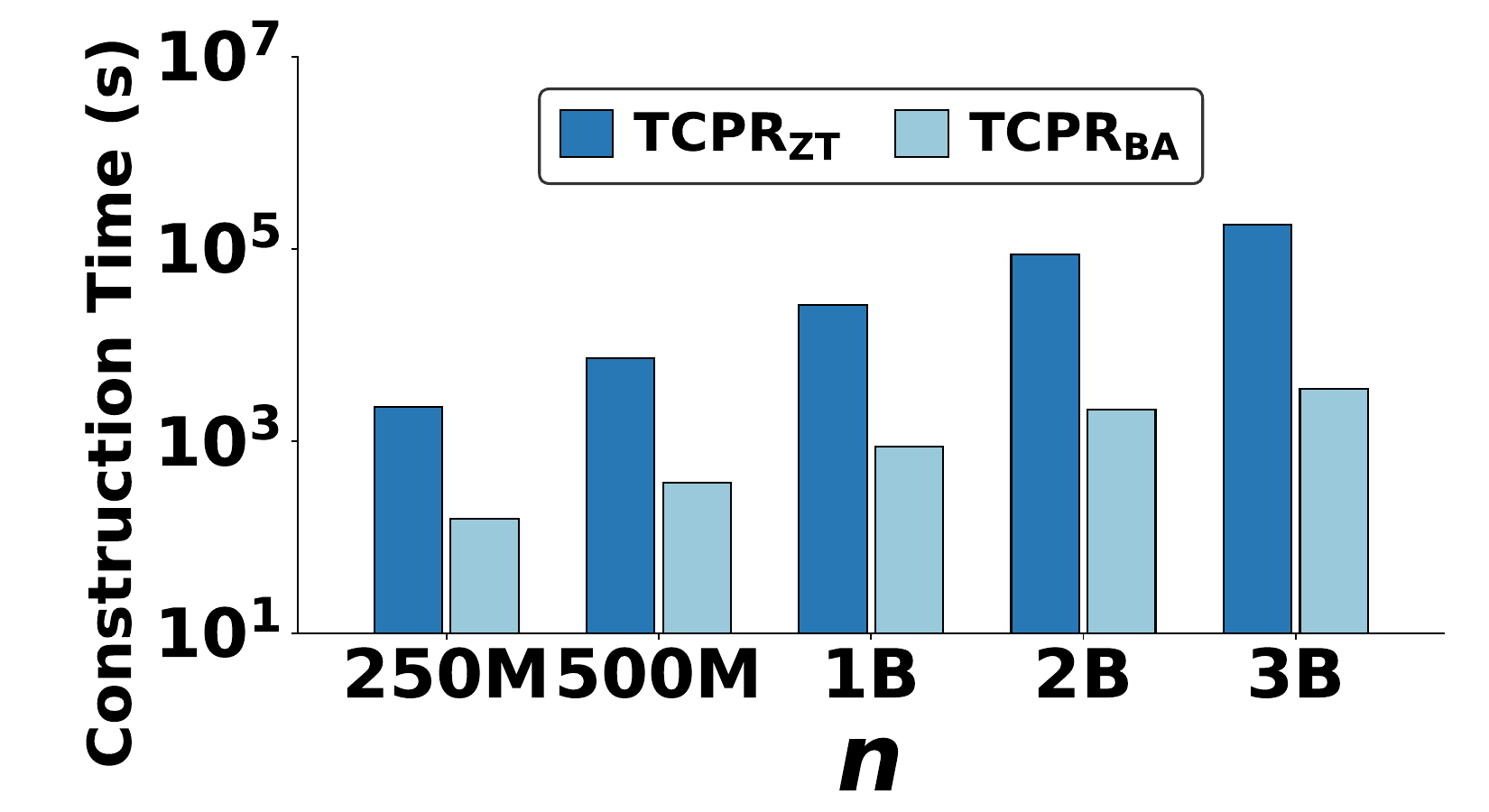}
    \caption{Constr.\ time vs. $n$}\label{fig:app:TP:n:build:SARS}
  \end{subfigure}
  \begin{subfigure}[t]{\appfigwidth}
    \includegraphics[width=\linewidth]{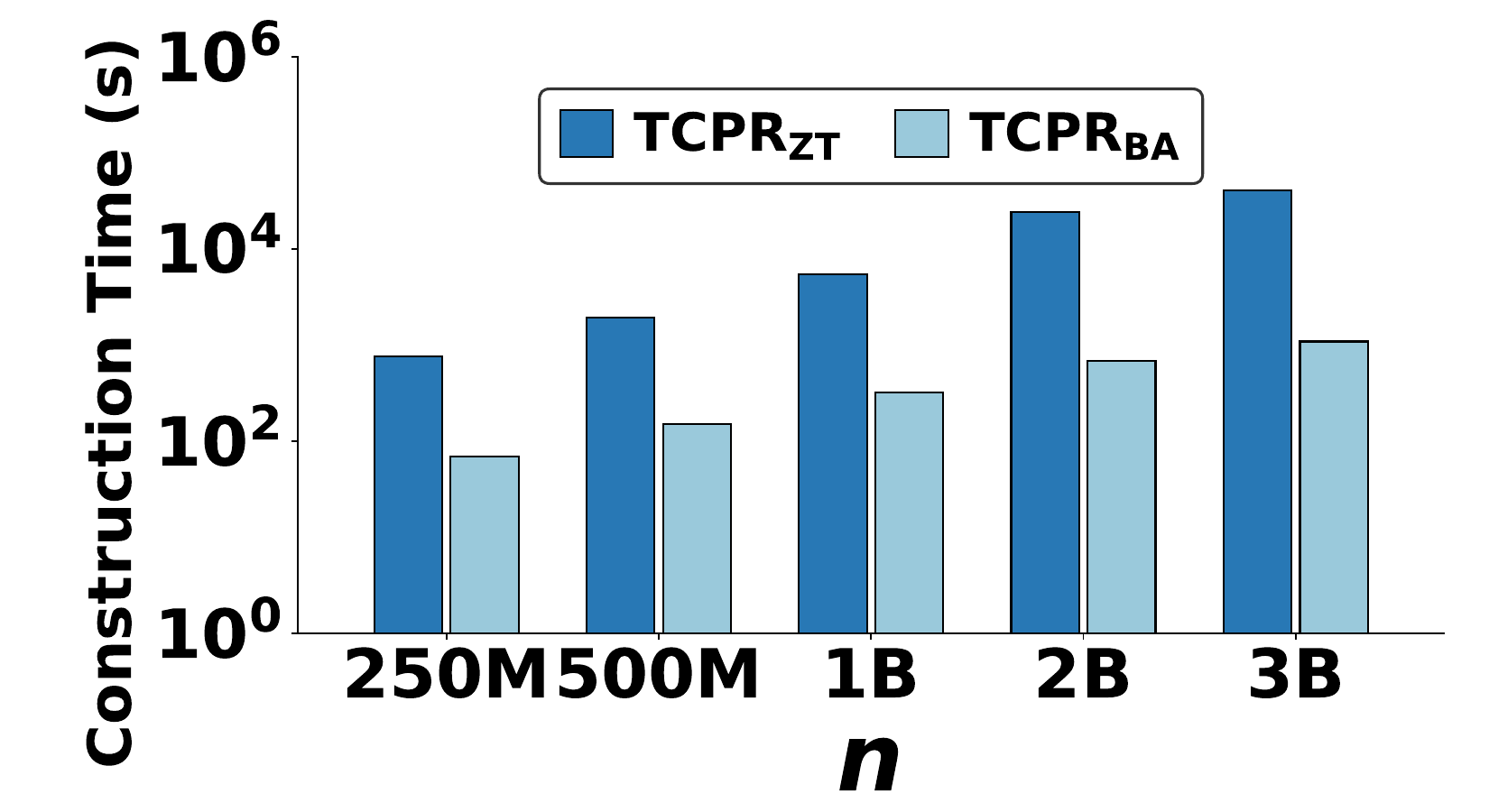}
    \caption{Constr.\ time vs. $n$}\label{fig:app:TP:n:build:SDSL}
  \end{subfigure}
  \begin{subfigure}[t]{\appfigwidth}
    \includegraphics[width=\linewidth]{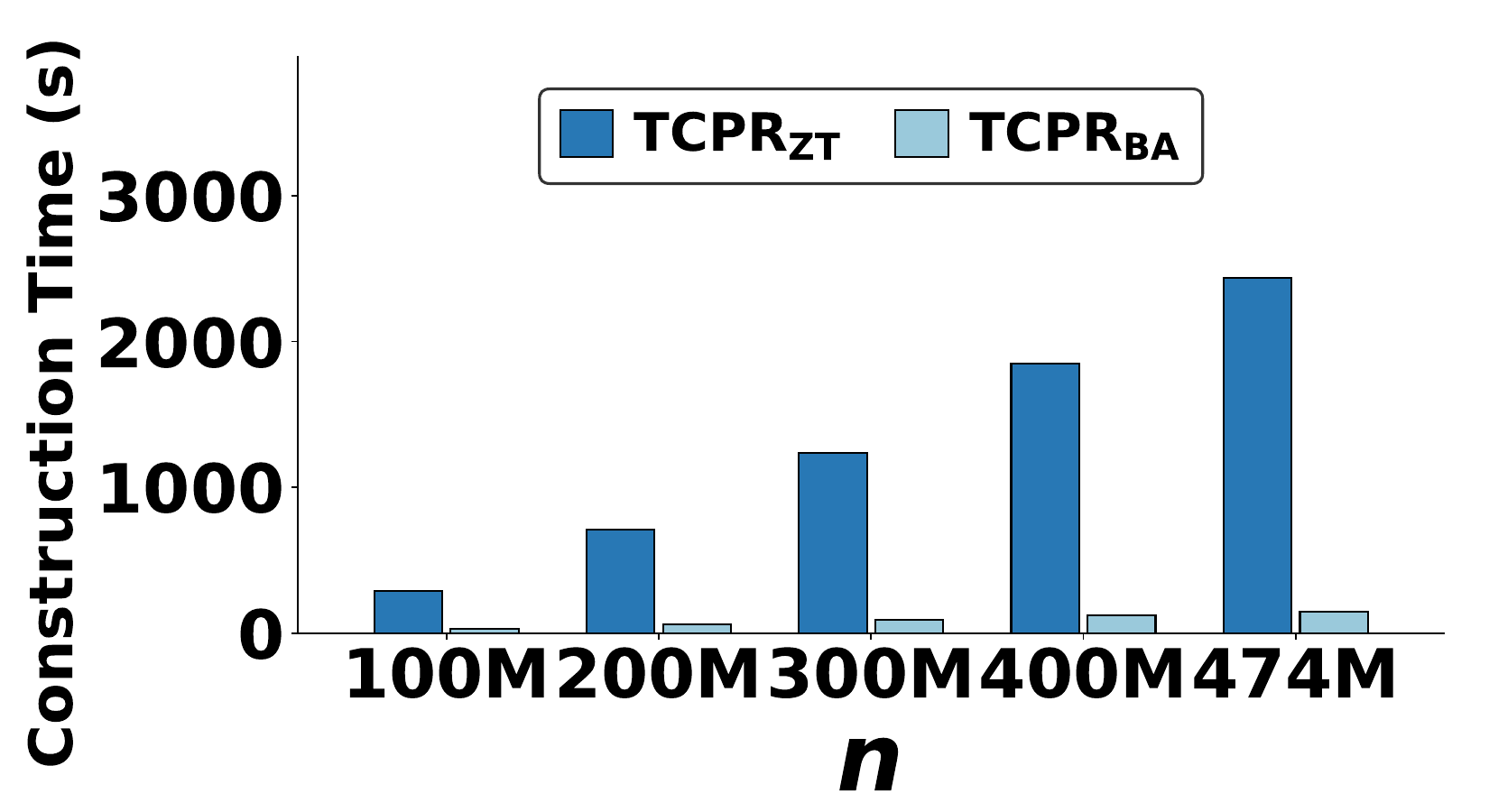}
    \caption{Constr.\ time vs. $n$}\label{fig:app:TP:n:build:WIKI}
  \end{subfigure}
  \vspace{\captionspacing}
  \vspace{+2mm}
  \caption{Index size of our \TCPR index with the \textsf{TP} scoring function vs. \TCPRBA on (a) \bst, (b) \sars, (c) \sdsl, and (d) \wiki vs. $n$; construction space of our \TCPR index with the \textsf{TP} scoring function vs. \TCPRBA on (e) \bst, (f) \sars, (g) \sdsl, and (h) \wiki vs. $n$; construction time of our \TCPR index with the \textsf{TP} scoring function vs. \TCPRBA on (i) \bst, (j) \sars, (k) \sdsl, and (l) \wiki vs. $n$.}\label{fig:app:TP:cost}
\end{figure}

\end{document}